\documentclass[12pt]{article}
\usepackage{amsthm,amssymb,bbm}
\usepackage{amsmath} % \numberwithin{equation} doesn't exist without this package.
\numberwithin{equation}{section}
\usepackage{xr-hyper}
\usepackage{mathtools}
\usepackage{floatpag}
\usepackage{natbib}
\usepackage{nameref}
\usepackage{multirow}
\usepackage[pdftex]{graphicx}
\usepackage[graphicx]{realboxes}
\usepackage{epstopdf}
\usepackage{subfigure}
\usepackage{enumerate}   
\usepackage{makecell}
\usepackage{booktabs}
\usepackage{relsize}
\usepackage[mathscr]{eucal}
\usepackage{omega}
\usepackage{array}
\usepackage{comment}
\usepackage{appendix}
\usepackage{mathabx} % for '\widecheck' macro
\usepackage{url}
\usepackage{algorithm}
\usepackage{algorithmic}
\usepackage{bm}
\usepackage{wrapfig}
\usepackage{lipsum}
\usepackage{mathrsfs}
\usepackage{dsfont}
\usepackage{titling}
\usepackage{color}
\usepackage{multirow}
\usepackage[usenames,dvipsnames,svgnames,table]{xcolor}
\usepackage[colorlinks,
linkcolor=red,
anchorcolor=blue,
citecolor=blue
]{hyperref}
 \usepackage{dutchcal}

\newcommand{\blind}{1}
\newcommand{\bds}{\boldsymbol}

\newcommand \rF{\mathrm{F}}

\def\##1\#{\begin{align}#1\end{align}}
\def\$#1\${\begin{align*}#1\end{align*}}

\newcommand {\HDBand}{\circ \{B_{k_2}({\bf 1}_{q})\otimes B_{k_1}({\bf 1}_{p})\}}
\newcommand {\DBCovZ}{\tilde{\M \Sigma}_{0,\mathcal{B}}(k_1,k_2)}
\newcommand {\DBCov}{\tilde{\M \Sigma}_{\mathcal{B}}(k_1,k_2)}
\newcommand {\HDTaper}{\circ \{T_{k_2}({\bf 1}_{q})\otimes T_{k_1}({\bf 1}_{p})\}}

\newcommand {\DTCovZ}{\tilde{\M \Sigma}_{0,\mathcal{T}}(k_1,k_2)}
\newcommand {\DTCov}{\tilde{\M \Sigma}_{\mathcal{T}}(k_1,k_2)}
\newcommand{\MB}{\mathcal{B}}

\newcommand {\cov}{\textnormal{cov}}
\newcommand {\vecc}{\textnormal {vec}}
\def\T{{ \mathrm{\scriptscriptstyle T} }} %%%transpose operator
\newcommand{\Rom}[1]{\text{\uppercase\expandafter{\romannumeral #1\relax}}}
\newcommand{\bee}{\begin{equation}\begin{aligned}}

\newcommand*{\p}{\text{Pr}}
\newcommand{\M}[1]{{{\mathbf{\MakeUppercase{#1}}}}}
\newcommand{\ee}{\end{aligned} \end{equation}}
\newcommand{\eq}{\end{quote}}
\newcommand{\bqp}{\begin{quote}\begin{parts}}
\newcommand{\epq}{\end{parts}\end{quote}}
\newcommand{\ki}{(k_1)}
\newcommand{\kii}{(k_2)}

\newcommand{\E}{\mathbb{E}}
\newcommand{\emm}{\end{bmatrix}}

\newcommand*{\vv}{\text{vec}}
\newcommand{\F}{\text{F}}

\newcommand{\beginsupplement}{%
        \setcounter{table}{0}
        \renewcommand{\thetable}{S\arabic{table}}%
        \setcounter{figure}{0}
        \renewcommand{\thefigure}{S\arabic{figure}}%
                \setcounter{section}{0}
        \renewcommand{\thesection}{S.\arabic{section}}%
     }

 \usepackage[top=1in,bottom=1in]{geometry}

\renewcommand{\thetable}{S\arabic{table}}
\renewcommand{\thefigure}{S\arabic{figure}}
\allowdisplaybreaks
\numberwithin{equation}{section} % This line resets equation numbering when starting a new section.
\protected\def\ignorethis#1\endignorethis{}
\let\endignorethis\relax
\def\TOCstop{\addtocontents{toc}{\ignorethis}}
\def\TOCstart{\addtocontents{toc}{\endignorethis}}

\usepackage{titletoc}
    \usepackage{lipsum} 
    \titlecontents{chapter}
    [5.5em] %5.3
    {\bigskip}
    {\contentslabel[\bfseries\textsc{\chaptername}~\thecontentslabel]{5.5em}\textbf}%\thecontentslabel
    {\hspace*{-5.5em}\textbf}% unnumbered chapters
    {\titlerule*[1pc]{.}\contentspage}[\smallskip]%
 \titlecontents{section}
 [0em] % i
 {\smallskip}
 {\thecontentslabel\enspace}%\thecontentslabel
 {\hspace*{-5.5em}}
 {\titlerule*[1pc]{.}\contentspage}%]

 \titlecontents{subsection}
 [2em] %
 {\smallskip}
 {\thecontentslabel\enspace}%\thecontentslabel
 {\hspace*{7.12em}}
 {\titlerule*[1pc]{.}\contentspage}

 \titlecontents{subsubsection}
 [5em] %
 {\smallskip}
 {\thecontentslabel\enspace}%\thecontentslabel
 {\hspace*{7.12em}}
 {\titlerule*[1pc]{.}\contentspage}

\begin{document}

\def\spacingset#1{\renewcommand{\baselinestretch}
{#1}\small\normalsize} \spacingset{1}

\if1\blind
{
  \title{\bf Covariance Estimation for Matrix-valued Data}
  \author{Yichi Zhang \\
  Department of Statistics, North Carolina State University\\
    Weining Shen \\
    Department of Statistics, University of California, Irvine\\
    Dehan Kong \\
    Department of Statistical Sciences, University of Toronto\\
   }
  \maketitle
} \fi

\if0\blind
{
   \title{\bf Covariance Estimation for Matrix-valued Data}
        \date{}
     \maketitle
} \fi
\TOCstop
\vspace{-1.2cm}
\begin{abstract}
Covariance estimation for matrix-valued data has received an increasing interest in applications. Unlike previous works that rely heavily on matrix normal distribution assumption and the requirement of fixed matrix size, we propose a class of distribution-free regularized covariance estimation methods for high-dimensional matrix data under a separability condition and a bandable covariance structure. Under these conditions, the original covariance matrix is decomposed into a Kronecker product of two bandable small covariance matrices representing the variability over row and column directions. We formulate a unified framework for estimating bandable covariance, and introduce an efficient algorithm based on rank one unconstrained Kronecker product approximation. The convergence rates of the proposed estimators are established, and the derived minimax lower bound shows our proposed estimator is rate-optimal under certain divergence regimes of matrix size. We further introduce a class of robust covariance estimators and provide theoretical guarantees to deal with heavy-tailed data. We demonstrate the superior finite-sample performance of our methods using simulations and real applications from a gridded temperature anomalies dataset and a S\&P 500 stock data analysis.
\end{abstract}

\vfill
\newpage
\spacingset{1.5} 
\section{Introduction}
 Matrix-valued data have received considerable interests in various applications. In environmental studies, the outcome of interest (e.g., temperature, humidity, air quality) is measured over a range of geographical regions. It is hence natural to represent the resulting data in a matrix form with two dimensions corresponding to latitude and longitude.  Examples of matrix-valued data also include two-dimensional digital imaging data, brain surface data and colorimetric sensor array data. There have been a few recent studies on regression analysis for matrix-valued data \citep{Zhou2014, wang2017generalized, kong2020l2rm, hu2020matrix, hu2021nonparametric}.

In this paper, we are interested in estimating the covariance of the matrix-valued data. Covariance estimation is a fundamental problem in multivariate data analysis. A large collection of statistical and machine learning methodologies including the principal component analysis, linear discriminant analysis, regression analysis and clustering analysis, require the knowledge of the covariance matrices. Denote $ \Xb \in  \RR^{p\times q}$ two dimensional matrix-valued data, and $\text{vec}(\cdot)$ the vectorization operator that stacks the columns of a matrix into a column vector. The covariance of $ \Xb $ is defined as $\text{cov}\{\vecc (\Xb)\}=\boldsymbol\Sigma^*\in \mathbb{R}^{pq\times pq} $. A naive estimate of $\boldsymbol\Sigma^*$ is the sample covariance. However, when $ pq>n $, it performs poorly. It has been shown in \citet{Wachter1978, Johnstone2001,  Johnstone2009} that when $ pq/n\rightarrow c\in (0,\infty] $, the largest eigenvalue of the sample covariance matrix is an inconsistent estimator of the largest eigenvalue for the population covariance matrix, and the eigenvectors of the sample covariance matrix can be nearly  orthogonal to the truth. 

To overcome the ultra-high-dimensionality, structural assumptions are needed to estimate the covariance consistently. Various types of structured covariance matrices have been introduced such as bandable covariance matrices, sparse covariance matrices, and spiked covariance matrices. Several regularization methods have been developed accordingly to estimate these matrices, including banded methods \citep{bickel2008regularized, Wu2009}, tapering methods \citep{Furrer2007, Cai2010}, and thresholding methods \citep{Bickel2008threshold, ElKaroui2008, Cai2011}.
Another issue with the sample covariance is that it does not utilize the knowledge that the data actually lie in a two dimensional matrix space. To address this issue, it is common to impose a separability assumption on the covariance of $ \text{vec}(\Xb) $, i.e., $\text{Cov}\{\text{vec}(\Xb)\}= \boldsymbol\Sigma^*_2\otimes\boldsymbol\Sigma^*_1 $, where $\boldsymbol\Sigma^*_1\in \RR^{p\times p}$ and $\boldsymbol\Sigma^*_2\in \RR^{q\times q}$ represent covariances among the rows and columns of the matrices, respectively. The separability assumption helps provide a stable and parsimonious alternative to an unrestricted version of $\text{Cov}\{\text{vec}(\Xb)\}$, and equally importantly, renders for a  simple-yet-meaningful scientific interpretation. For example,  
when analyzing temperature measurements over a geographical region, this assumption helps decompose the variability in the measurements onto spatial directions (e.g., longitude and latitude). 

To account for the separability assumption when estimating the covariance of matrix-valued data, a class of methods were proposed in the literature, all based on assuming a matrix normal distribution for the data. This idea was first proposed by \citet*{Dawid1981}, and then explored by \citet*{Dutilleul1999} as they introduced an iterative algorithm for maximum likelihood estimation. \citet*{Werner2008} developed two alternative estimation methods and derived the Cram\'{e}r-Lower bound for the problem in a compact form. Beyond the matrix case, \citet*{Galecki1994} and \citet*{Hoff2011} considered separable covariance matrices estimation for tensor data under the tensor normal model. To summarize, all the aforementioned methods rely heavily on the matrix normal distribution assumption since their estimation procedures are obtained using maximum likelihood estimation (MLE). Moreover, these methods can only handle the matrices with fixed dimensions, especially for the development of asymptotic theory. It remains unclear how those methodologies can be generalized under realistic situations where the data do not satisfy a matrix normal distribution (or any presumed distribution) and how the asymptotic theory works for matrices with high dimensions.

In this paper, we consider the covariance estimation problem for matrix-valued data under a much more challenging but realistic scenario. First, our method is distribution-free, which significantly differs from all the previous likelihood approaches. Second, we allow the dimensions of matrix-valued data to be much larger than the sample size, e.g., they can diverge at the exponential rate of the sample size. Under this scenario, even if the matrix normal assumption is true, the MLE still does not exist due to overfitting. Our solution is to impose a bandable assumption on $\boldsymbol\Sigma^*_1$ and $\boldsymbol\Sigma^*_2$. This assumption has been widely adopted for time series with scientific applications \citep{Visser1995}. The resulting bandable covariance structure exhibits a natural order among variables, thus can naturally depict the spatial and/or temporal correlation of the matrix-valued data. We then incorporate the separable and bandable properties into one unified estimation framework, and propose an efficient computational algorithm to obtain banded and tapering covariance estimates. The convergence rates of the proposed estimators are derived and shown to be minimax optimal under the high-dimensional setting and appropriate tail conditions. A phase transition phenomenon of optimal bandwidth selection is revealed by analyzing the impact of parameter complexity on the minimax optimality regime. Our proof makes use of some new matrix analysis techniques including an $\epsilon$-net argument that assesses the impact of doubly bandable covariance structure, the newly-derived random matrix inequality \citep{zajkowski2020bounds} and the unilateral singular space perturbation bound \citep{cai2018rate}, which shed new insights on high-dimensional regularized covariance estimation while accounting for matrix structure in the data.
{\color{black}To deal with potentially heavy-tailed  data, we further propose truncation-based robust banded and tapering covariance estimators. The truncation level is subtly analyzed to achieve an appropriate balance in bias-variance trade-off; and the corresponding convergence rate is derived.}

The rest of the article is organized as follows. We introduce our banded and tapering covariance estimates of matrix-valued data in Section \ref{sec:2}. Section \ref{Sec_MR} provides theoretical support of our method. In Section \ref{sec:RE}, we further propose a robust banded and tapering covariance estimation procedure to deal with heavy-tailed data and provide theoretical guarantees. 
Simulations are conducted in Section \ref{sim} to evaluate the finite-sample performance of the proposed methods. In Section \ref{realdata}, we apply our method to a gridded temperature anomalies dataset. We end with some discussions in Section \ref{discussion}. Technical proofs,  additional theoretical and numerical results, and an additional S$\&$P 500 stock data analysis are presented in the Supplementary File. 

\noindent{\bf Notation}: We summarize the notation used throughout the  paper here. For a vector $\mathbf{v} \in \RR^{d}$, we denote its Euclidean norm by $\|\mathbf{v}\|$.  For a matrix $\Ab=[A_{ij}]\in \RR^{d_1\times d_2}$, we denote $\textnormal{tr}(\Ab)$ its trace and $\|\Ab\|_{\rF}$ its Frobenius norm. We also define the following matrix norms,
\bee\label{T2:norm:collect}
&\|\Ab\|_2\equiv \sup\{\|\M Ax\|_2,\|x\|_2 = 1\},
\\
&\|\M A\|_1 \equiv \sup\{\|\M Ax\|_1,\|x\|_1 = 1\} = \max_{j}\sum_{i}|A_{ij}|,
\\ 
&\|\M A\|_{\infty} \equiv \sup\{\|\M Ax\|_{\max},\|x\|_{\max} = 1\} = \max_{i}\sum_{j}|A_{ij}|,
\\
&\|\M A\|_{\max} \equiv \max_{i,j}|A_{ij}|,\  \|\M A\|_{1,1} = \sum_{i}\sum_{j}|A_{ij}|.
\ee For two matrices $\Ab\in \RR^{q\times q}$ and $\Bb\in \RR^{p\times p}$, their Kronecker product $\Ab \otimes \Bb $ is a $pq\times pq$ matrix. Denote $\M 1_d$ a $d\times d$ matrix with all elements equal to 1. Let $ \circ $ be the Hadamard product of two matrices, i.e. element-wise product. We use $a\wedge b$ and $a \vee b$ as shorthand notation of  $\min\{a,b\}$ and $\max\{a,b\}$, respectively. For an arbitrary set $S$, we use $|S|$ to denote the cardinality of $S$. We let $\text{sgn}(\cdot)$ be the sign function and $\lfloor \cdot \rfloor$ be the floor function. We write $a \precsim b$ if there exists a universal constant $C > 0$ such that $a \leq C b$. 
\vspace{-.15in}
\section{Methodology}\label{sec:2} 
Denote $ \Xb \in  \RR^{p\times q}$ a two-dimensional random matrix, and $\vecc (\cdot)$ a vectorization operator that stacks the columns of a matrix into a vector. Let $\M \Sigma^{*}  \in \cS_+^{pq\times pq}$ be the true covariance matrix of $\vecc(\M X_i)$, where $\cS_+^{d\times d} $ denotes the space of $d\times d$ positive definite matrices. Assume that $\{\Xb_i: 1\leq i \leq n\}$ are independently and identically distributed (i.i.d.) matrix-valued samples generated from $ \Xb $. The main interest of the paper is to estimate $\boldsymbol\Sigma^*$ from the sampled data, where we allow $ p>n $ and $ q> n$. 

A naive estimator of $\boldsymbol\Sigma^*$ is the sample covariance $ \hat{\boldsymbol\Sigma} = n^{-1}\sum_{i=1}^n \{\vecc (\Xb_i)-\vecc (\hat{\boldsymbol\mu})\} \{\vecc (\Xb_i)-\vecc (\hat{\boldsymbol\mu})\}^{\T}$, where $ \hat{\boldsymbol\mu}=n^{-1}\sum_{i=1}^n  \Xb_i $. Although the sample covariance $ \hat{\boldsymbol\Sigma} $ is well-behaved for fixed $p$ and $ q$, it has undesired properties when $ pq>n$. In particular, the sample covariance matrix is singular, and it may not be a consistent estimator of $ \boldsymbol\Sigma^*$. In addition, the eigenvalues are often overdispersed and may be inconsistent \citep{Bickel2008threshold, bickel2008regularized}. 

Therefore, to estimate $\boldsymbol\Sigma^*$ in high-dimensional settings, we impose an additional assumption that the covariance of $ \vecc (\Xb) $ is separable, i.e.,
\#\nonumber 
\M \Sigma^{*} = \M \Sigma^{*}_2 \otimes \M \Sigma^{*}_1 \in \cS_+^{pq\times pq},
\#
where $\boldsymbol\Sigma^*_1\in \cS_+^{p\times p}$ and $\boldsymbol\Sigma^*_2\in \cS_+^{q\times q}$ represent covariances among the rows and columns of the matrices, respectively. The separability assumption provides a stable and parsimonious alternative to an unrestricted version of $\boldsymbol\Sigma^*$, and reduces the number of parameters from $pq(pq+1)/2$ to $\left\{p(p+1)/2+q(q+1)/2\right\}$.  This assumption is commonly used in modeling matrix-valued data \citep{Dawid1981, Hoff2011} and is satisfied for several matrix-variate distributions. 
For example, matrix normal distribution admits a separable covariance structure. A $p\times q$ matrix $\M X$ follows a matrix normal distribution $\mathbf{MN}_{p,q}({\boldsymbol \mu}, \M\Sigma_1,\M \Sigma_2)$ with ${\boldsymbol \mu}\in\RR^{p\times q}, \M\Sigma_1\in \RR^{p\times p}, \M\Sigma_2\in \RR^{q\times q}$, if and only if $\vecc(\M X) \sim \M N_{pq}\big\{\vecc(\boldsymbol \mu),\M\Sigma_2\otimes \M \Sigma_1\big\},$ where $\M N_{pq}$ represents a $pq$-dimensional multivariate normal distribution.  Therefore, the separability assumption holds because $\cov\big\{\vecc(\M X)\big\} = \M\Sigma_2\otimes \M\Sigma_1$. Another example is the matrix variate $t$-distribution, where the covariance is separable under mild conditions according to Theorem 4.3.1 in \citet{gupta2018matrix}.
\begin{remark}
The separability of true covariance is a key assumption in our framework, and is recommended to be tested in the data pre-processing stage. As a common assumption in spatial statistics, neuroimaging and functional data analysis, many methods have been proposed to test this assumption. For example, \citet{Lu2005} developed likelihood ratio tests and \citet{filipiak2016score} considered score test for separability of the covariance matrices under the matrix Gaussian assumption; \citet{aston2017tests} proposed projected-based bootstrap tests under both (parametric) matrix Gaussian and nonparametric conditions.% \citet{park2019permutation} proposed a distribution-free permutation-based test. 
In our real data applications, we implement \citet{aston2017tests}'s projected-based bootstrap test because it is both theoretically guaranteed and computationally fast under high-dimensional scenario. In addition, the test is distribution-free, which is suitable for our non-parametric framework.
\end{remark}
\vspace{-.15in}
\subsection{Bandable Covariance}\label{bandable}
To estimate the covariance matrix when $p>n$ or $ q>n$, regularizing large empirical covariance matrices has been widely used in literature \citep{bickel2008regularized}. One popular way is to band the sample covariance matrix. For any $\M A = [A_{l,m}]_{d\times d}$ and $k>0$, define
\$
\cB_d(k)=\left\{\M A \in \mathbb{R}^{d\times d}: A_{l,m}=0 
{\rm ~for~any~} |l-m|> k, 1\leq l,m\leq d\right\},
\$ 
and
\bee\label{def:b}
B_{k}(\M A) = \big[A_{l,m}\cdot \textbf{I}(|l - m|\leq k)\big]_{d\times d},
\ee
where $ \textbf{I}(\cdot) $ is an indicator function. We propose to solve the following optimization problem for a given pair of tuning parameters $(k_1, k_2)$: 
\bee\label{eq:band:0}
\big(\hat{\boldsymbol\Sigma}^{\mathcal{B}}_1(k_1), \hat{\boldsymbol\Sigma}^{\mathcal{B}}_2(k_2)\big)= \argmin_{\bSigma_1\in \cB_p(k_1), \bSigma_2\in \cB_q(k_2)}\big\| \hat{\bSigma}-\bSigma_2 \otimes \bSigma_1\big\|_\F^2 .
 \ee
And the banded covariance estimate corresponding to $(k_1, k_2)$ is $ \hat{\boldsymbol\Sigma}^{\MB}(k_1, k_2)=\hat{\boldsymbol\Sigma}^\MB_2(k_2) \otimes \hat{\boldsymbol\Sigma}^\MB_1(k_1)$. Here $k_1$ and $k_2$ control the regularization level of banding. 

Surprisingly, the above problem has a closed form solution. Define $\widetilde{\boldsymbol\Sigma}_\MB(k_1, k_2)$
 to be a $ pq \times pq $ matrix satisfying $ \widetilde{\boldsymbol\Sigma}_\MB(k_1, k_2)=\hat{\boldsymbol\Sigma}\circ \{B_{k_2}({\bf 1}_{q})\otimes B_{k_1}({\bf 1}_{p})\}$, where $ {\bf 1}_p $ and $ {\bf 1}_q $ are matrices of all $1$'s with dimensions $ p\times p $ and $ q\times q $, respectively. We call $ \widetilde{\boldsymbol\Sigma}_\MB(k_1, k_2) $ a doubly banded matrix of $\hat{\boldsymbol\Sigma}$ with bandwidths $k_1$ and $k_2$. We have the following proposition whose proof is deferred to the Supplementary File. 
\begin{proposition}\label{band:equivalence}
Solving \eqref{eq:band:0} is equivalent to solving the following optimization problem:
\bee\label{eq:band}
\big(\hat{\boldsymbol\Sigma}^{\mathcal{B}}_1(k_1), \hat{\boldsymbol\Sigma}^{\mathcal{B}}_2(k_2)\big) = \argmin_{\boldsymbol\Sigma_1, \boldsymbol\Sigma_2}\| \widetilde{\boldsymbol\Sigma}_\MB(k_1, k_2)-\boldsymbol\Sigma_2 \otimes \boldsymbol\Sigma_1\|_\rF^2 . 
\ee
\end{proposition}
This proposition provides an efficient way to solve the optimization problem \eqref{eq:band:0}. In particular, one can first obtain $\widetilde{\boldsymbol\Sigma}_\MB(k_1, k_2)$ by doubly banding $\hat{\boldsymbol\Sigma}$, and then solve the rank one unconstrained Kronecker product approximation \eqref{eq:band} based on the method in \citet*{van1993approximation} and \citet*{Pitsianis1997}. {\color{black}Solutions $\hat{\boldsymbol\Sigma}^{\mathcal{B}}_1(k_1), \hat{\boldsymbol\Sigma}^{\mathcal{B}}_2(k_2)$ are identified up to scale, while $\hat{\M \Sigma} =\hat{\boldsymbol\Sigma}^{\mathcal{B}}_2(k_2)\otimes \hat{\boldsymbol\Sigma}^{\mathcal{B}}_1(k_1) $ is unique. The procedure of solving \eqref{eq:band} implies that our proposed method is a \textit{spectral method}, which has  been  applied to a wide class of statistical problems. We refer interested readers to \citet{chen2020spectral} for a recent survey therein. More discussions on solving \eqref{eq:band} are deferred to Section \ref{sec:sKPA} in the Supplementary File.}

To implement the rank one unconstrained Kronecker product approximation, we adopt the algorithm proposed in \citet*{Batselier2015}. 
This algorithm is implemented using the Matlab package {\tt TKPSVD}, which can be downloaded at \href{https://github.com/kbatseli/TKPSVD}{https://github.com/kbatseli/TKPSVD}.

There are two tuning parameters $ k_1$ and $k_2$ involved in our estimation procedure. Theoretically, we will show in Section \ref{Sec_MR} that when $ k_1$ and $k_2$  are chosen appropriately, our estimator will be consistent even when both $p$ and $ q $ diverge at the exponential order of the sample size $n$. In practice, we apply the resampling procedure proposed in \citet{bickel2008regularized} to select the optimal bandwidths $ k_1 $ and $ k_2 $. In particular, we randomly split the original data into a training set and a test set, with sample sizes $ n_1 $ and $ n_2=n-n_1 $, respectively. We use the training set to estimate the covariance matrix $ \hat {\boldsymbol\Sigma}^{ta}(k_1,k_2) $ using our procedure, and compare with the sample covariance matrix of the test sample $ \hat {\boldsymbol\Sigma}^{te} $. We repeat the random split procedure for $ N $ times, and let $ \hat {\boldsymbol\Sigma}^{ta}_{\nu}(k_1,k_2) $ and $ \hat {\boldsymbol\Sigma}^{te}_{\nu} $ denote the estimates from the $ \nu $th split for $ \nu=1,\ldots, N $. We select the $ (k_1,k_2) $ that minimizes
\begin{equation*}
R(k_1,k_2)=N^{-1}\sum_{\nu=1}^N \|\hat {\boldsymbol\Sigma}^{ta}_{\nu}(k_1,k_2)-\hat {\boldsymbol\Sigma}^{te}_{\nu}\|_1.
\end{equation*}
Similar to \citet{bickel2008regularized}, we choose $n_1= \lfloor {\frac{n}{3}} \rfloor$ and $n_2=n-n_1$.  
We use $N=10$ random splits throughout the paper. Denote $\hat{k}_1$ and $\hat{k}_2$ the selected optimal bandwidths. Our final banded covariance estimate is $ \hat {\boldsymbol\Sigma}^\MB(\hat{k}_1, \hat{k}_2)=\hat{\boldsymbol\Sigma}^\MB_2(\hat{k}_2) \otimes \hat{\boldsymbol\Sigma}^\MB_1(\hat{k}_1)$, where 
\bee\label{optimal:formula}
\big(\hat{\boldsymbol\Sigma}^\MB_1(\hat{k}_1), \hat{\boldsymbol\Sigma}^\MB_2(\hat{k}_2)\big)= \argmin_{\bSigma_1\in \cB_p(\hat{k}_1), \bSigma_2\in \cB_q(\hat{k}_2)}\big\| \hat{\bSigma}-\bSigma_2 \otimes \bSigma_1\big\|_\rF^2.
\ee

\subsection{Tapering Covariance}\label{tapering}
\par
Another popular technique for covariance matrix regularization is tapering \citep{bickel2008regularized, Cai2010}. For any matrix $\M A = [A_{l,m}]_{d \times d}$ and any $k\geq 0$, we define $T_k(\M A) = [T_k(\M A)_{l,m}]_{d\times d}$, where
\bee\label{def:Tk}
T_k(\M A)_{l,m} = \begin{cases}
A_{l,m} & \text{ when } |l - m|\leq \lfloor k/2 \rfloor, \\
(2 - \frac{|l - m|}{\lfloor k/2 \rfloor})A_{l,m} & \text{ when } \lfloor k/2 \rfloor < |l - m| \leq k, \\
0 & \text{ otherwise.}
\end{cases}
\ee
\par
Consider
$$
 \widetilde{\M\Sigma}_{\mathcal{T}}(k_1,k_2) = \hat{\M \Sigma} \circ \{T_{k_2}\big(\M 1_q\big)\otimes T_{k_1}\big(\M 1_p\big)\}.
$$
Analogous to \eqref{eq:band}, we propose to solve
\#\label{eq:tapering}
\big(\hat{\boldsymbol\Sigma}^{\mathcal{T}}_1(k_1), \hat{\boldsymbol\Sigma}^{\mathcal{T}}_2(k_2)\big)= \argmin_{\bSigma_1, \bSigma_2}\big\| \widetilde{\M\Sigma}_{\mathcal{T}}(k_1,k_2)-\bSigma_2 \otimes \bSigma_1\big\|_\rF^2.
\#
Then we obtain the tapering covariance estimate as $ \hat {\boldsymbol\Sigma}^{\mathcal{T}}(k_1, k_2)=\hat{\boldsymbol\Sigma}^{\mathcal{T}}_2(k_2) \otimes \hat{\boldsymbol\Sigma}^{\mathcal{T}}_1(k_1)$. 

For solving \eqref{eq:tapering} and selecting the tapering tuning parameters ($k_1$, $k_2$), we adopt the same resampling procedure as the one proposed for the banded estimate in Section \ref{bandable}. 
{\remark\label{rm:2ptaper} For the doubly banded covariance, we have $\tilde{\M \Sigma}_{\mathcal{B}}(k_1,k_2) = \hat{\M\Sigma}$ when $k_1 \geq p$ and $k_2 \geq q$. However, for the doubly tapering covariance, we need $k_1 \geq 2p, k_2 \geq 2q$ to make $\tilde{\M \Sigma}_{\mathcal{T}}(k_1,k_2) = \hat{\M\Sigma}$. Therefore, in the theoretical analysis, we always assume $k_1$ and $ k_2 $ are at most the same order as $p$ and $q$, respectively.}
{\remark\label{rm:dg} When $p$ or $q$ equals $1$, the matrix-valued data degenerate to the vector-valued data, which we refer to as degenerate regime. By the forms of our proposed optimization problems in \eqref{band:equivalence} and \eqref{eq:tapering}, it is easy to see that our proposed estimators also simplify to \citet{Bickel2008threshold}'s banded estimator and \citet{Cai2010}'s tapering estimator under this degenerate regime. Similarly with those estimators under the degenerate regime, our proposed estimators are not guaranteed to be positive semi-definite or positive-definite. It is possible to apply the eigenvalue truncation technique (see e.g. Remark 3 in \citet{Cai2010}) to resolve this issue.  In particular, we can replace the negative eigenvalues of our proposed estimator with $0$ or a small positive constant.}

\vspace{-.15in}
\section{Theoretical Results}\label{Sec_MR}
\subsection{Notation}\label{Sec_MR_N} 
Following \citet{bickel2008regularized}, we define the following uniformity class of approximately bandable covariance matrices 
\bee\label{A1}
\mathcal{F}(\varepsilon_0, \alpha)=&\bigg\{ \bSigma : ~ \max_l\sum_m \{|\sigma_{l,m}|: |l-m|>k \}\leq C_0 k^{-\alpha}~\textnormal{for all}~k\geq 1,\\
&\textnormal{and}~0<\varepsilon_0\leq \lambda_{\min}(\bSigma)\leq \lambda_{\max}(\bSigma)\leq 1/\varepsilon_0 \bigg\},
\ee with fixed constants $C_0,\alpha >0$ and some $\varepsilon_0 < 1$, and $\lambda_{\min}(\bSigma), \lambda_{\max}(\bSigma)$ are the smallest and largest eigenvalues of $\bSigma$, respectively. We further define another important class of covariance matrices \citep{Cai2010},
\bee\label{A2}
\mathcal{M}(\varepsilon_0, \alpha) = &\bigg\{\M\Sigma:~\left|\sigma_{l,m}\right| \leq C_{1}|l-m|^{-\alpha-1} \text { for } l \neq m ,\\ &\textnormal{ and }~0<\varepsilon_0\leq \lambda_{\min}(\bSigma)\leq \lambda_{\max}(\bSigma)\leq 1/\varepsilon_0\bigg\},
\ee
with fixed constants $C_1,\alpha > 0$ and some $\varepsilon_0 < 1$. One can see that $\mathcal{M}(\varepsilon_0, \alpha)$ is a subset of $\mathcal{F}(\varepsilon_0, \alpha)$ when $C_1 \leq \alpha C_0$. In other words, $\mathcal{M}(\varepsilon_0, \alpha)$ is a more restrictive class. 
\par
{\color{black}Since we focus on the bandable covariance matrix classes $\mathcal{F}(\varepsilon_0, \alpha)$ and $\mathcal{M}(\varepsilon_0, \alpha)$, the true covariance matrix in the subsequent theoretical analysis $\M\Sigma^*$ will always have eigenvalues bounded away from $+\infty$, which is consistent with the previous works dealing with bandable covariance estimation for vector-valued data \citep{Bickel2008threshold,Cai2010,cai2012adaptive,cai2016estimating}. A more detailed discussion on bounded eigenvalues of $\M \Sigma^*$ is given in Section \ref{sec:bounde} of the Supplementary File. }

We say $\vecc(\M X)$ follows a sub-Gaussian distribution if for any $t>0$ and $\|{\bf v}\| = 1$, there exists a $\rho>0$ such that 
\bee\label{subga}
\Pr \left[\Big|{\bf v}^{\T} \big[\vecc(\M X) - \E \big\{\vecc(\M X)\big\}\big]
\Big|>t\right] \leq e^{-\rho t^2 }.
\ee 
Recall $\M X_1,\dots,\M X_n$ are i.i.d. $p\times q$ random matrix samples and $\vecc(\M  X_i)\in\RR^{pq}$ their vectorizations. Denote $x^{(i)}_{l_1,l_2}$ the $l_1l_2$th entry of $\M X_i$ for $1\leq l_1\leq p$ and $1\leq l_2\leq q$.  
Denote $\precsim$ and $\succsim$ inequalities up to multiplicative universal constants and $\asymp$ an equality up to a multiplicative universal constant. For sequences $a_n, b_n$, we also write $a_n = \Theta(b_n)$ if $a_n \asymp b_n$.
\subsection{Main Results}\label{sec:T}
\subsubsection{Convergence Rate of Proposed Estimators}\label{sec:T:up}
In this section, we derive the upper bound for the convergence rates of our banded and tapering covariance matrix estimates. We derive the $\mathcal{L}^2$ convergence rates of our proposed estimators under Frobenius norm. By Markov's inequality, the $\mathcal{L}^2$ convergence rates imply the same convergence rate in probability. We particularly utilize the property of low rank matrix approximation \citep{eckart1936approximation, Pitsianis1997}, unilateral subspace perturbation bound \citep{cai2018rate}, and the newly-generalized Hanson-Wright inequality \citep{zajkowski2020bounds} in  deriving the upper error bounds for our proposed estimators. 
\par
We consider two scenarios: (i) $\vecc{(\M X_i)}$ follows i.i.d. sub-Gaussian distribution; and (ii)  $\vecc{(\M X_i)}$ satisfies an element-wise finite fourth moment condition: 
\bee\label{am:moment}
\E(|x^{(i)}_{l_1,l_2}\cdot x^{(i)}_{m_1,m_2}|^2) \leq M < +\infty,
\ee 
where $x^{(i)}_{l_1,l_2}$ and $x^{(i)}_{m_1,m_2}$ are $l_1l_2$th and $m_1m_2$th entries of $\M X_i$, for any $1\leq l_1,m_1\leq p$ and $1\leq l_2,m_2 \leq q$. 
\par
\noindent {\bf Scenario (i):} 
We first present the convergence rates of our proposed covariance estimators when $\vecc{(\M X_i)}$ satisfies sub-Gaussian tail probability bound \eqref{subga}.  
{\color{black} We consider both cases that $\M \Sigma_1^*$ and $\M \Sigma_2^*$ reside in $\mathcal{F}(\varepsilon_0,\alpha)$ and $\M \Sigma_1^*$ and $\M \Sigma_2^*$ reside in $\mathcal{M}(\varepsilon_0,\alpha)$. The result is parallel to the results in \citet{bickel2008regularized} and \citet{Cai2010}, where they derive the convergence rates of the covariance estimates under these two cases when $\M X_i$ is a random vector, respectively. }
\par
Denote $\hat{\M \Sigma}_2^{\mathcal{B}}(k_2) \otimes \hat{\M \Sigma}_1^{\mathcal{B}}(k_1)$ the proposed banded estimator and $\hat{\M \Sigma}_2^{\mathcal{T}}(k_2) \otimes \hat{\M \Sigma}_1^{\mathcal{T}}(k_1)$ the proposed tapering estimator. To unify the notation, we define $\eta \in \{\mathcal{B},\mathcal{T}\}$. The following theorem presents an upper bound on the convergence rate for the proposed banded and tapering covariance estimators obtained from \eqref{eq:band:0} and \eqref{eq:tapering} under Scenario (i). 
\begin{theorem}\label{T2}
Let $\vecc(\M X_1), \vecc(\M X_2),\dots,\vecc(\M X_n)$ be i.i.d. sub-Gaussian random vectors in $\RR^{pq}$ with true covariance $\M \Sigma^* = \M \Sigma^*_2 \otimes \M \Sigma^*_1$. Let $\small\M I_{\eta,d}(k) = \M I(\eta = \mathcal{B}, k< d-1) + \M I( \eta = \mathcal{T},k< 2d-2)$ 
for $\eta\in\{\mathcal{B},\mathcal{T}\}, d \geq 1$, and let $\small\tilde{\alpha}_{a} = \begin{cases}
2\alpha_a & \text{when }\M \Sigma_1^{*}\in \mathcal{F}(\varepsilon_0, \alpha_1),\M \Sigma_2^{*} \in \mathcal{F}(\varepsilon_0, \alpha_2)
\\
2\alpha_a + 1 & \text{when } \M \Sigma_1^{*}\in \mathcal{M}(\varepsilon_0, \alpha_1),\M \Sigma_2^{*} \in \mathcal{M}(\varepsilon_0, \alpha_2)
\end{cases}$ for $a\in\{1,2\}$. Then we have,
\bee\label{T2:res1}
&\E\Big(\frac{\|\hat{\M\Sigma}^\eta_2\kii\otimes\hat{\M\Sigma}^\eta_1\ki - \M\Sigma^*\|_{\F}^2}{pq}\Big) 
\\
&\precsim \begin{cases}\frac{k_1}{qn} + \frac{k_2}{pn}+ \M I_{\eta,p}(k_1)\cdot k_1^{-\tilde{\alpha}_1} +  \M I_{\eta,q}(k_2)\cdot k_2^{-\tilde{\alpha}_2}, & pk_1 + qk_2 \precsim n;
\\
\big(\frac{k_1k_2}{n}\big)\wedge\big(\frac{pk^2_1}{qn^2} + \frac{qk^2_2}{pn^2}  \big)+ \M I_{\eta,p}(k_1)\cdot k_1^{-\tilde\alpha_1} +  \M I_{\eta,q}(k_2)\cdot k_2^{-\tilde\alpha_2}, & pk_1 + qk_2 \succ n.
\end{cases}
\ee
\end{theorem}
To better understand the error rate in \eqref{T2:res1}, consider a simple example where $q \asymp 1$ and $\M\Sigma_1^*,\M\Sigma_2^* \in \mathcal{M}(\varepsilon_0,\alpha)$. Then the rate becomes $\min\big\{n^{-\frac{2\alpha_1 +1}{2\alpha_1 + 2}},p/n\big\}$  with $k_1 = \min\big\{n^{\frac{1}{2\alpha_1 + 2}},2p\big\}$ and $k_2=2q$, which matches exactly with the minimax optimal rate in \citet{Cai2010}. Another example is when $p = q\precsim \sqrt{n}$ and $\alpha_1 = \alpha_2$. Then \eqref{T2:res1} becomes $(k_1/pn) + \M I_{\eta,p}(k_1) k_1^{-\tilde\alpha_1} \asymp \min\big\{(pn)^{-\tilde\alpha_1/(\tilde\alpha_1 +1)},1/n\big\}$ under the optimal choice for $k_1 = k_2$, which is $\min\big\{(pn )^{1/(\tilde\alpha_1 + 1)},2p\big\}$.  In general, the selection of $k_1,k_2$ to attain the optimal convergence rate can be quite complicated depending on the divergence regimes of $p,q$; and we provide a detailed discussion in Section \ref{optbd:T2}.
\par
{\color{black}

\noindent {\bf Scenario (ii):}
The sub-Gaussian assumption (i) can be relaxed to the finite fourth moment condition (ii), with a sacrifice of the convergence rate. 
The following theorem presents an upper bound on the convergence rate for the proposed banded and tapering covariance estimators obtained from \eqref{eq:band:0} and \eqref{eq:tapering} under Scenario (ii). 
\begin{theorem}\label{T3}
Let $\vecc(\M X_1), \vecc(\M X_2),\dots,\vecc(\M X_n)$ be i.i.d. random vectors in $\RR^{pq}$ with true covariance $\M \Sigma^* = \M \Sigma^*_2 \otimes \M \Sigma^*_1$. Assume $\E(|x^{(i)}_{l_1,l_2}\cdot x^{(i)}_{m_1,m_2}|^2) \leq M < +\infty$ where $M$ is a constant that does not depend on $i, l_1,m_1, l_2,m_2$. Let $\small\M I_{\eta,d}(k)$ and $\small\tilde{\alpha}_{a} $ be the same as in Theorem \ref{T2}. Then we have,
\bee\label{T3:res1}
&\E\Big(\frac{\|\hat{\M\Sigma}^\eta_2\kii\otimes\hat{\M\Sigma}^\eta_1\ki - \M\Sigma^*\|_{\F}^2}{pq}\Big) \precsim
\frac{k_1k_2}{n}+ \M I_{\eta,p}(k_1)\cdot k_1^{-\tilde{\alpha}_1} +  \M I_{\eta,q}(k_2)\cdot k_2^{-\tilde{\alpha}_2}.
\ee
\end{theorem}

\par
The selection of $k_1,k_2$ to attain the optimal convergence rate of \eqref{T3:res1}, under different divergence regimes of $p,q$, is discussed in Section \ref{optbd:T3}.

\subsubsection{Overall Lower Bound}
In this section, we give an overall lower bound for the convergence rates of covariance matrix estimates for a special scenario, where $\M \Sigma_1^*$ and $\M \Sigma_2^*$ are in $\mathcal{M}(\varepsilon_0,\alpha)$ and $\vecc(\M X_i)$ follows i.i.d. sub-Gaussian distribution. This scenario is a special case of Scenarios (i) and (ii) with either $\M \Sigma_1^{*}\in \mathcal{F}(\varepsilon_0, \alpha_1),\M \Sigma_2^{*} \in \mathcal{F}(\varepsilon_0, \alpha_2)$ or $\M \Sigma_1^{*}\in \mathcal{M}(\varepsilon_0, \alpha_1),\M \Sigma_2^{*} \in \mathcal{M}(\varepsilon_0, \alpha_2)$, as considered in Section \ref{sec:T:up}. So the lower bound we present here can be compared with the upper bounds presented in both Theorem \ref{T2} and Theorem \ref{T3}.
\begin{theorem}\label{T:low}
Let $\mathcal{P}^n_{\varepsilon_0,\alpha_1,\alpha_2}$ denote the class of distributions of $\{\vecc(\M X_i)\}_{i = 1}^n$, such that $\vecc(\M X_1), \vecc(\M X_2),\dots,\vecc(\M X_n)$ are i.i.d. sub-Gaussian random vectors in $\mathbb{R}^{pq}$ with any true covariance $\M \Sigma^* = \M \Sigma^*_2 \otimes \M \Sigma^*_1$, where $\M \Sigma_1^{*}\in \mathcal{M}(\varepsilon_0, \alpha_1),\M \Sigma_2^{*} \in \mathcal{M}(\varepsilon_0, \alpha_2)$. Let $\hat{\M \Sigma}_n$ be any possible covariance estimator based on $\{\vecc(\M X_i)\}_{i = 1}^n$, we have
\bee\label{Tlow:res}
&\inf_{\hat{\M \Sigma}_n}\sup_{\{\vecc(\M X_i)\}_{i = 1}^n\sim \mathbb{P},\atop\mathbb{P}\in \mathcal{P}^n_{\varepsilon_0,\alpha_1,\alpha_2}}\E\Bigg(\frac{\|\hat{\M \Sigma}_n - \M \Sigma_2^* \otimes \M \Sigma_1^*\|_\F^2}{pq}\Bigg)
\\
&\succsim \max\Bigg[\min\Big\{\frac{p}{nq},(nq)^{\frac{1}{2\alpha_1 + 2} - 1}\Big\}, \min\Big\{\frac{q}{np},(np)^{\frac{1}{2\alpha_2 + 2} - 1}\Big\}\Bigg].
\ee
\end{theorem}

\subsection{Additional theoretical results}\label{sec:compare}
We summarize other major theoretical findings as nine takeaway messages below. 
{\color{black} 
\begin{enumerate}[(1)]
\item By matching the derived lower and upper bounds, we are able to obtain sufficient conditions for which our obtained convergence rate is {\it minimax optimal} (similarly to \citet{Cai2010}, we focus on the case that $\M\Sigma_1,\M\Sigma_2$ are in the $\mathcal{M}(\varepsilon_0,\alpha)$ class). In particular, consider two regimes of $p,q$: (1) degenerate regime, where $p \wedge q = O(1)$; and (2) moderate high-dimensional regime, where $1\precsim p\precsim\max\big\{n\cdot (qn)^{-\frac{1}{2\alpha_1 + 2}},\sqrt{n}\big\} $  and $1\precsim q\precsim\max\big\{n\cdot (pn)^{-\frac{1}{2\alpha_2 + 2}},\sqrt{n}\big\}$. We can show that under Scenario (i), our proposed estimator is rate-optimal under both regimes; and under Scenario (ii), our estimator is rate-optimal under the degenerate regime.  These rate optimality findings are further verified by a simulation study in Section \ref{Sec28}. 
\item Under Scenario (i), the proposed estimator is always rate-optimal when $p,q \asymp \sqrt{n}$. In addition, when the bandable levels $\alpha_1,\alpha_2$ are large enough, the proposed estimator is also rate-optimal under the more challenging situation when both $p,q$ diverge at an asymptotic order close to $n$. For example, when $\alpha_1 = \alpha_2 = 2$, the proposed estimator is rate-optimal when $p = q = n^{0.7}$. See Section \ref{sec:effect:region} for more details. 
\item Our results reveal an interesting {\it phase transition} phenomenon in the sense that the optimal rate can be achieved without banding when $p$ and $q$ do not diverge fast enough in $n$. Take optimal $k_1$ as an example. If $p$ diverges not sufficiently fast compared to $q$, then no banding is needed for $k_1$. For example, consider $\alpha_1 =\alpha_2 = 2$. Under Scenario (i) with the proposed banded estimator, we let $p,q$ satisfy the moderate high-dimensional regime. If $p \precsim (nq)^{1/6}$, then the optimal $k_1 = p - 1$, which means no banding on $p$ direction is needed. If  $p \succ (nq)^{1/6}$, then the optimal $k_1 \asymp (nq)^{1/6}\ll p$. Similar symmetric results can be obtained for the optimal $k_2$. 
\item To provide more insights, we now focus on the example under Scenario (i) with $\alpha_1 = \alpha_2 = 2$, $\mathcal{M}$--class $\M\Sigma_1^*$, $\M\Sigma_2^*$, and give some graph illustration. We let $p = n^{\beta_1}$ and $q= n^{\beta_2}$, where different $\beta_1,\beta_2 > 0$ represent different divergence regimes of $p,q$. Results in Section \ref{optbd:T2} show that the upper bounds of Theorems \ref{T2} after selecting the optimal $k_1, k_2$, as well as the lower bound of Theorem \ref{T:low}, always have the polynomial form: 
$
n^{-r}.
$ 
As $r = -\log_n\big( n^{-r}\big)$, we call $r$ the negative log convergence rate (NLCR) and use it to measure the convergence rate of both upper and lower bounds. Similarly,  the optimal divergence regimes of $k_1, k_2$ always have the form: $n^{r'}$. We call $r'$ the log divergence rate (LDR) and use it to measure the divergence rate of optimal $k_1$.
\par
\begin{figure}[H]
  \centering
    \subfigure[NLCR: Lower Bound]{\includegraphics[width=0.32\textwidth]{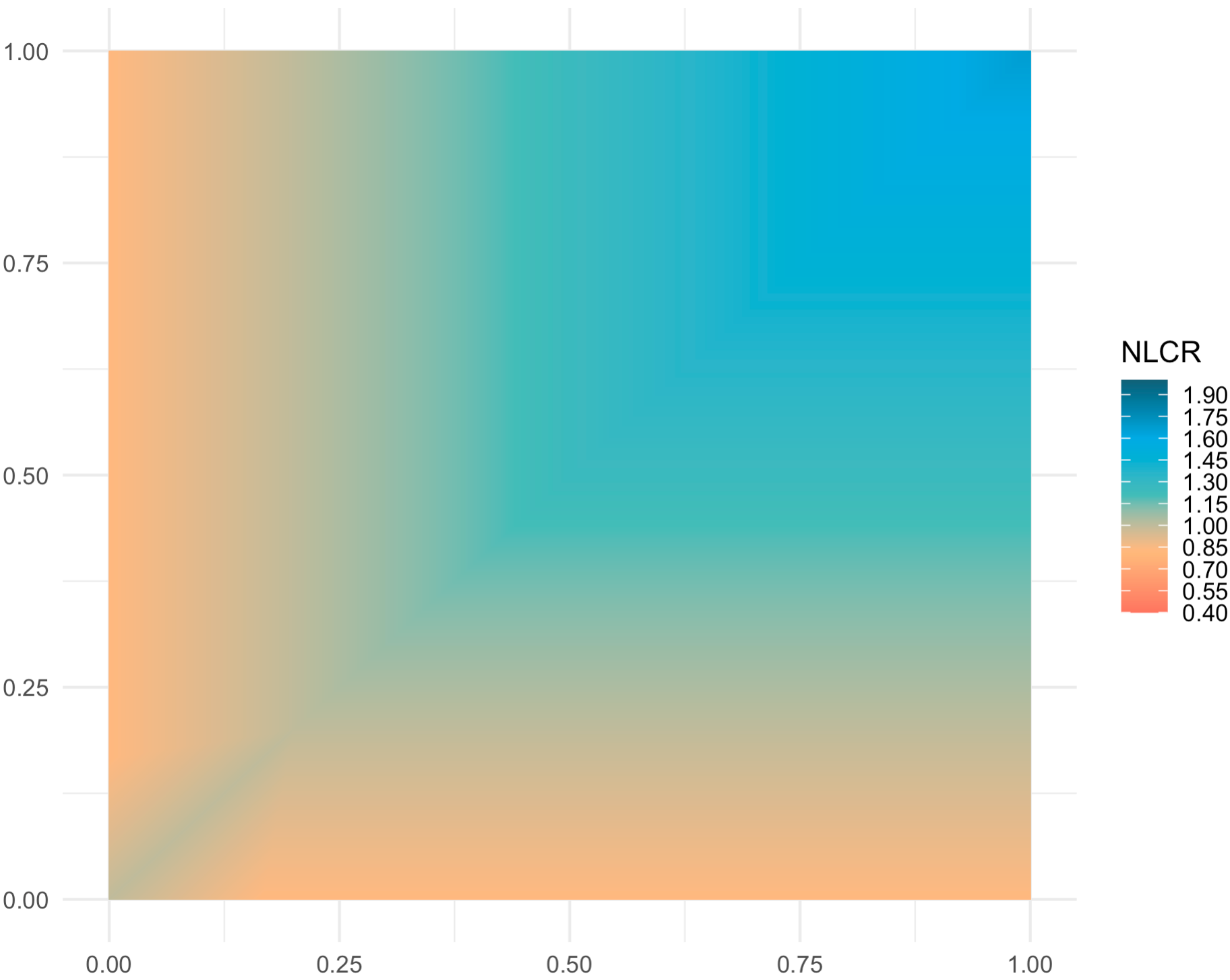}}
  \subfigure[NLCR: Upper Bound]{\includegraphics[width=0.32\textwidth]{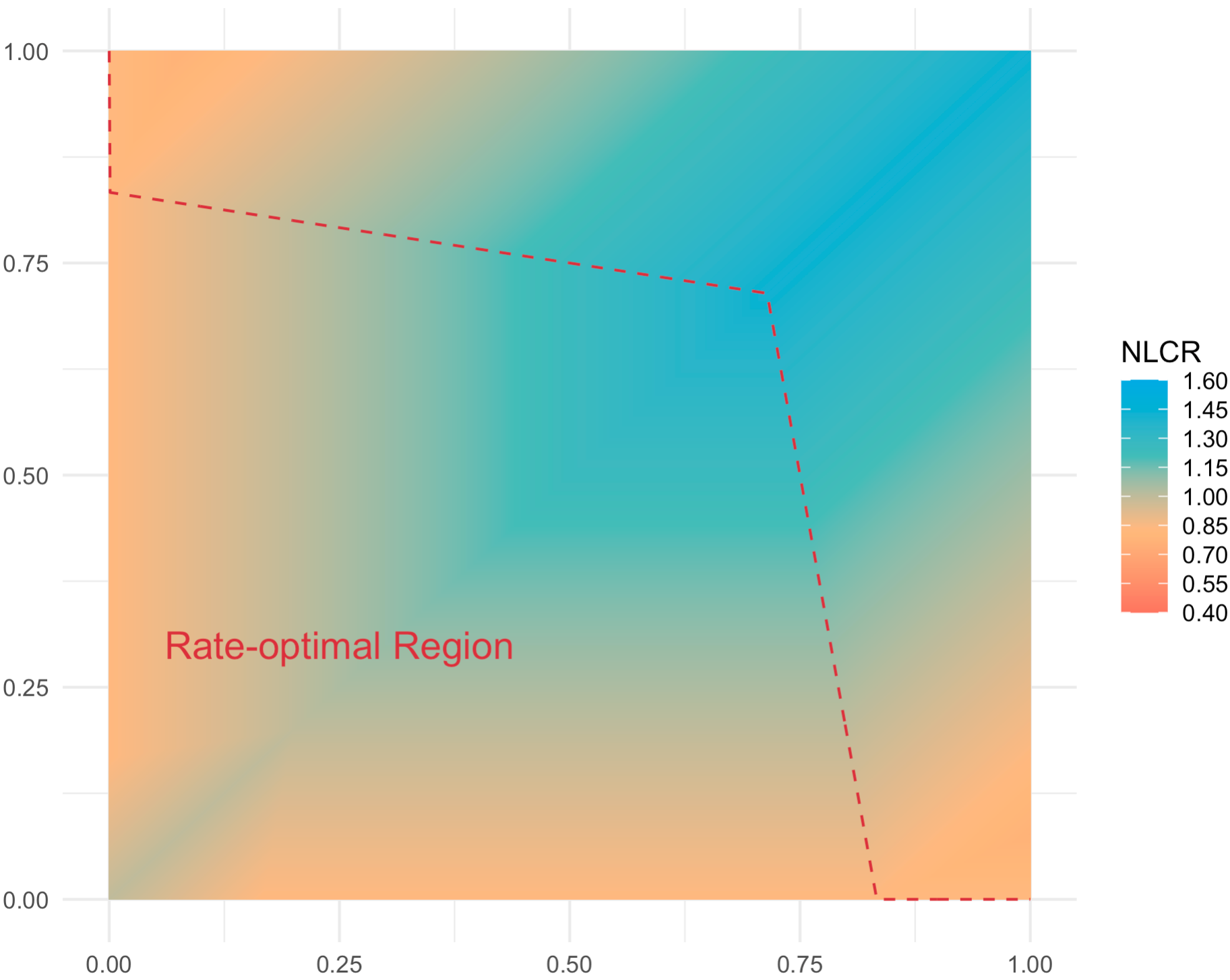}}
    \subfigure[LDR: Optimal $k_1$]{\includegraphics[width=0.32\textwidth]{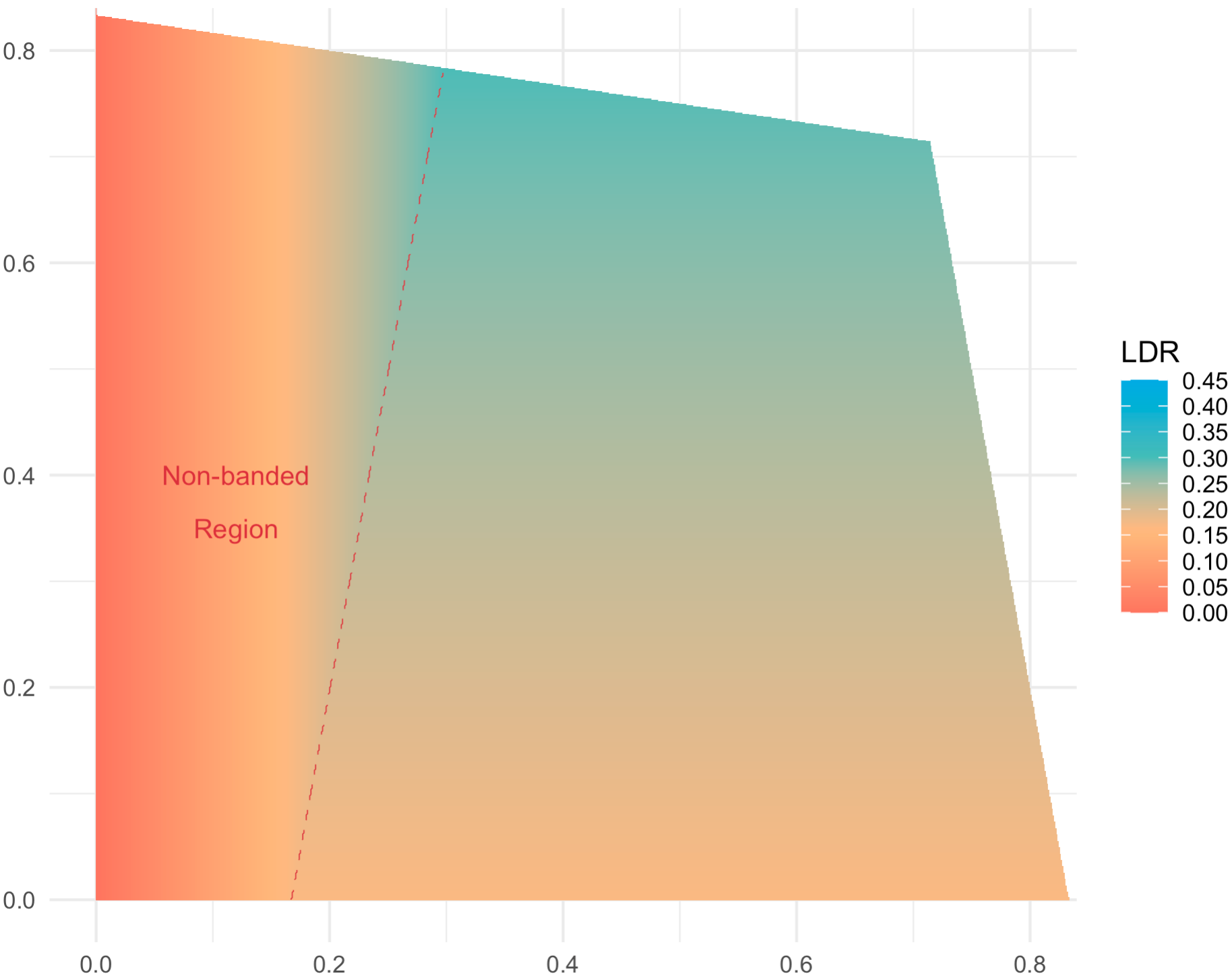}}
 \caption{The x-axis represents $\log p/\log n = \beta_1$ and the y-axis represents $\log q/\log n = \beta_2$. In panels (a) and (b), the color represents NLCR of the corresponding upper and lower bounds, where blue means a faster convergence rate. Panel (a) gives the lower bound by Theorem \ref{T:low}; panel (b) gives the error upper bound under Scenario (i) (sub-Gaussian scenario), and the regions in the bottom left corner, surrounded by red dashed lines, corresponds to the rate-optimal region. In panel (c),  colored regions correspond to the rate-optimal region in panels (b), and the color represents the LDR of the optimal $k_1$ value, where a deeper color indicates a larger divergence rate. The red dashed line sets the boundary for the non-bandable region (to its left), where no banding is necessary to achieve the optimal convergence rate.}
\label{fig:sgbd1:main}
\end{figure}
In Figure \ref{fig:sgbd1:main}~(a) and (b), we illustrate the NLCRs of the lower and upper bounds of the proposed estimator with different divergence regimes of $p,q$, under Scenario (i). The lower and upper bounds are matched in the rate-optimal region in Figure \ref{fig:sgbd1:main}~(b), in which the regime $p = q = n^{0.7}$ is included. When $p,q$ become larger, e.g., $p = q = n^{0.8}$, the proposed estimator is no longer rate-optimal by the current theoretical results. 
\par
In Figure \ref{fig:sgbd1:main}~(c), the LDR of optimal $k_1$ is in the no-banding region when $p$ is slowly divergent. When $p$ diverges faster than the phase-transition rate $(nq)^{1/6}$, the LDR of optimal $k_1$ is out of the no-banding region, and thus the banding on the $p$ direction is needed. The no-banding region in Figure \ref{fig:sgbd1:main}~(c) becomes wider when $\beta_2$ becomes larger. This is because the phase-transition rate $(nq)^{1/6}$ diverges faster when $q$ diverges faster.
\par
We leave more graphical illustrations in Section
\ref{sec:match:illustration}.
\item  In general, parameter complexity (such as matrix size and bandable level) has a complex effect on the error rate, the optimal bandwidth selection, and the rate-optimal regime. For example, under Scenario (i), with larger bandable levels $\alpha_1,\alpha_2$, a higher dimensional regime of $p,q$ can be shown to be rate-optimal by Theorem \ref{T2}.  A detailed discussion is given in Section \ref{sec:complexeffect}.
\item There are two technical reasons that make our estimator sub-optimal when $p$ and $q$ are both very large (e.g., $\gg \sqrt{n}$). One is that the spectral-norm bound in Lemma \ref{lemma:srate} is loosened from $\frac{k_1}{qn} + \frac{k_2}{pn}$ to $\frac{pk^2_1}{qn^2} + \frac{qk^2_2}{pn^2}$, when $pk_1 + qk_2 \succ n$, i.e., if the spectral-norm bound in Lemma \ref{lemma:srate} can be kept as the tight bound $\frac{k_1}{qn} + \frac{k_2}{pn}$, then the resulting optimal upper bounds by Theorem \ref{T3} will always match the lower bound in Theorem \ref{T:low}, regardless of $p,q$'s divergent regime. Another reason is that the lower bound in Theorem \ref{T:low} may not be tight. Improving these bounds is an interesting future work direction. 
\item For the degenerate regime (one of $p,q$ is $O(1)$), the rate optimality of {\color{black}both upper and lower bounds}, essentially agrees with those of Bickel's linear banded estimator \citep{bickel2008regularized, xiao2014theoretic} and Cai's linear tapering estimator in \citet{Cai2010}. See Section \ref{sec:dr} for more details.
\item We have derived the error bounds for individual matrix component estimation in Section \ref{sec:indcov}. The results show that our individual matrix estimator, multiplied by a proper constant, which is needed for the identifiability purpose, converges to the truth at the obtained rates under the Frobenius norm.   
\item We extend the rate results under the Frobenius norm to the spectral norm in Section \ref{sec:disc:spnorm}. It turns out to be very challenging to obtain spectral norm convergence results for our current estimator. Therefore we consider a rank one Kronecker product approximation based on minimizing the spectral norm rather than the Frobenius norm in our original estimation procedure. We then derive an upper bound for the convergence rate of the new estimator under the spectral norm. 
\end{enumerate}
 
}

\vspace{-.15in}
\section{Robust Covariance Estimation}\label{sec:RE}
Heavy-tailed data are commonly encountered in many applications. 
When modeling heavy-tailed data, Theorems \ref{T2} and \ref{T3} may not be suitable since they require sub-Gaussian tail and finite fourth moment conditions. In this section, we propose {\em robust} banded/tapering covariance estimators that improve the proposed non-robust estimators and enjoy desired theoretical properties for heavy-tailed data. 

\par
We first give a precise definition of heavy tail condition for the sample matrix $\M X_i, 1\leq i\leq n$. Following state-of-art robust covariance estimation literature \citep{avella2018robust, lu2021robust}, we quantify heavy-tailedness of $\M X_i$ by elementwise $2\zeta$-finite moment order, such that for any $1\leq l_1,m_1\leq p$ and $1\leq l_2,m_2 \leq q$, 
\bee\nonumber
\E \big(|x^{(i)}_{l_1,l_2}\cdot x^{(i)}_{m_1,m_2}|^\zeta\big)\leq M<+\infty.
\ee
We will show that the robust banded/tapering covariance estimators enjoy the following two nice properties. First, our proposed robust estimators converge to the truth under the Frobenius norm, even when data only have a finite moment order of $\zeta >1$. 
Compared to the proposed non-robust estimators, Theorem \ref{T3} needs at least $\zeta \geq 2$, i.e. a finite fourth or higher moment. Additionally, even when $\zeta \geq 2$, the convergence rate in Theorem \ref{T3} is always slower or equal to the convergence rate in Theorem \ref{T4}. Secondly, we show our proposed robust estimators have an adaptive convergence rate with different levels of heavy-tailedness. In particular, when the data $\M X_i$ has a higher moment condition $\zeta$, the convergence rate becomes faster, and closer to the rate in Theorem \ref{T2} under the sub-Gaussian condition, while the proposed non-robust estimators do not have such properties. 

We adopt the idea in \citet{fan2016shrinkage} and truncate the sample covariance as a preliminary step. When $\vecc(\mathbf{X}_i) \in \mathbb{R}^{pq}$, the modified estimator is defined as
\bee\nonumber
\widecheck{\M \Sigma} = \frac{1}{n}\sum_{i =1}^{n} \big\{\vecc(\widecheck{\M X}_i) - \vecc(\widecheck{{\boldsymbol\mu}})\big\}\cdot \big\{\vecc(\widecheck{\M X}_i)- \vecc(\widecheck{{\boldsymbol\mu}})\big\}^{\T},
\ee
where $\widecheck{\mathbf{X}}_i$ satisfies 
\bee\label{def:hatx}
\widecheck{x}^{(i)}_{l_1,l_2} = \text{sgn}(x^{(i)}_{l_1,l_2})(|x^{(i)}_{l_1,l_2}|\wedge \tau)
\ee 
with some selected $\tau >0$ for all $1\leq l_1\leq p, 1\leq l_2\leq q$, and $\widecheck{{\boldsymbol\mu}} = \frac{1}{n}\sum_{i = 1}^n\widecheck{\M X}_i$. {\color{black}Same as \citet{fan2016shrinkage}, we assume $\E[\M X_i] = \mathbf{0}$ for the ease of exposition.} For general covariance estimation problem, 
the results in \citet{fan2016shrinkage} suggest that if elementwisely $\E(|x^{(i)}_{l_1,l_2}\cdot x^{(i)}_{m_1,m_2}|^2) \leq M < +\infty$, where $M$ is a constant that does not depend on $n, l_1,m_1, l_2,m_2$, then $\|\widecheck{\mathbf{\Sigma}} -\mathbf{\Sigma}^*\|_{\max} = \mathcal{O}_{\M P}(\sqrt{\log \max(p,q) / n })$ by taking $\tau \asymp \big[\log\{\max(p,q)\}/n\big]^{-1/4}$, while sample covariance estimator needs the sub-Gaussian assumption to achieve  max-norm convergence.
\par
Similarly, instead of using the sample covariance $ \hat{\M \Sigma}$, we replace  $ \hat{\M \Sigma}$ with $\widecheck{\M \Sigma}$ when implementing our banded and/or tapering estimation procedure. 
In particular, define 
\bee\nonumber
&\widecheck{\mathbf{\Sigma}}_{\MB}(k_1,k_2) = \widecheck{\M \Sigma}\HDBand;
\\
&\widecheck{\mathbf{\Sigma}}_{\mathcal{T}}(k_1,k_2) = \widecheck{\M \Sigma}\HDTaper,
\ee
as the banded/tapering matrix of $\widecheck{\mathbf{\Sigma}}$ with parameters $k_1$ and $k_2$. We propose the robust banded/tapering estimators as $\hat{\mathbf{\Sigma}}^{\mathcal{R},\eta}(k_1,k_2)=\hat{\mathbf{\Sigma}}^{\mathcal{R},\eta}_2\kii\otimes \hat{\mathbf{\Sigma}}^{\mathcal{R},\eta}_1\ki$, where $\eta\in \{\MB,\mathcal{T}\}$ and $\hat{\mathbf{\Sigma}}^{\mathcal{R}, \eta}_1\ki, \hat{\mathbf{\Sigma}}^{\mathcal{R}, \eta}_2\kii$ are solutions from
\bee\label{rob_opt}
\big(\hat{\mathbf{\Sigma}}^{\mathcal{R},\eta}_1\ki, \hat{\mathbf{\Sigma}}^{\mathcal{R},\eta}_2\kii\big)=\argmin_{\M \Sigma_1, \M \Sigma_2}\| \widecheck{\boldsymbol\Sigma}_{\eta}(k_1, k_2)-{\M \Sigma}_2 \otimes {\M \Sigma}_1\|_\rF^2.
\ee

In practice, we adopt the resampling scheme in Section \ref{bandable} to select $\tau$ from a candidate pool $\mathbb{T} = \{|x|_{\iota}:\iota \in \mathscr{P} \}$, where $|x|_{\iota}$ is the $\iota$ percentile value among the set of all possible absolute values of coordinates in any $\M X_i$, i.e., $\{|x^{(i)}_{l_1,l_2}|: 1\leq i\leq n, 1 \leq l_1 \leq p, 1\leq l_2 \leq q\}$ and $\mathscr{P}$ is a candidate pool of percentiles under consideration. To solve \eqref{rob_opt} and select the tuning parameters ($k_1$, $k_2$), we adopt the same procedures as the ones used in non-robust banded/tapering covariance estimation in Section \ref{bandable}.

Before presenting our asymptotic result, we list two regularity assumptions for the truncated data $\widecheck{\M X}_1,\dots,\widecheck{\M X}_n$. Denote $\M \Sigma_{\mathcal{R}}^* \equiv \cov(\widecheck{\M X}_i)$ and its doubly banded or tapering version ${\M \Sigma}^{*,\eta}_{\mathcal{R}}(k_1,k_2)$ with $\eta\in\{\mathcal{B},\mathcal{T}\}$. 
\begin{assumption}\label{A:minev}
We assume eigenvalues of ${\M \Sigma}^*_{\mathcal{R}}$ satisfying $ \varepsilon_0' <\lambda_{\min}({\M \Sigma}^*_{\mathcal{R}})  \leq \lambda_{\max}({\M \Sigma}^*_{\mathcal{R}}) < 1/\varepsilon_0'$ for some constant $0 < \varepsilon_0' < 1$.
\end{assumption}
\begin{assumption}\label{A:psi2}
There exists $J_n $ such that $\big\|\vecc(\widecheck{\M X}_i)\big\|_{\psi_2} = \sup_{\|\mathbf{v}\| = 1}\|\mathbf{v}^\T\vecc(\widecheck{\M X}_i)\|_{\psi_2} \leq J_n\tau$, where $1\precsim J_n\precsim \sqrt{pq}$ and $J_n$ \textit{only} depends on $ n $.
\end{assumption}
Here $\|\cdot\|_{\psi_2}$ is the sub-Gaussian norm  defined in Section \ref{def:sesgrv}. Equation \eqref{def:hatx} implies that each element of $\vecc(\widecheck{\M X}_i)$ can be upper bounded by $C_J\tau$ under sub-Gaussian norm, with some constant $C_J>0$. Assumption \ref{A:psi2}  further introduces a general thresholding factor $J_n$ such that upper bound $J_n\tau$ holds uniformly for any $\|\mathbf{v}^\T\vecc(\widecheck{\M X}_i)\|_{\psi_2}$ when $\|\mathbf{v}\| = 1$, where $J_n\succsim 1$. A trivial bound of $J_n$ is $J_n\precsim\sqrt{pq}$ by triangle inequaltiy. In many conditions, based on the probabilistic structure of $\M X_i$, $J_n$ can be further reduced to a constant. We defer more discussions to Section \ref{sec:Jn} in the Supplementary File.
\par
Our main result is given in the following theorem; and the proof is included in the Supplementary File. 
\begin{theorem}\label{T4}
Let $\vecc(\M X_1), \vecc(\M X_2),\dots,\vecc(\M X_n)$ be i.i.d. random vectors in $\RR^{pq}$ with true covariance $\M \Sigma^* = \M \Sigma^*_2 \otimes \M \Sigma^*_1$. Assume $\E(|x^{(i)}_{l_1,l_2}\cdot x^{(i)}_{m_1,m_2}|^\zeta) \leq M < +\infty$ where $\zeta> 1$ is the  order of heavy-tailedness, and $M$ is a constant that does not depend on $n, l_1,m_1, l_2,m_2$.  Define $\small\M I_{\eta,d}(k)$ and $\small\tilde{\alpha}_{a} $ the same as in Theorem \ref{T2}, and define the error terms, $r^{\mathcal{R},\zeta}_1(k_1,k_2, p,q,n) = (k_1k_2)^{1/\zeta}\cdot\big(\frac{k_1J_n^4}{qn} + \frac{k_2J_n^4}{pn}\big)^{1 - 1/\zeta}$, and $r^{\mathcal{R},\zeta}_2(k_1,k_2, p,q,n) = (k_1k_2)^{1/\zeta}\cdot\big(\frac{pk^2_1J_n^4}{qn^2} + \frac{qk^2_2J_n^4}{pn^2}  \big)^{1 - 1/\zeta}$.
\par
Consider the following four cases:
$\M{C1.}$  $pk_1 + qk_2 \precsim n$ and $ 1< \zeta < 2$; $\M{C2.}$ $pk_1 + qk_2 \precsim n$ and $\zeta \geq 2$; $\M{C3.}$  $pk_1 + qk_2 \succ n$ and $ 1< \zeta < 2$; $\M{C4.}$ $pk_1 + qk_2 \succ n$ and $\zeta \geq 2$. Under Assumptions \ref{A:minev}--\ref{A:psi2}, by choosing the optimal $\tau$ given in Section  \ref{sec:opt:threshold}, 
we have
\bee\label{T4:res}
&\E\Big(\frac{\|\hat{\M \Sigma}_2^{\mathcal{R},\eta}\kii\otimes\hat{\M \Sigma}_1^{\mathcal{R},\eta}\ki - \M\Sigma^*\|_{\F}^2}{pq}\Big)
\\ 
&\small\precsim \begin{cases}
r^{\mathcal{R},\zeta}_1(k_1,k_2, p,q,n)+ \M I_{\eta,p}(k_1)\cdot k_1^{-\tilde{\alpha}_1} +  \M I_{\eta,q}(k_2)\cdot k_2^{-\tilde{\alpha}_2}, & \M{C1},
\\
\frac{k_1k_2}{n}\wedge r^{\mathcal{R},\zeta}_1(k_1,k_2, p,q,n)+ \M I_{\eta,p}(k_1)\cdot k_1^{-\tilde{\alpha}_1} +  \M I_{\eta,q}(k_2)\cdot k_2^{-\tilde{\alpha}_2}, & \M{C2},

\\
r^{\mathcal{R},\zeta}_2(k_1,k_2, p,q,n) + \M I_{\eta,p}(k_1)\cdot k_1^{-\tilde\alpha_1} +  \M I_{\eta,q}(k_2)\cdot k_2^{-\tilde\alpha_2}, & \M{C3},
\\
\frac{k_1k_2}{n}\wedge r^{\mathcal{R},\zeta}_2(k_1,k_2, p,q,n)+ \M I_{\eta,p}(k_1)\cdot k_1^{-\tilde\alpha_1} +  \M I_{\eta,q}(k_2)\cdot k_2^{-\tilde\alpha_2}.& \M{C4}.
\end{cases}
\ee
\end{theorem}
More discussion on the error rate and the choice of optimal $\tau$ is given in Section \ref{sec:supps3}.  
\vspace{-.15in}

\section{Simulation}\label{sim}
\subsection{Banded and Tapering Estimator}\label{sec:simu_pe}
We investigate the finite sample performance of the proposed estimator by simulations. We first consider the case where the data are generated from multivariate normal distributions. In particular, the $\{\vecc (\Xb_i)\}_{i = 1}^n $ are i.i.d generated from $\M N({\bf 0}, \M\Sigma) $, where $\M\Sigma = \boldsymbol\Sigma_2\otimes\boldsymbol\Sigma_1$ with $ \boldsymbol\Sigma_1\in \RR^{p\times p} $ and $ \boldsymbol\Sigma_2\in \RR^{q\times q} $. We consider following two covariance structures for $ \boldsymbol\Sigma_1$ and $ \boldsymbol\Sigma_2$. 

\noindent \textbf{Case 1.} {\it Moving average covariance structure}\\
We set $ \boldsymbol\Sigma_1 $ and $ \boldsymbol\Sigma_2 $ to be the covariances of an MA(1) process with
\begin{eqnarray*}
&&\sigma_{l_1, m_1}^{(1)}=\rho_1^{|l_1-m_1|}\cdot  \textbf{I}\{|l_1-m_1|\leq 1\}, {\rm ~~~} 1\leq l_1, m_1\leq p, \\
&&\sigma_{l_2, m_2}^{(2)}=\rho_2^{|l_2-m_2|}\cdot  \textbf{I}\{|l_2-m_2|\leq 1\}, {\rm ~~~} 1\leq l_2, m_2\leq q,
\end{eqnarray*}
where $ \rho_1=\rho_2=0.5 $.

\noindent \textbf{Case 2.} {\it Autoregressive covariance structure}\\
We take $ \boldsymbol\Sigma_1 $ and $ \boldsymbol\Sigma_2 $ to be the covariances of an AR(1) process with
\begin{eqnarray*}
&&\sigma_{l_1, m_1}^{(1)}=\rho_1^{|l_1-m_1|}, {\rm ~~~} 1\leq l_1, m_1\leq p,\\
&&\sigma_{l_2, m_2}^{(2)}=\rho_2^{|l_2-m_2|}, {\rm ~~~} 1\leq l_2, m_2\leq q,
\end{eqnarray*}
where we set $ (\rho_1, \rho_2)=(0.1, 0.1), (0.5, 0.5), (0.8, 0.8) $. 

We consider $ n=50, 100, 200 $, and $ (p,q)=(20,30), (100, 100) $. We compare our proposed estimators {\color{black}with the sample covariance estimator; the banded estimator \citep{bickel2008regularized}; the tapering estimator \citep{Cai2010}; the doubly banded and tapering estimators $\tilde{\M\Sigma}_{\mathcal{B}}(k_1,k_2)$ and $\tilde{\M\Sigma}_{\mathcal{T}}(k_1,k_2)$ defined in Sections \ref{bandable} and \ref{tapering}, respectively. 
We use the resampling scheme to choose the bandwidths for \citet{bickel2008regularized}'s banded estimator, \citet{Cai2010}'s tapering estimator and the corresponding doubly banded and tapering estimators. The random split procedure is repeated for $N = 10$ times. }

We report several quantities, $\|\hat{\boldsymbol\Sigma}-\boldsymbol\Sigma\|_{\F}$, $ \|\hat{\boldsymbol\Sigma}-\boldsymbol\Sigma\|_1 $ and $ \|\hat{\boldsymbol\Sigma}-\boldsymbol\Sigma\|_{2} $,  where $ \hat{\boldsymbol\Sigma} $ can be our proposed estimators and {\color{black} the other five comparison estimators.} These quantities characterize the estimation errors for the covariance matrices. We also report the $ \hat{k}_1 $ and $ \hat{k}_2 $ for our proposed methods and the doubly banded/tapering estimators; and $ \hat{k} $ for \citet{bickel2008regularized}'s banded estimator and \citet{Cai2010}'s tapering estimator. We summarize the averages of these quantities over 100 Monte Carlo repetitions in Table \ref{Table:simulation3} and Tables \ref{Table:simulation1}--\ref{addsims4} in Section \ref{sec:num:simu:1} of the  Supplementary File. Their associated standard errors are summarized in Section \ref{sec:simu:se} of the Supplementary File. 
\begin{table}[htp]
\thisfloatpagestyle{plain}
\setlength{\tabcolsep}{-1pt}
  \centering
    \caption{Simulation results for $(p,q,\rho_1, \rho_2)=(100,100, 0.5, 0.5)$ with the MA(1) covariance structure over 100  replications.
    The averages of $\|\hat{\boldsymbol\Sigma}-\boldsymbol\Sigma\|_{\F}$, $ \|\hat{\boldsymbol\Sigma}-\boldsymbol\Sigma\|_1 $ and $ \|\hat{\boldsymbol\Sigma}-\boldsymbol\Sigma\|_{2} $ for 
the proposed estimators (Proposed B and Proposed T), doubly banded and tapering estimators (Doubly B and Doubly T), Bickel's banded estimator (Banded), Cai's tapering estimator (Tapering) and the sample covariance estimator (Sample) are reported. The averages of $ \hat{k}_1 $ and $ \hat{k}_2 $ for the proposed and doubly banded/tapering estimators, the averages of  $ \hat{k} $ for Bickel's banded estimator and Cai's tapering estimator are also reported. }
\vspace{0.5pt}
\newsavebox{\tableboxb}
\begin{lrbox}{\tableboxb}
\begin{tabular}{p{0.24\textwidth}>{\centering}p{0.23\textwidth}p{0.15\textwidth}p{0.15\textwidth}p{0.15\textwidth}p{0.08\textwidth}p{0.08\textwidth}p{0.08\textwidth}}
  \hline
  
$ (n, p, q, \rho_1,\rho_2) $ & Method & $\|\hat{\boldsymbol\Sigma}-\boldsymbol\Sigma\|_{\F}$ & $ \|\hat{\boldsymbol\Sigma}-\boldsymbol\Sigma\|_1 $ & $ \|\hat{\boldsymbol\Sigma}-\boldsymbol\Sigma\|_{2} $ & $\hat{k}$ \quad & $\hat{k}_1$ \quad & $\hat{k}_2$ \quad \\ \hline
$(50,100,100,0.5,0.5)$ & Sample &1428.40&1605.38& 244.51\\& Banded &91.39& 4.08& 2.59& 1.09&&\\& Tapering &93.40& 4.49& 2.80& 2.00&&\\& Double B &47.82&4.32&2.58&&1.01&1.01\\& Double T &77.87&3.96&3.46&&1.78&1.82\\& Proposed B &8.41&0.79&0.49&&1.73&1.87\\& Proposed T &8.71&0.81&0.51&&2.00&2.00\\
$(100,100,100,0.5,0.5)$  & Sample &1004.87&1031.87& 130.38\\& Banded &88.68& 3.37& 2.29& 1.05&&\\& Tapering &89.76& 3.63& 2.42& 2.00&&\\& Double B &33.58&2.89&1.70&&1.01&1.00\\& Double T &51.88&4.85&2.37&&2.00&2.00\\& Proposed B &5.85&0.53&0.33&&1.65&1.84\\& Proposed T &6.13&0.56&0.35&&2.00&2.00\\
$(200,100,100,0.5,0.5)$& Sample &708.89&683.42& 71.37\\& Banded &87.37& 2.95& 2.15& 1.07&&\\& Tapering &87.93& 3.14& 2.22& 2.08&&\\& Double B &23.69&2.00&1.17&&1.01&1.00\\& Double T &36.58&3.32&1.63&&2.00&2.00\\& Proposed B &4.16&0.38&0.23&&1.79&1.82\\& Proposed T &4.32&0.39&0.25&&2.00&2.00\\
  \hline
\end{tabular}
\end{lrbox}
\label{Table:simulation3}
\scalebox{0.75}{\usebox{\tableboxb}}
\end{table}

From these tables, we can see that our proposed methods always perform better than all  comparison methods in terms of estimation errors. {\color{black}In addition, the doubly banded and tapering estimators perform generally better than the other three comparison estimators (sample/banded/tapering). 
} For Case 1, the oracle $ k_1$ and $k_2$ for the banded estimator are both $ 1$. Noticing that $ B_1({\bf 1}_p)=T_2({\bf 1}_p)$, the oracle $ k_1$ and $k_2$ for the tapering estimator are both $ 2$, so our method can select $ \hat{k}_1$ and $ \hat{k}_2$ accurately. For Case 2, when $ \rho_1 $ and $\rho_2$ increase, the selected bandwidths $ \hat{k}_1 $ and $ \hat {k}_2 $ for the proposed method also increase. 
\subsection{Robust Covariance Estimation for Heavy-tailed Data}
In this subsection, we investigate the finite sample performance of our proposed robust estimators when data follow heavy-tailed distributions. In particular, we consider i.i.d. $p\times q$ matrix-valued data $\M X_i$ with $\text{vec} ({\M X_i})$ following multivariate $t$ distributions with degrees of freedom $3$, i.e., $t_3(0,\M \Sigma)$ with $\M \Sigma = \M \Sigma_2 \otimes \M\Sigma_1$. Similar to Section \ref{sec:simu_pe}, we consider two covariance structures. Case 1: $\M \Sigma_1$, $\M \Sigma_2$ are the covariances of MA(1) process, where $(p,q,\rho_1,\rho_2) = (20,30,0.5,0.5)$ and $n = 50,100,200$. Case 2: $\M \Sigma_1$, $\M \Sigma_2$ are the covariances of AR(1) process, where $(\rho_1,\rho_2) = (0.1,0.1),(0.5,0.5),(0.8,0.8)$ and $(n,p,q) = (50,20,30)$.

\par
We compare the proposed robust banded and tapering estimators with the non-robust version of banded and tapering  estimators introduced in Section \ref{sec:2} as well as the sample covariance estimator. The results for both cases are summarized in Table \ref{tb:s5simu} of the Supplementary File and Table \ref{tb:rb_2} based on 100 Monte Carlo (MC) replications, respectively. We use the random splitting procedure introduced in Section \ref{bandable} and Section \ref{sec:RE} to select $\hat{k}_1,\hat{k}_2$ and choose the truncation threshold $\hat{\tau}$ based on  $$\mathscr{P} = \{100,99.995,99.99,99.97,99.95,99.93,99.91,99, 98,95,92,90,87,85,80\}.$$
\par

From the results, we can see the proposed robust estimators always outperform all the other methods. For Case 1 (Table \ref{tb:s5simu}), when $n$ becomes larger, the selected $\hat{\tau}$ for the proposed robust estimators decreases. For Case 2 (Table \ref{tb:rb_2}), we can see a clear improvement of estimation accuracy by adopting the robust covariance estimation as $\rho_1,\rho_2$ increase. Meanwhile, the selected bandwidths $\hat{k}_1$, $\hat{k}_2$ increase and the selected $\hat{\tau}$ decreases as expected.
\begin{table}[htp]
\thisfloatpagestyle{plain}
\setlength{\tabcolsep}{-1pt}
  \centering
    \caption{Simulation results for heavy-tailed data with $(\rho_1,\rho_2) = (0.1,0.1), (0.5,0.5), (0.8,0.8)$,  $(n,p,q)=(50,20,30)$, and AR(1) covariance structure over 100 replications. The averages of $\|\hat{\boldsymbol\Sigma}-\boldsymbol\Sigma\|_{\F}$, $ \|\hat{\boldsymbol\Sigma}-\boldsymbol\Sigma\|_1 $ and $ \|\hat{\boldsymbol\Sigma}-\boldsymbol\Sigma\|_{2} $ for robust estimators (Robust B and Robust T), 
our proposed estimators (Proposed B and Proposed T) and the naive sample covariance estimator (Sample) are summarized in this table. The averages of $ \hat{k}_1 $ and $ \hat{k}_2 $ for the proposed robust/non-robust methods, the averages of  $ \hat{\tau} $ for the proposed robust methods are also reported.}
\vspace{0.5pt}
\begin{lrbox}{\tableboxb}\label{tb:rb_2}\Rotatebox{0}{
\begin{tabular}{p{0.24\textwidth}>{\centering}p{0.23\textwidth}p{0.15\textwidth}p{0.15\textwidth}p{0.15\textwidth}p{0.08\textwidth}p{0.08\textwidth}p{0.08\textwidth}}
  \hline
$(n,p,q,\rho_1,\rho_2)$ & Method & $\| \hat{\Sigma} - \Sigma\|_{\F}$ & $\|\hat{\Sigma} - \Sigma\|_1$ & $\| \hat{\Sigma} - \Sigma \|_{2}$ & $\hat{k}_1$ & $\hat{k}_2$ & $\hat{\tau}$
\\
\hline
$(50,20,30,0.1,0.1)$& Sample &245.38&590.03&209.47\\& Proposed B &12.64&3.02&1.97&1.24&1.13\\& Proposed T &11.79&1.66&1.32&1.10&0.86\\& Robust B&9.05&1.75&1.05&1.45&1.34&7.16\\& Robust T&8.63&1.17&0.85&1.20&1.00&6.87\\
$(50,20,30,0.5,0.5)$& Sample &245.70&574.99&209.81\\& Proposed B &24.15&12.72&7.10&2.05&2.04\\& Proposed T &22.08&8.67&5.26&1.80&1.72\\& Robust B&16.05&7.71&4.38&2.37&2.43&5.71\\& Robust T&16.06&6.91&4.33&2.00&1.98&5.65\\
$(50,20,30,0.8,0.8)$& Sample &249.51&511.77&212.37\\& Proposed B &92.58&141.98&63.37&7.63&7.33\\& Proposed T &83.38&117.68&56.57&7.81&7.71\\& Robust B&48.68&49.61&28.75&8.69&8.60&2.85\\& Robust T&48.56&49.26&30.45&9.55&9.28&3.34\\
  \hline
\end{tabular}}
\end{lrbox}
\label{Table:simulation6}
\scalebox{0.75}{\usebox{\tableboxb}}
\end{table}

\vspace{-.15in}
\section{Gridded Temperature Anomaly Data Analysis}\label{realdata}
We analyze a gridded temperature anomalies dataset collected by the U.S. National Oceanic and Atmospheric Administration (NOAA) \citep{shen2017r,gu2020generalized} in this section. Another case study of a stock price dataset is presented in Section \ref{data:stock} of the Supplementary File. The temperature anomalies dataset contains the monthly air and marine temperature measurements from Jan 1880 to 2017 with a $5^{\circ}\times 5^{\circ}$ latitude-longitude resolution. It can be downloaded at \href{ftp://ftp.ncdc.noaa.gov/pub/data/noaaglobaltemp/operational}{ftp://ftp.ncdc.noaa.gov/pub/data/noaaglobaltemp/operational}. 

In our study, we focus on the temperature anomalies (the difference between an observed temperature and the baseline/normal value) in the past 20 years over the region marked in deep blue as shown in Figure \ref{climate_matrix} (Supplementary File) to avoid the missing values and to make sure the resulting data are in a matrix form (with two dimensions representing longitude and latitude). We have pre-processed the data to remove the mean trend and the dependence over the time. This is implemented by (i) first fitting a separate linear model for each spatial coordinate over the time and then removing the estimated time trend; and (ii) ``thinning'' the sequence of monthly measurements by taking a monthly record from a window of every 5 months. In Figure \ref{climate_trend_acf}, we use $5^\degree\times 5^\degree$ box centered at $57.5^\degree W$ longitude and $7.5^\degree S$ latitude as an example to show the effect of pre-processing. In (a) and (b), we show the data before and after the detrending; and in (c), we plot the estimated auto-correlation function for the thinned sequence. It can be seen that both detrending and thinning work quite well for that region. Similar results were also obtained for other spatial regions in our dataset.

 \begin{figure}
 \thisfloatpagestyle{plain}
  \centering
    \subfigure[]{\includegraphics[width=0.32\textwidth]{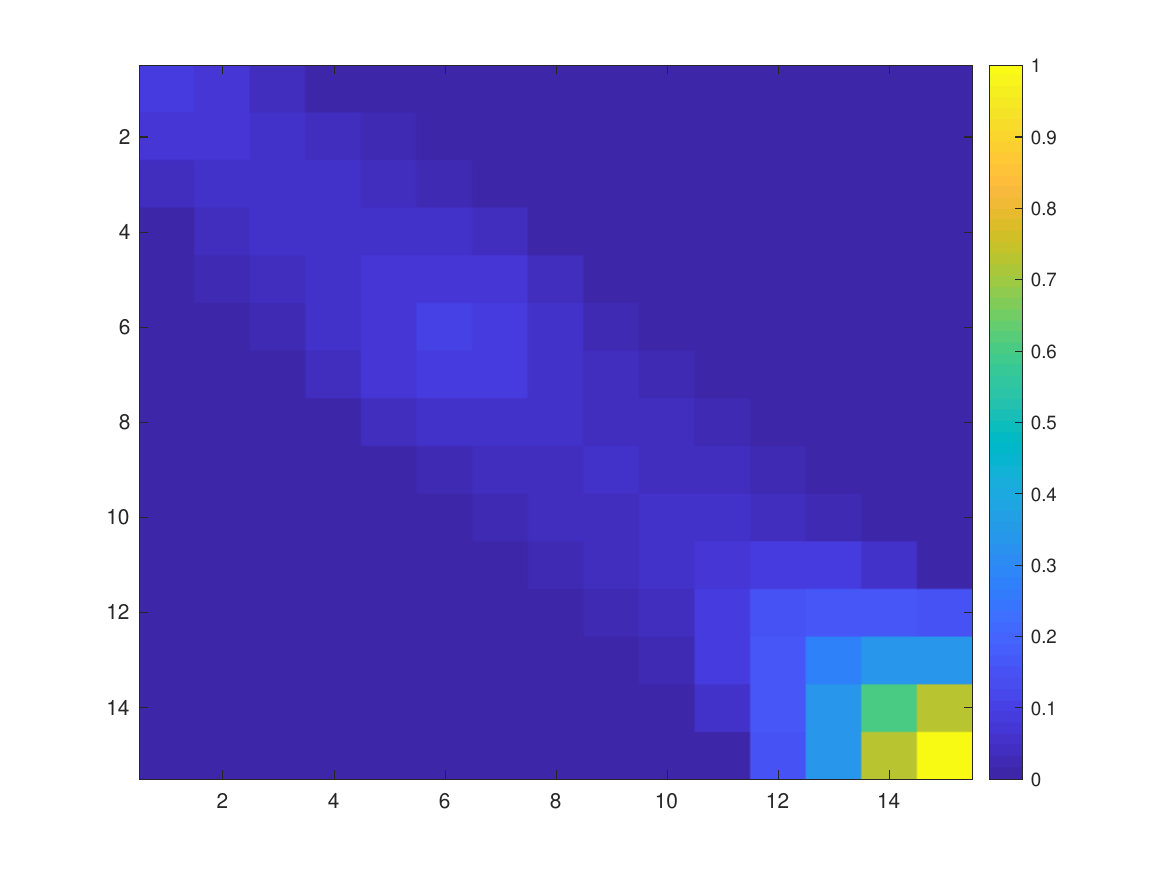}}
  \subfigure[]{\includegraphics[width=0.32\textwidth]{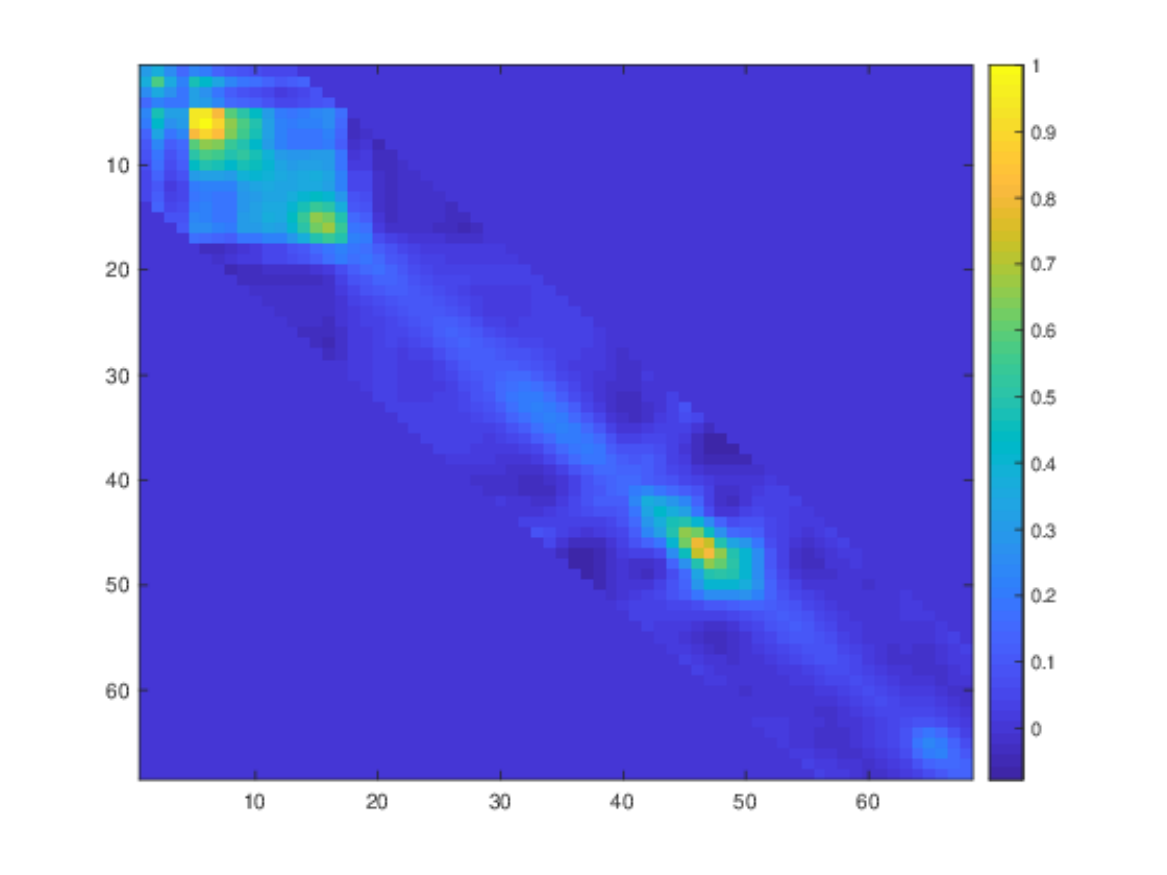}}
     \subfigure[]{\includegraphics[width=0.32\textwidth]{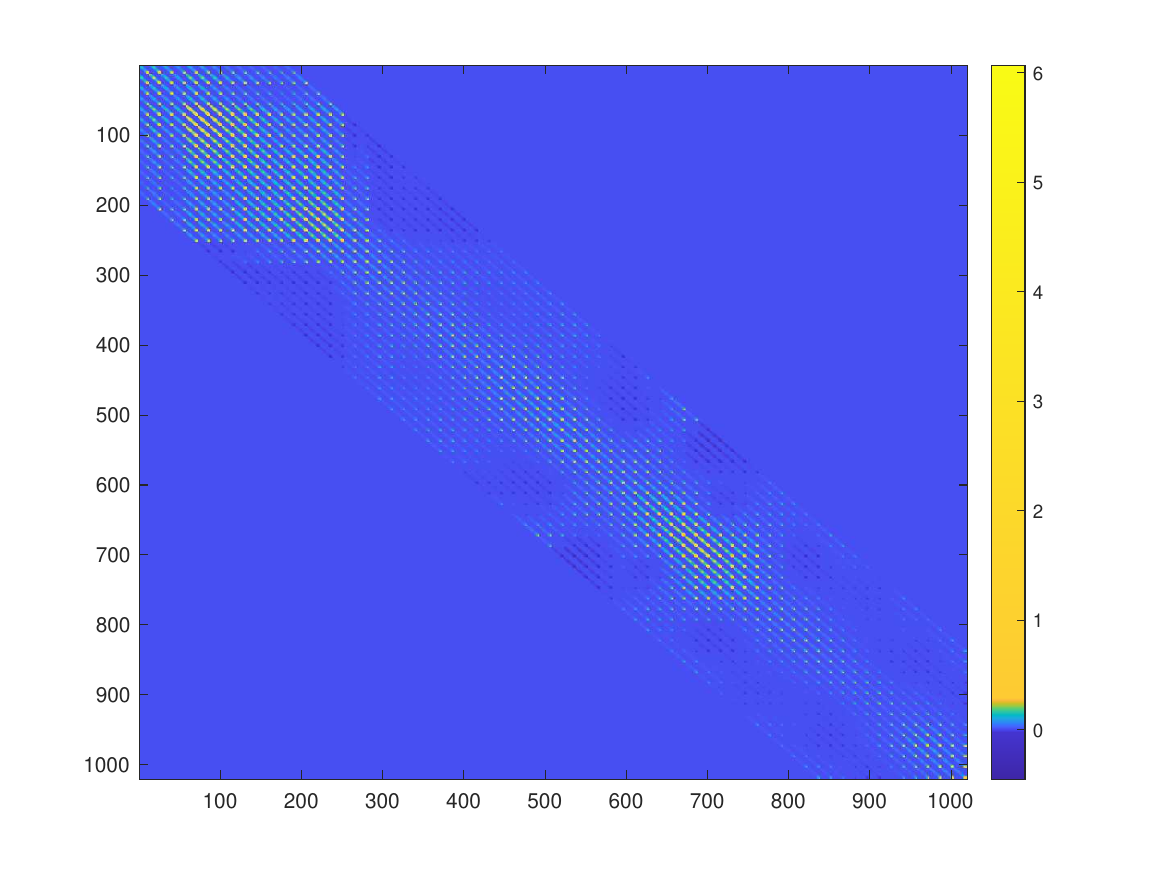}}
     \subfigure[]{\includegraphics[width=0.32\textwidth]{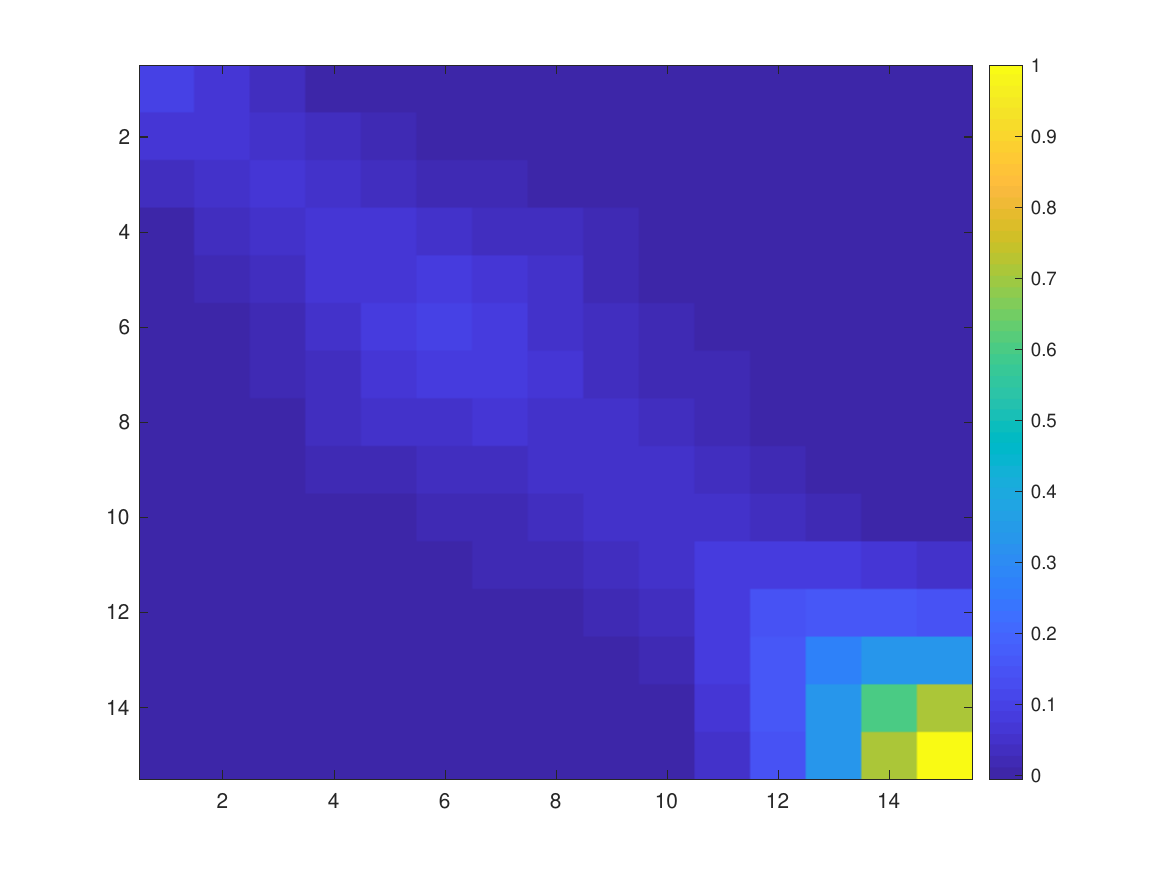}}
  \subfigure[]{\includegraphics[width=0.32\textwidth]{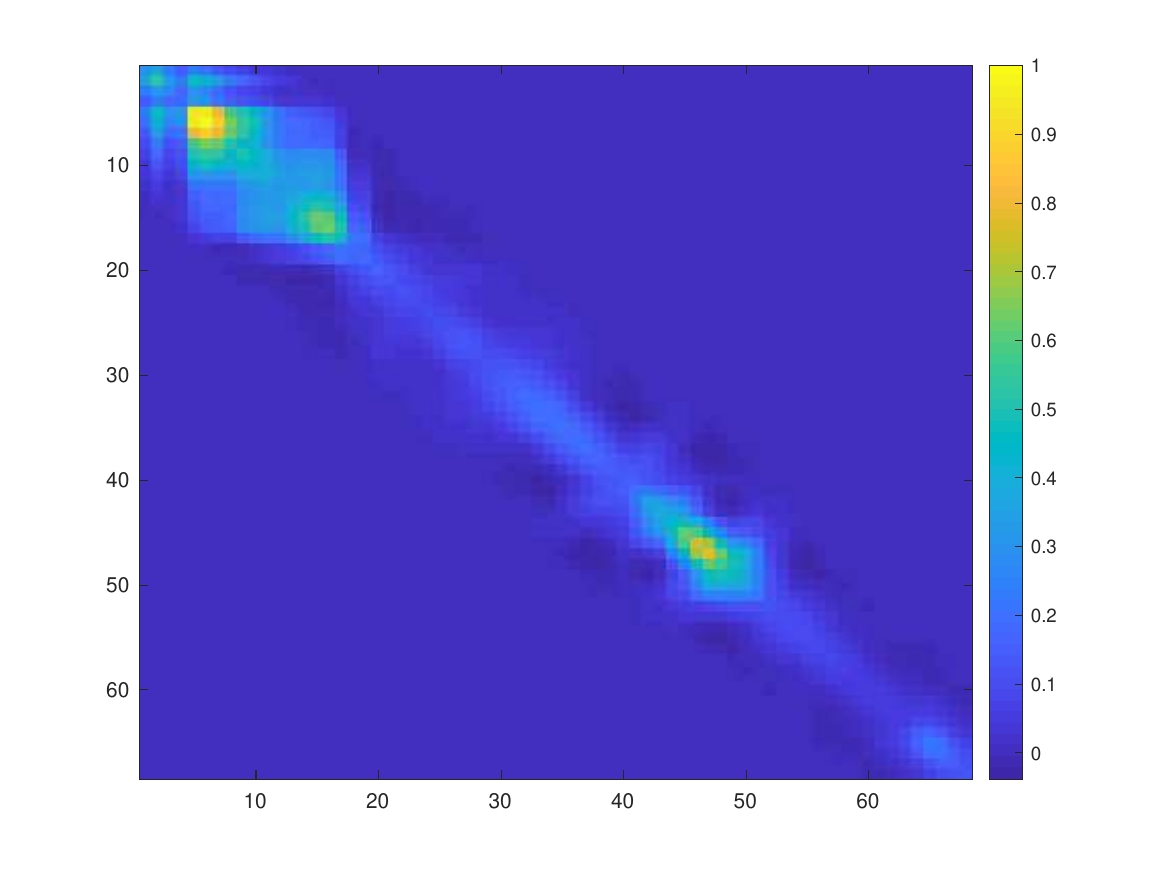}}
      \subfigure[]{\includegraphics[width=0.32\textwidth]{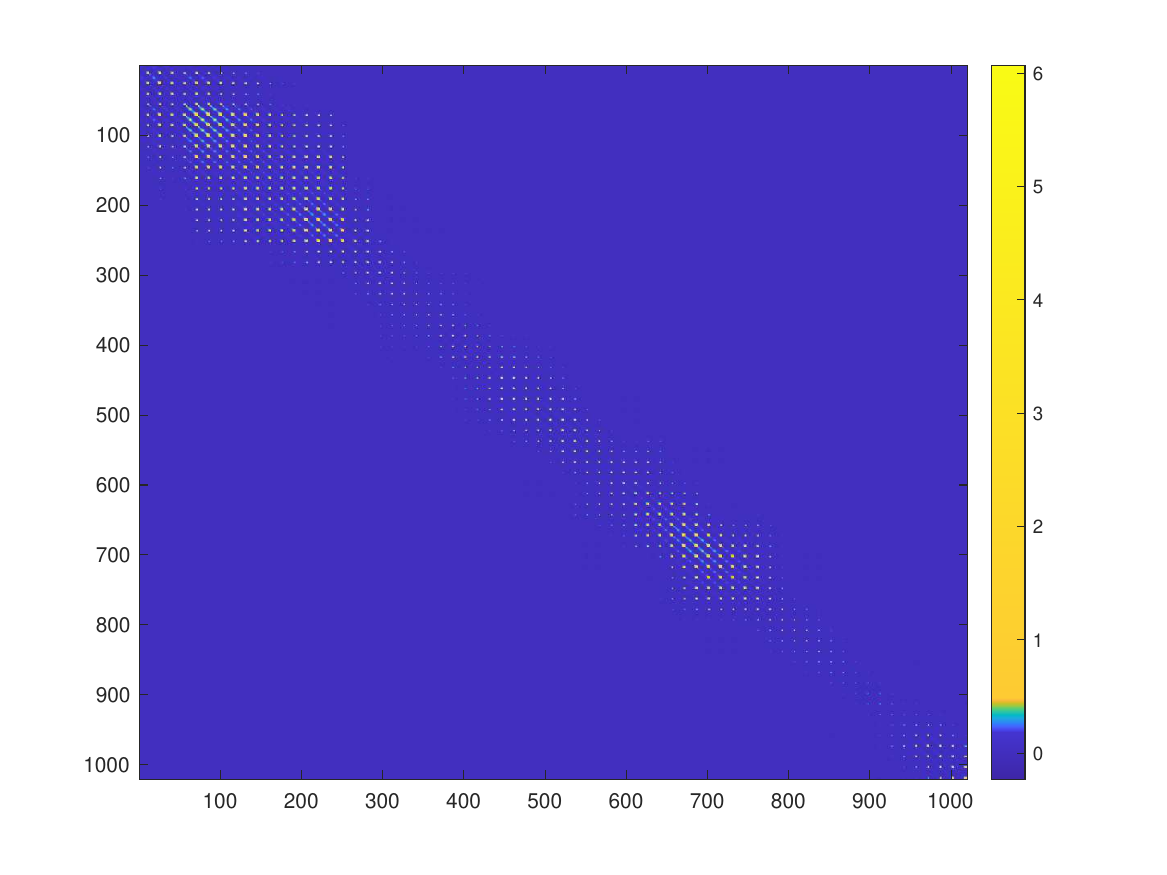}}

\caption{Temperature data analysis: Plots for the estimated covariance matrices obtained by proposed robust banded/tapering estimators: Panel (a) is the scaled robust banded latitude direction covariance $ \hat{\boldsymbol\Sigma}^{\mathcal{R},\MB}_{1}(\hat{k}^{\mathcal{B}}_{1}) $, Panel (b) is the scaled robust banded longitude direction covariance $ \hat{\boldsymbol\Sigma}^{\mathcal{R},\MB}_{2}(\hat{k}^{\mathcal{B}}_{2}) $ and Panel (c) is overall banded covariance $ \hat{\boldsymbol\Sigma}^{\mathcal{R},\MB}_{}(\hat{k}^{\mathcal{B}}_{1},\hat{k}^{\mathcal{B}}_{2})$. Panel (d) is the scaled robust tapering latitude direction covariance $ \hat{\boldsymbol\Sigma}^{\mathcal{R},\mathcal{T}}_{1}(\hat{k}^{\mathcal{T}}_{1})$, Panel (e) is the scaled robust tapering longitude direction covariance  $ \hat{\boldsymbol\Sigma}^{\mathcal{R},\mathcal{T}}_{2}(\hat{k}^{\mathcal{T}}_{2})$ and Panel (f) is the overall tapering covariance  $ \hat{\boldsymbol\Sigma}^{\mathcal{R},\mathcal{T}}_{}(\hat{k}^{\mathcal{T}}_{1},\hat{k}^{\mathcal{T}}_{2})$.}
\label{real_climate}
\end{figure}
\par
After preprocessing, we obtain a dataset of 15 (latitude) $\times$ 68 (longitude) matrix with a sample size of $n = 48$. {\color{black} We check the separability of our dataset's covariance structure via two procedures. First, we use the projection-based empirical bootstrap test \citep{aston2017tests}, implemented by function \textit{empirical\_bootstrap\_test} in R package {\tt covsep}. The p-value of the test is $0.133$, which confirms the validity of the separability assumption. Second, we compare the prediction error of covariance estimators with and without the separability. In particular, we randomly split the dataset into a training set with $n_1 = 48/2 = 24$ samples and a test set with the remaining half. The training set is used to calculate the sample covariance estimator (no separability) and our proposed separable estimator (separability assumed). For each estimator, we compare it with the sample covariance matrix obtained from the test sample as a benchmark, and then use the 1-norm and the Frobenius norm of their difference as the prediction errors. We repeat this procedure  500 times by using different random seeds for splitting the data. The average 1-norm and the  Frobenius  norm  prediction errors are $378.5$ $(\text{SE}=2.51)$ and $133.9$ $(0.48)$ without separability, in comparison with  $319.3 $ $(2.22)$ and $ 124.8$ $(0.50)$ with separability. The significant improvement in average prediction errors under both norms confirms the utility of imposed separability assumption.}
\par
Moreover, in Figure \ref{climate_trend_acf} (d), one can see outliers in the quantile-quantile (Q-Q) plot of all temperature anomalies. Therefore, we apply our proposed robust banded and tapering covariance estimation methods to the dataset. The threshold parameter $\tau$ for robust estimation is chosen by the method introduced in Section \ref{sec:RE} with 
\bee\label{candidate}
\mathscr{P} = \{100,99.995,99.99,99.97,99.95,99.93,99.91,99, 98,95,92,90,87,85,80\},
\ee 
after an initial evaluation over a wide range of values. The resampling scheme chooses $\hat{k}^{\mathcal{B}}_{1} = 3$, $\hat{k}^{\MB}_{2} = 12$ and  $\hat{\tau}^{\MB} = |x|_{99.97} = 5.28$ for the proposed robust banded estimator and picks $\hat{k}^{\mathcal{T}}_{1} = 6$, $\hat{k}^{\mathcal{T}}_{2} = 12$ and  $\hat{\tau}^{\mathcal{T}}_{} = |x|_{99.995} = 6.31$ for the proposed robust tapering estimator with a random split for $N = 10$ times. {\color{black}
\begin{remark}
The threshold parameter $\tau$ is chosen to be $99.99\%$ and $99.995\%$ for proposed banded and tapering estimators here. The large quantiles in the selected $\tau $ suggest the existence of outliers, which is a different data abnormality structure from the moment-based structure. This finding is also supported by the Q-Q plot in Figure \ref{climate_trend_acf} (d). In this case, our robust estimation procedure still works well because the initial truncation step can help remove these outliers. Meanwhile, our theoretical results mainly focus on the moment-based condition and it is still unclear how to extend these results when there exist outliers and the moment-based condition is violated. We leave it for future investigation. 
\end{remark}
}
\par
We plot the estimated covariance matrices in Figure \ref{real_climate}, including those over the latitude direction obtained by banding and tapering in (a) and (d), over the longitude direction in (b) and (e), and the overall covariance matrices in (c) and (f). All these matrices are scaled such that the maximum entry is $1$.  The results clearly suggest that the bandable assumption fits the data well, which is expected since the association between temperatures at two distant geographic areas is very weak. To further evaluate this assumption, we compare the magnitudes of entries removed from regularization (regularized entries) and unregularized entries, in latitude banded covariance $\hat{\M\Sigma}_1^{\mathcal{R},\mathcal{B}}(\hat{k}_1^{\mathcal{B}})$ and longitude banded covariance $\hat{\M\Sigma}_2^{\mathcal{R},\mathcal{B}}(\hat{k}_2^{\mathcal{B}})$. Over the latitude direction, the $l_1m_1$th entry in $\hat{\M\Sigma}_1^{\mathcal{R},\mathcal{B}}(\hat{k}_1^{\mathcal{B}})$ is a regularized entry if $ |l_1 - m_1|\leq \hat k_{1}^{ \mathcal{B}}= 3$, and is an unregularized entry otherwise. Similarly, we can define the regularized and unregularized entries in $\hat{\M\Sigma}_2^{\mathcal{R},\mathcal{B}}(\hat{k}_2^{\mathcal{B}})$. Due to \eqref{eq:band:0}, all regularized entries in $\hat{\M \Sigma}_{1}^{\mathcal{R},\MB}(\hat{k}^{\mathcal{B}}_{1})$ and $\hat{\M \Sigma}_{2}^{\mathcal{R},\MB}(\hat{k}^{\mathcal{B}}_{2})$ are zero. For comparison, we also implement our  proposed robust estimation procedure {\it without} banding (using same threshold $\hat{\tau}_{}^{\mathcal{R},\mathcal{B}}$), and obtain $\hat{\M \Sigma}_{1}^{\mathcal{R},\MB}(15)$ and $\hat{\M \Sigma}_{2}^{\mathcal{R},\MB}(68)$. We then plot the histograms of regularized entries in $\hat{\M \Sigma}_{1}^{\mathcal{R},\MB}(15)$, unregularized entries in $\hat{\M \Sigma}_{1}^{\mathcal{R},\MB}(15)$, and unregularized entries in $\hat{\M \Sigma}_{1}^{\mathcal{R},\MB}(\hat{k}_{1}^{\mathcal{B}})$ in Figure \ref{hist_climate} (a).  Similarly to Figure \ref{real_climate}, maximum magnitudes of $\hat{\M \Sigma}_{1}^{\mathcal{R},\MB}(15)$ and $\hat{\M \Sigma}_{1}^{\mathcal{R},\MB}(\hat{k}_1^{\mathcal{B}})$ are scaled to 1. The histogram clearly shows that the magnitudes of regularized entries in $\hat{\M \Sigma}_{1}^{\mathcal{R},\MB}(15)$ are generally much smaller than those of unregularized entries in $\hat{\M \Sigma}_{1}^{\mathcal{R},\MB}(15)$ and $\hat{\M \Sigma}_{1}^{\mathcal{R},\MB}(\hat{k}_{1}^{\mathcal{B}})$. This finding confirms the bandable structure over the latitude direction.   
Similar findings are observed in Figure \ref{hist_climate} (b) for the bandable structure over the longitude direction.
\par
In Figure \ref{real_climate}, for the latitude direction, there are two clusters of areas with strong correlation around $50^\degree E$ and $83^\degree W$. For the longitude direction, there is a cluster of areas around $40^\degree N$ that has a strong correlation (lower-right corner in (a) and (d)). Those coordinates correspond to Great Lakes (USA) and the Caspian Sea. To further illustrate the use of our estimated covariance for the matrix-valued data, we focus on the $5^\degree\times 5^\degree$ box centered at $107.5^\degree W$ longitude and $22.5^\degree N$ latitude (target region: west coast side of Mexico), and estimate its covariance with the temperature at other regions. The results obtained by banding and tapering are summarized in Figure \ref{climate_visi}. From the plot, there are two regions that have a strong correlation with the target region. The first one marked in red is the target region itself, which suggests a strong self-correlation in the neighborhood areas around the target region. The second marked in blue corresponds to ocean area (northeastern direction) near Hawaiian islands and they have a strong negative correlation, which may be related to the recent studies on Land–Ocean Surface Temperature ratio \citep{Lambert2011}.

\vspace{-.1in}
\section{Discussion}\label{discussion} 
In this paper, we propose banded and tapering covariance estimators for matrix-valued data under a separability condition. We adopt an efficient computational algorithm and derive the convergence rates of our covariance estimates under various scenarios. To deal with heavy-tailed data, we further propose robust banded and tapering covariance estimators, and show their theoretical advantages.  
\par
Bandable covariance structure plays an important role in our methodology. In practice, it is possible that the true covariance matrix is only bandable over one direction. We note this scenario can be naturally handled by our method. For instance, suppose $\M \Sigma_1^*$ in $\M \Sigma^* =\M \Sigma_2^*\otimes \M \Sigma_1^*$ is not bandable, then we can choose the bandwidth as $k_1 = p$ for the proposed banded estimator, and $k_1 = 2p$ for the proposed tapering estimator (see Remark \ref{rm:2ptaper}), to eliminate banding or tapering over the row direction.
\par
In terms of theory, although Theorems \ref{T2} and \ref{T3} are established under which both $\M \Sigma_1^*$ and $\M \Sigma_2^*$ are assumed belonging to regularity class $\mathcal{F}(\varepsilon_0,\alpha)$ or $\mathcal{M}(\varepsilon_0,\alpha)$, we find that the proof of Theorems \ref{T2} and \ref{T3} can adapt to this new case when $\M \Sigma_2^* \in \mathcal{F}(\varepsilon_0,\alpha_2)$ or $\mathcal{M}(\varepsilon_0,\alpha_2)$ with $\alpha_2 > 0$,  and $\M \Sigma_1^*$ only needs to satisfy $\|\M \Sigma_1^*\|_{2}\leq C_{\max}$ with some constant $C_{\max}>0$. Keeping other conditions unchanged and selecting $k_1 = p$ for the proposed banded estimator and $k_1 = 2p$ for the proposed tapering estimator, the error bound in Theorem \ref{T2} for sub-Gaussian scenario, becomes
\bee\label{T2:oned}
&\E\Big(\frac{\|\hat{\M\Sigma}^\eta_2\kii\otimes\hat{\M\Sigma}^\eta_1\ki - \M\Sigma^*\|_{\F}^2}{pq}\Big) \precsim
\begin{cases}\frac{p}{qn} + \frac{k_2}{pn}  +  \M I_{\eta,q}(k_2)\cdot k_2^{-\tilde{\alpha}_2}, & p^2 + qk_2 \precsim n
\\
\big(\frac{pk_2}{n}\big)\wedge\big(\frac{p^3}{qn^2} + \frac{qk^2_2}{pn^2}  \big) +  \M I_{\eta,q}(k_2)\cdot k_2^{-\tilde\alpha_2}, & p^2 + qk_2 \succ n.  \nonumber 
\end{cases}
\ee
 Additionally, the error bound in Theorem \ref{T3} for finite fourth moment scenario becomes
$
\E\Big(\frac{\|\hat{\M\Sigma}^\eta_2\kii\otimes\hat{\M\Sigma}^\eta_1\ki - \M\Sigma^*\|_{\F}^2}{pq}\Big) \precsim
\frac{pk_2}{n} +  \M I_{\eta,q}(k_2)\cdot k_2^{-\tilde\alpha_2}.  \nonumber 
$

Besides the spectral norm consistency problem discussed in Section \ref{sec:disc:spnorm}, there are a number of important directions for further investigation. First, it is still unclear whether the convergence of our proposed estimators achieves the lower bound when $p,q$ are beyond the rate-optimal regimes that we have discussed in Section \ref{sec:compare}. {\color{black}Second, as discussed in Section \ref{sec:disc:spnorm}, an efficient algorithm to solve the spectral-norm Kronecker product approximation is of interest.} 
Third, extension of the current approach to tensor-valued data is highly non-trivial in both computation and theory. Therefore, we leave it for future research. 

\section*{Acknowledgement}
The authors would like to thank the editor, associate editor, and three reviewers for their constructive comments, which have substantially improved the paper. Shen's research was partially supported by Simons Foundation Award 51262. Kong's research was partially supported by the Natural Science and Engineering Research Council of Canada. 
\bibliographystyle{chicago}
\bibliography{covariance}
\newpage
{\beginsupplement
\def\spacingset#1{\renewcommand{\baselinestretch}%
{#1}\small\normalsize} \spacingset{1}

%%%%%%%%%%%%%%%%%%%%%%%%%%%%%%%%%%%%%%%%%%%%%%%%%%%%%%%%%%%%%%%%%%%%%%%%%%%%%%

\if1\blind
{
  \title{\bf Supplementary File for ``Covariance Estimation for Matrix-valued Data''}
  \date{}
  \maketitle
} \fi

\if0\blind
{
   \title{\bf Supplementary File for ``Covariance Estimation for Matrix-valued Data''}
        \date{}
     \maketitle
} \fi
 
\spacingset{1.5} % DON'T change the spacing!

{\hypersetup{linkcolor=black}
 \tableofcontents
}
\TOCstart
\spacingset{1.5} 
\section{Notation}\label{sec:notation}
\subsection{Vector/Matrix Representation}
Let $\M A$  be an arbitrary matrix. We define the following notation. 
\begin{itemize}
\item For a matrix $\M A \in \RR^{p\times q}$, we denote the vectorization of $\M A$ by $\vecc(\M A) \equiv \big(A_{1,1}, A_{2,1},\dots\dots,A_{p,q}\big)^{\T}$, where $A_{l_1,l_2}$, is the $l_1l_2$th entry of $\M A$, for $1\leq l_1\leq p$ and $1\leq l_2\leq q$.
\item For a matrix $\M A \in\mathbb{R}^{pq \times pq}$, let $A_{\{(l_1,m_1),(l_2,m_2)\}}$ represent the $\big\{(l_2-1)\cdot p + l_1\big\},\big\{(m_2-1)\cdot p + m_1\big\}$th entry of $\M A$ for $1\leq l_1,m_1\leq p, 1\leq l_2,m_2\leq q$. 
\item For a matrix $\M A\in\RR^{d_1\times d_2}$ with arbitrary size $d_1,d_2>0$, we denote the $l_1m_1$th entry of $\M A$ by $A_{l_1,m_1}$, $1\leq l_1 \leq d_1, 1\leq m_1\leq d_2$, the rank of $\M A$ by $\text{Rank}(\M A)$, and the trace of $\M A$ by $\tr(\M A)$. We also use $[\M A]_{l_1,m_1}$ to represent the $l_1m_1$th entry of $\M A$, i.e., $[\M A]_{l_1,m_1} = A_{l_1,m_1}$.
\end{itemize}
Recall $\M X_1,\dots,\M X_n$ are i.i.d. $p\times q$ random matrix samples and $x^{(i)}_{l_1,l_2}$ is the $l_1l_2$th entry of $\M X_i$. The $\vecc(\M X_i)\equiv \big(x^{(i)}_{1,1}, x^{(i)}_{2,1},\dots\dots,x^{(i)}_{p,q}\big)^{\T} \in \RR^{pq}$ is the vectorization of $\M X_i$. %Denote $\vecc(\M X_i) \equiv (x^{i}_{1},\dots,x^{i}_{pq})^{\T}$ and $x^{i}_{l_1,l_2}$ the $(l_1,l_2)$th entry of $\M X_i$. 
\par
For the true covariance of $\M X_i$, denoted by $\M\Sigma^*$, we assume it is separable, i.e. $\M\Sigma^* = \M\Sigma_2^*\otimes\M\Sigma_1^*\in\RR^{pq\times pq}$, where $\M \Sigma^{*}_1 \equiv \big [ \sigma_{l_1,m_1}^{(1)}\big]_{p\times p}\in\RR^{p\times p},\M \Sigma^{*}_2 \equiv \big [ \sigma_{l_2,m_2}^{(2)}\big ]_{q\times q}\in\RR^{q\times q}$, and $\sigma_{l_a,m_a}^{(a)}$ is the $l_a m_a$th entry  of $\M \Sigma_{a}^*$ for $a \in \{1,2\}$. Let $\sigma_{(l_1,m_1),(l_2,m_2)}$ be the $\big\{(l_2-1)\cdot p + l_1\big\},\big\{(m_2-1)\cdot p + m_1\big\}$th entry of $\M \Sigma^*$. By the separability assumption, one has $\sigma_{(l_1,m_1),(l_2,m_2)} = \sigma_{l_1,m_1}^{(1)}\cdot \sigma_{l_2,m_2}^{(2)} $. 

Recall that in \eqref{def:b} and \eqref{def:Tk}, we define 
\bee\nonumber
B_{k}(\M A) = [A_{l,m}\cdot \textbf{I}(|l - m|\leq k)]_{d\times d}
\ee 
and $T_k(\M A) = [T_k(\M A)_{l,m}]_{d\times d}$, where
\bee\nonumber
T_k(\M A)_{l,m} = \begin{cases}
A_{l,m} & \text{ when } |l - m|\leq \lfloor k/2 \rfloor, \\
(2 - \frac{|l - m|}{\lfloor k/2 \rfloor})A_{l,m} & \text{ when } \lfloor k/2 \rfloor < |l - m| \leq k, \\
0 & \text{ otherwise.}
\end{cases}
\ee
We then define $\M \Sigma^{*,\MB}_a(k) \equiv B_k(\M \Sigma^*_a)$ and $\M \Sigma^{*,\mathcal{T}}_a(k) \equiv T_k(\M \Sigma^*_a)$ for $a = 1,2$. 
\subsection{Norms}
For a vector $\mathbf{v}=[v_l]_d\in\RR^{d}$, we denote $\|\mathbf{v}\|$ its Euclidean norm. We also define the following vector norms, 
$$
\|\mathbf{v}\|_{\max} = \max_{1\leq l\leq d}|v_l|,\quad\|\mathbf{v}\|_1 = \sum_{l = 1}^d|v_l|.
$$ 
\par
For a matrix $\Ab=[A_{l_1l_2}]_{d_1\times d_2}\in \RR^{d_1\times d_2}$, we denote $\|\Ab\|_{\rF}$ its Frobenius norm. We also define the following matrix norms,
\bee\nonumber
&\|\Ab\|_2\equiv \sup\{\|\M A\mathbf x\|_2,\|\mathbf x\|_2 = 1\},
\\
&\|\M A\|_1 \equiv \sup\{\|\M A\mathbf x\|_1,\|\mathbf x\|_1 = 1\} = \max_{1\leq l_2 \leq d_2}\sum_{l_1= 1}^{d_1}|A_{l_1l_2}|,
\\ 
&\|\M A\|_{\infty} \equiv \sup\{\|\M A\mathbf x\|_{\max},\|\mathbf x\|_{\max} = 1\} = \max_{1\leq l_1\leq d_1}\sum_{l_2= 1}^{d_2}|A_{l_1l_2}|,
\\
&\|\M A\|_{\max} \equiv \max_{1\leq l_1\leq d_1\atop 1 \leq l_2\leq d_2}|A_{l_1l_2}|, \quad\|\M A\|_{1,1} = \sum_{l_1 =1}^{d_1}\sum_{l_2 = 1}^{d_2}|A_{l_1l_2}|.
\ee
\par
For two probability measures $\bds P, \bds Q$ with density $\bds p, \bds q$, and with respect to a common measure $\mu$, we also define the $L_1$ norm of $\bds P$ as $\|\bds P\|_{1} \equiv \int|\bds p|d\mu$. 
\subsection{Sub-exponential and Sub-Gaussian Random Vectors}\label{def:sesgrv}
Same as the definition in \cite{wainwright2019high}, a random variable $X$ is a sub-exponential random variable if
\bee\label{subexp}
\E\Big[\exp\big\{\lambda(X - \mu)\big\}\Big]\leq \exp\Big(\frac{v^2\lambda^2}{2}\Big),
\ee
with some non-negative parameters $(v,\alpha)$ for all $|\lambda|<1/\alpha$.
For a sub-exponential random variable $X$, we define the sub-exponential norm $\|\cdot\|_{\psi_1}$ as
\bee\label{def:psi1}
\|X\|_{\psi_1} &\equiv \inf\{t>0:\E \exp(|X|/t)\leq 2\}.
\ee 

\par
We say a random vector $\bds X$ follows a sub-Gaussian distribution if for any $t>0$ and $\|{\bf v}\| = 1$, there exists a $\rho>0$ such that 
\bee\label{subga:supp}
\Pr \left\{\Big|{\bf v}^{\T} \big(\bds X - \E \bds X\big)
\Big|>t\right\} \leq e^{-\rho t^2 }.
\ee
For univariate  random variable $X\in\RR$, we define its sub-Gaussian norm $\|\cdot\|_{\psi_2}$ as 
\bee\label{def:psi2}
\|X\|_{\psi_2} &\equiv \inf\{t>0:\E \exp(X^2/t^2)\leq 2\}.
\ee 
For random vector $\bds X\in \RR^d$, we define its sub-Gaussian norm $\|\cdot\|_{\psi_2}$ as
\bee\nonumber
\|\bds X\|_{\psi_2} \equiv \sup_{\bds v\in \RR^{d},\|\bds v\| = 1}\|\bds v^\T\bds X\|_{\psi_2}.
\ee
\subsection{Notation for Proofs of Theorems \ref{T2}, \ref{T3}, \ref{T:indi} and \ref{Ts}}\label{notation:main}
\par
Denote $\mathcal{U}_{d} = \{\|\mathbf{x}\| = 1| \mathbf{x}\in \RR^d\}$ a unit sphere in $\RR^{d}$, and define $[k] = \{1,\dots,k\}$ for integer $k > 0$. Let $\hat{\M{\Sigma}}_0 = \frac{1}{n}\sum_{i=1}^{n}\vecc(\M X_i )\vecc(\M X_i)^\T = \Big[\hat{\sigma}^{(0)}_{(l_1,m_1),(l_2,m_2)}\Big]_{pq\times pq}$. We also define 
\bee\label{l:srate:1}
&\tilde{\M\Sigma}_{0,\MB}(k_1,k_2) \equiv \hat{\M \Sigma}_0 \HDBand \equiv \Big[\tilde{\sigma}^{\mathcal{B}}_{(l_1,m_1),(l_2,m_2)}\Big]_{pq\times pq}, 
\\
&\tilde{\M\Sigma}_{0,\mathcal{T}}(k_1,k_2) \equiv \hat{\M \Sigma}_0 \HDTaper \equiv \Big[\tilde{\sigma}^{\mathcal{T}}_{(l_1,m_1),(l_2,m_2)}\Big]_{pq\times pq},
\ee
 as the doubly banded and tapering matrices of $\hat{\M{\Sigma}}_0$. 
\par
For the proposed banded estimator, let $w_{l_1,m_1}^{p,k,\mathcal{B}}$ be the $l_1m_1$th element of the matrix $B_{k}(\textbf{1}_{p})$ and $w_{l_2,m_2}^{q,k,\mathcal{B}}$ be the $l_2m_2$th element of the matrix $B_{k}(\textbf{1}_{q})$.  By definition, it is easy to see
  \bee\nonumber
  \tilde{\sigma}^{\mathcal{B}}_{(l_1,m_1),(l_2,m_2)} = w_{l_2,m_2}^{q,k_2,\mathcal{B}} \cdot w_{l_1,m_1}^{p,k_1,\mathcal{B}}\cdot \hat{\sigma}^{(0)}_{(l_1,m_1),(l_2,m_2)}.
  \ee We further let $\boldsymbol\beta_1 = (l_1,m_1)^\T,\boldsymbol\beta_2 = (l_2,m_2)^\T$, $w^{\mathcal{B}}_{
\boldsymbol\beta_1,\boldsymbol\beta_2} = w_{l_1,m_1}^{p,k_1,\mathcal{B}} \cdot w_{l_2,m_2}^{q,k_2,\mathcal{B}}$ and $\sigma_{\boldsymbol\beta_1,\boldsymbol\beta_2} = \sigma_{l_1,m_1}^{(1)}\cdot \sigma_{l_2,m_2}^{(2)}$. 
\par
For the proposed tapering estimator, similarly define $w_{l_1,m_1}^{p,k,\mathcal{T}}, w_{l_2,m_2}^{q,k,\mathcal{T}}$ as the $l_1m_1$th and $l_2m_2$th elements of the matrix $T_{k}(\textbf{1}_{p})$ and $T_{k}(\textbf{1}_{q})$, respectively.  We can write
$\tilde{\M\Sigma}_{0,\mathcal{T}}(k_1,k_2) = \hat{\M \Sigma}_0 \HDTaper = \big[\tilde{\sigma}^{\mathcal{T}}_{(l_1,m_1),(l_2,m_2)}\big]_{pq\times pq}  
$
 as the doubly tapering matrix of $\hat{\M{\Sigma}}_0$ and have,
$
  \tilde{\sigma}^{\mathcal{T}}_{(l_1,m_1),(l_2,m_2)} = {w_{l_2,m_2}^{q,k_2,\mathcal{T}} \cdot w_{l_1,m_1}^{p,k_1,\mathcal{T}}}_{}\cdot \hat{\sigma}^{(0)}_{(l_1,m_1),(l_2,m_2)}
$. We also denote $w^{\mathcal{T}}_{\boldsymbol\beta_1,\boldsymbol\beta_2}\equiv w_{l_2,m_2}^{q,k_2,\mathcal{T}} \cdot w_{l_1,m_1}^{p,k_1,\mathcal{T}}$.

\par
Finally, we introduce some notation in matrix perturbation theory. For $\M V, \M W\in \RR^{p\times d}$ both having orthonormal columns, the $d$ principal angles between their column spaces is $\{\arccos(\varsigma_1), \arccos(\varsigma_2), \dots, \arccos(\varsigma_d)\}$, where $\varsigma_1\geq \varsigma_2 \geq \dots \geq \varsigma_d\geq 0$ are the singular values of $\M W^{\T} \M V$. The $\Theta(\M W, \M V)$ is a $d\times d$ diagonal matrix with $ll$th entry $l$th principal angle for $ l=1,\ldots, d$. And $\sin\Theta(\M W, \M V)$ is defined as a matrix by applying $\sin$ entrywise to the matrix $\Theta(\M W, \M V)$.

\subsection{Notation for Proof of Theorem \ref{T4}}\label{notation:robust}
Recall that when $\vecc(\mathbf{X}_i) \in \mathbb{R}^{pq}$, the modified estimator is defined as
\bee\nonumber
\widecheck{\M \Sigma} = \frac{1}{n}\sum_{i =1}^{n} \big\{\vecc(\widecheck{\M X}_i) - \vecc(\widecheck{{\boldsymbol\mu}})\big\}\cdot \big\{\vecc(\widecheck{\M X}_i)- \vecc(\widecheck{{\boldsymbol\mu}})\big\}^{\T},
\ee
where $\vecc(\widecheck{\mathbf{X}}_i)$ satisfies $\widecheck{x}^{(i)}_{l_1,l_2} = \text{sgn}(x^{(i)}_{l_1,l_2})(|x^{(i)}_{l_1,l_2}|\wedge \tau)$ with some $\tau >0$, and $\widecheck{{\boldsymbol\mu}} = \frac{1}{n}\sum_{i = 1}^n\widecheck{\M X}_i$. We denote ${\widecheck{\M X}}_{i}^{\mathcal{c}} \equiv {\widecheck{\M X}}_{i} - \E {\widecheck{\M X}}_{i}^{}$ as the centered ${\widecheck{\M X}}_{i}$. Since $\{\widecheck{\M X}_i\}_{i = 1}^n$ are i.i.d., $\widecheck{\M \Sigma}$ can be rewritten as, 
\bee\nonumber
\widecheck{\M \Sigma} &= \frac{1}{n}\sum_{i =1}^{n} \big[\vecc(\widecheck{\M X}_i) - \vecc(\widecheck{{\boldsymbol\mu}})\big]\cdot \big[\vecc(\widecheck{\M X}_i)- \vecc(\widecheck{{\boldsymbol\mu}})\big]^{\T}
\\
&=\frac{1}{n}\sum_{i =1}^{n} \Big[\big\{\vecc(\widecheck{\M X}_i) - \vecc(\E {\widecheck{\M X}}_{i})\big\} - \big\{\vecc(\widecheck{{\boldsymbol\mu}}) - \vecc(\E {\widecheck{\M X}}_{i})\big\}\Big]
\\
&\times \Big[\big\{\vecc(\widecheck{\M X}_i)-  \vecc(\E {\widecheck{\M X}}_{i})\big\}-\big\{\vecc(\widecheck{{\boldsymbol\mu}})-\vecc(\E {\widecheck{\M X}}_{i})\big\}\Big]^{\T}
\\
&=\frac{1}{n}\sum_{i =1}^{n}\big[\vecc({\widecheck{\M X}}_{i}^{\mathcal{c}}) - \vecc(\widecheck{{\boldsymbol\mu}}^{\mathcal{c}})\big]\cdot \big[\vecc({\widecheck{\M X}}_{i}^{\mathcal{c}}) - \vecc(\widecheck{{\boldsymbol\mu}}^{\mathcal{c}})\big]^{\T}
\ee
where $\widecheck{{\boldsymbol\mu}}^{\mathcal{c}} = \frac{1}{n}\sum_{i = 1}^n\widecheck{\M X}_i - \E\widecheck{\M X}_i = \frac{1}{n}\sum_{i = 1}^n\widecheck{\M X}^{\mathcal{c}}_i$. 
\par
Recall ${\M \Sigma}^{*}_{\mathcal{R}} \equiv \cov(\widecheck{\M X}_i) = \cov({\widecheck{\M X}}^{\mathcal{c}}_i)$. Let $\sigma^{\mathcal{R}}_{(l_1,m_1),(l_2,m_2)}$ be the $\big\{(l_2-1)\cdot p + l_1\big\},\big\{(m_2-1)\cdot p + m_1\big\}$th entry of $\M \Sigma^*_{\mathcal{R}}$. We note that unlike $\M \Sigma^* = \M \Sigma^*_2 \otimes \M \Sigma^*_1$, the ${\M \Sigma}^{*}_{\mathcal{R}}$ may not still have Kronecker product structure.  For $\eta\in\{\mathcal{B},\mathcal{T}\}$, recall that $\widecheck{\M \Sigma}_{\eta}(k_1,k_2)$ is the doubly banded or tapering matrix of $\widecheck{\M\Sigma}$ with bandwidths $k_1,k_2$, i.e.,
\bee\nonumber
&\widecheck{\M\Sigma}_{\MB}(k_1,k_2) \equiv \widecheck{\M\Sigma} \HDBand,
\\
&\widecheck{\M\Sigma}_{\mathcal{T}}(k_1,k_2) \equiv \widecheck{\M\Sigma} \HDTaper.
\ee 
 Similarly we denote 
\bee\nonumber
&{\M \Sigma}^{*,\mathcal{B}}_{\mathcal{R}}(k_1,k_2) \equiv {\M \Sigma}^{*}_{\mathcal{R}} \circ \{B_{k_2}({\bf 1}_{q})\otimes B_{k_1}({\bf 1}_{p})\},
\\
&{\M \Sigma}^{*,\mathcal{T}}_{\mathcal{R}}(k_1,k_2) \equiv {\M \Sigma}^{*}_{\mathcal{R}} \circ \{T_{k_2}({\bf 1}_{q})\otimes T_{k_1}({\bf 1}_{p})\}.
\ee

\section{Discussion of Theorems in Section \ref{Sec_MR}}
\subsection{Bounded Maximum Eigenvalues of $\mathbf{\Sigma}_1^*$ and $\mathbf{\Sigma}_2^*$}\label{sec:bounde}
In the theoretical investigation, we consider the scenario when the covariance matrices $\M\Sigma_1^*$ and $\M\Sigma_2^*$, representing  covariances among the rows and columns of the matrix-valued data $\M X$, are in the approximately bandable covariance classes $\mathcal{F}(\varepsilon_0, \alpha)$ or $\mathcal{M}(\varepsilon_0, \alpha)$. 
\par
As pointed out by one referee, the maximum eigenvalue of the covariance matrices may increase with the growing dimension in some applications. However, for the classes of covariance matrices $\mathcal{F}(\varepsilon_0, \alpha)$ or $\mathcal{M}(\varepsilon_0, \alpha)$, the bounded maximum eigenvalue assumption is indeed reasonable.

%In this section, we further show that, under the background of bandable covariance estimation specifically considered by this paper, the bounded eigenvalue conditions of $\M\Sigma_1^*$ and  $\M\Sigma_2^*$ presented in $\mathcal{F}(\varepsilon_0, \alpha)$ and $\mathcal{M}(\varepsilon_0, \alpha)$ are actually necessary. In another word, the bandable structures defined in $\mathcal{F}(\varepsilon_0, \alpha)$ and $\mathcal{M}(\varepsilon_0, \alpha)$ almost already imply bounded eigenvalues of $\M\Sigma_1^*$ and $\M\Sigma_2^*$, under some mild conditions.
\par In particular, we assume the following relaxed $\mathcal{F}^*(\alpha)$ and $\mathcal{M}^*(\alpha)$ matrix classes, which preserve the particular bandable covariance structures, but do not impose assumptions on bounded maximum eigenvalues, 
\bee\label{A1*}
\mathcal{F}^*(\alpha)=&\bigg\{ \bSigma : ~ \max_l\sum_m \{|\sigma_{l,m}|: |l-m|>k \}\leq C_0 k^{-\alpha}~\textnormal{for all}~k\geq 1 \bigg\},
\ee
and
\bee\label{A2*}
\mathcal{M}^*(\alpha) = &\bigg\{\M\Sigma:~\left|\sigma_{l,m}\right| \leq C_{1}|l-m|^{-\alpha-1} \text { for } l \neq m \bigg\},
\ee
where $\alpha,C_0,C_1 > 0$ are some fixed constants. 

%Now assume $\M \Sigma_1^*\in \mathcal{F}^*(\alpha)$ or $\mathcal{M}^*(\alpha)$  is a covariance matrix of some random vector $\mathbf{x}\in\RR^{p}$. Then with some extra, yet mild conditions, e.g., $\mathbf{x}\in\RR^{p}$ has bounded fourth order moment condition or $\mathbf{x}\in\RR^{p}$ is sub-Gaussian, we can show the bandable structure defined in $\mathcal{F}^*(\alpha)$ or  $\mathcal{M}^*(\alpha)$ immediately imply  $\lambda_{\max}(\M\Sigma_1^*)\precsim 1$; similarly $\lambda_{\max}(\M\Sigma_2^*)\precsim 1$ under symmetric conditions. The statement is formalized as follows.
%In this scenario, we actually can show the eigenvalues of $\M\Sigma^*$ are bounded, \textit{without} assuming the bounded eigenvalues conditions in advance, as assumed $\mathcal{F}(\varepsilon_0, \alpha)$ or $\mathcal{M}(\varepsilon_0, \alpha)$. The formal statement is in the following proposition.

We can show that the maximum eigenvalue of the covariance matrices in $\mathcal{F}^*(\alpha)$ or $\mathcal{M}^*(\alpha)$ is bounded under finite fourth-order moment conditions. The formal statement is included in the following proposition. 
\begin{proposition}\label{po:boundsig}
Let $\mathbf{x} = [x_l] \in \RR^{p}$ be a random vector with covariance matrix $\M\Sigma_1^*$, and $\mathbf{x} $ satisfies  the element-wise finite fourth moment condition:
$
\E(|x_{l_1}\cdot x_{m_1}|^2) \leq M < +\infty
$
for any $1 \leq l_1,m_1\leq p$, with some fixed constant $M > 0$. If $\M\Sigma_1^* \in \mathcal{F}^*(\alpha)$ or $\mathcal{M}^*(\alpha)$, then there exists some fixed constant $C_* > 0$, such that
$
\lambda_{\max}(\M\Sigma_1^*) \leq C_*.
$
\end{proposition}

\begin{remark}
As  discussed in Section \ref{sec:T:up}, the sub-Gaussian condition of $\mathbf{x}$ is stronger than the forth-order moment condition of $\mathbf{x}$. Therefore, the Proposition \ref{po:boundsig} also holds when $\mathbf{x}$ follows sub-Gaussian.
\end{remark}
\begin{comment}
In Proposition \ref{po:boundsig}, the bandable covariance structures \eqref{A11} and \eqref{A12} for the covariance matrix of each row and column of random variables in $\M X$, are the same assumptions as in the literature of bandable covariance estimation for vector-valued data \citep{Bickel2008threshold,Cai2010,cai2012adaptive,cai2016estimating}. We emphasize that we do not require the row-wise and column-wise covariance matrices to have bounded eigenvalues in advance. On the other hand, the assumption of a constant lower bound of $\|\M\Sigma^*\|_{\max}$ is mild, as it is always allowed for some random variables in $\M X$ to have nonvanishing variances. In fact, if all elements in $\M\Sigma^*$ are uniformly vanishing and upper bounded by some $\epsilon_n \rightarrow 0$, it is reasonable to impose $\epsilon_n$ as a factor multiplied  on the right-hand side of \eqref{A11} and \eqref{A12} for a tighter bandable covariance structure of $\M X$. Then Proposition \ref{po:boundsig} will still hold with such modification.
\par 
In conclusion, Proposition \ref{po:boundsig} implies  the bounded maximum eigenvalue condition of $\M\Sigma_1^*$, $\M \Sigma_2^*$ is actually an inherent 
%and thus, necessary condition, for our theoretical analysis of bandable covariance matrix of matrix-valued data.
\end{comment}
In conclusion, Proposition \ref{po:boundsig} implies the bounded maximum eigenvalue condition of $\M\Sigma_1^*$, $\M \Sigma_2^*$ is actually an inherent assumption if we impose bandable covariance structures on $\M\Sigma_1^*$ and $\M \Sigma_2^*$. Therefore, we add the bounded maximum eigenvalue condition in $\mathcal{F}(\varepsilon_0, \alpha)$ and $\mathcal{M}(\varepsilon_0, \alpha)$. Actually, the bounded maximum eigenvalue condition is also used in bandable covariance estimation for high-dimensional vector-valued data \citep{Bickel2008threshold,Cai2010,cai2012adaptive,cai2016estimating}. 

\begin{remark}
If one is interested in removing the bounded maximum eigenvalue assumption, some other covariance matrix models/classes can be considered instead of the bandable class.  Examples include sparse matrix \citep{huang2006covariance,Bickel2008threshold,bien2011sparse,cai2011adaptive}; graphical model \citep{meinshausen2006high,yuan2007model, friedman2008sparse,lam2009sparsistency,cai2011adaptive}; spike model \citep{Johnstone2001,paul2007asymptotics,donoho2018optimal,ding2021spiked}.
\end{remark}
\subsection{Discussion of Theorem \ref{T2}}\label{sec:discu:thm1}
Here we present some interpretation of  terms in \eqref{T2:res1}. The target error can be decomposed into two error terms via triangle inequality,
\bee\label{err:decom}
\frac{\E\|\hat{\M\Sigma}^\eta_2\kii\otimes\hat{\M\Sigma}^\eta_1\ki - \M\Sigma^*\|_{\F}^2}{pq} &\precsim \underbrace{{\frac{\E\|\hat{\M\Sigma}^\eta_2\kii\otimes\hat{\M\Sigma}^\eta_1\ki - \M\Sigma^{*,\eta}_2(k_2)\otimes \M\Sigma^{*,\eta}_1(k_1)\|_{\F}^2}{pq}}}_{E_1}
\\
&+\underbrace{{\frac{\E\|\M\Sigma^{*,\eta}_2(k_2)\otimes \M\Sigma^{*,\eta}_1(k_1) - \M\Sigma^*\|_{\F}^2}{pq}}}_{E_2} 
\ee

The error term $E_1$ is the entrywise mean squared errors of our proposed estimators, targeting at the estimation of doubly banded or tapering truth. An elegant technical Lemma (Lemma \ref{T1}) utilizing low rank matrix approximation property converts the Frobenius-norm target error term $E_1$ in \eqref{err:decom}, into a spectral-norm error term with matrices, whose components are reordered by $\xi(\cdot)$,
\bee\label{E1upper}
E_1 \precsim  \frac{\E\|\xi\{\tilde{\M \Sigma}_{\eta}(k_1,k_2)\} - \xi\{\M\Sigma^*_2(k_2)\otimes \M\Sigma^*_1(k_1)\}\|_{2}^2}{pq}.
\ee 
To bound $E_1$, we bound $\frac{\E\|\xi\{\tilde{\M \Sigma}_{\eta}(k_1,k_2)\} - \xi\{\M\Sigma^*_2(k_2)\otimes \M\Sigma^*_1(k_1)\}\|_{2}^2}{pq}$ with two types of technical arguments.
\par
\noindent\textbf{(i).} A simple property of $\xi(\cdot)$ (Lemma \ref{lemma:xi}) shows that $\E\|\xi\{\tilde{\M \Sigma}_{\eta}(k_1,k_2)\} -\xi\{\M\Sigma^*_2(k_2)\otimes \M\Sigma^*_1(k_1)\}\|_{2}^2/pq \leq \E\|\tilde{\M \Sigma}_{\eta}(k_1,k_2) - \M\Sigma^*_2(k_2)\otimes \M\Sigma^*_1(k_1)\|_{\F}^2/{pq}$. Then, based on the finite fourth moment condition implied by sub-Gaussian tailedness, one can obtain an entrywise moment bound for$\E\|\tilde{\M \Sigma}_{\eta}(k_1,k_2) - \M\Sigma^{*,\eta}_2(k_2)\otimes \M\Sigma^{*,\eta}_1(k_1)\|_{\F}^2/{pq}$, and therefore we have
\bee\label{varterm:1}
E_1\precsim \E\Big[\frac{\|\tilde{\M \Sigma}_{\eta}(k_1,k_2) - \M\Sigma^{*,\eta}_2(k_2)\otimes \M\Sigma^{*,\eta}_1(k_1)\|_{\F}^2}{pq}\Big]&\precsim \frac{k_1k_2}{n}
\\
&\equiv r^{(1)}_{\text{var}\mid \hat{\M \Sigma}_2\otimes \hat{\M\Sigma}_1}(k_1,k_2\mid p,q).
\ee
\par
\noindent{\textbf{(ii).}} An $\epsilon$-net argument addressing the complexity reduction effect of doubly banded and tapering is alternatively employed to bound $E_1$. It is incorporated with the newly-derived Hanson--Wright inequality for general sub-Gaussian random variables \citep{zajkowski2020bounds}. The inequality provides a subtle probabilistic error bound for each point on the $\varepsilon$-net, which results in another bound of $E_1$ as
\bee\label{varterm:2}
E_1 \precsim \E\Big[\frac{\|\xi\{\tilde{\M \Sigma}_{\eta}(k_1,k_2)\} - \xi\{\M\Sigma^*_2(k_2)\otimes \M\Sigma^*_1(k_1)\}\|_{2}^2}{pq}\Big] &\precsim 
\begin{cases}
\frac{k_1}{qn} + \frac{k_2}{pn}, & pk_1 + qk_2 \precsim n
\\
\frac{pk^2_1}{qn^2} + \frac{qk^2_2}{pn^2} \, & pk_1 + qk_2 \succ n
\end{cases}
\\
&\equiv r^{(2)}_{\text{var}\mid \hat{\M \Sigma}_2\otimes \hat{\M\Sigma}_1}(k_1,k_2\mid p,q).
\ee
Combining \eqref{varterm:1} and \eqref{varterm:2} and accounting for the fact that $\frac{k_1k_2}{n}\precsim \frac{k_1}{qn} + \frac{k_2}{pn}$, we finally derive the proposed error bound of $E_1$ as $E_1\precsim r_{\text{var}\mid \hat{\M \Sigma}_2\otimes \hat{\M\Sigma}_1}(k_1,k_2\mid p,q)$ in Theorem \ref{T2}, where $r_{\text{var}\mid \hat{\M \Sigma}_2\otimes \hat{\M\Sigma}_1}(k_1,k_2\mid p,q) \equiv r^{(1)}_{\text{var}\mid \hat{\M \Sigma}_2\otimes \hat{\M\Sigma}_1}(k_1,k_2\mid p,q)\wedge  r^{(2)}_{\text{var}\mid \hat{\M \Sigma}_2\otimes \hat{\M\Sigma}_1}(k_1,k_2\mid p,q)$ can be shown to have the following form,
\bee\label{def:rwhole}
r_{\text{var}\mid \hat{\M \Sigma}_2\otimes \hat{\M\Sigma}_1}(k_1,k_2\mid p,q) \equiv \begin{cases}\frac{k_1}{qn} + \frac{k_2}{pn}& pk_1 + qk_2 \precsim n
\\
\big(\frac{k_1k_2}{n}\big)\wedge\big(\frac{pk^2_1}{qn^2} + \frac{qk^2_2}{pn^2}  \big) & pk_1 + qk_2 \succ n.
\end{cases}
\ee
\par
Interestingly, both upper bounds in \eqref{varterm:1} and \eqref{varterm:2} are useful because one of them may be sharper than the other depending on different regimes of $p$ and $q$. First consider when $p,q$ diverge under an unbalanced regime like the degenerate regime ($q = 1$ and $p \rightarrow +\infty$). As shown later in Section \ref{sec:dr}, using \eqref{varterm:1} to bound the $E_1$  can lead to a minimax optimal convergence rate of the target error, while error bound in \eqref{varterm:2} is dominated by ${pk_1^2}/{n^2}$ when $p$ diverges sufficiently fast, which does not even converge to $0$ when $p\succsim n^2$. On the other hand, consider when $p,q$ diverge under the balanced regimes. As illustrated in Section \ref{sec:match:illustration} with examples, the Figures \ref{fig:sgbd1}--\ref{fig:wkbd2} (c) show the optimal convergence rate when only using \eqref{varterm:1} to bound $E_1$, while Figures \ref{fig:sgbd1}--\ref{fig:wkbd2} (b) show the optimal convergence rate after accounting for \eqref{varterm:2}. We can see when $p\approx q$, the bound in \eqref{varterm:2} can help to significantly sharpen the target error. 
\par
The error term $E_2$ could be seen as the thresholding error caused by doubly banded and tapering, and can be bounded by 
\bee\label{errorE2}
E_2 \precsim \M I_{\eta,p}(k_1)\cdot k_1^{-\tilde\alpha_1}+\M I_{\eta,q}(k_2)\cdot k_2^{-\tilde\alpha_2}
\ee 
after accounting for the off-diagonal decaying rates of entry magnitude for matrix classes $\mathcal{F}(\varepsilon_0,\alpha)$ and $\mathcal{M}(\varepsilon_0,\alpha)$.

\subsection{Optimal Bandwidth Selection in Theorem \ref{T2}}\label{optbd:T2}
Based on different divergence regimes of $(p,q)$, we aim to find optimal selection of $k_1,k_2$ to minimize the convergence rate of the following term:
\bee\label{T2:target}
r_{1}(k_1,k_2\mid p,q,\eta) \equiv \begin{cases}\frac{k_1}{qn} + \frac{k_2}{pn}+ \M I_{\eta,p}(k_1)\cdot k_1^{-\tilde{\alpha}_1} +  \M I_{\eta,q}(k_2)\cdot k_2^{-\tilde{\alpha}_2}, & pk_1 + qk_2 \precsim n
\\
\big(\frac{k_1k_2}{n}\big)\wedge\big(\frac{pk^2_1}{qn^2} + \frac{qk^2_2}{pn^2}  \big)+ \M I_{\eta,p}(k_1)\cdot k_1^{-\tilde\alpha_1} +  \M I_{\eta,q}(k_2)\cdot k_2^{-\tilde\alpha_2}, & pk_1 + qk_2 \succ n
\end{cases}.
\ee
Here 
\bee\label{t1:sup:alpha}
\small\tilde{\alpha}_{a} = \begin{cases}
2\alpha_a & \text{when }\M \Sigma_1^{*}\in \mathcal{F}(\varepsilon_0, \alpha_1),\M \Sigma_2^{*} \in \mathcal{F}(\varepsilon_0, \alpha_2)
\\
2\alpha_a + 1 & \text{when } \M \Sigma_1^{*}\in \mathcal{M}(\varepsilon_0, \alpha_1),\M \Sigma_2^{*} \in \mathcal{M}(\varepsilon_0, \alpha_2)
\end{cases},
\ee
with $a\in\{1,2\}$ and $\eta\in\big\{\mathcal{B},\mathcal{T}\big\}$. Given a divergence regime of $p,q$ and $\eta = \mathcal{B}\text{ or }\mathcal{T}$, we also define $k_{1,\text{opt}}^{(1)}, k_{2,\text{opt}}^{(1)}$ as the corresponding optimal selections of $k_1,k_2$ that give $r_1(k_1,k_2\mid p,q,\eta)$ the optimal convergence rate. 
\par
Since $r_{1}(k_1,k_2\mid p,q,\eta)$ has two phases, we further define
\bee\label{def:lu}
L(k_1,k_2\mid p,q,\eta) &\equiv \frac{k_1}{qn} + \frac{k_2}{pn}+ \M I_{\eta,p}(k_1)\cdot k_1^{-\tilde{\alpha}_1} +  \M I_{\eta,q}(k_2)\cdot k_2^{-\tilde{\alpha}_2},
\\
U(k_1,k_2\mid p,q,\eta)&\equiv \frac{pk^2_1}{qn^2} + \frac{qk^2_2}{pn^2}  + \M I_{\eta,p}(k_1)\cdot k_1^{-\tilde\alpha_1} +  \M I_{\eta,q}(k_2)\cdot k_2^{-\tilde\alpha_2}.
\ee
Additionally, we denote $k_{1,\text{opt}}^{(1),L}, k_{2,\text{opt}}^{(1),L}$ the corresponding optimal selection of $k_1,k_2$ for $L(k_1,k_2\mid p,q,n)$, and $k_{1,\text{opt}}^{(1),U}, k_{2,\text{opt}}^{(1),U}$ the corresponding optimal selection of $k_1,k_2$ for $U(k_1,k_2\mid p,q,n)$. For any $a,b,c\in\RR$, it is easy to see $(a\wedge b) + c = (a+c)\wedge (b+c)$. So we have
\bee\label{r1:high:rewrite}
&\big(\frac{k_1k_2}{n}\big)\wedge\big(\frac{pk^2_1}{qn^2} + \frac{qk^2_2}{pn^2}  \big)+ \M I_{\eta,p}(k_1)\cdot k_1^{-\tilde\alpha_1} +  \M I_{\eta,q}(k_2)\cdot k_2^{-\tilde\alpha_2}
\\
&=\Big\{\frac{k_1k_2}{n}+ \M I_{\eta,p}(k_1)\cdot k_1^{-\tilde\alpha_1} +  \M I_{\eta,q}(k_2)\cdot k_2^{-\tilde\alpha_2}\Big\}\wedge\Big\{ \frac{pk^2_1}{qn^2} + \frac{qk^2_2}{pn^2}  + \M I_{\eta,p}(k_1)\cdot k_1^{-\tilde\alpha_1} +  \M I_{\eta,q}(k_2)\cdot k_2^{-\tilde\alpha_2}\Big\}
\\
&=r_2(k_1,k_2\mid p,q,\eta) \wedge U(k_1,k_2\mid p,q,\eta).
\ee
By \eqref{def:lu} and \eqref{r1:high:rewrite}, we can rewrite $r_1(k_1,k_2\mid p, q,\eta)$ in \eqref{T2:target} as
\bee\label{r1:rewrite}
r_{1}(k_1,k_2\mid p,q,\eta) \equiv \begin{cases}L(k_1,k_2\mid p,q,\eta) & pk_1 + qk_2 \precsim n
\\
r_2(k_1,k_2\mid p,q,\eta) \wedge U(k_1,k_2\mid p,q,\eta) & pk_1 + qk_2 \succsim n
\end{cases}.
\ee
\begin{remark}
By simple algebra, one can see that $L(k_1,k_2\mid p,q,\eta) \asymp r_2(k_1,k_2\mid p,q,\eta) \wedge U(k_1,k_2\mid p,q,\eta)$ when $pk_1 + qk_2 \asymp n$. So we change the condition of the second case in \eqref{r1:rewrite} from $pk_1 + qk_2 \succ n$ to $pk_1 + qk_2 \succsim n$ for the simplicity of following discussion.
\end{remark}
The optimal convergence rate of $r_2(k_1,k_2\mid p,q,\eta)$ has been discussed in Section \ref{optbd:T3}. So we focus on deriving the optimal rate of $L(k_1,k_2\mid p,q,\eta)$ when $pk_1+qk_2\precsim n$ and the optimal rate of $U(k_1,k_2\mid p,q,\eta)$ when $pk_1 + qk_2\succ n$. By combining three optimal rates in \eqref{r1:rewrite}, we finally obtain the optimal rate of $r_1(k_1,k_2\mid p,q,\eta)$.
\par
\
\par
\noindent{\textbf{i. Optimal rate of $L(k_1,k_2\mid p,q,\eta)$ when $pk_1+qk_2\precsim n$: }}We first note that $L(k_1,k_2\mid p,q,\eta)$ can be decomposed into two parts with respect to $k_1$ and $k_2$ respectively,
\bee\label{l:split}
L(k_1,k_2\mid p,q,\eta) &=\frac{k_1}{qn} + \M I_{\eta,p}(k_1)\cdot k_1^{-\tilde{\alpha}_1} + \frac{k_2}{pn}+  \M I_{\eta,q}(k_2)\cdot k_2^{-\tilde{\alpha}_2}
\\
&=\underbrace{\frac{k_1}{qn} + \M I(k_1 <\tilde p)\cdot k_1^{-\tilde{\alpha}_1}}_{\equiv L_1(k_1\mid p,q,\eta)}+ \underbrace{\frac{k_2}{pn}+  \M I(k_2 <\tilde q)\cdot k_2^{-\tilde{\alpha}_2}}_{\equiv L_2(k_2\mid p,q,\eta)}.
\ee
Then to find the optimal convergence rate of $L(k_1,k_2\mid p,q,\eta)$ when $pk_1+qk_2\precsim n$, we can find the optimal rate of $L_1(k_1\mid p,q,\eta)$ when $pk_1\precsim n$, and the optimal rate of $L_2(k_2\mid p,q,\eta)$ when $qk_2\precsim n$, respectively. %By \eqref{l:split}, for any specific divergent rates of $p,q$, the final optimal rate of $L(k_1,k_2\mid p,q,\eta)$ is the maximum of $L_1(k_1\mid p,q,\eta)$'s optimal rate and $L_2(k_2\mid p,q,\eta)$'s optimal rate of. 
Let $k_{1,\text{opt}}^{(1),L}$ be the selected $k_1$ that gives $L_1(k_1\mid p,q,\eta)$ the optimal rate, and $k_{2,\text{opt}}^{(1),L}$ be the selected $k_2$ that gives $L_2(k_2\mid p,q,\eta)$ the optimal rate.
\par
We focus on the optimal rate of $L_1(k_1\mid p,q,\eta)$ when $pk_1\precsim n$. The optimal rate of $L_2(k_2\mid p,q,\eta)$ when $qk_2\precsim n$ can be derived similarly. Since $pk_1\precsim n$ directly implies $p\precsim n$ (by $k_1 \geq 1$), we only need to consider the following three regimes. 
\begin{itemize}
\item[\textbf{regime 1}]  ($p\precsim \sqrt{n}$)\textbf{:} Since $k_1\precsim p\precsim \sqrt{n}$, the constraint $k_1p\precsim n$ always holds when optimizing $L_1(k_1\mid p,q,n)$. 
\par
When $q$ satisfies $(qn)^{\frac{1}{\tilde\alpha_1 + 1}} \precsim p$, if we select $k_{1,\text{opt}}^{(1),L}$ less than $\tilde p$, then $L_1(k_1\mid p,q,\eta)$ becomes 
$$
\frac{k_1}{qn} + k_1^{-\tilde \alpha_1},
$$
where we can select $k_{1,\text{opt}}^{(1),L} \asymp (qn)^{\frac{1}{\tilde \alpha_1 + 1}}$, which yields $\frac{k_1}{qn}\asymp k_1^{-\tilde\alpha_1}$. Note that it is possible to select $k_{1,\text{opt}}^{(1),L}$ in such a divergence rate since $(qn)^{\frac{1}{\tilde \alpha_1 + 1}}\precsim p$. Then the optimal rate of $L_1(k_1\mid p,q,\eta)$ is  $(qn)^{\frac{1}{\tilde \alpha_1 + 1} - 1}$. If we select $k_{1,\text{opt}}^{(1),L} =\tilde p$, then $L_1(k_1\mid p,q,\eta)$ becomes $\frac{\tilde{p}}{qn}\asymp \frac{p}{qn}$, which is larger or equal to $(qn)^{\frac{1}{\tilde \alpha_1 + 1} - 1}$ under the condition of $(qn)^{\frac{1}{\tilde\alpha_1 + 1}} \precsim p$. Therefore we choose $k_{1,\text{opt}}^{(1),L} \asymp (qn)^{\frac{1}{\tilde \alpha_1 + 1}}$ and the corresponding optimal rate of $L_1(k_1\mid p,q,\eta)$ is  $(qn)^{\frac{1}{\tilde \alpha_1 + 1} - 1}$.
\par
When $q$ satisfies $(qn)^{\frac{1}{\tilde\alpha}_1 + 1}\succ p$, if $k_1<\tilde p$, the optimal rate of $L_1(k_1\mid p,q,\eta) = \frac{k_1}{qn} + k_1^{-\tilde\alpha_1} \asymp \frac{p}{qn} + p^{-\tilde\alpha_1}$ when selecting $k_1 = \tilde p - 1$. On the other hand, when selecting $k_1 = \tilde p$, then $L_1(k_1\mid p,q,\eta)$ becomes  $\frac{p}{qn}$, which is even smaller than $\frac{p}{qn} + p^{-\tilde\alpha_1}$. Therefore, when $(qn)^{\frac{1}{\tilde\alpha_1 + 1}}\succ p$, we select $k_{1,\text{opt}}^{(1),L} = \tilde p$ and the optimal rate of $L_1(k_1\mid p,q,\eta)$ is $\frac{p}{qn}$.
\par
By further noticing that $(qn)^{\frac{1}{\tilde \alpha_1 + 1} - 1} \precsim \frac{p}{qn}$ when $(qn)^{\frac{1}{\tilde\alpha_1 + 1}}\precsim p$ and vice versa, we finally obtain the optimal rate of $L_1(k_1 \mid p,q,n)$ under $p\precsim \sqrt{n}$ is $$
(qn)^{\frac{1}{\tilde\alpha_1 + 1} - 1}\wedge \frac{p}{qn},
$$
with $k_{1,\text{opt}}^{(1),L} = \begin{cases}
(qn)^{\frac{1}{\tilde \alpha_1 + 1}} & p\precsim \sqrt{n}, (qn)^{\frac{1}{\tilde\alpha_1 + 1}}\precsim p
\\
\tilde p & p\precsim \sqrt{n}, (qn)^{\frac{1}{\tilde\alpha_1 + 1}}\succ p
\end{cases}.$
\item[\textbf{regime 2}]($\sqrt{n}\prec p\precsim n$ and $p\cdot(qn)^{\frac{1}{\tilde \alpha_1 + 1}}\precsim n$)\textbf{:} Under regime 2, it is easy to see $(qn)^{\frac{1}{\tilde\alpha_1 + 1}}\precsim n/p\precsim \sqrt{n}\precsim p$. Thus $k_{1,\text{opt}}^{(1),L}\asymp (qn)^{\frac{1}{\tilde\alpha_1 + 1}}$ satisfies the two constraints $k_1 \precsim p$ and $k_1p\precsim n$. Similar arguments to regime 1 can be applied and imply that $k_{1,\text{opt}}^{(1),L}\asymp (qn)^{\frac{1}{\tilde\alpha_1 + 1}}$ gives the optimal rate of $L_1(k_1\mid p,q,n)$, which is $L_1(k_{1,\text{opt}}^{(1),L}\mid p,q,n)\asymp (qn)^{\frac{1}{\tilde\alpha_1 + 1} - 1}$.
\item[\textbf{regime 3}] ($\sqrt{n}\prec p\precsim n$ and $p\cdot(qn)^{\frac{1}{\tilde \alpha_1 + 1}}\succ n$)\textbf{:} Since $k_1p\precsim {n}$, we have 
\bee\label{rgm3:k1up}
k_1 \precsim {n}/p\prec\sqrt{n}\prec p,
\ee
 where the last two inequalities hold by $\sqrt{n}\prec p$. Thus we can not select $k_1 = \tilde p$ under regime 3. The $L_1(k_1\mid p,q,n)$ becomes
\bee\label{rgm3:main}
\frac{k_1}{qn} + k_1^{-\tilde\alpha_1}.
\ee
On the other hand, since $p\cdot(qn)^{\frac{1}{\tilde \alpha_1 + 1}}\succ n$ under regime 3, we know $k_1\prec (qn)^{\frac{1}{\tilde\alpha_1 + 1}}$ because of the  constraint $pk_1\precsim n$. Therefore \eqref{rgm3:main} is optimized when $k_1$ diverges maximally fast under the restriction of \eqref{rgm3:k1up}. By the upper bound of $k_1$ in \eqref{rgm3:k1up}, we select $k_{1,\text{opt}}^{(1),L}\asymp \frac{n}{p}$ and the corresponding optimal rate of $L_1(k_{1,\text{opt}}^{(1),L}\mid p,q,n)$ is $\big(\frac{n}{p}\big)^{-\tilde\alpha_1}$.
\end{itemize}
Summarizing the results above and by symmetry, the optimal rate of $L(k_1,k_2\mid p,q,\eta)$ when $pk_1+qk_2\precsim n$ is as follows:
\bee\label{optL}
L(k_{1,\text{opt}}^{(1),L},k_{2,\text{opt}}^{(2),L}\mid p,q,\eta)  = \max\Big\{L_1\big(k_{1,\text{opt}}^{(1),L}\mid p,q,\eta\big), L_2\big(k_{2,\text{opt}}^{(2),L}\mid p,q,\eta\big)\Big\},
\ee
where
\bee\label{optL12}
L_1&\big(k_{1,\text{opt}}^{(1),L}\mid p,q,\eta\big) \asymp 
\begin{cases}
(qn)^{\frac{1}{\tilde\alpha_1 + 1} - 1}\wedge \frac{p}{qn} & \text{when }p\precsim \sqrt{n}
\\
\frac{p}{qn} & \text{when } \sqrt{n}\prec p\precsim n \text{ and } p\cdot(qn)^{\frac{1}{\tilde \alpha_1 + 1}}\precsim n
\\
\big(\frac{n}{p}\big)^{-\tilde\alpha_1}&\text{when }\sqrt{n}\prec p\precsim n \text{ and }  p\cdot(qn)^{\frac{1}{\tilde \alpha_1 + 1}}\succ n
\end{cases},
\\
L_2&\big(k_{2,\text{opt}}^{(1),L}\mid p,q,\eta\big) \asymp 
\begin{cases}
(pn)^{\frac{1}{\tilde\alpha_2 + 1} - 1}\wedge \frac{q}{pn} & \text{when }q\precsim \sqrt{n}
\\
\frac{q}{pn} & \text{when } \sqrt{n}\prec q\precsim n \text{ and } q\cdot(pn)^{\frac{1}{\tilde \alpha_2 + 1}}\precsim n
\\
\big(\frac{n}{q}\big)^{-\tilde\alpha_2}&\text{when }\sqrt{n}\prec q\precsim n \text{ and }  q\cdot(pn)^{\frac{1}{\tilde \alpha_2 + 1}}\succ n
\end{cases}.
\ee
The corresponding $k_{1,\text{opt}}^{(1),L}, k_{2,\text{opt}}^{(1),L}$ are
\bee\label{optLk1k2}
k_{1,\text{opt}}^{(1),L} &= 
\begin{cases}
(qn)^{\frac{1}{\tilde \alpha_1 + 1}} & \text{when }p\precsim \sqrt{n}, (qn)^{\frac{1}{\tilde\alpha_1 + 1}}\precsim p \text{ or }\sqrt{n}\prec p\precsim n,p\cdot(qn)^{\frac{1}{\tilde \alpha_1 + 1}}\precsim n
\\
\tilde p & \text{when }p\precsim \sqrt{n}, (qn)^{\frac{1}{\tilde\alpha_1 + 1}}\succ p
\\
n/p &\text{when }\sqrt{n}\prec p\precsim n,p\cdot(qn)^{\frac{1}{\tilde \alpha_1 + 1}}\succ n
\end{cases},
\\
k_{2,\text{opt}}^{(1),L} &= 
\begin{cases}
(pn)^{\frac{1}{\tilde \alpha_2 + 1}} & \text{when }q\precsim \sqrt{n}, (pn)^{\frac{1}{\tilde\alpha_2 + 1}}\precsim q \text{ or }\sqrt{n}\prec q\precsim n,q\cdot(pn)^{\frac{1}{\tilde \alpha_2 + 1}}\precsim n
\\
\tilde q & \text{when }q\precsim \sqrt{n}, (pn)^{\frac{1}{\tilde\alpha_2 + 1}}\succ q
\\
n/q &\text{when }\sqrt{n}\prec q\precsim n,q\cdot(pn)^{\frac{1}{\tilde \alpha_2 + 1}}\succ n.
\end{cases}
\ee
\par
\
\par
\noindent{\textbf{ii. Optimal rate of $U(k_1,k_2\mid p,q,\eta)$ when $pk_1+qk_2\succsim n$: }}Similar to \textbf{i.}, $U(k_1,k_2\mid p,q,\eta)$ can be decomposed into two parts as follows,
\bee\nonumber
U(k_1,k_2\mid p,q,\eta)  \asymp \underbrace{\frac{pk^2_1}{qn^2} + \M I_{\eta,p}(k_1)\cdot k_1^{-\tilde\alpha_1}}_{\equiv U_1(k_1\mid p,q,\eta)} +  \underbrace{\frac{qk^2_2}{pn^2} +  \M I_{\eta,q}(k_2)\cdot k_2^{-\tilde\alpha_2}}_{\equiv U_2(k_2\mid p,q,\eta)}.
\ee
Also similar to part \textbf{i.}, the optimal rate of $U_1(k_1\mid p,q,\eta)$ is attained when $k_1 = k_{1,\text{opt}}^{(1),U}$, and the optimal rate of $U_2(k_2\mid p,q,\eta)$ is attained when $k_2 = k_{2,\text{opt}}^{(1),U}$. Due to symmetry, we only derive the optimal rate of $U_1(k_1\mid p,q,\eta)$ when $pk_1 + qk_2 \succsim n$. Note that we do not discuss the regime that $\max\{p,q\}\prec \sqrt{n}$ as it will not happen when $pk_1 + qk_2\succsim n$ by noticing that $k_1\precsim p$ and $k_2\precsim q$.
\begin{itemize}

\item[\textbf{regime 4}]  ($p\prec \sqrt{n}$)\textbf{:} When $p\prec \sqrt{n}$, we know $pk_1 \prec n$ since $k_1\precsim p$. Thus $qk_2\succsim n$ must hold; otherwise  $k_1p + k_2q\succsim n$ can not be satisfied. We discuss two scenarios of $q$: $q\precsim \frac{p^{\tilde\alpha_1 + 3}}{n^2}$ and $q\succ \frac{p^{\tilde\alpha_1 + 3}}{n^2}$. 
\par
When $q\precsim \frac{p^{\tilde\alpha_1 + 3}}{n^2}$, it is easy to see $\big(\frac{qn^2}{p}\big)^{\frac{1}{\tilde\alpha_1 + 2}}\precsim p$. Therefore, we can take $k_{1,\text{opt}}^{(1),U} \asymp \big(\frac{qn^2}{p}\big)^{\frac{1}{\tilde\alpha_1 + 2}}$ which satisfies $\frac{pk^2_1}{qn^2} \asymp k_1^{-\tilde\alpha_1}$ and get the optimal rate  $U_1(k_{1,\text{opt}}^{(1),U}\mid p,q,\eta) \asymp \big(\frac{qn^2}{p}\big)^{-\frac{\tilde \alpha_1}{\tilde\alpha_1 + 2}}$.
\par
When $q\succ \frac{p^{\tilde\alpha_1 + 3}}{n^2}$, we have $\big(\frac{qn^2}{p}\big)^{\frac{1}{\tilde\alpha_1 + 2}}\succ p$. Similar to arguments for regime 1, we take $k_{1,\text{opt}}^{(1),U} = \tilde p$ and the optimal rate  $U_1(k_{1,\text{opt}}^{(1),U}\mid p,q,\eta) \asymp \frac{p^3}{qn^2}$.

\item[\textbf{regime 5}]  ($\sqrt{n}\precsim p\prec n$)\textbf{:} We consider three scenarios based on the value of $q$: $q\prec \sqrt{n}$, $\sqrt{n}\precsim q \prec n$ and $n\precsim q$.
\par
\begin{itemize}
\item[(i)] When $q\prec \sqrt{n}$, since $k_2q\prec {n}$, we need an additional constraint that $k_1p\succsim n$ to guarantee $pk_1 + qk_2 \succsim n$. There are three possibilities. 
\begin{itemize}
\item[(1).] If $q\prec \frac{n^{\tilde\alpha_1}}{p^{\tilde\alpha_1 +1}}$, it is easy to check $p\big(\frac{qn^2}{p}\big)^{\frac{1}{\tilde\alpha_1 + 2}}\prec n$. Therefore, since $k_1 p \succsim {n}$ we can not solve $\frac{pk^2_1}{qn^2} \asymp k_1^{-\tilde\alpha_1}$ and take $k_1 \asymp \big(\frac{qn^2}{p}\big)^{\frac{1}{\tilde\alpha_1 + 2}}$ for the optimal rate of $\frac{pk^2_1}{qn^2} +  k_1^{-\tilde\alpha_1}$. It is easy to check $\frac{pk^2_1}{qn^2}$ always dominate $\frac{pk^2_1}{qn^2} +  k_1^{-\tilde\alpha_1}$ when $q\prec \frac{n^{\tilde\alpha_1}}{p^{\tilde\alpha_1 +1}}$. Therefore we take $k_{1,\text{opt}}^{(1),U} \asymp n/p$, the lowest divergence rate of $k_1$ when $k_1p\succsim n$, and $U_1(k_{1,\text{opt}}^{(1),U}\mid p,q,\eta) \asymp \frac{1}{pq}$. 
\item[(2).] If $q\succsim \frac{n^{\tilde\alpha_1}}{p^{\tilde\alpha_1 +1}}$ and $q\precsim \frac{p^{\tilde\alpha_1 + 3}}{n^2}$, we can see $p\big(\frac{qn^2}{p}\big)^{\frac{1}{\tilde\alpha_1 + 2}}\succsim n$ and $\big(\frac{qn^2}{p}\big)^{\frac{1}{\tilde\alpha_1 + 2}}\precsim p$. We solve $\frac{pk^2_1}{qn^2} \asymp k_1^{-\tilde\alpha_1}$ and can select $k_{1,\text{opt}}^{(1),U} \asymp \big(\frac{qn^2}{p}\big)^{\frac{1}{\tilde\alpha_1 + 2}}$ for the optimal rate, which is $U_1(k_{1,\text{opt}}^{(1),U}\mid p,q,\eta) \asymp \big(\frac{qn^2}{p}\big)^{-\frac{\tilde \alpha_1}{\tilde\alpha_1 + 2}}$. 
\item[(3).] If $q\succsim \frac{n^{\tilde\alpha_1}}{p^{\tilde\alpha_1 +1}}$ and $q\succ \frac{p^{\tilde\alpha_1 + 3}}{n^2}$, we can see $p\big(\frac{qn^2}{p}\big)^{\frac{1}{\tilde\alpha_1 + 2}}\succsim n$ and $\big(\frac{qn^2}{p}\big)^{\frac{1}{\tilde\alpha_1 + 2}}\succ p$. Therefore when we increase the divergence rate of $k_1$ from $1$ to $p$, we always have $k_1^{-\tilde{\alpha}_1}\succsim \frac{pk_1^2}{qn^2}$. Thus we take $k_{1,\text{opt}}^{(1),U} = \tilde p$ and the optimal rate  $U_1(k_{1,\text{opt}}^{(1),U}\mid p,q,\eta) \asymp \frac{p^3}{qn^2}$. 
\par
We further note that since $\sqrt{n}\precsim p\prec n$, we have $\frac{p^{\tilde\alpha_1 + 3}}{n^2} \succsim \frac{n^{\tilde\alpha_1}}{p^{\tilde\alpha_1 +1}}$. Thus conditions $q\succsim \frac{n^{\tilde\alpha_1}}{p^{\tilde\alpha_1 +1}}$ and $q\succ \frac{p^{\tilde\alpha_1 + 3}}{n^2}$ can be simplified to $q\succ \frac{p^{\tilde\alpha_1 + 3}}{n^2}$.
\end{itemize}
\par
\item[(ii)] When $\sqrt{n}\precsim q \prec n$, we have $\sqrt{n}\precsim p,q \prec n$ under regime $5$. We need either constraint $k_1p\succsim n$ or $k_2q\succsim n$ to guarantee $pk_1 + qk_2\succsim n$. We first derive the optimal rates  of $U(k_1,k_2\mid p,q,\eta) $, under the constraint $k_1p\succsim n$ and the constraint $k_2q\succsim n$ respectively. Then the final optimal rate of $U_1(k_1\mid p,q,\eta)$ is selected to be the smaller one between the two rates that are derived under these two constraints. Similar to previous arguments, we can show $\max\big\{\tilde U_1, \tilde U_2\big\}$ with 
\bee\label{u12}
\tilde{U}_1 &\asymp \begin{cases}
\frac{1}{pq}& q\prec \frac{n^{\tilde\alpha_1}}{p^{\tilde\alpha_1 +1}}
\\
\big(\frac{qn^2}{p}\big)^{-\frac{\tilde \alpha_1}{\tilde\alpha_1 + 2}}& q\succsim \frac{n^{\tilde\alpha_1}}{p^{\tilde\alpha_1 +1}} \text{ and } q\precsim \frac{p^{\tilde\alpha_1 + 3}}{n^2}
\\
\frac{p^3}{qn^2}& q\succ \frac{p^{\tilde\alpha_1 + 3}}{n^2}
\end{cases},
\\
\tilde{U}_2 &\asymp \begin{cases}
\big(\frac{pn^2}{q}\big)^{-\frac{\tilde \alpha_2}{\tilde\alpha_2 + 2}}& p\precsim q^{\tilde\alpha_2 + 3}/n^2
\\
\frac{q^3}{pn^2}&  p\succ q^{\tilde\alpha_2 + 3}/n^2
\end{cases},
\ee
is the optimal rate of $U(k_1,k_2\mid p,q,\eta)$ under the constraint $k_1p\succsim n$. And the corresponding selection of $k_1,k_2$ is 
\bee\label{k1:U1}
k_1 &= \begin{cases}
n/p& q\prec \frac{n^{\tilde\alpha_1}}{p^{\tilde\alpha_1 +1}}
\\
\big(\frac{qn^2}{p}\big)^{\frac{1}{\tilde\alpha_1 + 2}}& q\succsim \frac{n^{\tilde\alpha_1}}{p^{\tilde\alpha_1 +1}} \text{ and } q\precsim \frac{p^{\tilde\alpha_1 + 3}}{n^2}
\\
\tilde p & q\succ \frac{p^{\tilde\alpha_1 + 3}}{n^2}
\end{cases},
\\
k_2 &= \begin{cases}
\big(\frac{pn^2}{q}\big)^{\frac{1}{\tilde\alpha_2 + 2}}& p\precsim q^{\tilde\alpha_2 + 3}/n^2
\\
\tilde q&  p\succ q^{\tilde\alpha_2 + 3}/n^2
\end{cases}.
\ee
In addition, $\max\big\{\tilde U'_1, \tilde U'_2\big\}$ with 
\bee\label{u12'}
\tilde{U}_1' &\asymp \begin{cases}
\big(\frac{qn^2}{p}\big)^{-\frac{\tilde \alpha_1}{\tilde\alpha_1 + 2}}& q\precsim \frac{p^{\tilde\alpha_1 + 3}}{n^2}
\\
\frac{p^3}{qn^2}&  q\succ \frac{p^{\tilde\alpha_1 + 3}}{n^2}
\end{cases},
\\
\tilde{U}_2' &\asymp \begin{cases}
\frac{1}{pq}& p\prec \frac{n^{\tilde\alpha_2}}{q^{\tilde\alpha_2 +1}}
\\
\big(\frac{pn^2}{q}\big)^{-\frac{\tilde \alpha_2}{\tilde\alpha_2 + 2}}& p\succsim \frac{n^{\tilde\alpha_2}}{q^{\tilde\alpha_2 +1}} \text{ and }p\precsim q^{\tilde\alpha_2 + 3}/n^2
\\
\frac{q^3}{pn^2}& p\succ q^{\tilde\alpha_2 + 3}/n^2
\end{cases},
\ee
is the optimal rate of $U(k_1,k_2\mid p,q,\eta)$ under the constraint $k_1p\succsim n$. The corresponding selection of $k_1,k_2$ is 
\bee\label{k1:U1'}
k_1 &= \begin{cases}
\big(\frac{qn^2}{p}\big)^{\frac{1}{\tilde\alpha_1 + 2}}& q\precsim \frac{p^{\tilde\alpha_1 + 3}}{n^2}
\\
\tilde p & q\succ \frac{p^{\tilde\alpha_1 + 3}}{n^2}
\end{cases},
\\
k_2 &= \begin{cases}
n/q& p\prec \frac{n^{\tilde\alpha_2}}{p^{\tilde\alpha_2 +1}}
\\
\big(\frac{pn^2}{q}\big)^{\frac{1}{\tilde\alpha_2 + 2}}& p\prec \frac{n^{\tilde\alpha_2}}{p^{\tilde\alpha_2 +1}} \text{ and }p\precsim q^{\tilde\alpha_2 + 3}/n^2
\\
\tilde q&  p\succ q^{\tilde\alpha_2 + 3}/n^2
\end{cases}.
\ee
In summary, when  $\max\big\{\tilde U_1, \tilde U_2\big\} \precsim \max\big\{\tilde U_1', \tilde U_2'\big\}$, the optimal rate of $U_1(k_1\mid p,q,\eta)$ is $U_1(k_{1,\text{opt}}^{(1),U}\mid p,q,\eta) \asymp \tilde U_1$ with $k_{1,\text{opt}}^{(1),U}$ being the $k_1$  given in \eqref{k1:U1}. When  $\max\big\{\tilde U_1, \tilde U_2\big\} \succ \max\big\{\tilde U_1', \tilde U_2'\big\}$, the optimal rate of $U_1(k_1\mid p,q,\eta)$ is $U_1(k_{1,\text{opt}}^{(1),U}\mid p,q,\eta) \asymp \tilde U_1'$ with $k_{1,\text{opt}}^{(1),U}$ being the $k_1$ given in \eqref{k1:U1'}. 
\item[(iii)] When $q\succsim n$, the $pk_1 + qk_2 \succsim n$ naturally holds. Thus no additional lower bound constraint of $k_1$ is needed. Similar to the derivation in regime 1, we have $U_1(k_{1,\text{opt}}^{(1),U}\mid p,q,\eta) \asymp \begin{cases}
\big(\frac{qn^2}{p}\big)^{-\frac{\tilde \alpha_1}{\tilde\alpha_1 + 2}}& q\precsim \frac{p^{\tilde\alpha_1 + 3}}{n^2}\\
\frac{p^3}{qn^2}& q\succ \frac{p^{\tilde\alpha_1 + 3}}{n^2}
\end{cases}$, where $k_{1,\text{opt}}^{(1),U} = \begin{cases}
\big(\frac{qn^2}{p}\big)^{\frac{1}{\tilde\alpha_1 + 2}}& q\precsim \frac{p^{\tilde\alpha_1 + 3}}{n^2}
\\
\tilde p & q\succ \frac{p^{\tilde\alpha_1 + 3}}{n^2}
\end{cases}$.
\end{itemize}
\item[\textbf{regime 6}]  ($n\precsim p$)\textbf{:} It is easy to see $k_1p + k_2q \succsim p \succsim n$. The optimal rate and corresponding $k_{1,\text{opt}}^{(1),U}$ are hence the same with those for regime $5$ (iii).
\end{itemize}
Summarizing the results above and by symmetry, the optimal rate of $U(k_1,k_2\mid p,q,\eta)$ when $pk_1+qk_2\succsim n$ is obtained as follows:
\bee\label{Uopt}
U(k_{1,\text{opt}}^{(1),U},k_{2,\text{opt}}^{(1),U}\mid p,q,\eta)  = \max\Big\{U_1\big(k_{1,\text{opt}}^{(1),U}\mid p,q,\eta\big), U_2\big(k_{2,\text{opt}}^{(1),U}\mid p,q,\eta\big)\Big\},
\ee
where
\bee\label{Uopt12}
U_1&\big(k_{1,\text{opt}}^{(1),U}\mid p,q,\eta\big) = 
\begin{cases}
\frac{1}{pq} & \text{when }(p,q)\in \mathcal{C_1}
\\
\big(\frac{qn^2}{p}\big)^{-\frac{\tilde \alpha_1}{\tilde\alpha_1 + 2}} & \text{when } (p,q)\in\mathcal{C_2}
\\
\frac{p^3}{qn^2}&\text{when } (p,q)\in\mathcal{C_3}
\end{cases},
\\
U_2&\big(k_{2,\text{opt}}^{(1),U}\mid p,q,\eta\big) = 
\begin{cases}
\frac{1}{pq} & \text{when }(p,q)\in \mathcal{C_4}
\\
\big(\frac{pn^2}{q}\big)^{-\frac{\tilde \alpha_2}{\tilde\alpha_2 + 2}} & \text{when } (p,q)\in\mathcal{C_5}
\\
\frac{q^3}{pn^2}&\text{when } (p,q)\in\mathcal{C_6}
\end{cases}.
\ee
The corresponding $k_{1,\text{opt}}^{(1),U}, k_{2,\text{opt}}^{(1),U}$ are
\bee\label{Uoptk}
k_{1,\text{opt}}^{(1),U}& = 
\begin{cases}
n/p & \text{when }(p,q)\in \mathcal{C_1}
\\
\big(\frac{qn^2}{p}\big)^{\frac{1}{\tilde\alpha_1 + 2}}& \text{when } (p,q)\in\mathcal{C_2}
\\
\tilde p&\text{when } (p,q)\in\mathcal{C_3}
\end{cases},
\\
k_{2,\text{opt}}^{(1),U}& = 
\begin{cases}
n/q & \text{when }(p,q)\in \mathcal{C_4}
\\
\big(\frac{pn^2}{q}\big)^{\frac{1}{\tilde\alpha_2 + 2}} & \text{when } (p,q)\in\mathcal{C_5}
\\
\tilde q &\text{when } (p,q)\in\mathcal{C_6}
\end{cases}.
\ee
\par
Here $\mathcal{C_1}$--$\mathcal{C_6}$ are defined as follows,
\bee\nonumber
\mathcal{C_1} &= \Bigg\{(p,q)\mid \sqrt{n}\precsim p\prec n, q\prec \min\{\sqrt{n},\frac{n^{\tilde\alpha_1}}{p^{\tilde\alpha_1 +1}}\}; 
\\ &\max\big\{\tilde U_1, \tilde U_2\big\} \precsim \max\big\{\tilde U_1', \tilde U_2'\big\}, \sqrt{n}\precsim p \prec n, \sqrt{n}\precsim q\prec \min\big\{n,\frac{n^{\tilde\alpha_1}}{p^{\tilde\alpha_1 + 1}}\big\}\Bigg\},
\\
\mathcal{C_2} &= \Bigg\{(p,q)\mid p\prec \sqrt{n},q\precsim \frac{p^{\tilde\alpha_1 + 3}}{n^2} ; \ \sqrt{n}\precsim p\prec n, \frac{n^{\tilde\alpha_1}}{p^{\tilde\alpha_1 +1}}\precsim q\precsim \min\big\{\sqrt{n},\frac{p^{\tilde\alpha_1 + 3}}{n^2}\big\};
\\ 
&\max\big\{\tilde U_1, \tilde U_2\big\} \precsim \max\big\{\tilde U_1', \tilde U_2'\big\}, \sqrt{n}\precsim p \prec n, \max\{\sqrt{n},\frac{n^{\tilde\alpha_1}}{p^{\tilde\alpha_1 + 1}}\}\precsim q\prec \min\big\{n,\frac{p^{\tilde\alpha_1 + 3}}{n^2}\big\};
\\
&\max\big\{\tilde U_1, \tilde U_2\big\} \succ \max\big\{\tilde U_1', \tilde U_2'\big\}, \sqrt{n}\precsim p \prec n, \sqrt{n}\precsim q\prec \min\big\{n,\frac{p^{\tilde\alpha_1 + 3}}{n^2}\big\};
\\
&\sqrt{n}\precsim p \prec n, n\precsim q\precsim \frac{p^{\tilde\alpha_1 + 3}}{n^2}; \ n\precsim p , q\precsim \frac{p^{\tilde\alpha_1 + 3}}{n^2}\Bigg\},
\\
\mathcal{C_3} &= \Bigg\{(p,q)\mid p\prec \sqrt{n},q\succ \frac{p^{\tilde\alpha_1 + 3}}{n^2}; \ \sqrt{n}\precsim p\prec n, \frac{p^{\tilde\alpha_1 + 3}}{n^2}\precsim q\prec \sqrt{n};\ 
\\
&\max\big\{\tilde U_1, \tilde U_2\big\} \precsim \max\big\{\tilde U_1', \tilde U_2'\big\}, \sqrt{n}\precsim p \prec n, \max\{\sqrt{n},\frac{p^{\tilde\alpha_1 + 3}}{n^2}\}\precsim q\prec n;
\\
&\max\big\{\tilde U_1, \tilde U_2\big\} \succ \max\big\{\tilde U_1', \tilde U_2'\big\}, \sqrt{n}\precsim p \prec n,\max\{\sqrt{n},\frac{p^{\tilde\alpha_1 + 3}}{n^2}\}\precsim q\prec n;
\\
&\sqrt{n}\precsim p \prec n, \max\big\{n,\frac{p^{\tilde\alpha_1 + 3}}{n^2}\big\}\precsim q; \ n\precsim p ,  \frac{p^{\tilde\alpha_1 + 3}}{n^2} \prec q
\Bigg\},
\\
\mathcal{C_4} &= \Bigg\{(p,q)\mid \sqrt{n}\precsim q\prec n, p\prec \min\{\sqrt{n},\frac{n^{\tilde\alpha_2}}{q^{\tilde\alpha_2 +1}}\}; 
\\ &\max\big\{\tilde U_1, \tilde U_2\big\} \succ \max\big\{\tilde U_1', \tilde U_2'\big\}, \sqrt{n}\precsim q \prec n, \sqrt{n}\precsim p\prec \min\big\{n,\frac{n^{\tilde\alpha_2}}{p^{\tilde\alpha_2 + 1}}\big\}\Bigg\},
\\
\mathcal{C_5} &= \Bigg\{(p,q)\mid q\prec \sqrt{n},p\precsim q^{\tilde\alpha_2 + 3}/n^2 ; \ \sqrt{n}\precsim q\prec n, \frac{n^{\tilde\alpha_2}}{p^{\tilde\alpha_2 +1}}\precsim p\precsim \min\big\{\sqrt{n},\frac{q^{\tilde\alpha_2 + 3}}{n^2}\big\};
\\ 
&\max\big\{\tilde U_1, \tilde U_2\big\} \succ \max\big\{\tilde U_1', \tilde U_2'\big\}, \sqrt{n}\precsim q \prec n, \max\{\sqrt{n},\frac{n^{\tilde\alpha_2}}{q^{\tilde\alpha_2 + 1}}\}\precsim p\prec \min\big\{n,q^{\tilde\alpha_2 + 3}/n^2\big\};
\\
&\max\big\{\tilde U_1, \tilde U_2\big\} \precsim \max\big\{\tilde U_1', \tilde U_2'\big\}, \sqrt{n}\precsim q \prec n, \sqrt{n}\precsim p\prec \min\big\{n,q^{\tilde\alpha_2 + 3}/n^2\big\};
\\
&\sqrt{n}\precsim q \prec n, n\precsim p\precsim q^{\tilde\alpha_2 + 3}/n^2; \ n\precsim q , p\precsim q^{\tilde\alpha_2 + 3}/n^2\Bigg\},
\ee
\bee\nonumber
\mathcal{C_6} &= \Bigg\{(p,q)\mid q\prec \sqrt{n},p\succ q^{\tilde\alpha_2 + 3}/n^2; \ \sqrt{n}\precsim q\prec n, q^{\tilde\alpha_2 + 3}/n^2\precsim p\prec \sqrt{n};\ 
\\
&\max\big\{\tilde U_1, \tilde U_2\big\} \succ \max\big\{\tilde U_1', \tilde U_2'\big\}, \sqrt{n}\precsim q \prec n, \max\{\sqrt{n},q^{\tilde\alpha_2 + 3}/n^2\}\precsim p\prec n;
\\
&\max\big\{\tilde U_1, \tilde U_2\big\} \precsim \max\big\{\tilde U_1', \tilde U_2'\big\}, \sqrt{n}\precsim q \prec n,\max\{\sqrt{n},q^{\tilde\alpha_2 + 3}/n^2\}\precsim p\prec n;
\\
&\sqrt{n}\precsim q \prec n, \max\big\{n,q^{\tilde\alpha_2 + 3}/n^2\big\}\precsim p; \ n\precsim q ,  q^{\tilde\alpha_2 + 3}/n^2 \prec p
\Bigg\}.
\ee
Note that $\mathcal{C}_1$--$\mathcal{C}_3$ are parallel to $\mathcal{C}_4$--$\mathcal{C}_6$, and the definitions of $\tilde U_1,\tilde U_2,\tilde U_1',\tilde U_2'$ can be found in \eqref{u12} and \eqref{u12'}.
\par
\
\par
\noindent{\textbf{iii. Optimal rate of $r_{1}(k_1,k_2\mid p,q,\eta)$: }}We consider the optimal rate of $r_{1}(k_1,k_2\mid p,q,\eta)$ under three scenarios: (1)  $\max\{p,q\}\prec \sqrt{n}$; (2)  Either $p\succsim \sqrt{n}$ or $q\succsim \sqrt{n}$, and $\max\{p,q\}\prec n$; (3) Either $p\succsim n$ or $q\succsim n$.
\par
When $\max\{p,q\}\prec \sqrt{n}$, since $k_1\precsim \sqrt{n}$ and $k_2\precsim \sqrt{n}$, then $pk_1+qk_2 \prec n$ always hold. By \eqref{r1:rewrite}, we have $r_{1}(k_1,k_2\mid p,q,\eta) \asymp L(k_1,k_2\mid p,q,\eta)$ when $\max\{p,q\}\prec\sqrt{n}$. Therefore, the optimal rate of $r_{1}(k_1,k_2\mid p,q,\eta)$ is $L(k^{(1),L}_{1,\text{opt}},k^{(1),L}_{2,\text{opt}}\mid p,q,\eta)$ shown in \eqref{optL}--\eqref{optL12} based on different regimes of $p,q$. The corresponding $k^{(1)}_{1,\text{opt}}, k^{(1)}_{2,\text{opt}}$ are $k^{(1),L}_{1,\text{opt}}, k^{(1),L}_{2,\text{opt}}$ shown in \eqref{optLk1k2}.
\par
When either $p\succsim \sqrt{n}$ or $q\succsim \sqrt{n}$, and $\max\{p,q\}\prec n$, both $k_1p + k_2q \precsim n$ and $k_1p + k_2q \succsim  n$ can be satisfied with appropriate choices of $k_1$ and $k_2$. By \eqref{r1:rewrite}, the optimal rate of $r_1(k_1,k_2\mid p,q,n)$ is $L(k^{(1),L}_{1,\text{opt}},k^{(1),L}_{2,\text{opt}}\mid p,q,\eta)$  when choosing $k_1,k_2$ under the constraint $k_1p + k_2q \precsim n$, and the optimal rate is $r_2(k^{(2)}_{1,\text{opt}},k^{(2)}_{2,\text{opt}}\mid p,q,\eta) \wedge U(k^{(1),U}_{1,\text{opt}},k^{(1),U}_{2,\text{opt}}\mid p,q,\eta) $ when choosing $k_1,k_2$ under the constraint that $k_1p + k_2q \succsim n$. Here $r_2(k^{(2)}_{1,\text{opt}},k^{(2)}_{2,\text{opt}}\mid p,q,\eta)$ is defined in \eqref{r2:opt}--\eqref{r2:kopt} and $U(k^{(1),U}_{1,\text{opt}},k^{(1),U}_{2,\text{opt}}\mid p,q,\eta)$ is defined in \eqref{Uopt}--\eqref{Uoptk}. So the final optimal rate of $r_1(k_1,k_2\mid p,q,n)$ is 
\bee\label{optr1(2)}
L(k^{(1),L}_{1,\text{opt}},k^{(1),L}_{2,\text{opt}}\mid p,q,\eta)\wedge r_2(k^{(2)}_{1,\text{opt}},k^{(2)}_{2,\text{opt}}\mid p,q,\eta) \wedge U(k^{(1),U}_{1,\text{opt}},k^{(1),U}_{2,\text{opt}}\mid p,q,\eta),
\ee
and the optimal selection of $(k_1,k_2)$ is $(k^{(1),L}_{1,\text{opt}},k^{(1),L}_{2,\text{opt}})$ or $(k^{(2)}_{1,\text{opt}},k^{(2)}_{2,\text{opt}})$ or $(k^{(1),U}_{1,\text{opt}},k^{(1),U}_{2,\text{opt}})$, depending on which term in \eqref{optr1(2)} has the fastest convergence rate.
\par
When $p\succsim n$ or $q\succsim n$, the $pk_1 + qk_2 \succsim n$ always hold. Thus similar to previous arguments, the optimal rate of $r_1(k_1,k_2\mid p,q,\eta)$ is $r_2(k^{(2)}_{1,\text{opt}},k^{(2)}_{2,\text{opt}}\mid p,q,\eta) \wedge U(k^{(1),U}_{1,\text{opt}},k^{(1),U}_{2,\text{opt}}\mid p,q,\eta)$ by \eqref{r1:rewrite}. The optimal selection of $(k_1,k_2)$ is $(k^{(2)}_{1,\text{opt}},k^{(2)}_{2,\text{opt}})$ or $(k^{(1),U}_{1,\text{opt}},k^{(1),U}_{2,\text{opt}})$, depending on whether $r_2(k^{(2)}_{1,\text{opt}},k^{(2)}_{2,\text{opt}}\mid p,q,\eta)$ or $U(k^{(1),U}_{1,\text{opt}},k^{(1),U}_{2,\text{opt}}\mid p,q,\eta)$ has a faster convergence rate.
\par
To be more specific, the optimal rate of $r_1(k_1,k_2\mid p,q,\eta)$ and corresponding selection of $k_1,k_2$ is as follows:
\bee\label{r1:opt}
r&_1(k_{1,\text{opt}}^{(1)},k_{2,\text{opt}}^{(1)}\mid p,q,\eta) 
\\
&\asymp 
\begin{cases}
L(k^{(1),L}_{1,\text{opt}},k^{(1),L}_{2,\text{opt}}\mid p,q,\eta) \quad \quad \quad \quad \quad \quad \quad \quad \quad \quad \quad \quad\text{when } \max\{p,q\}\precsim \sqrt{n},
\\
L(k^{(1),L}_{1,\text{opt}},k^{(1),L}_{2,\text{opt}}\mid p,q,\eta)
\wedge r_2(k^{(2)}_{1,\text{opt}},k^{(2)}_{2,\text{opt}}\mid p,q,\eta) \wedge U(k^{(1),U}_{1,\text{opt}},k^{(1),U}_{2,\text{opt}}\mid p,q,\eta)
\\
\quad \quad \quad \quad \quad \quad \quad \quad \quad \quad \quad \quad\quad \quad \quad \quad \quad \quad \quad \quad \quad  \ \ \text{when }p\text{ or } q\succsim \sqrt{n} \text{, and }\max\{p,q\}\prec n,
\\
r_2(k^{(2)}_{1,\text{opt}},k^{(2)}_{2,\text{opt}}\mid p,q,\eta) \wedge U(k^{(1),U}_{1,\text{opt}},k^{(1),U}_{2,\text{opt}}\mid p,q,\eta) \quad \text{when } p\text{ or }q\succsim n,
\end{cases}
\ee
\bee\nonumber
(&k_{1,\text{opt}}^{(1)},k_{2,\text{opt}}^{(1)})
\\
&=\begin{cases}
\big(k^{(1),L}_{1,\text{opt}},k^{(1),L}_{2,\text{opt}}\big) \quad \text{when }\max\{p,q\}\precsim \sqrt{n},
\\
\big(k^{(1),L}_{1,\text{opt}},k^{(1),L}_{2,\text{opt}}\big)\quad \text{when }
p\text{ or } q\succsim \sqrt{n}, \text{ and }\ \max\{p,q\}\prec n, 
\\
\quad  \quad  \quad  \quad  \quad  \quad  \ \ \ \text{and }L(k^{(1),L}_{1,\text{opt}},k^{(1),L}_{2,\text{opt}}\mid p,q,\eta)\precsim r_2(k^{(2)}_{1,\text{opt}},k^{(2)}_{2,\text{opt}}\mid p,q,\eta) \wedge U(k^{(1),U}_{1,\text{opt}},k^{(1),U}_{2,\text{opt}}\mid p,q,\eta),
\\
\big(k^{(2)}_{1,\text{opt}},k^{(2)}_{2,\text{opt}}\big)\quad\text{when } p\text{ or } q\succsim \sqrt{n},\text{ and }\ \max\{p,q\}\prec n,
\\
\quad  \quad  \quad  \quad  \quad  \quad  \ \ \ \text{and }r_2(k^{(2)}_{1,\text{opt}},k^{(2)}_{2,\text{opt}}\mid p,q,\eta)\precsim  L(k^{(1),L}_{1,\text{opt}},k^{(1),L}_{2,\text{opt}}\mid p,q,\eta)\wedge U(k^{(1),U}_{1,\text{opt}},k^{(1),U}_{2,\text{opt}}\mid p,q,\eta),
\\
\big(k^{(1),U}_{1,\text{opt}},k^{(1),U}_{2,\text{opt}}\big) \quad \text{when }
p\text{ or } q\succsim \sqrt{n},\text{ and }\ \max\{p,q\}\prec n,
\\
\quad  \quad  \quad  \quad  \quad  \quad  \ \ \ \text{and }U(k^{(1),U}_{1,\text{opt}},k^{(1),U}_{2,\text{opt}}\mid p,q,\eta)\precsim r_2(k^{(2)}_{1,\text{opt}},k^{(2)}_{2,\text{opt}}\mid p,q,\eta) \wedge L(k^{(1),L}_{1,\text{opt}},k^{(1),L}_{2,\text{opt}}\mid p,q,\eta),
\\
\big(k^{(2)}_{1,\text{opt}},k^{(2)}_{2,\text{opt}}\big) \quad \text{when } p\text{ or }q\succsim n, \text{ and } r_2(k^{(2)}_{1,\text{opt}},k^{(2)}_{2,\text{opt}}\mid p,q,\eta)\precsim U(k^{(1),U}_{1,\text{opt}},k^{(1),U}_{2,\text{opt}}\mid p,q,\eta),
\\
\big(k^{(1),U}_{1,\text{opt}},k^{(1),U}_{2,\text{opt}}\big)\quad \text{when } p\text{ or }q\succsim n, \text{ and }U(k^{(1),U}_{1,\text{opt}},k^{(1),U}_{2,\text{opt}}\mid p,q,\eta)\precsim r_2(k^{(2)}_{1,\text{opt}},k^{(2)}_{2,\text{opt}}\mid p,q,\eta).
\end{cases}
\ee
Here $L(k^{(1),L}_{1,\text{opt}},k^{(1),L}_{2,\text{opt}}\mid p,q,\eta)$, $k^{(1),L}_{1,\text{opt}},k^{(1),L}_{2,\text{opt}}$ are defined in \eqref{optL}-\eqref{optLk1k2}, $r_2(k^{(2)}_{1,\text{opt}},k^{(2)}_{2,\text{opt}}\mid p,q,\eta)$, $k^{(2)}_{1,\text{opt}},k^{(2)}_{2,\text{opt}}$ are defined in \eqref{r2:opt}-\eqref{r2:kopt}, and $U(k^{(1),U}_{1,\text{opt}},k^{(1),U}_{2,\text{opt}}\mid p,q,\eta)$, $k^{(1),U}_{1,\text{opt}}, k^{(1),U}_{2,\text{opt}}$ are defined in \eqref{Uopt}-\eqref{Uoptk}.
\begin{remark}\label{rm:mhdr}
Consider the case when $\M\Sigma_1^*$ and $\M\Sigma_2^*$ are in the matrix class $\mathcal{M}(\varepsilon_0,\alpha)$, and therefore $\tilde\alpha_1 = 2\alpha_1 + 1$ and $\tilde\alpha_2 = 2\alpha_2 + 1$. Based on the derived results, we derive regimes of $p,q$ under which the convergence rate satisfies
\bee\label{r1low}
r_1(k_{1,\text{opt}}^{(1)},k_{2,\text{opt}}^{(1)}\mid p,q,\eta)&\asymp (qn)^{\frac{1}{\tilde\alpha_1 + 1} - 1}\wedge \frac{p}{qn} + (pn)^{\frac{1}{\tilde\alpha_2 + 1} - 1}\wedge \frac{q}{pn}
\\
&=(qn)^{\frac{1}{2\alpha_1 + 2} - 1}\wedge \frac{p}{qn} + (pn)^{\frac{1}{2\alpha_2 + 2} - 1}\wedge \frac{q}{pn},
\ee 
when $p,q$ are not under the degenerate regime, i.e., $\min\{p,q\}\rightarrow +\infty$ when $n\rightarrow +\infty$. Based on the lower bound in Theorem \ref{T:low}, when \eqref{r1low} is satisfied, the rate of our proposed estimator is minimax rate-optimal. In addition, we derive the corresponding optimal $k_{1,\text{opt}}^{(1)},k_{2,\text{opt}}^{(1)}$.
\par
We only consider $\max\{p,q\}\precsim n$, and note that we can not find a non-degenerate while minimax-optimal regimes, when $\max\{p,q\}\succ n$. In \eqref{r1:opt} we have shown,
\bee\nonumber
r_1(k_{1,\text{opt}}^{(1)},k_{2,\text{opt}}^{(1)}\mid p,q,\eta) &\asymp L(k^{(1),L}_{1,\text{opt}},k^{(1),L}_{2,\text{opt}}\mid p,q,\eta)
\\
&=\max\Big\{L_1\big(k_{1,\text{opt}}^{(1),L}\mid p,q,\eta\big), L_2\big(k_{2,\text{opt}}^{(2),L}\mid p,q,\eta\big)\Big\}.
\ee
\par
When $p\precsim \sqrt{n}$, by \eqref{optL12}, we have $L_1\big(k_{1,\text{opt}}^{(1),L}\mid p,q,\eta\big) \asymp (qn)^{\frac{1}{2\alpha_1 + 2} - 1}\wedge \frac{p}{qn}$.
\par
 When both $\sqrt{n}\prec p\precsim n$ and $p\cdot(qn)^{\frac{1}{2\alpha_1 + 2}}\precsim n$ hold, by \eqref{optL12}, we have $L_1\big(k_{1,\text{opt}}^{(1),L}\mid p,q,\eta\big)\asymp \frac{p}{qn}$. Note that  $p\cdot(qn)^{\frac{1}{2\alpha_1 + 2}}\precsim n$ implies $(qn)^{\frac{1}{2\alpha_1 + 2}-1}\precsim \frac{1}{pq}$, and $\sqrt{n}\prec p\precsim n$ implies $n/p \prec \sqrt{n}\prec p$, we have
\bee\nonumber
(qn)^{\frac{1}{2\alpha_1 + 2}-1}\precsim \frac{1}{pq}\prec \frac{1}{(n/p)q} = \frac{p}{qn},
\ee
when both $\sqrt{n}\prec p\precsim n$ and $p\cdot(qn)^{\frac{1}{2\alpha_1 + 2}}\precsim n$ hold. Therefore $L_1\big(k_{1,\text{opt}}^{(1),L}\mid p,q,\eta\big)\asymp \frac{p}{qn} \asymp (qn)^{\frac{1}{2\alpha_1 + 2} - 1}\wedge \frac{p}{qn}$ when both $\sqrt{n}\prec p\precsim n$ and $p\cdot(qn)^{\frac{1}{2 \alpha_1 + 2}}\precsim n$ hold. 
\par
Since $p\cdot(qn)^{\frac{1}{2 \alpha_1 + 2}}\precsim n$ can be rewritten as $p\precsim n\cdot (qn)^{-\frac{1}{2\alpha_1 + 2}}$. Based on the above discussions, $L_1\big(k_{1,\text{opt}}^{(1),L}\mid p,q,\eta\big)\asymp  (qn)^{\frac{1}{2\alpha_1 + 2} - 1}\wedge \frac{p}{qn}$ when $1\precsim p\precsim\max\big\{n\cdot (qn)^{-\frac{1}{2\alpha_1 + 2}},\sqrt{n}\big\}$. Similarly, we also have $L_2\big(k_{2,\text{opt}}^{(1),L}\mid p,q,\eta\big)\asymp  (pn)^{\frac{1}{2\alpha_2 + 2} - 1}\wedge \frac{q}{pn}$ when $1\precsim q\precsim\max\big\{n\cdot (pn)^{-\frac{1}{2\alpha_2 + 2}},\sqrt{n}\big\}$. Finally we obtain
$$
r_1(k_{1,\text{opt}}^{(1)},k_{2,\text{opt}}^{(1)}\mid p,q,\eta) \asymp (qn)^{\frac{1}{2\alpha_1 + 2} - 1}\wedge \frac{p}{qn} + (pn)^{\frac{1}{2\alpha_2 + 2} - 1}\wedge \frac{q}{pn},
$$ and therefore our proposed estimator is minimax rate-optimal when 
\bee\label{minimaxregion}
\text{(i) }&1\precsim p\precsim\max\big\{n\cdot (qn)^{-\frac{1}{2\alpha_1 + 2}},\sqrt{n}\big\}
\\
\text{(ii) }&1\precsim q\precsim\max\big\{n\cdot (pn)^{-\frac{1}{2\alpha_2 + 2}},\sqrt{n}\big\},
\ee
hold simultaneously. 
\par
Since $p\cdot(qn)^{\frac{1}{2 \alpha_1 + 2}}\precsim n$ and $q\cdot(pn)^{\frac{1}{2 \alpha_2 + 2}}\precsim n$ always hold under \eqref{minimaxregion}, based on the optimal $k_1,k_2$ selection rule \eqref{optLk1k2} for $L(k_1,k_2\mid p,q,\eta)$, the corresponding optimal $k_1,k_2$ under \eqref{minimaxregion} are
\bee\nonumber
&k_{1,\text{opt}}^{} = 
\begin{cases}
\tilde{p} & \text{when }p\precsim \sqrt{n}\text{ and } p\prec (nq)^{\frac{1}{2\alpha_1 + 2}} 
\\
(nq)^{\frac{1}{2 \alpha_1 + 2}} & \text{when }p\precsim \sqrt{n}\text{ and } (nq)^{\frac{1}{2\alpha_1 + 2}}\precsim p, \text{ or }\sqrt{n}\prec p,
\end{cases}
\\
&k_{2,\text{opt}}^{} = 
\begin{cases}
\tilde{q} & \text{when }q\precsim \sqrt{n}\text{ and } q\prec (np)^{\frac{1}{2{\alpha}_2 + 2}} 
\\
(np)^{\frac{1}{2 \alpha_2 + 2}} & \text{when }q\precsim \sqrt{n}\text{ and } (np)^{\frac{1}{2\alpha_2 + 2}}\precsim q, \text{ or }\sqrt{n}\prec q.
\end{cases}
\ee
Note that when $p\asymp (nq)^{\frac{1}{2\alpha_1 + 2}}$, we have $\tilde p \asymp (nq)^{\frac{1}{2\alpha_1 + 2}}$. Consequently, when $p\precsim \sqrt{n}$ and $p\asymp (nq)^{\frac{1}{2\alpha_1 + 2}}$, the two phases of $k_{1,\text{opt}}$ above actually lead to the same rate. Therefore, we can rewrite the selection of $k_{1,\text{opt}}$ as 
$$
k_{1,\text{opt}}^{} = 
\begin{cases}
\tilde{p} & \text{when }p\precsim \sqrt{n}\text{ and } p\precsim (nq)^{\frac{1}{2{\alpha}_1 + 2}} 
\\
(nq)^{\frac{1}{2 \alpha_1 + 2}} & \text{when }p\precsim \sqrt{n}\text{ and } (nq)^{\frac{1}{2\alpha_1 + 2}}\prec p, \text{ or }\sqrt{n}\prec p.
\end{cases}
$$Then we can simplify the above conditions of $k_{1,\text{opt}}$ to
\bee\nonumber
k_{1,\text{opt}}^{} = 
\begin{cases}\tilde{p} & \text{when }p\precsim \sqrt{n}\wedge (nq)^{\frac{1}{2\alpha_1 + 2}}
\\
(nq)^{\frac{1}{2 \alpha_1 + 2}} & \text{when }p\succ \sqrt{n}\wedge (nq)^{\frac{1}{2\alpha_1 + 2}}.
\end{cases}
\ee
And similarly for $k_{2,\text{opt}}$,
\bee
k_{2,\text{opt}}^{} = 
\begin{cases}\tilde{q} & \text{when }q\precsim \sqrt{n}\wedge (np)^{\frac{1}{2\alpha_2 + 2}}
\\
(np)^{\frac{1}{2 \alpha_2 + 2}} & \text{when }q\succ \sqrt{n}\wedge (np)^{\frac{1}{2\alpha_2 + 2}}.
\end{cases}
\ee
\end{remark} 

\subsection{Discussion of Theorem \ref{T3}}
The interpretation of error bounds in Theorem \ref{T3} is very similar to that of Theorem \ref{T2}. Comparing the rates in these two theorems, the only difference is that, under Scenario (ii), we cannot bound the error term $E_1$ of the target error through \eqref{varterm:2}, because the sub-Gaussian tail assumption is not satisfied under Scenario (ii). Other error terms can be decomposed and interpreted in the same way as the ones in Section \ref{sec:discu:thm1}.
\begin{remark}\label{rk:T3dg}
Since $\frac{k_1}{qn} + \frac{k_2}{pn}\precsim \frac{\max\{k_1,k_2\}}{n}\precsim \frac{k_1k_2}{n}$, the rate \eqref{T2:res1} in Theorem \ref{T2} can be rewritten as
\bee\nonumber
\Big(\frac{k_1k_2}{n}\Big)\wedge r(k_1,k_2,p,q) +  \M I_{\eta,p}(k_1)\cdot k_1^{-\tilde{\alpha}_1} +  \M I_{\eta,q}(k_2)\cdot k_2^{-\tilde{\alpha}_2},
\ee
where $r(k_1,k_2,p,q) = 
\begin{cases}
\frac{k_1}{qn} + \frac{k_2}{pn}, & pk_1 + qk_2 \precsim n
\\
\frac{pk^2_1}{qn^2} + \frac{qk^2_2}{pn^2}, & pk_1 + qk_2 \succsim n
\end{cases}$. It is easy to see $\frac{k_1k_2}{n} \succsim  \big(\frac{k_1k_2}{n}\big)\wedge r(k_1,k_2,p,q)$ for any regimes of $p,q,k_1,k_2$. Therefore, the rate \eqref{T2:res1} of Theorem \ref{T2} is always faster or equal to the rate \eqref{T3:res1} of Theorem \ref{T3}. This is reasonable because Scenario (i) is more restrictive than Scenario (ii).
\end{remark}

 \subsection{Optimal Bandwidth Selection in Theorem \ref{T3}}\label{optbd:T3}
Based on different divergence regimes of $(p,q)$, we aim to find optimal selection of $k_1,k_2$ to minimize the following convergence rate:
\bee\label{T2:target3}
r_{2}(k_1,k_2\mid p,q,\eta)  = \frac{k_1k_2}{n}+\M I_{\eta,p}(k_1)\cdot k_1^{-\tilde{\alpha}_1} +  \M I_{\eta,q}(k_2)\cdot k_2^{-\tilde{\alpha}_2}.
\ee
Here $\tilde{\alpha}_a$
 is equal to \eqref{t1:sup:alpha}, for $a\in\{1,2\}$ and $\eta\in\big\{\mathcal{B},\mathcal{T}\big\}$. Given a divergence regime of $p,q$ and $\eta = \mathcal{B}\text{ or }\mathcal{T}$, we also define $k_{1,\text{opt}}^{(2)}, k_{2,\text{opt}}^{(2)}$ as the corresponding optimal selection of $k_1,k_2$ that give $r_2(k_1,k_2\mid p,q,\eta)$ the optimal convergence rate.

For simplicity, we define 
\bee\nonumber
\tilde{d} = \begin{cases}
d - 1, & \eta = \mathcal{B};
\\
2d - 1, & \eta = \mathcal{T},
\end{cases}
\ee
for $d\in\{p,q\}$. It is easy to see $ \M I_{\eta,d}(k) = \M I(k<\tilde{d})$ for any $d\in\{p,q\}$ and $k>0$.
\\
\par
\noindent{\textbf{Scenario 1 (both $p,q$ diverge fast):}} When $p,q$ are both {\color{black} diverging sufficiently fast}, such that 
\bee\label{T1:cd1:1}
p \succ {n}^{\tilde\alpha_2/(\tilde\alpha_1+\tilde\alpha_2+\tilde\alpha_1\tilde\alpha_2)},~\text{and}~ q\succ {n}^{\tilde\alpha_1/(\tilde\alpha_1+\tilde\alpha_2+\tilde\alpha_1\tilde\alpha_2)},
\ee
we claim that the optimal  rate of \eqref{T2:target3}'s right-hand side is attained when $k_1 <\tilde{p}, k_2 <\tilde{q}$ and both 
\bee\label{T1:cd1:2}
&\frac{k_1k_2}{n}  \asymp k_{1}^{-\tilde\alpha_1};\quad \frac{k_1k_2}{n} \asymp k_2^{-\tilde\alpha_2},
\ee
hold. 
\par
To see this, first, one can easily check that $\frac{\tilde p\tilde q}{n}\asymp \frac{pq}{n} \succ \max\{ p^{-\tilde\alpha_1},q^{-\tilde\alpha_2}\}$ under \eqref{T1:cd1:1} by simple algebra. Thus, compared to the case when $k_1=\tilde p , k_2 =\tilde q $, the right-hand side of \eqref{T2:target3} decays faster when $k_1\prec p, k_2 \prec q$. Second, when $k_1\prec p, k_2 \prec q$, the right-hand side of \eqref{T2:target3} becomes 
\bee\label{optk1k2:target:1}
\frac{k_1k_2}{n}+\M I_{\eta,p}(k_1)\cdot k_1^{-\tilde{\alpha}_1} +  \M I_{\eta,q}(k_2)\cdot k_2^{-\tilde{\alpha}_2} &\asymp \frac{k_1k_2}{n}+ k_1^{-\tilde{\alpha}_1} +  k_2^{-\tilde{\alpha}_2}
\\
&\asymp\max\Big\{\frac{k_1k_2}{n},k_{1}^{-\tilde\alpha_1},k_2^{-\tilde\alpha_2}\Big\}.
\ee 
 And the above rate is minimized, as long as the three terms on the right-hand side diverge at the same asymptotic order, i.e., \eqref{T1:cd1:2} is satisfied. By simple algebra, \eqref{T1:cd1:2} is equivalent to \bee\label{T1:cd1:3}
k_{1,\text{opt}}^{(2)} \asymp {n}^{\tilde\alpha_2/(\tilde\alpha_1+\tilde\alpha_2+\tilde\alpha_1\tilde\alpha_2)},k_{2,\text{opt}}^{(2)} \asymp {n}^{\tilde\alpha_1/(\tilde\alpha_1+\tilde\alpha_2+\tilde\alpha_1\tilde\alpha_2)},
\ee
which is attainable for $k_1, k_2$ under the condition of regime 1, i.e., $p \succ {n}^{\tilde\alpha_2/(\tilde\alpha_1+\tilde\alpha_2+\tilde\alpha_1\tilde\alpha_2)},~\text{and}~ q\succ {n}^{\tilde\alpha_1/(\tilde\alpha_1+\tilde\alpha_2+\tilde\alpha_1\tilde\alpha_2)}$. The corresponding convergence rate is $r_{2}(k_{1,\text{opt}}^{(2)},k_{2,\text{opt}}^{(2)}\mid p,q,\eta) \asymp {n}^{-\frac{\tilde\alpha_1\tilde\alpha_2}{\tilde\alpha_1 + \tilde\alpha_2 + \tilde\alpha_1\tilde\alpha_2}}$.
\\
\par
\noindent{\textbf{Scenario 2 (one of $p,q$ diverges slowly):} }By symmetry of $p, q$, we only consider when $p$ can not diverge as fast as ${n}^{\tilde\alpha_2/(\tilde\alpha_1+\tilde\alpha_2+\tilde\alpha_1\tilde\alpha_2)}$ while $q$ can diverge faster than $(n/p)^{\frac{1}{\tilde{\alpha}_2 + 1}}$. In other words, 
\bee\label{T1:cd2:con}
 p \precsim {n}^{\tilde\alpha_2/(\tilde\alpha_1+\tilde\alpha_2+\tilde\alpha_1\tilde\alpha_2)}, q\succ \big({n}/p\big)^{\frac{1}{\tilde\alpha_2 + 1}}.
\ee
Note that under \eqref{T1:cd2:con}, we can further {\color{black} obtain} $q\succ \big({n}/p\big)^{\frac{1}{\tilde\alpha_2 + 1}} \succsim {n}^{\tilde\alpha_1/(\tilde\alpha_1+\tilde\alpha_2+\tilde\alpha_1\tilde\alpha_2)}$.
\par
If $k_1 < \tilde p$, we note the $k_1^{-\tilde\alpha_1}$ term always exists on the right-hand side of \eqref{T2:target3}. The target rate we optimize under this scenario is 
\bee\label{k1k2pick:s2:target}
r_{2}(k_1,k_2\mid p,q,\eta) \asymp\frac{k_1k_2}{n}+k_1^{-\tilde{\alpha}_1} +  \M I_{\eta,q}(k_2)\cdot k_2^{-\tilde{\alpha}_2}.
\ee
Since $p \precsim {n}^{\tilde\alpha_2/(\tilde\alpha_1+\tilde\alpha_2+\tilde\alpha_1\tilde\alpha_2)}$, we have $k_1^{^{-\tilde\alpha_1}} \succsim p^{^{-\tilde\alpha_1}} \succsim \big({n}^{\tilde\alpha_2/(\tilde\alpha_1+\tilde\alpha_2+\tilde\alpha_1\tilde\alpha_2)}\big)^{-\tilde{\alpha}_1}\asymp {n}^{-\tilde\alpha_1\tilde\alpha_2/(\tilde\alpha_1 + \tilde\alpha_2 + \tilde\alpha_1\tilde\alpha_2)}$. So by \eqref{k1k2pick:s2:target}, the optimal rate we can attain when $k_1 < \tilde p$ is no better than $k_1^{-\tilde{\alpha}_1}$ while $k_1^{-\tilde{\alpha}_1}\succsim {n}^{-\tilde\alpha_1\tilde\alpha_2/(\tilde\alpha_1 + \tilde\alpha_2 + \tilde\alpha_1\tilde\alpha_2)}$. 
\par
On the other hand, if $k_1 = \tilde p$, the right-hand side of \eqref{T2:target3} becomes $pk_2\cdot\frac{1}{n}  + k_2^{-\tilde\alpha_2} \M I(k_2 < \tilde q)$. By simple algebra, one can show that when $q \asymp \tilde q$ satisfies \eqref{T1:cd2:con}, the optimal rate is attained when $pk_2\cdot\frac{1}{n} \asymp k_2^{-\tilde\alpha_2}$, i.e., $k_2 \asymp (n/p)^{\frac{1}{\tilde\alpha_2 + 1}} \prec q$. And then
\bee\label{T1:cd2:2}
r_{2}(k_1,k_2\mid p,q,\eta)\asymp (n/p)^{-\frac{\tilde\alpha_2}{\tilde\alpha_2 + 1}}.
\ee 
Moreover, since $p \precsim n^{\tilde\alpha_2/(\tilde\alpha_1+\tilde\alpha_2+\tilde\alpha_1\tilde\alpha_2)}$, we can show 
\bee\nonumber
(n/p)^{-\frac{\tilde\alpha_2}{\tilde\alpha_2 + 1}}\precsim n^{-\tilde\alpha_1\tilde\alpha_2/(\tilde\alpha_1 + \tilde\alpha_2 + \tilde\alpha_1\tilde\alpha_2)},
\ee
which implies the rate we get when setting $k_1 = \tilde p$ is no worse than the rate when $k_1 < \tilde p$. Thus the optimal rate is attained by \eqref{T1:cd2:2} under Scenario 2, after choosing $k_{1,\text{opt}}^{(2)} = \tilde p$ and $ k_{2,\text{opt}}^{(2)} \asymp (n/p)^{\frac{1}{\tilde{\alpha}_2 + 1}}$.
\\
\par
\noindent{\textbf{Scenario 3 (both $p,q$ diverge slowly):}} By symmetry of $p, q$, we only consider when $p$ can not diverge faster than ${n}^{\tilde\alpha_2/(\tilde\alpha_1+\tilde\alpha_2+\tilde\alpha_1\tilde\alpha_2)}$ and at the same time $q$ can not diverge faster than $(n/p)^{\frac{1}{\tilde{\alpha}_2 + 1}}$.  In other words, 
\bee \label{T1:cd3:con}
p \precsim n^{\tilde\alpha_2/(\tilde\alpha_1+\tilde\alpha_2+\tilde\alpha_1\tilde\alpha_2)}, q\precsim \big(n/p\big)^{\frac{1}{\tilde\alpha_2 + 1}}
\ee
\par
If $k_1 = \tilde{p}$, similar to  Scenario 2, we want to optimize {\color{black}the following term with respect to $k_2$}:
\bee\label{T1:cd3:con:2}
r_{2}(k_1,k_2\mid p,q,\eta)\asymp pk_2\cdot\frac{1}{n}  + k_2^{-\tilde\alpha_2} \M I(k_2 < \tilde q).
\ee However, by \eqref{T1:cd3:con}, we can easily show that the term $k_2^{-\tilde\alpha_2}$ always dominates the terms in \eqref{T1:cd3:con:2} when $k_2 \precsim q$ and $k_2 < \tilde q$. Thus we can take $k_2 = \tilde q $ in \eqref{T1:cd3:con:2}, and the term $k_2^{-\tilde\alpha_2} \M I(k_2 < \tilde q)$ becomes zero. With $k_{1,\text{opt}}^{(2)} = \tilde{p}, k_{2,\text{opt}}^{(2)} = \tilde{q}$, we get the optimal convergence rate of \eqref{T1:cd3:con:2}: 
\bee\nonumber
r_{2}(k_1,k_2\mid p,q,\eta) \asymp \frac{pq}{n},
\ee
which is also faster than $n^{-\tilde\alpha_1\tilde\alpha_2/(\tilde\alpha_1 + \tilde\alpha_2 + \tilde\alpha_1\tilde\alpha_2)}$ and is the optimal rate we can get under Scenario 3. 
\par
\
\par
\noindent{\textbf{Summary:}} Under different divergence regimes of $p,q$, we summarize the optimal rate of $r_2(k_1,k_2\mid p,q,n) = r_2\big(k_{1,\text{opt}}^{(2)},k_{2,\text{opt}}^{(2)}\mid p,q,n\big)$ and corresponding $k_{1,\text{opt}}^{(2)}, k_{2,\text{opt}}^{(2)}$ as follows:
\bee\label{r2:opt}
r_2\big(k_{1,\text{opt}}^{(2)},k_{2,\text{opt}}^{(2)}\mid p,q,n\big) \asymp 
\begin{cases}
{n}^{-\frac{\tilde\alpha_1\tilde\alpha_2}{\tilde\alpha_1 + \tilde\alpha_2 + \tilde\alpha_1\tilde\alpha_2}} &\text{when } p \succ {n}^{\tilde\alpha_2/(\tilde\alpha_1+\tilde\alpha_2+\tilde\alpha_1\tilde\alpha_2)}, q\succ {n}^{\tilde\alpha_1/(\tilde\alpha_1+\tilde\alpha_2+\tilde\alpha_1\tilde\alpha_2)}
\\
(n/p)^{-\frac{\tilde\alpha_2}{\tilde\alpha_2 + 1}} & \text{when }  p \precsim {n}^{\tilde\alpha_2/(\tilde\alpha_1+\tilde\alpha_2+\tilde\alpha_1\tilde\alpha_2)}, q\succ \big({n}/p\big)^{\frac{1}{\tilde\alpha_2 + 1}}
\\
(n/q)^{-\frac{\tilde\alpha_1}{\tilde\alpha_1 + 1}} & \text{when }  q \precsim {n}^{\tilde\alpha_1/(\tilde\alpha_1+\tilde\alpha_2+\tilde\alpha_1\tilde\alpha_2)}, p\succ \big({n}/q\big)^{\frac{1}{\tilde\alpha_1 + 1}}
\\
\frac{pq}{n} & \text{when }
\begin{cases}
p \precsim {n}^{\tilde\alpha_2/(\tilde\alpha_1+\tilde\alpha_2+\tilde\alpha_1\tilde\alpha_2)}, q\precsim \big({n}/p\big)^{\frac{1}{\tilde\alpha_2 + 1}}
\\
q \precsim {n}^{\tilde\alpha_1/(\tilde\alpha_1+\tilde\alpha_2+\tilde\alpha_1\tilde\alpha_2)}, p\precsim \big({n}/q\big)^{\frac{1}{\tilde\alpha_1 + 1}}
\end{cases},
\end{cases}
\ee
with,
\bee\label{r2:kopt}
\small{\big(k_{1,\text{opt}}^{(2)},k_{2,\text{opt}}^{(2)}\big) =
\begin{cases}
\big({n}^{\tilde\alpha_2/(\tilde\alpha_1+\tilde\alpha_2+\tilde\alpha_1\tilde\alpha_2)},{n}^{\tilde\alpha_1/(\tilde\alpha_1+\tilde\alpha_2+\tilde\alpha_1\tilde\alpha_2)}\big) & \text{when }p \succ {n}^{\tilde\alpha_2/(\tilde\alpha_1+\tilde\alpha_2+\tilde\alpha_1\tilde\alpha_2)}, q\succ {n}^{\tilde\alpha_1/(\tilde\alpha_1+\tilde\alpha_2+\tilde\alpha_1\tilde\alpha_2)}
\\
\big(\tilde p, \{n/p\}^{\frac{1}{\tilde{\alpha}_2 + 1}}\big)& \text{when }  p \precsim {n}^{\tilde\alpha_2/(\tilde\alpha_1+\tilde\alpha_2+\tilde\alpha_1\tilde\alpha_2)}, q\succ \big({n}/p\big)^{\frac{1}{\tilde\alpha_2 + 1}}
\\
\big(\tilde q, \{n/q\}^{\frac{1}{\tilde{\alpha}_1 + 1}}\big)& \text{when }  q \precsim {n}^{\tilde\alpha_1/(\tilde\alpha_1+\tilde\alpha_2+\tilde\alpha_1\tilde\alpha_2)}, p\succ \big({n}/p\big)^{\frac{1}{\tilde\alpha_1 + 1}}
\\
\big(\tilde p,\tilde q\big) & \text{when }\begin{cases}
p \precsim {n}^{\tilde\alpha_2/(\tilde\alpha_1+\tilde\alpha_2+\tilde\alpha_1\tilde\alpha_2)}, q\precsim \big({n}/p\big)^{\frac{1}{\tilde\alpha_2 + 1}}
\\
q \precsim {n}^{\tilde\alpha_1/(\tilde\alpha_1+\tilde\alpha_2+\tilde\alpha_1\tilde\alpha_2)}, p\precsim \big({n}/q\big)^{\frac{1}{\tilde\alpha_1 + 1}}
\end{cases}.
\end{cases}}
\ee

\subsection{Minimax Optimal Regime}\label{supp:minimax1}
We discuss the minimax optimal regime for the error rates in Section \ref{sec:T}, where the lower bound in Theorem \ref{T:low} and upper bound in Theorems \ref{T2} and \ref{T3} (all in main paper) match with each other. We consider two regimes for $p,q$: the degenerate regime and the moderate high-dimensional regime (precise definition to be given later). We will show that under Scenario (i) defined in Section \ref{Sec_MR} of the main paper, our proposed estimator is rate-optimal under both regimes, and under Scenario (ii) defined in Section \ref{Sec_MR} of the main paper, our proposed estimator is rate-optimal under the degenerate regime. 
%In Section \ref{sec:match:illustration}, with some specific examples, we illustrate that derived regimes precisely capture all the rate-optimal regions of our derived upper and lower bounds, under both scenarios. 
\subsubsection{Degenerate Regime}\label{sec:dr}
As mentioned in Remark \ref{rm:dg}, we refer the case that $p$ or $q$ equals 1, as the degenerate regime and show that our estimators are the same as \citet{Bickel2008threshold}'s banded estimator and \citet{Cai2010}'s tapering estimator for vector-valued data under the degenerate regime. \citet{Cai2010} showed that banded/tapering covariance estimates for sub-Gaussian vector data are rate-optimal under the Frobenius norm when true covariance belongs to  $\mathcal{M}(\varepsilon_0,\alpha)$. Also, we note \citet{Cai2010}'s theorem and proof can directly adapt to vector data with finite fourth moment condition \eqref{am:moment}.  Therefore, our proposed estimators are rate-optimal under the degenerate regime for both Scenarios (i) and (ii), when $\M \Sigma_1^*\in \mathcal{M}(\varepsilon_0,\alpha_1), \M \Sigma_2^*\in \mathcal{M}(\varepsilon_0,\alpha_2)$. 
\par
The optimality can also be shown based on our main theorems. Let $q = 1$, $\M \Sigma_1^*\in \mathcal{M}(\varepsilon_0,\alpha_1), \M \Sigma_2^*\in \mathcal{M}(\varepsilon_0,\alpha_2)$ and take $k_2 = q = 1$. By simple algebra, the convergence rates of Theorems \ref{T2} and \ref{T3} become
$
\frac{k_1}{n} + \M I_{\eta,p}(k_1)\cdot k_1^{-2\alpha_1 - 1}.
$
The lower bound of Theorem \ref{T:low} becomes
$
\min\{\frac{p}{n}, n^{\frac{1}{2\alpha_1 + 2} - 1}\}
$, which aligns with \citet{Cai2010}'s lower bound for vector-valued data.
%since 
%\bee
%\min\Big\{\frac{q}{np},(np)^{\frac{1}{2\alpha_2 + 2} - 1}\Big\} = \min\Big\{\frac{q}{np},(np)^{\frac{1}{2\alpha_2 + 2} - 1}\Big\}........
%\ee
We can choose $k_1 = 2p$ when $p/n\precsim n^{\frac{1}{2\alpha_1 + 2} - 1}$, and $k_1 \asymp n^{1/(2\alpha_1 + 2)}$ otherwise. Then the upper and lower bounds  are matched. Thus the proposed banded/tapering estimators are rate-optimal. The same optimality result also holds symmetrically when $p = 1$ and $q$ diverges. In addition, $q = 1$ or $p = 1$ can be generalized to $p\leq C_{\text{D}}$ or $q\leq C_{\text{D}}$ for some fixed constant $C_{\text{D}} > 0$. And in this case, the proposed estimators are still rate-optimal.
\subsubsection{Moderate High-dimensional Regime}\label{sec:mhr}
We consider the following regime for $p,q$, 
\bee\label{mhdr}
\text{(i) }&1\precsim p\precsim\max\big\{n\cdot (qn)^{-\frac{1}{2\alpha_1 + 2}},\sqrt{n}\big\};
\\
\text{(ii) }&1\precsim q\precsim\max\big\{n\cdot (pn)^{-\frac{1}{2\alpha_2 + 2}},\sqrt{n}\big\}.
\ee
Under the condition $\M \Sigma_1^*\in \mathcal{M}(\varepsilon_0,\alpha_1), \M \Sigma_2^*\in \mathcal{M}(\varepsilon_0,\alpha_2)$, the $\tilde{\alpha}_1,\tilde{\alpha}_2$ in Theorem \ref{T2} become $2\alpha_1 + 1, 2\alpha_2 + 1$. Then Remark \ref{rm:mhdr} shows that under this regime, the convergence rate of Theorem \ref{T2} becomes \eqref{Tlow:res} with  optimal selection of $k_1,k_2$, i.e. matches the lower bound of Theorem \ref{T:low}.

So our proposed estimator is rate-optimal under \eqref{mhdr}, when $\M\Sigma_1^*,\M\Sigma_2^*$ are in $\mathcal{M}(\varepsilon_0,\alpha)$. It is easy to see $p\asymp\sqrt{n}$ and $q\asymp \sqrt{n}$ satisfy \eqref{mhdr}. Since $p\asymp\sqrt{n}$ and $q\asymp \sqrt{n}$ imply the whole dimension $pq\asymp n$, we call $p,q$'s divergence regime \eqref{mhdr}, the \textit{moderate high-dimensional regime}.

\subsection{Parameter Space Complexity}\label{sec:complexeffect}
We provide further details on how the complexity parameters such as the matrix size, bandable levels can affect our derived theoretical rates. With a separable structure, we can divide the whole parameter space of $\M\Sigma^* = \M\Sigma_2^*\otimes \M\Sigma_1^*$ into two subspaces with parameters residing in $\M\Sigma_1^*$ and $\M\Sigma_2^*$, respectively. It is easy to see the dimensions of subspace  $\M\Sigma_1^*$ and $\M\Sigma_2^*$ are $\Theta(p^2)$ and $\Theta(q^2)$, and thus the whole dimension is $\Theta(p^2 + q^2)$. The dimension of parameter space, together with bandable levels $\alpha_1,\alpha_2$ of two parameter subspaces, can be regarded as complexity measures of the whole parameter space $\M\Sigma^*$. With a more bandable parameter structure (larger $\alpha_1,\alpha_2$), more off-diagonal elements in $\M\Sigma_1^*$ and $\M\Sigma_2^*$ become negligible and thus  $\M\Sigma^*$ becomes less complex. %On the other hand, with lower dimensions of $\M\Sigma_1^*,\M\Sigma_2^*$, the number of parameters is reduced, and therefore $\M\Sigma^*$ becomes less complex. 
\par
With lower complexity of the parameter space, the optimal $k_1$ and $ k_2$ selection changes, and $\M\Sigma^*$ can be estimated more accurately. In this subsection, we study the effects of the  parameter space complexity, on (i) the error rates of our newly-derived upper and lower bounds, (ii) the optimal $k_1,k_2$ selection, and (iii) the minimax rate-optimal regions of $(p,q)$. 
\subsubsection{Effect on Error Rates}\label{sec:effect:rate}
 First, we consider the effects of bandable levels on the upper and lower bounds of the error rates. For any given divergence regimes of $p,q,k_1,k_2$, when the bandable levels $\alpha_1,\alpha_2$ become larger, the term $\M I_{\eta,p}(k_1)\cdot k_1^{-\tilde{\alpha}_1}+ \M I_{\eta,q}(k_2)\cdot k_2^{-\tilde{\alpha}_2}$ in \eqref{T2:res1} and \eqref{T3:res1}, and the terms $(nq)^{\frac{1}{2\alpha_1 + 2} - 1}$ and $(np)^{\frac{1}{2\alpha_2 + 2} - 1}$ in \eqref{Tlow:res}, converge no slower than the original ones, while other terms in the corresponding upper and lower bounds do not change. Therefore, with larger bandable levels, the error rates generally decay faster than or equal the original rate.
\par
To exhibit the effect of parameter space dimension, consider the scenario where the divergence rate of $pq$ is fixed and $\alpha_1 = \alpha_2$ (equal bandable levels of two subspaces).  In this case, the parameter dimension $\Theta(p^2 + q^2)$  is minimized when $p\asymp q$, which intuitively leads to the minimal complexity and the smallest error rate of our proposed estimators. Theoretically, one can see this via a simple example. When the rate of $pq$ is fixed with $pq\precsim \min \{n^{1/(2\alpha_1 + 2)}, n^{1/(2\alpha_2 + 2)}\}$, $p$ and $q$ are always under moderate high-dimensional regime since $p \vee q\precsim pq\precsim \sqrt{n}$. Then the error rate of our proposed estimators is minimax-optimal, and is equal to ${p}/{qn} + {q}/{pn}$ by \eqref{Tlow:res}. It is easy to  see that ${p}/{qn} + {q}/{pn}$ has the minimal rate $1/n$ when $p\asymp q$. Similarly, when $\alpha_1 \approx \alpha_2$ but $\alpha_1\neq \alpha_2$, the minimal error rate is attained when $p,q$ are approximately equal.
\par
In general, with lower parameter space complexity, the error rates of our proposed estimator become smaller and therefore $\M\Sigma^*$ can be estimated more accurately.
\subsubsection{Effect on Optimal $k_1,k_2$ Selection}\label{sec:effect:k1k2}
In this subsection, we discuss the optimal $k_1,k_2$ selection. To make the discussion more meaningful, we only focus on the regimes of $p,q$ that our proposed estimator is minimax rate-optimal. In particular, we focus on the Scenario (i) with $\M\Sigma_1^*,\M\Sigma_2^*$ in the $\mathcal{M}(\varepsilon_0,\alpha)$ class. As shown in Section \ref{sec:compare}, we consider the following two rate-optimal regimes.
\begin{itemize}
\item \textbf{Degenerate regime}
When $q = 1$, \citet{Cai2010} has shown that the optimal $k_1$ is $k_{1,\text{opt}} 
= \begin{cases}
\tilde{p} & \text{when }p\precsim n^{1/(2\alpha_1 + 2)}
\\
n^{1/(2\alpha_1 + 2)} & \text{when }p\succ n^{1/(2\alpha_1 + 2)}
\end{cases}$, where $\tilde{p} = \begin{cases}
p - 1 &\eta = \mathcal{B}
\\
2p - 1 & \eta = \mathcal{T}
\end{cases}$ is the minimal thresholding level that makes the proposed estimator not banded along the $p$ direction.  %When $p = 1$, $k_{1,\text{opt}}$ is trivially $\tilde{p}$. 
 When $ p=1$, we can observe similar phenomenon.
\item \textbf{Moderate high-dimensional regime.}
When $p,q$ satisfy the moderate high-dimensional condition \eqref{mhdr}, Remark \ref{rm:mhdr} shows that the optimal $k_1$ is,
\bee\label{summarize:kopt}
k_{1,\text{opt}}^{} = 
\begin{cases}
\tilde{p} & \text{when }p\precsim \sqrt{n}\wedge (nq)^{1/(2\alpha_1 + 2)}
\\
(nq)^{1/(2 \alpha_1 + 2)} & \text{when }p\succ \sqrt{n}\wedge (nq)^{1/(2\alpha_1 + 2)}
\end{cases}.
\ee 
\end{itemize}
By simple algebra, the above optimal choice for $k_1$ under both regimes in fact can be summarized as \eqref{summarize:kopt},
when $p,q$ are under either of the two rate-optimal regimes.
\par
Therefore, $\sqrt{n}\wedge (nq)^{1/({2\alpha_1 + 2})}$ is the phrase transition threshold such that only when $p\succ \sqrt{n}\wedge (nq)^{1/({2\alpha_1 + 2})}$,  regularization is needed for our proposed estimators to attain the minimax optimal rate. In addition, $(nq)^{1/({2\alpha_1 + 2})}$ is the minimax optimal regularization level of $k_1$. This result matches with our empirical finding (see more details in Section \ref{sec:match:illustration}), that is, the non-banded region becomes wider as $\alpha_1$ decreases, and when $\alpha_1 \rightarrow 0$ the threshold for $p$ is $\sqrt{n}$. Intuitively, when $\alpha_1$ becomes larger or the dimension of subspace $\M\Sigma_2^*$, i.e., $\Theta(q^2)$, becomes smaller, the regularization is beneficial in a lower or equal-dimensional case, and the regularization level of $k_1$ becomes higher, because $(nq)^{1/(2\alpha_1 + 2)}$ diverges slower. Intuitively this makes sense because a higher regularization level adapts to a more bandable structure in $\M\Sigma_1^*$. On the other hand, the dimensionality of subspace $\M\Sigma_1^*$, i.e., $\Theta(p^2)$,  determines whether to perform regularization on the $p$ direction or not, since the regularization is only helpful when $p$ is sufficiently large.  Similar results can be shown for the $q$  dimension as well.
%Also, the dimensionality of $\M\Sigma_1^*$, the $\Theta(p^2)$, determines the occurrence of regularization via $k_1$. The effects of $\alpha_1$ and $\Theta(p^2)$ can be interpreted similarly to the degenerate regime. It is worthy mentioning that the dimension of the counterpart subspace $\M\Sigma_2^*$, influences the optimal $k_1$ selection procedure by changing the corresponding regularization threshold and the regularization level. 
\subsubsection{Effect on the Minimax Rate-optimal Regions of $(p,q)$}\label{sec:effect:region}
In this subsection, we discuss the effect of $\alpha_1,\alpha_2$ on the minimax rate-optimal regions of $(p,q)$. Among two rate-optimal regimes discussed previously under Scenario (i) with  $\M\Sigma_1^*,\M\Sigma_2^*$ in the $\mathcal{M}(\varepsilon_0,\alpha)$ class, the degenerate regime does not change for different values of $\alpha_1,\alpha_2$. Therefore, we focus on studying the effect of $\alpha_1,\alpha_2$ on the minimax rate-optimal regions of $(p,q)$ under the moderate high-dimensional regime. 
\par
We consider the case when $p,q$ are polynomially divergent such that $p = n^{\beta_1}, q = n^{\beta_2}$ for some constants $\beta_1,\beta_2 > 0$. So we can use a pair of $(\beta_1,\beta_2)$ to represent a specific divergent rate of $(p,q)$. By simple algebra, the polynomially divergent  $p,q$ satisfy the moderate high-dimensional regime \eqref{mhdr} if and only if $(\beta_1,\beta_2)$ belongs to the following region,
\bee\label{def:rmh}
\mathcal{R}_{\text{MH}} = \Bigg\{(\beta_1,\beta_2)\mid &(i).\ 0\leq \beta_1 \leq \M I(\beta_2\leq \alpha_1)\cdot \frac{2\alpha_1 + 1 - \beta_2}{2\alpha_1 + 2} + \M I(\beta_2 > \alpha_1) \cdot\frac{1}{2};
\\
&(ii). 0\leq \beta_2 \leq \M I(\beta_1\leq \alpha_2)\cdot \frac{2\alpha_2 + 1 - \beta_1}{2\alpha_2 + 2} + \M I(\beta_1 > \alpha_2) \cdot\frac{1}{2}\Bigg\}.
\ee
\par
In Figure \ref{fig:region} (a), we show the region of $\mathcal{R}_{\text{MH}}$ when both $\alpha_1$ and $\alpha_2$ go to $+\infty$. The region is actually $\max\{\beta_1,\beta_2\}\leq 1$, because for any fixed $\beta_2$ the upper bound in \eqref{def:rmh} (i) is
\bee\nonumber
\M I(\beta_2\leq \alpha_1)\cdot \frac{2\alpha_1 + 1 - \beta_2}{2\alpha_1 + 2} + \M I(\beta_2 > \alpha_1) \cdot\frac{1}{2} \longrightarrow 1
\ee
when $\alpha_1 \rightarrow +\infty$. A similar result holds for \eqref{def:rmh} (ii) when $\alpha_2 \rightarrow \infty$. In Figure \ref{fig:region} (b) and (c), we show the regions of $\mathcal{R}_{\text{MH}}$ with $\alpha_1 = \alpha_2$ being $1.5$ and $0.2$, respectively. In Figure \ref{fig:region} (d), we plot $\mathcal{R}_{\text{MH}}$ when both $\alpha_1$ and $\alpha_2$ converge to $0$. The \eqref{def:rmh} directly implies that $\mathcal{R}_{\text{MH}}$ becomes $\max\{\beta_1,\beta_2\}\leq .5$ in this case. Through Figures \ref{fig:region} (a)-(d), one can easily see that when $\alpha_1$ and $\alpha_2$ decrease from $+\infty$ to $0$, the corresponding $\mathcal{R}_{\text{MH}}$ shrinks from $\max\{\beta_1,\beta_2\} \leq 1$, and finally to $\max\{\beta_1,\beta_2\} \leq .5$. These observations match with the theory, since when $\max\{\alpha_1,\alpha_2\}\rightarrow +\infty$, the whole region of moderate high-dimensional regime becomes  $\max\{p,q\}\precsim n$. On the other hand, when $\max\{\alpha_1,\alpha_2\}\rightarrow 0$, the whole region of moderate high-dimensional regime collapses to $\max\{p,q\}\precsim n^{1/2}$.
\begin{figure}
  \centering
    \subfigure[$\alpha_1,\alpha_2 \rightarrow +\infty$]{\includegraphics[width=0.49\textwidth]{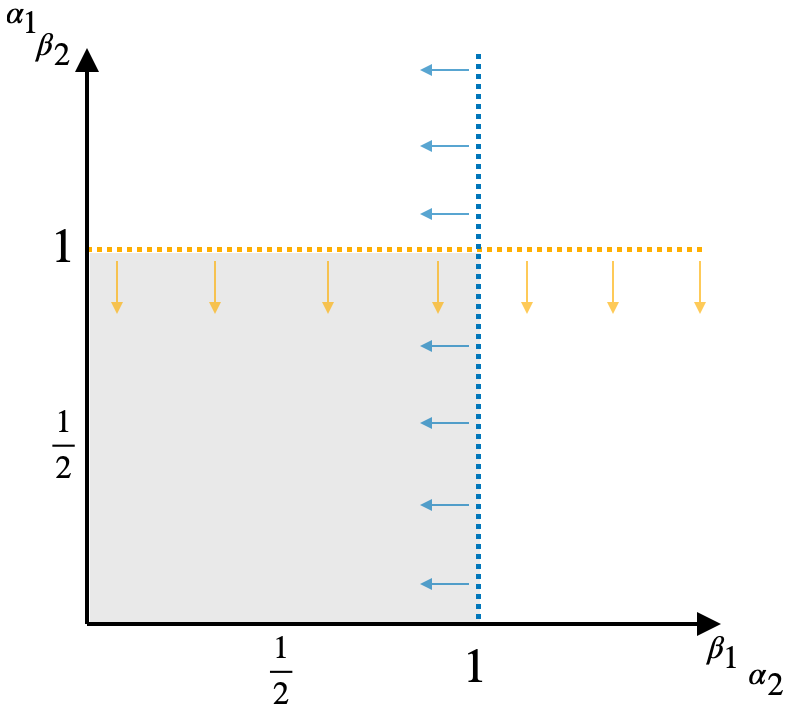}}
  \subfigure[$\alpha_1=\alpha_2 = 1.5$]{\includegraphics[width=0.49\textwidth]{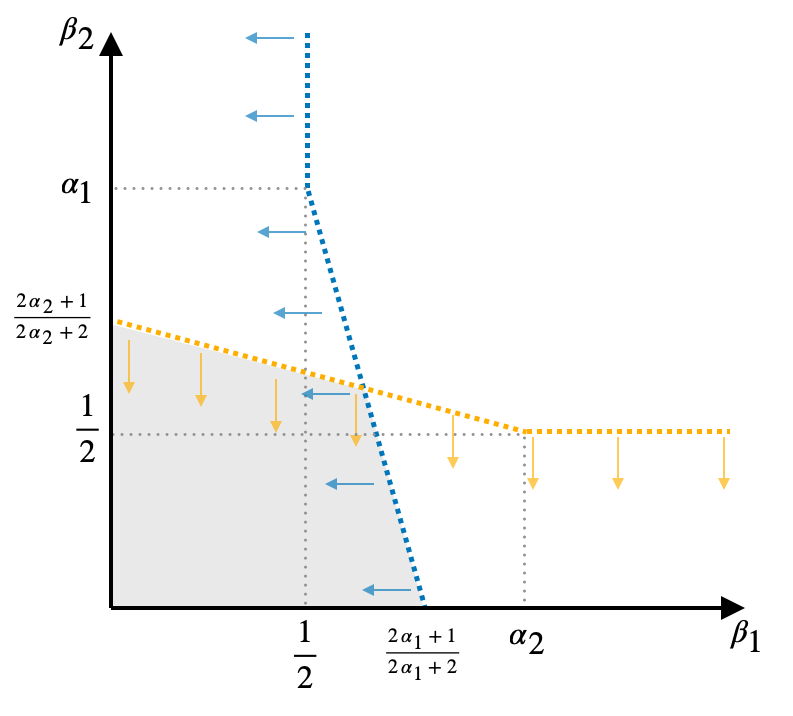}}
\par
    \subfigure[$\alpha_1=\alpha_2=0.2$]{\includegraphics[width=0.49\textwidth]{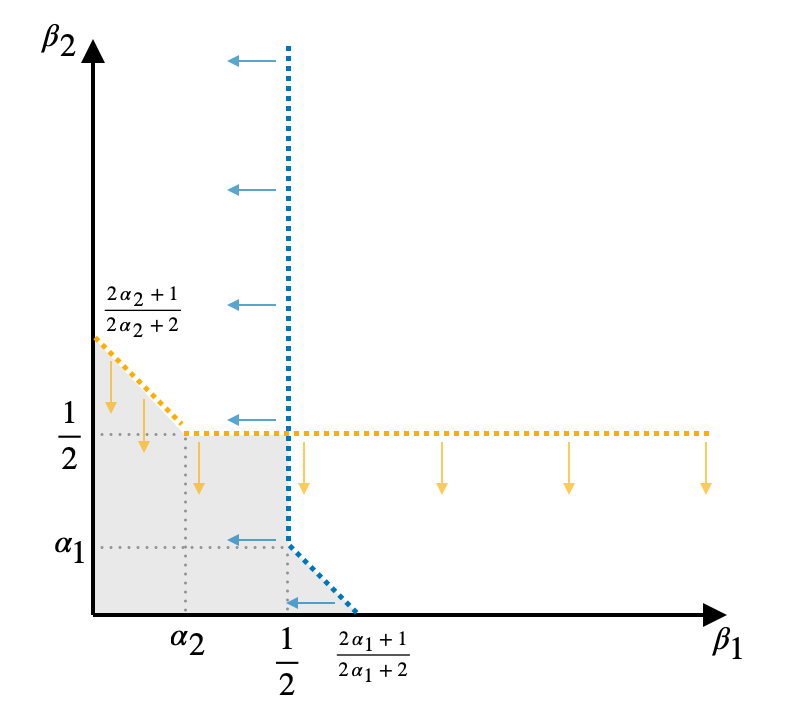}}
  \subfigure[$\alpha_1, \alpha_2 \rightarrow 0$]{\includegraphics[width=0.49\textwidth]{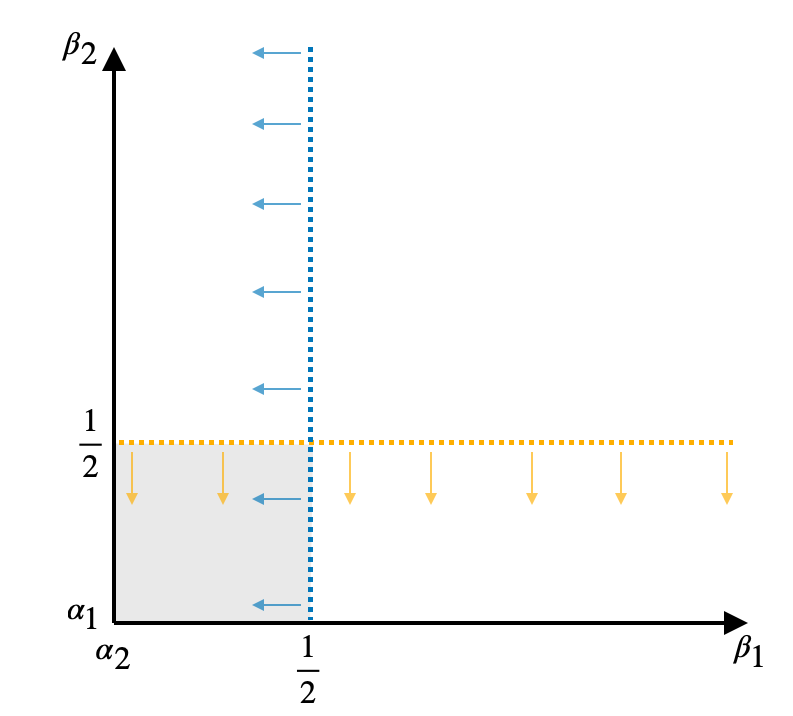}}
 \caption{The demonstration of $\mathcal{R}_{\text{MH}}$. Panel (a) shows the $\mathcal{R}_{\text{MH}}$ when $\alpha_1,\alpha_2 \rightarrow +\infty$. Panel (b)-(d) show the $\mathcal{R}_{\text{MH}}$ with $\alpha_1,\alpha_2$ are both $1.5$, $0.2$ and $0$, respectively. In all panels, the constraint (i) in the definition of $\mathcal{R}_{\text{MH}}$ in \eqref{def:rmh}, is presented by the region on the left-hand side of the blue curve, and the constraint (ii) in the definition of $\mathcal{R}_{\text{MH}}$ is presented by the region under the yellow curve. Therefore, the exact region of $\mathcal{R}_{\text{MH}}$ is the grey region, which is the intersection of blue and yellow regions.}
\label{fig:region}
\end{figure}

%and thus the moderate high-dimensional regime becomes broader. For simplicity, we consider a fixed divergence regime of $q$, such that $q\precsim \sqrt{n}$. Then $q$ automatically satisfies constraint (ii) in \eqref{mhdr}, and therefore $p,q$ satisfies the moderate high-dimensional regime condition \eqref{mhdr} as long as $1\precsim p\precsim\max\big\{n\cdot (qn)^{-\frac{1}{2\alpha_1 + 2}},\sqrt{n}\big\}$. With sufficiently larger $\alpha_1$, the upper bound of $p$, $\max\big\{n\cdot (qn)^{-\frac{1}{2\alpha_1 + 2}},\sqrt{n}\big\}$ diverges faster than the original one, and therefore a higher dimensional regime of $p$ can also satisfy the moderate high-dimensional regime. Symmetric argument shows that higher-dimensional regime of $q$ can satisfy \eqref{mhdr} when $\alpha_2$ increases and $p\precsim \sqrt{n}$. When $p$ or $q\succsim \sqrt{n}$, a similar, yet more massive argument can also show moderate high-dimensional regime becomes broader, when both $\alpha_1$ and $\alpha_2$ become larger.
%

\subsection{Minimax Optimal Regime: Graphical Illustration}\label{sec:match:illustration}
In this subsection, we illustrate the upper and lower bounds derived in Theorems \ref{T2}, \ref{T3} and \ref{T:low} using specific examples. In particular, we consider the case when $p,q$ are polynomially divergent and let $$p = n^{\beta_1}, q = n^{\beta_2},$$  with $\beta_1,\beta_2 > 0$. We consider following four specific settings of $(\alpha_1,\alpha_2)$:  $(\alpha_1,\alpha_2) = (2,2)$, $(2,1)$, $(0.4,0.4)$, and $(0.4,0.2)$.  When $\M \Sigma_1^*\in \mathcal{M}(\varepsilon_0,\alpha_1), \M \Sigma_2^*\in \mathcal{M}(\varepsilon_0,\alpha_2)$ with $(\alpha_1,\alpha_2) = (2,2)$ and $(2,1)$, as $\alpha_1,\alpha_2$ are relatively large, the entries of $\M \Sigma_1^*$ and $\M \Sigma_2^*$ vanish to zero very fast when they are far away from diagonal. We name these scenarios as \textit{strongly bandable examples}. When $(\alpha_1,\alpha_2) = (0.4,0.4)$ and $(0.4,0.2)$, as $\alpha_1,\alpha_2$ are relatively small, similarly, we name these senarios as   \textit{weakly bandable examples}. 
\par
From Sections \ref{optbd:T2} and \ref{optbd:T3}, for $p = n^{\beta_1}, q = n^{\beta_2}$ with $\beta_1,\beta_2 >0$, the upper bounds of Theorems \ref{T2} and \ref{T3} after selecting the optimal $k_1, k_2$, as well as the lower bound of Theorem \ref{T:low}, always have the polynomial form: 
$
n^{-r}.
$ 
%Here $r$ is a constant determined by $\beta_1,\beta_2$ for Theorems \ref{T2}, \ref{T3}, and \ref{T:low}. We also note that the relationships between $r$ and $(\beta_1,\beta_2)$ are different for different theorems. 
As $r = -\log_n\big( n^{-r}\big)$, we call $r$ the negative log convergence rate (NLCR) and use it to measure the convergence rate of both upper and lower bounds. For example, an NLCR value of $1$ implies an error rate (bound) of $n^{-1}$. Therefore, a larger NLCR value (e.g., color blue in Figure \ref{fig:sgbd1}) implies a better (faster) convergence rate. Also implied by the optimal $k_1,k_2$ selection procedure in Sections \ref{optbd:T2} and \ref{optbd:T3}, the optimal divergence regimes of $k_1, k_2$ always have the form: $n^{r'}$. %where $r'$ is  determined by $\beta_1,\beta_2 $ for each upper bound and the lower bound. 
As $r' = \log_n\big( n^{r'}\big)$, we call $r'$ the log divergence rate (LDR) and use it to measure the divergence rate of optimal $k_1,k_2$.
\par
\
\par
\noindent
\textbf{Strongly Bandable Examples: } In these examples, we first plot the lower bounds of Theorem \ref{T:low} in Figures \ref{fig:sgbd1}--\ref{fig:sgbd2} (a). Under the Scenario (i) defined in Section \ref{Sec_MR} of the main paper, we plot the NLCR of the optimal convergence rate of Theorem \ref{T2}, in Figures \ref{fig:sgbd1}--\ref{fig:sgbd2} (b). The optimal convergence rate of Theorem \ref{T2} is minimax rate-optimal for both degenerate and moderate high-dimensional regimes. Based on discussions in Sections \ref{sec:dr}, \ref{sec:mhr} and formula \eqref{def:rmh}, we derive rate-optimal regions of $(\beta_1,\beta_2)$ when $(\alpha_1,\alpha_2) = (2,2)$ and $(2,1)$, denoted by $\mathcal{R}_{2,2}$ and $\mathcal{R}_{2,1}$, respectively:
%the following regions of $\beta_1,\beta_2$: the $\mathcal{R}_{2,2}$ and $\mathcal{R}_{2,1}$ that make our proposed estimators rate-optimal when $(\alpha_1,\alpha_2) = (2,2)$ and $(2,1)$ respectively,
\bee\nonumber
&\mathcal{R}_{2,2} = \underbrace{\Big\{(\beta_1,\beta_2)\mid 0\leq\beta_2 \leq \min\{-\frac{1}{6}\beta_1 + \frac{5}{6},-6\beta_1 + 5\}, 0\leq\beta_1\Big\}}_{\text{moderate high-dimensional regime}}\bigcup \underbrace{\Big\{(\beta_1,\beta_2)\mid\beta_1 = 0 \text{ or }\beta_2 = 0\Big\}}_{\text{degenerate regime}}
\\
&\mathcal{R}_{2,1} = \underbrace{\Big\{(\beta_1,\beta_2)\mid 0\leq\beta_2 \leq \min\{-\frac{1}{4}\beta_1 + \frac{3}{4},-6\beta_1 + 5\}, 0\leq\beta_1\Big\}}_{\text{moderate high-dimensional regime}}\bigcup \underbrace{\Big\{(\beta_1,\beta_2)\mid\beta_1 = 0 \text{ or }\beta_2 = 0\Big\}}_{\text{degenerate regime}}.
\ee 
To visualize the regions $\mathcal{R}_{2,2}$ and $\mathcal{R}_{2,1}$, we plot them in Figure \ref{fig:sgbd1}(b) and Figure \ref{fig:sgbd2} (b), respectively. The regions are in the bottom left corner, surrounded by red dashed lines. 
%The regions $\mathcal{R}_{2,2}$ and $\mathcal{R}_{2,1}$ are surrounded by red dashed lines, shown in Figure \ref{fig:sgbd1}(b) and Figure \ref{fig:sgbd2} (b), respectively.. 
%Comparing Figures \ref{fig:sgbd1}--\ref{fig:sgbd2} (a) and (b), we can see regions $\mathcal{R}_{2,2}$ and $\mathcal{R}_{2,1}$ precisely capture all the regimes that derived upper and lower bounds are matched under Scenario (i) in both strongly bandable examples. 
When divergence regimes of $p,q$ are in the minimax rate-optimal regions, we present the corresponding LDR of  optimal $k_1$ in Figure \ref{fig:phase} (a)-(b). The corresponding LDR of optimal $k_2$ can be found in Figure \ref{fig:sup:ratecompare} (a)-(b). In Figures \ref{fig:phase} and \ref{fig:sup:ratecompare}, the non-banded regions, where no banding of $k_1$ or $k_2$ is necessary to achieve optimal convergence rate, are in the left corner of each panel, surrounded by red dashed lines.
\par
 For Scenario (ii) defined in Section \ref{Sec_MR} of the main paper, the convergence rate of proposed estimators is given in Theorem \ref{T3}, and is minimax rate-optimal only for the degenerate regime. 
We show the rate-optimal degenerate regime region, which is $
{\big\{(\beta_1,\beta_2)\mid\beta_1 = 0 \text{ or }\beta_2 = 0\big\}}_{}
$, by the two red dashed lines in Figures \ref{fig:sgbd1}--\ref{fig:sgbd2} (c). 
\begin{figure}[H]
  \centering
    \subfigure[Lower Bound]{\includegraphics[width=0.49\textwidth]{22low}}
  \subfigure[Upper Bound under Scenario (i) (i.e., sub-Gaussian tail)]{\includegraphics[width=0.49\textwidth]{22up}}
\\
  \subfigure[Upper Bound under Scenario (ii) (i.e., finite fourth moment) ]{\includegraphics[width=0.49\textwidth]{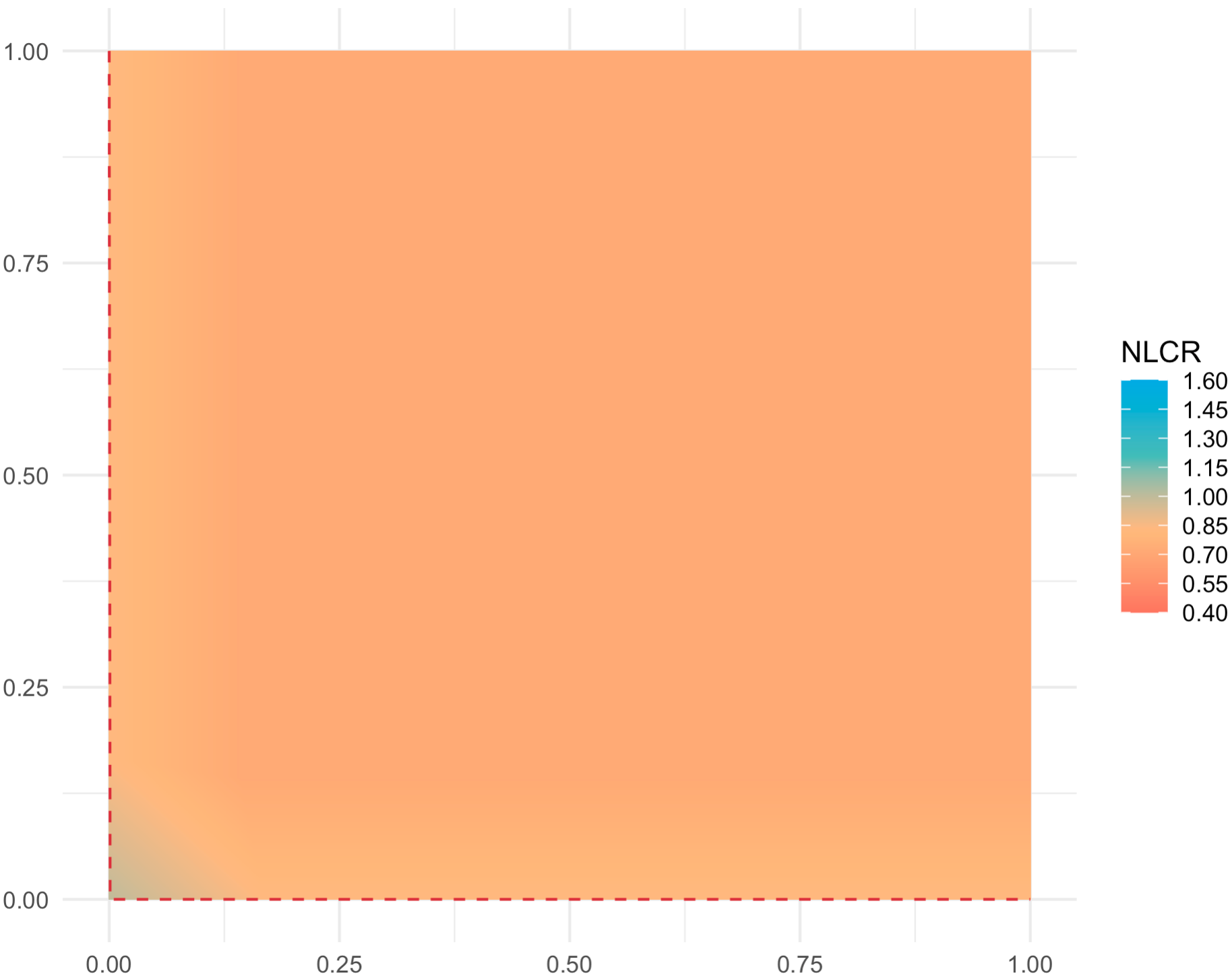}}
 \caption{Strongly Bandable Example ($\alpha_1 = 2, \alpha_2 = 2$): The x-axis represents $\log p/\log n = \beta_1$ and y-axis represents $\log q/\log n = \beta_2$. The color represents the negative log convergence rates (NLCR) of the corresponding upper and lower bounds, where blue means a faster convergence rate. Panel (a) gives the overall error lower bound that applies for both Scenarios (i) and (ii); panel (b) gives the error upper bound under Scenario (i) (sub-Gaussian scenario), and the regions in the bottom left corner, surrounded by red dashed lines, corresponds to the rate-optimal region; panel (c) gives the upper bound under Scenario (ii) (finite fourth moment scenario), and the two red dashed lines represent the rate-optimal region under this scenario. }
\label{fig:sgbd1}
\end{figure}

\begin{figure}[H]
  \centering
    \subfigure[Lower Bound]{\includegraphics[width=0.49\textwidth]{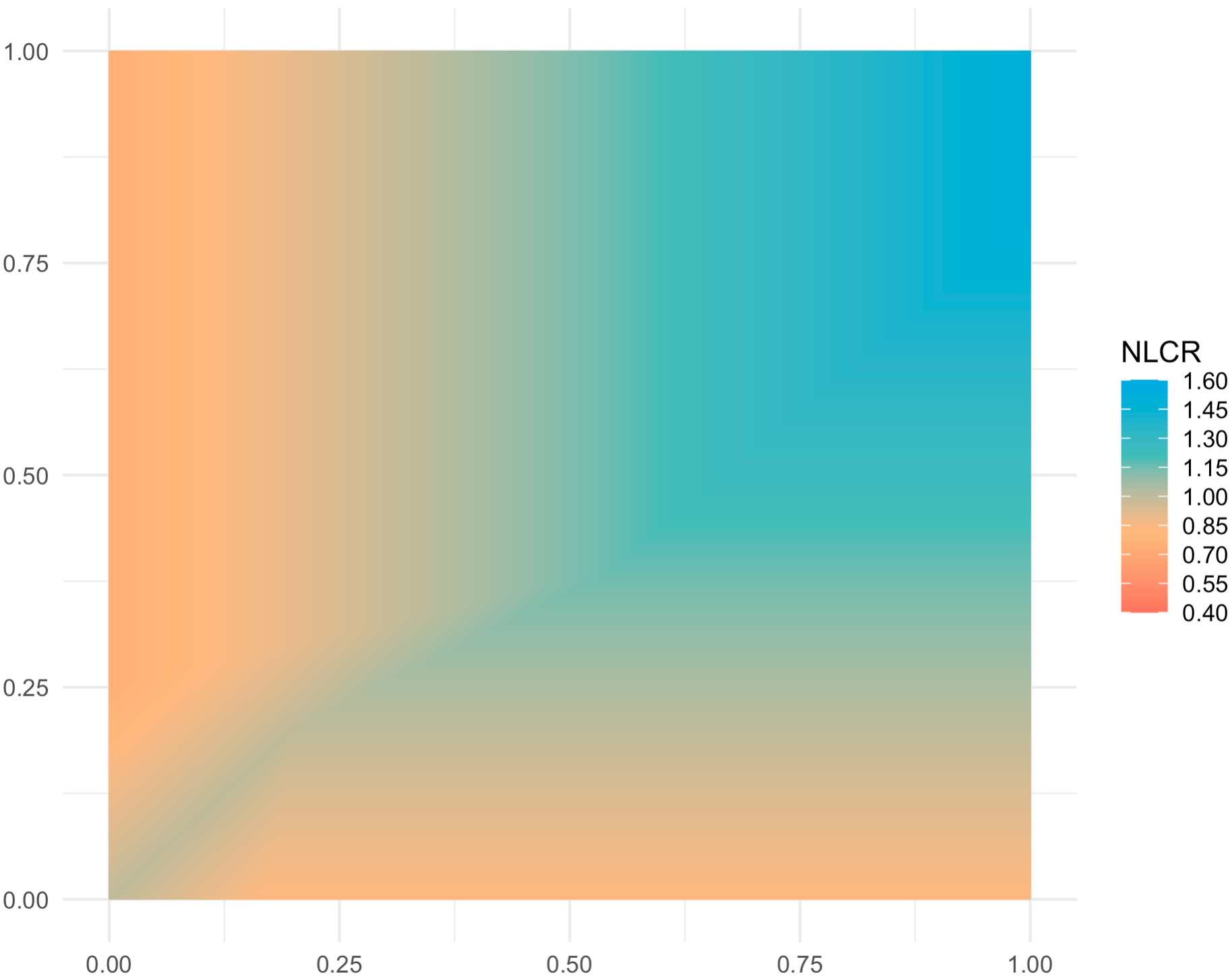}}
  \subfigure[Upper Bound for Scenario]{\includegraphics[width=0.49\textwidth]{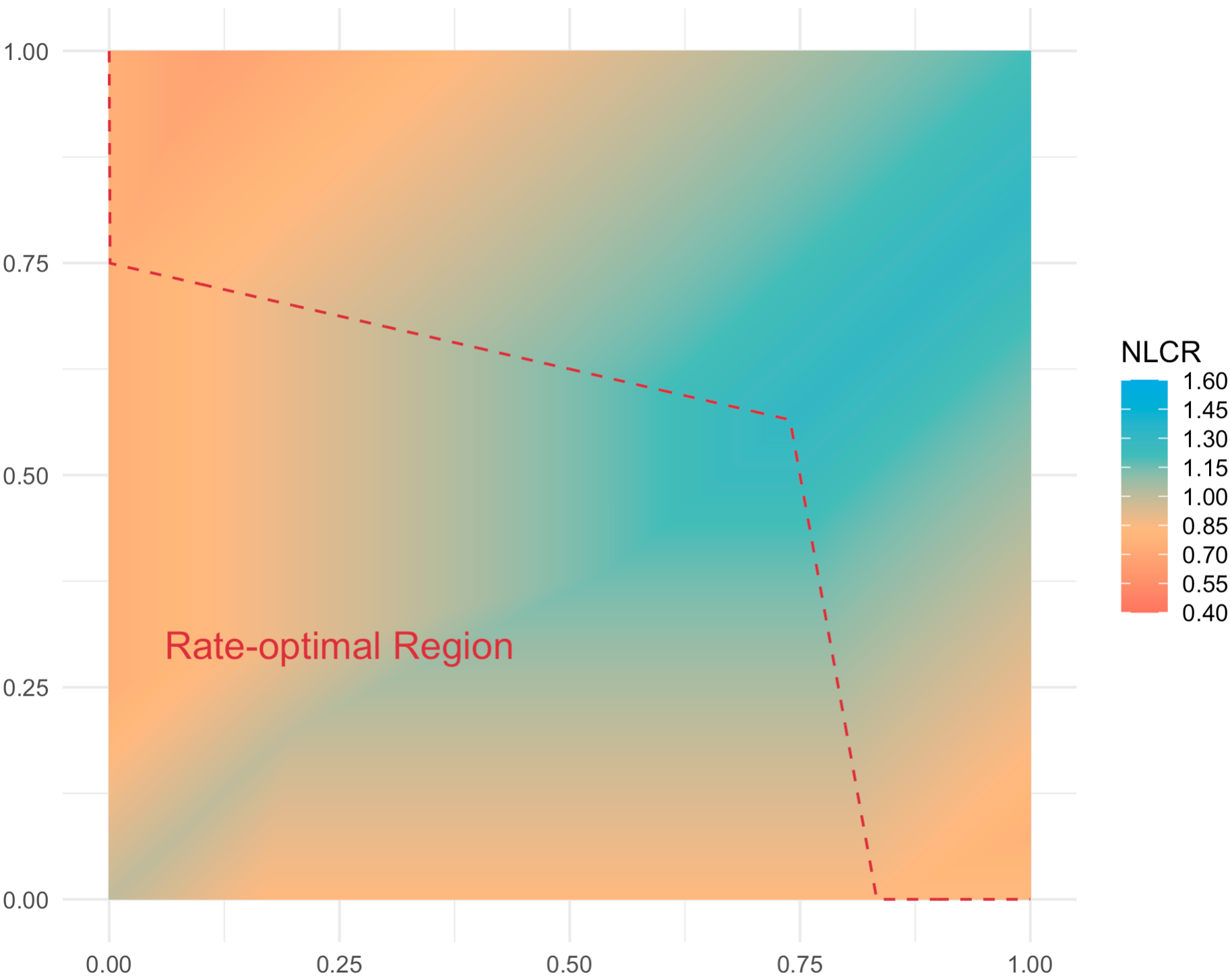}}
\\
  \subfigure[Upper Bound for Scenario (ii)]{\includegraphics[width=0.49\textwidth]{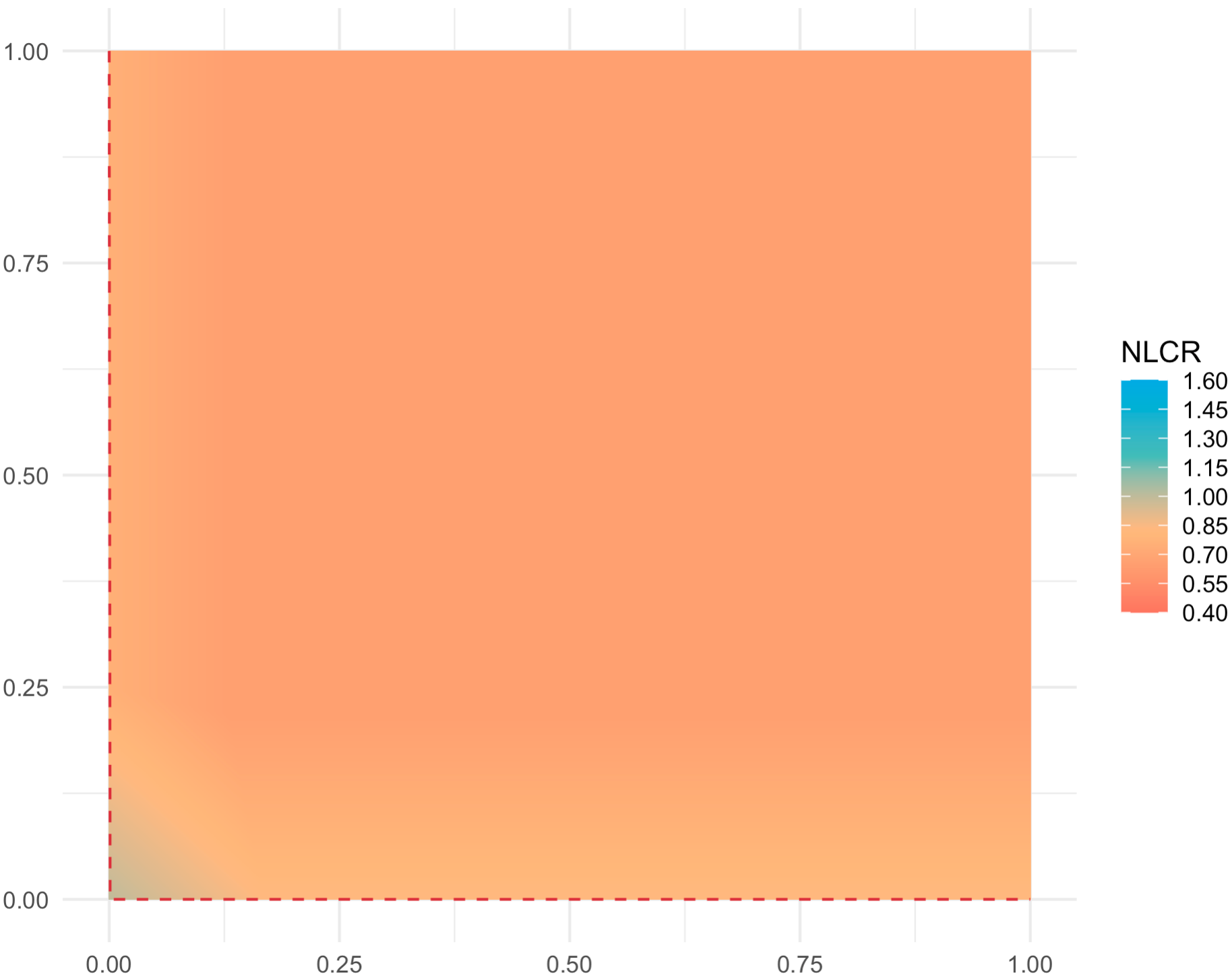}}
 \caption{Strongly Bandable Example ($\alpha_1 = 2, \alpha_2 = 1$):  The x-axis represents $\log p/\log n = \beta_1$ and y-axis represents $\log q/\log n = \beta_2$. The color represents the negative log convergence rates (NLCR) of the corresponding upper and lower bounds, where blue means a faster convergence rate. Panel (a) gives the overall error lower bound that applies for both Scenarios (i) and (ii); panel (b) gives the error upper bound under Scenario (i) (sub-Gaussian scenario), and the regions in the bottom left corner, surrounded by red dashed lines, corresponds to the rate-optimal region; panel (c) gives the upper bound under Scenario (ii) (finite fourth moment scenario), and the two red dashed lines represent the rate-optimal region under this scenario. }
\label{fig:sgbd2}
\end{figure}
\noindent
\textbf{Weakly Bandable Examples: } Similar to the strongly bandable examples, we plot the lower bounds of Theorem \ref{T:low} in Figures \ref{fig:wkbd1}--\ref{fig:wkbd2} (a). Under the Scenario (i), we plot the NLCR of the optimal convergence rate of Theorem \ref{T2}, in Figures \ref{fig:wkbd1}--\ref{fig:wkbd2} (b). Similar to the strongly bandable examples, we derive rate-optimal regions of $(\beta_1,\beta_2)$ when $(\alpha_1,\alpha_2) = (0.4,0.4)$ and $(0.4,0.2)$, denoted by $\mathcal{R}_{0.4,0.4}$ and $\mathcal{R}_{0.4,0.2}$, respectively:
\bee\label{rateopt:case2}
\mathcal{R}_{0.4,0.4} = &\underbrace{\Big\{(\beta_1,\beta_2)\mid 0\leq\beta_2 \leq -\frac{5}{14}\beta_1 + \frac{9}{14}, 0\leq\beta_1\leq 0.4\Big\}}_{\text{moderate high-dimensional regime}}\bigcup \underbrace{\Big\{(\beta_1,\beta_2)\mid 0\leq\beta_2 \leq 0.5, 0.4\leq\beta_1\leq 0.5\Big\}}_{\text{moderate high-dimensional regime}}
\\
&\bigcup\underbrace{\Big\{(\beta_1,\beta_2)\mid 0\leq\beta_2 \leq -2.8\beta_1 + 1.8, 0.5\leq\beta_1\leq 1\Big\}}_{\text{moderate high-dimensional regime}}\bigcup \underbrace{\Big\{(\beta_1,\beta_2)\mid\beta_1 = 0 \text{ or }\beta_2 = 0\Big\}}_{\text{degenerate regime}},
\\
\mathcal{R}_{0.4,0.2} = &\underbrace{\Big\{(\beta_1,\beta_2)\mid 0\leq\beta_2 \leq -\frac{5}{12}\beta_1 + \frac{7}{12}, 0\leq\beta_1\leq 0.2\Big\}}_{\text{moderate high-dimensional regime}}\bigcup \underbrace{\Big\{(\beta_1,\beta_2)\mid 0\leq\beta_2 \leq 0.5, 0.2\leq\beta_1\leq 0.5\Big\}}_{\text{moderate high-dimensional regime}}
\\
&\bigcup\underbrace{\Big\{(\beta_1,\beta_2)\mid 0\leq\beta_2 \leq -2.8\beta_1 + 1.8, 0.5\leq\beta_1\leq 1\Big\}}_{\text{moderate high-dimensional regime}}\bigcup \underbrace{\Big\{(\beta_1,\beta_2)\mid\beta_1 = 0 \text{ or }\beta_2 = 0\Big\}}_{\text{degenerate regime}}.
\ee
To visualize the regions $\mathcal{R}_{0.4,0.4}$ and $\mathcal{R}_{0.4,0.2}$, we plot them in Figure \ref{fig:wkbd1} (b) and Figure \ref{fig:wkbd2} (b), respectively. The regions are in the bottom left corner, surrounded by red dashed lines. 
%The regions $\mathcal{R}_{2,2}$ and $\mathcal{R}_{2,1}$ are surrounded by red dashed lines, shown in Figure \ref{fig:sgbd1}(b) and Figure \ref{fig:sgbd2} (b), respectively.. 
%Comparing Figures \ref{fig:sgbd1}--\ref{fig:sgbd2} (a) and (b), we can see regions $\mathcal{R}_{2,2}$ and $\mathcal{R}_{2,1}$ precisely capture all the regimes that derived upper and lower bounds are matched under Scenario (i) in both strongly bandable examples. 
When divergence regimes of $p,q$ are in the minimax rate-optimal regions, we present the corresponding LDR of optimal $k_1$ in Figure \ref{fig:phase} (c)-(d). The corresponding LDR of optimal $k_2$ can be found in Figure \ref{fig:sup:ratecompare} (c)-(d). In Figures \ref{fig:phase} and \ref{fig:sup:ratecompare}, the non-banded regions, where no banding of $k_1$ or $k_2$ is necessary to achieve optimal convergence rate, are in the left corner of each panel, surrounded by red dashed lines.
\par
 For Scenario (ii), the convergence rate is minimax rate-optimal only for the degenerate regime. 
We show the rate-optimal degenerate regime region, which is $
{\big\{(\beta_1,\beta_2)\mid\beta_1 = 0 \text{ or }\beta_2 = 0\big\}}
$, by the two red dashed lines in Figures \ref{fig:wkbd1}--\ref{fig:wkbd2} (c). 
\begin{figure}
  \centering
    \subfigure[Lower Bound]{\includegraphics[width=0.49\textwidth]{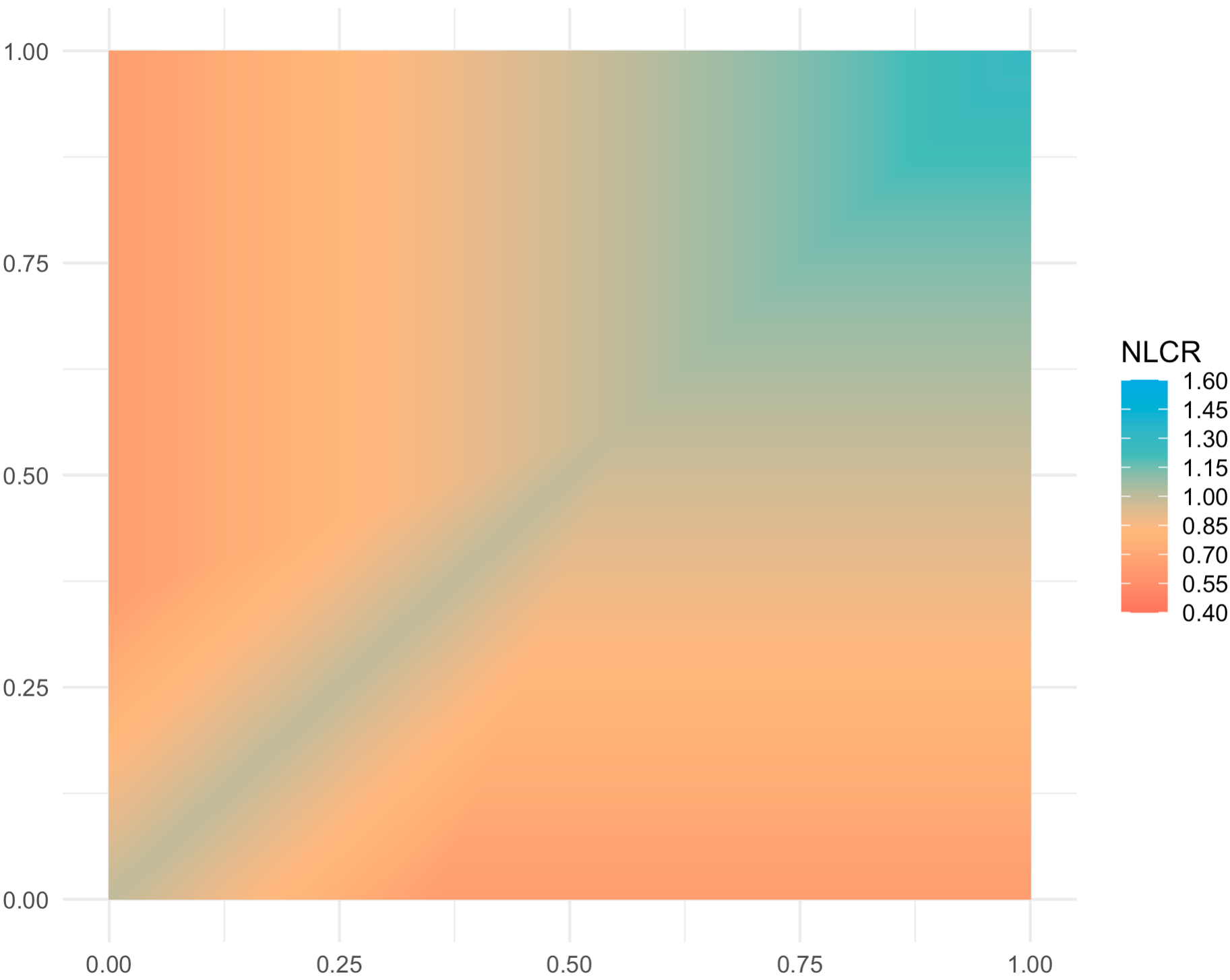}}
  \subfigure[Upper Bound for Scenario (i)]{\includegraphics[width=0.49\textwidth]{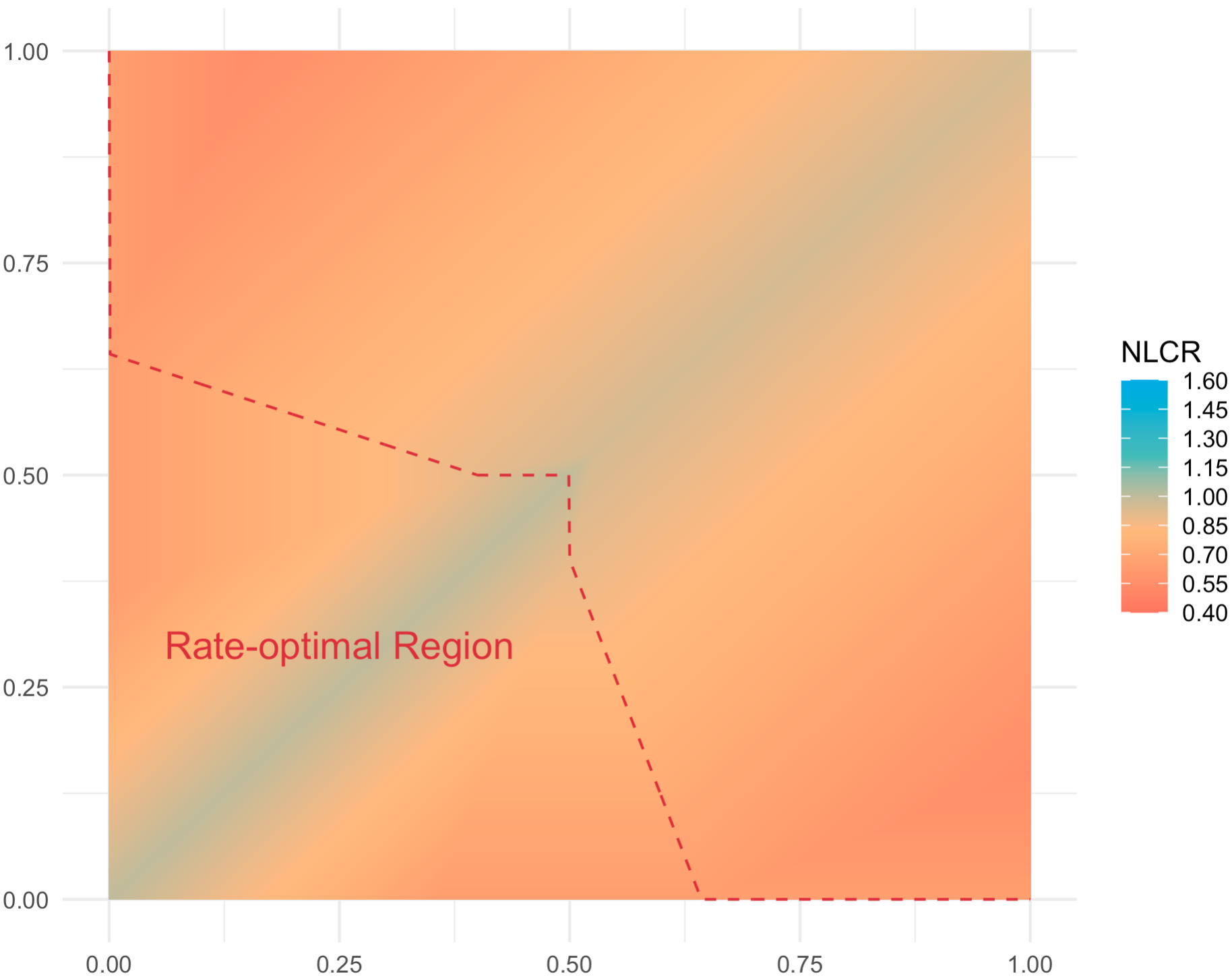}}
\\
  \subfigure[Upper Bound for Scenario (ii)]{\includegraphics[width=0.49\textwidth]{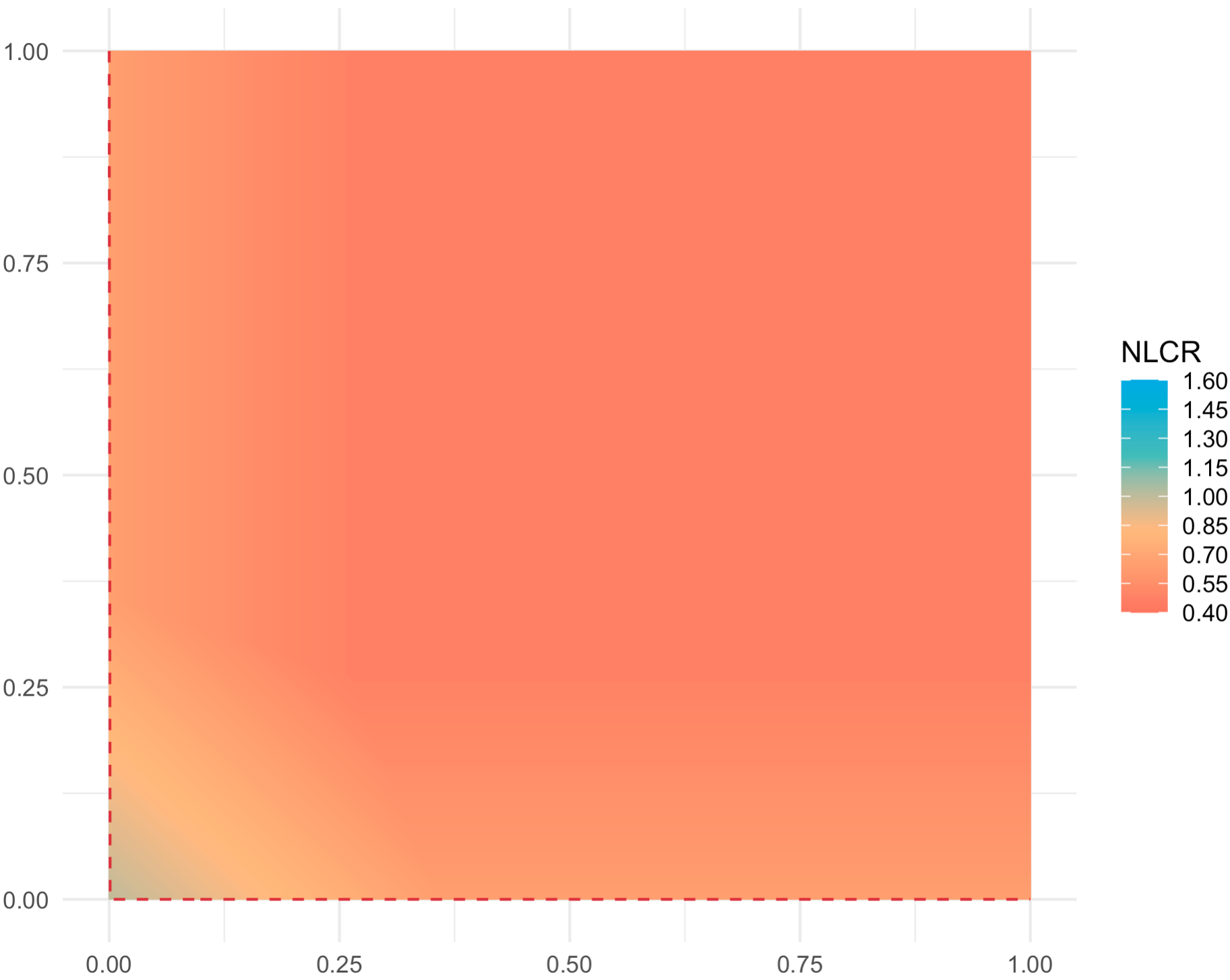}}
\caption{Weakly Bandable Example ($\alpha_1 = 0.4, \alpha_2 = 0.4$): The x-axis represents $\log p/\log n = \beta_1$ and y-axis represents $\log q/\log n = \beta_2$. The color represents the negative log convergence rates (NLCR) of the corresponding upper and lower bounds, where blue means a faster convergence rate. Panel (a) gives the overall error lower bound that applies for both Scenarios (i) and (ii); panel (b) gives the error upper bound under Scenario (i) (sub-Gaussian scenario), and the regions in the bottom left corner, surrounded by red dashed lines, corresponds to the rate-optimal region; panel (c) gives the upper bound under Scenario (ii) (finite fourth moment scenario), and the two red dashed lines represent the rate-optimal region under this scenario. }
\label{fig:wkbd1}
\end{figure}
\begin{figure}[H]
  \centering
    \subfigure[Lower Bound]{\includegraphics[width=0.49\textwidth]{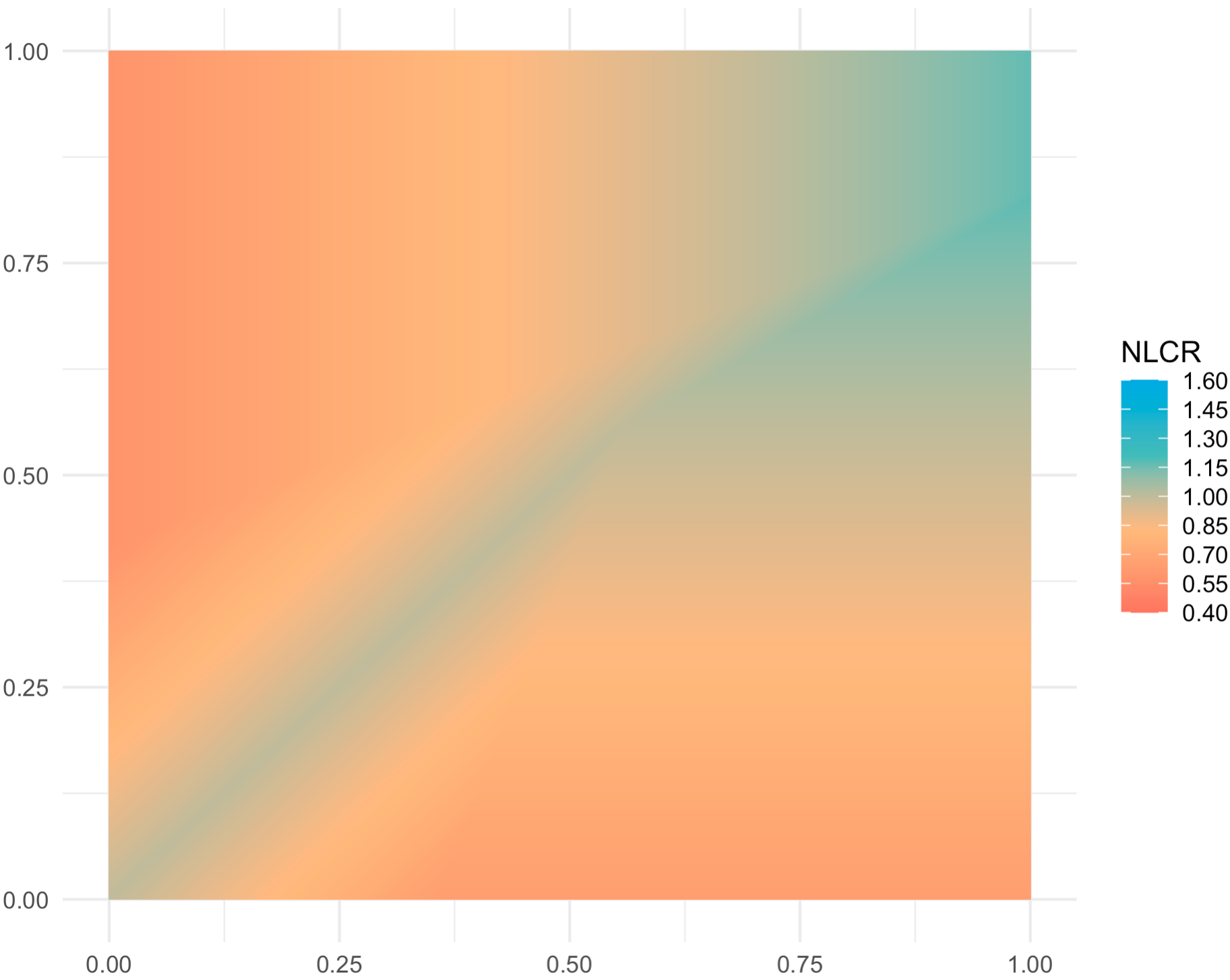}}
  \subfigure[Upper Bound for Scenario (i)]{\includegraphics[width=0.49\textwidth]{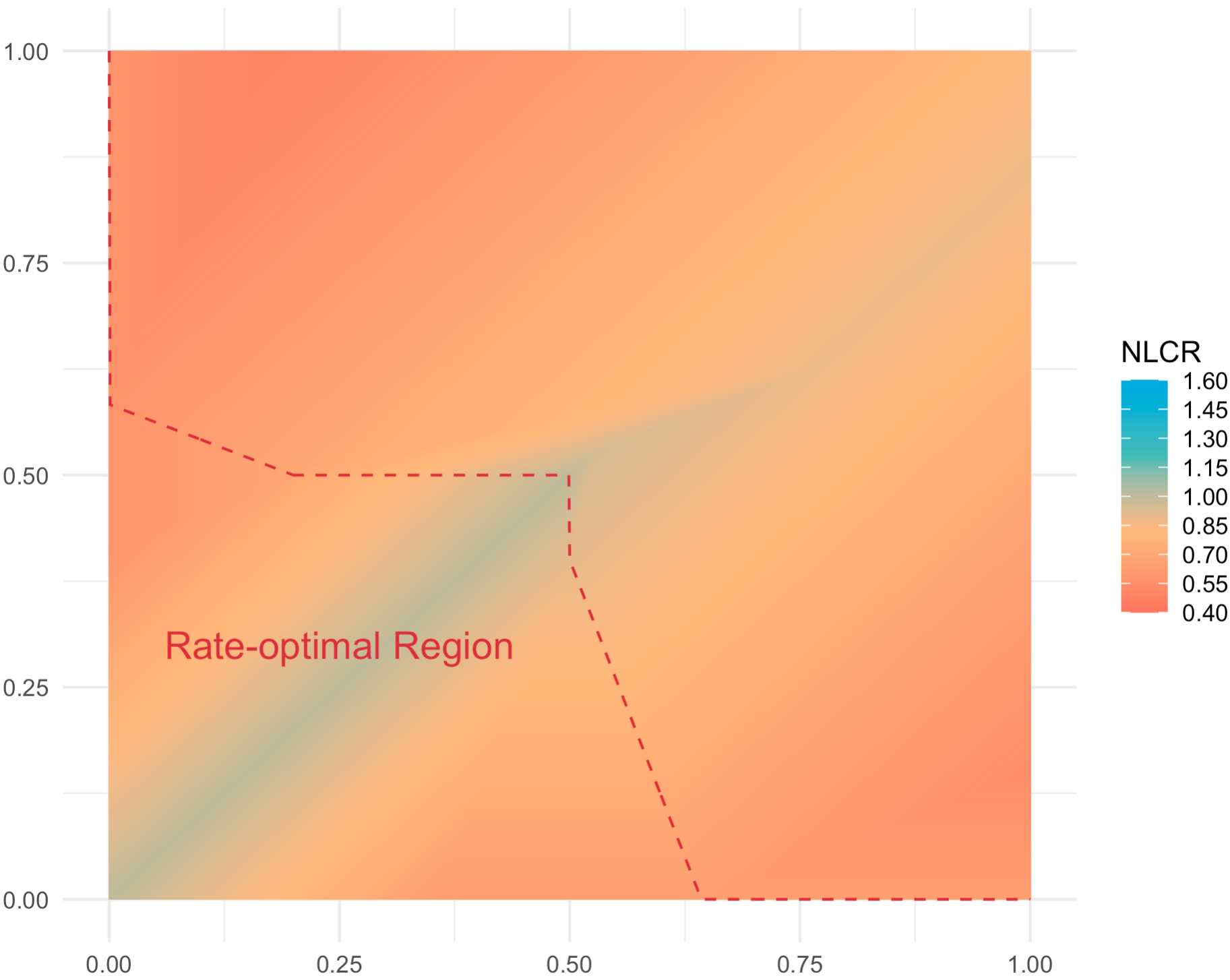}}
\\
  \subfigure[Upper Bound for Scenario (ii)]{\includegraphics[width=0.49\textwidth]{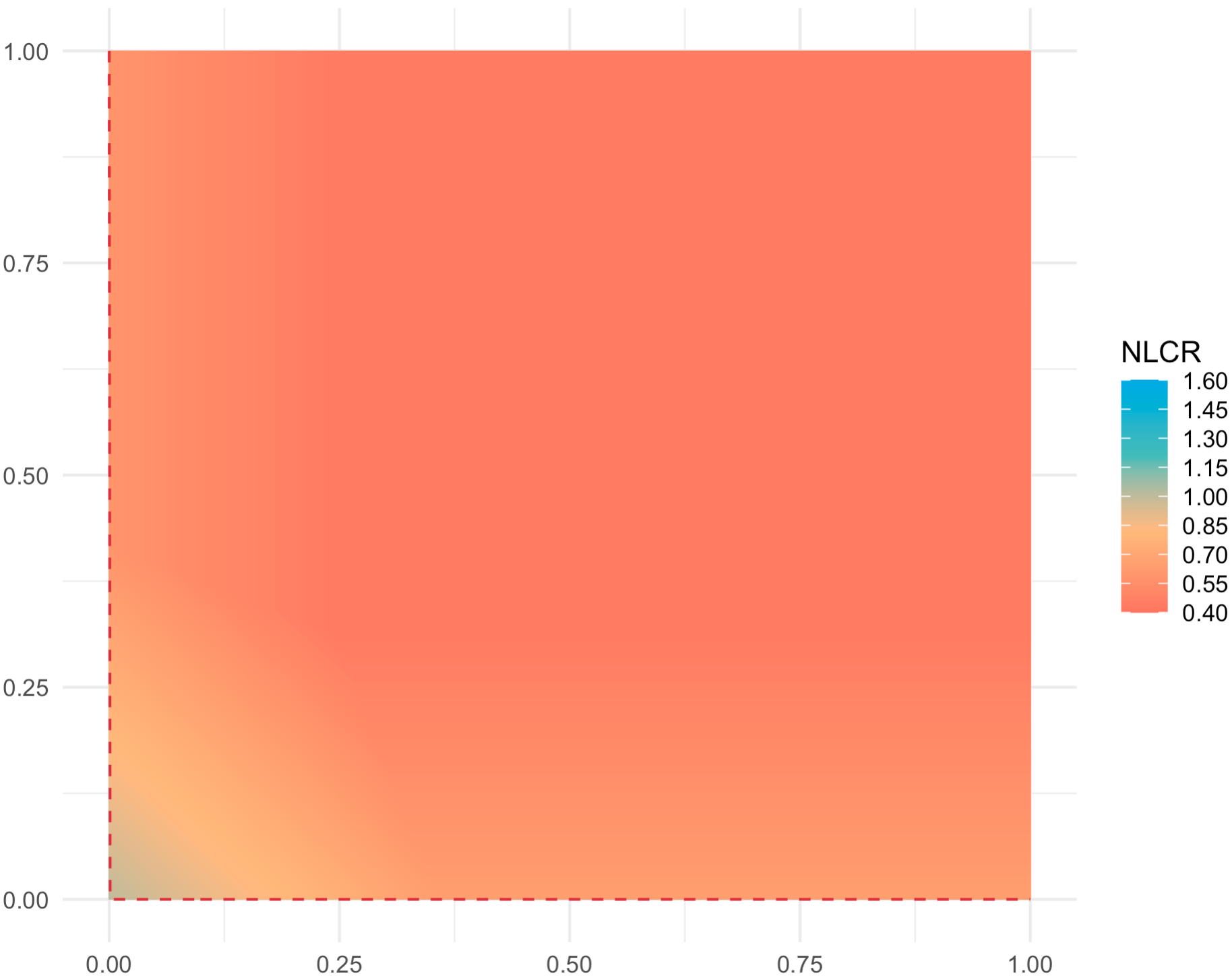}}
\caption{Weakly Bandable Example ($\alpha_1 = 0.4, \alpha_2 = 0.2$):  The x-axis represents $\log p/\log n = \beta_1$ and y-axis represents $\log q/\log n = \beta_2$. The color represents the negative log convergence rates (NLCR) of the corresponding upper and lower bounds, where blue means a faster convergence rate. Panel (a) gives the overall error lower bound that applies for both Scenarios (i) and (ii); panel (b) gives the error upper bound under Scenario (i) (sub-Gaussian scenario), and the regions in the bottom left corner, surrounded by red dashed lines, corresponds to the rate-optimal region; panel (c) gives the upper bound under Scenario (ii) (finite fourth moment scenario), and the two red dashed lines represent the rate-optimal region under this scenario. }
\label{fig:wkbd2}
\end{figure}
\begin{figure}[H]
  \subfigure[$\alpha_1 = 2, \alpha_2 = 2$]{\includegraphics[width=0.49\textwidth]{22k1}}
  \subfigure[$\alpha_1 = 2, \alpha_2 = 1$]{\includegraphics[width=0.49\textwidth]{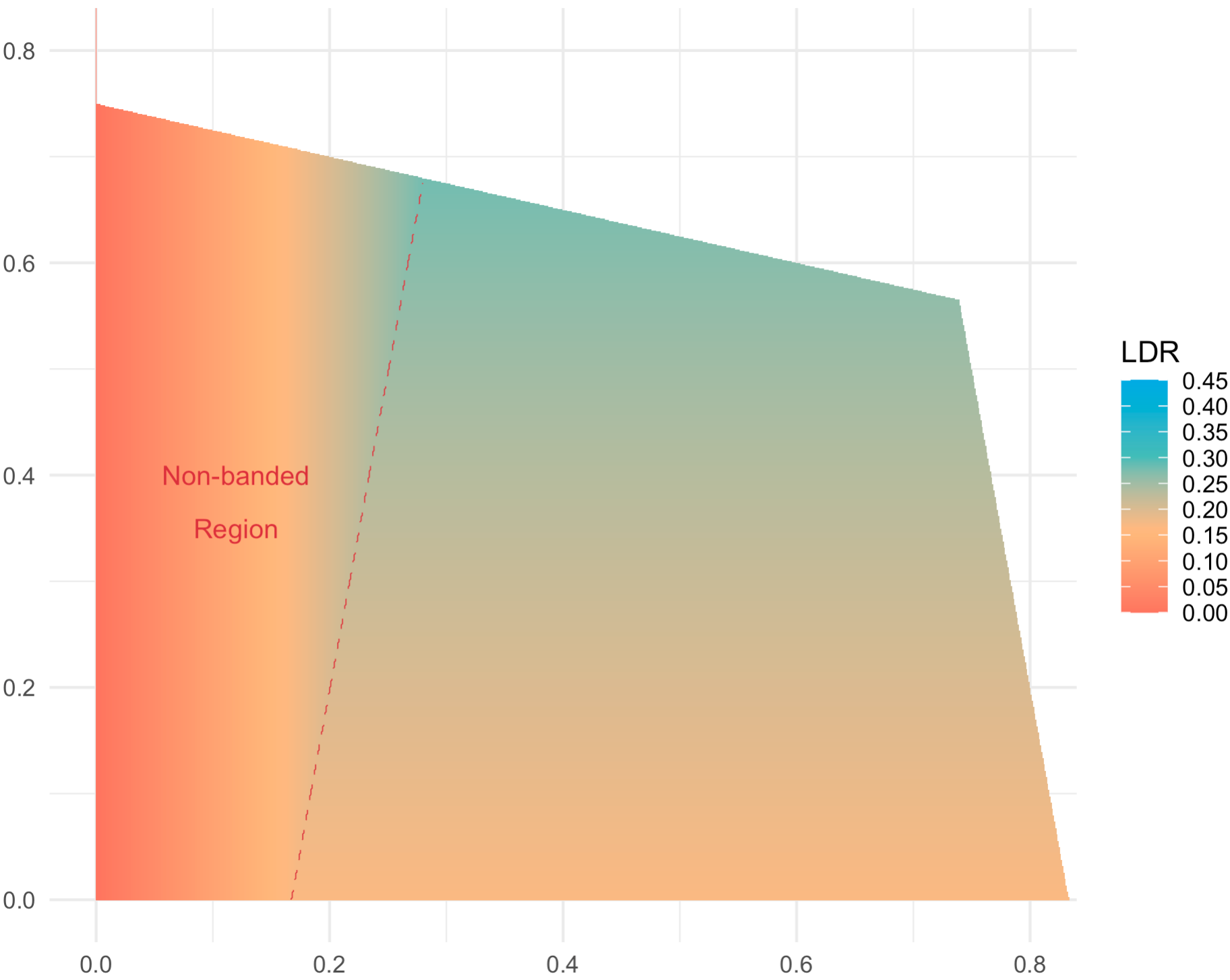}}
  \subfigure[$\alpha_1 = 0.4, \alpha_2 = 0.4$]{\includegraphics[width=0.49\textwidth]{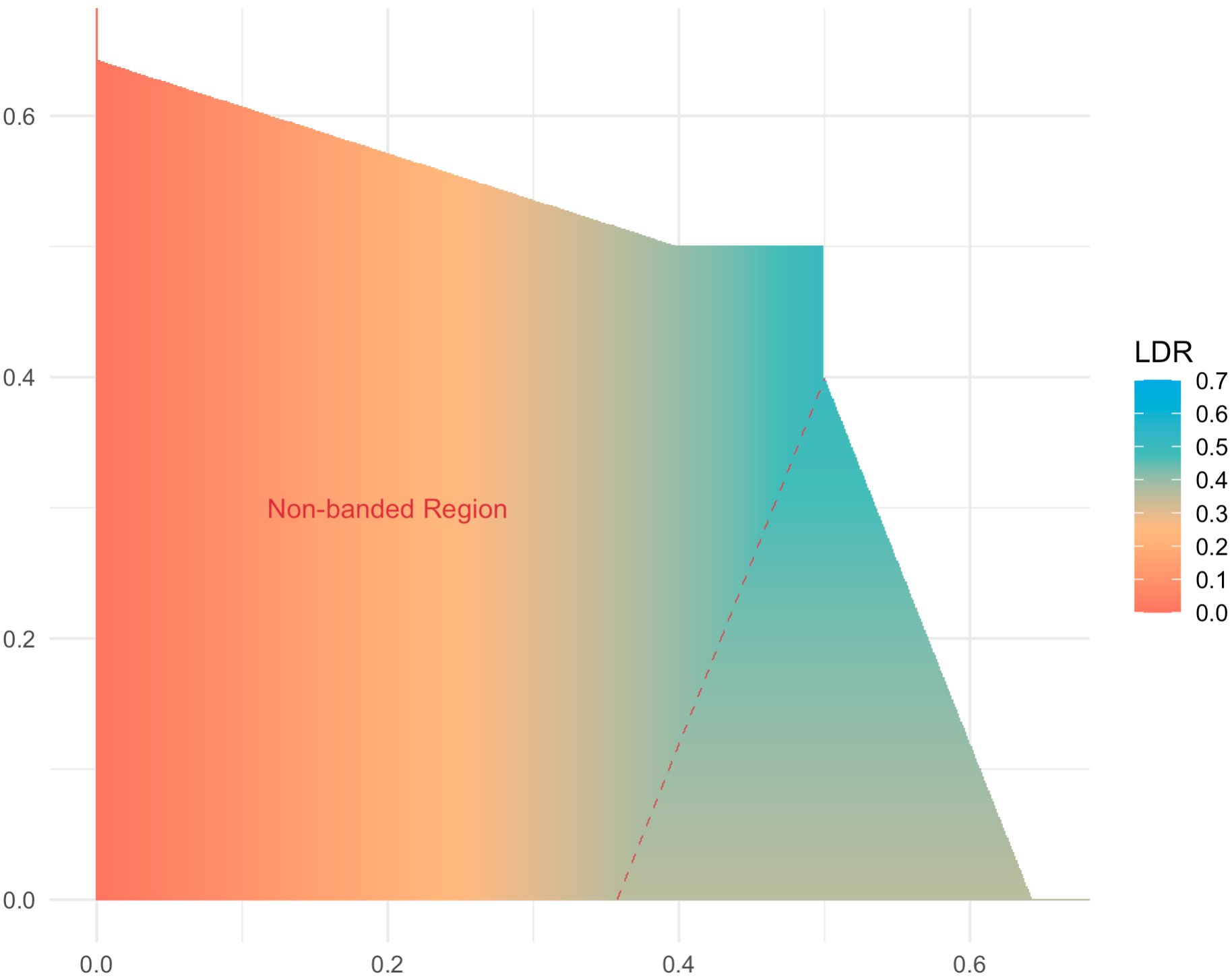}}
  \subfigure[$\alpha_1 = 0.4, \alpha_2 = 0.2$]{\includegraphics[width=0.49\textwidth]{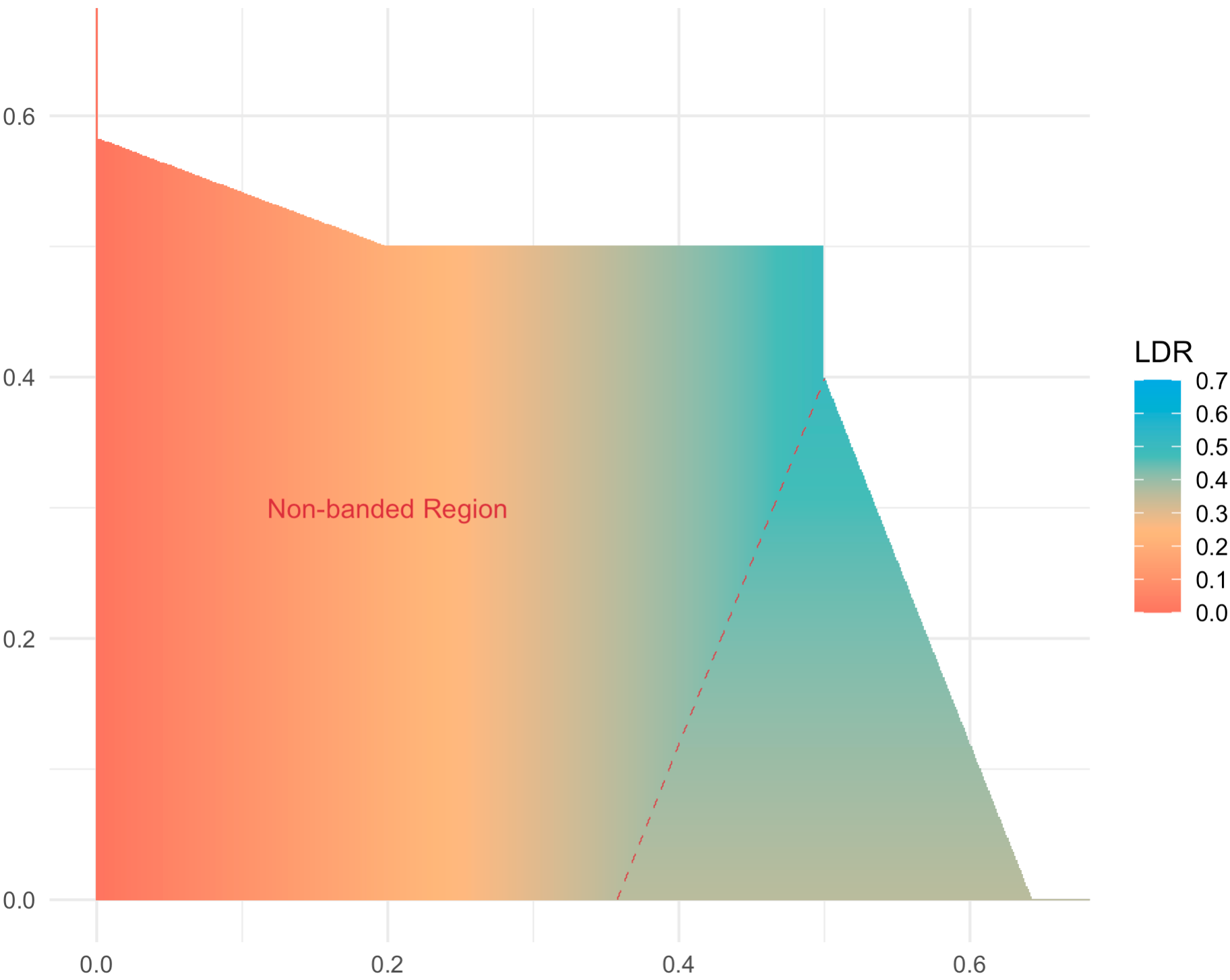}}
\caption{Phase transition phenomenon for $k_1$: The x-axis represents $\log p/\log n = \beta_1$ and y-axis represents $\log q/\log n = \beta_2$. Colored regions correspond to the rate-optimal regions for different $(\alpha_1,\alpha_2)$ values in panels (a)-(d), and the color represents the log divergence rates (LDR) of the optimal $k_1$ value, where a deeper color indicates a larger divergence rate. Red dashed line sets the boundary for non-banded region (to its left), where no banding is necessary to achieve optimal convergence rate.}
\label{fig:phase}
\end{figure}
\begin{figure}
   \centering
  \subfigure[Strongly Bandable Example with $(\alpha_1,\alpha_2) = (2,2)$]{\includegraphics[width=0.49\textwidth]{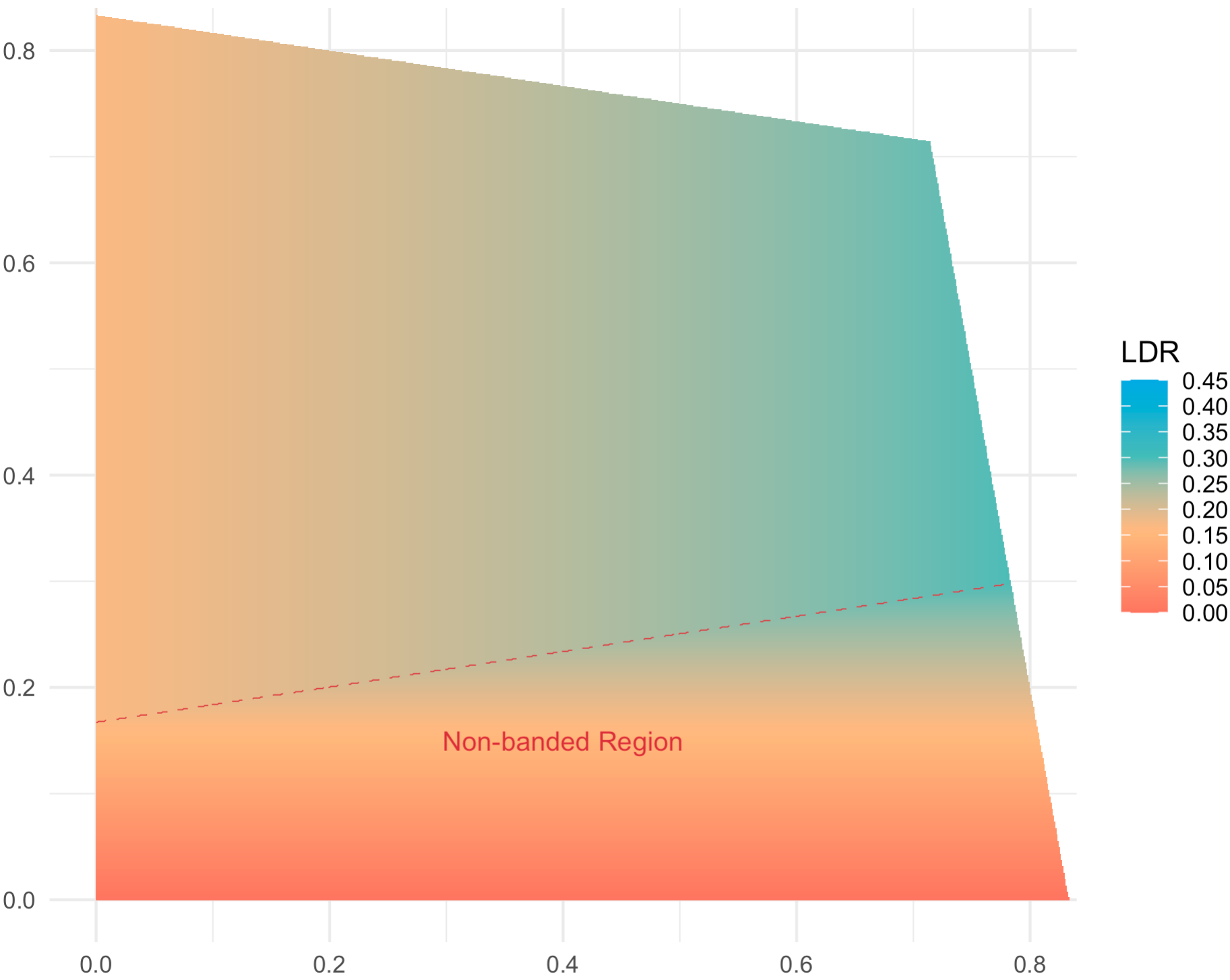}}
  \subfigure[Strongly Bandable Example with $(\alpha_1,\alpha_2) = (2,1)$]{\includegraphics[width=0.49\textwidth]{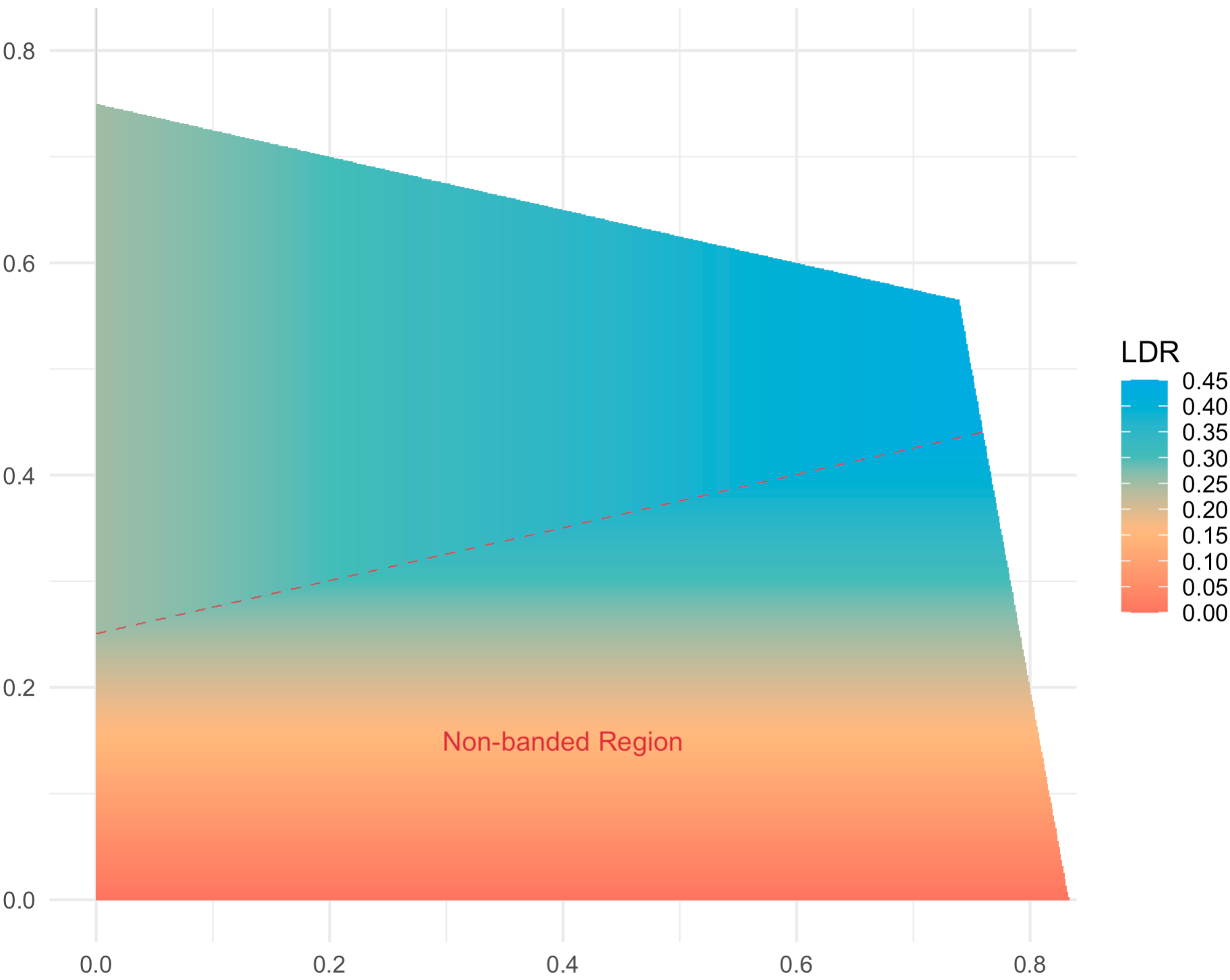}}

 \subfigure[Weakly Bandable Example with $(\alpha_1,\alpha_2) = (0.4,0.4)$]{\includegraphics[width=0.49\textwidth]{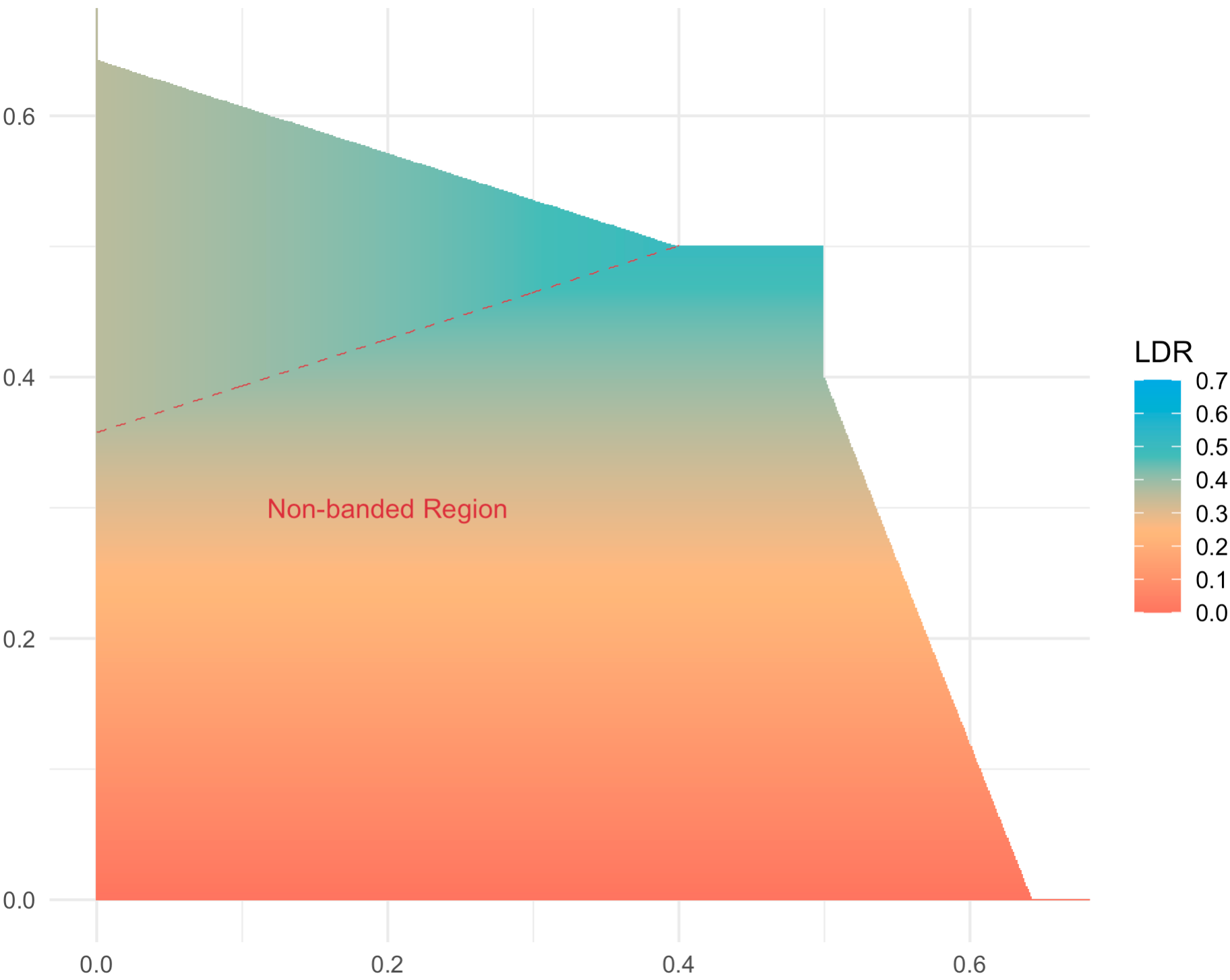}}
 \subfigure[Weakly Bandable Example with $(\alpha_1,\alpha_2) = (0.4,0.2)$]{\includegraphics[width=0.49\textwidth]{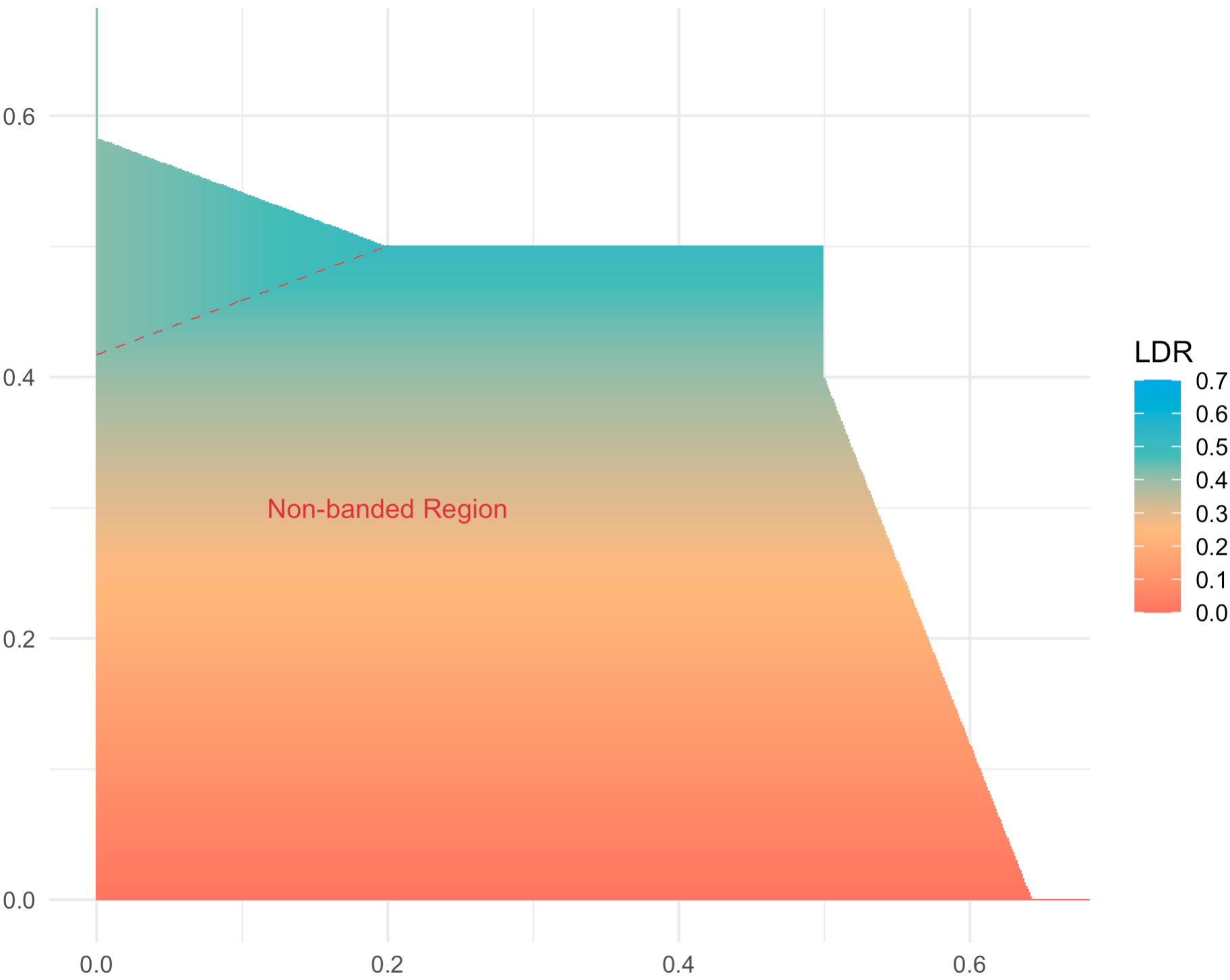}}

\caption{Phase transition phenomenon for $k_2$: The x-axis represents $\log p/\log n = \beta_1$ and y-axis represents $\log q/\log n = \beta_2$. Colored regions correspond to the rate-optimal regions for different $(\alpha_1,\alpha_2)$ values in panels (a)-(d), and the color represents the log divergence rates (LDR) of the optimal $k_1$ value, where a deeper color indicates a larger divergence rate. Red dashed line sets the boundary for non-banded region (to its left), where no banding is necessary to achieve optimal convergence rate.}
\label{fig:sup:ratecompare}
\end{figure}

%Through the above examples, we have the following observations, which comply with the theoretically identified features  in Sections  \ref{sec:T} and \ref{sec:complexeffect}.
The above examples suggest the following findings:
\begin{itemize}
\item Comparing Figures \ref{fig:sgbd1}--\ref{fig:wkbd2} (b) and (c), for fixed regimes of $p,q$, the convergence rate we obtain under Scenario (i) is much faster than the convergence rate  we obtain under Scenario (ii), which is reasonable as suggested by Remark \ref{rk:T3dg} in the main paper. 
\item In Figures \ref{fig:sgbd1}--\ref{fig:wkbd2}, while keeping $p,q$ unchanged, all upper and lower bounds converge faster when $\alpha_1,\alpha_2$ become larger. 
\par
Within the rate-optimal region, consider the case where $pq = n^{\beta_1 + \beta_2}$ diverges under a fixed rate $n^{\tilde{r}}$, i.e.,  $\beta_1 + \beta_2 = \tilde{r}$ is fixed. Figures \ref{fig:sgbd1} (b) and \ref{fig:wkbd1} (b) show that when $\alpha_1 = \alpha_2$, the error rate of our proposed estimator is the sharpest when $p$ and $q$ diverge equally fast. Figures \ref{fig:sgbd2} (b) and \ref{fig:wkbd2} (b) show that, when $\alpha_1 \approx \alpha_2$, the error rate of our proposed estimator is the sharpest when $\beta_1 \approx \beta_2$, i.e., $p$ and $q$ are approximately equally divergent.
%\par
%The above observations exactly comply with the discussion in Section \ref{sec:effect:rate}.
\item From Figures \ref{fig:sgbd1}--\ref{fig:wkbd2}, we can see that when $\alpha_1,\alpha_2$ become smaller, i.e., $\M\Sigma_1^*,\M\Sigma_2^*$ become less bandable, the rate-optimal region becomes smaller and eventually collapses to the $[0,0.5]\times [0,0.5]$ square region. %, which can be precisely predicted by our discussion in Section \ref{sec:effect:region}.  
\item In Figure \ref{fig:phase}, for any specific $\alpha_1$ and $\alpha_2$, when $p,q$ are under the moderate high-dimensional regime, the phase transition phenomenon of $k_1$ selection can be observed for both strongly and weakly bandable examples. For a fixed divergence rate of $ q $, when $p$ diverges slower than a threshold rate there is no benefit of regularizing over the $p$ direction, i.e., $k_1 =\tilde{p}$ already guarantees an optimal rate in Theorem \ref{T2}. When $p$ diverges faster than the threshold rate, the optimal banded rate over the $p$ direction will keep the same regardless of the value of $p$. Moreover, when $\beta_2$ becomes larger (i.e., $q$ diverges faster), the corresponding threshold of $p$ also becomes larger. \par
On the other hand, comparing Figure \ref{fig:phase} panels (a) and (c), we find that the  non-banded region becomes larger, i.e., the threshold for phase transition of $k_1$ becomes larger, as $\alpha_1$ and $\alpha_2$ decrease. Moreover,  the entire rate-optimal region will collapse to the $[0,0.5]\times [0,0.5]$ square region as $\max(\alpha_1,\alpha_2) \rightarrow 0$. This region is also the non-bandable region when $\max(\alpha_1,\alpha_2) \rightarrow 0$.  For $k_2$, we observe similar findings, as shown in Figure \ref{fig:sup:ratecompare}. 
%we expect the entire rate-optimal region to coincide with the non-banded region as $\max(\alpha_1,\alpha_2) \rightarrow 0$, i.e., no banding is necessary for $\max(p,q) =O(n^{1/2})$ while $\max(\alpha_1,\alpha_2) \rightarrow 0$. We now explain these findings with precise theoretical analysis.
\end{itemize}

\subsection{Empirical Justification of Theoretical Error Rates}\label{Sec28}
 In this subsection, we justify the established error rates using simulations. In particular, we consider $p = q = \lfloor n^{\beta} \rfloor$, where $\beta>0$ is the divergence rate of $p,q$. The $\{\vecc (\Xb_i)\}_{i = 1}^n $ are $n$ i.i.d $\M N({\bf 0}, \boldsymbol\Sigma_2\otimes\boldsymbol\Sigma_1),$  with $\M\Sigma_1,\M\Sigma_2$ satisfying,
\begin{eqnarray}\label{sigma12:structure}
\sigma_{l_a, m_a}^{(a)}&= & 
\begin{cases}
2, & l_a = m_a
\\
|l_a-m_a|^{-\alpha_a - 1} & l_a \neq m_a
\end{cases},
\end{eqnarray}
where $a\in\{1,2\}$, $\alpha_1 = \alpha_2 = 2$. So both $\M\Sigma_1,\M\Sigma_2$ are in the matrix class $\mathcal{M}(\varepsilon_0,\alpha = 2)$.  
\par
For a specific setting of $(p,q,n)$, we numerically estimate the empirical error of our proposed estimator as follows. We first use $30$ Monte-Carlo (MC) pre-rounds to select optimal $\hat{k}_1$ and $\hat{k}_2$. We take the medians among selected $\hat{k}_1$s and $\hat{k}_2$s over 30 pre-rounds, as the estimation of optimal $\hat{k}_{1,\text{opt}}, \hat{k}_{2,\text{opt}}$. We then run $300$ main MC replicates  under the corresponding parameter setting $(p,q,n)$. In the $j$th replicate, we fit our proposed banded estimator with $k_1 = \hat{k}_{1,\text{opt}}, k_2 = \hat{k}_{2,\text{opt}}$, over $n$ i.i.d. generated $p\times q$ samples $\{\M X_i\}_{i = 1}^n$, and calculate the  normalized Frobenius-norm error $E_{p,q,n}^{(j)} = \big\|\hat{\M\Sigma}_2^{\mathcal{B}}(\hat{k}_{2,\text{opt}})\otimes \hat{\M\Sigma}_1^{\mathcal{B}}(\hat{k}_{1,\text{opt}}) - \M\Sigma\big\|^2_\F/pq$. We plot the log empirical error $\log(\bar{E}_{p,q,n})$, where $\bar{E}_{p,q,n} = \sum_{j = 1}^{300}E_{p,q,n}^{(j)}/300$.
\par
We consider nine different divergence regimes of $p,q$, in terms of different $\beta$, with $\beta = 0.3,0.35,\dots,0.70$. Under each divergence regime, we run simulations with $n = 100,200,\dots,1500$ when $\beta \leq 0.6$, and with $n = 100,200,\dots,1000$ when $\beta = 0.65,0.70$. Here we use a larger sample size for $\beta \leq 0.6$ to ensure a small rounding error for small $\beta$. 
\par
Figure \ref{fig:sgbd1} (b) shows the minimax rate-optimal region of our proposed banded estimator with $\alpha_1 = \alpha_2 = 2$. Since $p = q = \lfloor n^{\beta} \rfloor \asymp n^{\beta}$ and $0.3\leq \beta\leq 0.7$, our proposed estimator is minimax rate-optimal. As $\M\Sigma_1, \M\Sigma_2$ belong to the matrix class $\mathcal{M}({\varepsilon_0,2})$, the theoretical error rate should attain the lower bound of \eqref{T:low} with $\alpha_1 = \alpha_2 = 2$ such that, 
\bee\label{simu:them1}
\E\big({\|\hat{\M \Sigma}_n - \M \Sigma_2^* \otimes \M \Sigma_1^*\|_\F^2}/{pq}\big) \asymp (np)^{-5/6}.
\ee
\par
\begin{figure}
  \centering
    \subfigure[Fitting plots with $\beta = 0.55$]{\includegraphics[width=0.49\textwidth]{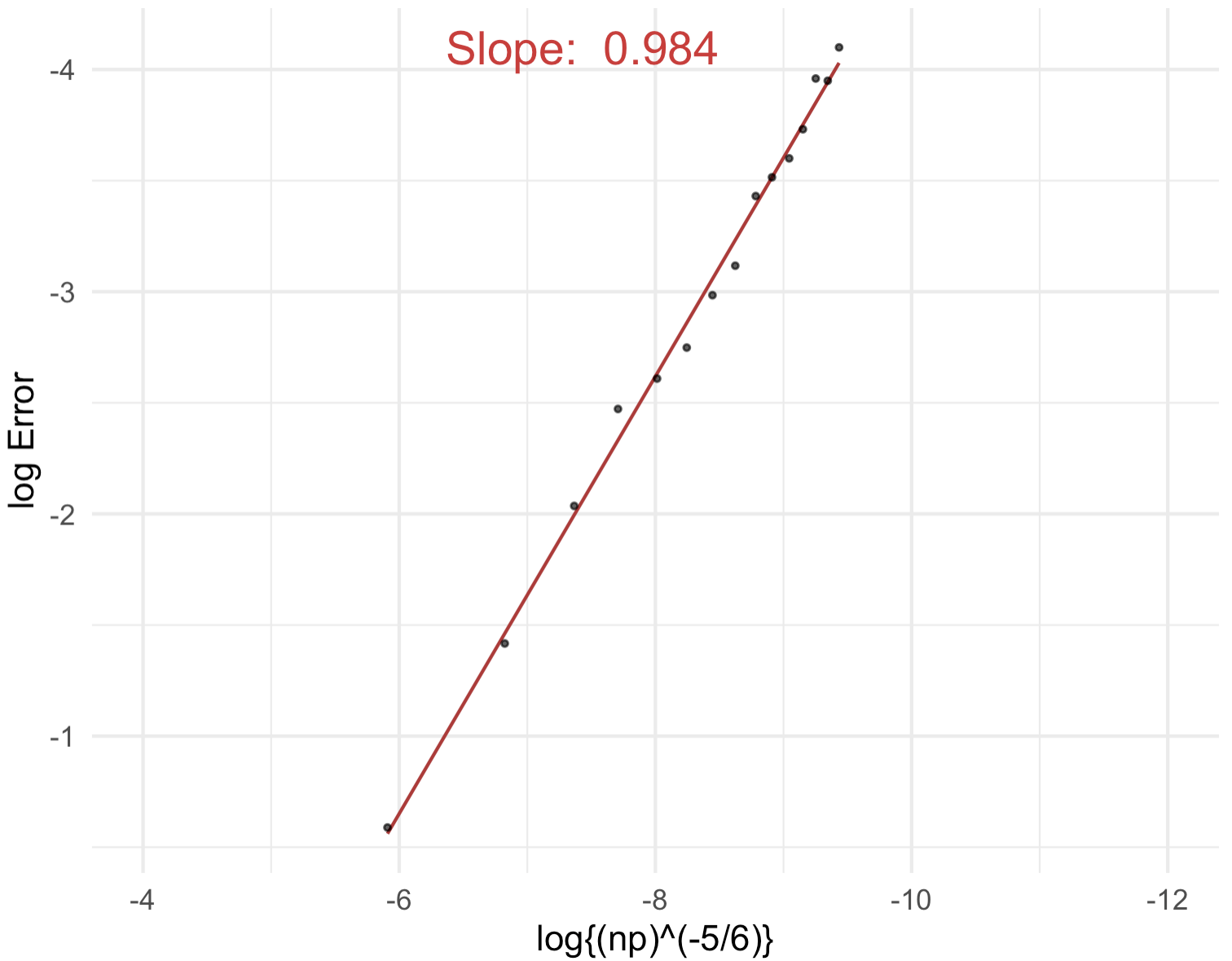}}
  \subfigure[Summary of $\hat{s}_\beta$ and $\tilde{s}_\beta$]{\includegraphics[width=0.49\textwidth]{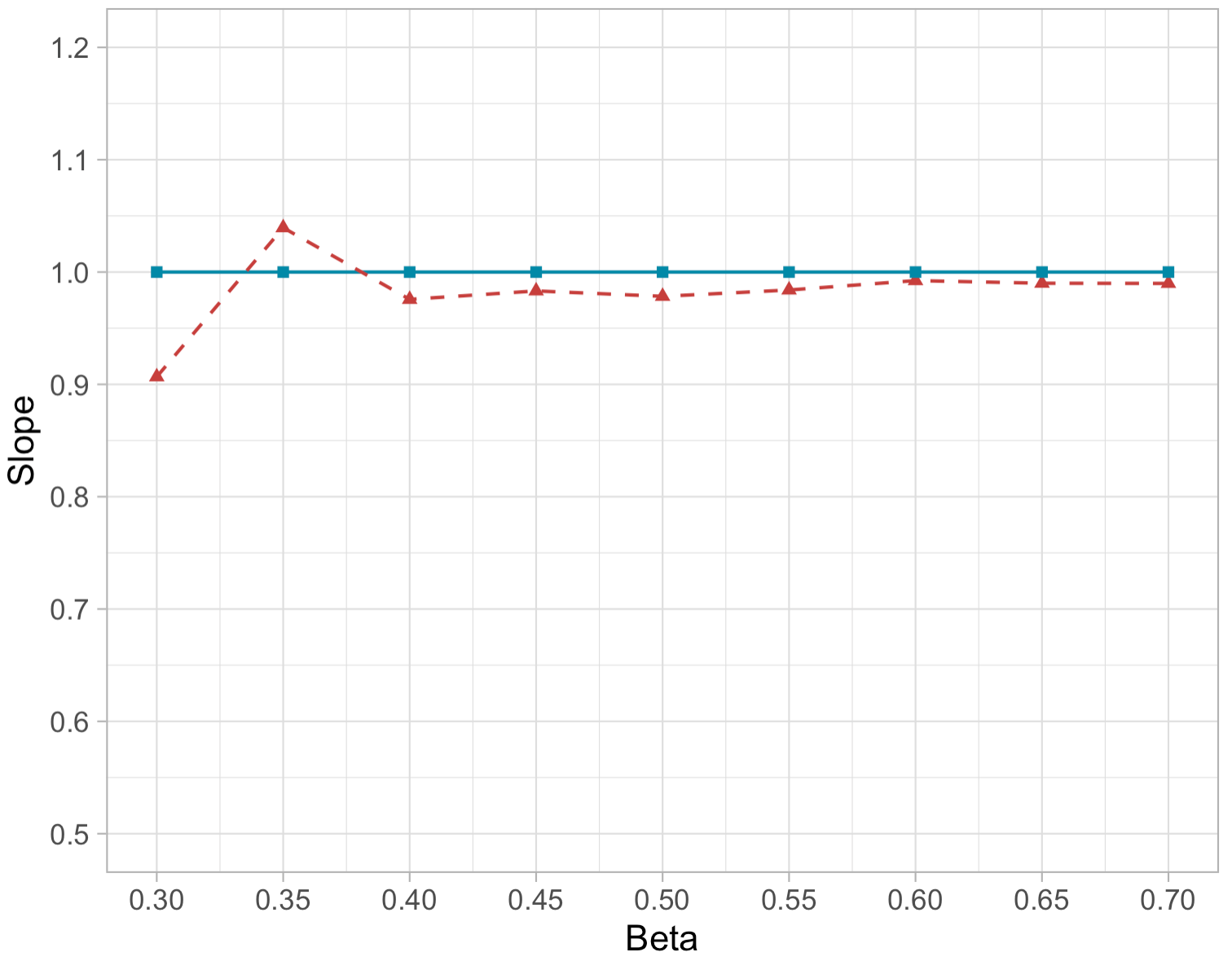}}
\caption{Panel (a): pairs of $(\log\big\{(np)^{-5/6}\big\},  \log\bar{E}_{p,q,n})$ for different values of $n$ with $\beta = 0.55$ and the regression line (red). Panel (b): estimated slopes (red) with $\beta$ ranging from $0.3$ to $0.7$, and the blue line representing the benchmark value of $1$.
}
\label{fig:simu:rate}
\end{figure}
We first fix $\beta = 0.55$  and plot the pairs of $(\log\big\{(np)^{-5/6}\big\},  \log\bar{E}_{p,q,n})$ for different values of $n$ in Figure \ref{fig:simu:rate} (a). We observe a clear linear relationship (estimated slope of $.984$ obtained by fitting a simple linear regression) and all points align closely to the fitted line,  which confirms our derived rate in \eqref{simu:them1}. %Fitted $\hat{s}_\beta $ and $\tilde{s}_\beta $ are very close to the theoretically predicted true slopes $1.291$ and $1$, which indicates that our rate-optimal error rate derived in Section \ref{sec:T} for our proposed estimator is sharp under this particular divergence regime of $p,q$. Similar results for other $\beta$-divergence regimes, are included in Figure \ref{supp:fig:slope} in the Supplementary File.\par
We further let $\beta$ take different values in $\{.30,.35,\ldots,.70\}$,  repeat the above procedure of slope estimation, and present the estimated slopes in Figure \ref{fig:simu:rate} (b) and Figure \ref{supp:fig:slope}. We find that the estimated slopes are very close to the ground truth $1$, which again confirms our theoretical error rate. The slight bias when $\beta = 0.3$ and $ 0.35$ can be explained by the rounding error caused by small values of $p$ and $q$, e.g., when $\beta = 0.3$ and $n$ is $1500$, $p,q=\lfloor n^{\beta} \rfloor$ are only around $8$.
\begin{figure}
  \centering
    \subfigure[$\beta=0.3$: Estimated Slope 0.907.]{\includegraphics[width=0.32\textwidth]{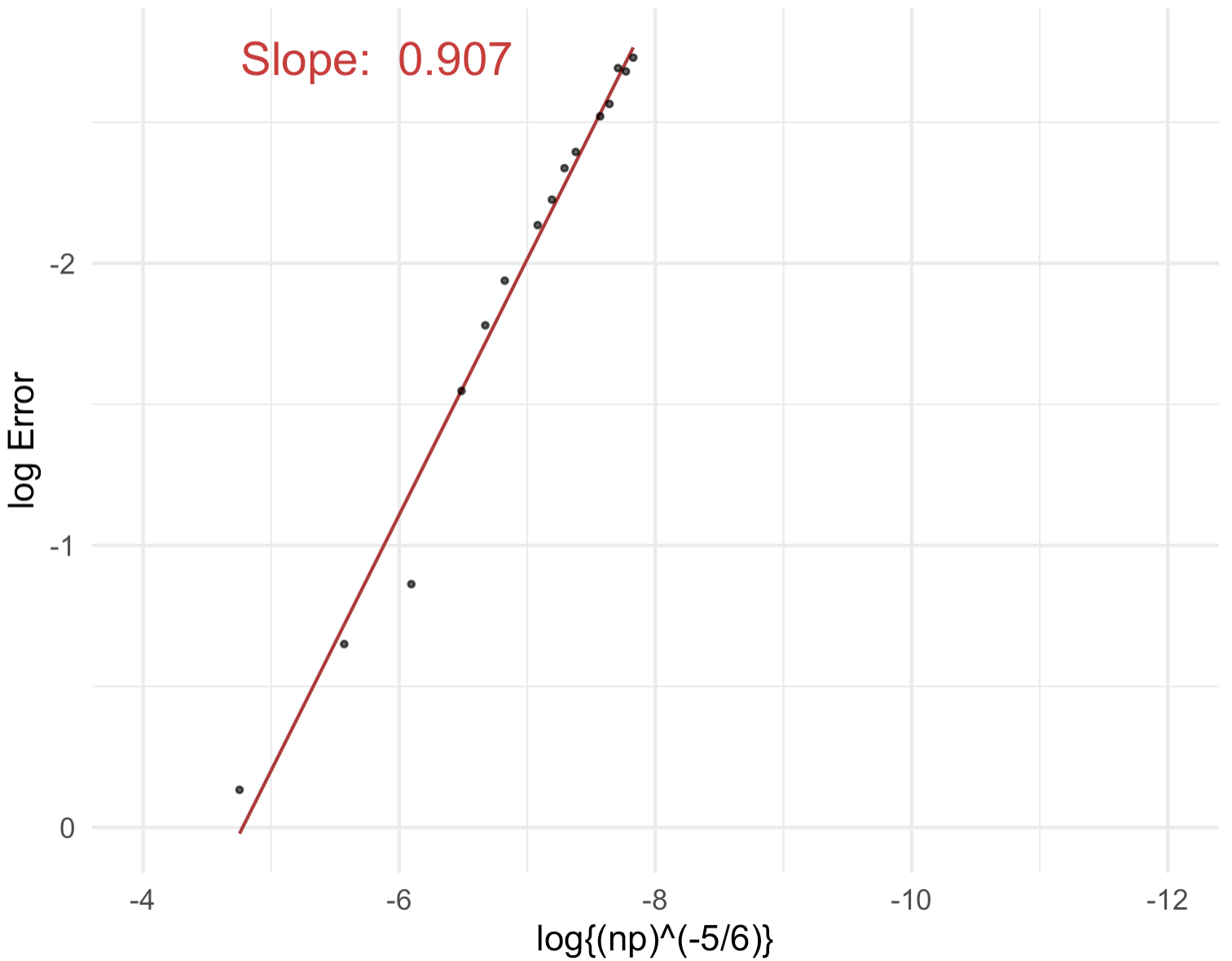}}
  \subfigure[$\beta=0.35$: Estimated Slope 1.039.]{\includegraphics[width=0.32\textwidth]{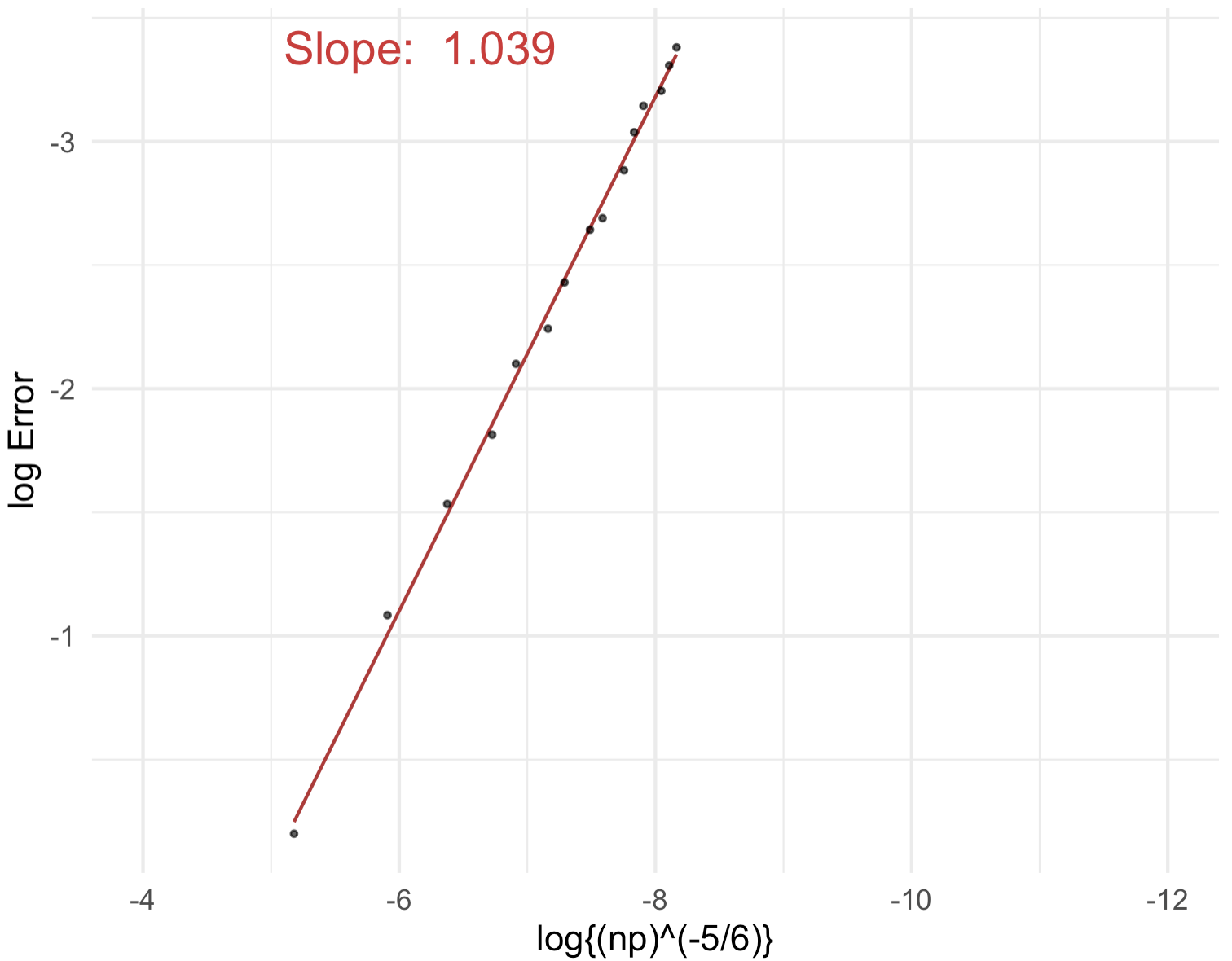}}
\subfigure[$\beta=0.4$: Estimated Slope 0.976.]{\includegraphics[width=0.32\textwidth]{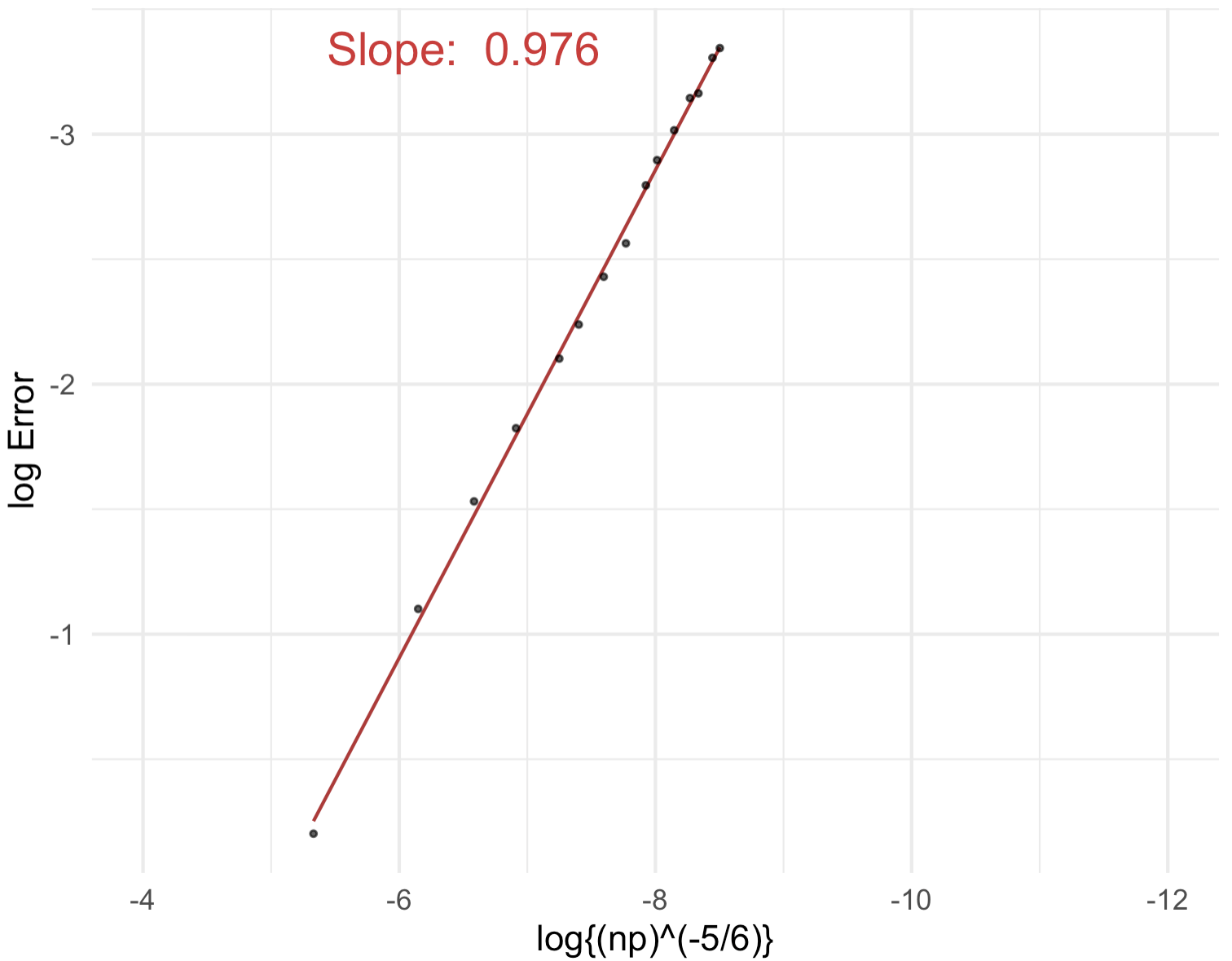}}
\\
  \subfigure[$\beta=0.45$: Estimated Slope 0.983.]{\includegraphics[width=0.33\textwidth]{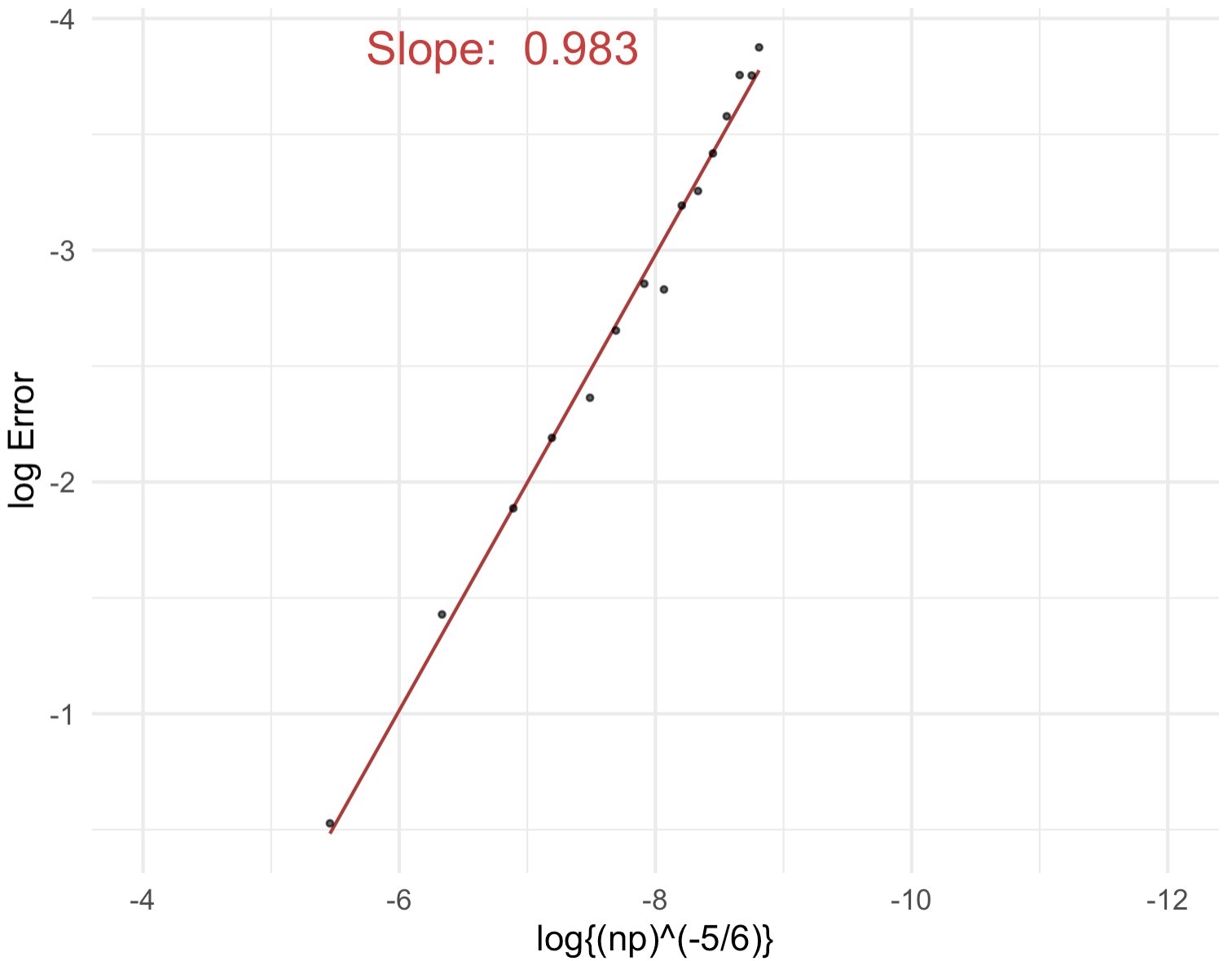}}
\subfigure[$\beta=0.5$: Estimated Slope 0.978.]{\includegraphics[width=0.32\textwidth]{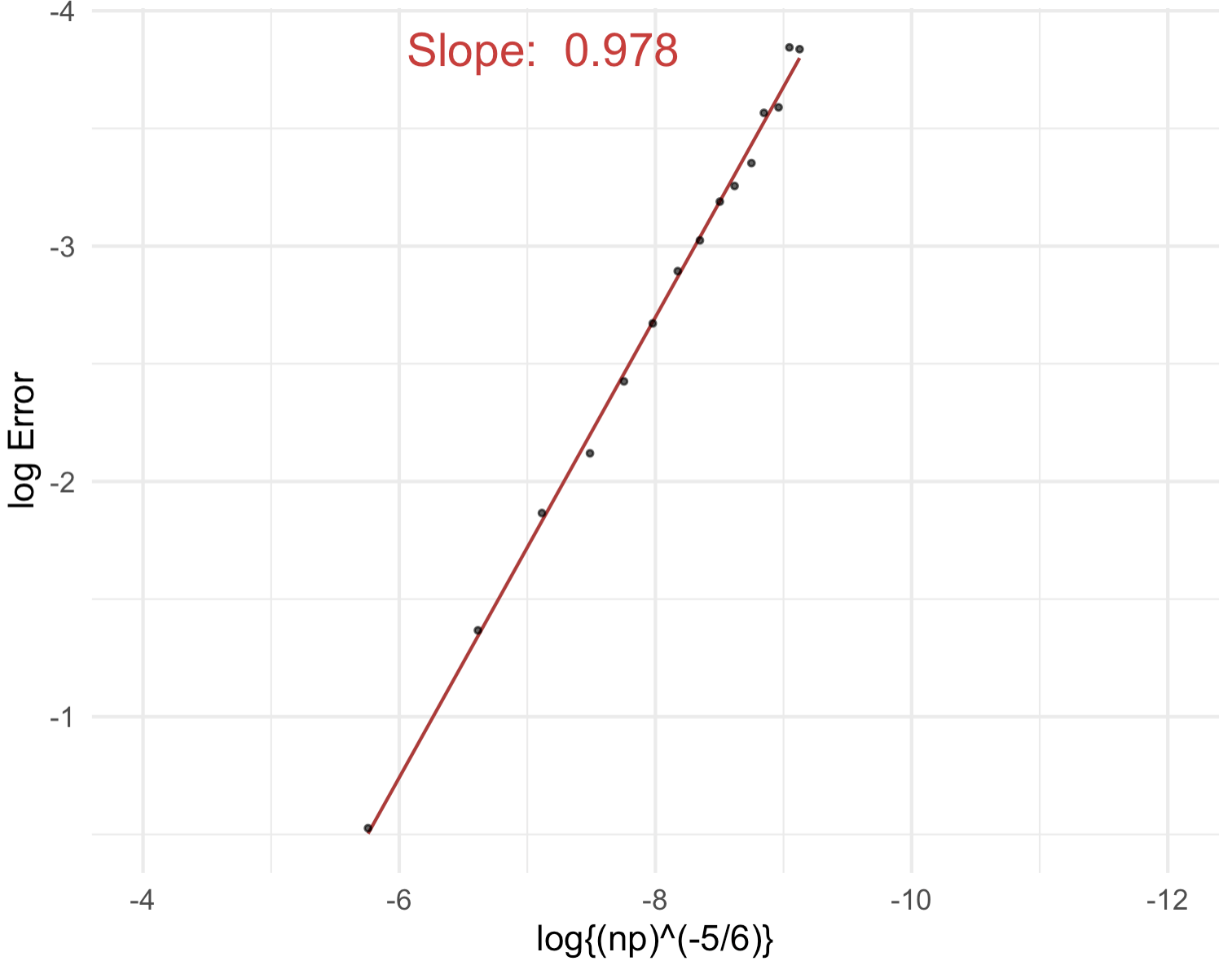}}
\subfigure[$\beta=0.55$: Estimated Slope 0.984.]{\includegraphics[width=0.32\textwidth]{simu_rate_055}}
\\
  \subfigure[$\beta=0.6$: Estimated Slope 0.992.]{\includegraphics[width=0.32\textwidth]{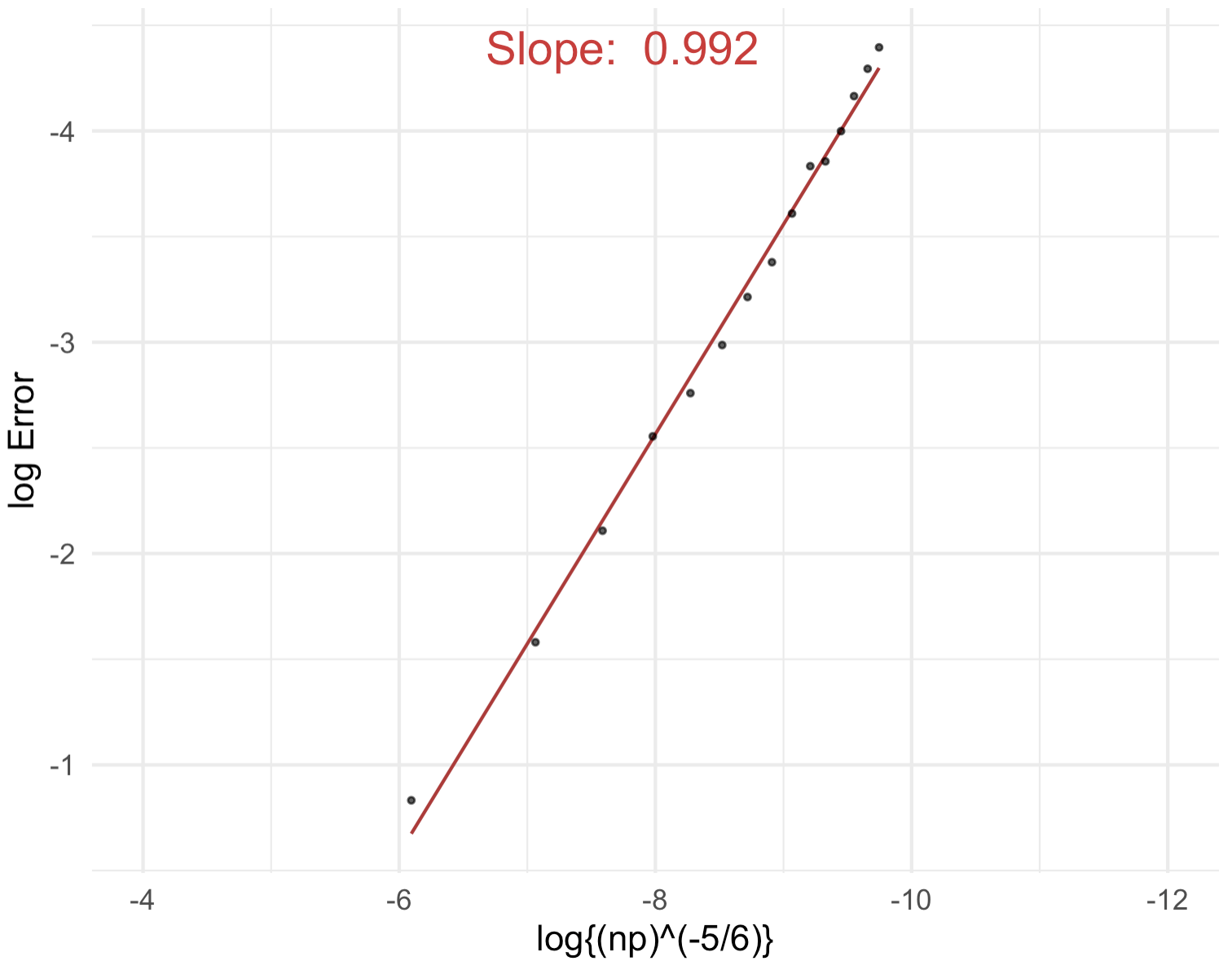}}
    \subfigure[$\beta=0.65$: Estimated Slope 0.990.]{\includegraphics[width=0.32\textwidth]{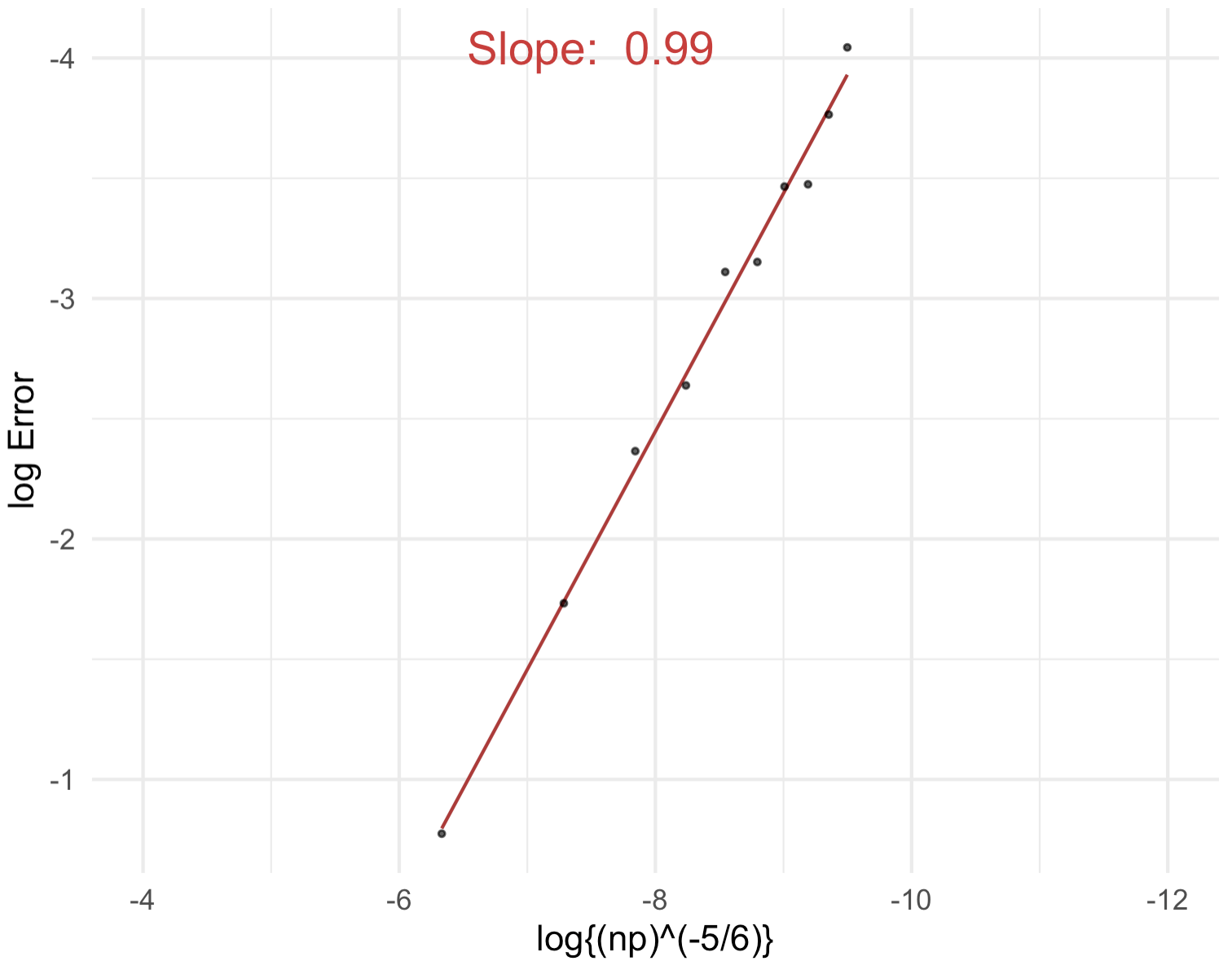}}
  \subfigure[$\beta=0.7$: Estimated Slope 0.990.]{\includegraphics[width=0.32\textwidth]{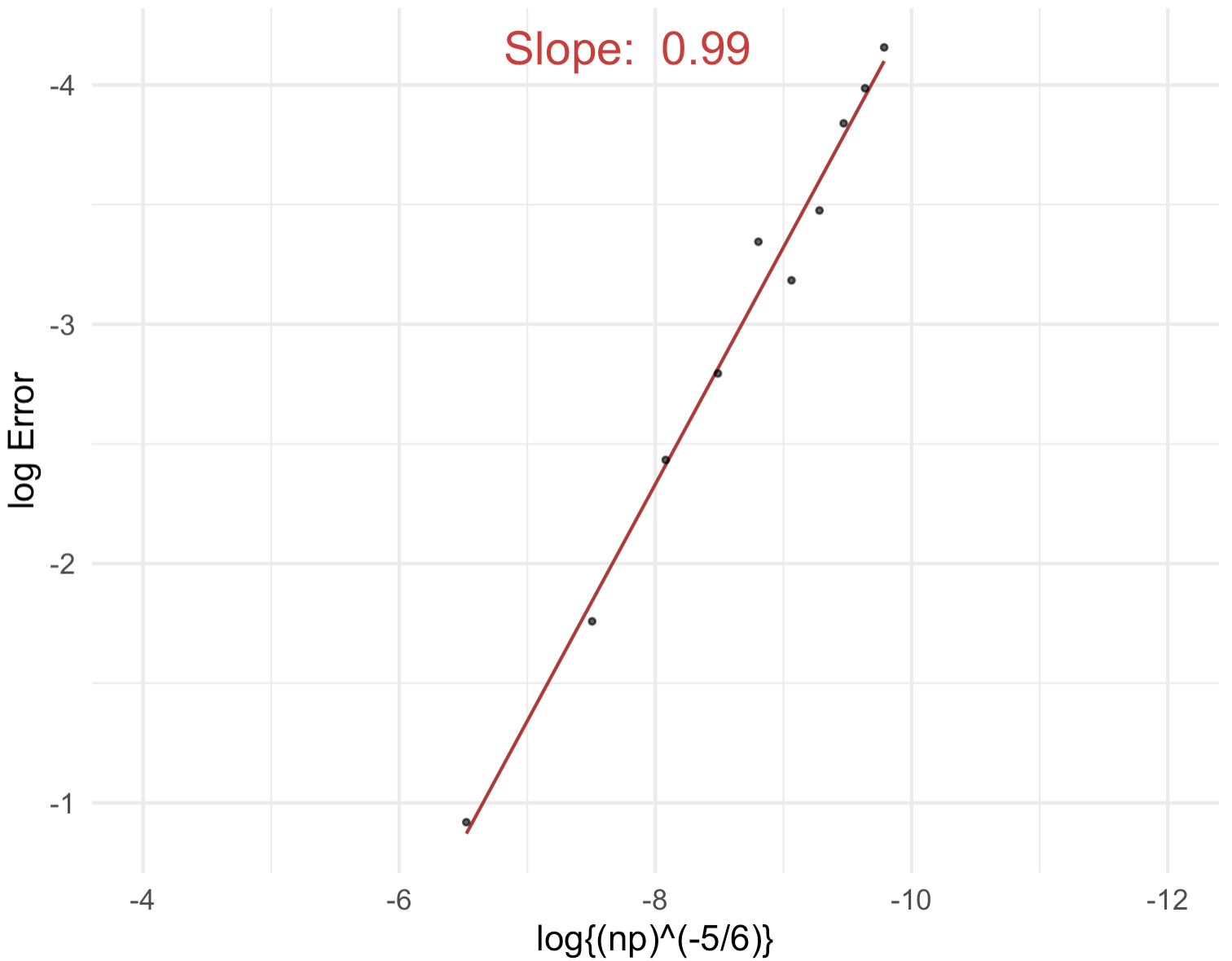}}

\caption{In each panel, all simulated pairs of $(\log\big\{(np)^{-5/6}\big\},  \log\bar{E}_{p,q,n})$ (bullet points) are plotted with specific $\beta$. The fitted regression line (red) and the corresponding estimated slope is reported, while the true slope is $1$ for all panels.
}
\label{supp:fig:slope}
\end{figure}

\subsection{Concentration of Individual Covariance Components}\label{sec:indcov} 
In practice, it may also be of interest to study the covariance structure over the row and/or column directions of the matrix-valued data. Therefore, we study the concentration of our proposed individual matrix estimators $\hat{\M \Sigma}_1^{\eta}(k_1),\hat{\M \Sigma}_2^{\eta}(k_2)$ for $\eta\in\big\{\mathcal{B},\mathcal{T}\big\}$. %to the corresponding truths $\M\Sigma^*_1$ and $\M\Sigma_2^*$. 
\par
As shown in Proposition \ref{decomposition:uniqueness}, $\hat{\M \Sigma}_1^{\eta}(k_1),\hat{\M \Sigma}_2^{\eta}(k_2)$ are unique up to some constant $c\neq 0$ multiplication. In other words, $\hat{\M \Sigma}_1^{\eta}(k_1),\hat{\M \Sigma}_2^{\eta}(k_2)$ are identifiable only up to scale. To make the theoretical analysis meaningful,  for $a = 1$ or $2$, we normalize a particular $\hat{\M \Sigma}_a^{\eta}(k_a)$ to $c_{a,\eta}\hat{\M \Sigma}_a^{\eta}(k_a)$. Here $c_{a,\eta}$ is a constant determined by the truth ${\M \Sigma}_a^{*,\eta}$, the selected bandwidth $k_a$, and  $\hat{\M \Sigma}_a^{\eta}(k_a)$.
\par
We will show the asymptotic results under Scenario (i). Recall 
\bee\nonumber
r_{\text{var}\mid \hat{\M \Sigma}_2\otimes \hat{\M\Sigma}_1}(k_1,k_2\mid p,q) \equiv \begin{cases}\frac{k_1}{qn} + \frac{k_2}{pn},& pk_1 + qk_2 \precsim n
\\
\big(\frac{k_1k_2}{n}\big)\wedge\big(\frac{pk^2_1}{qn^2} + \frac{qk^2_2}{pn^2}  \big), & pk_1 + qk_2 \succ n.
\end{cases}
\ee
as defined in \eqref{def:rwhole}, which is the error bound of term $E_1$. Similarly, we define
\bee\label{def:ra}
&r_{\text{var}\mid \hat{\M\Sigma}_1}(k_1,k_2\mid p,q) \equiv \begin{cases}\frac{k_1}{qn}& pk_1 \precsim n
\\
\frac{k_1k_2}{n}\wedge \frac{pk^2_1}{qn^2} & pk_1 \succ n,
\end{cases}
\\
&r_{\text{var}\mid \hat{\M\Sigma}_2}(k_1,k_2\mid p,q) \equiv \begin{cases}\frac{k_2}{pn}& qk_2 \precsim n
\\
\frac{k_1k_2}{n}\wedge \frac{qk^2_2}{pn^2} & qk_2 \succ n.
\end{cases}
\ee
%Denote $r_{\text{var}\mid \hat{\M\Sigma}_1}(k_1,k_2\mid p,q)$ and $r_{\text{var}\mid \hat{\M\Sigma}_2}(k_1,k_2\mid p,q)$ the terms in $r_{\text{var}\mid \hat{\M \Sigma}_2\otimes \hat{\M\Sigma}_1}(k_1,k_2\mid p,q)$, with respect to the row directions and column directions. 
\begin{theorem}\label{T:indi}
Suppose the conditions in Theorem \ref{T2} hold, and suppose the divergence regimes of  $k_1,k_2,p,q$ can guarantee 
\bee\label{thm:submatrix:con1}
r_{\text{var}\mid \hat{\M \Sigma}_2\otimes \hat{\M\Sigma}_1}(k_1,k_2\mid p,q) \rightarrow 0,
\ee
when $n\rightarrow +\infty$. Denote $d_a = \M I(a = 1)p + \M I(a = 2)q$. For $a \in\{1,2\}$, define $a^c = \{1,2\} \setminus a$. Then there exists normalization constants $c_{a,\eta}$ determined by $\M\Sigma_a^*$, $k_a$ and $\hat{\M \Sigma}_a^{\eta}(k_a)$, such that,
\bee\label{thm:sub:res1}
&\frac{1}{d_a}\big\|c_{a,\eta}\hat{\M \Sigma}_a^{\eta}(k_a) - \M\Sigma_a^*\big\|_\F^2 
\\
&=\mathcal{O}_{\mathbb{P}}\Big\{ r_{\text{var}\mid \hat{\M \Sigma}_a}(k_1,k_2\mid p,q) + r_{\text{var}\mid \hat{\M \Sigma}_2\otimes \hat{\M\Sigma}_1}(k_1,k_2\mid p,q) \cdot  r_{\text{var}\mid \hat{\M \Sigma}_{a^c}}(k_1,k_2\mid p,q) + \M I_{n,d_a}(k_a)\cdot k_a^{-\tilde\alpha_a}\Big\}.
\ee
\end{theorem}
Analytically finding $k_1,k_2$ to attain the optimal convergence rate of \eqref{thm:sub:res1} is complicated. To this end, we propose a numerical method to select optimal $k_1,k_2$ when $p,q$ are polynomially divergent in Section \ref{optbd:num}.
\begin{remark}
To better understand the rate in \eqref{thm:sub:res1}, we simply consider the degenerate regime where $q\asymp 1$ and $\M\Sigma_1,\M\Sigma_2 \in \mathcal{M}(\varepsilon_0,\alpha)$. Then let $c_{1,\eta} = 1$, our target error is $\frac{1}{p}\|\hat{\M\Sigma}_1^\eta(k_1) - \M\Sigma_1^*\|^2_\F$ for either $\eta = \mathcal{B}$ or $\mathcal{T}$. By definition and the fact that $k_2\asymp q \asymp k_2q \asymp 1$, the corresponding rates in \eqref{thm:sub:res1} become
\bee\nonumber
&r_{\text{var}\mid \hat{\M \Sigma}_2\otimes \hat{\M\Sigma}_1}(k_1,k_2\mid p,q) \asymp \begin{cases}\frac{k_1}{n} + \frac{1}{pn},& pk_1 \precsim n
\\
\big(\frac{k_1}{n}\big)\wedge\big(\frac{pk^2_1}{n^2} + \frac{1}{pn^2}  \big), & pk_1 \succ n,
\end{cases}
\\
&r_{\text{var}\mid \hat{\M\Sigma}_1}(k_1,k_2\mid p,q) \equiv \begin{cases}\frac{k_1}{n}& pk_1 \precsim n
\\
\frac{k_1}{n}\wedge \frac{pk^2_1}{n^2} & pk_1 \succ n,
\end{cases}
\\
&r_{\text{var}\mid \hat{\M\Sigma}_2}(k_1,k_2\mid p,q) \asymp \frac{1}{pn}.
\ee
Moreover, for $r_{\text{var}\mid \hat{\M\Sigma}_1}(k_1,k_2\mid p,q)$, since $\frac{k_1}{n} \succsim \frac{1}{pn}$ always hold as $k_1, p \geq 1$ and when $pk_1 \succ n$ one has $\frac{pk_1^2}{n^2} \succ \frac{nk_1}{n^2} = \frac{k_1}{n} \succ \frac{1}{pn^2}$, we have $r_{\text{var}\mid \hat{\M \Sigma}_2\otimes \hat{\M\Sigma}_1}(k_1,k_2\mid p,q) \asymp \frac{k_1}{n}$. Similarly, we can also show $r_{\text{var}\mid \hat{\M\Sigma}_1}(k_1,k_2\mid p,q) \asymp \frac{k_1}{n}$. 
\par
With all the results above, the convergence rate in \eqref{thm:sub:res1} becomes \begin{eqnarray*}
\frac{1}{p}\|\hat{\M\Sigma}_1^\eta(k_1) - \M\Sigma_1^*\|^2_\F &=&\mathcal{O}_{\mathbb{P}}\Big\{\frac{k_1}{n} + \frac{k_1}{n}\cdot \frac{1}{pn} + \M I_{n,p}(k_1)\cdot k_1^{-2\alpha_1 - 1}\Big\}
\\
&=&\mathcal{O}_{\mathbb{P}}\Big\{\frac{k_1}{n} + \M I_{n,p}(k_1)\cdot k_1^{-2\alpha_1 - 1}\Big\}
\\
& = & \mathcal{O}_{\mathbb{P}}\Big\{\min\left(n^{-\frac{2\alpha_1 + 1}{2\alpha_1 + 2}},\frac{p}{n}\right)\Big\},
\end{eqnarray*}
when selecting $k_1= \min\{n^{\frac{1}{2\alpha_1 + 2}},2p\}$. The second equality above holds because $\frac{1}{pn}\rightarrow 0$ when $n\rightarrow +\infty$, we have $\frac{k_1}{n}\succ \frac{k_1}{n}\cdot \frac{1}{pn}$. The above error rate matches with the minimax optimal Frobenius-norm convergence rate for vector-valued data in \citet{Cai2010}.
\end{remark}
\begin{remark}
 As discussed in Remark \ref{rmk:spec} of the main paper, the proposed individual matrix estimators $\hat{\M \Sigma}^{\eta}_1(k_1)$ and $\hat{\M \Sigma}^{\eta}_2(k_2)$ are obtained by reordering the components of the top right and left singular vectors of $\xi\big\{\hat{\boldsymbol\Sigma}^{\eta}(k_1, k_2)\big\}$ up to some constants. Therefore some singular subspace perturbation bounds can be applied to derive the error rates of the target error in \eqref{thm:sub:res1}. We employ the unilateral singular subspace perturbation bound derived recently by \citet{cai2018rate}, which can specifically and precisely control the perturbation errors of left and right singular subspaces, respectively. \citet{cai2018rate}'s perturbation bound is shown to be rate-optimal in a general condition. As a result, we can show that in \eqref{thm:sub:res1} the divergence regimes of $p,q$ have asymmetric effects on the error rates of $\M\Sigma_1^*$ and $\M\Sigma_2^*$ estimation, which is also intuitively reasonable.  One may also consider to use the classical Wedin's theorem to derive the error bounds. Because Wedin's theorem can only give a uniform perturbation bound for both left and right singular subspaces, it leads to an identical error rate for both $\M\Sigma_1^*$ and $\M\Sigma_2^*$, which, as shown in \citet{cai2018rate}, is likely to be sub-optimal when $p$ and $q$ have unbalanced divergence rates.
\end{remark}
\begin{remark}
As $r_{\text{var}\mid \hat{\M \Sigma}_2\otimes \hat{\M\Sigma}_1}(k_1,k_2\mid p,q) $ is one term in the entire error upper bound \eqref{T2:res1} of the proposed estimator, as long as $p,q,k_1,k_2$ can guarantee a convergence of the proposed estimator, condition \eqref{thm:submatrix:con1} is satisfied.
\end{remark}

%Condition \eqref{thm:submatrix:con1} in Theorem \ref{T:indi} is easy to satisfy. Recall $r_{\text{var}\mid \hat{\M \Sigma}_2\otimes \hat{\M\Sigma}_1}(k_1,k_2\mid p,q) $ is one error term within the whole error rate \eqref{T2:res1} of the original proposed estimator. As long as divergence regimes $p,q,k_1,k_2$ can guarantee a convergence of the whole proposed estimator, condition \eqref{thm:submatrix:con1} is satisfied immediately.
%\end{remark}
\begin{remark}
We interpret the error rates in \eqref{thm:sub:res1} as follows. Take $a = 1$ as an example. The target error $\big\| c_{1,\eta}\hat{\M \Sigma}_1^{\eta}(k_1) - \M\Sigma_1^*\big\|_\F^2/p$ can be decomposed into two error terms via triangle inequality,
\bee\nonumber
\frac{1}{p}\big\|c_{1,\eta}\hat{\M \Sigma}_1^{\eta}(k_1) - \M\Sigma_1^*\big\|_\F^2 &\precsim \underbrace{\frac{1}{p}\big\| c_{1,\eta}\hat{\M \Sigma}_1^{\eta}(k_1) -  \M\Sigma_1^{*,\eta}(k_1)\big\|_\F^2}_{E_1^{(p)}} 
+\underbrace{\frac{1}{p}\big\|\M\Sigma^{*,\eta}_1(k_1) - \M\Sigma_1^*\big\|_\F^2}_{E_2^{(p)}}.
\ee
Similar to \eqref{err:decom}, $E_1^{(p)}$ is the entrywise mean squared error of the normalized proposed component estimator $c_{1,\eta}\hat{\M \Sigma}_1^{\eta}(k_1)$ to the banded or tapering true covariance component, while $E_2^{(p)}$ is the thresholding error caused by banded or tapering over the $p$ direction.
\par
In \eqref{thm:sub:res1}, the first two terms correspond to the error rate of $E_1^{(p)}$, 
%which is determined by $r_{\text{var}\mid \hat{\M\Sigma}_1}(k_1,k_2\mid p,q)$ and  $r_{\text{var}\mid \hat{\M \Sigma}_2\otimes \hat{\M\Sigma}_1}(k_1,k_2\mid p,q) \cdot  r_{\text{var}\mid \hat{\M \Sigma}_{\bar{a}}}(k_1,k_2\mid p,q) $, where the first term plays the key role in the error.  
%\par
and the last term is the error rate of $E_2^{(p)}$. When $\M \Sigma_1^{*}\in \mathcal{M}(\varepsilon_0, \alpha_1),\M \Sigma_2^{*} \in \mathcal{M}(\varepsilon_0, \alpha_2)$, the last term becomes the bias term under the degenerate regime when $\alpha$ equals $\alpha_1\text{ or }\alpha_2$, the error rate of which has been derived in (37) of \citet{Cai2010}.
\end{remark}
\subsection{Spectral-norm Concentration}\label{sec:disc:spnorm}
Our theorems in Section \ref{sec:T} mainly focus on the convergence under the Frobenius norm; and it is also of interest to study the convergence under the spectral norm. For vector-valued data (i.e., the degenerate regime), spectral-norm convergence results for banded and tapering covariance estimators have been studied in \citet{Bickel2008threshold, Cai2010}. We refer readers to \citet{cai2016estimating} for a thorough survey. Next we discuss the challenge and possible future work directions for spectral-norm convergence. 
\par
For vector data, the spectral-norm error is usually obtained by utilizing the inherent linearity of the banded and tapering covariance estimator via some entrywise or block-wise technical arguments. For example, \citet{Bickel2008threshold} obtained a spectral-norm error bound for the banded estimator  based on some entrywise inequality scaling  and an entrywise maximal error bound. \citet{Cai2010} used  explicit blockwise sub-matrix decomposition analysis and obtained a minimax rate-optimal error rate.
\par
However, for matrix data, our proposed estimator is generally not linear when both $p$ and $q$ diverge. This is because the proposed estimation components $\hat{\M\Sigma}_1^{\eta}(k_1)$ and $\hat{\M\Sigma}_2^\eta(k_2)$ are derived by reordering the components of the leading right and left singular vectors of $\xi\{\tilde{\M\Sigma}_\eta(k_1,k_2)\}$ and multiplying with the leading singular value of $\xi\{\tilde{\M\Sigma}_\eta(k_1,k_2)\}$. Therefore, the usual entrywise and block-wise techniques can not be similarly applied under our framework. In fact, the existing literature in singular subspace perturbation mainly focus on $\ell_2$-norm perturbation bound \citep{wedin1972perturbation, yu2015useful, o2018random, cai2018rate} and $\ell_{2,\infty}$-norm perturbation bound \citep{cape2019two, lei2019unified, abbe2020entrywise}, which are not sufficient to provide an accurate spectral-norm perturbation analysis for our estimators obtained from the leading singular vectors. Developing more precise matrix perturbation bounds is an important future work direction.{\color{black}
\subsubsection{Spectral-norm Approximation Variation}
An alternative way to attain the spectral-norm concentration to $\M\Sigma^*$, is to consider a variation of the proposed method: an estimator based on the Kronecker product approximation of the doubly banded or tapering estimator under  the \textit{spectral} norm, instead of the original Frobenius norm used in \eqref{eq:band:0}. In particular, we consider $(\hat{\M\Sigma}_1^{\mathcal{\eta,\mathcal{S}}}(k_1),\hat{\M\Sigma}_2^{\mathcal{\eta,\mathcal{S}}}(k_2))$ that solves the approximation problem:
\bee\label{sec:dis:s}
  \argmin_{\bSigma_1, \bSigma_2}\big\| \widetilde{\M\Sigma}_{\eta}(k_1,k_2)-\bSigma_2 \otimes \bSigma_1\big\|_2^2,
\ee 
with $\eta\in\{\mathcal{B},\mathcal{T}\}$. Then $\hat{\M \Sigma}^{\eta,\mathcal{S}} = \hat{\M\Sigma}_2^{\mathcal{\eta,\mathcal{S}}}(k_2) \otimes \hat{\M\Sigma}_1^{\mathcal{\eta,\mathcal{S}}}(k_1)$ can be treated as an alternative estimator of $\M \Sigma^*$. %{\color{red} Why do we need the following paragraph? Can we start to state the theorem directly? ``with at least the same rate'' is very vague. And if the rate is the same, it seems to be a weakness.} 
In fact, we have the following spectral-concentration result for our newly proposed tapering estimator.
\begin{theorem}\label{Ts}
Let $\vecc(\M X_1), \vecc(\M X_2),\dots,\vecc(\M X_n)$ be i.i.d. sub-Gaussian random vectors in $\RR^{pq}$ with true covariance $\M \Sigma^* = \M \Sigma^*_2 \otimes \M \Sigma^*_1$, where $\M \Sigma_1^{*}\in \mathcal{F}(\varepsilon_0, \alpha_1),\M \Sigma_2^{*} \in \mathcal{F}(\varepsilon_0, \alpha_2)$ or $\M \Sigma_1^{*}\in \mathcal{M}(\varepsilon_0, \alpha_1),\M \Sigma_2^{*} \in \mathcal{M}(\varepsilon_0, \alpha_2)$. If at least one of $p,q$ is polynomially divergent (faster than some $n^b$ with $b > 0$), we have
\bee\label{Ts:res1}
\E\Big(\big\|\hat{\M \Sigma}^{\mathcal{T},\mathcal{S}} - \M\Sigma^*\big\|_{2}^2\Big)\precsim \frac{k_1k_2 + \log(\max\{p,q\})}{n} + \M I(k_1< 2p-2)\cdot k_1^{-2\alpha_1} +  \M I(k_2< 2q-2)\cdot k_2^{-2\alpha_2}.
\ee
Otherwise, we can simply take $k_1 = 2p-2, k_2 = 2q-2$, and have
\bee\label{Ts:res2}
\E\Big(\big\|\hat{\M \Sigma}^{\mathcal{T},\mathcal{S}} - \M\Sigma^*\big\|^2_{2}\Big) \precsim  \frac{pq}{n}.
\ee 
\end{theorem}
The selection of $k_1,k_2$ to attain the optimal convergence rate of \eqref{Ts:res1} under different divergence regimes of $p,q$ is discussed in Section \ref{optbd:Ts}.
\begin{remark}
Due to the term $\frac{\log\{\max(p,q)\}}{n}$, the rate \eqref{Ts:res1} converges to zero under the ultrahigh-dimensional regime such that $p,q$ satisfy $p,q\precsim \exp(n).$
\end{remark}
\begin{remark}
In fact $\hat{\M \Sigma}^{\mathcal{B},\mathcal{S}}$ is also consistent under the spectral norm when $p,q \precsim \exp(n)$.  However, showing the rate-optimality of the bandable estimator under the spectral norm is an unresolved problem, even under the degenerate regime. Therefore, we omit the sub-optimal results of $\hat{\M \Sigma}^{\mathcal{B},\mathcal{S}}$ for simplicity. 
\end{remark}
\begin{remark}
For the degenerate regime such as $q=1$, we can take $k_2 = 1$. Then the convergence rate of \eqref{Ts:res1} in Theorem \ref{Ts} becomes 
\bee\nonumber
&\frac{k_1\times 1 + \log(\max\{p,1\})}{n} + \M I(k_1< 2p-2)\cdot k_1^{-2\alpha_1} +\M I(1< 0)\cdot k_1^{-2\alpha_2}
\\
&= \frac{k_1 + \log(p)}{n} + \M I(k_1< 2p-2)\cdot k_1^{-2\alpha_1}, 
\ee
which agrees with the spectral-norm convergence rate of \citet{Cai2010}'s rate-optimal tapering covariance estimator,  for vector-valued data. Thus $\hat{\M \Sigma}^{\mathcal{T},\mathcal{S}}$ is rate-optimal under the degenerate regime. 
\end{remark}
To this end, we have shown that the newly proposed estimator enjoys desired spectral concentration property to $\M\Sigma^*$ while maintaining the separability structure. However, to the best of our knowledge, unlike the Frobenius-norm case, it remains an open question to find an efficient algorithm to solve the spectral-norm Kronecker product approximation problem as in \eqref{sec:dis:s}. Therefore, a more thorough study of this newly proposed estimator is beyond the scope of this paper, and we leave it for future work.
\begin{remark}
One can show the solution of the spectral-norm approximation problem \eqref{sec:dis:s} is generally different from the solution of the Frobenius-norm approximation  problem \eqref{eq:band:0}.

%Recall that Eckart-Young-Mirsky Theorem \citep{eckart1936approximation, mirsky1960symmetric} has shown for a given $\M M^* \in \RR^{d_1\times d_2}$ and $k>0$, low rank matrix  approximation problem:
%\bee\nonumber
%\M M_k = \argmin_{\text{rank}(\M M) = k}\|\M M^* - \M M\|_\star
%\ee
%can be solved simultaneously for all unitarily invariant norms $\|\cdot\|_\star$ (including Frobenius norm and spectral norm), by a rank-$k$ matrix $\hat{\M M}_k$, which is formed by the $k$ leading spikes through singular value decomposition.
\par
%We note such analogy is generally not true for the Kronecker product approximation problem. 
We see this via a simple example. Consider a $4\times 4$ doubly tapering or banded estimator $\widetilde{\M\Sigma}_{\eta}(k_1,k_2) = \diag(1,2,3,4)$. Our original proposed estimator by Frobenius-norm approximation (rounded for display) is 
$$
\hat{\M\Sigma}_2^{\mathcal{\eta}}(k_2) \otimes \hat{\M\Sigma}_1^{\mathcal{\eta}}(k_1)=
\begin{pmatrix}
1.2736&&&&
\\
& 1.8072&&
\\&&2.8790&
\\
&&&4.0853
\end{pmatrix},
$$ where a particular group of solution $\big(\hat{\M\Sigma}_2^{\mathcal{\eta}}(k_2), \hat{\M\Sigma}_1^{\mathcal{\eta}}(k_1)\big)$ is
\bee
\hat{\M\Sigma}_1^{\mathcal{\eta}}(k_1) = \begin{pmatrix}
3.1481 & 
\\
& 4.4672
\end{pmatrix}, \hat{\M\Sigma}_2^{\mathcal{\eta}}(k_2) = \begin{pmatrix}
0.4046 & 
\\
& 0.9145
\end{pmatrix}.
\ee
Then the spectral-norm approximation error is $\big\| \widetilde{\M\Sigma}_{\eta}(k_1,k_2)-\hat{\M\Sigma}_2^{\mathcal{\eta}}(k_2) \otimes \hat{\M\Sigma}_1^{\mathcal{\eta}}(k_1)\big\|_2 = 0.2736$. However, if we sightly perturb $\big(\hat{\M\Sigma}_2^{\mathcal{\eta}}(k_2), \hat{\M\Sigma}_1^{\mathcal{\eta}}(k_1)\big)$ as
\bee
\hat{\M\Sigma}_1^{\mathcal{\eta,\mathcal{P}}}(k_1) = \begin{pmatrix}
3.1481 + 0.01 & 
\\
& 4.4672 + 0.01
\end{pmatrix}, \hat{\M\Sigma}_2^{\mathcal{\eta,\mathcal{P}}}(k_2) = \begin{pmatrix}
0.4046 - 0.01 & 
\\
& 0.9145 - 0.01
\end{pmatrix},\nonumber
\ee  
then the spectral approximation error of $\hat{\M\Sigma}_2^{\mathcal{\eta},\mathcal{P}}(k_2) \otimes \hat{\M\Sigma}_1^{\mathcal{\eta},\mathcal{P}}(k_1)$ becomes $\big\| \widetilde{\M\Sigma}_{\eta}(k_1,k_2)-\hat{\M\Sigma}_2^{\mathcal{\eta,\mathcal{P}}}(k_2) \otimes \hat{\M\Sigma}_1^{\mathcal{\eta,\mathcal{P}}}(k_1)\big\|_2 = 0.2460$ which is significantly smaller than the approximation error by $\hat{\M\Sigma}_2^{\mathcal{\eta}}(k_2) \otimes \hat{\M\Sigma}_1^{\mathcal{\eta}}(k_1)$. This indicates that $(\hat{\M\Sigma}_1^{\mathcal{\eta}}(k_1), \hat{\M\Sigma}_2^{\mathcal{\eta}}(k_2))$ is not a solution to the spectral-norm approximation problem \eqref{sec:dis:s}. 
%\par
%To summarize, this example suggests that unlike low-rank matrix approximation, the spectral-norm and Frobenius-norm Kronecker product approximation problems do not agree on a unique solution. Therefore  solving spectral-norm approximation problem like \eqref{sec:dis:s} requires further investigation.   
\end{remark}
}
\subsubsection{Optimal Bandwidth Selection of Theorem \ref{Ts}}\label{optbd:Ts}
Based on different divergence regimes of $(p,q)$, we aim to find optimal selection of $k_1,k_2$ to minimize the convergence rate of the following term:
\bee\label{Tts:target}
r_3(k_1,k_2\mid p,q) = \frac{k_1k_2 + \log(\max\{p,q\})}{n} + k_1^{-2\alpha_1}I(k_1 <2{p}-2) + k_2^{-2\alpha_2}I(k_2 < 2{q}-2).
\ee
\par
We note $r_3(k_1,k_2\mid p,q)$ can be reordered as
\bee\label{ts:select:r3}
r_3(k_1,k_2\mid p,q) &= \frac{k_1k_2}{n} + k_1^{-2\alpha_1}I(k_1 <2{p}-2) + k_2^{-2\alpha_2}I(k_2 < 2{q}-2) + \frac{\log(\max\{p,q\})}{n}
\\
&= r_2(k_1,k_2\mid p,q,\mathcal{T}) + \frac{\log(\max\{p,q\})}{n}.
\ee
For first term in \eqref{ts:select:r3}, the optimal rate of $r_2(k_1,k_2\mid p,q,\mathcal{T})$ is shown in Section \ref{optbd:T3}. Since the last term $\frac{\log(\max\{p,q\})}{n}$ in \eqref{ts:select:r3} is not to do with $k_1,k_2$, the optimal $k_{1},k_2$ selection for $r_3(k_1,k_2\mid p,q)$, is same with the optimal $k_{1},k_2$ selection for $r_2(k_1,k_2\mid p,q)$. Thus $$k_{1,\text{opt}}^{(3)} = k_{1,\text{opt}}^{(2)}\text{ and } k_{2,\text{opt}}^{(3)} = k_{2,\text{opt}}^{(2)}$$
for all different divergence regimes of $p,q$.  Then we take the maximum between optimal rate $r_2\big(k_{1,\text{opt}}^{(2)},k_{2,\text{opt}}^{(2)}\mid p,q,n\big)$ and the last term $\log(\max\{p,q\})/n$, and get the final optimal convergence rate of $r_3(k_1,k_2\mid p,q)$.
\par
Similar to Section \ref{optbd:T3}, we derive the optimal rates of $r_3(k_1,k_2\mid p,q)$ under following three scenarios.
\par
\
\par
\noindent{\textbf{Scenario 1 (both $p,q$ diverge fast):}}
First we consider when $p \succ n^{\alpha_2/(\alpha_1+\alpha_2+2\alpha_1\alpha_2)}, q\succ n^{\alpha_1/(\alpha_1+\alpha_2+2\alpha_1\alpha_2)}$. By \eqref{r2:opt} and \eqref{r2:kopt}, we can take $k_{1,\text{opt}}^{(3)} = k_{1,\text{opt}}^{(2)} = n^{\alpha_2/(\alpha_1+\alpha_2+2\alpha_1\alpha_2)}$ and $k_{2,\text{opt}}^{(3)} = k_{2,\text{opt}}^{(2)} = n^{\alpha_1/(\alpha_1+\alpha_2+2\alpha_1\alpha_2)}$. Then the optimal rate of $r_2(k_1,k_2\mid p,q,\mathcal{T})$ becomes $n^{-2\alpha_1\alpha_2/(\alpha_1 +\alpha_2 + 2\alpha_1\alpha_2)}$. Taking the term $\log(\max\{p,q\})/n$ in \eqref{ts:select:r3} into account,  we have
\bee\nonumber
r_3\big(k_{1,\text{opt}}^{(3)},k_{2,\text{opt}}^{(3)}\mid p,q\big)\asymp \max\Bigg\{n^{-2\alpha_1\alpha_2/(\alpha_1 + \alpha_2 + 2\alpha_1\alpha_2)},\frac{\log(\max\{p,q\})}{n}\Bigg\}.
\ee
\\
\par
\noindent{\textbf{Scenario 2 (one of $p,q$ diverges slowly):}} By symmetry, we only consider when $p\precsim n^{\alpha_2/(\alpha_1+\alpha_2+2\alpha_1\alpha_2)}, q\succ (n/p)^{\frac{1}{2\alpha_2 + 1}}$. We can take $k_{1,\text{opt}}^{(3)} = k_{1,\text{opt}}^{(2)}  = 2p-1$ (here $\eta = \mathcal{T}$), and $k_{2,\text{opt}}^{(3)} = k_{2,\text{opt}}^{(2)} = (n/p)^{\frac{1}{2\alpha_2 + 1}}$. The optimal rate of $r_2(k_1,k_2,\mid p,q,\mathcal{T})$ is $\big(n/p\big)^{-\frac{2\alpha_2}{2\alpha_2 + 1}}$. Taking the term $\log(\max\{p,q\})/n$ in \eqref{ts:select:r3} into account,  we have
\bee\nonumber
r_3\big(k_{1,\text{opt}}^{(3)},k_{2,\text{opt}}^{(3)}\mid p,q\big) &\asymp \max\Bigg\{\big(n/p\big)^{-\frac{2\alpha_2}{2\alpha_2 + 1}},\frac{\max\{\log(p),\log(q)\}}{n}\Bigg\}
\\
&\asymp\max\Bigg\{\big(n/p\big)^{-\frac{2\alpha_2}{2\alpha_2 + 1}},\frac{\log(q)}{n}\Bigg\},
\ee
since $\log(p)/n \precsim \log(n)/n \prec \big(n/p\big)^{-\frac{2\alpha_2}{2\alpha_2 + 1}}$ under this scenario.
\\
\par
\noindent{\textbf{Scenario 3 (both $p,q$ diverge slowly):}}
By symmetry, we only consider when $p\precsim n^{\alpha_2/(\alpha_1+\alpha_2+2\alpha_1\alpha_2)}, q\precsim {n}^{\frac{1}{2\alpha_2 + 1}}$. Noting here $\eta = \mathcal{T}$, we can take $k_{1,\text{opt}}^{(3)} = k_{1,\text{opt}}^{(2)}  = 2p-1$ and $k_{2,\text{opt}}^{(3)} = k_{2,\text{opt}}^{(2)} = 2q - 1$.  The optimal rate of $r_2(k_1,k_2\mid p,q,\mathcal{T})$ is $pq/n$. Since $\frac{\log\{\max(p,q)\}}{n}\asymp \frac{\log(p) + \log(q)}{n}\asymp \frac{\log(pq)}{n}$, it is easy to see 
$$
\frac{pq}{n} \succsim \frac{\log(pq)}{n} \asymp\frac{\log\{\max(p,q)\}}{n}.
$$ Then finally we have
\bee\nonumber
r_3\big(k_{1,\text{opt}}^{(3)},k_{2,\text{opt}}^{(3)}\mid p,q\big) \asymp \max\Big\{pq/n, \log(\max\{p,q\})/n\Big\} \asymp pq/n.
\ee
\par
\
\par
\noindent{\textbf{Summarization:}} Under different divergence regimes of $p,q$, we summarize the optimal rate of $r_3(k_1,k_2\mid p,q) = r_3\big(k_{1,\text{opt}}^{(3)},k_{2,\text{opt}}^{(3)}\mid p,q\big) $ and corresponding $k_{1,\text{opt}}^{(3)}, k_{2,\text{opt}}^{(3)}$ as following:
\bee\label{r3:opt}
&r_3\big(k_{1,\text{opt}}^{(3)},k_{2,\text{opt}}^{(3)}\mid p,q\big)  
\\
&\asymp 
\begin{cases}
\max\big\{n^{-2\alpha_1\alpha_2/(\alpha_1 + \alpha_2 + 2\alpha_1\alpha_2)},\frac{\log(\max\{p,q\})}{n}\big\} &\text{when } p \succ {n}^{\tilde\alpha_2/(\tilde\alpha_1+\tilde\alpha_2+\tilde\alpha_1\tilde\alpha_2)}, q\succ {n}^{\tilde\alpha_1/(\tilde\alpha_1+\tilde\alpha_2+\tilde\alpha_1\tilde\alpha_2)}
\\
\max\big\{\big(n/p\big)^{-\frac{2\alpha_2}{2\alpha_2 + 1}},\frac{\log(q)}{n}\big\} & \text{when }  p \precsim {n}^{\tilde\alpha_2/(\tilde\alpha_1+\tilde\alpha_2+\tilde\alpha_1\tilde\alpha_2)}, q\succ \big({n}/p\big)^{\frac{1}{\tilde\alpha_2 + 1}}
\\
\max\big\{\big(n/q\big)^{-\frac{2\alpha_1}{2\alpha_1 + 1}},\frac{\log(p)}{n}\big\} & \text{when }  q \precsim {n}^{\tilde\alpha_1/(\tilde\alpha_1+\tilde\alpha_2+\tilde\alpha_1\tilde\alpha_2)}, p\succ \big({n}/q\big)^{\frac{1}{\tilde\alpha_1 + 1}}
\\
\frac{pq}{n} & \text{when }
\begin{cases}
p \precsim {n}^{\tilde\alpha_2/(\tilde\alpha_1+\tilde\alpha_2+\tilde\alpha_1\tilde\alpha_2)}, q\precsim \big({n}/p\big)^{\frac{1}{\tilde\alpha_2 + 1}}
\\
q \precsim {n}^{\tilde\alpha_1/(\tilde\alpha_1+\tilde\alpha_2+\tilde\alpha_1\tilde\alpha_2)}, p\precsim \big({n}/q\big)^{\frac{1}{\tilde\alpha_1 + 1}}
\end{cases},
\end{cases}
\ee
with,
\bee\label{r3:kopt}
\big(k_{1,\text{opt}}^{(3)},k_{2,\text{opt}}^{(3)}\big) &= \big(k_{1,\text{opt}}^{(2)},k_{2,\text{opt}}^{(2)}\big) 
\\
&=
\begin{cases}
\big({n}^{\tilde\alpha_2/(\tilde\alpha_1+\tilde\alpha_2+\tilde\alpha_1\tilde\alpha_2)},{n}^{\tilde\alpha_1/(\tilde\alpha_1+\tilde\alpha_2+\tilde\alpha_1\tilde\alpha_2)}\big) & \text{when }p \succ {n}^{\tilde\alpha_2/(\tilde\alpha_1+\tilde\alpha_2+\tilde\alpha_1\tilde\alpha_2)}, q\succ {n}^{\tilde\alpha_1/(\tilde\alpha_1+\tilde\alpha_2+\tilde\alpha_1\tilde\alpha_2)}
\\
\big(\tilde p, \{n/p\}^{\frac{1}{\tilde{\alpha}_2 + 1}}\big)& \text{when }  p \precsim {n}^{\tilde\alpha_2/(\tilde\alpha_1+\tilde\alpha_2+\tilde\alpha_1\tilde\alpha_2)}, q\succ \big({n}/p\big)^{\frac{1}{\tilde\alpha_2 + 1}}
\\
\big(\tilde q, \{n/q\}^{\frac{1}{\tilde{\alpha}_1 + 1}}\big)& \text{when }  q \precsim {n}^{\tilde\alpha_1/(\tilde\alpha_1+\tilde\alpha_2+\tilde\alpha_1\tilde\alpha_2)}, p\succ \big({n}/p\big)^{\frac{1}{\tilde\alpha_1 + 1}}
\\
\big(\tilde p,\tilde q\big) & \text{when }\begin{cases}
p \precsim {n}^{\tilde\alpha_2/(\tilde\alpha_1+\tilde\alpha_2+\tilde\alpha_1\tilde\alpha_2)}, q\precsim \big({n}/p\big)^{\frac{1}{\tilde\alpha_2 + 1}}
\\
q \precsim {n}^{\tilde\alpha_1/(\tilde\alpha_1+\tilde\alpha_2+\tilde\alpha_1\tilde\alpha_2)}, p\precsim \big({n}/q\big)^{\frac{1}{\tilde\alpha_1 + 1}}
\end{cases}.
\end{cases}
\ee

\section{Discussion of Theorems in Section \ref{sec:RE}}\label{sec:supps3}
\subsection{Thresholding Factor $J_n$}\label{sec:Jn}{\color{black}
Recall that we have defined the thresholding factor $J_n$ in Assumption \ref{A:psi2} such that 
\bee\label{robust:def}
\|\vecc(\widecheck{\M X}_i)\|_{\psi_2} \leq J_n\tau,
\ee when $n\rightarrow +\infty$. As discussed in Section \ref{sec:RE} in the main paper, a trivial upper bound of \eqref{robust:def} holds with $J_n \asymp \sqrt{pq}$, by basic property of sub-Gaussian norm and triangle inequality. In this section, we further show that, for the bandable covariance structure that is specifically considered in this paper,  the trivial bound of \eqref{robust:def} can be typically sharpened with a slowly growing $J_n$ or $J_n \asymp 1$, for any heavy-tailed distribution. Our key observation is that the growing rate of $J_n$ can be determined by the dependence structure of $\M X$. %{\color{red} The following sentence is vague. Need to rewrite or delete if not important.} When separable covariance matrices $\M\Sigma_1^*$ and $\M\Sigma_2^*$ are bandable, two entries in the data matrix $\M X$ have very weak dependence if they are long-distant. 
This observation makes it possible to better control $J_n$.
\par To see this, consider the following example of a block-wise dependent structure of data matrix $\M X$: 
\begin{itemize}
\item For the ease of exposition, let $B$ be a fixed integer, $\M X$ be a square data matrix in $\RR^{p\times p}$, and $p = A_n B$ with some positive integer $A_n$, where $A_n$ is allowed to grow with $n$.
\item Divide $\M X$ into $(p/B)^2$ different $B\times B$ block matrices and assume the entries in $\M X$ are independent of each other across the blocks. 
\end{itemize}This example can be related to the spatial-temporal setting, and will result in  bandable structures of $\M\Sigma_1^*$ and $\M\Sigma_2^*$ when the covariance matrix of $\M X$ is separable. In particular, $\M\Sigma_1^*$ and $\M \Sigma_2^*\in\RR^{p\times p}$ can only have non-zero elements on the diagonal $B\times B$ block matrices. We further illustrate the block-wise dependence structure of $\M X$ and the corresponding bandable covariance structures of $\M \Sigma_1^*$ and $\M \Sigma_2^*$ in Figure \ref{fig:growJn}.
\begin{figure}
   \centering
\includegraphics[width=0.6\textwidth]{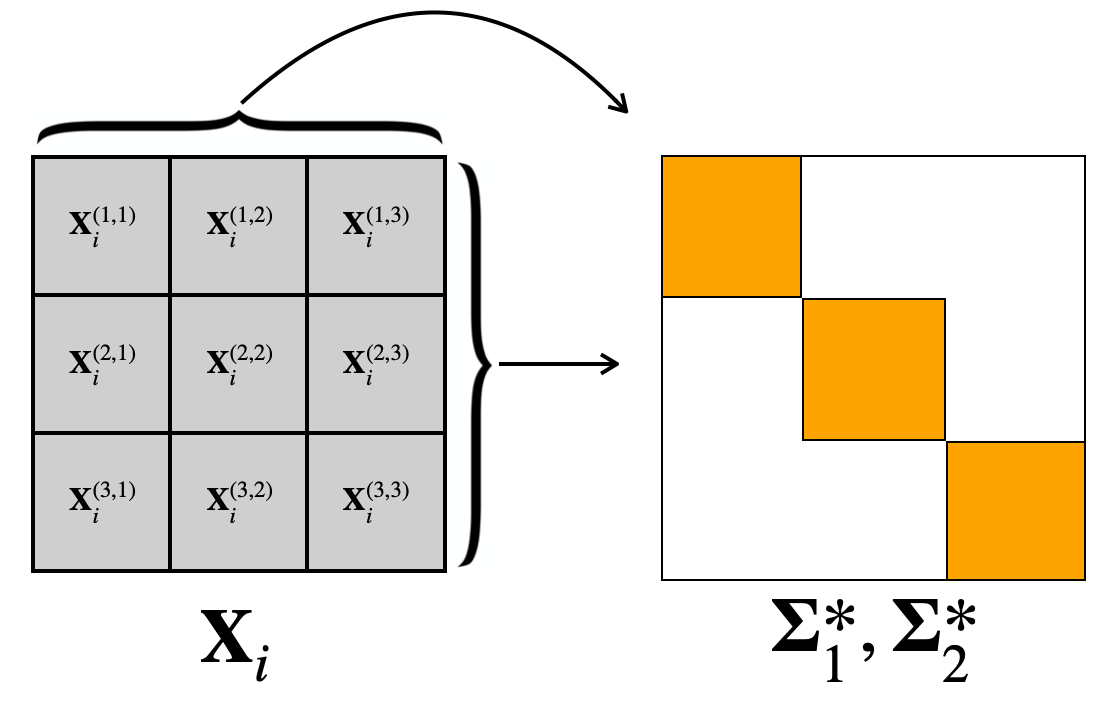}
\caption{Illustration of $\M X_i$'s dependence structure and $\M\Sigma_1^*,\M\Sigma_2^*$'s bandable structures with $p/B = 3$. Different sub-matrices (grey) on the left-hand side represents different block matrices $\M X_i^{(b_1,b_2)}\in \RR^{B\times B}$ within $\M X_i$. They are independent of each other. The matrix on the right-hand side represents either $\M \Sigma_1^*$ or $\M\Sigma_2^*$, which can take non-zero values only on the $B\times B$ diagonal block matrices (orange).}
\label{fig:growJn}
\end{figure}
\par
Now we show that under the block-wise dependent structure of $\M X$, the thresholding factor $J_n$ in \eqref{robust:def} can be sharpened to a constant. Let $\widecheck{\M X}^{(b_1,b_2)}_i$ be the $b_1b_2$th block in $\widecheck{\M X}_i$, corresponding to the original block matrix $\M X_i^{(b_1,b_2)}$ as illustrated in Figure \ref{fig:growJn}. 
%For any $\mathbf{v}\in\RR^{p^2}$ with $\|\mathbf{v}\| = 1$, we can write
%\bee
%\mathbf{v}^\T\vecc(\widecheck{\M X}_i) = \sum_{b_1 = 1}^{p/B}\sum_{b_2 = 1}^{p/B} \mathbf{v}_{b_1,b_2}^\T\vecc(\widecheck{\M X}^{(b_1,b_2)}_i)
%\ee
%where $\widecheck{\M X}^{(b_1,b_2)}_i$ is the $b_1b_2$th block in $\widecheck{\M X}_i$, corresponding to the original block matrix $\M X_i^{(b_1,b_2)}$ as illustrated in Figure \ref{}, and $\mathbf{v}_{b_1,b_2} \in \RR^{B^2}$ is some sub-vector of $\mathbf{v}$ such that $\sqrt{\sum_{b_1 = 1}^{p/B}\sum_{b_2 = 1}^{p/B}\|\mathbf{v}_{b_1,b_2}\|^2} = \|\mathbf{v}\| = 1$. 
First, for any specific group of $(b_1,b_2)$ and any $\mathbf{v}_{b_1,b_2}\in\RR^{B^2}$ such that $\|\mathbf{v}_{b_1,b_2}\|=1$, by triangle inequality and the definition of sub-Gaussian norm in Section \ref{def:sesgrv}, we have 
\bee\label{bind:1}
&\|\vecc(\widecheck{\M X}^{(b_1,b_2)}_i) - \E\vecc(\widecheck{\M X}^{(b_1,b_2)}_i)\|_{\psi_2}
\\
&=\sup_{\|\mathbf{v}_{b_1,b_2}\|=1}\big\|\mathbf{v}_{b_1,b_2}^\T\big\{\vecc(\widecheck{\M X}^{(b_1,b_2)}_i) - \E\vecc(\widecheck{\M X}^{(b_1,b_2)}_i)\big\}\big\|_{\psi_2} 
\\
&\leq \sup_{\|\mathbf{v}_{b_1,b_2}\|=1}\sum_{j = 1}^{B^2}\big|\big[\mathbf{v}_{b_1,b_2}\big]_j\big|\cdot\big\|\big[\vecc(\widecheck{\M X}^{(b_1,b_2)}_i)\big]_j\big\|_{\psi_2}
\\
&\leq B^2C_J\cdot\tau
\\
&\precsim \tau
,
\ee
where we use $[\mathbf{u}]_j$ to represent the $j$th coordinate of any vector $\mathbf{u}$; the second inequality is by $\big|\big[\mathbf{v}_{b_1,b_2}\big]_j\big| \leq 1$ and $\big\|\big[\vecc(\widecheck{\M X}^{(b_1,b_2)}_i)\big]_j\big\|_{\psi_2} \leq C_J\tau $ for some fixed constant $C_J > 0$ (see e.g. Example 2.5.8 in \citet{vershynin2018high}), recalling that $\big[\vecc(\widecheck{\M X}^{(b_1,b_2)}_i)\big]_j$ is an entry in $\widecheck{\M X}_i$ which is always smaller or equal to $\tau$, and $B$ is fixed. Second, by the block-wise dependent structure of $\M X_i$, it is easy to see that the set of random vectors $\big\{\vecc(\widecheck{\M X}^{(b_1,b_2)}_i) - \E\vecc(\widecheck{\M X}^{(b_1,b_2)}_i)\big\}_{b_1,b_2 = 1,\dots,p/B}$ contains independent, mean-zero, and disjoint sub-vectors of $\vecc(\widecheck{\M X}_i) - \E\vecc(\widecheck{\M X}_i)$. Then by the sub-Gaussian norm bound with independent coordinates (see e.g. Lemma 3.4.2 in \citet{vershynin2018high}) and \eqref{bind:1}, there exists some fixed constant $C'_J > 0$ such that,
\bee
\ &\|\vecc(\widecheck{\M X}_i)\|_{\psi_2} 
\\
&\leq \|\vecc(\widecheck{\M X}_i) - \E\vecc(\widecheck{\M X}_i)\|_{\psi_2}  + \|\E\vecc(\widecheck{\M X}_i)\|_{\psi_2}
\\
&\leq C'_J\max_{(b_1,b_2)}\|\vecc(\widecheck{\M X}^{(b_1,b_2)}_i) - \E\vecc(\widecheck{\M X}^{(b_1,b_2)}_i)\|_{\psi_2} + \|\E\vecc(\widecheck{\M X}_i)\|_{\psi_2}
\\
&\precsim \tau + \|\E\vecc(\widecheck{\M X}_i)\|_{\psi_2}.
\ee
Next, we bound $\|\E\vecc(\widecheck{\M X}_i)\|_{\psi_2}$. Since $\E\vecc(\widecheck{\M X}_i)$ is a deterministic vector, of which the coordinates are independent and uniformly bounded by $\tau$ in magnitude, we directly have $\|\E\vecc(\widecheck{\M X}_i)\|_{\psi_2} \precsim \tau$ by Example 2.5.8 and Lemma 3.4.2 in \citet{vershynin2018high}. Therefore, we have shown $\|\vecc(\widecheck{\M X}_i)\|_{\psi_2} \precsim \tau$, which implies that the upper bound \eqref{robust:def} holds with $J_n\precsim 1$ when $\M X$ has a block-wise dependent structure and can be generated from any heavy-tailed distributions.
\par
For the general case, the bandable structures of $\M\Sigma_1^*$, $\M\Sigma_2^*$ still indicate the dependence between two far-away entries in $\M X$ might be weak, but they may not be exactly independent. Rigorous quantification of the relationship between the growth rate of $J_n$, and the dependence structure (or more generally, the probabilistic structure) of $\M X$ is rather complicated. Therefore, we leave it for future research.
\par
Recall that $\M X$ is heavy-tailed and satisfies the 2$\zeta$-finite moment condition. 
When $\M X$ satisfies Assumption \ref{A:psi2} with $J_n$  slowly growing or $J_n\precsim 1$, one interesting theoretical property is that, when $\zeta$ becomes large, the proposed robust estimator will have an approximately same convergence rate as the originally-proposed estimator with a sub-Gaussian $\M X$. See Remark \ref{rk:rb:subgaussian} for detailed discussions, where we consider the specific scenario $J_n\precsim 1$ as an example. Therefore, even if $\M X$ is not heavy-tailed, e.g., $\M X$ is sub-exponential, one can still consider replacing the originally-proposed estimator with the proposed robust estimator, which can potentially improve the theoretical performance of the covariance estimation and make it comparable with the sub-Gaussian case; as long as $\M X$ has a good probabilistic structure that guarantees $J_n$ can be well-controlled.
}
\subsection{Rate interpretation of Theorem \ref{T4}}
We interpret the error rate \eqref{T4:res} as follows. With some algebra, the target error can be bounded by three error terms,
\bee\nonumber
&\E\Big(\frac{\|\hat{\M \Sigma}_2^{\mathcal{R},\eta}\kii\otimes\hat{\M \Sigma}_1^{\mathcal{R},\eta}\ki - \M\Sigma^*\|_{\F}^2}{pq}\Big) 
\\
&\precsim \underbrace{\frac{1}{pq}\E\big\|\xi\{\widecheck{\M \Sigma}_{\eta}(k_1,k_2)\} - \xi\{{\M \Sigma}^{*,\eta}_{\mathcal{R}}(k_1,k_2)\}\big\|_2^2}_{E_1^{\mathcal{R}}} +\underbrace{\frac{1}{pq}\big\|{\M \Sigma}^{*,\eta}_{\mathcal{R}}(k_1,k_2)-\M\Sigma_2^{*,{\eta}}(k_2)\otimes \M\Sigma_1^{*,{\eta}}(k_1)\big\|_\F^2}_{E_2^{\mathcal{R}}}
\\ \nonumber
&+\underbrace{\frac{1}{pq}\big\|\M\Sigma_2^{*,{\eta}}(k_2)\otimes \M\Sigma_1^{*,{\eta}}(k_1) - \M\Sigma^*\big\|_\F^2}_{E_2}. 
\ee
\par
Here $E_1^{\mathcal{R}}$ has a similar form to the right-hand side of \eqref{E1upper}, which serves as an upper bound of the mean squared error $E_1$. When $\zeta \geq 2$, we use similar techniques to \eqref{varterm:1} and \eqref{varterm:2} for the proposed non-robust estimator to bound $E_1^{\mathcal{R}}$, while accounting for the effect of diverging $\tau$. In particular, a modified technique to \eqref{varterm:1} yields 
\bee\label{error:E1r:2}
E_1^{\mathcal{R}}\precsim r^{(1)}_{\text{var}\mid \hat{\M \Sigma}_2\otimes \hat{\M\Sigma}_1}(k_1,k_2\mid p,q),
\ee and a modified argument to \eqref{varterm:2} yields 
\bee\label{error:E1r}
E_1^{\mathcal{R}}\precsim (J_n\tau)^4\cdot r^{(2)}_{\text{var}\mid \hat{\M \Sigma}_2\otimes \hat{\M\Sigma}_1}(k_1,k_2\mid p,q).
\ee When $\zeta <2$, since the finite fourth moment condition no longer holds, a similar argument to \eqref{varterm:1} can not be employed. Therefore, $E_1^{\mathcal{R}}$ can not be bounded by \eqref{error:E1r:2}. On the other hand, $E_1^{\mathcal{R}}$ can still be bounded by \eqref{error:E1r}.
\par
The term $E_2^{\mathcal{R}}$ is the truncation error between the doubly banded/tapering true covariances of the truncated data and the original data. We employ higher order Markov inequality and H\"{o}lder  inequality to precisely control it based on both truncation level $\tau$ and the moment condition $\zeta$, as $E_2^{\mathcal{R}}  \precsim{k_1k_2}\cdot\tau^{-4(\zeta - 1)}$. The term $E_2$ is the original thresholding error, which can be bounded in the same way as \eqref{errorE2}.
\par
After incorporating the error bounds derived above, and balancing the rate trade-off in terms of $\tau$, the error rate \eqref{T4:res} and the optimal $\tau$ in \eqref{tau:choose} can be derived. 
\begin{remark}
As shown in Theorems \ref{T2} and \ref{T4}, when $\zeta \geq 2$, both our non-robust estimator and the robust estimator are consistent. Here, we briefly compare these two convergence rates. 
%If $\zeta \geq 2$, our original proposed estimators have the convergence rate given in Theorem \ref{T2}, while our proposed robust estimators have the convergence rate given in Theorem \ref{T4}. When $\zeta \geq 2$, we briefly compare these two rates.
\par
As shown later in Section \ref{sec:opt:threshold}, when $\zeta \geq 2$ and $\frac{k_1k_2}{n}$ converges faster than both $r^{\mathcal{R},\zeta}_1(k_1,k_2, p,q,n)$ and $r^{\mathcal{R},\zeta}_2(k_1,k_2, p,q,n)$, the optimal $\tau$ is $\tau = +\infty$, i.e., no truncation is needed in our robust estimation procedure. Then our proposed robust estimators degenerate to the non-robust estimators and the convergence rate becomes $\frac{k_1k_2}{n} + \M I_{\eta,p}(k_1)\cdot k_1^{-\tilde{\alpha}_1} +  \M I_{\eta,q}(k_2)\cdot k_2^{-\tilde{\alpha}_2} $. This coincides with the error rates given in Theorem \ref{T3} for our non-robust estimators. We note that this also implies our proposed robust estimators are rate-optimal under the degenerate regime.
\par
On the other hand, if $r^{\mathcal{R},\zeta}_1(k_1,k_2, p,q,n)$ converges faster than $\frac{k_1k_2}{n}$ under condition $\M C\M 2$, it is easy to see the convergence rate in Theorem \ref{T4} is faster than the one in Theorem \ref{T2}. Similar phenomenon also holds if $r^{\mathcal{R},\zeta}_2(k_1,k_2, p,q,n)$ converges faster than $\frac{k_1k_2}{n}$ under condition $\M C\M 4$. Therefore, under these scenarios, our proposed robust estimators improve the  convergence rate of the non-robust estimators, after an appropriate choice of $\tau$. %We note that we do not consider $\M C\M 1$ and $\M C\M 3$ as under $\M C\M 1$ and $\M C\M 3$, $\zeta$ is no longer larger than $2$.
\end{remark}
\begin{remark}\label{rk:rb:subgaussian}
When $\zeta\geq 2$ and $J_n \asymp 1$, we have
\bee\label{remark:rb:rate}
&\E\Big(\frac{\|\hat{\M \Sigma}_2^{\mathcal{R},\eta}\kii\otimes\hat{\M \Sigma}_1^{\mathcal{R},\eta}\ki - \M\Sigma^*\|_{\F}^2}{pq}\Big)
\\ 
&\small\precsim \begin{cases}
\big(\frac{k_1k_2}{n}\big)\wedge\Big[(k_1k_2)^{1/\zeta}\cdot\big(\frac{k_1}{qn} + \frac{k_2}{pn}\big)^{1 - 1/\zeta}\Big]+ \M I_{\eta,p}(k_1)\cdot k_1^{-\tilde{\alpha}_1} +  \M I_{\eta,q}(k_2)\cdot k_2^{-\tilde{\alpha}_2}, & pk_1 + qk_2 \precsim n,
\\
\big(\frac{k_1k_2}{n}\big)\wedge\Big[(k_1k_2)^{1/\zeta}\cdot\big(\frac{pk^2_1}{qn^2} + \frac{qk^2_2}{pn^2}  \big)^{1 - 1/\zeta}\Big]+ \M I_{\eta,p}(k_1)\cdot k_1^{-\tilde\alpha_1} +  \M I_{\eta,q}(k_2)\cdot k_2^{-\tilde\alpha_2}, & pk_1 + qk_2 \succ n ,
\end{cases}
\ee
%where we denote $r_2(k_1,k_2\mid p,q,\eta) \equiv \big(\frac{k_1k_2}{n}\big) + \M I_{\eta,p}(k_1)\cdot k_1^{-\tilde{\alpha}_1} +  \M I_{\eta,q}(k_2)\cdot k_2^{-\tilde{\alpha}_2}$ the rate \eqref{T3:res1} of Theorem \ref{T3}, and denote 
%\bee\nonumber
%r_r(k_1,k_2\mid p,q,\eta, \zeta) \equiv \begin{cases}
%(k_1k_2)^{1/\zeta}\cdot\big(\frac{k_1}{qn} + \frac{k_2}{pn}\big)^{1 - 1/\zeta}+ \M I_{\eta,p}(k_1)\cdot k_1^{-\tilde{\alpha}_1} +  \M I_{\eta,q}(k_2)\cdot k_2^{-\tilde{\alpha}_2}, & pk_1 + qk_2 \precsim n
%\\
%(k_1k_2)^{1/\zeta}\cdot\big(\frac{pk^2_1}{qn^2} + \frac{qk^2_2}{pn^2}  \big)^{1 - 1/\zeta}+ \M I_{\eta,p}(k_1)\cdot k_1^{-\tilde\alpha_1} +  \M I_{\eta,q}(k_2)\cdot k_2^{-\tilde\alpha_2}, & pk_1 + qk_2 \succ n.
%\end{cases}
%\ee
Then it is easy to see that when $\zeta \rightarrow +\infty$, for any given $k_1,k_2,p,q$, the asymptotic order of the error rate in \eqref{remark:rb:rate} converges to the sub-Gaussian rate \eqref{T2:res1} in Theorem \ref{T2}. In summary, our robust estimators can improve the convergence rate compared to the non-robust estimators when data are heavy-tailed, and will not sacrifice the convergence rate when the data are sub-Gaussian.
%Therefore, the theoretical performance of our robust proposed estimators can adapt with different moment conditions, and can approach the performance of our original proposed estimators on sub-Gaussian data, when the heavy-tailed data have large $\zeta$. 
\par
 We illustrate this observation via a strongly bandable example defined in Section \ref{sec:match:illustration}. Suppose $\alpha_1 = 2,\alpha_2 = 2$, $\M \Sigma_1^{*}\in \mathcal{M}(\varepsilon_0, \alpha_1),\M \Sigma_2^{*} \in \mathcal{M}(\varepsilon_0, \alpha_2)$, and $J_n\asymp 1$. We show the optimal convergence rates in Theorem \ref{T4} with different $\zeta$ in Figure \ref{fig:sgbd:rb}. In Figure \ref{fig:sgbd:rb} (a), we show the  convergence rate of robust proposed estimators when $\zeta = 1.5$, while no consistency can be shown for the proposed non-robust estimators.  Under the exact finite fourth moment condition ($\zeta = 2$), Figure \ref{fig:sgbd:rb} (b) has the same pattern with that of Figure \ref{fig:sgbd1} (c), so the convergence rates we obtain for robust proposed estimators are the same with the rates of non-robust estimators in Theorem \ref{T3}. Moreover, with an increasing $\zeta$, patterns of Figures \ref{fig:sgbd:rb} (c)-(f) become closer and closer to the pattern of Figure \ref{fig:sgbd1} (b). This confirms that the theoretical performance of the robust proposed estimator will approach that of the proposed non-robust estimator under the sub-Gaussian condition, when $\zeta\rightarrow +\infty$. 
\begin{figure}
  \centering
    \subfigure[$\zeta = 1.5$]{\includegraphics[width=0.32\textwidth]{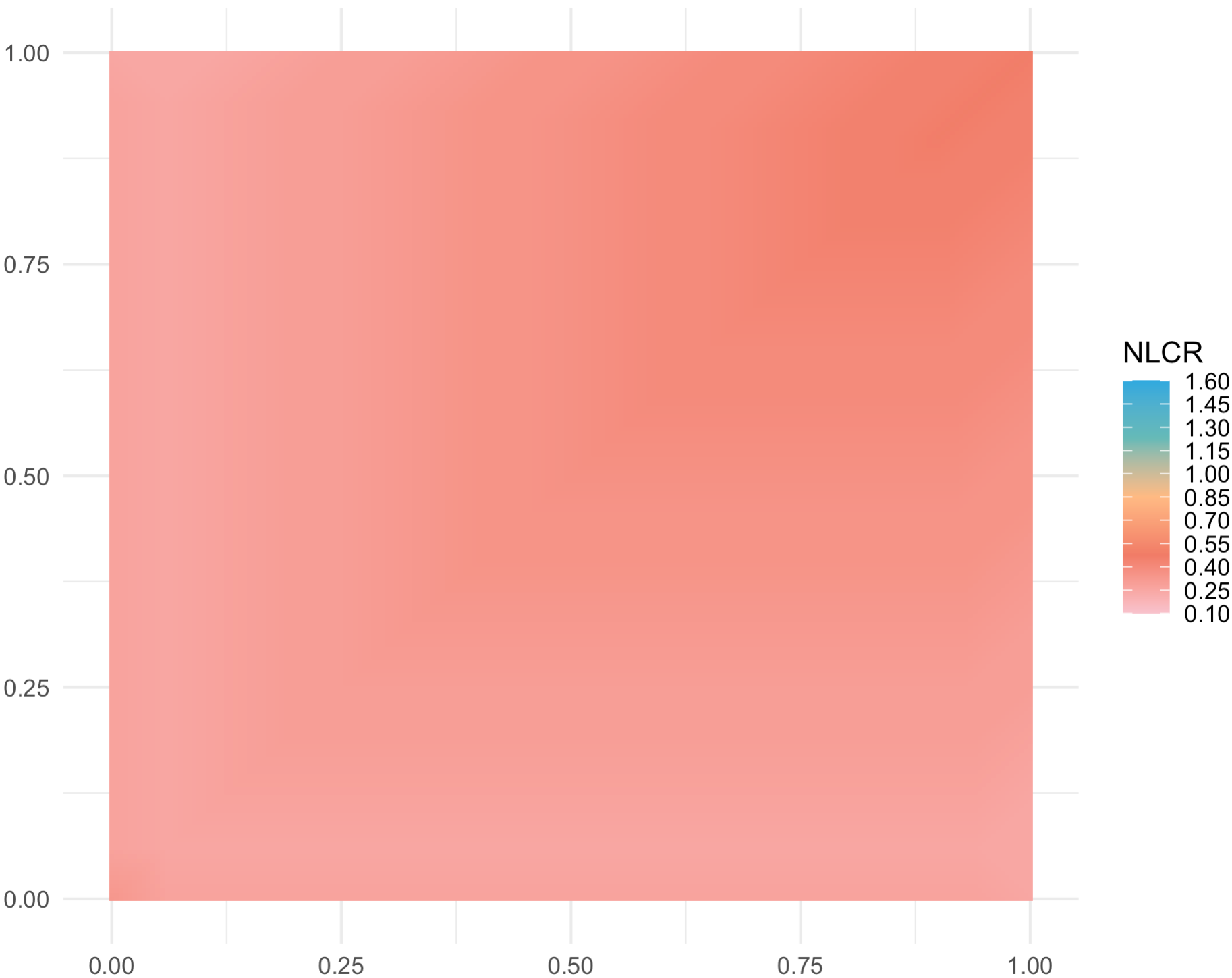}}
  \subfigure[$\zeta = 2$]{\includegraphics[width=0.32\textwidth]{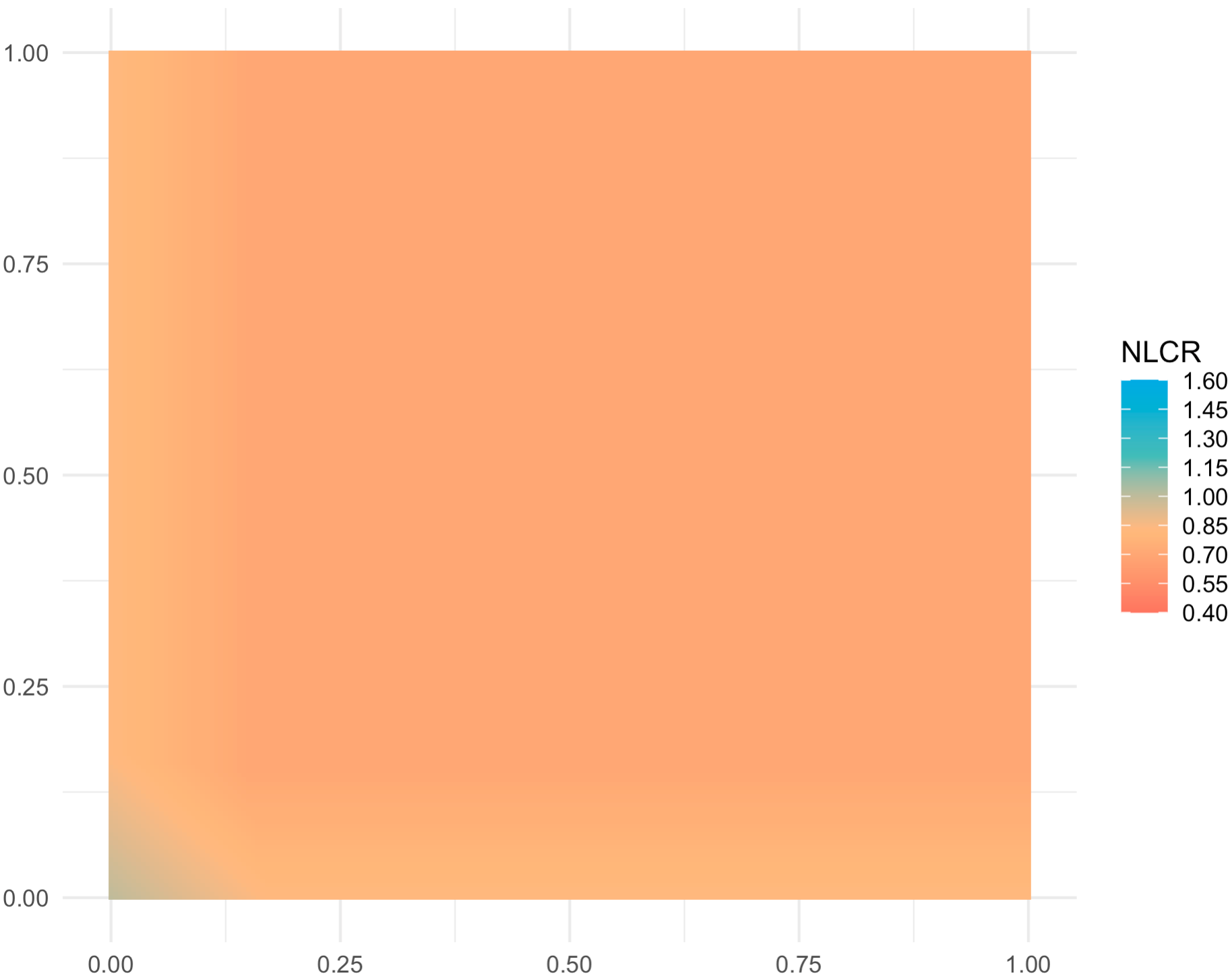}}
  \subfigure[$\zeta = 3$]{\includegraphics[width=0.32\textwidth]{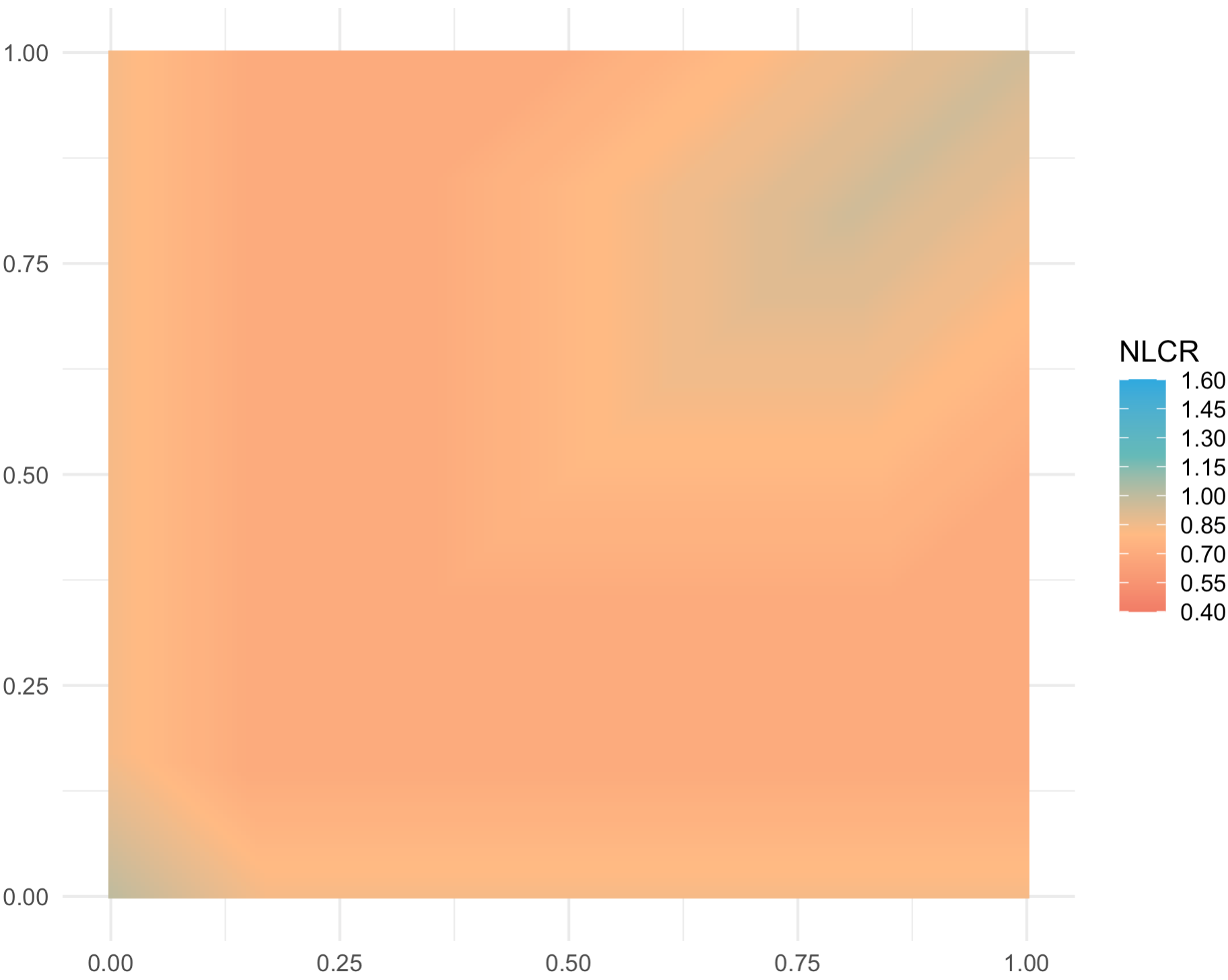}}
    \subfigure[$\zeta = 4$]{\includegraphics[width=0.32\textwidth]{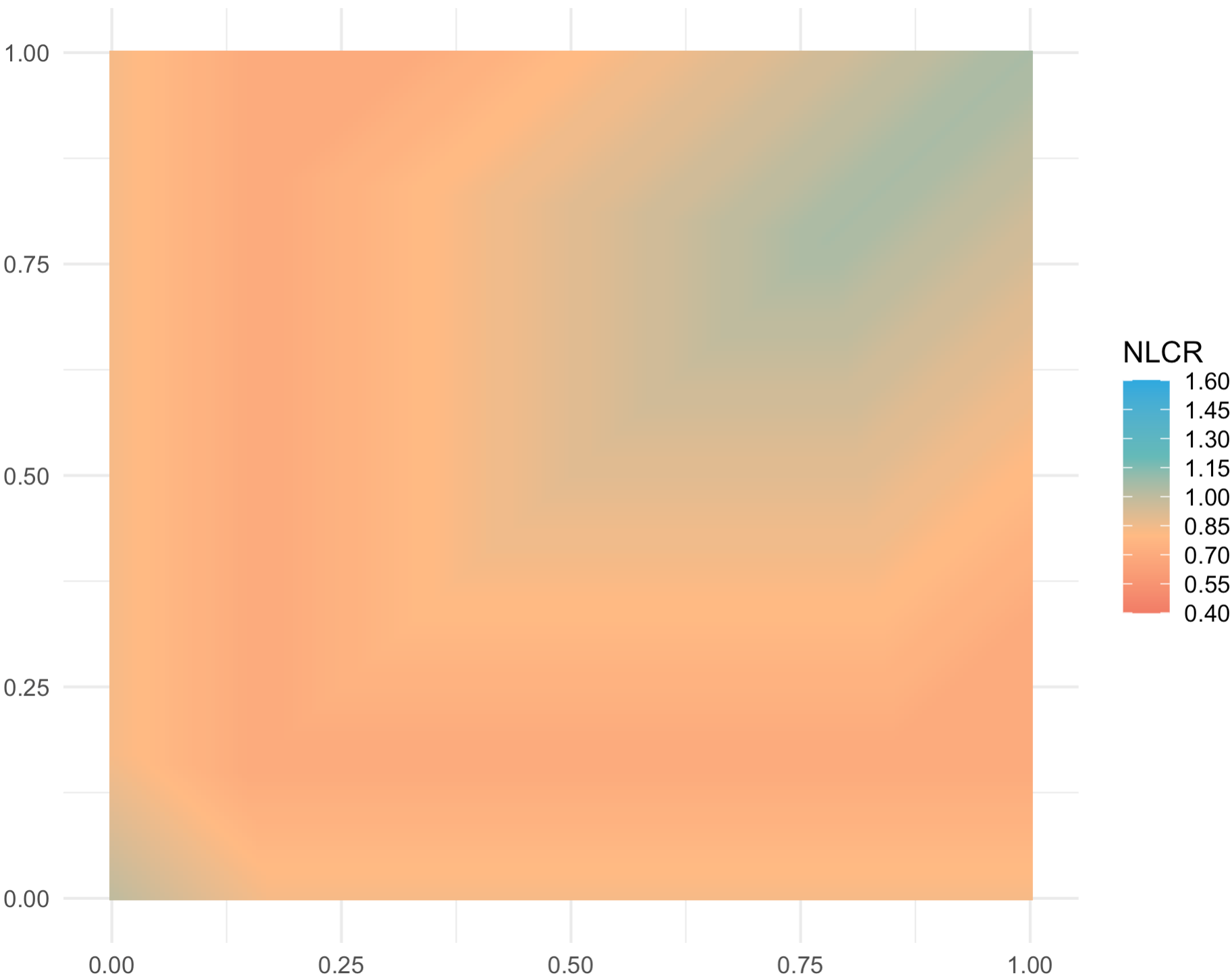}}
  \subfigure[$\zeta = 10$]{\includegraphics[width=0.32\textwidth]{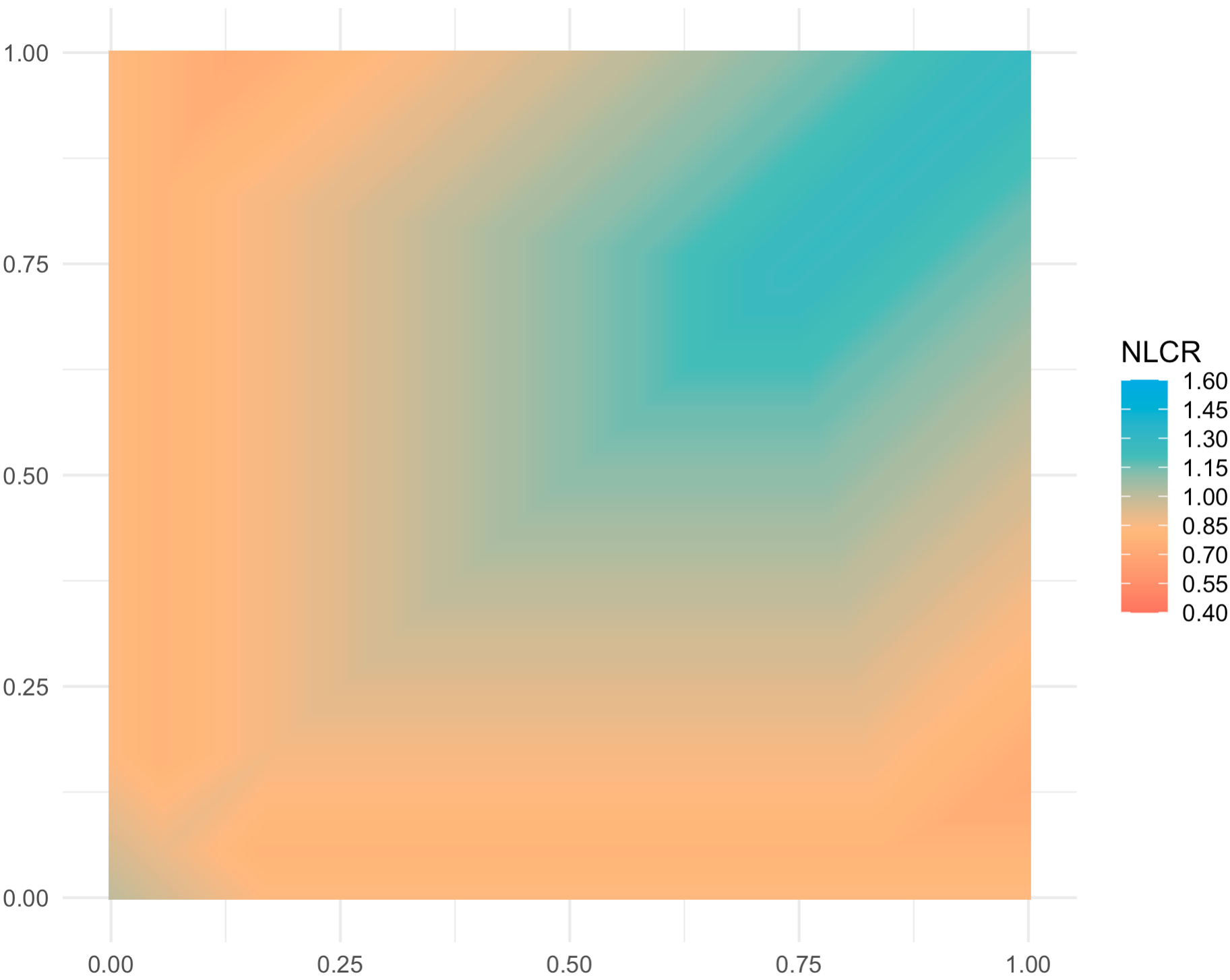}}
  \subfigure[$\zeta = 20$]{\includegraphics[width=0.32\textwidth]{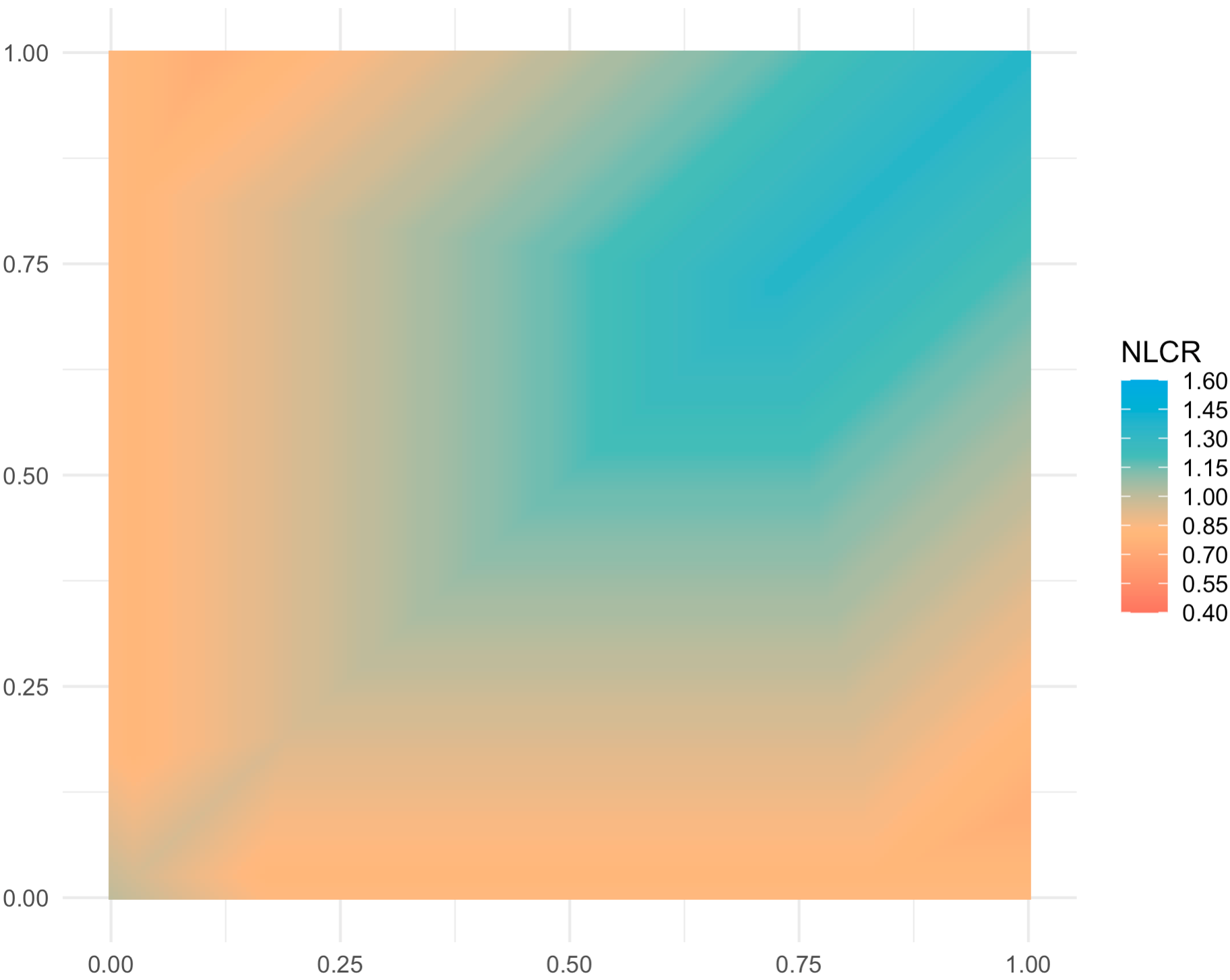}}
\caption{Optimal convergence rates of Theorem \ref{T4} in the strongly bandable example ($\alpha_1 = 2, \alpha_2 = 2$) with difference $\zeta$ and $J_n\asymp 1$. The x-axis represents $\log p = \beta_1$ and y-axis represents $\log q = \beta_2$. The color represents the negative log convergence rate (NLCR).}
\label{fig:sgbd:rb}
\end{figure}
\end{remark}
Analytically finding $k_1,k_2$ to attain the optimal convergence rate of \eqref{T4:res} in the main paper is complicated. In Section \ref{optbd:num}, we propose a numerical method to select optimal $k_1,k_2$, when $p,q$ are polynomially divergent.
\par
We will give an explicit formula for the optimal choice of $\tau$ in the next subsection. 
\subsection{Optimal Threshold Selection for Theorem \ref{T4}}\label{sec:opt:threshold}
Recall that $r^{\mathcal{R},\zeta}_1(k_1,k_2, p,q,n) = (k_1k_2)^{1/\zeta}\cdot\big(\frac{k_1J_n^4}{qn} + \frac{k_2J_n^4}{pn}\big)^{1 - 1/\zeta}$, and $r^{\mathcal{R},\zeta}_2(k_1,k_2, p,q,n) = (k_1k_2)^{1/\zeta}\cdot\big(\frac{pk^2_1J_n^4}{qn^2} + \frac{qk^2_2J_n^4}{pn^2}  \big)^{1 - 1/\zeta}$. We choose the optimal $\tau$ as follows to attain the error rate \eqref{T4:res} in Theorem \ref{T4}, 
\bee\label{tau:choose}
\small{
\tau \asymp \begin{cases}
\big\{pqk_1k_2n/(pk_1J_n^4 + qk_2J_n^4) \big\}^{1/(4\zeta)} & \text{when }pk_1 + qk_2\precsim n \text{ and }
\scriptsize{\begin{cases}\zeta \geq 2\text{ and }r^{\mathcal{R},\zeta}_1(k_1,k_2, p,q,n)\precsim \frac{k_1k_2}{n}
\\
1<\zeta <2
\end{cases};}
\\
\big\{pqk_1k_2n^2/(pk_1J_n^2 + qk_2J_n^2)^2 \big\}^{1/(4\zeta)}  & \text{when }pk_1 + qk_2\succ n\ \text{and }\scriptsize{\begin{cases}\zeta \geq 2 \text{ and }
r^{\mathcal{R},\zeta}_2(k_1,k_2, p,q,n)\precsim \frac{k_1k_2}{n}
\\
1 <\zeta <2,
\end{cases};}
\\
+\infty & \text{otherwise}.
\end{cases}}
\ee

\subsection{Numerical Optimal Bandwidth Selection in Theorems \ref{T4} and \ref{T:indi}}\label{optbd:num}
We discuss a numerical approach for selecting the optimal bandwidths when $p,q$ diverge polynomially in $n$, i.e., $p\asymp n^{\beta_1}$ and $q\asymp n^{\beta_2}$. We consider a large candidate set $\mathcal{K}$ for  $k_1$ and $k_2$. For example, $\mathcal{K}= \{(k_1, k_2) = (n^{\beta_{k_1}}, n^{\beta_{k_2}})\mid \beta_{k_1} = 0,0.1,\dots, 1, \beta_{k_2} = 0,0.1,\dots,1\}$. Then based on the rates in Theorems \ref{T4} and \ref{T:indi}, each pair of $(k_1, k_2) = (n^{\beta_{k_1}}, n^{\beta_{k_2}})$ combined with $p\asymp n^{\beta_1}, q\asymp n^{\beta_1}$, leads to a particular polynomial convergence rate in Theorems \ref{T4} and \ref{T:indi}, denoted by $n^{-r(\beta_1,\beta_2,\beta_{k_1},\beta_{k_2})}$. We take the pair of $(k_1, k_2) = (n^{\beta_{k_1}}, n^{\beta_{k_2}})$ in $\mathcal{K}$ that maximizes $r(\beta_1,\beta_2,\beta_{k_1},\beta_{k_2})$ as the selected $k_1$ and $k_2$. %When the ranges of $\beta_{k_1}, \beta_{k_2}$   in $\mathcal{K}$ are wide enough, and the grids of $\beta_{k_1}, \beta_{k_2}$ are sufficiently fine, the numerically obtained optimal $k_1,k_2$ regimes can approximate the true optimal $k_1,k_2$ regimes sufficiently well.

\section{Solving Frobenius-norm Kronecker Product Approximation}\label{sec:sKPA}
We briefly introduce how to solve a general rank one unconstrained Kronecker product approximation problem under the Frobenius norm. Suppose we have matrices $\Ab\in \RR^{p_1q_1\times p_2q_2}$, $\Bb\in \RR^{q_1\times q_2}$ and $\Cb\in \RR^{p_1\times p_2}$. The goal is to solve
\begin{eqnarray}\label{kronecker}
\min_{\Bb, \Cb}\| \Ab-\Bb \otimes \Cb\|^2_\F.
\end{eqnarray} 
For a $p_1q_1 \times p_2q_2$ matrix $\M A$ with submatrix structure,
\$
\M A = \begin{bmatrix}
\M A_{1,1} \cdots  \M A_{1,q_2} \\
\vdots \ \ \ \ \ \ \ \ \ \ \ \vdots \\
\M A_{q_1,1} \cdots  \M A_{q_1q_2},
\end{bmatrix}
\$
where the submatrix $\M A_{l,m} \in \RR^{p_1\times p_2}$ for all $1\leq l\leq q_1$ and $1\leq m\leq q_2$. We define a matrix transformation $\xi(\cdot):\RR^{p_1q_1 \times p_2q_2} \longrightarrow \RR^{q_1q_2\times p_1p_2}$ as,
\begin{comment}
\#
\label{xi:fun}
\xi(\M A) = \begin{bmatrix}
\vv (\M A_{1,1})^{\T} \\
\vdots \\
\vv (\M A_{p_B,1})^{\T} \\
\vdots \\
\vv (\M A_{p_B,q_B})^{\T}
\end{bmatrix}.
\#
\end{comment}
\#
\label{xi:fun}
\xi(\M A) = \begin{bmatrix}
\vv (\M A_{1,1}), \cdots, 
\vv (\M A_{q_1,1}), \cdots, 
\vv (\M A_{q_1,q_2}) 
\end{bmatrix}^{\T}.
\#
\citet*{van1993approximation,Pitsianis1997} show that solving (\ref{kronecker}) is equivalent to solving a rank one singular value decomposition (SVD) of $ \xi({\Ab}) $. In particular, they have shown the following propositions. 
\begin{proposition}\label{equivalence}
$\| \Ab-\Bb \otimes \Cb\|^2_\F=\| \xi({\Ab})-{\bb}{\bf c}^{\T}\|^2_\F$, where $ {\bb}=\vecc ({\Bb}) $ and $ {\bf c}=\vecc ({\Cb}) $.
\end{proposition}
\begin{proposition}\label{equivalence2}
The minimizer of $\min_{\bb, \bf c}\| \xi({\Ab})-{\bb}{\bf c}^{\T}\|^2_\F$ is the same as the minimizer of $\min_{\bb, \bf c}\| \xi({\Ab})-{\bb}{\bf c}^{\T}\|_2$.
\end{proposition}

By Propositions \ref{equivalence} and \ref{equivalence2}, solving (\ref{kronecker}) is equivalent to solving 
\begin{eqnarray}\label{kroneckerequivalence}
\min_{\bb, {\bf c}}\| \xi({\Ab})-{\bb}{\bf c}^{\T}\|_2. 
\end{eqnarray}

\begin{proposition}\label{solution}
If the SVD of matrix $  \xi({\Ab})={\Ub}{\bf S}{\Vb} $, where $ {\Ub} \in \RR^{q_1q_2\times q_1q_2} $ and $ {\Vb}\in \RR^{p_1p_2\times p_1p_2} $, and $ \sigma_1 $ is the largest singular value. Then the minimizer of (\ref{kroneckerequivalence}) is $ \hat{\bb}=C\sigma_1 {\ub}_1 $ and $ \hat{\bf c}=C^{-1}{\vb}_1 $ for any constant $ C\neq 0 $, where $ {\ub}_1 $ and $ {\vb}_1 $ are the first columns of matrices $ {\Ub} $ and $ {\Vb} $.
\end{proposition}
%Therefore, the solution to (\ref{kronecker}) would be $ \hat {\Bb} $ and $ \hat {\Cb} $ with  $ \hat{\bb}=\vecc (\hat{\Bb}) $ and $ \hat{\bf c}=\vecc (\hat{\Cb}) $. 
Propositions \ref{equivalence}, \ref{equivalence2} and \ref{solution} are directly taken from \citet{Pitsianis1997,golub1996matrix}, so we omit the proof. From these propositions, one can easily obtain the following proposition. 

\begin{proposition}\label{decomposition:uniqueness}
The Kronecker approximation $ \Bb \otimes \Cb $ is unique. In addition, if $ (\Bb^*, \Cb^*) $ is a solution of \eqref{kronecker}, $ (c\Bb^*, c^{-1}\Cb^*) $ is also a solution of \eqref{kronecker} for any constant $ c\!\neq\! 0 $. 
\end{proposition}
Proposition \ref{decomposition:uniqueness} implies that our banded covariance estimate $\hat{\boldsymbol\Sigma}^{\MB}(k_1, k_2)$ is unique. 
%\noindent {\bf Proof of Proposition \ref{decomposition:uniqueness}}: By Theorem 5.8 in \citet{van1993approximation}, there exists symmetric positive definite matrices $ \hat{\boldsymbol\Sigma}_1^* $ and $\hat{\boldsymbol\Sigma}_2^*$ such that $ (\hat{\boldsymbol\Sigma}_1^*, \hat{\boldsymbol\Sigma}_2^*) $ is a solution of \eqref{eq:sample}. By propositions \ref{equivalence}, \ref{equivalence2} and \ref{solution}, $ \vecc(\hat{\boldsymbol\Sigma}_1) $ and $\vecc(\hat{\boldsymbol\Sigma}_2)$ are obtained from rank-1 approximation of a specific permutation of $ \hat{\boldsymbol\Sigma} $ by using singular value decomposition. By uniqueness of the singular value decomposition (up to scale for singular vectors), one has $ \vecc(\hat{\boldsymbol\Sigma}_1) $ and $\vecc(\hat{\boldsymbol\Sigma}_2)$ are unique subject to scale constant $ c $, and their product $  \vecc(\hat{\boldsymbol\Sigma}_1) (\vecc(\hat{\boldsymbol\Sigma}_2))^{\T} $ is unique. Therefore, $ \hat{\boldsymbol\Sigma}_1 $ and $\hat{\boldsymbol\Sigma}_2$ are unique subject to scale constant $ c $, and their Kronecker product $ \hat{\boldsymbol\Sigma}_1\otimes \vecc(\hat{\boldsymbol\Sigma}_2)$ is unique. Thus, all the solutions of \eqref{eq:sample} are of the form $ (c\hat{\boldsymbol\Sigma}_1^*, c^{-1}\hat{\boldsymbol\Sigma}_2^*) $ for any constant $ c\!\neq\! 0 $. 

\begin{remark}\label{rmk:spec}
Propositions \ref{equivalence}--\ref{solution}  show the proposed optimization problem \eqref{eq:band} can be solved by reordering the components of the top right and left singular vectors of $\xi\big\{\tilde{\M\Sigma}_{\mathcal{B}}(k_1,k_2)\big\}$ and multiplying with the top singular value of $\xi\big\{\tilde{\M\Sigma}_{\mathcal{B}}(k_1,k_2)\big\}$. Therefore, our proposed method can be classified as a spectral method. The spectral methods have become increasingly popular in recent years due to their elegant form based on SVD, and the availability of a rich class of efficient SVD algorithms \citep{drmavc2008new,drmavc2008new2, halko2011finding}. More importantly, spectral methods enjoy nice theoretical properties in various contexts, including network analysis \citep{rohe2011spectral, sussman2012consistent}, matrix completion and denoising \citep{achlioptas2007fast}, spiked covariance estimation \citep{Johnstone2001}, etc; see Section 3.10 in \citet{chen2020spectral} for a comprehensive review. %In general, the theoretical properties of spectral methods can often be derived using the tools in matrix analysis and random matrix literature. We particularly utilize the property of low rank matrix approximation \citep{eckart1936approximation, Pitsianis1997}, unilateral subspace perturbation bound \citep{cai2018rate}, and the newly-generalized Hanson-Wright inequality \citep{zajkowski2020bounds} in  deriving the upper error bounds for our proposed estimators.
\end{remark}
\section{Additional Real Data Example:  S$\&$P 500 Stock Data Analysis}\label{data:stock}
We analyze the S$\&$P 500 stock price dataset, which was first collected by Yahoo! company, processed by \citet{yang2020estimating} and made available in R package {\tt loggle}. S$\&$P data contains the daily closing prices of  stocks from January 1st, 2007 to December 31th, 2016. The stocks in the original dataset are classified into $5$ Global Industry Classification Standard (GICS) sectors. We are interested in 90 stocks that belongs to sectors of Information Technology (IT) and Consumer Staples (CS). The primary variable of interest is the daily percentage changes (DPCs) of each stock, defined as, 
\bee\nonumber
\text{DPC} = \frac{\text{closing price today  (\textdollar) - closing price yesterday (\textdollar)}}{\text{closing price yesterday (\textdollar)}} \times  100\%.
\ee
Our main focus is to estimate the covariance of DPCs between different stocks at different weekdays. In our framework, we represent the data as a matrix, denoted by $ {\bf X}_i$, for week $i$,  where the rows represent $90$ different stocks and the columns represent the five business days of a week (Monday to Friday). To help better explore the sparse correlation patterns in the stocks, we use Isomap \citep{wagaman2009discovering} to reorder the 90 stocks and encourage a potential bandable structure in the covariance matrix. Also the bandwidths selected by resampling scheme in the following discussions indicate the bandable structure fits the data well.
\par
We then pre-process the dataset by removing the mean trend and time-dependence in a similar way as for the temperature anomaly data in Section \ref{realdata} of the main paper. This is implemented by (i) fitting a separate linear model on DPCs for each stock and each business day, and then removing the estimated time trend, (ii) ``thinning" the sequence of weekly measurements by only picking the data matrix of first week in each month, and (iii) dropping the data matrix with missing entries because of holiday closings. To demonstrate its performance, we use the stock values of Colgate-Palmolive Company (CL) on Monday as an example. In Figure \ref{stock_trend_acf}, panels (a) and (b) show the data before and after the detrending, and (c) show the estimated auto-correlation function for the thinned sequence of data. It is clear that both the detrending and thinning steps have achieved a satisfactory performance. 
%To get a valid bandable structure of covariance in the stock (row) direction, we further reorder the stocks via Isomap \citep{wagaman2009discovering}.
\begin{figure}
   \centering
    \subfigure[]{\includegraphics[width=0.4\textwidth]{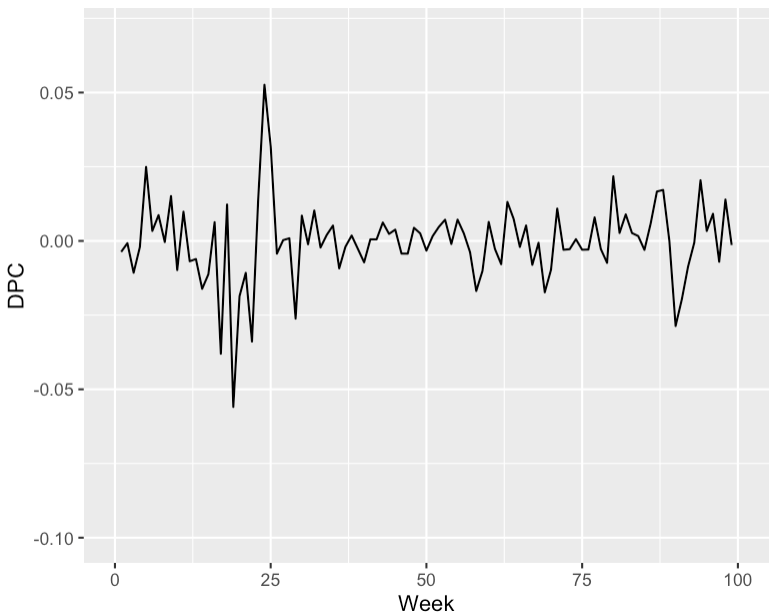}}
  \subfigure[]{\includegraphics[width=0.4\textwidth]{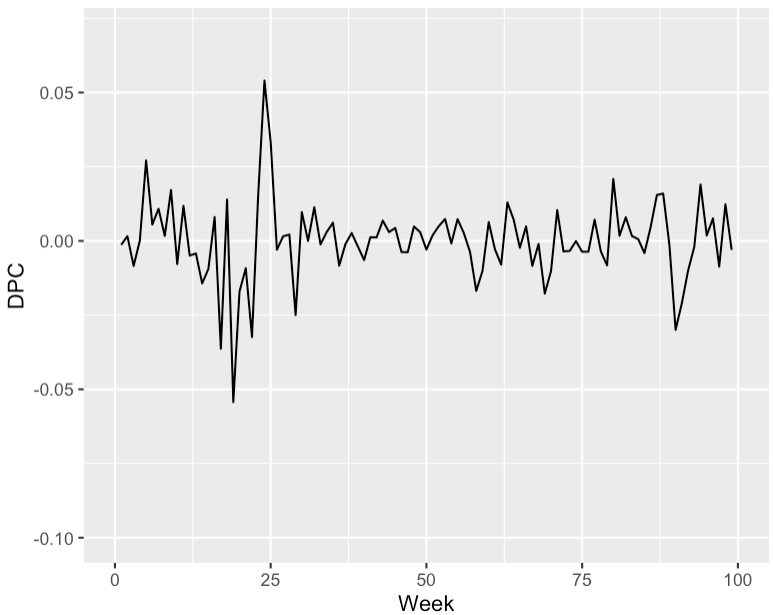}}
   \subfigure[]{\includegraphics[width=0.4\textwidth]{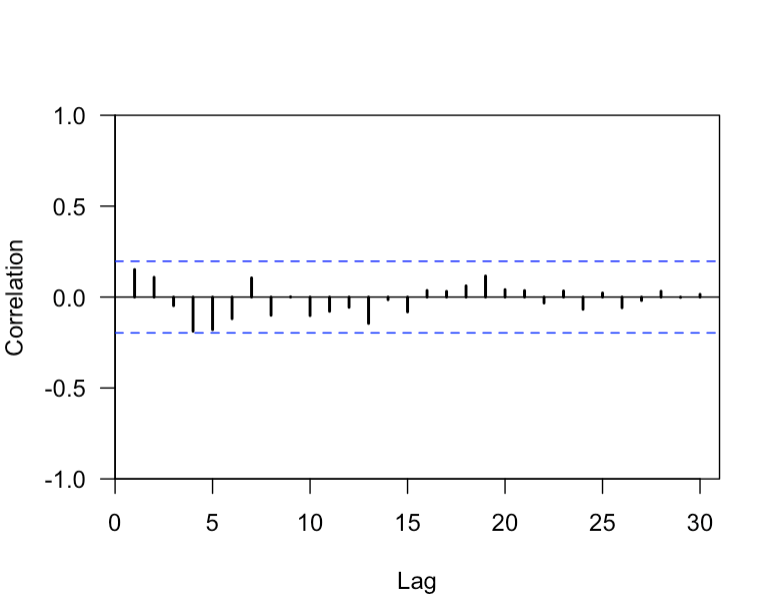}}
     \subfigure[]{ \includegraphics[width=0.4\textwidth]{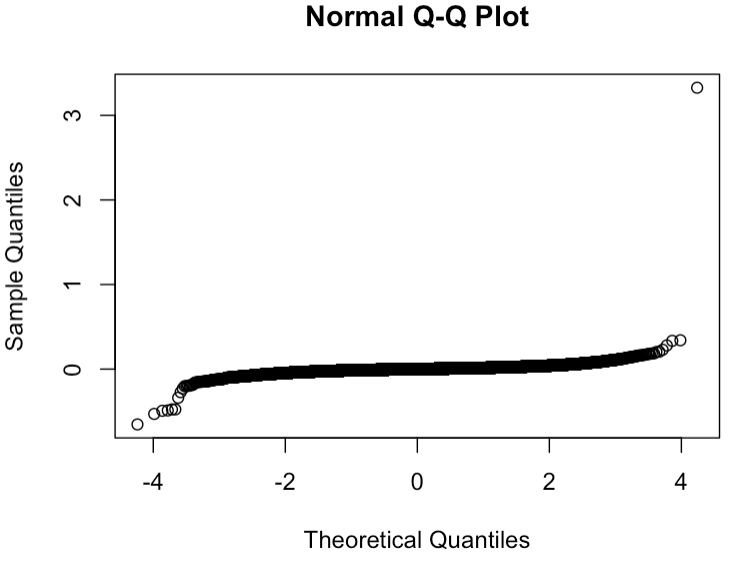}}
\caption{S\&P 500 stock data analysis, (a) and (b): Daily percentage changes before and after detrending; (c): Auto-correlation function for the thinned sequence; (d): The Q-Q plot of all sampled daily percentage changes.}
\label{stock_trend_acf}
\end{figure}
 \par
After pre-processing, we obtain a dataset of $90$ (stock) $\times$ $5$ (day of a week) with a sample size of $n = 99$. Similar to the gridded temperature anomaly data,  we further verify the separability of covariance assumption via two procedures. First we implement the projection-based empirical bootstrap test \citep{aston2017tests}. The p-value of separable test is $0.346$. Hence it is reasonable to assume the separability of the covariance matrix. Second, we calculate the two-fold cross-validation prediction error under 1-norm and Frobenius norm when separability of the dataset is being or being not assumed. {\color{black} Similar to the temperature gridded dataset, the average prediction errors (based on 500 random splits) are 
$R_{\text{ns},1} = 0.589
$ $(\text{SE}=0.0005)$  and $R_{\text{ns},\F} =  0.239$ $(0.0001)$ when no separability is assumed, in comparison with  $R_{\text{ns},1} = 0.424$ $(0.0058)$ and $R_{\text{ns},\F} =  0.236$ $(0.0001)$ when separability is assumed, which also suggests that separability is a reasonable structure assumption for the underlying true covariance. Also, in Figure \ref{stock_trend_acf} (d), clearly, there are several outliers in the Q-Q plot. Therefore, we apply our proposed robust banded and tapering covariance estimation methods. The candidate set $\mathscr{P}$ of threshold parameter $\tau$ is chosen to be the same as in \eqref{candidate}. The resampling scheme chooses $\hat{k}^{\MB}_{1} = 10$, $\hat{k}^{\MB}_{2} = 1$ and  $\hat{\tau}^{\MB}_{} = |x|_{99.95} = 0.186$  for the proposed robust banded estimator and picks $\hat{k}^{\mathcal{T}}_{1} = 12$, $\hat{k}^{\mathcal{T}}_{2} = 2$ and  $\hat{\tau}^{\mathcal{T}}_{} = |x|_{99.99} = 0.489$ for the robust tapering estimator with a random split of $N = 10$ times. 
\begin{figure}
  \centering
    \subfigure[]{\includegraphics[width=0.32\textwidth]{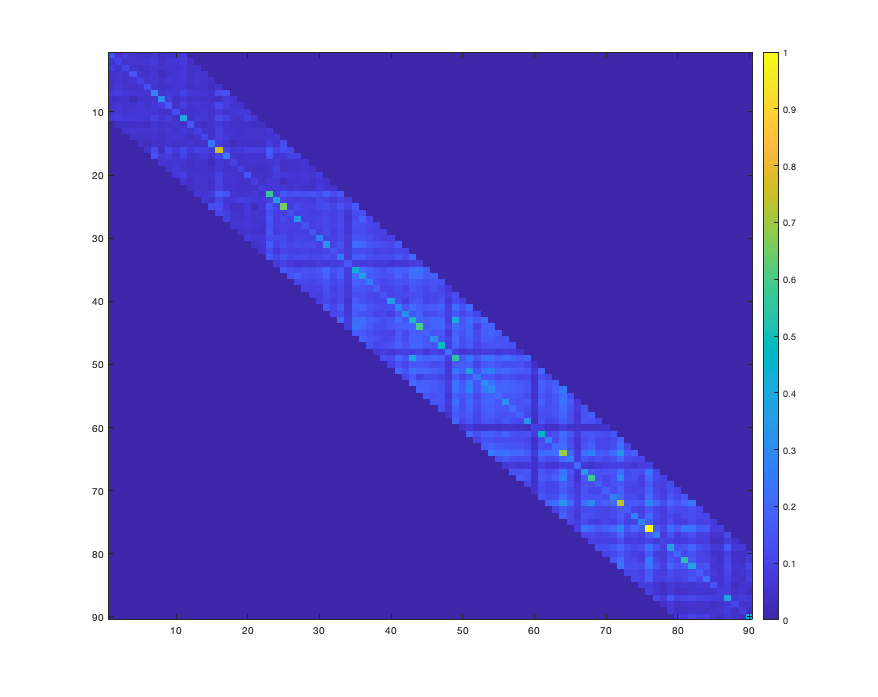}}
  \subfigure[]{\includegraphics[width=0.32\textwidth]{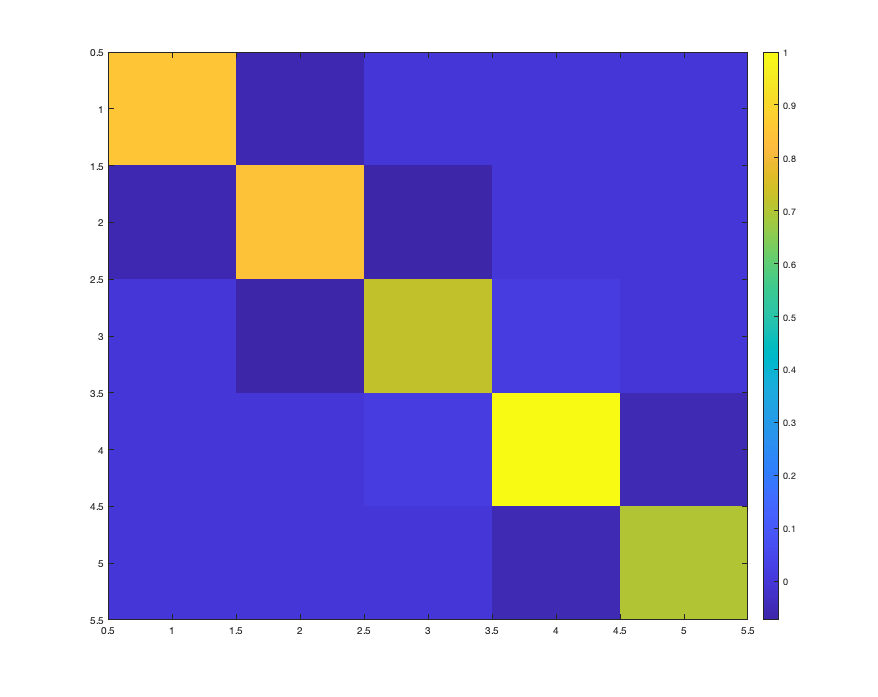}}
     \subfigure[]{\includegraphics[width=0.32\textwidth]{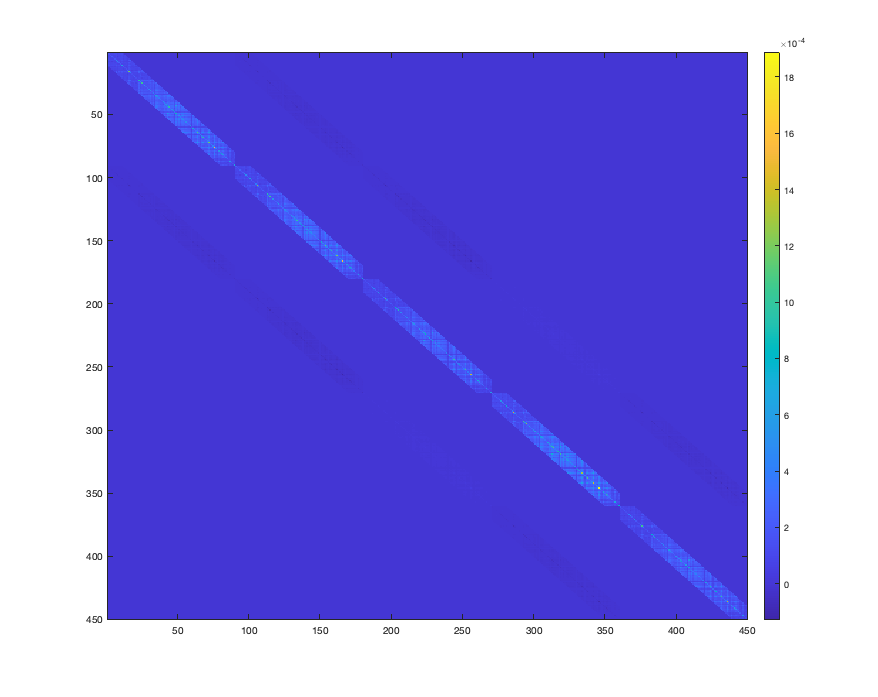}}
          \subfigure[]{\includegraphics[width=0.32\textwidth]{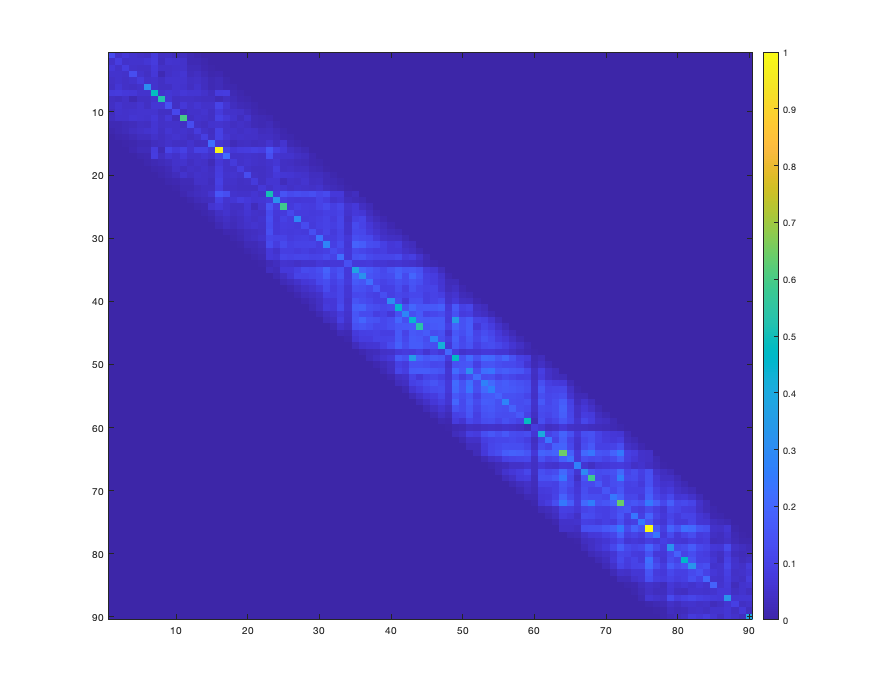}}
  \subfigure[]{\includegraphics[width=0.32\textwidth]{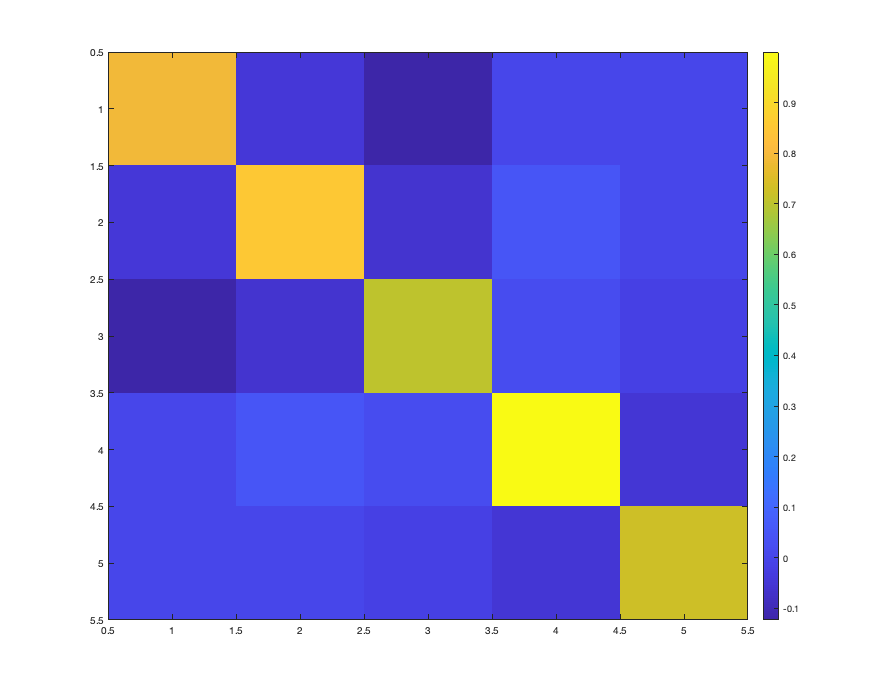}}
      \subfigure[]{\includegraphics[width=0.32\textwidth]{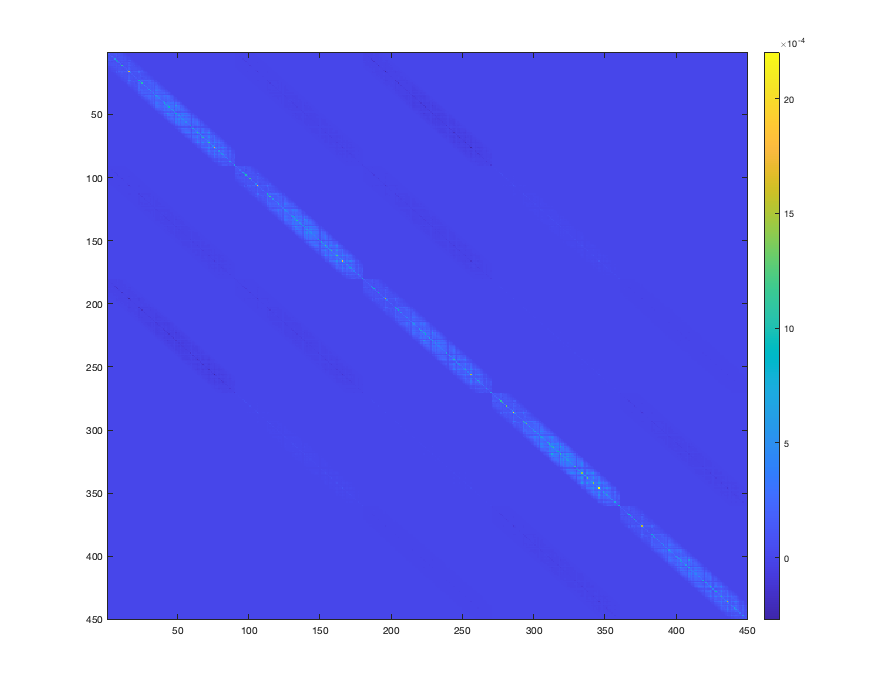}}
      %change to climate_taper_whole.pdf and climate_band_whole.pdf for original plot
\caption{S$\&$P 500 stock data analysis:  estimated covariance matrices obtained by robust banded and tapering methods: (a)  scaled banded covariance for stock direction $ \hat{\boldsymbol\Sigma}^{\mathcal{R},\MB}_{1}(\hat{k}^{\mathcal{B}}_{1}) $, (b) scaled banded covariance for weekday direction $ \hat{\boldsymbol\Sigma}^{\mathcal{R},\MB}_{2}(\hat{k}^{\mathcal{B}}_{2}) $, (c) overall  banded covariance $ \hat{\boldsymbol\Sigma}^{\mathcal{R},\MB}_{}(\hat{k}^{\mathcal{B}}_{1},\hat{k}^{\mathcal{B}}_{2})$. Panel (d), scaled 
tapering covariance for stock direction $ \hat{\boldsymbol\Sigma}^{\mathcal{R},\mathcal{T}}_{1}(\hat{k}^{\mathcal{T}}_{1})$, (e) scaled tapering covariance for weekday direction $ \hat{\boldsymbol\Sigma}^{\mathcal{R},\mathcal{T}}_{2}(\hat{k}^{\mathcal{T}}_{2})$, (f), overall tapering covariance $ \hat{\boldsymbol\Sigma}^{\mathcal{R},\mathcal{T}}_{}(\hat{k}^{\mathcal{T}}_{1},\hat{k}^{\mathcal{T}}_{2})$.}
\label{real_stock}
\end{figure}}
\par
We plot the estimated covariance matrices in Figure \ref{real_stock}. In (a) and (d), we present the estimated covariance matrices over the stock direction obtained by banded and tapering (denoted by $\hat{\M \Sigma}_{1}^{\mathcal{R},\MB}(\hat{k}^{\mathcal{B}}_{1})$,$\hat{\M\Sigma}_{1}^{\mathcal{R},\mathcal{T}}(\hat{k}^{\mathcal{T}}_{1})\in \RR^{90 \times 90}$), and in (b) and (e), we have the estimated covariance matrices for weekday direction, which are denoted by $\hat{\M \Sigma}_{2}^{\mathcal{R},\MB}(\hat{k}^{\mathcal{R},\mathcal{B}}_{2})$,$\hat{\M\Sigma}_{2}^{\mathcal{R},\mathcal{T}}(\hat{k}^{\mathcal{T}}_{2})\in \RR^{5 \times 5}$. All these matrices are scaled such that the maximum entry is $1$. The overall covariance matrices ($\hat{\M \Sigma}_{}^{\mathcal{R},\MB}(\hat{k}^{\mathcal{B}}_{1},\hat{k}^{\mathcal{R},\mathcal{B}}_{2})$, $\hat{\M \Sigma}_{}^{\mathcal{R},\mathcal{T}}(\hat{k}^{\mathcal{T}}_{1},\hat{k}^{\mathcal{T}}_{2}) \in \RR^{450\times 450}$) are plotted in Figure \ref{real_stock} (c) and (f). In general we find that the covariance patterns are very similar between banded and tapering estimates, and both methods clearly demonstrate a bandable structure in the estimated covariance matrix along both stock and weekday directions. For the weekday direction, as shown in Figure \ref{real_stock} (b) and (e), there is a strong self-correlation in DPCs at each weekday, and negative correlations in DPCs among Monday, Tuesday and Wednesday, and between Thursday and Friday. For the stock direction, small values for selected bandwidths ($10$ for banded and $12$ for tapering) confirm the utility of reordering the stocks by Isomap \citep{wagaman2009discovering} on our dataset, and also suggest a sparse correlation pattern in the 90 stocks that we consider. This finding is consistent with a previous study by \citet{yang2020estimating}, where the authors analyzed the same dataset using a Gaussian graphical model and showed that stocks from IT and CS sectors have a low percentage of within-sector edges as well as a weak cross-sector interaction in the graphical model. 
\section{Additional Numerical Results}\label{sec:num}
\subsection{Additional Simulation Results}\label{sec:num:simu:1}
In this section, we present additional simulation results for Section \ref{sim} of the main paper. Table \ref{Table:simulation1} contains results for $(p,q,\rho_1,\rho_2)=(20,30,0.5,0.5)$ with covariance structure of MA(1); Tables \ref{Table:simulation2:01}--\ref{addsim1} contains results for $(p,q, \rho_1, \rho_2)=(20,30, 0.1, 0.1)$, $(20,30, 0.5, 0.5)$, and $(20,30, 0.8, 0.8)$ with AR(1) covariance structure;  Tables \ref{Table:simulation4}--\ref{addsims4} contains results for $(p,q, \rho_1,\rho_2)=(100, 100, 0.1, 0.1),(100, 100, 0.5, 0.5), (100, 100, 0.8, 0.8)$ with AR(1) covariance structure; Table \ref{tb:s5simu} contains results for heavy-tailed data with $(p,q,\rho_1,\rho_2)=(20,30,0.5,0.5)$ and MA(1) covariance structure. All those tables confirm the excellent estimation accuracy performance of our proposed method.

\begin{table}[htp]
\setlength{\tabcolsep}{-1pt}
  \centering
    \caption{Additional Simulation results for $(p,q,\rho_1,\rho_2)=(20,30,0.5,0.5)$ with covariance structure of MA(1) over 100 replications.
    The averages of $\|\hat{\boldsymbol\Sigma}-\boldsymbol\Sigma\|_{\F}$, $ \|\hat{\boldsymbol\Sigma}-\boldsymbol\Sigma\|_1 $ and $ \|\hat{\boldsymbol\Sigma}-\boldsymbol\Sigma\|_{2} $ for 
the proposed estimators (Proposed B and Proposed T), doubly banded and tapering estimators (Doubly B and Doubly T), Bickel's banded estimator (Banded), Cai's tapering estimator (Tapering) and the sample covariance estimator (Sample) are reported. The averages of $ \hat{k}_1 $ and $ \hat{k}_2 $ for the proposed and doubly banded/tapering estimators, and the averages of  $ \hat{k} $ for Bickel's banded estimator and Cai's tapering estimator are also reported.  }
\vspace{0.5pt}
\begin{lrbox}{\tableboxb}
\begin{tabular}{p{0.24\textwidth}>{\centering}p{0.23\textwidth}p{0.15\textwidth}p{0.15\textwidth}p{0.15\textwidth}p{0.08\textwidth}p{0.08\textwidth}p{0.08\textwidth}}
  \hline
$ (n, p, q, \rho_1,\rho_2) $ & Method & $\|\hat{\boldsymbol\Sigma}-\boldsymbol\Sigma\|_{\F}$ & $ \|\hat{\boldsymbol\Sigma}-\boldsymbol\Sigma\|_1 $ & $ \|\hat{\boldsymbol\Sigma}-\boldsymbol\Sigma\|_{2} $ & $\hat{k}$ \quad & $\hat{k}_1$ \quad & $\hat{k}_2$ \quad \\ \hline
$(50,20,30,0.5,0.5)$ & Sample &85.98&94.73&22.02\\& Banded &22.04& 3.61& 2.40& 1.11&&\\& Tapering &22.56& 3.97& 2.54& 2.10&&\\& Doubly B &11.46&3.37&2.02&&1.00&1.00\\& Doubly T &19.38&5.36&2.81&&1.64&1.66\\&Proposed B &3.95&1.24&0.78&&1.65&1.41\\& Proposed T &4.30&1.40&0.85&&2.00&2.00\\
$(100,20,30,0.5,0.5)$ & Sample &60.44&62.36&13.08\\& Banded &21.33& 3.03& 2.18& 1.01&&\\& Tapering &21.61& 3.27& 2.27& 2.02&&\\& Doubly B &8.07&2.23&1.35&&1.01&1.00\\& Doubly T &12.26&3.96&1.93&&2.00&2.00\\& Proposed B &2.71&0.82&0.52&&1.43&1.47\\& Proposed T &2.97&0.93&0.57&&2.00&2.00\\
$(200,20,30,0.5,0.5)$ & Sample &42.65&41.82& 8.13\\& Banded &21.01& 2.71& 2.08& 1.05&&\\& Tapering &21.15& 2.88& 2.13& 2.02&&\\& Doubly B &5.65&1.52&0.91&&1.00&1.00\\& Doubly T &8.61&2.71&1.32&&2.00&2.00\\& Proposed B &1.87&0.55&0.35&&1.57&1.37\\& Proposed T &2.05&0.63&0.39&&2.00&2.00\\

  \hline
\end{tabular}
\end{lrbox}
\label{Table:simulation1}
\scalebox{0.75}{\usebox{\tableboxb}}
\end{table}
\begin{table}[htp]
\setlength{\tabcolsep}{-1pt}
  \centering
    \caption{Additional Simulation results for $(p,q, \rho_1, \rho_2)=(20,30, 0.1, 0.1)$ with AR(1) covariance structure over 100 replications.
    The averages of $\|\hat{\boldsymbol\Sigma}-\boldsymbol\Sigma\|_{\F}$, $ \|\hat{\boldsymbol\Sigma}-\boldsymbol\Sigma\|_1 $ and $ \|\hat{\boldsymbol\Sigma}-\boldsymbol\Sigma\|_{2} $ for 
the proposed estimators (Proposed B and Proposed T), doubly banded and tapering estimators (Doubly B and Doubly T), Bickel's banded estimator (Banded), Cai's tapering estimator (Tapering) and the sample covariance estimator (Sample) are reported. The averages of $ \hat{k}_1 $ and $ \hat{k}_2 $ for the proposed and doubly banded/tapering estimators, the averages of  $ \hat{k} $ for Bickel's banded estimator and Cai's tapering estimator are reported.}
\vspace{0.01pt}
\begin{lrbox}{\tableboxb}
\begin{tabular}{p{0.24\textwidth}>{\centering}p{0.23\textwidth}p{0.15\textwidth}p{0.15\textwidth}p{0.15\textwidth}p{0.08\textwidth}p{0.08\textwidth}p{0.08\textwidth}}
  \hline
$ (n, p, q, \rho_1,\rho_2) $ & Method & $\|\hat{\boldsymbol\Sigma}-\boldsymbol\Sigma\|_{\F}$ & $ \|\hat{\boldsymbol\Sigma}-\boldsymbol\Sigma\|_1 $ & $ \|\hat{\boldsymbol\Sigma}-\boldsymbol\Sigma\|_{2} $ & $\hat{k}$ \quad & $\hat{k}_1$ \quad & $\hat{k}_2$ \quad \\ \hline
$(50,20,30,0.1,0.1)$& Sample &85.87&92.61&18.92\\& Banded &7.86&1.50&0.96&1.02&&\\& Tapering &6.99&1.24&0.81&1.03&&\\& Doubly B &7.26&1.35&0.88&&0.17&0.23\\& Doubly T &7.13&1.30&0.84&&0.04&0.16\\& Proposed B &2.64&0.49&0.31&&1.19&1.20\\& Proposed T &3.03&0.56&0.35&&1.76&1.60\\
$(100,20,30,0.1,0.1)$ & Sample &60.40&60.59&10.85\\& Banded &6.10&1.12&0.67&1.07&&\\& Tapering &6.04&1.01&0.61&1.08&&\\& Doubly B &6.05&1.09&0.65&&0.35&0.34\\& Doubly T &6.13&1.05&0.63&&0.24&0.14\\& Proposed B &1.88&0.35&0.22&&1.27&1.23\\& Proposed T &2.12&0.39&0.24&&1.84&1.84\\
$(200,20,30,0.1,0.1)$& Sample &42.60&40.61& 6.46\\& Banded &4.93&0.85&0.48&1.04&&\\& Tapering &5.44&0.86&0.53&1.18&&\\& Doubly B &4.98&0.86&0.50&&0.48&0.45\\& Doubly T &5.42&0.89&0.53&&0.38&0.32\\& Proposed B &1.25&0.25&0.15&&1.67&1.64\\& Proposed T &1.27&0.26&0.15&&2.00&1.98\\
\hline
\end{tabular}
\end{lrbox}
\label{Table:simulation2:01}
\scalebox{0.75}{\usebox{\tableboxb}}
\end{table}

\begin{table}[htp]
\setlength{\tabcolsep}{-1pt}
  \centering
    \caption{Additional Simulation results for $(p,q, \rho_1, \rho_2)=(20,30, 0.5, 0.5)$ with AR(1) covariance structure over 100 replications.
    The averages of $\|\hat{\boldsymbol\Sigma}-\boldsymbol\Sigma\|_{\F}$, $ \|\hat{\boldsymbol\Sigma}-\boldsymbol\Sigma\|_1 $ and $ \|\hat{\boldsymbol\Sigma}-\boldsymbol\Sigma\|_{2} $ for 
the proposed estimators (Proposed B and Proposed T), doubly banded and tapering estimators (Doubly B and Doubly T), Bickel's banded estimator (Banded), Cai's tapering estimator (Tapering) and the sample covariance estimator (Sample) are reported. The averages of $ \hat{k}_1 $ and $ \hat{k}_2 $ for the proposed and doubly banded/tapering estimators, the averages of  $ \hat{k} $ for Bickel's banded estimator and Cai's tapering estimator are reported.}
\vspace{0.01pt}
\begin{lrbox}{\tableboxb}
\begin{tabular}{p{0.24\textwidth}>{\centering}p{0.23\textwidth}p{0.15\textwidth}p{0.15\textwidth}p{0.15\textwidth}p{0.08\textwidth}p{0.08\textwidth}p{0.08\textwidth}}
  \hline
$ (n, p, q, \rho_1,\rho_2) $ & Method & $\|\hat{\boldsymbol\Sigma}-\boldsymbol\Sigma\|_{\F}$ & $ \|\hat{\boldsymbol\Sigma}-\boldsymbol\Sigma\|_1 $ & $ \|\hat{\boldsymbol\Sigma}-\boldsymbol\Sigma\|_{2} $ & $\hat{k}$ \quad & $\hat{k}_1$ \quad & $\hat{k}_2$ \quad \\ \hline
$(50,20,30,0.5,0.5)$ & Sample &86.11&96.07&23.51\\& Banded &26.94& 8.50& 6.06& 2.35&&\\& Tapering &26.94& 8.50& 6.05& 2.58&&\\& Doubly B &19.72&8.87&4.05&&1.64&1.71\\& Doubly T &19.62&9.22&3.69&&2.00&2.08\\& Proposed B &8.06&4.35&2.03&&2.82&2.89\\& Proposed T &9.99&4.75&2.65&&2.00&2.00\\
$(100,20,30,0.5,0.5)$ & Sample &60.46&62.78&14.30\\& Banded &26.08& 7.76& 5.96& 2.71&&\\& Tapering &26.07& 7.73& 5.96& 3.04&&\\& Doubly B &15.55&7.05&3.28&&1.96&1.91\\& Doubly T &14.95&6.99&3.06&&2.08&2.02\\& Proposed B &5.82&3.20&1.49&&4.02&3.90\\& Proposed T &5.50&2.97&1.47&&3.94&3.86\\
$(200,20,30,0.5,0.5)$ & Sample &42.70&42.53& 9.06\\& Banded &17.21& 7.57& 3.38&22.04&&\\& Tapering &17.50& 7.64& 3.70&32.66&&\\& Doubly B &12.04&5.70&2.62&&2.09&2.17\\& Doubly T &12.05&5.69&2.65&&2.12&2.22\\& Proposed B &4.16&2.38&1.05&&4.36&4.23\\& Proposed T &4.18&2.36&1.14&&3.98&4.00\\
\hline
\end{tabular}
\end{lrbox}
\label{Table:simulation2}
\scalebox{0.75}{\usebox{\tableboxb}}
\end{table}

\begin{table}[htp]
\setlength{\tabcolsep}{-1pt}
  \centering
    \caption{Additional Simulation results for $(p,q, \rho_1,\rho_2)=(20,30, 0.8, 0.8)$ with AR(1) covariance structure over 100 replications.
    The averages of $\|\hat{\boldsymbol\Sigma}-\boldsymbol\Sigma\|_{\F}$, $ \|\hat{\boldsymbol\Sigma}-\boldsymbol\Sigma\|_1 $ and $ \|\hat{\boldsymbol\Sigma}-\boldsymbol\Sigma\|_{2} $ for 
the proposed estimators (Proposed B and Proposed T), doubly banded and tapering estimators (Doubly B and Doubly T), Bickel's banded estimator (Banded), Cai's tapering estimator (Tapering) and the sample covariance estimator (Sample) are reported. The averages of $ \hat{k}_1 $ and $ \hat{k}_2 $ for the proposed and doubly banded/tapering estimators, the averages of  $ \hat{k} $ for Bickel's banded estimator and Cai's tapering estimator are also reported. }
\vspace{0.5pt}
\begin{lrbox}{\tableboxb}
\begin{tabular}{p{0.24\textwidth}>{\centering}p{0.23\textwidth}p{0.15\textwidth}p{0.15\textwidth}p{0.15\textwidth}p{0.08\textwidth}p{0.08\textwidth}p{0.08\textwidth}}
  \hline
$ (n, p, q, \rho_1,\rho_2) $ & Method & $\|\hat{\boldsymbol\Sigma}-\boldsymbol\Sigma\|_{\F}$ & $ \|\hat{\boldsymbol\Sigma}-\boldsymbol\Sigma\|_1 $ & $ \|\hat{\boldsymbol\Sigma}-\boldsymbol\Sigma\|_{2} $ & $\hat{k}$ \quad & $\hat{k}_1$ \quad & $\hat{k}_2$ \quad \\ \hline
$(50,20,30,0.8,0.8)$ & Sample & 87.53&112.82& 39.48\\& Banded &60.37&59.83&27.35&63.23&&\\& Tapering &61.52&56.92&31.35&68.01&&\\& Doubly B &54.10&57.18&26.65&&4.50&5.05\\& Doubly T &49.75&52.63&25.72&&5.94&6.30\\& Proposed B &33.92&40.19&17.30&& 6.81& 6.97\\& Proposed T &32.06&36.27&18.13&& 7.66& 7.84\\
$(100,20,30,0.8,0.8)$& Sample &61.43&75.67&26.67\\& Banded &46.85&47.32&21.10&85.12&&\\& Tapering & 46.38& 45.50& 23.16&104.08&&\\& Doubly B &42.89&46.04&21.03&&5.69&5.73\\& Doubly T &38.90&41.95&20.63&&7.24&7.22\\& Proposed B &24.34&28.64&12.63&& 8.80& 8.57\\& Proposed T &23.12&26.35&13.07&& 9.64& 9.72\\
$(200,20,30,0.8,0.8)$ & Sample &43.44&52.14&17.76\\& Banded & 35.28& 37.91& 15.50&117.91&&\\& Tapering & 35.59& 37.04& 17.75&134.40&&\\& Doubly B &33.91&37.57&16.42&&6.69&6.58\\& Doubly T &30.49&34.31&16.23&&8.26&8.56\\& Proposed B &17.30&21.08& 9.13&&10.38&10.26\\& Proposed T &17.21&20.39& 9.79&&11.55&11.42\\
%\hline Max SE & &  0.40 & 0.80 & 0.39 & 0.98 & 0.16 & 0.15\\ 
\hline
\end{tabular}\label{addsim1}
\end{lrbox}
\scalebox{0.75}{\usebox{\tableboxb}}
\end{table}

\begin{table}[htp]
\setlength{\tabcolsep}{-1pt}
  \centering
    \caption{Simulation results for $(p,q, \rho_1, \rho_2)=(100,100, 0.1, 0.1)$  with AR(1) covariance structure over 100 replications.
    The averages of $\|\hat{\boldsymbol\Sigma}-\boldsymbol\Sigma\|_{\F}$, $ \|\hat{\boldsymbol\Sigma}-\boldsymbol\Sigma\|_1 $ and $ \|\hat{\boldsymbol\Sigma}-\boldsymbol\Sigma\|_{2} $ for 
the proposed estimators (Proposed B and Proposed T), doubly banded and tapering estimators (Doubly B and Doubly T), Bickel's banded estimator (Banded), Cai's tapering estimator (Tapering) and the sample covariance estimator (Sample) are reported. The averages of $ \hat{k}_1 $ and $ \hat{k}_2 $ for the proposed and doubly banded/tapering estimators, the averages of  $ \hat{k} $ for Bickel's banded estimator and Cai's tapering estimator are also reported. }
\vspace{0.01pt}
\begin{lrbox}{\tableboxb}
\begin{tabular}{p{0.24\textwidth}>{\centering}p{0.23\textwidth}p{0.15\textwidth}p{0.15\textwidth}p{0.15\textwidth}p{0.08\textwidth}p{0.08\textwidth}p{0.08\textwidth}}
  \hline
$ (n, p, q, \rho_1,\rho_2) $ & Method & $\|\hat{\boldsymbol\Sigma}-\boldsymbol\Sigma\|_{\F}$ & $ \|\hat{\boldsymbol\Sigma}-\boldsymbol\Sigma\|_1 $ & $ \|\hat{\boldsymbol\Sigma}-\boldsymbol\Sigma\|_{2} $ & $\hat{k}$ \quad & $\hat{k}_1$ \quad & $\hat{k}_2$ \quad \\ \hline
$(50,100,100,0.1,0.1)$ & Sample &1428.38&1598.50& 231.64\\& Banded &32.14& 1.84& 1.19& 1.02&&\\& Tapering &29.33& 1.53& 1.03& 1.09&&\\& Doubly B &30.40&1.67&1.11&&0.26&0.27\\& Doubly T &29.92&1.59&1.06&&0.16&0.16\\& Proposed B &6.16&0.36&0.20&&2.37&2.28\\& Proposed T &6.44&0.37&0.21&&2.80&2.78\\
$(100,100,100,0.1,0.1)$ & Sample &1004.94&1029.06& 120.68\\& Banded &24.84& 1.31& 0.80& 1.03&&\\& Tapering &25.03& 1.19& 0.72& 1.11&&\\& Doubly B &24.89&1.29&0.79&&0.38&0.43\\& Doubly T &25.81&1.29&0.77&&0.24&0.40\\& Proposed B &4.25&0.26&0.14&&2.78&2.49\\& Proposed T &4.09&0.25&0.14&&3.00&2.98\\
$(200,100,100,0.1,0.1)$& Sample &708.84&679.93& 64.43\\& Banded &22.58& 1.16& 0.62& 2.04&&\\& Tapering &22.49& 1.15& 0.62& 2.00&&\\& Doubly B &20.35&1.00&0.58&&0.52&0.50\\& Doubly T &22.44&1.05&0.59&&0.52&0.52\\& Proposed B &3.41&0.24&0.11&&4.44&4.38\\& Proposed T &2.94&0.19&0.10&&4.01&4.03\\
\hline
\end{tabular}
\end{lrbox}
\label{Table:simulation4}
\scalebox{0.75}{\usebox{\tableboxb}}  
\end{table}

\begin{table}[htp]
\setlength{\tabcolsep}{-1pt}
  \centering
    \caption{Simulation results for $(p,q, \rho_1, \rho_2)=(100,100, 0.5, 0.5)$ with AR(1) covariance structure over 100 replications.
    The averages of $\|\hat{\boldsymbol\Sigma}-\boldsymbol\Sigma\|_{\F}$, $ \|\hat{\boldsymbol\Sigma}-\boldsymbol\Sigma\|_1 $ and $ \|\hat{\boldsymbol\Sigma}-\boldsymbol\Sigma\|_{2} $ for 
the proposed estimators (Proposed B and Proposed T), doubly banded and tapering estimators (Doubly B and Doubly T), Bickel's banded estimator (Banded), Cai's tapering estimator (Tapering) and the sample covariance estimator (Sample) are reported. The averages of $ \hat{k}_1 $ and $ \hat{k}_2 $ for the proposed and doubly banded/tapering estimators, the averages of  $ \hat{k} $ for Bickel's banded estimator and Cai's tapering estimator are also reported. }
\vspace{0.01pt}
\begin{lrbox}{\tableboxb}
\begin{tabular}{p{0.24\textwidth}>{\centering}p{0.23\textwidth}p{0.15\textwidth}p{0.15\textwidth}p{0.15\textwidth}p{0.08\textwidth}p{0.08\textwidth}p{0.08\textwidth}}
  \hline
$ (n, p, q, \rho_1,\rho_2) $ & Method & $\|\hat{\boldsymbol\Sigma}-\boldsymbol\Sigma\|_{\F}$ & $ \|\hat{\boldsymbol\Sigma}-\boldsymbol\Sigma\|_1 $ & $ \|\hat{\boldsymbol\Sigma}-\boldsymbol\Sigma\|_{2} $ & $\hat{k}$ \quad & $\hat{k}_1$ \quad & $\hat{k}_2$ \quad \\ \hline
$(50,100,100,0.5,0.5)$& Sample &1428.37&1606.57& 250.19\\& Banded &113.01&  9.54&  6.45&  2.68&&\\& Tapering &112.91&  9.50&  6.46&  2.96&&\\& Double B &83.49&11.04&4.54&&1.81&1.77\\& Double T &83.36&11.63&4.36&&2.06&2.08\\& Proposed B &19.27& 3.20& 1.46&& 5.30& 5.11\\& Proposed T &17.81& 2.85& 1.42&& 5.48& 5.06\\
$(100,100,100,0.5,0.5)$& Sample &1005.05&1029.25& 134.55\\& Banded &109.18&  8.45&  6.33&  2.91&&\\& Tapering &109.11&  8.40&  6.33&  3.32&&\\& Doubly B &65.54&8.49&3.59&&2.09&1.98\\& Doubly T &64.29&8.53&3.50&&2.18&2.10\\& Proposed B &13.69& 2.47& 1.07&& 6.65& 6.43\\& Proposed T &13.09& 2.26& 1.06&& 7.18& 7.14\\
$(200,100,100,0.5,0.5)$ & Sample &708.97&684.42& 74.77\\& Banded &107.05&  7.79&  6.24&  3.23&&\\& Tapering &106.94&  7.73&  6.23&  3.88&&\\& Doubly B &51.87&6.81&3.01&&2.16&2.29\\& Doubly T &51.78&6.77&3.04&&2.30&2.30\\& Proposed B &10.32& 1.70& 0.80&& 5.52& 5.29\\& Proposed T &9.69&1.61&0.83&&5.68&5.74\\
\hline
\end{tabular}
\end{lrbox}
\label{Table:simulation4:2}
\scalebox{0.75}{\usebox{\tableboxb}}  
\end{table}
\begin{table}[htp]
\setlength{\tabcolsep}{-1pt}
  \centering
    \caption{Additional Simulation results for $(p,q, \rho_1,\rho_2)=(100, 100, 0.8, 0.8)$ with AR(1) covariance structure over 100 replications.
    The averages of $\|\hat{\boldsymbol\Sigma}-\boldsymbol\Sigma\|_{\F}$, $ \|\hat{\boldsymbol\Sigma}-\boldsymbol\Sigma\|_1 $ and $ \|\hat{\boldsymbol\Sigma}-\boldsymbol\Sigma\|_{2} $ for 
the proposed estimators (Proposed B and Proposed T), doubly banded and tapering estimators (Doubly B and Doubly T), Bickel's banded estimator (Banded), Cai's tapering estimator (Tapering) and the sample covariance estimator (Sample) are reported. The averages of $ \hat{k}_1 $ and $ \hat{k}_2 $ for the proposed and doubly banded/tapering estimators, the averages of  $ \hat{k} $ for Bickel's banded estimator and Cai's tapering estimator are also reported.}
\vspace{0.5pt}
\begin{lrbox}{\tableboxb}
\begin{tabular}{p{0.24\textwidth}>{\centering}p{0.23\textwidth}p{0.15\textwidth}p{0.15\textwidth}p{0.15\textwidth}p{0.08\textwidth}p{0.08\textwidth}p{0.08\textwidth}}
  \hline
$ (n, p, q, \rho_1,\rho_2) $ & Method & $\|\hat{\boldsymbol\Sigma}-\boldsymbol\Sigma\|_{\F}$ & $ \|\hat{\boldsymbol\Sigma}-\boldsymbol\Sigma\|_1 $ & $ \|\hat{\boldsymbol\Sigma}-\boldsymbol\Sigma\|_{2} $ & $\hat{k}$ \quad & $\hat{k}_1$ \quad & $\hat{k}_2$ \quad \\ \hline
$(50,100,100,0.8,0.8)$& Sample &1429.90&1642.93& 335.63\\& Banded &400.02& 81.89& 70.20& 11.38&&\\& Tapering &398.82& 80.70& 70.60& 12.32&&\\& Double B &238.94&92.80&35.97&&6.44&6.48\\& Double T &221.82&82.78&36.64&&8.31&8.18\\& Proposed B &86.55&42.92&17.88&&14.13&13.12\\& Proposed T &78.24&37.44&17.17&&16.48&15.78\\
$(100,100,100,0.8,0.8)$& Sample &1006.18&1069.56& 196.81\\& Banded &396.39& 78.84& 70.13& 11.54&&\\& Tapering &395.99& 78.21& 70.39& 13.46&&\\& Double B &190.17&72.81&30.62&&7.15&6.88\\& Double T &174.2&64.95&30.64&&9.02&8.95\\& Proposed B &62.98&32.43&12.78&&16.75&16.89\\& Proposed T &58.92&29.22&12.60&&20.78&20.70\\
$(200,100,100,0.8,0.8)$& Sample &709.59&711.09&120.51\\& Banded &394.56& 76.71& 70.15& 11.21&&\\& Tapering &394.34& 76.33& 70.28& 13.68&&\\& Double B &151.33&58.09&25.81&&7.70&7.76\\& Double T &137.57&53.45&25.56&&10.10 &10.18\\& Proposed B &44.19&22.22& 9.00&&14.66&14.87\\& Proposed T &41.61&20.62& 9.21&&19.24&19.24\\
%\hline Max SE & &  0.68 & 0.71 & 0.30 & 0.23 & 0.52 & 0.54\\ 
\hline
\end{tabular}\label{addsims4}
\end{lrbox}
\scalebox{0.75}{\usebox{\tableboxb}}
\end{table}

\begin{table}[htp]
\setlength{\tabcolsep}{-1pt}
  \centering
    \caption{Additional Simulation results for heavy-tailed data with $(p,q,\rho_1,\rho_2)=(20,30,0.5,0.5)$ and MA(1) covariance structure over 100 replications. The averages of $\|\hat{\boldsymbol\Sigma}-\boldsymbol\Sigma\|_{\F}$, $ \|\hat{\boldsymbol\Sigma}-\boldsymbol\Sigma\|_1 $ and $ \|\hat{\boldsymbol\Sigma}-\boldsymbol\Sigma\|_{2} $ for our robust estimators (Robust B and Robust T), 
our proposed estimators (Proposed B and Proposed T) and the naive sample covariance estimator (Sample) are summarized in this table. The averages of $ \hat{k}_1 $ and $ \hat{k}_2 $ for the proposed robust/non-robust methods, the averages of  $ \hat{\tau} $ for the proposed robust methods are also reported.}
\vspace{0.5pt}
\begin{lrbox}{\tableboxb}
\Rotatebox{0}{

\begin{tabular}{p{0.24\textwidth}>{\centering}p{0.23\textwidth}p{0.15\textwidth}p{0.15\textwidth}p{0.15\textwidth}p{0.08\textwidth}p{0.08\textwidth}p{0.08\textwidth}}
   \hline
$(n,p,q,\rho_1,\rho_2)$ & Method & $\| \hat{\Sigma} - \Sigma\|_{\F}$ & $\|\hat{\Sigma} - \Sigma\|_1$ & $\| \hat{\Sigma} - \Sigma \|_{2}$ & $\hat{k}_1$ & $\hat{k}_2$ & $\hat{\tau}$
\\
\hline
$(50,20,30,0.5,0.5)$
& Sample &244.85&553.94&208.93\\& Proposed B &19.54&6.97&4.43&1.34&1.49\\& Proposed T &16.99&4.27&2.91&1.82&1.84\\& Robust B&13.42&3.97&2.50&1.46&1.59&6.85\\& Robust T&11.75&2.89&1.98&1.98&2.00&5.94\\
$(100,20,30,0.5,0.5)$
& Sample &181.03&391.29&148.41\\& Proposed B &15.46&6.52&3.95&1.67&1.55\\& Proposed T &14.33&4.14&2.79&1.82&1.94\\& Robust B&11.46&3.44&2.15&1.76&1.77&7.93\\& Robust T&9.89&2.49&1.70&2.00&2.00&7.26\\
$(200,20,30,0.5,0.5)$& Sample &126.68&270.67&97.48\\& Proposed B &10.75&3.39&2.11&1.62&1.47\\& Proposed T &10.54&2.63&1.75&1.92&1.88\\& Robust B&9.65&2.69&1.73&1.69&1.51&10.98\\& Robust T&8.83&2.14&1.48&2.00&2.00&9.75\\
  \hline
\end{tabular}}
\label{tb:s5simu}
\end{lrbox}
\scalebox{0.75}{\usebox{\tableboxb}}
\end{table}

\subsection{Tables for Standard Errors}\label{sec:simu:se}
We present the corresponding standard errors for all tables in main article and Supplementary File in Tables \ref{Table:simulation3SE}--\ref{Table:simulation5SE}. %All standard errors are kept two decimal places.

\begin{table}[htp]
\setlength{\tabcolsep}{-1pt}
  \centering
    \caption{Standard errors for Table \ref{Table:simulation3}}
\vspace{0.5pt}
\begin{lrbox}{\tableboxb}
\begin{tabular}{p{0.24\textwidth}>{\centering}p{0.23\textwidth}p{0.15\textwidth}p{0.15\textwidth}p{0.15\textwidth}p{0.08\textwidth}p{0.08\textwidth}p{0.08\textwidth}}
  \hline
$ (n, p, q, \rho_1,\rho_2) $ & Method & $\|\hat{\boldsymbol\Sigma}-\boldsymbol\Sigma\|_{\F}$ & $ \|\hat{\boldsymbol\Sigma}-\boldsymbol\Sigma\|_1 $ & $ \|\hat{\boldsymbol\Sigma}-\boldsymbol\Sigma\|_{2} $ & $\hat{k}$ \quad & $\hat{k}_1$ \quad & $\hat{k}_2$ \quad \\ \hline
$(50,100,100,0.5,0.5)$ & Sample &0.42&3.85&0.21\\& Banded &0.08&0.03&0.01&0.04&&\\& Tapering &0.01&0.03&0.01&0.00&&\\& Double B &0.16&0.04&0.02&&0.01&0.01\\& Double T &0.81&0.01&0.04&&0.06&0.06\\& Proposed B &0.09&0.01&0.01&&0.09&0.09\\& Proposed T &0.03&0.01&0.01&&0&0\\

$(100,100,100,0.5,0.5)$ & Sample &0.20&2.00&0.13\\& Banded &0.03&0.02&0.01&0.03&&\\& Tapering &0.00&0.01&0.01&0.00&&\\& Double B &0.16&0.04&0.02&&0.01&0.01\\& Double T &0.81&0.01&0.04&&0.06&0.06\\& Proposed B &0.06&0.01&0.00&&0.08&0.08\\& Proposed T &0.03&0.01&0.00&&0&0\\
$(200,100,100,0.5,0.5)$ & Sample &0.09&1.06&0.06\\& Banded &0.02&0.01&0.00&0.03&&\\& Tapering &0.01&0.01&0.00&0.04&&\\& Double B &0.06&0.02&0.01&&0.01&0\\& Double T &0.01&0.02&0.01&&0&0\\& Proposed B &0.05&0.01&0.00&&0.09&0.09\\& Proposed T &0.02&0.00&0.00&&0&0\\
 %\hline Max SE & &  0.09 & 0.03 & 0.01 & 0.04 & 0.09 & 0.09 \\
  \hline
\end{tabular}
\end{lrbox}
\scalebox{0.75}{\usebox{\tableboxb}}
\label{Table:simulation3SE}
\end{table}

\begin{table}[htp]
\setlength{\tabcolsep}{-1pt}
  \centering
    \caption{Standard errors for Table \ref{Table:simulation4}}
\vspace{0.5pt}
\begin{lrbox}{\tableboxb}
\begin{tabular}{p{0.24\textwidth}>{\centering}p{0.23\textwidth}p{0.15\textwidth}p{0.15\textwidth}p{0.15\textwidth}p{0.08\textwidth}p{0.08\textwidth}p{0.08\textwidth}}
  \hline
$ (n, p, q, \rho_1,\rho_2) $ & Method & $\|\hat{\boldsymbol\Sigma}-\boldsymbol\Sigma\|_{\F}$ & $ \|\hat{\boldsymbol\Sigma}-\boldsymbol\Sigma\|_1 $ & $ \|\hat{\boldsymbol\Sigma}-\boldsymbol\Sigma\|_{2} $ & $\hat{k}$ \quad & $\hat{k}_1$ \quad & $\hat{k}_2$ \quad \\ \hline
$(50,100,100,0.1,0.1)$ & Sample &0.29&3.80&0.16\\& Banded &0.08&0.02&0.01&0.01&&\\& Tapering &0.27&0.03&0.01&0.03&&\\& Doubly B &0.21&0.03&0.01&&0.04&0.05\\& Doubly T &0.35&0.03&0.02&&0.05&0.06\\& Proposed B &0.23&0.01&0.00&&0.12&0.13\\& Proposed T &0.28&0.01&0.01&&0.11&0.12\\
$(100,100,100,0.1,0.1)$ & Sample &0.14&1.78&0.06\\& Banded &0.06&0.01&0.01&0.02&&\\& Tapering &0.12&0.02&0.01&0.03&&\\& Doubly B &0.09&0.02&0.01&&0.05&0.06\\& Doubly T &0.19&0.03&0.01&&0.07&0.09\\& Proposed B &0.11&0.01&0.00&&0.15&0.16\\& Proposed T &0.11&0.00&0.00&&0.10&0.10\\
$(200,100,100,0.1,0.1)$& Sample &0.07&1.03&0.03\\& Banded &0.04&0.01&0.00&0.02&&\\& Tapering &0.01&0.01&0.00&0.00&&\\& Doubly B &0.06&0.01&0.00&&0.06&0.05\\& Doubly T &0.00&0.01&0.00&&0.09&0.09\\& Proposed B &0.04&0.00&0.00&&0.15&0.15\\& Proposed T &0.03&0.00&0.00&&0.13&0.13\\
\hline
\end{tabular}
\end{lrbox}
\scalebox{0.75}{\usebox{\tableboxb}}
\end{table}

\begin{table}[htp]
\setlength{\tabcolsep}{-1pt}
  \centering
    \caption{Standard errors for Table \ref{Table:simulation4:2}}
\vspace{0.5pt}
\begin{lrbox}{\tableboxb}
\begin{tabular}{p{0.24\textwidth}>{\centering}p{0.23\textwidth}p{0.15\textwidth}p{0.15\textwidth}p{0.15\textwidth}p{0.08\textwidth}p{0.08\textwidth}p{0.08\textwidth}}
  \hline
$ (n, p, q, \rho_1,\rho_2) $ & Method & $\|\hat{\boldsymbol\Sigma}-\boldsymbol\Sigma\|_{\F}$ & $ \|\hat{\boldsymbol\Sigma}-\boldsymbol\Sigma\|_1 $ & $ \|\hat{\boldsymbol\Sigma}-\boldsymbol\Sigma\|_{2} $ & $\hat{k}$ \quad & $\hat{k}_1$ \quad & $\hat{k}_2$ \quad \\ \hline
$(50,100,100,0.5,0.5)$& Sample &0.47&4.14&0.27\\& Banded &0.09&0.05&0.01&0.08&&\\& Tapering &0.06&0.05&0.01&0.12&&\\& Double B &0.25&0.13&0.04&&0.06&0.05\\& Double T &0.28&0.12&0.05&&0.03&0.04\\& Proposed B &0.26&0.04&0.02&&0.15&0.15\\& Proposed T &0.33&0.03&0.02&&0.09&0.12\\
$(100,100,100,0.5,0.5)$& Sample &0.21&1.84&0.15\\& Banded &0.03&0.02&0.01&0.07&&\\& Tapering &0.02&0.02&0.01&0.10&&\\& Doubly B &0.32&0.06&0.03&&0.05&0.04\\& Doubly T &0.17&0.06&0.01&&0.06&0.04\\& Proposed B &0.09&0.03&0.01&&0.23&0.20\\& Proposed T &0.09&0.02&0.01&&0.23&0.22\\
$(200,100,100,0.5,0.5)$ & Sample &0.11&1.03&0.07\\& Banded &0.02&0.02&0.01&0.06&&\\& Tapering &0.03&0.02&0.01&0.08&&\\& Doubly B &0.12&0.05&0.02&&0.04&0.05\\& Doubly T &0.02&0.05&0.02&&0.07&0.07\\& Proposed B &0.18&0.02&0.01&&0.12&0.12\\& Proposed T &0.13&0.02&0.01&&0.07&0.07\\
\hline
\end{tabular}
\end{lrbox}
\scalebox{0.75}{\usebox{\tableboxb}}
\end{table}

\begin{table}[htp]
\setlength{\tabcolsep}{-1pt}
  \centering
    \caption{Standard errors for Table \ref{Table:simulation6}}
\vspace{0.5pt}
\begin{lrbox}{\tableboxb}\Rotatebox{0}{
\begin{tabular}{p{0.24\textwidth}>{\centering}p{0.23\textwidth}p{0.15\textwidth}p{0.15\textwidth}p{0.15\textwidth}p{0.08\textwidth}p{0.08\textwidth}p{0.08\textwidth}}
  \hline
$(n,p,q,\rho_1,\rho_2)$ & Method & $\| \hat{\Sigma} - \Sigma\|_{\F}$ & $\|\hat{\Sigma} - \Sigma\|_1$ & $\| \hat{\Sigma} - \Sigma \|_{2}$ & $\hat{k}_1$ & $\hat{k}_2$ & $\hat{\tau}$
\\
\hline
$(50,20,30,0.1,0.1)$& Sample &35.68&100.98&36.41\\& Proposed B&1.72&0.62&0.42&0.11&0.11\\& Proposed T&1.47&0.22&0.21&0.10&0.10\\& Robust B&0.75&0.23&0.12&0.10&0.11&0.32\\& Robust T&0.74&0.11&0.08&0.10&0.10&0.37\\
$(50,20,30,0.5,0.5)$& Sample &36.30&99.56&37.02\\& Proposed B&2.86&1.57&0.95&0.09&0.10\\& Proposed T&1.70&0.48&0.28&0.06&0.07\\& Robust B&0.61&0.52&0.23&0.07&0.08&0.24\\& Robust T&0.51&0.26&0.13&0.00&0.02&0.24\\
$(50,20,30,0.8,0.8)$& Sample &37.44&85.72&38.24\\& Proposed B&14.57&28.31&12.54&0.16&0.11\\& Proposed T&12.64&23.40&10.72&0.19&0.17\\& Robust B&1.63&2.76&1.32&0.24&0.21&0.15\\& Robust T&1.40&2.22&1.12&0.27&0.26&0.17\\
  \hline
\end{tabular}}
\end{lrbox}
\scalebox{0.75}{\usebox{\tableboxb}}
\end{table}
\begin{table}[htp]
\setlength{\tabcolsep}{-1pt}
  \centering
    \caption{Standard errors for Table \ref{Table:simulation1}}
\vspace{0.5pt}
\begin{lrbox}{\tableboxb}
\begin{tabular}{p{0.24\textwidth}>{\centering}p{0.23\textwidth}p{0.15\textwidth}p{0.15\textwidth}p{0.15\textwidth}p{0.08\textwidth}p{0.08\textwidth}p{0.08\textwidth}}
  \hline
$ (n, p, q, \rho_1,\rho_2) $ & Method & $\|\hat{\boldsymbol\Sigma}-\boldsymbol\Sigma\|_{\F}$ & $ \|\hat{\boldsymbol\Sigma}-\boldsymbol\Sigma\|_1 $ & $ \|\hat{\boldsymbol\Sigma}-\boldsymbol\Sigma\|_{2} $ & $\hat{k}$ \quad & $\hat{k}_1$ \quad & $\hat{k}_2$ \quad \\ \hline
$(50,20,30,0.5,0.5)$ & Sample &0.12&0.38&0.08\\& Banded &0.03&0.04&0.02&0.05&&\\& Tapering &0.02&0.04&0.02&0.04&&\\& Doubly B &0.03&0.06&0.03&&0&0\\& Doubly T &0.24&0.13&0.04&&0.08&0.09\\& Proposed B &0.05&0.03&0.02&&0.08&0.07\\& Proposed T &0.04&0.03&0.02&&0&0\\
$(100,20,30,0.5,0.5)$ & Sample &0.06&0.18&0.04\\& Banded &0.00&0.01&0.01&0.01&&\\& Tapering &0.01&0.02&0.01&0.02&&\\& Doubly B &0.03&0.03&0.02&&0.01&0\\& Doubly T &0.02&0.03&0.02&&0&0\\& Proposed B &0.04&0.02&0.01&&0.07&0.07\\& Proposed T &0.02&0.01&0.01&&0&0\\
$(200,20,30,0.5,0.5)$ & Sample &0.03&0.11&0.02\\& Banded &0.00&0.01&0.00&0.02&&\\& Tapering &0.00&0.01&0.00&0.02&&\\& Doubly B &0.01&0.02&0.01&&0&0\\& Doubly T &0.02&0.02&0.01&&0&0\\& Proposed B &0.03&0.01&0.01&&0.08&0.07\\& Proposed T &0.02&0.01&0.01&&0&0\\
 %\hline Max SE & &  0.05 & 0.04 & 0.02 & 0.05 & 0.08 & 0.07 \\
  \hline
\end{tabular}
\end{lrbox}
\label{Table:simulation1SE}
\scalebox{0.75}{\usebox{\tableboxb}}
\end{table}

\begin{table}[htp]
\setlength{\tabcolsep}{-1pt}
  \centering
    \caption{Standard errors for Table \ref{Table:simulation2:01}}
\vspace{0.5pt}
\begin{lrbox}{\tableboxb}
\begin{tabular}{p{0.24\textwidth}>{\centering}p{0.23\textwidth}p{0.15\textwidth}p{0.15\textwidth}p{0.15\textwidth}p{0.08\textwidth}p{0.08\textwidth}p{0.08\textwidth}}
  \hline
$ (n, p, q, \rho_1,\rho_2) $ & Method & $\|\hat{\boldsymbol\Sigma}-\boldsymbol\Sigma\|_{\F}$ & $ \|\hat{\boldsymbol\Sigma}-\boldsymbol\Sigma\|_1 $ & $ \|\hat{\boldsymbol\Sigma}-\boldsymbol\Sigma\|_{2} $ & $\hat{k}$ \quad & $\hat{k}_1$ \quad & $\hat{k}_2$ \quad \\ \hline
$(50,20,30,0.1,0.1)$ & Sample &0.08&0.35&0.05\\& Banded &0.03&0.02&0.01&0.01&&\\& Tapering &0.04&0.02&0.01&0.02&&\\& Doubly B &0.05&0.02&0.02&&0.04&0.04\\& Doubly T &0.07&0.03&0.02&&0.03&0.06\\& Proposed B &0.07&0.01&0.01&&0.06&0.06\\& Proposed T &0.07&0.01&0.01&&0.07&0.08\\
$(100,20,30,0.1,0.1)$ & Sample &0.04&0.17&0.03\\& Banded &0.03&0.01&0.01&0.03&&\\& Tapering &0.03&0.01&0.01&0.03&&\\& Doubly B &0.03&0.02&0.01&&0.06&0.05\\& Doubly T &0.04&0.02&0.01&&0.07&0.05\\& Proposed B &0.07&0.01&0.01&&0.06&0.05\\& Proposed T &0.07&0.01&0.01&&0.05&0.05\\
$(200,20,30,0.1,0.1)$& Sample &0.02&0.09&0.01\\& Banded &0.01&0.01&0.00&0.02&&\\& Tapering &0.00&0.01&0.00&0.04&&\\& Doubly B &0.02&0.01&0.00&&0.05&0.05\\& Doubly T &0.00&0.01&0.00&&0.08&0.07\\& Proposed B &0.03&0.00&0.00&&0.08&0.08\\& Proposed T &0.02&0.00&0.00&&0.00&0.02\\
\hline
\end{tabular}
\end{lrbox}
\scalebox{0.75}{\usebox{\tableboxb}}
\end{table}

\begin{table}[htp]
\setlength{\tabcolsep}{-1pt}
  \centering
    \caption{Standard errors for Table \ref{Table:simulation2}}
\vspace{0.5pt}
\begin{lrbox}{\tableboxb}
\begin{tabular}{p{0.24\textwidth}>{\centering}p{0.23\textwidth}p{0.15\textwidth}p{0.15\textwidth}p{0.15\textwidth}p{0.08\textwidth}p{0.08\textwidth}p{0.08\textwidth}}
  \hline
$ (n, p, q, \rho_1,\rho_2) $ & Method & $\|\hat{\boldsymbol\Sigma}-\boldsymbol\Sigma\|_{\F}$ & $ \|\hat{\boldsymbol\Sigma}-\boldsymbol\Sigma\|_1 $ & $ \|\hat{\boldsymbol\Sigma}-\boldsymbol\Sigma\|_{2} $ & $\hat{k}$ \quad & $\hat{k}_1$ \quad & $\hat{k}_2$ \quad \\ \hline
$(50,20,30,0.5,0.5)$ & Sample &0.13&0.46&0.12\\& Banded &0.02&0.04&0.01&0.06&&\\& Tapering &0.01&0.04&0.01&0.09&&\\& Doubly B &0.08&0.11&0.03&&0.06&0.07\\& Doubly T &0.1&0.11&0.04&&0.03&0.05\\& Proposed B &0.07&0.06&0.03&&0.04&0.03\\& Proposed T &0.03&0.05&0.02&&0&0\\
$(100,20,30,0.5,0.5)$ & Sample &0.06&0.21&0.06\\& Banded &0.01&0.03&0.01&0.08&&\\& Tapering &0.01&0.02&0.01&0.11&&\\& Doubly B &0.09&0.06&0.04&&0.05&0.05\\& Doubly T &0.03&0.06&0.02&&0.04&0.02\\& Proposed B &0.07&0.04&0.03&&0.09&0.10\\& Proposed T &0.07&0.04&0.03&&0.03&0.05\\
$(200,20,30,0.5,0.5)$ & Sample &0.03&0.12&0.04\\& Banded &0.02&0.04&0.02&0.05&&\\& Tapering &0.02&0.03&0.02&0.15&&\\& Doubly B &0.01&0.04&0.02&&0.03&0.04\\& Doubly T &0.01&0.04&0.02&&0.05&0.06\\& Proposed B &0.04&0.03&0.02&&0.07&0.08\\& Proposed T &0.03&0.02&0.02&&0.02&0\\
\hline
\end{tabular}
\end{lrbox}
\scalebox{0.75}{\usebox{\tableboxb}}
\end{table}
\begin{table}[htp]
\setlength{\tabcolsep}{-1pt}
  \centering
    \caption{Standard errors for Table \ref{addsim1}}
\vspace{0.5pt}
\begin{lrbox}{\tableboxb}
\begin{tabular}{p{0.24\textwidth}>{\centering}p{0.23\textwidth}p{0.15\textwidth}p{0.15\textwidth}p{0.15\textwidth}p{0.08\textwidth}p{0.08\textwidth}p{0.08\textwidth}}
  \hline
$ (n, p, q, \rho_1,\rho_2) $ & Method & $\|\hat{\boldsymbol\Sigma}-\boldsymbol\Sigma\|_{\F}$ & $ \|\hat{\boldsymbol\Sigma}-\boldsymbol\Sigma\|_1 $ & $ \|\hat{\boldsymbol\Sigma}-\boldsymbol\Sigma\|_{2} $ & $\hat{k}$ \quad & $\hat{k}_1$ \quad & $\hat{k}_2$ \quad \\ \hline
$(50,20,30,0.8,0.8)$ & Sample &0.39&1.18&0.64\\& Banded &0.19&0.53&0.30&0.22&&\\& Tapering &0.20&0.34&0.30&0.47&&\\& Doubly B &0.3&0.66&0.39&&0.12&0.15\\& Doubly T &0.29&0.56&0.38&&0.17&0.2\\& Proposed B &0.40&0.80&0.39&&0.12&0.11\\& Proposed T &0.34&0.54&0.35&&0.08&0.05\\
$(100,20,30,0.8,0.8)$& Sample &0.23&0.69&0.43\\& Banded &0.14&0.35&0.27&0.64&&\\& Tapering &0.24&0.26&0.31&0.98&&\\& Doubly B &0.25&0.39&0.35&&0.16&0.14\\& Doubly T &0.27&0.35&0.35&&0.18&0.17\\& Proposed B &0.26&0.54&0.29&&0.16&0.15\\& Proposed T &0.27&0.39&0.30&&0.08&0.07\\
$(200,20,30,0.8,0.8)$ & Sample &0.15&0.45&0.28\\& Banded &0.14&0.36&0.21&0.68&&\\& Tapering &0.16&0.30&0.23&0.83&&\\& Doubly B &0.22&0.38&0.31&&0.14&0.12\\& Doubly T &0.2&0.34&0.3&&0.15&0.16\\& Proposed B &0.20&0.43&0.22&&0.15&0.14\\& Proposed T &0.23&0.36&0.24&&0.09&0.10\\
%\hline Max SE & &  0.40 & 0.80 & 0.39 & 0.98 & 0.16 & 0.15\\ 
\hline
\end{tabular}
\end{lrbox}
\scalebox{0.75}{\usebox{\tableboxb}}
\end{table}

\begin{table}[htp]
\setlength{\tabcolsep}{-1pt}
  \centering
    \caption{Standard errors for Table \ref{addsims4}}
\vspace{0.5pt}
\begin{lrbox}{\tableboxb}
\begin{tabular}{p{0.24\textwidth}>{\centering}p{0.23\textwidth}p{0.15\textwidth}p{0.15\textwidth}p{0.15\textwidth}p{0.08\textwidth}p{0.08\textwidth}p{0.08\textwidth}}
  \hline
$ (n, p, q, \rho_1,\rho_2) $ & Method & $\|\hat{\boldsymbol\Sigma}-\boldsymbol\Sigma\|_{\F}$ & $ \|\hat{\boldsymbol\Sigma}-\boldsymbol\Sigma\|_1 $ & $ \|\hat{\boldsymbol\Sigma}-\boldsymbol\Sigma\|_{2} $ & $\hat{k}$ \quad & $\hat{k}_1$ \quad & $\hat{k}_2$ \quad \\ \hline
$(50,100,100,0.8,0.8)$& Sample &1.46&5.43&1.49\\& Banded &0.05&0.12&0.02&0.10&&\\& Tapering &0.05&0.13&0.03&0.16&&\\& Double B &0.69&0.97&0.19&&0.08&0.08\\& Double T &0.52&0.93&0.28&&0.11&0.11\\& Proposed B &0.68&0.71&0.30&&0.36&0.33\\& Proposed T &0.57&0.51&0.24&&0.23&0.23\\
$(100,100,100,0.8,0.8)$& Sample &0.61&2.85&0.65\\& Banded &0.03&0.08&0.02&0.13&&\\& Tapering &0.02&0.08&0.03&0.23&&\\& Double B &0.22&0.54&0.21&&0.10&0.08\\& Double T &0.28&0.39&0.22&&0.11&0.11\\& Proposed B &0.47&0.57&0.23&&0.52&0.54\\& Proposed T &0.39&0.37&0.20&&0.47&0.50\\
$(200,100,100,0.8,0.8)$& Sample &0.33&1.46&0.37\\& Banded &0.01&0.04&0.02&0.09&&\\& Tapering &0.00&0.04&0.01&0.07&&\\& Double B &0.27&0.40&0.19&&0.08&0.08\\& Double T &0.3&0.30&0.20&&0.13&0.15\\& Proposed B &0.21&0.27&0.12&&0.12&0.12\\& Proposed T &0.24&0.22&0.13&&0.12&0.12\\
%\hline Max SE & &  0.68 & 0.71 & 0.30 & 0.23 & 0.52 & 0.54\\ 
\hline
\end{tabular}\label{addsim2SE}
\end{lrbox}
\scalebox{0.75}{\usebox{\tableboxb}}
\end{table}

\begin{table}[htp]
\setlength{\tabcolsep}{-1pt}
  \centering
    \caption{Standard errors for Table \ref{tb:s5simu}}
\vspace{0.5pt}
\begin{lrbox}{\tableboxb}\label{tb:rb_1}
\Rotatebox{0}{

\begin{tabular}{p{0.24\textwidth}>{\centering}p{0.23\textwidth}p{0.15\textwidth}p{0.15\textwidth}p{0.15\textwidth}p{0.08\textwidth}p{0.08\textwidth}p{0.08\textwidth}}
   \hline
$(n,p,q,\rho_1,\rho_2)$ & Method & $\| \hat{\Sigma} - \Sigma\|_{\F}$ & $\|\hat{\Sigma} - \Sigma\|_1$ & $\| \hat{\Sigma} - \Sigma \|_{2}$ & $\hat{k}_1$ & $\hat{k}_2$ & $\hat{\tau}$
\\
\hline
$(50,20,30,0.5,0.5)$
& Sample &36.19&92.56&36.91\\& Proposed B&2.61&1.15&0.76&0.08&0.09\\& Proposed T&1.68&0.35&0.25&0.06&0.05\\& Robust B&1.05&0.48&0.27&0.08&0.08&0.31\\& Robust T&0.50&0.16&0.10&0.02&0.00&0.24\\
$(100,20,30,0.5,0.5)$
& Sample &26.87&67.07&27.43\\& Proposed B&2.07&1.61&0.93&0.09&0.08\\& Proposed T&1.45&0.59&0.40&0.06&0.03\\& Robust B&0.72&0.35&0.21&0.09&0.09&0.40\\& Robust T&0.37&0.16&0.10&0.00&0.00&0.30\\
$(200,20,30,0.5,0.5)$& Sample &11.96&34.16&12.31\\& Proposed B&0.85&0.38&0.22&0.09&0.08\\& Proposed T&0.77&0.22&0.14&0.04&0.05\\& Robust B&0.61&0.27&0.16&0.08&0.08&0.58\\& Robust T&0.49&0.17&0.10&0.00&0.00&0.44\\  \hline
\end{tabular}}
\end{lrbox}
\scalebox{0.75}{\usebox{\tableboxb}}
\label{Table:simulation5SE}
\end{table}
\newpage
\subsection{Figures for Gridded Temperature Anomaly Data Analysis}
We present additional figures for Section \ref{realdata} of the main paper. In Figure \ref{climate_matrix}, we mark the regions that our data matrices are obtained from in deep blue. In Figure \ref{climate_trend_acf}, we use $5^\degree\times 5^\degree$ box centered at $57.5^\degree W$ longitude and $7.5^\degree S$ latitude as an example to show the effect of pre-processing. In Figure \ref{climate_trend_acf} (a) and (b), we show the data before and after the detrending; and in Figure \ref{climate_trend_acf} (c), we plot the estimated auto-correlation function for the thinned sequence. In Figure \ref{climate_trend_acf} (d), we show the quantile-quantile (Q-Q) plot of all temperature anomalies. It can be seen that both detrending and thinning work quite well for that region. Similar results were also obtained for other spatial regions in our dataset. In Figure \ref{hist_climate}, we compare the histograms of regularized and unregularized entries along latitude and longitude directions, respectively. In Figure \ref{hist_climate} (a), we compare the histograms of regularized entries in $\hat{\M \Sigma}_{1}^{\mathcal{R},\MB}(15)$, unregularized entries in $\hat{\M \Sigma}_{1}^{\mathcal{R},\MB}(15)$, and unregularized entries in $\hat{\M \Sigma}_{1}^{\mathcal{R},\MB}(\hat{k}_{1}^{\mathcal{B}})$ over the latitude direction. In Figure \ref{hist_climate} (b), we compare the  histograms of regularized entries in $\hat{\M \Sigma}_{2}^{\mathcal{R},\MB}(68)$, unregularized entries in $\hat{\M \Sigma}_{2}^{\mathcal{R},\MB}(68)$, and unregularized entries in $\hat{\M \Sigma}_{2}^{\mathcal{R},\MB}(\hat{k}_{2}^{\mathcal{B}})$ over the longitude direction. Figure \ref{climate_visi} visualizes the covariance between each lat-lon box and the lat-lon box centered at $107.5^\circ W$ longitude and $22.5^\circ N$ latitude. In particular, Figure \ref{climate_visi} (a) plots the corresponding covariance estimation of the proposed banded estimator. Figure \ref{climate_visi} (b) plots the corresponding covariance estimation of the proposed tapering estimator.
\begin{figure}
  \centering
\includegraphics[width=0.8\textwidth]{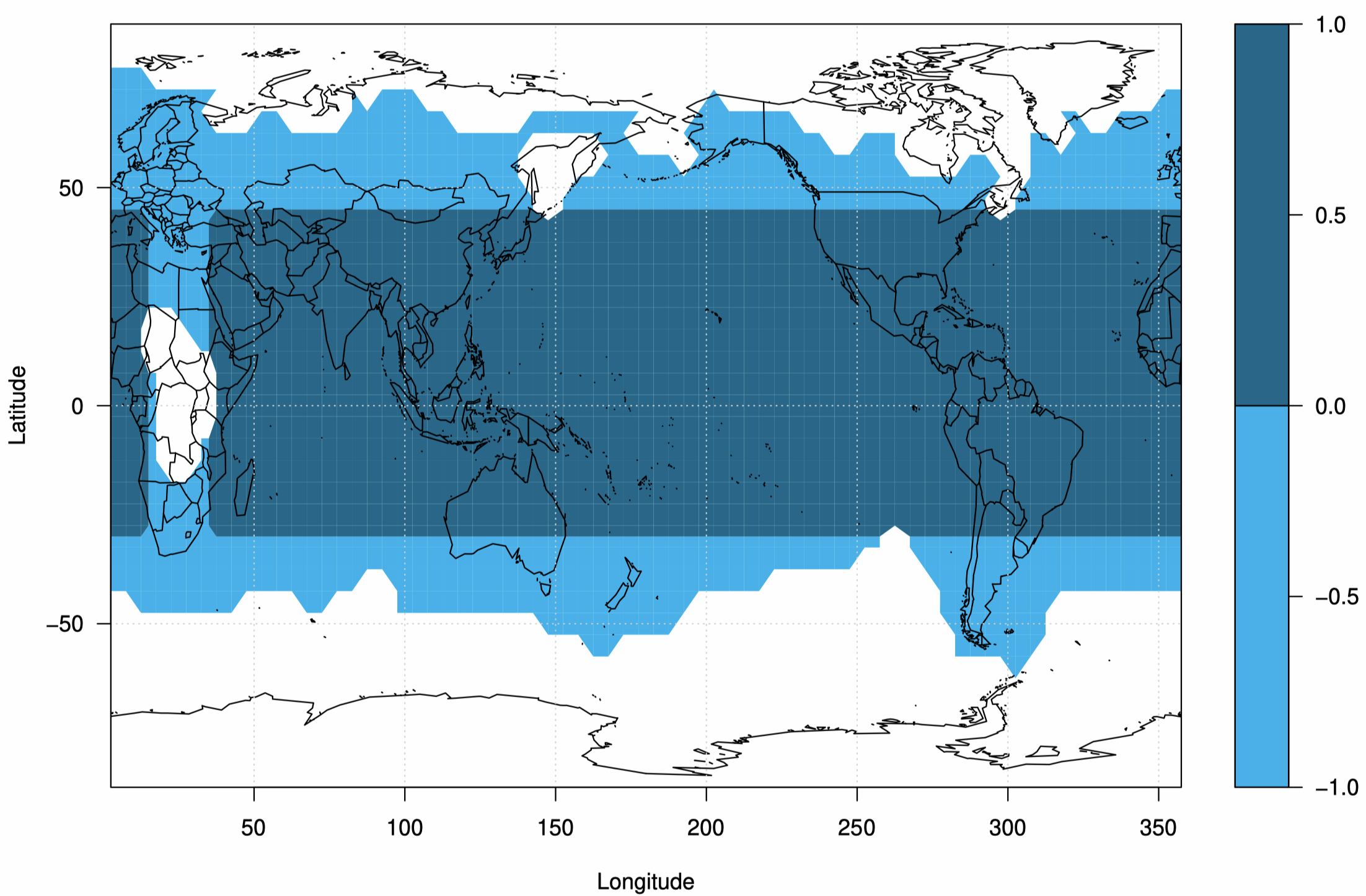}
\caption{Temperature data analysis: The whole blue region in the plot shows all non-missing spatial coordinates in \citet{gu2020generalized}'s dataset. The deep blue region shows the spatial coordinates we extract from all non-missing coordinates to form our matrix type dataset.}
\label{climate_matrix}
\end{figure}
\begin{figure}
   \centering
    \subfigure[]{\includegraphics[width=0.4\textwidth]{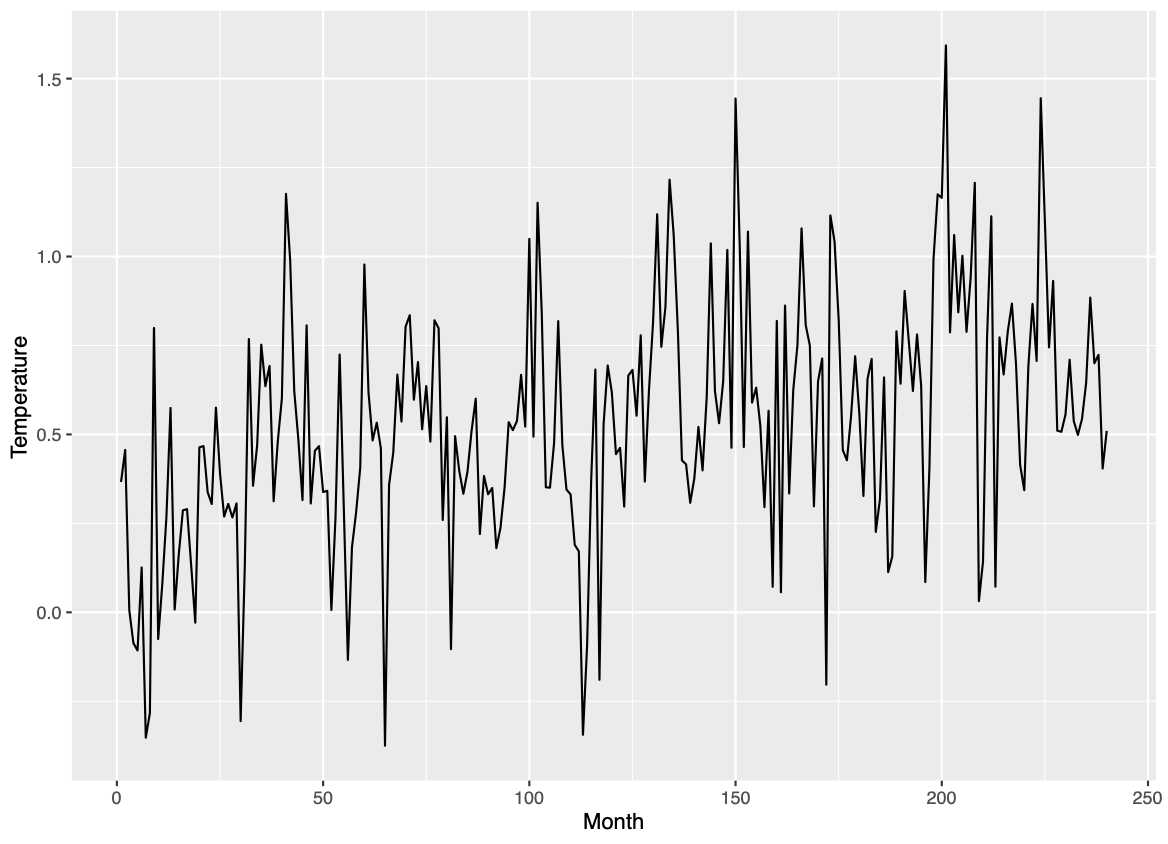}}
  \subfigure[]{\includegraphics[width=0.4\textwidth]{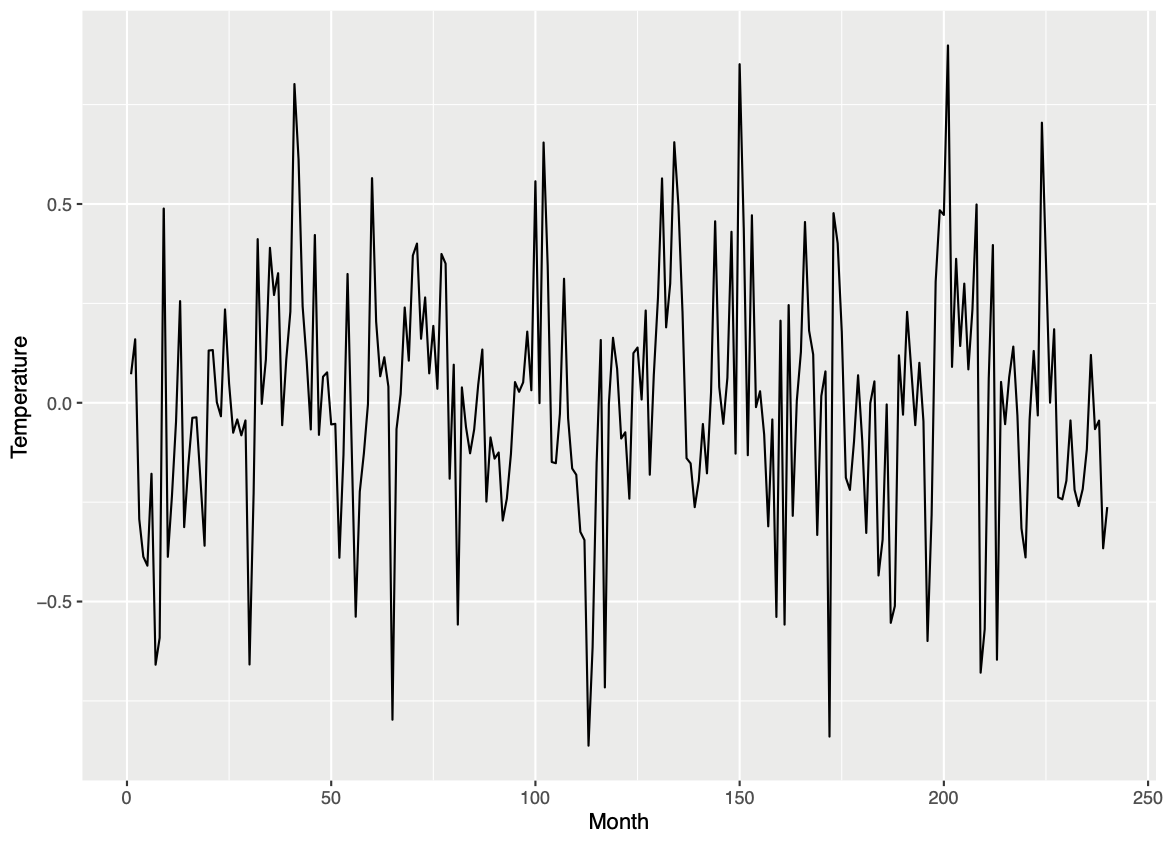}}
    \subfigure[]{\includegraphics[width=0.4\textwidth]{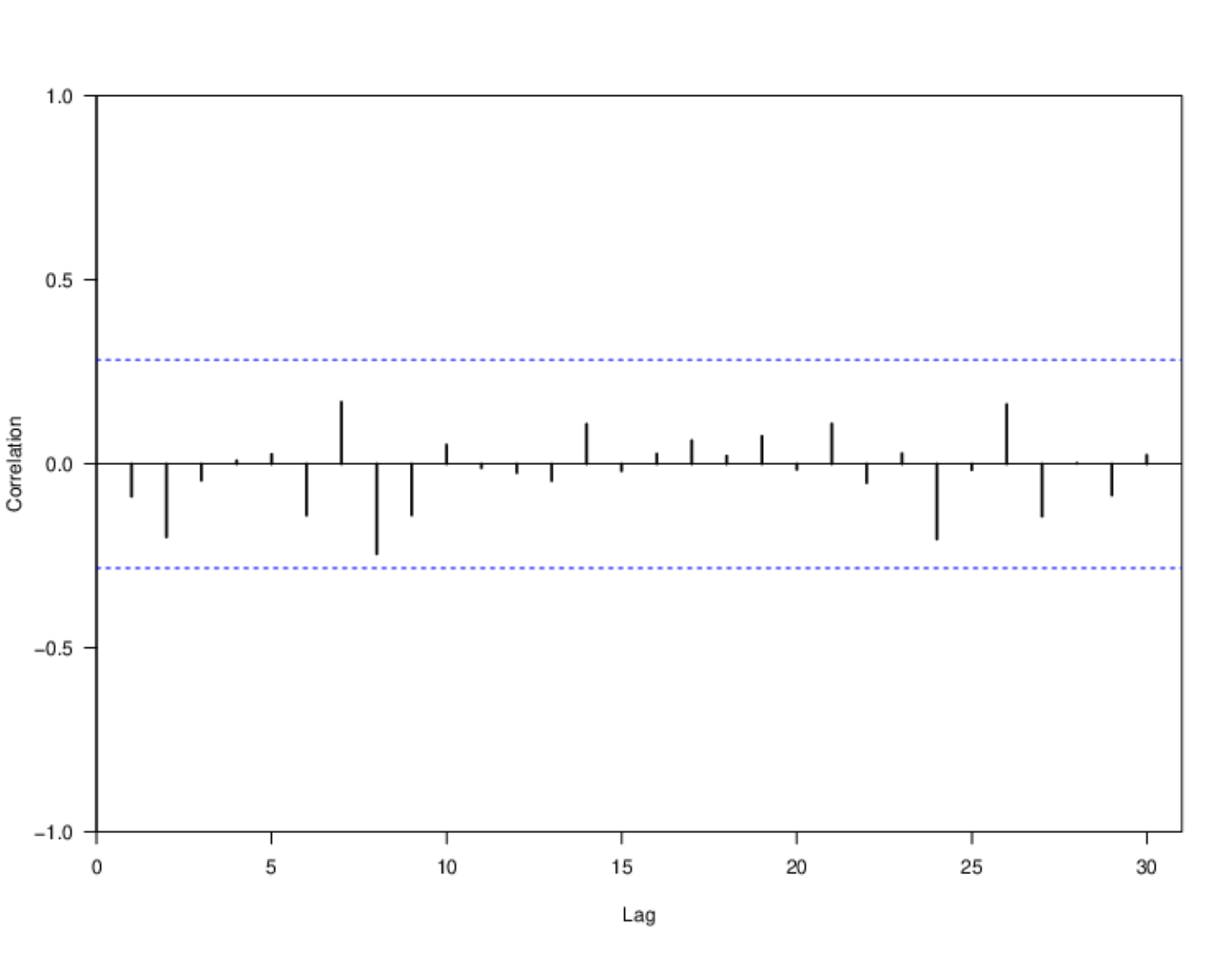}}
     \subfigure[]{ \includegraphics[width=0.4\textwidth]{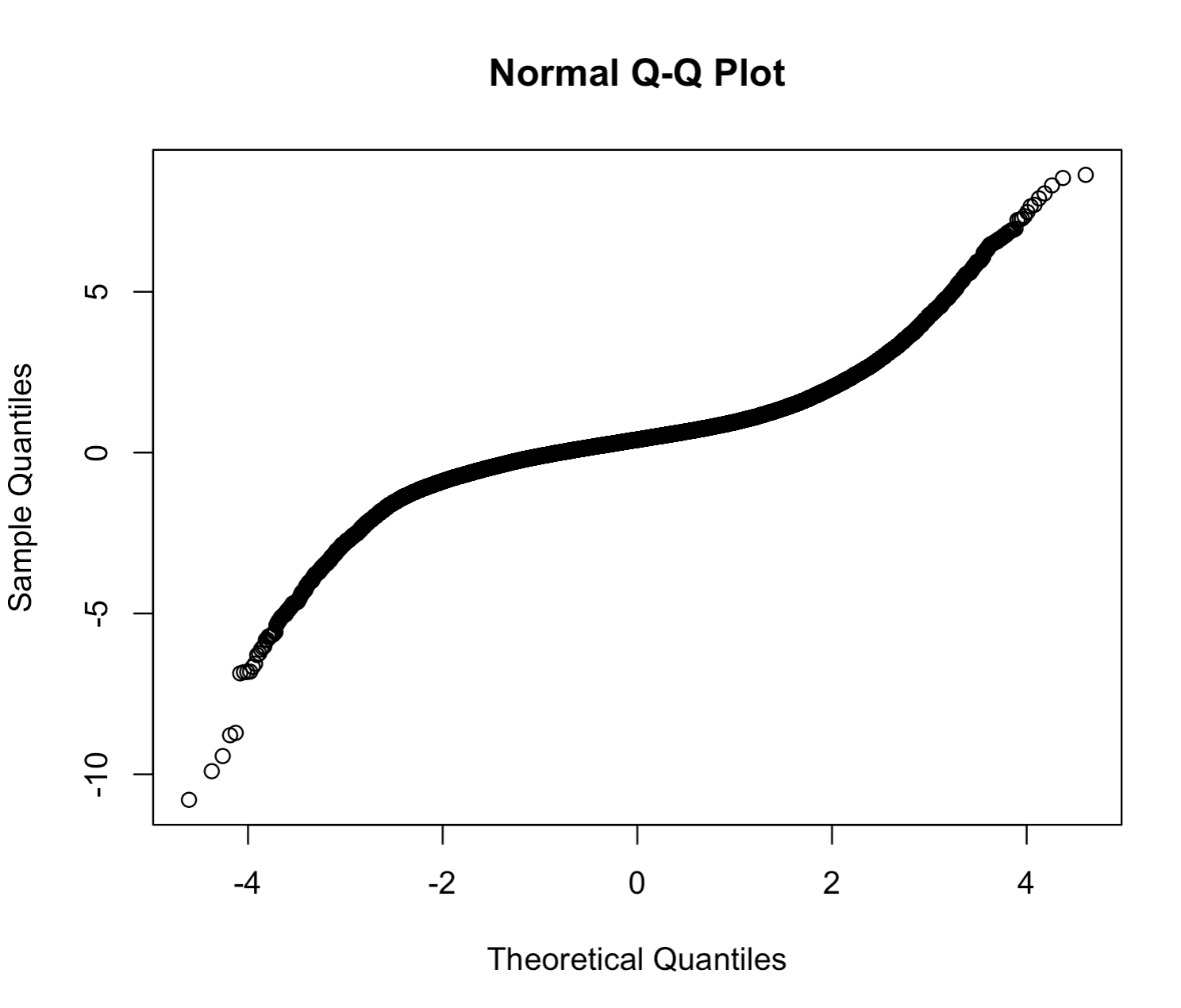}}
\caption{Temperature data analysis - (a) and (b): Monthly temperature anomalies before and after detrending; (c): Auto-correlation function for the thinned sequence (Benchmark: white noise confidence bands in blue); (d): The Q-Q plot of all sampled temperature anomalies in all $5^{\circ} \times 5^{\circ}$ latitude-longitude boxes.}
\label{climate_trend_acf}
\end{figure}
\begin{figure}
  \centering
    \subfigure[Latitude Direction]{\includegraphics[width=0.49\textwidth]{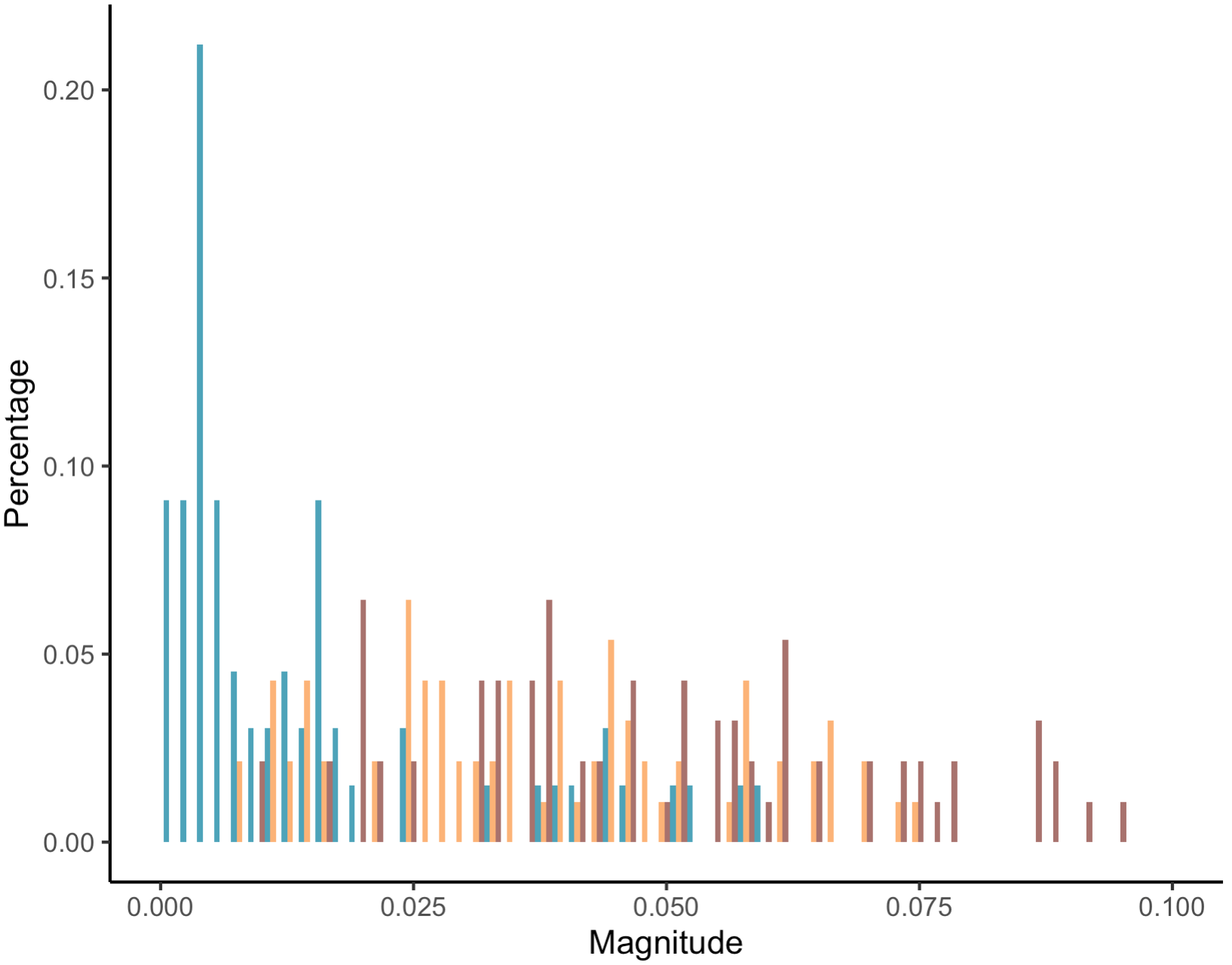}}
  \subfigure[Longitude Direction]{\includegraphics[width=0.49\textwidth]{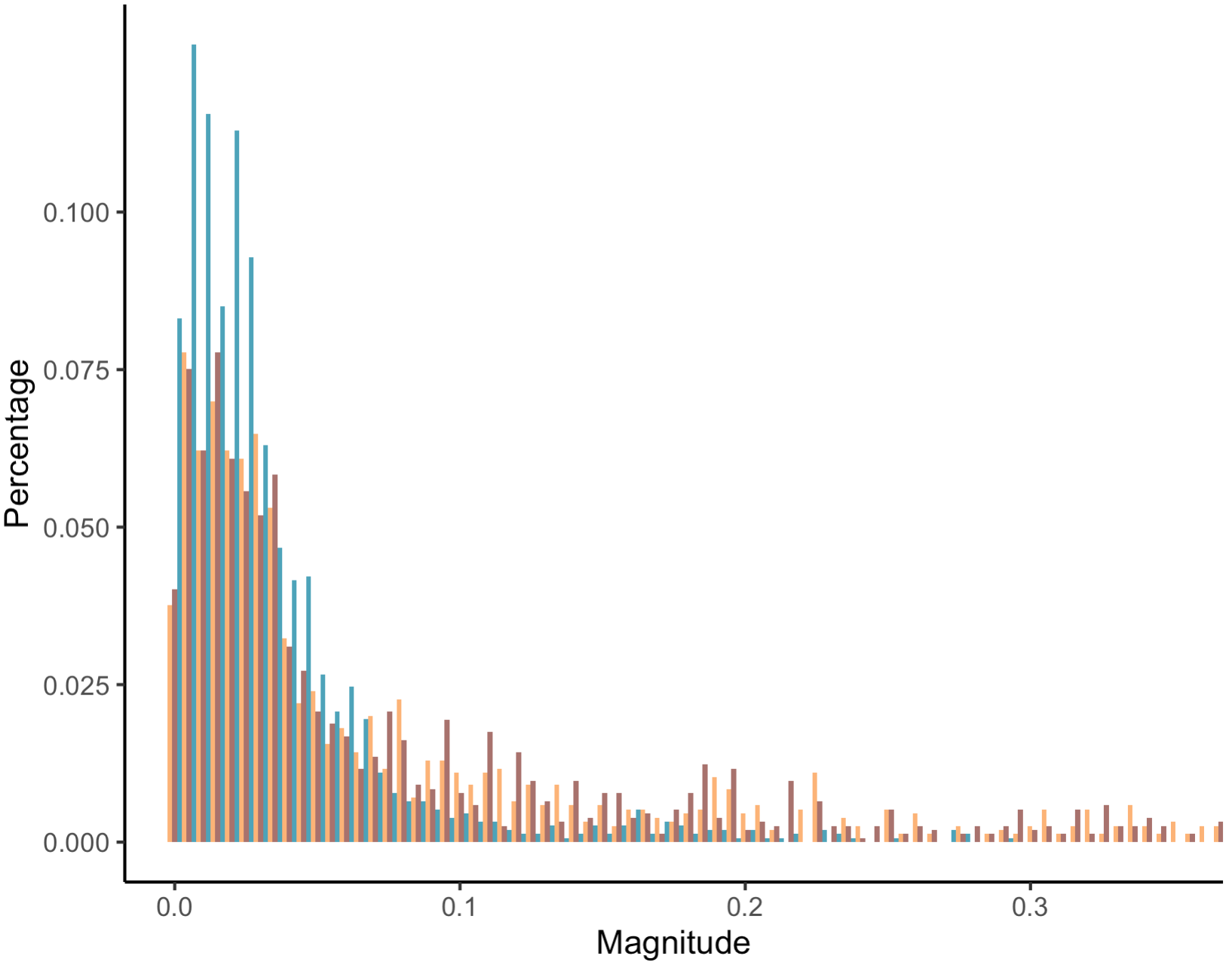}}
\caption{Histograms of regularized and unregularized entries in the covariance estimation along latitude and longitude directions; Panel (a) plots the histograms of regularized entries in $\hat{\M \Sigma}_{1,\text{tp}}^{\mathcal{R},\MB}(15)$ (blue), unregularized entries in $\hat{\M \Sigma}_{1,\text{tp}}^{\mathcal{R},\MB}(15)$ (orange), and unregularized entries in $\hat{\M \Sigma}_{1,\text{tp}}^{\mathcal{R},\MB}(\hat{k}_{1,\text{tp}}^{\mathcal{R},\mathcal{B}})$ (brown); Panel (b) plots the histograms of regularized entries in $\hat{\M \Sigma}_{2,\text{tp}}^{\mathcal{R},\MB}(68)$ (blue), unregularized entries in $\hat{\M \Sigma}_{2,\text{tp}}^{\mathcal{R},\MB}(68)$ (orange), and unregularized entries in $\hat{\M \Sigma}_{2,\text{tp}}^{\mathcal{R},\MB}(\hat{k}_{2,\text{tp}}^{\mathcal{R},\mathcal{B}})$ (brown); Magnitude domains of panel (a) and (b) contain the maximal magnitudes of all regularized entries in $\hat{\M \Sigma}_{1,\text{tp}}^{\mathcal{R},\MB}(15)$ and $\hat{\M \Sigma}_{2,\text{tp}}^{\mathcal{R},\MB}(68)$ respectively, while some unregularized entries have magnitudes larger than the domains.}
\label{hist_climate}
\end{figure}
\begin{figure}
   \centering
    \subfigure[]{\includegraphics[width=0.49\textwidth]{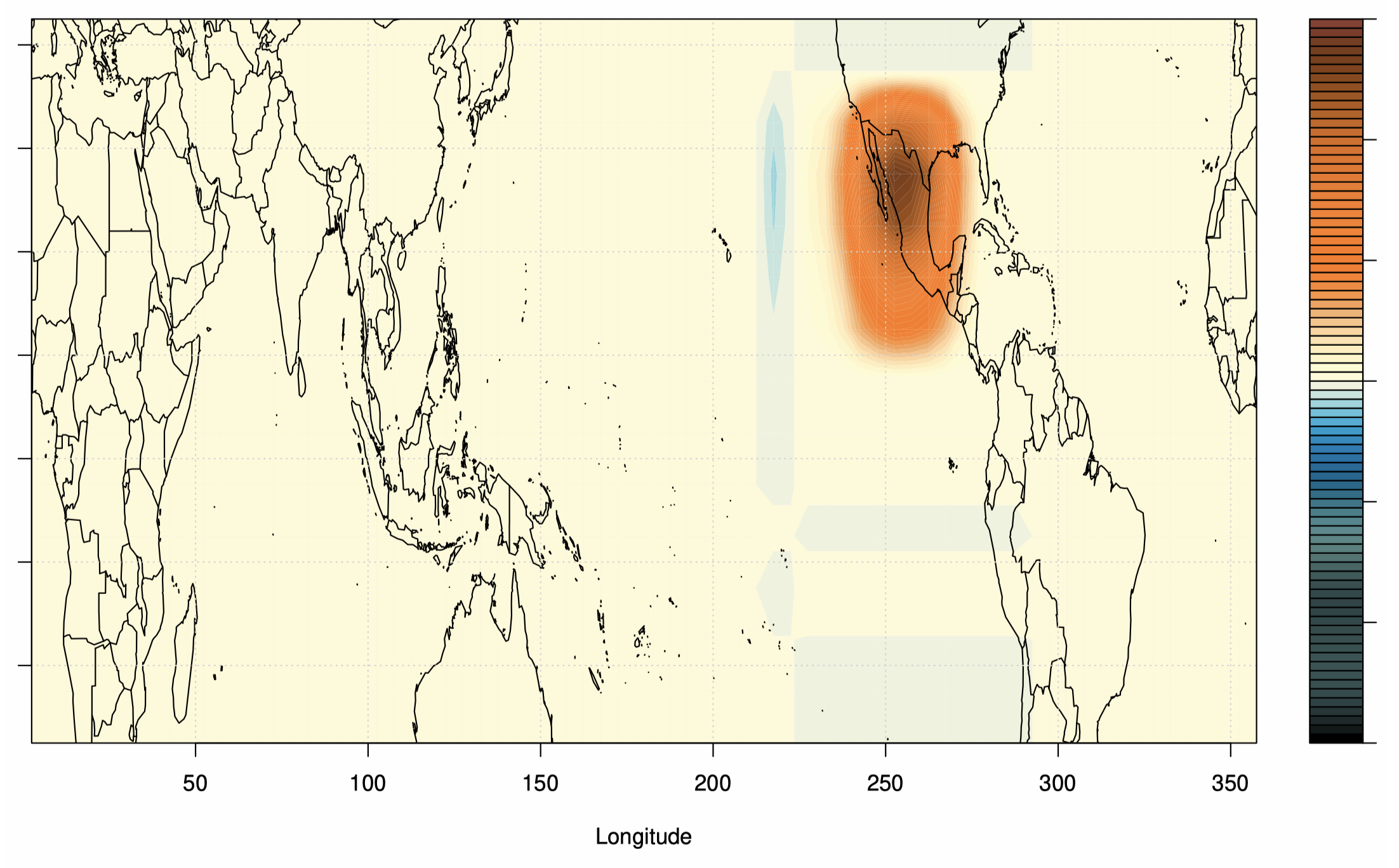}}
  \subfigure[]{\includegraphics[width=0.49\textwidth]{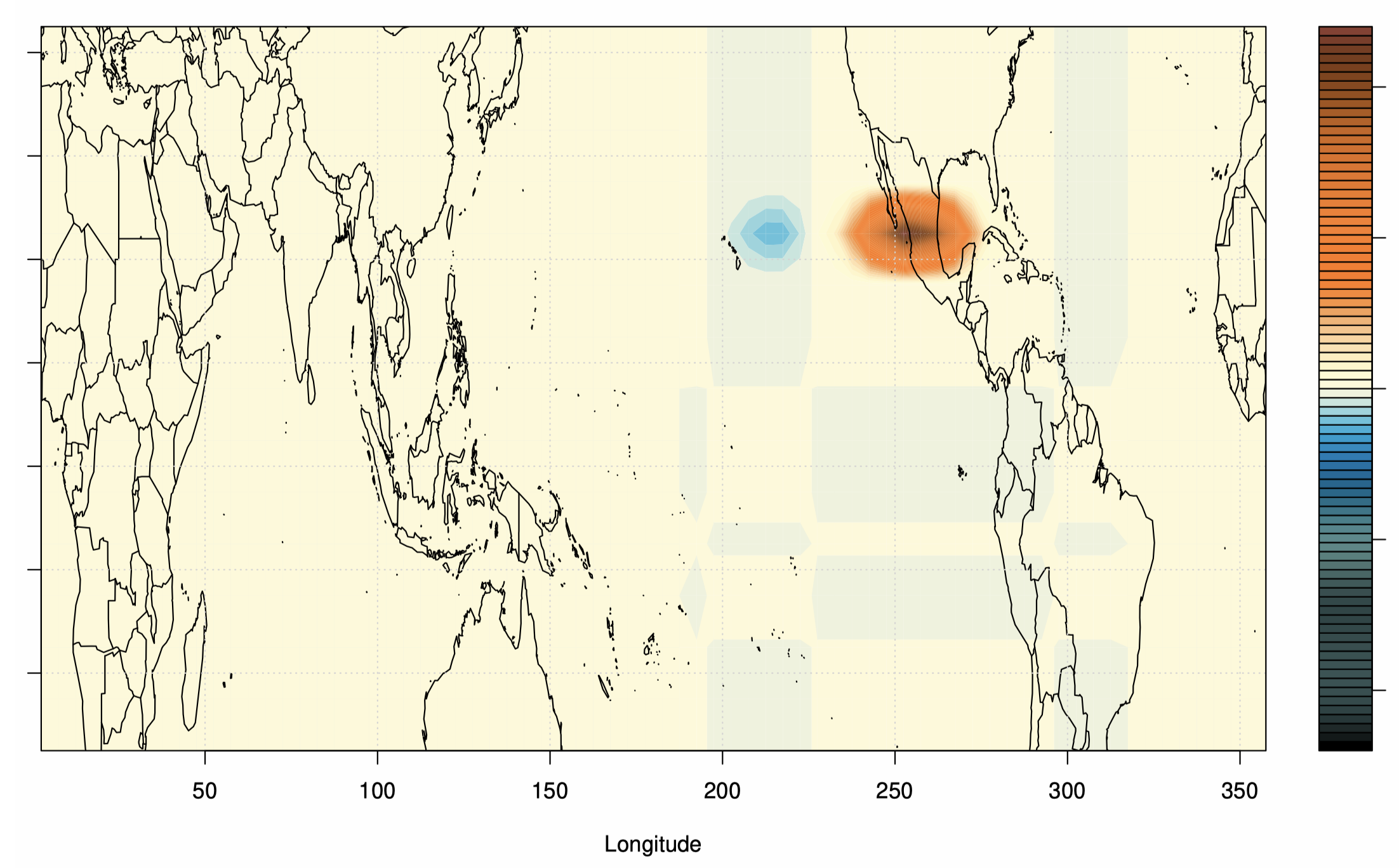}}
\caption{Temperature data analysis: Panel (a) plots the proposed banded covariance estimation of each lat-lon box with the lat-lon box centered at $107.5^\degree W$ longitude and $22.5^\degree N$ latitude. Panel (b) plots the proposed tapering covariance estimation of each lat-lon box with the lat-lon box centered at $107.5^\degree W$ longitude and $22.5^\degree N$ latitude.}
\label{climate_visi}
\end{figure}

\section{Proof of Propositions}\label{sec:propo}
\subsection{Proof of Proposition \ref{band:equivalence}}
\label{sec:prop}
Define 
\begin{eqnarray*}
\hat{\boldsymbol\Sigma}=\left(
      \begin{array}{ccc}
    \hat{\boldsymbol\Sigma}_{1,1}   & \cdots &\hat{\boldsymbol\Sigma}_{1,p}  \\
    \vdots  & & \vdots \\
   \hat{\boldsymbol\Sigma}_{p,1}   & \cdots & \hat{\boldsymbol\Sigma}_{p,p}  \\
      \end{array}
    \right),
\end{eqnarray*} 
where $ \hat{\boldsymbol\Sigma}_{l_2,m_2}\in\RR^{q\times q} $.

One can easily see that solving \eqref{eq:band:0} is equivalent to solving 
\bee\label{prop1:1}
\argmin_{\boldsymbol\Sigma_1\in \mathcal{B}_p(k_1), b_{l_2,m_2}} \sum_{l_2=1}^p \sum_{m_2=1}^p \|\hat{\boldsymbol\Sigma}_{l_2,m_2}-b_{l_2,m_2}\boldsymbol\Sigma_1\|_\F^2
\ee
subject to $b_{l_2,m_2}=0$ for all $ |l_2-m_2|>k_2$ and $ b_{l_2,m_2} $ is the $l_2m_2$th element of $ \boldsymbol\Sigma_2 $ for $ 1\leq l_2, m_2\leq p $. Since all summand terms in (\ref{prop1:1}) with $ |l_2-m_2|>k_2$ are constants $\|\hat{\M \Sigma}_{l_2,m_2}\|_\F$ due to $b_{l_2,m_2}=0$ for all $ |l_2-m_2|>k_2$, solving (\ref{prop1:1}) is equivalent to solving 
\begin{eqnarray*}
\argmin_{\boldsymbol\Sigma_1\in \mathcal{B}_p(k_1), b_{l_2,m_2}} \sum_{|l_2-m_2|\leq k_2} \|\hat{\boldsymbol\Sigma}_{l_2,m_2}-b_{l_2,m_2}\boldsymbol\Sigma_1\|_\F^2
\end{eqnarray*} 
subject to $b_{l_2,m_2}=0$ for all $ |l_2-m_2|>k_2$. This is also equivalent to solving 
\begin{eqnarray*}
\argmin_{\boldsymbol\Sigma_1\in \mathcal{B}_p(k_1), b_{l_2,m_2}} \sum_{|l_2-m_2|\leq k_2} \|\hat{\boldsymbol\Sigma}_{l_2,m_2}\circ B_{k_1}({\bf 1}_{p})-b_{l_2,m_2}\boldsymbol\Sigma_1\|_\F^2
\end{eqnarray*} 
subject to $b_{l_2,m_2}=0$ for all $ |l_2-m_2|>k_2$ due to $\M \Sigma_1$ is a $k_1$-bandable matrix. By some algebra, this is equivalent to solving 
\begin{eqnarray*}
\argmin_{\boldsymbol\Sigma_1\in \mathcal{B}_p(k_1), \boldsymbol\Sigma_2\in \mathcal{B}_q(k_2)}\| \widetilde{\boldsymbol\Sigma}_\MB(k_1, k_2)-\boldsymbol\Sigma_2 \otimes \boldsymbol\Sigma_1\|_\rF^2 . 
\end{eqnarray*} 
which implies that solving \eqref{eq:band:0} is equivalent to solving
\bee\label{propfinal}
\argmin_{\boldsymbol\Sigma_1, \boldsymbol\Sigma_2}\| \widetilde{\boldsymbol\Sigma}_\MB(k_1, k_2)-\boldsymbol\Sigma_2 \otimes \boldsymbol\Sigma_1\|_\rF^2 
\ee
because solutions of (\ref{propfinal}):  $(\M \Sigma_1,\M \Sigma_2)$ will automatically belong to  $(\mathcal{B}_p(k_1), \mathcal{B}_q(k_2))$ due to the doubly banded structure of $\widetilde{\boldsymbol\Sigma}_\MB(k_1, k_2)$. \qed
\subsection{Proof of Proposition \ref{po:boundsig}}
In the proof, we focus on the scenario that $\M\Sigma_1^* \in \mathcal{F}^*(\alpha)$. Similar arguments can be applied if $\M\Sigma_1^* \in \mathcal{M}^*(\alpha)$.
\par
By the finite fourth order moment of $\mathbf{x}$ and Cauchy-Schwarz inequality, there exists a fixed constant $C_u > 0$ such that all elements in the covariance matrix $\M\Sigma_1^*$ are uniformly upper bounded by $C_u$, i.e.,
\bee\label{sigma*max}
\|\M \Sigma_1^*\|_{\max} \leq C_u.
\ee
We can now decompose $\M\Sigma_{1}^{*}$ as
\bee\label{main:4:1}
\M\Sigma_{1}^{*} = \mathcal{P}_{\text{diag}}(\M\Sigma_{1}^{*}) + \big\{\M\Sigma_{1}^{*} - \mathcal{P}_{\text{diag}}(\M\Sigma_{1}^{*})\big\},
\ee
where $\mathcal{P}_{\text{diag}}(\cdot)$ is the projection onto the matrix subspace consisting of all diagonal matrices. We immediately have $\|\mathcal{P}_{\text{diag}}(\M\Sigma_{1}^{*})\|_2 = \|\mathcal{P}_{\text{diag}}(\M\Sigma_{1}^{*})\|_{\max}\leq \|\M\Sigma_{1}^{*}\|_{\max} \leq C_{{u}}$. In addition, by taking $k = 1$ in \eqref{A1*}, we have 
\bee\label{main:4:2}
\|\M\Sigma_{1}^{*} - \mathcal{P}_{\text{diag}}(\M\Sigma_{1}^{*})\|_{\infty} &= \max_{1\leq l_1\leq p}\sum_{m_1 = 1}^p |\sigma^{(1)}_{l_1,m_1}|
\\
&\leq 2\max_{l_1,m_1\atop|l_1 - m_1| = 1}|\sigma^{(1)}_{l_1,m_1}|+ \max_{1\leq l_1\leq p}\sum_{m_1 = 1}^p \{|\sigma^{(1)}_{l_1,m_1}|: |l_1-m_1|>1 \}
\\
&\leq 2C_{{u}} + C_0.
\ee 
By symmetry, we also have $\|\M\Sigma_{1}^{*} - \mathcal{P}_{\text{diag}}(\M\Sigma_{1}^{*})\|_{1}\leq 2C_{{u}} + C_0$. Then the matrix H{\"o}lder's inequality \citep{golub1996matrix} implies 
\bee\label{main:4:3}
\|\M\Sigma_{1}^{*} - \mathcal{P}_{\text{diag}}(\M\Sigma_{1}^{*})\|_{2}&\leq \sqrt{\|\M\Sigma_{1}^{*} - \mathcal{P}_{\text{diag}}(\M\Sigma_{1}^{*})\|_{1}\|\M\Sigma_{1}^{*} - \mathcal{P}_{\text{diag}}(\M\Sigma_{1}^{*})\|_{\infty}}
\\
&\leq 2C_{{u}} + C_0.
\ee
Combining \eqref{main:4:1}--\eqref{main:4:3}, we then show by triangle inequality that 
\bee\nonumber
\lambda_{\max}(\M\Sigma_{1}^{*}) &= \|\M\Sigma_{1}^{*}\|_2 
\\
&\leq \|\mathcal{P}_{\text{diag}}(\M\Sigma_{1}^{*})\|_2 + \|\M\Sigma_{1}^{*} - \mathcal{P}_{\text{diag}}(\M\Sigma_{1}^{*})\|_2
\\
&\precsim 1.
\ee\qed

\section{Lemmas and their Proofs}\label{sec:lmpf}
\begin{lemma}[Weyl's Theorem \citep{weyl1912asymptotische}]\label{lm:weyl}
For symmetric matrices $\M M,\M M' \in \RR^{d\times d}$, where $\M M$ has eigenvalues $\lambda_1(\M M)\geq \lambda_2(\M M)\geq \cdots \geq \lambda_{d}(\M M)$ and $\M M'$ has eigenvalues $\lambda_1(\M M')\geq \lambda_2(\M M')\geq \cdots \geq \lambda_{d}(\M M')$. Then we have
\bee\nonumber
\max_{1\leq l\leq d}|\lambda_l(\M M) - \lambda_l(\M M')|\leq \|\M M - \M M'\|_2
%\\
%&\text{(ii). } \sum_{l = 1}^d|\lambda_l(\M M) - \lambda_l(\M M')|^2\leq \|\M M - \M M'\|^2_\F
\ee
\end{lemma}
\begin{lemma}[Matrix H{\"o}lder's Inequality]\label{lm:matrixholder} For matrix $\M M\in\RR^{d_1\times d_2}$, we have
\bee\label{lm:matrixholder:1}
\|\M M\|_2 \leq \sqrt{\|\M M\|_{\infty}\|\M M\|_1}.
\ee
Specially, for symmetric matrix $\M M \in \RR^{d\times d}$ we have $\|\M M\|_2\leq \|\M M\|_1$.
\end{lemma}
\begin{proof}[Proof of Lemma \ref{lm:matrixholder}]For general form \eqref{lm:matrixholder:1}, see {Corollary 2.3.2} in \citet{golub1996matrix} for the proof. When $\M M$ is symmetric, by directly checking the definition, one can see $\|\M M\|_1 = \|\M M\|_{\infty}$. Then applying \eqref{lm:matrixholder:1} yields $\|\M M\|_2\leq \sqrt{\|\M M\|^2_1} = \|\M M\|_1.$

\end{proof}
%\begin{lemma}\label{lemma:sin_f}
%For $\M U_1,\M U_0\in \mathbb{R}^{p}$, the $\sin \Theta$ distance of $\M U_1,\M U_0$ under the Frobenius norm is given by $
%|\sin\Theta(\M U_1,\M U_0)| = \|\M U_1^{\T} \M U_0^{\perp}\|_\F$, where $\M U_0^{\perp}$ is the orthogonal extension of $\M U_0$ in $\mathbb{R}^{p}$. Then we have
%\bee\nonumber
%\|\M U_1 - \M U_0\| \leq \sqrt{2}|\sin\Theta(\M U_1,\M U_0)| 
%\ee
%\end{lemma}
%\begin{proof}[Proof of Lemma \ref{lemma:sin_f}]
%The lemma is a direct corollary of Lemma 1 in \citet{cai2018rate}.
%\end{proof}
\begin{lemma}\label{lemma:xi}
Let $\xi(\cdot)$ be the matrix transformation function defined in \eqref{xi:fun} of the main paper. Let $\M M_1,\M M_2 \in \mathbb{R}^{pq \times pq}$, $\bds m_1 = \big[m^{(1)}_{l_1m_1}\big]\in \RR^{p\times p}, \bds m_2 = \big[m^{(2)}_{l_2m_2}\big]\in \RR^{q\times q}$ be arbitrary matrices, with $1\leq l_1,m_1\leq p$ and $1\leq l_2,m_2\leq q$. We have the following properties of $\xi(\cdot)$,
\bee\nonumber
\text{(i).}&\ \|\xi(\M M_1)\|_2\leq\|\xi(\M M_1)\|_\F =\|\M M_1\|_\F,
\\
\text{(ii).}&\ \xi(c_1 \cdot \M M_1 + c_2 \cdot \M M_2) = c_1\cdot\xi(\M M_1) + c_2 \cdot \xi(\M M_2),
\\
\text{(iii).}&\ \E\big[\xi(\M M_1)\big] = \xi(\E \M M_1),
\\
\text{(iv).}&\ \xi\big(\bds m_2 \otimes \bds m_1\big) =  \xi(\bds m_2)\cdot\xi(\bds m_1)^\T.
\ee
\end{lemma}
\begin{proof}[Proof of Lemma \ref{lemma:xi}]
For (i), since $\|\cdot\|_2\leq \|\cdot\|_\F$, we have
\bee\nonumber
\|\xi(\M M_1)\|_2\leq \|\xi(\M M_1)\|_\F.
\ee
The equality $\|\xi(\M M_1)\|_\F = \|\M M_1\|_\F$, and the results in (ii) and (iii), hold by the fact that  $\xi(\cdot)$ only reorders the components of matrix $\M M_1, \M M_2$ while not changing the corresponding entries.
\par
For (iv), by definitions, we have
\bee\nonumber
&\xi(\bds m_2\otimes \bds m_1) 
\\
&= 
\begin{bmatrix}
m_{11}^{(2)} \vv(\bds m_1)^{\T} &
\dots
&
m_{q1}^{(2)}\vv(\bds m_1)^{\T} &
m_{12}^{(2)}\vv(\bds m_1)^{\T} &
\dots
&
m_{q2}^{(2)} \vv(\bds m_1)^{\T}  &
\dots
\dots &
m_{qq}^{(2)} \vv(\bds m_1)^{\T} 
\end{bmatrix}^\T
\\
&=\begin{bmatrix}
m_{11}^{(2)} &
\dots
&
m_{q1}^{(2)} &
m_{12}^{(2)} &
\dots
&
m_{q2}^{(2)}&
\dots &
\dots &
m_{qq}^{(2)} 
\end{bmatrix}^\T \cdot \vv(\bds m_1)^{\T}
\\
&= \vv(\bds m_2) \cdot \vv(\bds m_1)^{\T}.
\ee
\end{proof}
\begin{lemma}\label{lm:exp}
Let $Z \in \RR$ be a non-negative random variable. For any $x>0$, we have
\bee\nonumber
\E Z \leq x + \int_{x}^{+\infty} \p(Z\geq t)dt.
\ee
\end{lemma}
\begin{proof}[Proof of Lemma \ref{lm:exp}] 
For any non-negative random variable $Z$, we have $\E Z = \int_{0}^{+\infty}\p(Z > t)dt$ (see e.g. \citet{casella2002statistical}). Then we have
\bee\nonumber
\E Z &=\int_{0}^{x} \p(Z\geq t)dt + \int_{x}^{+\infty} \p(Z\geq t)dt
\\
&\leq (x - 0)\cdot 1 + \int_{x}^{+\infty} \p(Z\geq t)dt.
\ee
%where the inequality is by that $\p(Z >t) < 1$ always holds as a probability.
\end{proof}
%\begin{proof}[Proof of Lemma \ref{lm:exp}] 
%We can decompose 
%\bee\label{lm:exp:1}
%\E Z &= \E\big[Z\cdot\M I(Z\leq x)\big] + \E\big[Z\cdot \M I(Z> x)\big] 
%\\
%&\leq x + \E\big[Z\cdot \M I(Z> x)\big].
%\ee
%Next we bound $\E\{Z\cdot \M I(Z> x)\}$. For any non-negative random variable $Z^*$, we have $\E Z^* = \int_{0}^{+\infty}\p(Z^* > t)dt$ (see e.g. \citet{casella2002statistical}). Taking $Z^* = Z\cdot\M I(Z> x)$, we have
%$$
%\E Z\cdot\M I(Z> x) = \int_0^{+\infty}\p\{Z\cdot\M I(Z> x)>t\}dt.
%$$
%We notice that: First, when $x\leq t$, the event $Z\cdot\M I(Z> x)>t$ happens if and only if $Z > t$ ; Second, when $x > t$,  the event $Z\cdot\M I(Z> x)>t$ happens if and only if $Z > x$. Thus
%\bee\label{lm:exp:2}
%\E Z\cdot\M I(Z> x) &= \int_0^{x}\p\{Z\cdot\M I(Z> x)>t\}dt +\int_x^{+\infty}\p\{Z\cdot\M I(Z> x)>t\}dt
%\\
%&=\int_0^{x}\p(Z>x)dt +\int_x^{+\infty}\p(Z>t)dt 
%\\
%&= x\cdot\p(Z>x) +\int_x^{+\infty}\p(Z>t)dt
%\ee
%Combining \eqref{lm:exp:1} and \eqref{lm:exp:2}, we show
%\bee
%\E Z &\leq x + x\cdot\p(Z>x) +\int_x^{+\infty}\p(Z>t)dt 
%\\
%&\leq 2x  + \int_x^{+\infty}\p(Z>t)dt
%\ee
%since $\p(Z>x) \leq 1$, which finishes our proof.
%\end{proof}

\begin{lemma}\label{T1}
Let $\hat{\M \Sigma}^{\diamond}_1$,$\hat{\M \Sigma}^{\diamond}_2$ be the solution to the optimization problem,
\bee\label{lemma:trans}
(\hat{\M \Sigma}^{\diamond}_1, \hat{\M \Sigma}^{\diamond}_2) = \argmin_{(\boldsymbol\Sigma_1, \boldsymbol\Sigma_2)}\| \widetilde{\boldsymbol\Sigma}^\diamond-\boldsymbol\Sigma_2 \otimes \boldsymbol\Sigma_1\|_\rF^2. 
\ee
For any $\tilde{\M \Sigma}_1^* = [\tilde\sigma_{l_1m_1}]\in \RR^{p\times p}$ and $\tilde{\M \Sigma}_2^* = [\tilde{\sigma}_{l_2m_2}] \in \RR^{q\times q}$ such that  $\tilde{\M\Sigma}^* = \tilde{\M\Sigma}^{*}_2 \otimes \tilde{\M\Sigma}^{*}_1 \in \RR^{pq\times pq}$,  we have
\#\nonumber
\|\tilde{\M\Sigma}^{*} - \hat{\M\Sigma}^{\diamond}_2 \otimes \hat{\M\Sigma}^{\diamond}_1\|_\F &\leq 2\sqrt{2}\|\xi(\tilde{\M\Sigma}^{*}) - \xi(\widetilde{\boldsymbol\Sigma}^\diamond)\|_2.
\#
\end{lemma}
\begin{proof}[Proof of Lemma \ref{T1}] By Lemma \ref{lemma:xi}, we have
\bee\nonumber
\|\tilde{\M\Sigma}^{*} - \hat{\M\Sigma}^{\diamond}_2 \otimes \hat{\M\Sigma}^{\diamond}_1\|_\F &= \|\xi(\tilde{\M\Sigma}^{*}) - \xi(\hat{\M\Sigma}^{\diamond}_2 \otimes \hat{\M\Sigma}^{\diamond}_1)\|_\F
\\
&=\|\xi(\tilde{\M\Sigma}^{*}_2)\cdot\xi(\tilde{\M\Sigma}^{*}_1)^\T - \xi(\hat{\M\Sigma}^{\diamond}_2) \cdot\xi( \hat{\M\Sigma}^{\diamond}_1)^\T\|_\F
\\
&\leq \sqrt{2}\|\xi(\tilde{\M\Sigma}^{*}_2)\cdot\xi(\tilde{\M\Sigma}^{*}_1)^\T - \xi(\hat{\M\Sigma}^{\diamond}_2) \cdot\xi( \hat{\M\Sigma}^{\diamond}_1)^\T\|_2.
\ee
The last inequality holds because both $\xi(\tilde{\M\Sigma}^{*}_2)\cdot\xi(\tilde{\M\Sigma}^{*}_1)^\T$ and $\xi(\hat{\M\Sigma}^{\diamond}_2) \cdot\xi( \hat{\M\Sigma}^{\diamond}_1)^\T$ are rank-one matrices, which implies the rank of $\xi(\tilde{\M\Sigma}^{*}_2)\cdot\xi(\tilde{\M\Sigma}^{*}_1)^\T - \xi(\hat{\M\Sigma}^{\diamond}_2) \cdot\xi( \hat{\M\Sigma}^{\diamond}_1)^\T$ is at most $2$, and therefore $\|\xi(\tilde{\M\Sigma}^{*}_2)\cdot\xi(\tilde{\M\Sigma}^{*}_1)^\T - \xi(\hat{\M\Sigma}^{\diamond}_2) \cdot\xi( \hat{\M\Sigma}^{\diamond}_1)^\T\|_\F^2 \leq 2\|\xi(\tilde{\M\Sigma}^{*}_2)\cdot\xi(\tilde{\M\Sigma}^{*}_1)^\T - \xi(\hat{\M\Sigma}^{\diamond}_2) \cdot\xi( \hat{\M\Sigma}^{\diamond}_1)^\T\|_2^2$. Then by triangle inequality,
\bee\label{lm:approx:1}
\|\tilde{\M\Sigma}^{*} - \hat{\M\Sigma}^{\diamond}_2 \otimes \hat{\M\Sigma}^{\diamond}_1\|_\F &\leq \sqrt{2}\|\xi(\tilde{\M\Sigma}^{*}_2)\cdot\xi(\tilde{\M\Sigma}^{*}_1)^\T - \xi(\tilde{\M\Sigma}^\diamond) + \xi(\tilde{\M\Sigma}^\diamond) - \xi(\hat{\M\Sigma}^{\diamond}_2) \cdot\xi( \hat{\M\Sigma}^{\diamond}_1)^\T\|_2
\\
&\leq \sqrt{2}\|\xi(\tilde{\M\Sigma}^{*}_2)\cdot\xi(\tilde{\M\Sigma}^{*}_1)^\T - \xi(\tilde{\M\Sigma}^\diamond)\|_2 +\sqrt{2}\| \xi(\hat{\M\Sigma}^{\diamond}_2) \cdot\xi( \hat{\M\Sigma}^{\diamond}_1)^\T- \xi(\tilde{\M\Sigma}^\diamond) \|_2.
\ee
\par
\citet{Pitsianis1997} have shown that $\hat{\M \Sigma}^{\diamond}_1, \hat{\M \Sigma}^{\diamond}_2$ are obtained by reordering the components of the leading singular vectors of  $\xi(\tilde{\M\Sigma}^\diamond)$.  By Eckart-Young-Mirsky theorem \citep{eckart1936approximation}, among all the rank one matrix ${\bf A}$, $\xi(\hat{\M\Sigma}^{\diamond}_2) \cdot\xi( \hat{\M\Sigma}^{\diamond}_1)^\T$ minimizes the spectral-norm error $ \|{\bf A}-\xi(\tilde{\M\Sigma}^\diamond) \|_2$. Since $\xi(\tilde{\M\Sigma}^{*}_2)\cdot\xi(\tilde{\M\Sigma}^{*}_1)^\T $ is also a rank-one matrix, we have 
\bee\label{lm:approx:2}
\| \xi(\hat{\M\Sigma}^{\diamond}_2) \cdot\xi( \hat{\M\Sigma}^{\diamond}_1)^\T- \xi(\tilde{\M\Sigma}^\diamond) \|_2 \leq \|\xi(\tilde{\M\Sigma}^{*}_2)\cdot\xi(\tilde{\M\Sigma}^{*}_1)^\T - \xi(\tilde{\M\Sigma}^\diamond)\|_2.
\ee 
%The Eckart-Young-Mirsky theorem \citep{eckart1936approximation} guarantees $\xi(\hat{\M\Sigma}^{\diamond}_2) \cdot\xi( \hat{\M\Sigma}^{\diamond}_1)^\T$ is the rank one matrix having the smallest spectral-norm error to $\xi(\tilde{\M\Sigma}^\diamond)$, which implies 

%since $\xi(\tilde{\M\Sigma}^{*}_2)\cdot\xi(\tilde{\M\Sigma}^{*}_1)^\T $ is also a rank one matrix. 
Summarizing \eqref{lm:approx:1} and \eqref{lm:approx:2}, we conclude 
\bee\nonumber
\|\tilde{\M\Sigma}^{*} - \hat{\M\Sigma}^{\diamond}_2 \otimes \hat{\M\Sigma}^{\diamond}_1\|_\F &\leq 2\sqrt{2}\|\xi(\tilde{\M\Sigma}^{*}_2)\cdot\xi(\tilde{\M\Sigma}^{*}_1)^\T - \xi(\tilde{\M\Sigma}^\diamond)\|_2
\\
&= 2\sqrt{2}\|\xi(\tilde{\M\Sigma}^{*}) - \xi(\widetilde{\boldsymbol\Sigma}^\diamond)\|_2 \quad (\text{By Lemma \ref{lemma:xi}}). 
\ee
\end{proof}

{\color{black} We introduce the following Lemmas \ref{l:outer}--\ref{lemma:srate}, which are used to prove Theorems \ref{T2} and \ref{T3}.}
\begin{lemma}\label{l:outer}
We have the following two sets of results: \\
\noindent \textbf{(a)} When $\M \Sigma_1^{*}\in \mathcal{F}(\varepsilon_0, \alpha_1),\M \Sigma_2^{*} \in \mathcal{F}(\varepsilon_0, \alpha_2)$ and $\M\Sigma^* = \M \Sigma_2^*\otimes \M \Sigma_1^*$, one has
\bee\label{l:outer:res1}
&\frac{1}{pq}\big\|\M \Sigma_2^{*,\mathcal{B}}(k_2)\otimes\M \Sigma_1^{*,\mathcal{B}}(k_1) - \M \Sigma^*\big\|_\F^2 \precsim \M I(k_1< p - 1)\cdot k_1^{-2\alpha_1} +  \M I(k_2< q - 1)\cdot k_2^{-2\alpha_2},
\\
&\frac{1}{pq}\big\|\M \Sigma_2^{*,\mathcal{T}}(k_2)\otimes\M \Sigma_1^{*,\mathcal{T}}(k_1) - \M \Sigma^*\big\|_\F^2 \precsim \M I(k_1< 2p-2\}\cdot k_1^{-2\alpha_1} +  \M I(k_2< 2q - 2)\cdot k_2^{-2\alpha_2}.
\ee
\par
\noindent{\textbf{(b)}} When $\M \Sigma_1^{*}\in \mathcal{M}(\varepsilon_0, \alpha_1),\M \Sigma_2^{*} \in \mathcal{M}(\varepsilon_0, \alpha_2)$ and $\M\Sigma^* = \M \Sigma_2^*\otimes \M \Sigma_1^*$, one has
\bee\label{l:outer:res2}
&\frac{1}{pq}\big\|\M \Sigma_2^{*,\mathcal{B}}(k_2)\otimes\M \Sigma_1^{*,\mathcal{B}}(k_1) - \M \Sigma^*\big\|_\F^2  \precsim \M I(k_1< p - 1)\cdot k_1^{-2\alpha_1-1} +  \M I(k_2< q - 1)\cdot k_2^{-2\alpha_2-1},
\\
&\frac{1}{pq}\big\|\M \Sigma_2^{*,\mathcal{T}}(k_2)\otimes\M \Sigma_1^{*,\mathcal{T}}(k_1) - \M \Sigma^*\big\|_\F^2\precsim \M I(k_1< 2p - 2)\cdot k_1^{-2\alpha_1-1} +  \M I(k_2< 2q - 2)\cdot k_2^{-2\alpha_2-1}.
\ee
\end{lemma}
\begin{proof}[Proof of Lemma \ref{l:outer}] By Lemma \ref{lemma:xi}, we have
\bee\nonumber
\xi\big\{\M \Sigma_2^{*,\mathcal{T}}(k_2)\otimes\M \Sigma_1^{*,\mathcal{T}}(k_1)\big\} - \xi\big\{\M \Sigma^*\big\} = \vecc\big\{\M \Sigma_2^{*,\mathcal{T}}(k_2)\big\}\cdot\vecc\big\{\M \Sigma_1^{*,\mathcal{T}}(k_1)\big\}^\T - \vecc\big\{\M \Sigma_2^{*}\big\}\cdot\vecc\big\{\M \Sigma_1^{*}\big\}^\T.
\ee
Since $\vecc\big\{\M \Sigma_2^{*,\mathcal{T}}(k_2)\big\}\cdot\vecc\big\{\M \Sigma_1^{*,\mathcal{T}}(k_1)\big\}^\T $ and $\vecc\big\{\M \Sigma_2^{*}\big\}\cdot\vecc\big\{\M \Sigma_1^{*}\big\}^\T$ are both rank 1 matrices, it is easy to see the rank of $\xi\big\{\M \Sigma_2^{*,\mathcal{T}}(k_2)\otimes\M \Sigma_1^{*,\mathcal{T}}(k_1)\big\} - \xi\big\{\M \Sigma^*\big\}$ is at most $2$. By the properties of Frobenius norm, we have
\bee\label{lm:bias:convert}
&\big\|\M \Sigma_2^{*,\mathcal{T}}(k_2)\otimes\M \Sigma_1^{*,\mathcal{T}}(k_1) - \M \Sigma^*\big\|_\F^2
\\
&= \big\|\xi\big\{\M \Sigma_2^{*,\mathcal{T}}(k_2)\otimes\M \Sigma_1^{*,\mathcal{T}}(k_1)\big\} - \xi\big\{\M \Sigma^*\big\}\big\|_\F^2 \quad (\text{By Lemma \ref{lemma:xi}})
\\
&\leq 2\big\|\xi\big\{\M \Sigma_2^{*,\mathcal{T}}(k_2)\otimes\M \Sigma_1^{*,\mathcal{T}}(k_1)\big\} - \xi\big\{\M \Sigma^*\big\}\big\|_2^2.
\ee
So to show the desired results, it is enough to show the upper bound of $\big\|\xi\big\{\M \Sigma_2^{*,\mathcal{B}}(k_2)\otimes\M \Sigma_1^{*,\mathcal{B}}(k_1)\big\} - \xi\big\{\M \Sigma^*\big\}\big\|_2^2/pq$.
\\
\par
Next, we focus on bounding $\big\|\xi\big\{\M \Sigma_2^{*,\eta}(k_2)\otimes\M \Sigma_1^{*,\eta}(k_1)\big\} - \xi\big\{\M \Sigma^*\big\}\big\|_2^2/pq$ to show the desired results. We first prove the first inequality in part (a).  %$\M \Sigma_1^{*}\in \mathcal{F}(\varepsilon_0, \alpha_1),\M \Sigma_2^{*} \in \mathcal{F}(\varepsilon_0, \alpha_2)$.
\par
%\textbf{1. Banded Condition: } 
%\par 
We first derive error rates of several quantities. As $\M\Sigma^*_1\in\mathcal{F}(\epsilon_0,\alpha_1),\M\Sigma^*_2 \in \mathcal{F}(\epsilon_0,\alpha_2)$, by definition, one has
\bee\nonumber
\|\M\Sigma^*_1-\M\Sigma^{*,\mathcal{B}}_1(k_1)\|_{1,1} 
&= \sum_{1\leq l_1,m_1\leq p\atop |l_1 - m_1| > k_1} \sigma^{(1)}_{l_1,m_1}
\\
&\leq p\cdot\Big(\max_{l_1}\sum_{|l_1 - m_1|>k_1} \sigma^{(1)}_{l_1,m_1}\Big)
\\
&\precsim p\cdot k_1^{-\alpha_1}.
\ee
Similarly $\|\M\Sigma^*_2-\M\Sigma^{*,\mathcal{B}}_2(k_2)\|_{1,1} \precsim  q\cdot k_2^{-\alpha_2}$.
Meanwhile, we have
\bee\nonumber
\|\M\Sigma^*_1 - \M\Sigma^{*,\mathcal{B}}_1(k_1)\|_{\max} &\leq \|\M\Sigma^*_1 - \M\Sigma^{*,\mathcal{B}}_1(k_1)\|_1 \precsim k_1^{-\alpha_1},
\ee
and $\|\M\Sigma^*_2 - \M\Sigma^{*,\mathcal{B}}_2(k_2)\|_{\max}\precsim k_2^{-\alpha_2}$. Since $\M \Sigma_1^{*}\in \mathcal{F}(\varepsilon_0, \alpha_1),\M \Sigma_2^{*} \in \mathcal{F}(\varepsilon_0, \alpha_2)$, it is easy to see 
 $\|\M \Sigma^*_1\|_{\max}\leq \|\M \Sigma^*_1\|_{2} = \lambda_{\max}(\M\Sigma^*_1)\leq 1/\epsilon_0$, and thus $\|\M \Sigma^*_1\|_{1,1}\leq p\|\M \Sigma^*_1\|_{\max}\leq p\|\M \Sigma^*_1\|_{2}\precsim p$ and $\|\M\Sigma^{*,\mathcal{B}}_1(k_1)\|_{\max}\leq \|\M\Sigma^{*}_1\|_{\max}\precsim 1$. Similarly, $\|\M \Sigma^*_2\|_{1,1}\precsim q, \|\M\Sigma^{*,\mathcal{B}}_2(k_2)\|_{\max}\precsim 1$. 
 
As a result, when $k_1<p, k_2 < q$,
\begin{align}
\nonumber
&\frac{1}{pq}\|\xi\{\M\Sigma^{*,\mathcal{B}}_2(k_2)\otimes\M\Sigma^{*,\mathcal{B}}_1(k_1)\} - \xi\{\M\Sigma^*\}\|^2_2 
\\\nonumber
&\leq \frac{1}{pq}\|\xi\{\M\Sigma^{*,\mathcal{B}}_2(k_2)\otimes\M\Sigma^{*,\mathcal{B}}_1(k_1) - \M\Sigma^*_2 \otimes \M\Sigma^*_1\}\|_1 \cdot \|\xi\{\M\Sigma^{*,\mathcal{B}}_2(k_2)\otimes\M\Sigma^{*,\mathcal{B}}_1(k_1) - \M\Sigma^*_2 \otimes \M\Sigma^*_1\}\|_{\infty}
\\\nonumber
&\precsim\frac{1}{pq}\max\Big\{\|\M\Sigma^{*,\mathcal{B}}_1(k_1)\|_{\max}\|\M\Sigma^{*,\mathcal{B}}_2(k_2) - \M\Sigma^*_2\|_{1,1}, \|\M\Sigma^*_1 - \M\Sigma^{*,\mathcal{B}}_1(k_1)\|_{\max}\|\M\Sigma^*_2\|_{1,1}\Big\} 
\\\label{lm:outer:band:1:1}
&~~~\times \max\Big\{\|\M\Sigma^{*,\mathcal{B}}_2(k_2)\|_{\max}\|\M\Sigma^{*,\mathcal{B}}_1(k_1) - \M\Sigma^*_1\|_{1,1}, \|\M\Sigma^*_2 - \M\Sigma^{*,\mathcal{B}}_2(k_2)\|_{\max}\|\M\Sigma^*_1\|_{1,1}\Big\}
\\\nonumber
&\precsim\frac{1}{pq}\max\Big\{qk_2^{-\alpha_2},qk_1^{-\alpha_1}\Big\} \times \max\Big\{pk_1^{-\alpha_1},pk_2^{-\alpha_2}\Big\}
\\\nonumber
&\precsim \max\{k_{1}^{-2\alpha_1},k_2^{-2\alpha_2}\}
\\\label{rate2}
&\asymp k_{1}^{-2\alpha_1} + k_2^{-2\alpha_2},
\end{align}
where the first inequality holds by matrix H\"older's inequality (Lemma \ref{lm:matrixholder}). For the second inequality, by Lemma \ref{lemma:xi} of $\xi(\cdot)$, we know each column of $\xi\{\M\Sigma^{*,\mathcal{B}}_2(k_2)\otimes\M\Sigma^{*,\mathcal{B}}_1(k_1) - \M\Sigma^*_2 \otimes \M\Sigma^*_1\} = \xi\{\M\Sigma^{*,\mathcal{B}}_2(k_2)\} \cdot \xi\{\M\Sigma^{*,\mathcal{B}}_1(k_1)\}^\T - \xi\{\M\Sigma^*_2\} \cdot \xi\{\M\Sigma^*_1\}^\T$ is either in the form of  $\sigma_{l,m}^{(1)} \cdot \vecc\big\{\M \Sigma_2^{*,\mathcal{B}}(k_2) -  \M \Sigma_2^*\big\} \text{ when } |l - m|\leq k_1$ (so $\sigma_{l,m}^{(1)}$ is the $lm$th entry of both $\M \Sigma_1^{*,\mathcal{B}}(k_1)$ and $\M \Sigma_1^{*,\mathcal{B}}$), or in the form of $0 \times\vecc\{\M \Sigma_2^{*,\mathcal{B}}(k_2)\} -\sigma_{l,m}^{(1)} \vecc\{\M \Sigma_2^*\}$ when $|l-m|>k_1$ (so $0$ is the $lm$th entry of $\M \Sigma_1^{*,\mathcal{B}}(k_1)$ and $\sigma_{l,m}^{(1)}$ is the $lm$th entry of $\M \Sigma_1^{*,\mathcal{B}}$). Since a matrix's $\|\cdot\|_1$ norm is the maximum among all its column vectors' $\ell_1$ norm, we have
\bee\label{lm:outer:add1}
&\|\xi\{\M\Sigma^{*,\mathcal{B}}_2(k_2)\otimes\M\Sigma^{*,\mathcal{B}}_1(k_1) - \M\Sigma^*_2 \otimes \M\Sigma^*_1\}\|_1
\\
&\leq \max\Big\{\underbrace{\max_{|l-m|\leq k_1}|\sigma^{(1)}_{l,m}|}_{=\|\M \Sigma^{*,\mathcal{B}}_1(k_1)\|_{\max}}\cdot\underbrace{\big\|\vecc\big\{\M \Sigma_2^{*,\mathcal{B}}(k_2) -  \M \Sigma_2^*\big\}\big\|_1}_{=\big\|\M \Sigma_2^{*,\mathcal{B}}(k_2) -  \M \Sigma_2^*\big\|_{1,1}},
\underbrace{\max_{|l-m|>k_1}|\sigma_{l,m}^{(1)}|}_{=\|\M \Sigma_1^* - \M \Sigma_1^{*,\mathcal{B}}(k_2)\|_{\max}}\cdot\underbrace{\|\vecc\{\M\Sigma_2^*\}\|_1}_{=\|\M \Sigma_2^*\|_{1,1}}\Big\}
\\
&=\max\Big\{\|\M\Sigma^{*,\mathcal{B}}_1(k_1)\|_{\max}\|\M\Sigma^{*,\mathcal{B}}_2(k_2) - \M\Sigma^*_2\|_{1,1}, \|\M\Sigma^*_1 - \M\Sigma^{*,\mathcal{B}}_1(k_1)\|_{\max}\|\M\Sigma^*_2\|_{1,1}\Big\}.
\ee 
By symmetry, we also have
\bee\label{lm:outer:add2}
&\|\xi\{\M\Sigma^{*,\mathcal{B}}_2(k_2)\otimes\M\Sigma^{*,\mathcal{B}}_1(k_1) - \M\Sigma^*_2 \otimes \M\Sigma^*_1\}\|_\infty
\\
&\leq \max\Big\{\|\M\Sigma^{*,\mathcal{B}}_2(k_2)\|_{\max}\|\M\Sigma^{*,\mathcal{B}}_1(k_1) - \M\Sigma^*_1\|_{1,1}, \|\M\Sigma^*_2 - \M\Sigma^{*,\mathcal{B}}_2(k_2)\|_{\max}\|\M\Sigma^*_1\|_{1,1}\Big\}.
\ee 
Combining \eqref{lm:outer:add1} and \eqref{lm:outer:add2} yields the second inequality of \eqref{lm:outer:band:1:1}. The final rate \eqref{rate2} follows by using the previously derived error rates. Similarly, we can also consider the cases when $k_1 \geq p - 1$ or $k_2 \geq q - 1$, under which either $\M \Sigma_1^{*,\mathcal{B}}(k_1) = \M \Sigma_1^{*}$ or $\M \Sigma_2^{*,\mathcal{B}}(k_2) = \M \Sigma_2^{*}$.  The only difference is, when $\M \Sigma_1^{*,\mathcal{B}}(k_1) = \M \Sigma_1^{*}$ or $\M \Sigma_2^{*,\mathcal{B}}(k_2) = \M \Sigma_2^{*}$, one of the terms in each $\max\{\cdot,\cdot\}$ of \eqref{lm:outer:band:1:1} is zero and thus the final rate $k_{1}^{-2\alpha_1} + k_2^{-2\alpha_2}$ will degenerate to either $k_{1}^{-2\alpha_1}$ or $k_2^{-2\alpha_2}$. Finally we can show
\bee\label{T2:p:b:case2}
\frac{1}{pq}\|\xi\{\M\Sigma^{*,\mathcal{B}}_2(k_2)\otimes\M\Sigma^{*,\mathcal{B}}_1(k_1)\} - \xi\{\M\Sigma^*\}\|^2_2
\precsim 
\begin{cases}
k_1^{-2\alpha_1} + k_2^{-2\alpha_2} & k_1 <p - 1, k_2 < q - 1
\\
k_2^{-2\alpha_2} & k_1 \geq p - 1,k_2<q - 1
\\
k_1^{-2\alpha_1} & k_2 \geq q - 1,k_1<p - 1
\\
0 & k_1 \geq p - 1, k_2 \geq q - 1.
\end{cases}
\ee
Combining \eqref{lm:bias:convert} and \eqref{T2:p:b:case2}, we show the first inequality in part (a). 
\\
\par
We then prove the second inequality in part (a). By definition, it is easy to see that $\|\M\Sigma^{*,\mathcal{T}}_{u}(k_{u})\|_{\max}\leq \|\M\Sigma^{*,\mathcal{B}}_{u}(k_{u})\|_{\max}$ and 
\bee\nonumber
&\|\M\Sigma^{*,\mathcal{T}}_{u}(k_{u}) - \M\Sigma^*_{u}\|_{1,1}\leq \|\M\Sigma^{*,\mathcal{B}}_{u}(\lfloor k_{u}/2\rfloor) - \M\Sigma^*_{u}\|_{1,1}, 
\\
&\|\M\Sigma^{*,\mathcal{T}}_{u}(k_{u}) - \M\Sigma^*_{u}\|_{\max}\leq \|\M\Sigma^{*,\mathcal{B}}_{u}(\lfloor k_{u}/2\rfloor) - \M\Sigma^*_{u}\|_{\max},
\ee
where $u \in \{1,2\}$. Then, when $(k_1,k_2) \neq (p,q)$, similar to \eqref{rate2}, we can show
\begin{align}\nonumber
&\frac{1}{pq}\big\|\xi\{\M\Sigma^{*,\mathcal{T}}_2(k_2)\otimes\M\Sigma^{*,\mathcal{T}}_1(k_1)\} - \xi\{\M\Sigma^*\}\big\|^2_2 
\\
\nonumber
&\leq \frac{1}{pq}\big\|\xi\big\{\M\Sigma^{*,\mathcal{T}}_2(k_2)\otimes\M\Sigma^{*,\mathcal{T}}_1(k_1) - \M\Sigma^*_2 \otimes \M\Sigma^*_1\big\}\big\|_1 \cdot \big\|\xi\big\{\M\Sigma^{*,\mathcal{T}}_2(k_2)\otimes\M\Sigma^{*,\mathcal{T}}_1(k_1) - \M\Sigma^*_2 \otimes \M\Sigma^*_1\big\}\big\|_{\infty}
\\
\nonumber
&\precsim\frac{1}{pq}\max\Big\{\|\M\Sigma^{*,\mathcal{T}}_1(k_1)\|_{\max}\|\M\Sigma^{*,\mathcal{T}}_2(k_2) - \M\Sigma^*_2\|_{1,1}, \|\M\Sigma^*_1 - \M\Sigma^{*,\mathcal{T}}_1(k_1)\|_{\max}\|\M\Sigma^*_2\|_{1,1}\Big\} 
\\\label{lm:outer:taper:1:1}
&\times \max\Big\{\|\M\Sigma^{*,\mathcal{T}}_2(k_2)\|_{\max}\|\M\Sigma^{*,\mathcal{T}}_1(k_1) - \M\Sigma^*_1\|_{1,1}, \|\M\Sigma^*_2 - \M\Sigma^{*,\mathcal{T}}_2(k_2)\|_{\max}\|\M\Sigma^*_1\|_{1,1}\Big\}
\\\nonumber
&\precsim\frac{1}{pq}\max\Big\{\|\M\Sigma^{*,\mathcal{B}}_1(k_1)\|_{\max}\|\M\Sigma^{*,\mathcal{B}}_2(\lfloor k_2/2\rfloor) - \M\Sigma^*_2\|_{1,1}, \|\M\Sigma^*_1 - \M\Sigma^{*,\mathcal{B}}_1(\lfloor k_1/2\rfloor)\|_{\max}\|\M\Sigma^*_2\|_{1,1}\Big\} 
\\\nonumber
&\times \max\Big\{\|\M\Sigma^{*,\mathcal{B}}_2(k_2)\|_{\max}\|\M\Sigma^{*,\mathcal{B}}_1(\lfloor k_1/2\rfloor) - \M\Sigma^*_1\|_{1,1}, \|\M\Sigma^*_2 - \M\Sigma^{*,\mathcal{B}}_2(\lfloor k_2/2\rfloor)\|_{\max}\|\M\Sigma^*_1\|_{1,1}\Big\}
\\\label{lm:outer:taper:1}
&\precsim \max\{k_{1}^{-2\alpha_1},k_2^{-2\alpha_2}\}
\\
&\asymp k_{1}^{-2\alpha_1} + k_2^{-2\alpha_2}.
\end{align}
Further considering the cases of  $k_1 \geq 2p - 2$ or $k_2 \geq 2q - 2$, we have either $\M \Sigma_1^{*,\mathcal{T}}(k_1) = \M \Sigma_1^{*}$ or $\M \Sigma_2^{*,\mathcal{T}}(k_2) = \M \Sigma_2^{*}$. Same argument as \eqref{lm:outer:taper:1} can also be applied to these cases. The only difference is, when $\M \Sigma_1^{*,\mathcal{T}}(k_1) = \M \Sigma_1^{*}$ or $\M \Sigma_2^{*,\mathcal{T}}(k_2) = \M \Sigma_2^{*}$, one of the terms in each $\max\{\cdot,\cdot\}$ of \eqref{lm:outer:taper:1:1} is zero and thus the final rate $k_{1}^{-2\alpha_1} + k_2^{-2\alpha_2}$ will degenerate to either $k_{1}^{-2\alpha_1}$ or $k_2^{-2\alpha_2}$. And finally we have
\bee\label{rate5}
&\frac{1}{pq}\big\|\xi\big\{\M \Sigma_2^{*,\mathcal{T}}(k_2)\otimes \M \Sigma_1^{*,\mathcal{T}}(k_1)\big\} - \xi\big\{\M \Sigma^*\big\}\big\|^2_2
\precsim 
\begin{cases}
k_1^{-2\alpha_1} + k_2^{-2\alpha_2} & k_1 <2p - 2, k_2 < 2q - 2
\\
k_2^{-2\alpha_2} & k_1 \geq 2p - 2,k_2<2q - 2
\\
k_1^{-2\alpha_1} & k_2 \geq 2q - 2,k_1<2p - 2
\\
0 & k_1 \geq 2p - 2, k_2 \geq 2q - 2.
\end{cases}
\ee
Combining \eqref{lm:bias:convert} and \eqref{rate5}, we show the second inequality in part (a). 
\\
\par

Next we prove the first inequality in part (b). Recall that $\M \Sigma^* = [
\sigma_{(l_1,m_1),(l_2,m_2)} 
]$ where $\sigma_{(l_1,m_1),(l_2,m_2)} = \sigma_{l_1,m_1}^{(1)}\cdot \sigma_{l_2,m_2}^{(2)} $ is the $\big\{(l_2-1)\cdot p + l_1\big\},\big\{(m_2-1)\cdot p + m_1\big\}$th element of $\M \Sigma^*$. For the simplicity of notation, we write $\bds \beta_1 = (l_1,m_1)^\T, \bds \beta_2 = (l_2,m_2)^\T$ and define $d(\bds \beta_1) = |l_1 - m_1|, d(\bds \beta_2) = |l_2 - m_2|.$
\par
We can show that 
\bee\label{l:ourter:1}
&\frac{1}{pq}\|\xi\{\M\Sigma_2^{*,\mathcal{B}}(k_2)\otimes \M\Sigma_1^{*,\mathcal{B}}(k_1)\}-\xi\{\M\Sigma^{*}\}\|^2_2 
\\
&\leq \frac{1}{pq}\|\xi\{\M\Sigma_2^{*,\mathcal{B}}(k_2)\otimes \M\Sigma_1^{*,\mathcal{B}}(k_1)\}-\xi\{\M\Sigma^{*}\}\|^2_\F 
\\
&=\frac{1}{pq}\|\M\Sigma_2^{*,\mathcal{B}}(k_2)\otimes \M\Sigma_1^{*,\mathcal{B}}(k_1)-\M\Sigma^{*}\|^2_\F 
\\
&= \frac{1}{pq} \sum_{d(\boldsymbol\beta_1)>k_1 \atop \text{ or }d(\boldsymbol\beta_2)>k_2} \sigma_{\boldsymbol\beta_1,\boldsymbol\beta_2}^2
\\
& = \frac{1}{pq}\sum_{d(\boldsymbol\beta_1)\leq k_1, \atop d(\boldsymbol\beta_2)> k_2}\sigma_{\boldsymbol\beta_1,\boldsymbol\beta_2}^2 + \frac{1}{pq}\sum_{d(\boldsymbol\beta_1)> k_1,\atop d(\boldsymbol\beta_2)\leq k_2}\sigma_{\boldsymbol\beta_1,\boldsymbol\beta_2}^2 + \frac{1}{pq}\sum_{d(\boldsymbol\beta_1)>k_1,\atop d(\boldsymbol\beta_2)> k_2}\sigma_{\boldsymbol\beta_1,\boldsymbol\beta_2}^2.
\ee
Since $\M \Sigma_1^{*}\in \mathcal{M}(\varepsilon_0, \alpha_1),\M \Sigma_2^{*} \in \mathcal{M}(\varepsilon_0, \alpha_2)$, when $d(\bds \beta_1)\geq 1, d(\bds \beta_2)\geq 1,$ we have 
\bee\nonumber
&\big(\sigma_{\boldsymbol\beta_1}^{(1)}\big)^2\leq C_1^2 d(\boldsymbol\beta_1)^{-2\alpha_1-2} = C_1^2|l_1 - m_1|^{-2\alpha_1 - 2},
\\
&\big(\sigma_{\boldsymbol\beta_2}^{(2)}\big)^2 \leq C_1^2d(\boldsymbol\beta_2)^{-2\alpha_2-2} = C_1^2|l_2 - m_2|^{-2\alpha_2 - 2},
\ee
for some constant $C>0$. Then we can show
\bee\label{lm:pf:bias:betak1}
\sum_{\substack{d(\boldsymbol\beta_1)\leq k_1}} {\big(\sigma_{\boldsymbol\beta_1}^{(1)}\big)}^2&=\sum_{1\leq l_1\leq p}\sum_{m_1: |m_1 - l_1|\leq k_1}{\big(\sigma_{\boldsymbol\beta_1}^{(1)}\big)}^2
\\
&\leq \sum_{1\leq l_1\leq p}\sum_{1\leq m_1\leq p}{\big(\sigma_{\boldsymbol\beta_1}^{(1)}\big)}^2
\\
&\leq \sum_{1\leq l_1\leq p}\Big\{1/\varepsilon_0^2 +\sum_{1\leq m_1\leq p\atop m_1\neq l_1}C_1^2|l_1 - m_1|^{-2\alpha_2 - 2}\Big\}
\\
&\precsim p,
\ee
where $\varepsilon_0$ is given in the definition of $\mathcal{M}(\varepsilon_0,\alpha_1)$ such that $|\sigma_{l_1m_1}|\leq \|\M\Sigma_1^*\|_{\max}\leq \|\M\Sigma_1^*\|_2\leq 1/\varepsilon_0 $ for all $l_1 =m_1$. Therefore,
\bee\label{T3:1}
\frac{1}{pq}\sum_{d(\boldsymbol\beta_1)\leq k_1,\atop d(\boldsymbol\beta_2)> k_2}\sigma_{\boldsymbol\beta_1,\boldsymbol\beta_2}^2 &= \frac{1}{pq}\sum_{d(\boldsymbol\beta_1)\leq k_1,\atop d(\boldsymbol\beta_2)> k_2}{\big(\sigma_{\boldsymbol\beta_1}^{(1)}\big)}^2\cdot {\big(\sigma_{\boldsymbol\beta_2}^{(2)}\big)}^2 
\\
&= \frac{1}{pq}\Big\{\sum_{\substack{d(\boldsymbol\beta_1)\leq k_1}} {\big(\sigma_{\boldsymbol\beta_1}^{(1)}\big)}^2 \Big\}\cdot \Big\{\sum_{d(\boldsymbol\beta_2)> k_2} {\big(\sigma_{\boldsymbol\beta_2}^{(2)}\big)}^2 \Big\}
\\
&\precsim \frac{1}{pq}\cdot p \cdot \Big\{\sum_{d(\boldsymbol\beta_2)>k_2}{\big(\sigma_{\boldsymbol\beta_2}^{(2)}\big)}^2\Big\}\cdot\M I(k_2< q - 1) \quad (\text{By \eqref{lm:pf:bias:betak1}}) 
\\
&\precsim \frac{1}{pq} \cdot p \cdot qk_2^{-2\alpha_2 - 1}\cdot\M I( k_2< q - 1) \\
&= {k_2^{-2\alpha_2-1}}\cdot\M I( k_2< q-1).
\ee
The multiplication of $\M I(k_2 < q-1)$ in the first inequality appears because $\sum_{d(\boldsymbol\beta_2)>k_2}{\big(\sigma_{\boldsymbol\beta_2}^{(2)}\big)}^2=0$ for $k_2 \geq q-1$.
%by the fact that when $k_2 \geq q-1$, there will be no term in $\sum_{d(\boldsymbol\beta_2)>k_2}{\big(\sigma_{\boldsymbol\beta_2}^{(2)}\big)}^2$. 
The last inequality holds as $\big(\sigma_{\boldsymbol\beta_2}^{(2)}\big)^2 \leq C_1^2d(\boldsymbol\beta_2)^{-2\alpha_2-2}$ and we can show that
\bee\nonumber
\frac{1}{q}\sum_{d(\boldsymbol\beta_2)>k_2}{\big(\sigma_{\boldsymbol\beta_2}^{(2)}\big)}^2 &= \frac{1}{q}\sum_{1\leq l_2\leq q}\sum_{m_2:|m_2 - l_2|>k_2}\big(\sigma_{l_2,m_2}^{(2)}\big)^2
\\
&\leq \frac{1}{q}\sum_{1\leq l_2\leq q}\sum_{m_2:|m_2 - l_2|>k_2}|l_2 - m_2|^2
\\
&\precsim  q\cdot \frac{1}{q}\cdot k_2^{-2\alpha_2-1}\cdot\M I(k_2 < q-1)
\\
&\precsim k_2^{-2\alpha_2-1}\cdot\M I(k_2 < q - 1) .
\ee
Similar to (\ref{T3:1}), one has
\bee\label{T3:2}
\frac{1}{pq}\sum_{d(\boldsymbol\beta_1)> k_1, \atop d(\boldsymbol\beta_2)\leq k_2}\sigma_{\boldsymbol\beta_1,\boldsymbol\beta_2}^2\precsim  k_1^{-2\alpha_1-1}\cdot\M I(k_1< p - 1).
\ee
Also by the fact $\M \Sigma_1^* \in \mathcal{M}(\epsilon_0,\alpha_1), \M \Sigma_2^* \in \mathcal{M}(\epsilon_0,\alpha_2)$ and the definition of $\mathcal{M}(\epsilon,\alpha)$ class,
\bee\label{T3:3}
\frac{1}{pq}\sum_{d(\boldsymbol\beta_1)>k_1 \atop d(\boldsymbol\beta_2)> k_2}\sigma_{\boldsymbol\beta_1,\boldsymbol\beta_2}^2 &= \frac{1}{pq}\Big\{\sum_{d(\boldsymbol\beta_1)>k_1}{\big(\sigma_{\boldsymbol\beta_1}^{(1)}\big)}^2\Big\}\cdot \Big\{\sum_{d(\boldsymbol\beta_2)>k_2}{\big(\sigma_{\boldsymbol\beta_2}^{(2)}\big)}^2\Big\}
\\
&\precsim k_1^{-2\alpha_1-1} \cdot k_2^{-2\alpha_2-1}\cdot\M I(k_1< p-1, k_2< q-1).
\ee
Summarizing the results above, we have 
\bee\nonumber
&\frac{1}{pq}\|\xi\{\M\Sigma_2^{*,\mathcal{B}}(k_2)\otimes \M\Sigma_1^{*,\mathcal{B}}(k_1)\}-\xi\{\M\Sigma^{*}\}\|^2_2  
\\
&\precsim k_1^{-2\alpha_1-1}\cdot\M I(k_1 < p-1) + k_2^{-2\alpha_2-1}\cdot\M I(k_2 < q-1) 
+k_1^{-2\alpha_1-1} \cdot k_2^{-2\alpha_2-1}\cdot\M I(k_1< p-1, k_2< q-1)
\\
&\asymp \M I(k_1 < p-1)\cdot k_1^{-2\alpha_1-1} + \M I(k_2 < q-1)\cdot k_2^{-2\alpha_2-1}.
\ee
Combining it with \eqref{lm:bias:convert}, we can obtain the first inequality in part (b).
\\
\par
Finally we prove the second inequality in part (b). To give a similar upper bound for $\frac{1}{pq}\|\xi\{\M\Sigma_2^{*,\mathcal{T}}(k_2)\otimes \M\Sigma_1^{*,\mathcal{T}}(k_1)\}-\xi\{\M\Sigma^{*}\}\|^2_2$, we can replace $k_1,k_2$ with $\lfloor\frac{k_1}{2}\rfloor,\lfloor\frac{k_2}{2}\rfloor$ in the proof for proposed banded estimator above. By definition, compared with $\M \Sigma^*$, the corresponding element in $\M\Sigma_2^{*,\mathcal{T}}(k_2)\otimes \M\Sigma_1^{*,\mathcal{T}}(k_1)$ is equal to $\sigma_{(l_1,m_1),(l_2,m_2)}$ when $|l_1 - m_1|\leq \lfloor k_1/2\rfloor, |l_2 - m_2|\leq \lfloor k_2/2\rfloor$ and is a shrinkage of $\sigma_{(l_1,m_1),(l_2,m_2)}$ otherwise. Thus, when $|l_1 - m_1|> \lfloor k_1/2\rfloor$ or $|l_2 - m_2|> \lfloor k_2/2\rfloor$, the absolute value of corresponding element in $\M\Sigma_2^{*,\mathcal{T}}(k_2)\otimes \M\Sigma_1^{*,\mathcal{T}}(k_1)-\M\Sigma^{*}$ is less or equal to $|\sigma_{(l_1,m_1),(l_2,m_2)}|$, which implies $\frac{1}{pq}\|\M\Sigma_2^{*,\mathcal{T}}(k_2)\otimes \M\Sigma_1^{*,\mathcal{T}}(k_1)-\M\Sigma^{*}\|^2_\F \leq \frac{1}{pq} \sum_{d(\boldsymbol\beta_1)>\lfloor k_1/2 \rfloor \atop \text{ or }d(\boldsymbol\beta_2)>\lfloor k_2/2\rfloor} \sigma_{\boldsymbol\beta_1,\boldsymbol\beta_2}^2$. Same as \eqref{l:ourter:1}, we can show
\bee\label{l:outer:2}
&\frac{1}{pq}\|\xi\{\M\Sigma_2^{*,\mathcal{T}}(k_2)\otimes \M\Sigma_1^{*,\mathcal{T}}(k_1)\}-\xi\{\M\Sigma^{*}\}\|^2_2 
\\
&\leq \frac{1}{pq}\|\xi\{\M\Sigma_2^{*,\mathcal{T}}(k_2)\otimes \M\Sigma_1^{*,\mathcal{T}}(k_1)\}-\xi\{\M\Sigma^{*}\}\|^2_\F 
\\
&=\frac{1}{pq}\|\M\Sigma_2^{*,\mathcal{T}}(k_2)\otimes \M\Sigma_1^{*,\mathcal{T}}(k_1)-\M\Sigma^{*}\|^2_\F 
\\
&\leq \frac{1}{pq} \sum_{d(\boldsymbol\beta_1)>\lfloor k_1/2 \rfloor \atop \text{ or }d(\boldsymbol\beta_2)>\lfloor k_2/2\rfloor} \sigma_{\boldsymbol\beta_1,\boldsymbol\beta_2}^2
\\
&= \frac{1}{pq}\sum_{d(\boldsymbol\beta_1)\leq \lfloor k_1/2 \rfloor , \atop d(\boldsymbol\beta_2)> \lfloor k_2/2\rfloor}\sigma_{\boldsymbol\beta_1,\boldsymbol\beta_2}^2 + \frac{1}{pq}\sum_{d(\boldsymbol\beta_1)> \lfloor k_1/2\rfloor,\atop d(\boldsymbol\beta_2)\leq \lfloor k_2/2\rfloor }\sigma_{\boldsymbol\beta_1,\boldsymbol\beta_2}^2 + \frac{1}{pq}\sum_{d(\boldsymbol\beta_1)>\lfloor k_1/2\rfloor,\atop d(\boldsymbol\beta_2)> \lfloor k_2/2\rfloor }\sigma_{\boldsymbol\beta_1,\boldsymbol\beta_2}^2
\\
&=\frac{1}{pq}\|\M\Sigma_2^{*,\mathcal{B}}(\lfloor k_2/2\rfloor)\otimes \M\Sigma_1^{*,\mathcal{B}}(\lfloor k_1/2\rfloor)-\M\Sigma^{*}\|^2_\F.
\ee
\par
Using similar argument in bounding $\frac{1}{pq}\|\M\Sigma_2^{*,\mathcal{B}}(k_2)\otimes \M\Sigma_1^{*,\mathcal{B}}(k_1)-\M\Sigma^{*}\|^2_\F$, one has
\bee\nonumber
\frac{1}{pq}\|\M\Sigma_2^{*,\mathcal{B}}(\lfloor k_2/2\rfloor)\otimes \M\Sigma_1^{*,\mathcal{B}}(\lfloor k_1/2\rfloor)-\M\Sigma^{*}\|^2_\F\precsim \M I(k_1< 2p - 2)\cdot k_1^{-2\alpha_1-1} +  \M I(k_2< 2q - 2)\cdot k_2^{-2\alpha_2-1}.
\ee Combining with \eqref{l:outer:2}, we finally have
\bee\label{T3:main:res1:pf}
\frac{1}{pq}\big\|\xi\big\{\M \Sigma_2^{*,\mathcal{T}}(k_2)\otimes\M \Sigma_1^{*,\mathcal{T}}(k_1)\big\} - \xi\big\{\M \Sigma^*\big\}\big\|_2^2 \precsim \M I(k_1< 2p - 2)\cdot k_1^{-2\alpha_1-1} +  \M I(k_2< 2q - 2)\cdot k_2^{-2\alpha_2-1}.
\ee
Combine it  with \eqref{lm:bias:convert}, we can obtain the second inequality in part (b).
\end{proof}
\begin{lemma}\label{l:frate}
Let $\vecc(\M X_1),\vecc(\M  X_2),\cdots,\vecc(\M X_n)$ be i.i.d random vectors in $\RR^{pq}$ with true covariance $\M \Sigma^* = \M \Sigma^*_2 \otimes \M \Sigma^*_1$. Assume $\E(|x^{(i)}_{l_1,l_2}\cdot x^{(i)}_{m_1,m_2}|^2) \leq M < +\infty$, where $M$ is a constant that does not depend on $n, l_1,m_1, l_2,m_2$.
\par
For $\eta\in\{\mathcal{B},\mathcal{T}\}$, we have
\begin{align}\label{l:frate:res1}
&\frac{1}{pq}\E{\|\tilde{\M \Sigma}_{\mathcal{\eta}}(k_1,k_2) - \M\Sigma^{*,\eta}_2(k_2)\otimes \M\Sigma^{*,\eta}_1(k_1)\|_{\F}^2}{} \precsim \frac{k_1k_2}{n
},
\end{align}
when $\M \Sigma_1^{*}\in \mathcal{F}(\varepsilon_0, \alpha_1),\M \Sigma_2^{*} \in \mathcal{F}(\varepsilon_0, \alpha_2)$ or $\M \Sigma_1^{*}\in \mathcal{M}(\varepsilon_0, \alpha_1),\M \Sigma_2^{*} \in \mathcal{M}(\varepsilon_0, \alpha_2)$.
\end{lemma}
\begin{proof}[Proof of Lemma \ref{l:frate}]
The notation used in this proof is mainly introduced in Section \ref{notation:main}.  Without loss of generality, we assume $\E \vecc(\M X_i) = \bds 0_{pq}$. We first prove the bound \eqref{l:frate:res1} for the proposed banded estimator. Since $\mathcal{M}(\varepsilon_0,\alpha)$ is a subset of $\mathcal{F}(\varepsilon_0,\alpha)$ for any $\alpha>0$, we only need to prove the result for the case that $\M \Sigma_1^{*}\in \mathcal{F}(\varepsilon_0, \alpha_1),\M \Sigma_2^{*} \in \mathcal{F}(\varepsilon_0, \alpha_2)$. Then the bound for $\M \Sigma_1^{*}\in \mathcal{M}(\varepsilon_0, \alpha_1),\M \Sigma_2^{*} \in \mathcal{M}(\varepsilon_0, \alpha_2)$ holds as a special case. 
%and treat $\M \Sigma_1^{*}\in \mathcal{M}(\varepsilon_0, \alpha_1),\M \Sigma_2^{*} \in \mathcal{M}(\varepsilon_0, \alpha_2)$ as a special condition. 
\par
Recal that $x^{(i)}_{l_1,l_2}$ is the $l_1l_2$th entry of $\M X_i$, it is easy to see that,
\bee\label{T3:var}
\mathbb{E} \hat{\sigma}_{\boldsymbol\beta_1 ,\boldsymbol\beta_2}^{(0)}&=\sigma_{\boldsymbol\beta_1,\boldsymbol\beta_2},
\\
\text{Var}\big(\hat{\sigma}_{\boldsymbol\beta_1,\boldsymbol\beta_2}^{(0)}\big) &= \frac{1}{n}\text{Var}\big(x^{(i)}_{l_1,l_2}\cdot x^{(i)}_{m_1,m_2}\big)\leq\frac{1}{n} \mathbb{E}\big(x^{(i)}_{l_1,l_2}\cdot x^{(i)}_{m_1,m_2}\big)^2\leq \frac{M}{n},
\ee
So for any $\boldsymbol\beta_1,\boldsymbol\beta_2$,
\bee\label{l:frate:1}
\mathbb{E}\big(\tilde{\sigma}^{\mathcal{B}}_{\boldsymbol\beta_1,\boldsymbol\beta_2} - \sigma_{\boldsymbol\beta_1,\boldsymbol\beta_2}\big)^2 &= \mathbb{E}\big(w^{\mathcal{B}}_{\boldsymbol\beta_1,\boldsymbol\beta_2}\cdot \hat{\sigma}_{\boldsymbol\beta_1,\boldsymbol\beta_2}^{(0)} - \sigma_{\boldsymbol\beta_1,\boldsymbol\beta_2}\big)^2 
\\
& =  \mathbb{E}\big(w^{\mathcal{B}}_{\boldsymbol\beta_1,\boldsymbol\beta_2}\cdot \hat{\sigma}_{\boldsymbol\beta_1,\boldsymbol\beta_2}^{(0)} -w^{\mathcal{B}}_{\boldsymbol\beta_1,\boldsymbol\beta_2}\cdot\sigma_{\boldsymbol\beta_1,\boldsymbol\beta_2}+ w^{\mathcal{B}}_{\boldsymbol\beta_1,\boldsymbol\beta_2}\cdot\sigma_{\boldsymbol\beta_1,\boldsymbol\beta_2} - \sigma_{\boldsymbol\beta_1,\boldsymbol\beta_2}\big)^2 
\\
&= \mathbb{E}\big(w^{\mathcal{B}}_{\boldsymbol\beta_1,\boldsymbol\beta_2}\cdot \hat{\sigma}_{\boldsymbol\beta_1,\boldsymbol\beta_2}^{(0)} -w^{\mathcal{B}}_{\boldsymbol\beta_1,\boldsymbol\beta_2}\cdot\sigma_{\boldsymbol\beta_1,\boldsymbol\beta_2}\big)^2 + \big(w^{\mathcal{B}}_{\boldsymbol\beta_1,\boldsymbol\beta_2}\cdot\sigma_{\boldsymbol\beta_1,\boldsymbol\beta_2} - \sigma_{\boldsymbol\beta_1,\boldsymbol\beta_2}\big)^2
\\
&\leq {\big(w^{\mathcal{B}}_{\boldsymbol\beta_1,\boldsymbol\beta_2}\big)}^2 \cdot\frac{M}{n} + \big(w^{\mathcal{B}}_{\boldsymbol\beta_1,\boldsymbol\beta_2} - 1\big)^2\sigma_{\boldsymbol\beta_1,\boldsymbol\beta_2}^2.
\ee
\par
By triangle inequality, we have
\bee\label{T_3_1}
&\frac{1}{pq}\E\|{\tilde{\M \Sigma}_{\mathcal{B}}(k_1,k_2) - \M\Sigma^{*,\mathcal{B}}_2(k_2)\otimes \M\Sigma^{*,\mathcal{B}}_1(k_1)\|_{\F}^2}{} 
\\
&\leq \frac{1}{pq}\E\Big\{{\|\tilde{\M \Sigma}_{\mathcal{B}}(k_1,k_2)-\DBCovZ\|_\F}
+{\|\DBCovZ-\M\Sigma_2^{*,\mathcal{B}}(k_2)\otimes \M\Sigma_1^{*,\mathcal{B}}(k_1)\|_\F}\Big\}^2
\\
&\leq\frac{2}{pq}\E\|\tilde{\M \Sigma}_{\mathcal{B}}(k_1,k_2)-\DBCovZ\|_\F^2 +\frac{2}{pq}\E\|\DBCovZ-\M\Sigma_2^{*,\mathcal{B}}(k_2)\otimes \M\Sigma_1^{*,\mathcal{B}}(k_1)\|^2_\F.
\ee
Now we bound the two terms on the right-hand side, respectively. First, we  bound $\E\|\tilde{\M \Sigma}_{\mathcal{B}}(k_1,k_2)-\DBCovZ\|_\F^2/pq$. Recall the sample  covariance estimator is
\bee\nonumber
\hat{\M \Sigma} &= \frac{1}{n}\sum_{k = 1}^{n}\{\vecc(\M X_k) - \vecc(\bar{\M X})\}\{\vecc(\M X_k) - \vecc(\bar{\M X})\}^{\T} 
\\
&= \hat{\M\Sigma}_0 - \vecc(\bar{\M X})\vecc(\bar{\M X})^{\T}.
\ee
By definition, we have $\hat{\M \Sigma}_0 - \hat{\M\Sigma} = \vecc(\bar{\M X})\vecc(\bar{\M X})^{\T}$ and,
\bee\label{lm:frate:core1}
\DBCov - \DBCovZ &= \hat{\M \Sigma} \HDBand - \hat{\M \Sigma}_0\HDBand 
\\
&=\{\hat{\M \Sigma} - \hat{\M \Sigma}_0\}\HDBand
\\
&=\vecc(\bar{\M X})\vecc(\bar{\M X})^{\T}\HDBand,
\ee
which implies
\bee\label{l:frate:b1}
\E\|\tilde{\M \Sigma}_{\mathcal{B}}(k_1,k_2)-\DBCovZ\|_\F^2  &= \E\|\vecc(\bar{\M X})\vecc(\bar{\M X})^{\T}\HDBand\|^2_\F 
\\
&\leq pq(2k_1+1)(2k_2+1) \times \max_{1\leq l_1,m_1\leq p \atop 1\leq l_2,m_2\leq q}\E\Big(\big|\bar{x}_{l_1,m_1}\big|^2\big|\bar{ x}_{l_2,m_2}\big|^2\Big) 
\\
&\asymp n^{-4}pqk_1k_2\max_{1\leq l_1,m_1\leq p \atop 1\leq l_2,m_2\leq q}\E\Big\{\Big(\sum_{i = 1}^n x^{(i)}_{l_1,m_1}\Big)^2\Big(\sum_{i = 1}^n x^{(i)}_{l_2,m_2}\Big)^2\Big\}
\\
&=n^{-4}pqk_1k_2\max_{1\leq l_1,m_1\leq p \atop 1\leq l_2,m_2\leq q}\E\Big\{\sum_{1\leq i_1,i_2,i_3,i_4\leq n}x^{(i_1)}_{l_1,m_1}x^{(i_2)}_{l_1,m_1}x^{(i_3)}_{l_2,m_2}x^{(i_4)}_{l_2,m_2}\Big\}.
\ee
Since $\{\M X_i\}_{i = 1}^n$ are i.i.d. and $\E \M X_i = \M{0}$, for any $l_1,l_2,m_1,m_2$, 
\bee\nonumber
&\E\Big\{\sum_{1\leq i_1,i_2,i_3,i_4\leq n}x^{(i_1)}_{l_1,m_1}x^{(i_2)}_{l_1,m_1}x^{(i_3)}_{l_2,m_2}x^{(i_4)}_{l_2,m_2}\Big\} 
\\
&= \sum_{1\leq i_1\neq i_3\leq n\atop i_1 = i_2, i_3 = i_4}\E|x^{(i_1)}_{l_1,m_1}|^2\E| x^{(i_3)}_{l_2,m_2}|^2 + \sum_{1\leq i_1\leq n\atop i_1 = i_2 = i_3 = i_4}\E|x^{(i_1)}_{l_1,m_1} x^{(i_1)}_{l_2,m_2}|^2
+\sum_{1\leq i_1\neq i_2\leq n\atop i_1 = i_3,i_2 = i_4}\E x^{(i_1)}_{l_1,m_1}x^{(i_1)}_{l_2,m_2} \times\E x^{(i_2)}_{l_1,m_1}x^{(i_2)}_{l_2,m_2}.
\ee
%where all the other cross terms are $0$ after taking expectation by the independence, except for the three terms in the right-hand side of the equation. 
Then we can apply Cauchy-Schwarz inequality and show
\bee\nonumber
&\max_{1\leq l_1,m_1\leq p \atop 1\leq l_2,m_2\leq q}\E\Big\{\sum_{1\leq i_1,i_2,i_3,i_4\leq n}x^{(i_1)}_{l_1,m_1}x^{(i_2)}_{l_1,m_1}x^{(i_3)}_{l_2,m_2}x^{(i_4)}_{l_2,m_2}\Big\}
\\
&\leq\max_{1\leq l_1,m_1\leq p \atop 1\leq l_2,m_2\leq q} \sum_{1\leq i_1\neq i_3\leq n\atop i_1 = i_2, i_3 = i_4}\E|x^{(i_1)}_{l_1,m_1}|^2\E| x^{(i_3)}_{l_2,m_2}|^2 + \max_{1\leq l_1,m_1\leq p \atop 1\leq l_2,m_2\leq q}\sum_{1\leq i_1\leq n\atop i_1 = i_2 = i_3 = i_4}\E|x^{(i_1)}_{l_1,m_1} x^{(i_1)}_{l_2,m_2}|^2
\\
&+\max_{1\leq l_1,m_1\leq p \atop 1\leq l_2,m_2\leq q}\sum_{1\leq i_1\neq i_2\leq n\atop i_1 = i_3,i_2 = i_4}\E x^{(i_1)}_{l_1,m_1}x^{(i_1)}_{l_2,m_2} \times\E x^{(i_2)}_{l_1,m_1}x^{(i_2)}_{l_2,m_2}
\\
&\leq n^2\times \underbrace{\max_{l_1,m_1,l_2,m_2}\Big\{\E|x^{(i_1)}_{l_1,m_1}|^4\Big\}^{1/2}}_{\leq \sqrt{M}}\times \underbrace{\max_{l_1,m_1,l_2,m_2}\Big\{\E|x^{(i_3)}_{l_1,m_1}|^4\Big\}^{1/2}}_{\leq \sqrt{M}}
+ n\times\underbrace{\max_{l_1,l_2,m_1,m_2}{\E|x^{(i_1)}_{l_1,m_1} x^{(i_1)}_{l_2,m_2}}|^2}_{\leq M}
\\
&+ n^2 \underbrace{\Big\{\max_{l_1,l_2,m_1,m_2}\E |x^{(i_1)}_{l_1,m_1}x^{(i_1)}_{l_2,m_2}|^2\Big\}^{1/2}}_{\leq \sqrt{M}} \times\underbrace{\Big\{\max_{l_1,l_2,m_1,m_2}\E |x^{(i_2)}_{l_1,m_1}x^{(i_2)}_{l_2,m_2}|^2\Big\}^{1/2}}_{\leq \sqrt{M}}
\\
&\precsim n^2,
\ee
where the last inequality holds by the finite fourth moment condition of $\E(|x^{(i)}_{l_1,l_2}\cdot x^{(i)}_{m_1,m_2}|^2)$. Finally we obtain
\bee\label{l:fbound:2}
\frac{1}{pq}\E\|\tilde{\M \Sigma}_{\mathcal{B}}(k_1,k_2)-\DBCovZ\|_\F^2 \precsim \frac{k_1k_2}{n^2}.
\ee
Next, we  bound $\frac{1}{pq}\E\|\DBCovZ-\M\Sigma_2^{*,\mathcal{B}}(k_2)\otimes \M\Sigma_1^{*,\mathcal{B}}(k_1)\|^2_\F$. Define $d(\boldsymbol\beta_1) = |l_1-m_1|,d(\boldsymbol\beta_2) = |l_2 - m_2|$, by (\ref{T3:var}), we have
\bee\label{l:frate:5}
&\frac{1}{pq}\mathbb{E}\|\tilde{\M{\Sigma}}_{0,\MB}(k_1,k_2) - \M\Sigma_2^{*,\mathcal{B}}(k_2)\otimes \M\Sigma_1^{*,\mathcal{B}}(k_1)\|_\F^2
\\
&= \frac{1}{pq}\sum_{d(\boldsymbol\beta_1)\leq k_1,\atop d(\boldsymbol\beta_2)\leq k_2}\E(\tilde{\sigma}^{\mathcal{B}}_{\boldsymbol\beta_1,\boldsymbol\beta_2} - \sigma_{\boldsymbol\beta_1,\boldsymbol\beta_2})^2 
\\
&\leq \frac{1}{pq}\sum_{d(\boldsymbol\beta_1)\leq k_1,\atop d(\boldsymbol\beta_2)\leq k_2}\big(w_{\boldsymbol\beta_1,\boldsymbol\beta_2}^{\mathcal{B}}\big)^2\cdot\frac{M}{n} +  \frac{1}{pq}\sum_{d(\boldsymbol\beta_1)\leq k_1,\atop d(\boldsymbol\beta_2)\leq k_2}\big(\underbrace{w^{\mathcal{B}}_{\boldsymbol\beta_1,\boldsymbol\beta_2} - 1}_{:=0 \text{ by definition}}\big)^2\sigma_{\boldsymbol\beta_1,\boldsymbol\beta_2}^2
\\
&= \frac{1}{pq}\sum_{d(\boldsymbol\beta_1)\leq k_1,\atop d(\boldsymbol\beta_2) \leq k_2}\frac{M}{n}  \\
&\precsim \frac{k_1k_2}{n},
\ee
which implies 
\bee\label{l:fbound:3}
\frac{1}{pq}\mathbb{E}\|\tilde{\M{\Sigma}}_{0,\MB}(k_1,k_2) - \M\Sigma_2^{*,\mathcal{B}}(k_2)\otimes \M\Sigma_1^{*,\mathcal{B}}(k_1)\|_\F^2\precsim \frac{k_1k_2}{n}.
\ee
Combining \eqref{T_3_1}, \eqref{l:fbound:2} and \eqref{l:fbound:3}, one has
\bee\label{l:f:final:1}
\frac{1}{pq}\E\|{\tilde{\M \Sigma}_{\mathcal{B}}(k_1,k_2) - \M\Sigma^{*,\mathcal{B}}_2(k_2)\otimes \M\Sigma^{*,\mathcal{B}}_1(k_1)\|_{\F}^2}{} \precsim \frac{k_1k_2}{n}.
\ee
\\
\par
Next we derive bound \eqref{l:frate:res1} for the proposed tapering estimator. Same argument as \eqref{l:frate:1} and \eqref{T_3_1} yields
\bee\label{l:frate:t1}
\text{(i). }&\mathbb{E}\big(\tilde{\sigma}^{\mathcal{T}}_{\boldsymbol\beta_1,\boldsymbol\beta_2} - w^{\mathcal{T}}_{\boldsymbol\beta_1,\boldsymbol\beta_2}\cdot\sigma_{\boldsymbol\beta_1,\boldsymbol\beta_2}\big)^2 \precsim {\big(w^{\mathcal{T}}_{\boldsymbol\beta_1,\boldsymbol\beta_2}\big)}^2 \cdot\frac{M}{n};
\\
\text{(ii). }&\frac{1}{pq}\E\|{\tilde{\M \Sigma}_{\mathcal{T}}(k_1,k_2) - \M\Sigma^{*,\mathcal{T}}_2(k_2)\otimes \M\Sigma^{*,\mathcal{T}}_1(k_1)\|_{\F}^2} 
\\
&\precsim\frac{1}{pq}\E\|\tilde{\M \Sigma}_{\mathcal{T}}(k_1,k_2)-\DTCovZ\|_\F^2 
+\frac{1}{pq}\E\|\DTCovZ-\M\Sigma_2^{*,\mathcal{T}}(k_2)\otimes \M\Sigma_1^{*,\mathcal{T}}(k_1)\|^2_\F.
\ee 
Similar to \eqref{l:frate:b1}, by definitions of $T_k(\cdot)$ and $B_k(\cdot)$,
\bee\label{l:frate:t2}
&\frac{1}{pq}\E\|\tilde{\M \Sigma}_{\mathcal{T}}(k_1,k_2)-\DTCovZ\|_\F^2 
\\
& =\frac{1}{pq}\E\|\vecc(\bar{\M X})\vecc(\bar{\M X})^{\T}\HDTaper\|^2_\F
\\
&\leq \frac{1}{pq}\E\|\vecc(\bar{\M X})\vecc(\bar{\M X})^{\T}\circ\{B_{k_2}(\M 1_q)\otimes B_{k_1}(\M 1_p)\}\|^2_\F
\\
&\precsim \frac{k_1k_2}{n^2},
\ee
where the last inequality holds by \eqref{l:fbound:2}. For $\E\|\DTCovZ-\M\Sigma_2^{*,\mathcal{T}}(k_2)\otimes \M\Sigma_1^{*,\mathcal{T}}(k_1)\|^2_\F/pq$, similar to \eqref{l:frate:5}, we have
\bee\label{l:frate:t3}
&\frac{1}{pq}\mathbb{E}\|\tilde{\M{\Sigma}}_{0,\mathcal{T}}(k_1,k_2) - \M\Sigma_2^{*,\mathcal{T}}(k_2)\otimes \M\Sigma_1^{*,\mathcal{T}}(k_1)\|_\F^2
\\
&= \frac{1}{pq}\sum_{d(\boldsymbol\beta_1)\leq k_1,\atop d(\boldsymbol\beta_2)\leq k_2}\E(\tilde{\sigma}^{\mathcal{T}}_{\boldsymbol\beta_1,\boldsymbol\beta_2} - w^{\mathcal{T}}_{\boldsymbol\beta_1,\boldsymbol\beta_2}\cdot\sigma_{\boldsymbol\beta_1,\boldsymbol\beta_2})^2 
\\&\precsim \frac{k_1k_2}{n} .
\ee
Combining \eqref{l:frate:t1}, \eqref{l:frate:t2} and \eqref{l:frate:t3}, we finally derive the bound \eqref{l:frate:res1} by
\bee\nonumber
&\frac{1}{pq}\E{\|\tilde{\M \Sigma}_{\mathcal{\mathcal{T}}}(k_1,k_2) - \M\Sigma^{*,\mathcal{T}}_2(k_2)\otimes \M\Sigma^{*,\mathcal{T}}_1(k_1)\|_{\F}^2}{}\precsim\frac{k_1k_2}{n
}.
\ee
\end{proof}

\begin{lemma}\label{lemma:srate}
Let $\vecc(\M X_1),\vecc(\M  X_2),\cdots,\vecc(\M X_n)$ be i.i.d sub-Gaussian random vectors in $\RR^{pq}$ with true covariance $\M \Sigma^* = \M \Sigma^*_2 \otimes \M \Sigma^*_1$, where $\M \Sigma_1^{*}\in \mathcal{F}(\varepsilon_0, \alpha_1),\M \Sigma_2^{*} \in \mathcal{F}(\varepsilon_0, \alpha_2)$ or $\M \Sigma_1^{*}\in \mathcal{M}(\varepsilon_0, \alpha_1),\M \Sigma_2^{*} \in \mathcal{M}(\varepsilon_0, \alpha_2)$. For $\eta\in\{\mathcal{B},\mathcal{T}\}$, we have
\bee\label{lemma:srate:res}
\frac{1}{pq}\E{\big\|\xi\big\{\tilde{\M \Sigma}_{\mathcal{\eta}}(k_1,k_2)\big\} - \xi\big\{\M\Sigma^{*,\eta}_2(k_2)\otimes \M\Sigma^{*,\eta}_1(k_1)\big\}\big\|_{2}^2}\precsim \begin{cases}
\frac{k_1}{qn} + \frac{k_2}{pn} & {\rm ~if~} pk_1 + qk_2 \precsim n
\\
\frac{pk^2_1}{qn^2} + \frac{qk^2_2}{pn^2} & {\rm ~if~} pk_1 + qk_2 \succ n. 
\end{cases}
\ee 
\end{lemma}
\begin{proof}[Proof of Lemma \ref{lemma:srate}] The notation used in this proof is mainly introduced in Section \ref{notation:main}. The proof borrows ideas from Gaussian chaos concentration   (see e.g., \citet{wagaman2009discovering,tsiligkaridis2013covariance,zhou2014gemini}). Similar to the Proof of Lemma \ref{l:frate}, we only consider the case of 
$\M \Sigma_1^{*}\in \mathcal{F}(\varepsilon_0, \alpha_1),\M \Sigma_2^{*} \in \mathcal{F}(\varepsilon_0, \alpha_2)$, and the bound for $\M \Sigma_1^{*}\in \mathcal{M}(\varepsilon_0, \alpha_1),\M \Sigma_2^{*} \in \mathcal{M}(\varepsilon_0, \alpha_2)$ holds as a special case. 
\par
By \eqref{lm:frate:core1}, we have
$
\DBCov - \DBCovZ =\vecc(\bar{\M X})\vecc(\bar{\M X})^{\T}\HDBand 
$;
$\DTCov - \DTCovZ =\vecc(\bar{\M X})\vecc(\bar{\M X})^{\T}\HDTaper$. Without loss of generality, we assume $\E \vecc(\M X_i) = \bds 0_{pq}$.
\par
We first prove Lemma \ref{lemma:srate} when $\eta = \mathcal{B}$ for the proposed banded estimator. By Lemma \ref{lemma:xi} and \eqref{l:srate:1}, we have
\bee\label{lm:srate:bstart}
&\frac{1}{pq}\E{\big\|\xi\big\{\tilde{\M \Sigma}_{\mathcal{B}}(k_1,k_2)\big\} - \xi\{\M\Sigma^{*,\mathcal{B}}_2(k_2)\otimes \M\Sigma^{*,\mathcal{B}}_1(k_1)\}\big\|_{2}^2}
\\
&\leq\frac{1}{pq}\E\Big[{\|\xi\big\{\tilde{\M \Sigma}_{0,\mathcal{B}}(k_1,k_2)\big\} - \xi\{\M\Sigma^{*,\mathcal{B}}_2(k_2)\otimes \M\Sigma^{*,\mathcal{B}}_1(k_1)\}\|_{2} + \|\xi\big\{\tilde{\M \Sigma}_{\mathcal{B}}(k_1,k_2)\big\} - \xi\big\{\tilde{\M \Sigma}_{0,\mathcal{B}}(k_1,k_2)\big\}\|_{2}\Big]^2}
\\
&\precsim \frac{1}{pq}\E\big\|\xi\big\{\tilde{\M \Sigma}_{0,\mathcal{B}}(k_1,k_2)-\M\Sigma^{*,\mathcal{B}}_2(k_2)\otimes \M\Sigma^{*,\mathcal{B}}_1(k_1)\big\}\big\|_{2}^2
\\
&+ \frac{1}{pq}\E\big\|\xi\big[\vecc(\bar{\M X})\vecc(\bar{\M X})^{\T}\HDBand \big]\big\|_{2} ^ 2.
\ee
%We summarize our proof strategy as follows. We bound $\frac{1}{pq}\E{\|\xi\big\{\tilde{\M \Sigma}_{\mathcal{B}}(k_1,k_2)\big\} - \xi\{\M\Sigma^{*,\mathcal{B}}_2(k_2)\otimes \M\Sigma^{*,\mathcal{B}}_1(k_1)\}\|_{2}^2}$ via the above inequality. 
We will then bound the above two terms on the right-hand side of \eqref{lm:srate:bstart} by the following steps. In Steps 1.1--1.4, we bound the first term $\frac{1}{pq}\E\big\|\xi\big\{\tilde{\M \Sigma}_{0,\mathcal{B}}(k_1,k_2)-\M\Sigma^{*,\mathcal{B}}_2(k_2)\otimes \M\Sigma^{*,\mathcal{B}}_1(k_1)\big\}\big\|_{2}^2$ via an $\epsilon$-net argument that can simultaneously address the effect of the bandable structure of our proposed estimators. In Step 2, we show the bound of the second term $\frac{1}{pq}\E\big\|\xi\big[\vecc(\bar{\M X})\vecc(\bar{\M X})^{\T}\HDBand \big]\big\|_{2} ^ 2$ is in the same asymptotic order with the bound of $\frac{1}{pq}\E\big\|\xi\big\{\tilde{\M \Sigma}_{0,\mathcal{B}}(k_1,k_2)-\M\Sigma^{*,\mathcal{B}}_2(k_2)\otimes \M\Sigma^{*,\mathcal{B}}_1(k_1)\big\}\big\|_{2}^2$. In Step 3, we combine the bounds of two terms together and finally show \eqref{lemma:srate:res} for $\eta = \mathcal{B}$. 
\\
\par
{\noindent\textbf{Step 1.1:}}  For simplicity, denote $\Delta_n^{\mathcal{B}} \equiv \xi\big\{\tilde{\M \Sigma}_{0,\mathcal{B}}(k_1,k_2)-\M\Sigma^{*,\mathcal{B}}_2(k_2)\otimes \M\Sigma^{*,\mathcal{B}}_1(k_1)\big\}, \hat{\M \Sigma}_{i,0} \equiv \vecc(\M X_i)\vecc(\M X_i)^\T$. Note that $\M\Sigma^{*,\mathcal{B}}_2(k_2)\otimes \M\Sigma^{*,\mathcal{B}}_1(k_1) = \E\tilde{\M \Sigma}_{0,\mathcal{B}}(k_1,k_2)$, $\tilde{\M \Sigma}_{0,\mathcal{B}}(k_1,k_2) = \Big\{\frac{1}{n}\sum_{i = 1}^{n}\hat{\M \Sigma}_{i,0}\Big\} \HDBand = \frac{1}{n}\sum_{i = 1}^n\hat{\M \Sigma}_{i,0}\HDBand$, we have
\bee\nonumber
\Delta_n^{\mathcal{B}} &= \frac{1}{n}\sum_{i = 1}^n\xi\Big[\hat{\M \Sigma}_{i,0}\HDBand - \E\big[\hat{\M \Sigma}_{i,0}\HDBand\big]\Big]
\\
&=\frac{1}{n}\sum_{i = 1}^n\underbrace{\xi\Big[\big\{\hat{\M \Sigma}_{i,0}-\E\hat{\M \Sigma}_{i,0}\big\}\HDBand\Big]}_{I_i^{\mathcal{B}}}.
\ee
Now we first study the term $I_i^{\mathcal{B}}$. The $\hat{\M \Sigma}_{i,0}\in\mathbb{R}^{pq\times pq}$ can be written as
\bee\nonumber
\hat{\M \Sigma}_{i,0} = \begin{pmatrix}
\hat{\M \Sigma}_{i,0}^{(1,1)} & \dots & \hat{\M \Sigma}_{i,0}^{(1,q)}
\\
\vdots & \ddots & \vdots
\\
\hat{\M \Sigma}_{i,0}^{(q,1)} & \dots & \hat{\M \Sigma}_{i,0}^{(q,q)},
\end{pmatrix}
\ee
where $\hat{\M \Sigma}_{i,0}^{(l_2,m_2)}$ is the $l_2m_2$th $p\times p$ sub-block matrix of $\hat{\M \Sigma}_{i,0}$ for $1\leq l_2,m_2 \leq q$. Similar to the proof of Lemma \ref{l:frate}, let $w_{l_2,m_2}^{q,k,\mathcal{B}}$ be the $l_2m_2$th element of the matrix $B_{k}(\textbf{1}_{q})$. We can write the $q^2 \times p^2$ matrix $I_i^{\mathcal{B}}$ as 
\bee\nonumber
I_i^{\mathcal{B}} &= \xi\Big[\big\{\hat{\M \Sigma}_{i,0}-\E\hat{\M \Sigma}_{i,0}\big\}\HDBand\Big]
\\
&=\xi\left(\rule{0cm}{1.5cm}\begin{pmatrix}
\hat{\M \Sigma}_{i,0}^{(1,1)} - \E \hat{\M \Sigma}_{i,0}^{(1,1)} & \dots & \hat{\M \Sigma}_{i,0}^{(1,q)} - \E\hat{\M \Sigma}_{i,0}^{(1,q)}
\\
\vdots & \ddots & \vdots
\\
\hat{\M \Sigma}_{i,0}^{(q,1)} - \E \hat{\M \Sigma}_{i,0}^{(q,1)}& \dots & \hat{\M \Sigma}_{i,0}^{(q,q)} - \E \hat{\M \Sigma}_{i,0}^{(q,q)},
\end{pmatrix}\circ \begin{pmatrix}
B_{k_1}(\M 1_p)\cdot w_{1,1}^{q,k_2,\mathcal{B}} & \dots &  B_{k_1}(\M 1_p) \cdot w_{1,q}^{q,k_2,\mathcal{B}}
\\
\vdots & \ddots & \vdots
\\
 B_{k_1}(\M 1_p)\cdot w_{q,1}^{q,k_2,\mathcal{B}} & \dots &  B_{k_1}(\M 1_p)\cdot w_{q,q}^{q,k_2,\mathcal{B}},
\end{pmatrix}
\right)\rule{0cm}{1.5cm}
\\
&=\xi \left(\rule{0cm}{1.5cm}\begin{pmatrix}
w_{1,1}^{q,k_2,\mathcal{B}}\cdot(\hat{\M \Sigma}_{i,0}^{(1,1)} - \E\hat{\M \Sigma}_{i,0}^{(1,1)})\circ B_{k_1}(\M 1_p) & \dots & w_{1,q}^{q,k_2,\mathcal{B}}\cdot(\hat{\M \Sigma}_{i,0}^{(1,q)} - \E \hat{\M \Sigma}_{i,0}^{(1,q)})\circ B_{k_1}(\M 1_p)
\\
\vdots & \ddots & \vdots
\\
 w_{q,1}^{q,k_2,\mathcal{B}}\cdot(\hat{\M \Sigma}_{i,0}^{(q,1)} - \E \hat{\M \Sigma}_{i,0}^{(q,1)})\circ B_{k_1}(\M 1_p) & \dots &  w_{q,q}^{q,k_2,\mathcal{B}}\cdot(\hat{\M \Sigma}_{i,0}^{(q,q)} - \E \hat{\M \Sigma}_{i,0}^{(q,q)})\circ B_{k_1}(\M 1_p)
\end{pmatrix}
\right)\rule{0cm}{1.5cm}
\\
&=\begin{pmatrix}
\vecc^\T\big[\big\{\hat{\M \Sigma}_{i,0}^{(1,1)} - \E \hat{\M \Sigma}_{i,0}^{(1,1)}\big\}\circ B_{k_1}(\bds 1_p)\big]\cdot w_{1,1}^{q,k_2,\mathcal{B}}
\\
\vdots
\\
\vecc^\T\big[\big\{\hat{\M \Sigma}_{i,0}^{(q,1)} - \E \hat{\M \Sigma}_{i,0}^{(q,1)}\big\}\circ B_{k_1}(\bds 1_p)\big]\cdot w_{q,1}^{q,k_2,\mathcal{B}}
\\
\vdots
\\
\vecc^\T\big[\big\{\hat{\M \Sigma}_{i,0}^{(q,q)} - \E \hat{\M \Sigma}_{i,0}^{(q,q)}\big\}\circ B_{k_1}(\bds 1_p)\big]\cdot w_{q,q}^{q,k_2,\mathcal{B}}
\end{pmatrix},
\ee
where the last equality is by the definition of $\xi(\cdot)$ in \eqref{xi:fun}. Then let $\bds v, \bds u$ be any vectors such that $\bds v\in \mathcal{U}_{q^2}, \bds u\in \mathcal{U}_{p^2}$, by the property of spectral norm we have 
\bee\label{l:srate:vdu:0}
&\big\|\xi\big\{\tilde{\M \Sigma}_{0,\mathcal{B}}(k_1,k_2)-\M\Sigma^{*,\mathcal{B}}_2(k_2)\otimes \M\Sigma^{*,\mathcal{B}}_1(k_1)\big\}\big\|_{2} 
\\
&= \sup_{\bds v\in \mathcal{U}_{q^2}, \bds u\in \mathcal{U}_{p^2}}\Big|\bds v^\T\xi\big\{\tilde{\M \Sigma}_{0,\mathcal{B}}(k_1,k_2)-\M\Sigma^{*,\mathcal{B}}_2(k_2)\otimes \M\Sigma^{*,\mathcal{B}}_1(k_1)\big\}\bds u\Big|
\\
&=\sup_{\bds v\in \mathcal{U}_{q^2}, \bds u\in \mathcal{U}_{p^2}}|\bds v^\T\Delta_n^{\mathcal{B}}\bds u|, 
\ee
where
\bee\label{l:srate:vdu}
\bds v^\T\Delta_n^{\mathcal{B}}\bds u &= \frac{1}{n}\sum_{i = 1}^n\bds v^\T I^{\mathcal{B}}_i \bds u 
\\
&= \frac{1}{n}\sum_{i = 1}^n \bds v^\T \begin{pmatrix}
\vecc^\T\big[\big\{\hat{\M \Sigma}_{i,0}^{(1,1)} - \E \hat{\M \Sigma}_{i,0}^{(1,1)}\big\}\circ B_{k_1}(\bds 1_p)\big]\cdot w_{1,1}^{q,k_2,\mathcal{B}}
\\
\vdots
\\
\vecc^\T\big[\big\{\hat{\M \Sigma}_{i,0}^{(q,1)} - \E \hat{\M \Sigma}_{i,0}^{(q,1)}\big\}\circ B_{k_1}(\bds 1_p)\big]\cdot w_{q,1}^{q,k_2,\mathcal{B}}
\\
\vdots
\\
\vecc^\T\big[\big\{\hat{\M \Sigma}_{i,0}^{(q,q)} - \E \hat{\M \Sigma}_{i,0}^{(q,q)}\big\}\circ B_{k_1}(\bds 1_p)\big]\cdot w_{q,q}^{q,k_2,\mathcal{B}}
\end{pmatrix}\bds u
\\
&= \frac{1}{n}\sum_{i = 1}^n \bds v^\T \begin{pmatrix}
w_{1,1}^{q,k_2,\mathcal{B}}\cdot\Big[\vecc^\T\big\{\hat{\M \Sigma}_{i,0}^{(1,1)} - \E \hat{\M \Sigma}_{i,0}^{(1,1)}\big\}\circ \vecc^\T\big\{B_{k_1}(\bds 1_p)\big\}\Big]\cdot\bds u
\\
\vdots
\\
w_{q,q}^{q,k_2,\mathcal{B}}\cdot\Big[\vecc^\T\big\{\hat{\M \Sigma}_{i,0}^{(q,q)} - \E \hat{\M \Sigma}_{i,0}^{(q,q)}\big\}\circ\vecc^\T\big\{ B_{k_1}(\bds 1_p)\big\}\Big]\cdot\bds u
\end{pmatrix}
\\
& = \frac{1}{n}\sum_{i = 1}^n \bds v^\T \begin{pmatrix}
w_{q,q}^{q,k_2,\mathcal{B}}\cdot\vecc^\T\big\{\hat{\M \Sigma}_{i,0}^{(1,1)} - \E \hat{\M \Sigma}_{i,0}^{(1,1)}\big\}\cdot \Big[\bds u\circ \vecc\big\{B_{k_1}(\bds 1_p)\big\}\Big]
\\
\vdots
\\
w_{q,q}^{q,k_2,\mathcal{B}}\cdot\vecc^\T\big\{\hat{\M \Sigma}_{i,0}^{(q,q)} - \E \hat{\M \Sigma}_{i,0}^{(q,q)}\big\}\cdot \Big[\bds u\circ\vecc\big\{ B_{k_1}(\bds 1_p)\big\}\Big]
\end{pmatrix}
\\
& = \frac{1}{n}\sum_{i = 1}^n \bds v^\T \begin{pmatrix}
w_{1,1}^{q,k_2,\mathcal{B}}\cdot\vecc^\T\big\{\hat{\M \Sigma}_{i,0}^{(1,1)} - \E \hat{\M \Sigma}_{i,0}^{(1,1)}\big\}
\\
\vdots
\\
w_{q,q}^{q,k_2,\mathcal{B}}\cdot\vecc^\T\big\{\hat{\M \Sigma}_{i,0}^{(q,q)} - \E \hat{\M \Sigma}_{i,0}^{(q,q)}\big\}
\end{pmatrix}\times\Big[\bds u\circ\vecc\big\{ B_{k_1}(\bds 1_p)\big\}\Big].
\ee
The fourth equality above holds by the fact that for any vectors $\bds a_1,\bds a_2, \bds a_3 \in \RR^d$, $(\bds a_1 \circ \bds a_2)^\T\cdot \bds a_3 = \sum_{i = 1}^d a_i^{(1)}a_i^{(2)}a_i^{(3)} = \bds a_1^\T\cdot(\bds a_3 \circ \bds a_2)$, where $a_i^{(u)}$ is the $i$th coordinate of $\bds a_u$. Similar argument can also be applied on $\bds v^\T$ side in \eqref{l:srate:vdu} and one has
\bee\label{l:srate:vdu2}
\bds v^\T\Delta_n^{\mathcal{B}}\bds u &= \frac{1}{n}{\mathlarger{\mathlarger{\sum}}}_{i = 1}^n\Big[\bds v\circ\vecc\big\{ B_{k_2}(\bds 1_q)\big\}\Big]^\T\begin{pmatrix}
\vecc^\T\big\{\hat{\M \Sigma}_{i,0}^{(1,1)} - \E \hat{\M \Sigma}_{i,0}^{(1,1)}\big\}
\\
\vdots
\\
\vecc^\T\big\{\hat{\M \Sigma}_{i,0}^{(q,q)} - \E \hat{\M \Sigma}_{i,0}^{(q,q)}\big\}
\end{pmatrix}\Big[\bds u\circ\vecc\big\{ B_{k_1}(\bds 1_p)\big\}\Big]
\\
&=\Big[\bds v\circ\vecc\big\{ B_{k_2}(\bds 1_q)\big\}\Big]^\T\times\Big[\frac{1}{n}\sum_{i = 1}^n\xi\big\{\hat{\M \Sigma}_{i,0} - \E\hat{\M \Sigma}_{i,0}\big\}\Big]\times\Big[\bds u\circ\vecc\big\{ B_{k_1}(\bds 1_p)\big\}\Big]
\\
&=\big[\bds v\circ\vecc\big\{ B_{k_2}(\bds 1_q)\big\}\big]^\T\times \xi\big\{\hat{\M \Sigma}_0 - \E\hat{\M \Sigma}_0\big\}\times\big[\bds u\circ\vecc\big\{ B_{k_1}(\bds 1_p)\big\}\big].
\ee
Now we define the banded unit spheres from $\mathcal{U}_{q^2}, \mathcal{U}_{p^2}$ such that 
\bee\label{def:Ub}
&\mathcal{U}_{p^2}^{\mathcal{B}}(k_1)=\Big\{\bds u \mid \|\bds u\| = 1; \bds u \circ \vecc\big\{ B_{k_1}(\bds 1_p)\big\} = \bds u, \bds u\in \RR^{p^2}\Big\};
\\
&\mathcal{U}_{q^2}^{\mathcal{B}}(k_2)=\Big\{\bds v \mid \|\bds v\| = 1; \bds v \circ \vecc\big\{ B_{k_2}(\bds 1_q)\big\} = \bds v, \bds v\in \RR^{q^2}\Big\}.
\ee
It is easy to check $\mathcal{U}_{p^2}^{\mathcal{B}}(k_1) \subseteq \mathcal{U}_{p^2}, \mathcal{U}_{q^2}^{\mathcal{B}}(k_2) \subseteq \mathcal{U}_{q^2}$. In other words, $\mathcal{U}_{p^2}^{\mathcal{B}}(k_1)$ contains the unit vectors in $\RR^{p^2}$, whose coordinates are non-zero only if the corresponding coordinates in $\vecc\big\{ B_{k_1}(\bds 1_p)\big\}$ are $1$. A symmetric result also holds for $\mathcal{U}_{q^2}^{\mathcal{B}}(k_2)$.
\par
Thus for any $\bds u \in \mathcal{U}_{p^2}, \bds v\in \mathcal{U}_{q^2}$, we can define $\tilde{\bds u} = \bds u\circ \vecc\{B_{k_1}(\bds 1_p)\}/\|\bds u\circ \vecc\{B_{k_1}(\bds 1_p)\}\|$ and $\tilde{\bds v} = \bds v\circ \vecc\{B_{k_2}(\bds 1_q)\}/\|\bds v\circ \vecc\{B_{k_2}(\bds 1_q)\}\|$ and it is easy to check that $\tilde{\bds u}\in \mathcal{U}_{p^2}^{\mathcal{B}}(k_1),\tilde{\bds v} \in \mathcal{U}_{q^2}^{\mathcal{B}}(k_2)$. By \eqref{l:srate:vdu2}, we have
\bee\label{l:srate:vdu:3}
|\bds v^\T&\Delta_n^{\mathcal{B}}\bds u| = \Big|\big[\bds v\circ\vecc\big\{ B_{k_2}(\bds 1_q)\big\}\big]^\T\times \xi\big\{\hat{\M \Sigma}_0 - \E\hat{\M \Sigma}_0\big\}\times\big[\bds u\circ\vecc\big\{ B_{k_1}(\bds 1_p)\big\}\big]\Big|
\ee
and 
\bee\label{l:srate:vdu:4}
&|\tilde{\bds v}^\T\Delta_n^{\mathcal{B}}\tilde{\bds u}| 
\\
&= \Big|\big[\tilde{\bds v}\circ\vecc\big\{ B_{k_2}(\bds 1_q)\big\}\big]^\T\times \xi\big\{\hat{\M \Sigma}_0 - \E\hat{\M \Sigma}_0\big\}\times\big[\tilde{\bds u}\circ\vecc\big\{ B_{k_1}(\bds 1_p)\big\}\big]\Big|
\\
&=\frac{1}{\|\bds u\circ \vecc\{B_{k_1}(\bds 1_p)\}\|\|\bds v\circ \vecc\{B_{k_2}(\bds 1_q)\}\|} \Big|\big[\bds v\circ\vecc\big\{ B_{k_2}(\bds 1_q)\big\}\big]^\T\times \xi\big\{\hat{\M \Sigma}_0 - \E\hat{\M \Sigma}_0\big\}\times\big[\bds u\circ\vecc\big\{ B_{k_1}(\bds 1_p)\big\}\big]\Big|
\\
&\geq |\bds v^\T\Delta_n^{\mathcal{B}}\bds u|,
\ee
by the fact that $\|\bds u\circ \vecc\{B_{k_1}(\bds 1_p)\}\|^{-1}, \|\bds v\circ \vecc\{B_{k_2}(\bds 1_q)\}\|^{-1}\geq 1$ as $\bds u \circ \vecc\{B_{k_1}(\bds 1_p)\}$ and $ \bds v \circ \vecc\{B_{k_2}(\bds 1_q)\}$ are shrinkage vectors from unit vectors. Since $\bds u, \bds v$ are arbitrary, by \eqref{l:srate:vdu:0}, \eqref{l:srate:vdu:3} and \eqref{l:srate:vdu:4}, we have
\bee\label{l:srate:up}
\big\|\xi\big\{\tilde{\M \Sigma}_{0,\mathcal{B}}(k_1,k_2)-\M\Sigma^{*,\mathcal{B}}_2(k_2)\otimes \M\Sigma^{*,\mathcal{B}}_1(k_1)\big\}\big\|_{2} &= \sup_{\bds v\in \mathcal{U}_{q^2}, \bds u\in \mathcal{U}_{p^2}}|\bds v^\T\Delta_n^{\mathcal{B}}\bds u| 
\\
&\leq \sup_{{\bds v}^*\in \mathcal{U}^{\mathcal{B}}_{q^2}(k_2), \bds u^*\in \mathcal{U}^{\mathcal{B}}_{p^2}(k_1)}|{\bds v^*}^\T\Delta_n^{\mathcal{B}}\bds u^*|.
\ee
On the other hand, since $\mathcal{U}_{p^2}^{\mathcal{B}}(k_1) \subseteq \mathcal{U}_{p^2}, \mathcal{U}_{q^2}^{\mathcal{B}}(k_2) \subseteq \mathcal{U}_{q^2}$, we have \bee\label{l:srate:low}\sup_{\bds v\in \mathcal{U}_{q^2}, \bds u\in \mathcal{U}_{p^2}}|\bds v^\T\Delta_n^{\mathcal{B}}\bds u| \geq \sup_{{\bds v}^*\in \mathcal{U}^{\mathcal{B}}_{q^2}(k_2), \bds u^*\in \mathcal{U}^{\mathcal{B}}_{p^2}(k_1)}|{\bds v^*}^\T\Delta_n^{\mathcal{B}}\bds u^*|,\ee
and thus by \eqref{l:srate:up}, \eqref{l:srate:low} and \eqref{l:srate:vdu2}
\bee\label{l:srate:vdu:final}
\big\|&\xi\big\{\tilde{\M \Sigma}_{0,\mathcal{B}}(k_1,k_2)-\M\Sigma^{*,\mathcal{B}}_2(k_2)\otimes \M\Sigma^{*,\mathcal{B}}_1(k_1)\big\}\big\|_{2} 
\\
&= \sup_{{\bds v}^*\in \mathcal{U}^{\mathcal{B}}_{q^2}(k_2), \bds u^*\in \mathcal{U}^{\mathcal{B}}_{p^2}(k_1)}|{\bds v^*}^\T\Delta_n^{\mathcal{B}}\bds u^*| 
\\
&= \sup_{{\bds v}^*\in \mathcal{U}^{\mathcal{B}}_{q^2}(k_2), \bds u^*\in \mathcal{U}^{\mathcal{B}}_{p^2}(k_1)}\Big|\big[\bds v^*\circ\vecc\big\{ B_{k_2}(\bds 1_q)\big\}\big]^\T\times \xi\big\{\hat{\M \Sigma}_0 - \E\hat{\M \Sigma}_0\big\}\times\big[\bds u^*\circ\vecc\big\{ B_{k_1}(\bds 1_p)\big\}\big]\Big|
\\
&= \sup_{{\bds v}^*\in \mathcal{U}^{\mathcal{B}}_{q^2}(k_2), \bds u^*\in \mathcal{U}^{\mathcal{B}}_{p^2}(k_1)}\Big|\bds v^*\cdot\xi\big\{\hat{\M \Sigma}_0 - \E\hat{\M \Sigma}_0\big\}\cdot{\bds u^*}^\T \Big|.
\ee
Now we focus on $\bds v^*\cdot\xi\big\{\hat{\M \Sigma}_0 - \E\hat{\M \Sigma}_0\big\}\cdot{\bds u^*}^\T$ with ${\bds v}^*\in \mathcal{U}^{\mathcal{B}}_{q^2}(k_2), \bds u^*\in \mathcal{U}^{\mathcal{B}}_{p^2}(k_1)$. Similar to the proof of Theorem \ref{l:frate}, let $\boldsymbol\beta_1 = (l_1,m_1)^\T,\boldsymbol\beta_2 = (l_2,m_2)^\T$ with $1\leq l_1,m_1\leq p, 1\leq l_2,m_2 \leq q$. Denote the $(m_1-1)*p+l_1$th coordinate of $\bds u^*$ as $u^*_{\bds \beta_1}$; $(m_2 - 1)*q + l_2$th coordinate of $\bds v^*$ by $v^*_{\bds \beta_2}$; and the $(l_2-1)*q + l_1, (m_2 - 1)*q + m_1$th element of $\hat{\M \Sigma}_{i,0}$ by $\hat{\sigma}^{(i,0)}_{\bds \beta_1,\bds \beta_2}$. We also define $\bds{\mathscr{U}}^*(\bds u^*) \in \mathbb{R}^{p\times p}, \bds{\mathscr{V}}^*(\bds v^*)\in\mathbb{R}^{q\times q}$ such that $\vecc\{\bds{\mathscr{U}}^*(\bds u^*)\} = \bds u^*, \vecc\{\bds{\mathscr{V}}^*(\bds v^*)\} = \bds v^*$. It is easy to see that for certain $\bds u^*$ and $\bds v^*$, $\bds{\mathscr{U}}^*(\bds u^*)$ and $\bds{\mathscr{V}}^*(\bds v^*)$ are unique. Then similar to \eqref{l:srate:vdu2}, one has
\bee\label{l:srate:vsu:1}
&{\bds v^*}^\T\cdot\xi\big\{\hat{\M \Sigma}_0 - \E\hat{\M \Sigma}_0\big\}\cdot{\bds u^*} 
\\
&= \frac{1}{n}{{{\sum}}}_{i = 1}^n{\bds v^*}^\T\begin{pmatrix}
\vecc^\T\big\{\hat{\M \Sigma}_{i,0}^{(1,1)} - \E \hat{\M \Sigma}_{i,0}^{(1,1)}\big\}
\\
\vdots
\\
\vecc^\T\big\{\hat{\M \Sigma}_{i,0}^{(q,q)} - \E \hat{\M \Sigma}_{i,0}^{(q,q)}\big\}
\end{pmatrix}\bds u^*
\\
&=\frac{1}{n}\sum_{i = 1}^n \sum_{\bds \beta_2\in [q]\times[q]} v^*_{\bds \beta_2}\vecc\{\hat{\M \Sigma}_{i,0}^{\bds \beta_2} - \E\hat{\M \Sigma}_{i,0}^{\bds \beta_2}\}\bds u^*
\\
&=\frac{1}{n}\sum_{i = 1}^n \sum_{\bds \beta_1\in[p]\times [p]\atop\bds \beta_2\in [q]\times[q]} v^*_{\bds \beta_2}u^*_{\bds \beta_1}\big(\hat{\sigma}^{(i,0)}_{\bds \beta_1,\bds \beta_2} - \E\hat{\sigma}^{(i,0)}_{\bds \beta_1,\bds \beta_2}\big)
\\
&=\frac{1}{n}\sum_{i = 1}^n \Big[\sum_{\bds \beta_1\in[p]\times [p]\atop\bds \beta_2\in [q]\times[q]} v^*_{\bds \beta_2}u^*_{\bds \beta_1}\hat{\sigma}^{(i,0)}_{\bds \beta_1,\bds \beta_2}\Big] - \Big[\E\Big\{\sum_{\bds \beta_1\in[p]\times [p]\atop\bds \beta_2\in [q]\times[q]}v^*_{\bds \beta_2}u^*_{\bds \beta_1}\hat{\sigma}^{(i,0)}_{\bds \beta_1,\bds \beta_2}\Big\}\Big].
\ee
Denote $\vecc(\M X_i)_{(l_1,l_2)}$ the $(l_2 - 1)*p + l_1$th element in $\vecc(\M X_i)$, i.e., $\vecc(\M X_i)_{(l_1,l_2)} = x^{(i)}_{l_1,l_2}$. For the first term on the right-hand side,  we have
\bee\label{l:srate:vsu:2}
&\sum_{\bds \beta_1 = (l_1,m_1)^\T\in[p]\times [p]\atop\bds \beta_2 = (l_2,m_2)^\T\in [q]\times[q]} v^*_{\bds \beta_2}u^*_{\bds \beta_1}\hat{\sigma}^{(i,0)}_{\bds \beta_1,\bds \beta_2} 
\\
&=\sum_{\bds \beta_1 = (l_1,m_1)^\T\in[p]\times [p]\atop\bds \beta_2 = (l_2,m_2)^\T\in [q]\times[q]} v^*_{\bds \beta_2}u^*_{\bds \beta_1}x^{(i)}_{l_1,l_2}x^{(i)}_{m_1,m_2}
\\
&= \sum_{\bds \beta_1 = (l_1,m_1)^\T\in[p]\times [p]\atop\bds \beta_2 = (l_2,m_2)^\T\in [q]\times[q]} v^*_{\bds \beta_2}u^*_{\bds \beta_1}\vecc(\M X_i)_{(l_1,l_2)}\vecc(\M X_i)_{(m_1,m_2)}
\\
&= \sum_{\bds \beta_1 = (l_1,m_1)^\T\in[p]\times [p]\atop\bds \beta_2 = (l_2,m_2)^\T\in [q]\times[q]} \vecc(\M X_i)_{(l_1,l_2)}\times \Big[\mathscr{V}^*({\bds v}^*)\otimes\mathscr{U}({\bds u^*})\Big]_{\{(l_1,m_1),(l_2,m_2)\}}\times\vecc(\M X_i)_{(m_1,m_2)}
\\
&=\vecc(\M X_i)^\T\times\Big[\mathscr{V}^*({\bds v}^*)\otimes\mathscr{U}({\bds u^*})\Big]\times\vecc(\M X_i).
\ee
Combining \eqref{l:srate:vsu:1}, \eqref{l:srate:vsu:2} and \eqref{l:srate:vdu:final}, we finally show
\bee\label{l:srate:back}
&\big({\bds v^*}\big)^\T\cdot\xi\big\{\hat{\M \Sigma}_0 - \E\hat{\M \Sigma}_0\big\}\cdot{\bds u^*}
\\
&= \frac{1}{n}\sum_{i = 1}^n\underbrace{\vecc(\M X_i)^\T\cdot\mathscr{V}^*({\bds v}^*)\otimes\mathscr{U}({\bds u^*})\cdot\vecc(\M X_i) - \E\Big\{\vecc(\M X_i)^\T\cdot\mathscr{V}^*({\bds v}^*)\otimes\mathscr{U}({\bds u^*})\cdot\vecc(\M X_i)\Big\}}_{\equiv I_i(\bds u^*,\bds v^*)}
\ee
and 
\bee\label{step1:final}
\big\|&\xi\big\{\tilde{\M \Sigma}_{0,\mathcal{B}}(k_1,k_2)-\M\Sigma^{*,\mathcal{B}}_2(k_2)\otimes \M\Sigma^{*,\mathcal{B}}_1(k_1)\big\}\big\|_{2}  
\\
&= \sup_{{{\bds v}^*}^\T \in \mathcal{U}^{\mathcal{B}}_{q^2}(k_2), \bds u^*\in \mathcal{U}^{\mathcal{B}}_{p^2}(k_1)}\Big|\bds v^*\cdot\xi\big\{\hat{\M \Sigma}_0 - \E\hat{\M \Sigma}_0\big\}\cdot{\bds u^*}\Big|
\\
&=\sup_{{\bds v}^*\in \mathcal{U}^{\mathcal{B}}_{q^2}(k_2), \bds u^*\in \mathcal{U}^{\mathcal{B}}_{p^2}(k_1)}\Bigg|\frac{1}{n}\sum_{i = 1}^n\vecc(\M X_i)^\T\cdot\mathscr{V}^*({\bds v}^*)\otimes\mathscr{U}({\bds u^*})\cdot\vecc(\M X_i) 
\\
&- \E\Big\{\vecc(\M X_i)^\T\cdot\mathscr{V}^*({\bds v}^*)\otimes\mathscr{U}({\bds u^*})\cdot\vecc(\M X_i)\Big\}\Bigg| 
\\
&= \sup_{{\bds v}^*\in \mathcal{U}^{\mathcal{B}}_{q^2}(k_2), \bds u^*\in \mathcal{U}^{\mathcal{B}}_{p^2}(k_1)}\Big|\frac{1}{n}\sum _{i = 1}^nI_i(\bds u^*,\bds v^*)\Big|.
\ee
\par
\noindent{\textbf{Step 1.2: }}Next, we use Hanson-Wright type inequality \citep{rudelson2013hanson, zajkowski2020bounds} to study  the concentration of $\frac{1}{n}\sum _{i = 1}^nI_i(\bds u^*,\bds v^*)$ for some fixed ${\bds v}^*\in \mathcal{U}^{\mathcal{B}}_{q^2}(k_2), \bds u^*\in \mathcal{U}^{\mathcal{B}}_{p^2}(k_1)$. From Corollary 2.8 of \citet{zajkowski2020bounds}, for a sub-Gaussian vector $\mathbf v_{s} \in \mathbb{R}^d$ such that $\E\mathbf v_s = \M 0, \cov(\mathbf v_s ) = \M I_d $, and for any $K$  such that $\|\mathbf{v}^\T\mathbf v_s\|_{\psi_2}\leq K$ (see \eqref{def:psi2} for definition of $\|\cdot\|_{\psi_2}$) for any $\|\mathbf{v}\| = 1$, we have
\bee\label{l:srate:hwi}
\p\Big(|\mathbf{v}_s^\T \M M \mathbf{v}_s - \E\mathbf{v}_s^\T \M M \mathbf{v}_s|\geq t\Big)\leq 2\exp\Bigg[-\min\Big\{\frac{t^2}{C^2K^4\|\M M\|_\F^2},\frac{t}{CK^2\|\M M\|_\F}\Big\}\Bigg],
\ee
where $\M M$ is any matrix in $\mathbb{R}^{d\times d}$ and $C$ is a fixed constant. 
\par
In this proof, we take $\M M = (\M \Sigma^*)^{1/2}\mathscr{V}^*({\bds v}^*)\otimes\mathscr{U}^*({\bds u^*})(\M \Sigma^*)^{1/2}$  and $ \mathbf{v}_s = (\M \Sigma^*)^{-1/2}\vecc(\M X_i)$ in \eqref{l:srate:hwi}.  As the maximal eigenvalues of $(\M \Sigma^*_1)^{1/2},(\M \Sigma^*_2)^{1/2}$ are smaller than $1/\epsilon_0^{1/2}$ by \eqref{A1}, and all eigenvalues of these positive definitive matrices are positive, thus $\|(\M \Sigma^*)^{1/2}\|_2 = \lambda_{\max}\{(\M \Sigma^*)^{1/2}\} = \lambda^{1/2}_{\max}\{(\M \Sigma^*)\} =  \lambda^{1/2}_{\max}\{(\M \Sigma^*_1)\}\times \lambda^{1/2}_{\max}\{(\M \Sigma^*_2)\} \leq 1/\varepsilon_0^{}$, by the property of Kronecker product matrix's eigenvalues. Combining with the fact that $\|\mathscr{V}^*({\bds v}^*)\otimes\mathscr{U}({\bds u^*})\|_\F = \|\mathscr{U}({\bds u^*})\|_\F\|\mathscr{V}({\bds v^*})\|_\F = \|\vecc(\bds u^*)\| \times \|\vecc(\bds v^*)\| = 1$ (for first equality, see e.g. \citet{lancaster1972norms}), we have 
\bee\label{l:srate:con1}
\|\M M\|_\F &\leq \|(\M \Sigma^*)^{1/2}\|_2^2\times \|\mathscr{V}^*({\bds v}^*)\otimes\mathscr{U}({\bds u^*})\|_\F \leq 1/\varepsilon^2_0.
\ee
Also, it is easy to see $\E \mathbf{v}_s= \E\{(\M \Sigma^*)^{-1/2}\vecc(\M X_i)\} = \M 0$      and $ \cov(\mathbf{v}_s) = (\M \Sigma^*)^{-1/2} \M \Sigma^* (\M \Sigma^*)^{-1/2} = \M I_{pq}$. And by \eqref{subga:supp}, we know
$
\Pr \left[\Big|{\bf v}^{\T} \Big\{\vecc(\M X) - \E \big(\vecc(\M X)\big)\Big\}
\Big|>t\right] \leq e^{-\rho t^2 }
$
for any $\mathbf v \in \RR^{pq}$ such that $\|\mathbf v\| = 1$, which implies, 
\bee\nonumber
\Pr \left[\Big|{\bf v}^{\T} \Big\{\mathbf v_s- \E \big(\mathbf v_s\big)\Big\}
\Big|>t\right] &= \Pr \left[\Big|{\bf v}^{\T}(\M \Sigma^*)^{-1/2}\vecc(\M X_i)
\Big|>t\right] 
\\
&= \Pr \left[\Big|\Big\{\frac{(\M \Sigma^*)^{-1/2}\bf v}{\|(\M \Sigma^*)^{-1/2}\bf v\|}\Big\}^{\T}\vecc(\M X_i)
\Big|>\|(\M \Sigma^*)^{-1/2}{\bf v}\| t \right]
\\
&\leq e^{-\rho \|(\M \Sigma^*)^{-1/2}{\bf v}\|^2t^2} 
\\
&\leq e^{-\rho \varepsilon_0^2 t^2},
\ee
where the first inequality holds because $\Big\|\frac{(\M \Sigma^*)^{-1/2}\bf v}{\|(\M \Sigma^*)^{-1/2}\bf v\|}\Big\| = 1$. The second inequality holds because $\|(\M \Sigma^*)^{-1/2}{\bf v}\|^2 \geq \lambda^2_{\min}\big\{(\M \Sigma^*)^{-1/2}\big\} = \big\{\lambda^{-1/2}_{\max}\big(\M \Sigma^*\big)\big\}^2\geq (1/\varepsilon^2_0)^{-1 } = \varepsilon^2_0$. Here the inequality $\big\{\lambda^{-1/2}_{\max}\big(\M \Sigma^*\big)\big\}^2\geq (1/\varepsilon^2_0)^{-1 }$ holds because of  the property of Kronecker product, {\color{black} i.e.} for positive-definitive $\M \Sigma_1^*$, $\M \Sigma_2^*$, we have $\lambda_{\max}(\M \Sigma^*) = \lambda_{\max}(\M \Sigma^*_2 \otimes \M \Sigma_1^*) =\lambda_{\max}(\M \Sigma^*_1)\cdot \lambda_{\max}(\M \Sigma^*_2) \leq 1/\varepsilon_0^2$. Thus $\mathbf{v}^\T\mathbf v_s$ is a sub-Gaussian random vector parametrized with $\rho\varepsilon_0^2$ by the definition \eqref{subga:supp} for any $\|\mathbf{v}\| = 1$. By Proposition 2.5.2 in \citet{vershynin2018high}, there exists a fixed constant $C'>0$ such that
\bee\label{l:srate:con2}
\|\mathbf{v}^\T\mathbf v_s\|_{\psi_2}\leq C'/(\sqrt{\rho}\varepsilon_0)
\ee
for any $\|\mathbf{v}\| = 1$. Combining \eqref{l:srate:hwi}, \eqref{l:srate:con1} and \eqref{l:srate:con2}, by definition we have
\bee\label{l:srate:subexp}
&\p(|I_{i}(\bds u^*,\bds v^*)|\geq t)
\\
&=\p\Big[\Big|\vecc(\M X_i)^\T\cdot\mathscr{V}^*({\bds v}^*)\otimes\mathscr{U}({\bds u^*})\cdot\vecc(\M X_i) 
- \E\Big\{\vecc(\M X_i)^\T\cdot\mathscr{V}^*({\bds v}^*)\otimes\mathscr{U}({\bds u^*})\cdot\vecc(\M X_i)\Big\}\Big|\geq t\Big] 
\\
&=\p\Big(|\mathbf{v}_s^\T \M M \mathbf{v}_s - \E\mathbf{v}_s^\T \M M \mathbf{v}_s|\geq t\Big) 
\\
&\leq 2\exp\Bigg[-\min\Big\{\frac{t^2}{C^2(C'/\sqrt{\rho}\varepsilon_0)^4\|\M M\|_\F^2},\frac{t}{C(C'/\sqrt{\rho}\varepsilon_0)^2\|\M M\|_\F}\Big\}\Bigg]
\\
&\leq 2\exp\Bigg[-\min\Big\{\frac{t^2}{C^2(C'/\sqrt{\rho}\varepsilon_0)^4\varepsilon_0^2},\frac{t}{C(C'/\sqrt{\rho}\varepsilon_0)^2\varepsilon_0}\Big\}\Bigg]
\\
& = 2\exp\Big[-\min\Big\{C_1(\rho,\varepsilon_0)t^2,C_2(\rho,\varepsilon_0)t\Big\}\Big]
\\
&=\begin{cases}
2\exp\{-C_1(\rho,\varepsilon_0)t^2\} & 0\leq t\leq C_2(\rho,\varepsilon_0)/C_1(\rho,\varepsilon_0)
\\
2\exp\{-C_2(\rho,\varepsilon_0)t\} & t > C_2(\rho,\varepsilon_0)/C_1(\rho,\varepsilon_0),
\end{cases}
\ee
where we define  $C_1(\rho,\varepsilon_0) \equiv \rho^2\varepsilon_0^2/C^2C'^4, C_2(\rho,\varepsilon_0) \equiv \rho\varepsilon_0/CC'^2$. The above inequality implies that there exists a sufficiently large $K_1(\rho,\varepsilon_0)$ only depending on $\rho,\varepsilon_0$, such that $\p\Big\{|I_i(\bds u^*,\bds v^*)|>t\Big\} \leq \exp\big\{1 - t/K_1(\rho,\varepsilon_0)\big\}$. This tail probability bound of $I_i(\bds u^*,\bds v^*)$ satisfies (5.14) in \citet{eldar2012compressed}. Thus, by results in \citet{eldar2012compressed}, there exits a $K_2(\rho,\varepsilon_0)$ { satisfying} (5.15) in \citet{eldar2012compressed} { such that} $(\E|I_i(\bds u^*,\bds v^*)|^s)^{1/s}\leq K_2(\rho,\varepsilon_0)s$ for all $s\geq 1$. Then combining with Proposition 2.7.1 and Definition 2.7.5 in \citet{vershynin2018high}, we have $\|I_i(\bds u^*,\bds v^*)\|_{\psi_1} \leq K_3(\rho,\varepsilon_0)$ for some constant $K_3(\rho,\varepsilon_0)>0$ only determined by $\rho,\varepsilon_0$. Therefore the $\{I_i(\bds u^*,\bds v^*)\}_{i = 1}^n$ are i.i.d. sub-Exponential random variables (see definitions of $\|\cdot\|_{\psi_1}$ and sub-Exponential in \eqref{def:psi1} and \eqref{subexp}). By Bernstein inequality for sum of sub-Gaussian random variables (see, e.g. Theorem 2.8.1 in \citet{vershynin2018high}), we finally show
\bee\label{lm:srate:coreprob}
&\p\Big\{\Big|\frac{1}{n}\sum _{i = 1}^nI_i(\bds u^*,\bds v^*)\Big|\geq t\Big\}
\\
&\leq 2\exp\Big[-C''\min\Big\{\frac{n^2t^2}{n\|I_{i}(\bds u^*,\bds v^*)\|_{\psi_1}^2},\frac{nt}{\|I_{i}(\bds u^*,\bds v^*)\|_{\psi_1}}\Big\}\Big]
\\
&\leq2\exp\Big[-nC''\min\Big\{\frac{t^2}{K^2_3(\rho,\varepsilon_0)},\frac{t}{K_3(\rho,\varepsilon_0)}\Big\}\Big].
\ee
\par
\noindent\textbf{Step 1.3: } For any $\bds u\in\mathcal{U}^{\mathcal{B}}_{p^2}(k_1)$, we claim that one of its coordinates can be non-zero, only if the corresponding coordinate in $\vecc\big\{B_{k_1}(\bds 1_p)\big\}$ equals $1$. We will prove this by contradiction. By the definition of  $\mathcal{U}^{\mathcal{B}}_{p^2}(k_1)$ in \eqref{def:Ub}, we have for any $\bds u \in \mathcal{U}^{\mathcal{B}}_{p^2}(k_1)$, it must satisfy $\|\bds u\| = \|\bds u\circ \vecc\{B_{k_1}(\bds 1_{p})\}\| = 1$. Therefore, if there exists a $\bds u\in\mathcal{U}^{\mathcal{B}}_{p^2}(k_1)$ such that one of its coordinates  is non-zero while the corresponding  coordinate in $\vecc\big\{B_{k_1}(\bds 1_p)\big\}$ also equals $0$, we must have $\|\bds u\circ \vecc\{B_{k_1}(\bds 1_{p})\}\| < \|\bds u\| = 1$. Contradiction!
\par
By the above argument, we know that although the dimension of space $\mathcal{U}_{p^2}^{\mathcal{B}}(k_1)$ is $p^2$, the number of coordinates that can take non-zero values for vectors in $\mathcal{U}_{p^2}^{\mathcal{B}}(k_1)$ is far less than $p^2$. And this number actually equals the number of non-zero entries in $\vecc\{B_{k_1}(\bds 1_p)\}$. In the following, we use this result to construct an $\epsilon$-net over $\mathcal{U}^{\mathcal{B}}_{p^2}(k_1)$ with reduced complexity. 
\par 
By definition, the number of coordinates in $\vecc\big\{ B_{k_1}(\bds 1_p)\big\}$ that equal $1$ is in the same order of $k_1p$. Therefore, the number of all coordinates in $\vecc\big\{ B_{k_1}(\bds 1_p)\big\}$ that equal $1$, {\color{black} or equivalently, the number of all possible non-zero coordinates of $\bds u\in \mathcal{U}^{\mathcal{B}}_{p^2}(k_1)$ is $c_uk_1p$}, where $c_u>0$ is determined by $k_1$ and $p$. In particular, $c_u$ can be upper bounded by 3. This is because by definition, there are at most $2k_1 + 1$ entries in each row of $\bds 1_p$ equalling $1$, and thus there are at most $(2k_1 + 1)p \leq 3k_1p$ entries equalling $1$ in $\vecc\big\{ B_{k_1}(\bds 1_p)\big\}$. Similarly, it is easy to check that $k_1p\leq (k_1+1)p\leq c_upk_1$. Thus $c_u$ is bounded in a constant interval $[1,3]$. Therefore, $c_u$ can be treated as a constant in this proof. Then, for each $\bds u^*\in \mathcal{U}^{\mathcal{B}}_{p^2}(k_1)$, we define  $\mathcal{R}_{u}(\bds u^*)$ as a vector in $\RR^{c_upk_1}$ that only preserve the coordinates of  $\bds u^*$, whose corresponding coordinates are equal to $1$ in $\vecc\{B_{k_1}(\bds 1_{p})\}$.
\par
Recall that $\bds u^*\in \mathcal{U}^{\mathcal{B}}_{p^2}(k_1)$, it is easy to see that $\|\mathcal{R}_{u}(\bds u^*)\| = \|\bds u^*\| = 1$,  i.e. %and $\mathcal{R}_{u}(\bds u^*)$ is in the unit sphere of $\RR^{c_upk_1}$. So 
$\mathcal{R}_u(\bds u^*)\in\mathcal{U}_{c_{u}pk_1}.$ Also, here $c_u pk_1$ is the dimension of $\mathcal{R}_{u}(\bds u^*)$, where $c_u$ is determined by $k_1$ and $p$.  From definition, one can see that $\mathcal{R}_{u}(\bds u^*)$ is a bijection from $\mathcal{U}^{\mathcal{B}}_{p^2}(k_1)$ to $\mathcal{U}_{c_upk_1}$ and $\mathcal{R}_u\big\{\mathcal{U}^{\mathcal{B}}_{p^2}(k_1)\big\} = \mathcal{U}_{c_upk_1}, \mathcal{R}^{-1}_u\big\{\mathcal{U}_{c_upk_1}\big\} = \mathcal{U}^{\mathcal{B}}_{p^2}(k_1).$ Similarly we can define $\mathcal{R}_v(\bds v^*): \mathcal{U}^{\mathcal{B}}_{q^2}(k_2) \longrightarrow \mathcal{U}_{c_vqk_2}$ for any $\bds v^*\in \mathcal{U}_{q^2}^{\mathcal{B}}(k_2)$ with some $c_v \in [1,3]$.
\par
Next, we build $\epsilon$-nets over $\mathcal{U}^{\mathcal{B}}_{p^2}(k_1),\mathcal{U}^{\mathcal{B}}_{q^2}(k_2)$ and bound $\big\|\xi\big\{\tilde{\M \Sigma}_{0,\mathcal{B}}(k_1,k_2)-\M\Sigma^{*,\mathcal{B}}_2(k_2)\otimes \M\Sigma^{*,\mathcal{B}}_1(k_1)\big\}\big\|_{2} ^2 = \sup_{{\bds v}^*\in \mathcal{U}^{\mathcal{B}}_{q^2}(k_2), \bds u^*\in \mathcal{U}^{\mathcal{B}}_{p^2}(k_1)}|\frac{1}{n}\sum _{i = 1}^nI_i(\bds u^*,\bds v^*)|$ via combining the $\epsilon$-net arguments, with the concentration of $|\frac{1}{n}\sum _{i = 1}^nI_i(\bds u^*,\bds v^*)|$ for certain $\bds u^*,\bds v^*$ shown in \eqref{lm:srate:coreprob}.
\par
Since $\mathcal{U}_{c_upk_1}$ is the unit sphere of $\mathbb{R}^{c_upk_1}$, by Lemma 5.2 in \citet{eldar2012compressed}, we know there exists a $1/3$-net of $\mathcal{U}_{c_upk_1}$ in the Euclidean space $\RR^{c_upk_1}$ denoted by $\mathcal{N}_{c_upk_1}$, such that $|\mathcal{N}_{c_upk_1}| \leq 7^{c_upk_1}$, and for any $\tilde{\bds u}\in \mathcal{U}_{c_upk_1}$ there exists $\tilde{\bds u}^{(1/3)} \in \mathcal{N}_{c_upk_1}$ satisfying $\|\tilde{\bds u}^{(1/3)} - \tilde{\bds u}\|\leq 1/3$. For $\mathcal{R}_u$ as a bijection, we can define $\mathcal{N}_u \equiv \mathcal{R}_u^{-1}(\mathcal{N}_{c_upk_1})$ with $|\mathcal{N}_u| = |\mathcal{N}_{c_upk_1}| \leq 7^{c_u pk_1}$. By definition of $\mathcal{R}_u(\bds u^*)$, we can also see that $\mathcal{N}_u \subseteq \mathcal{U}^{\mathcal{B}}_{p^2}(k_1)$. In addition, for any $\bds u^* \in \mathcal{U}^{\mathcal{B}}_{p^2}(k_1)$, there exists $\tilde{\bds u}^{(1/3),*}\in \mathcal{N}_{c_upk_1}$ satisfying $\|\tilde{\bds u}^{(1/3),*} - \mathcal{R}_u(\bds u^*)\| \leq 1/3$ since  $\mathcal{R}_u(\bds u^*)\in \mathcal{U}_{c_u pk_1}$. Thus $\mathcal{R}_u^{-1}(\tilde{\bds u}^{(1/3),*}) \in \mathcal{N}_u$ since $\tilde{\bds u}^{(1/3),*}\in \mathcal{N}_{c_upk_1}$. By the fact that $\mathcal{R}_u(\bds u^*)$ only removes zero coordinates of $\bds u^*$, we have 
\bee\nonumber
\|\mathcal{R}_u^{-1}(\tilde{\bds u}^{(1/3),*}) - \bds u^*\| &= \|\tilde{\bds u}^{(1/3),*} - \mathcal{R}_u(\bds u^*)\|\leq 1/3,
\ee
which implies that $\mathcal{N}_u$ is a $1/3$-net of $\mathcal{U}^{\mathcal{B}}_{p^2}(k_1)$ and $|\mathcal{N}_u|\leq 7^{c_upk_1}$. Similarly, we can build $\mathcal{N}_v\subseteq \mathcal{U}_{q^2}^{\mathcal{B}}(k_2)$ as a $1/3$-net of $\mathcal{U}_{q^2}^{\mathcal{B}}(k_2) $ and $|\mathcal{N}_v|\leq 7^{c_vqk_2}$.
\par
Then for any ${\bds v}^*\in \mathcal{U}^{\mathcal{B}}_{q^2}(k_2), \bds u^*\in \mathcal{U}^{\mathcal{B}}_{p^2}(k_1)$, there exist $\bds v^{(1/3),*}\in \mathcal{N}_v, \bds u^{(1/3),*}\in\mathcal{N}_u$ with $\|\bds v^{(1/3),*} - \bds v\|\leq 1/3, \|\bds u^{(1/3),*} - \bds u\|\leq 1/3$. Thus we have
\bee\nonumber
&|\bds u^{(1/3),*}\Delta_n^{\mathcal{B}}\bds v^{(1/3),*} - \bds u^*\Delta_n^{\mathcal{B}}\bds v^*| 
\\
&= |(\bds u^{(1/3),*}-\bds u^*)\Delta_n^{\mathcal{B}}\bds v^{(1/3),*} - \bds u^*\Delta_n^{\mathcal{B}}(\bds v^*-\bds v^{(1/3),*})|
\\
&\leq |(\bds u^{(1/3),*}-\bds u^*)\Delta_n^{\mathcal{B}}\bds v^{(1/3),*}| + |\bds u^*\Delta_n^{\mathcal{B}}(\bds v^*-\bds v^{(1/3),*})|
\\
&\leq \|\bds u^{(1/3),*}-\bds u^*\|\|\Delta_n^{\mathcal{B}}\|_2\|\bds v^{(1/3),*}\| + \|\bds u^*\|\|\Delta_n^{\mathcal{B}}\|_2\|\bds v^*-\bds v^{(1/3),*}\| 
\\
&\leq 2\times 1/3\times \|\Delta_n^{\mathcal{B}}\|_2
\\
&= \frac{2}{3}\|\Delta_n^{\mathcal{B}}\|_2
\ee
By \eqref{l:srate:vdu:final}, we have
\bee\nonumber
\|\Delta_n^{\mathcal{B}}\|_2 &= \big\|\xi\big\{\tilde{\M \Sigma}_{0,\mathcal{B}}(k_1,k_2)-\M\Sigma^{*,\mathcal{B}}_2(k_2)\otimes \M\Sigma^{*,\mathcal{B}}_1(k_1)\big\}\big\|_{2} 
\\
&= \sup_{{\bds v}^*\in \mathcal{U}^{\mathcal{B}}_{q^2}(k_2), \bds u^*\in \mathcal{U}^{\mathcal{B}}_{p^2}(k_1)}|{\bds v^*}^\T\Delta_n^{\mathcal{B}}\bds u^*| 
\\
&\leq\sup_{\bds u^{(1/3),*}\in \mathcal{N}_u\atop
\bds v^{(1/3),*}\in \mathcal{N}_v}|\bds u^{(1/3),*}\Delta_n^{\mathcal{B}}\bds v^{(1/3),*}| + \frac{2}{3}\|\Delta_n^{\mathcal{B}}\|_2,
\ee
which implies that $\|\Delta_n^{\mathcal{B}}\|_2\leq 3\sup_{\bds u^{(1/3),*}\in \mathcal{N}_u\atop
\bds v^{(1/3),*}\in \mathcal{N}_v}|\bds u^{(1/3),*}\Delta_n^{\mathcal{B}}\bds v^{(1/3),*}|$. %after merging the terms on the left-hand and right-hand sides of the above inequality. 
Furthermore, by \eqref{l:srate:back},  we have
\bee\label{lm:srate:enet}
&\big\|\xi\big\{\tilde{\M \Sigma}_{0,\mathcal{B}}(k_1,k_2)-\M\Sigma^{*,\mathcal{B}}_2(k_2)\otimes \M\Sigma^{*,\mathcal{B}}_1(k_1)\big\}\big\|_2 
\\
&= \|\Delta_n^{\mathcal{B}}\|_2
\\
&=\sup_{{\bds v}^*\in \mathcal{U}^{\mathcal{B}}_{q^2}(k_2), \bds u^*\in \mathcal{U}^{\mathcal{B}}_{p^2}(k_1)}|{\bds v^*}^\T\Delta_n^{\mathcal{B}}\bds u^*| 
\\
&\leq 3\sup_{\bds u^{(1/3),*}\in \mathcal{N}_u\atop
\bds v^{(1/3),*}\in \mathcal{N}_v}|\bds u^{(1/3),*}\Delta_n^{\mathcal{B}}\bds v^{(1/3),*}| 
\\
&= 3\sup_{\bds u^{(1/3),*}\in \mathcal{N}_u\atop
\bds v^{(1/3),*}\in \mathcal{N}_v}\Big|\frac{1}{n}\sum _{i = 1}^nI_i\big(\bds u^{(1/3),*},\bds v^{(1/3),*}\big)\Big|
\ee
\par
\noindent\textbf{Step 1.4:} Combining \eqref{lm:srate:enet} and \eqref{lm:srate:coreprob}, for any $\bds u^* \in  \mathcal{U}^{\mathcal{B}}_{p^2}(k_1), \bds v^*\in\mathcal{U}^{\mathcal{B}}_{q^2}(k_2)$, we have  
\bee\label{lm:srate:con}
&\p\Big\{\|\xi\big\{\tilde{\M \Sigma}_{0,\mathcal{B}}(k_1,k_2)-\M\Sigma^{*,\mathcal{B}}_2(k_2)\otimes \M\Sigma^{*,\mathcal{B}}_1(k_1)\big\}\big\|^2_2 \geq t\Big\}
\\
&\leq\p\Bigg\{3^2\sup_{\bds u^{(1/3),*}\in \mathcal{N}_u\atop
\bds v^{(1/3),*}\in \mathcal{N}_v}\Big|\frac{1}{n}\sum _{i = 1}^nI_i\big({\bds u}^{(1/3),*},{\bds v}^{(1/3),*}\big)\Big|^2\geq t\Bigg\} 
\\
&= \p\Bigg\{\sup_{\bds u^{(1/3),*}\in \mathcal{N}_u\atop
\bds v^{(1/3),*}\in \mathcal{N}_v}\Big|\frac{1}{n}\sum _{i = 1}^nI_i\big({\bds u}^{(1/3),*},{\bds v}^{(1/3),*}\big)\Big|^2\geq t/9\Bigg\}
\\
&\leq |\mathcal{N}_u|\times |\mathcal{N}_v|\times\p\Bigg\{\Big|\frac{1}{n}\sum _{i = 1}^nI_i\big({\bds u}^{*},{\bds v}^{*}\big)\Big|^2\geq t/9\Bigg\}
\\
&= 7^{c_upk_1 + c_v qk_2}\times\p\Bigg\{\Big|\frac{1}{n}\sum _{i = 1}^nI_i\big({\bds u}^{*},{\bds v}^{*}\big)\Big|\geq t^{1/2}/3\Bigg\}
\\
&\leq 2\times 7^{c_upk_1 + c_v qk_2}\times\exp\Bigg[-nC''\min\Big\{\frac{t}{9K^2_3(\rho,\varepsilon_0)},\frac{\sqrt{t}}{3K_3(\rho,\varepsilon_0)}\Big\}\Bigg]
\\
&\leq  2a_1^{pk_1 + qk_2}\exp\big[-n\min\{a_2t,a_3\sqrt{t}\}\big],
\ee 
where we define three constants $a_1 \equiv 7^{\max(c_u,c_v)}>1, a_2 \equiv C''/(9K_3^2(\rho,\varepsilon_0)), a_3 \equiv C''/3K_3(\rho,\epsilon)$. For any chosen $t_n$ {depending on} $n$, we bound $\E\|\xi\big\{\tilde{\M \Sigma}_{0,\mathcal{B}}(k_1,k_2)-\M\Sigma^{*,\mathcal{B}}_2(k_2)\otimes \M\Sigma^{*,\mathcal{B}}_1(k_1)\big\}\big\|^2_2$ by Lemma \ref{lm:exp} and \eqref{lm:srate:con}
\bee\label{lm:srate:targetf}
\E&\|\xi\big\{\tilde{\M \Sigma}_{0,\mathcal{B}}(k_1,k_2)-\M\Sigma^{*,\mathcal{B}}_2(k_2)\otimes \M\Sigma^{*,\mathcal{B}}_1(k_1)\big\}\big\|^2_2 
\\
\leq &t_n + \int_{t_n}^{+\infty}\p\Big\{\|\xi\big\{\tilde{\M \Sigma}_{0,\mathcal{B}}(k_1,k_2)-\M\Sigma^{*,\mathcal{B}}_2(k_2)\otimes \M\Sigma^{*,\mathcal{B}}_1(k_1)\big\}\big\|^2_2 \geq t\Big\}dt
\\
\precsim& t_n+\int_{t_n}^{+\infty}a_1^{pk_1 + qk_2}\exp\big[-n\min\{a_2t,a_3\sqrt{t}\}\big]dt.
\ee
We give the optimal choice of $t_n$ and the convergence rate of $\E\|\xi\big\{\tilde{\M \Sigma}_{0,\mathcal{B}}(k_1,k_2)-\M\Sigma^{*,\mathcal{B}}_2(k_2)\otimes \M\Sigma^{*,\mathcal{B}}_1(k_1)\big\}\big\|^2_2 $, under {the following three scenarios,} respectively. In the following discussion of three scenarios, with a bit abuse of notation, we use $ C $ to denote some constant terms, though the $C$ in different places may denote different constants.
\par
\noindent{\textbf{(a) When} $pk_1 + qk_2\succ n$ \textbf{: }}  take $t_n = C\Big(\frac{pk_1 + qk_2}{n}\Big)^2\rightarrow +\infty$ as $n\rightarrow +\infty$. When $n$ and $t_n$ are sufficiently large, we have $a_3\sqrt{t}<a_2t$ for all $t\geq t_n$. Then the right-hand side of \eqref{lm:srate:targetf} becomes
\bee\label{lm:srate:con1target}
&t_n +\int_{t_n}^{+\infty}a_1^{pk_1 + qk_2}\exp\big[-na_3\sqrt{t}\big]dt 
\\
&= t_n+\int_{\sqrt{t_n}}^{+\infty}2a_1^{pk_1+qk_2}\tilde{t}\exp\Big[-na_3{\tilde{t}}\ \Big]d\tilde{t}  \ \ \ \ \ \ \ ( \tilde{t} = \sqrt{t})
\\
& = t_n + \frac{-2a_1^{pk_1 + qk_2}}{na_3}\int_{\sqrt{t_n}}^{+\infty}-na_3\tilde{t}\exp(-na_3\tilde{t})d\tilde{t}
\\
& = t_n+\frac{2a_1^{pk_1+qk_2}\sqrt{t_n}\exp(-na_3\sqrt{t_n})}{na_3} + \frac{2a_1^{pk_1+qk_2}\exp(-na_3\sqrt{t_n})}{n^2a_3^2}
\\
& \precsim t_n + \frac{a_1^{pk_1+qk_2}\sqrt{t_n}\exp(-na_3\sqrt{t_n})}{n} \\
& = C\Big(\frac{pk_1 + qk_2}{n}\Big)^2 + \big\{a_1\exp(-\sqrt{C}a_3)\big\}^{pk_1 + qk_2}\sqrt{C}\frac{pk_1 + qk_2}{n^2}.
\ee
We know that when $C$ is large enough, $a_1\exp(-\sqrt{C}a_3)<1$ and 
\bee\label{lm:srate:con1bridge}
\big\{a_1\exp(-\sqrt{C}a_3)\big\}^{pk_1 + qk_2}\times\big(pk_1 + qk_2\big)\rightarrow 0,
\ee
which implies the second term on the right-hand side of \eqref{lm:srate:con1target} {\color{black} converges} to $0$ as $n \rightarrow +\infty$. Combining \eqref{lm:srate:targetf}, \eqref{lm:srate:con1target} and \eqref{lm:srate:con1bridge}, we finally have
\bee\nonumber
\E\|\xi\big\{\tilde{\M \Sigma}_{0,\mathcal{B}}(k_1,k_2)-\M\Sigma^{*,\mathcal{B}}_2(k_2)\otimes \M\Sigma^{*,\mathcal{B}}_1(k_1)\big\}\big\|^2_2 \precsim\Big(\frac{pk_1 + qk_2}{n}\Big)^2
\ee
when $pk_1 + qk_2 \succ n$.
\par
\noindent\textbf{(b) When }$pk_1 + qk_2 \prec n$\textbf{ :}  take $t_n = C(\frac{pk_1 + qk_2 }{n})\rightarrow 0$ as $n\rightarrow +\infty$. It is easy to see $\min\{a_2 t,a_3\sqrt{t}\} = \begin{cases}a_2 t & t\leq a_3^2/a_2^2\\
a_3\sqrt{t} & t > a_3^2/a_2^2
\end{cases}$. When $n$ is sufficiently large {\color{black} such that} $t_n \leq a_3^2/a_2^2$, the right-hand side of \eqref{lm:srate:targetf} becomes
\bee\label{lm:srate:con2target}
&t_n+ \int_{t_n}^{a_3^2/a_2^2}a_1^{pk_1 + qk_2}\exp\big[-na_2{t}\big]dt + \int_{a_3^2/a_2^2}^{+\infty}a_1^{pk_1 + qk_2}\exp\big[-na_3\sqrt{t}\big]dt
\\
&=t_n+ \frac{a_1^{pk_1 + qk_2}}{na_2}\Big[\exp\{-na_2t_n\} - \exp\{-na_3^2/a_2\}\Big]+ \int_{a_3^2/a_2^2}^{+\infty}a_1^{pk_1 + qk_2}\exp\big[-na_3\sqrt{t}\big]dt \\
%&=t_n+ \frac{a_1^{pk_1 + qk_2}}{na_2}\underbrace{\Big[\exp\{-na_2t_n\} - \exp\{-na_3^2/a_2\}\Big]}_{\leq \exp(-na_2t_n)}+ \underbrace{\int_{a_3^2/a_2^2}^{+\infty}a_1^{pk_1 + qk_2}\exp\big[-na_3\sqrt{t}\big]dt}_{\precsim\frac{a_1^{pk_1+qk_2}\exp(-na_3^2/a_2)}{n} + \frac{a_1^{pk_1+qk_2}\exp(-na_3^2/a_2)}{n^2}\text{ by \eqref{lm:srate:con1target}}}
%\\
&\precsim t_n + \frac{a_1^{pk_1 + qk_2}}{n}\exp(-na_2t_n) + \frac{a_1^{pk_1+qk_2}\exp(-na_3^2/a_2)}{n} +\frac{a_1^{pk_1+qk_2}\exp(-na_3^2/a_2)}{n^2}  
\\
&\precsim t_n + \frac{a_1^{pk_1 + qk_2}}{n}\exp(-na_2t_n) + \frac{a_1^{pk_1+qk_2}\exp(-na_3^2/a_2)}{n}
\\
&\asymp \frac{pk_1 + qk_2}{n} + \{a_1\exp(-a_2C)\}^{pk_1 + qk_2}/n + \frac{a_1^{pk_1+qk_2}\exp(-na_3^2/a_2)}{n},
\ee
{\color{black} where the first inequality holds by \eqref{lm:srate:con1target}.}
On the right-hand side of \eqref{lm:srate:con2target}, the first term $\frac{pk_1 + qk_2}{n} \succsim 1/n$ since $p,q,k_1,k_2\geq 1$, the second term $\{a_1\exp(-a_2C)\}^{pk_1 + qk_2}/n \precsim 1/n$ when $C$ is large enough and $a_1\exp(-a_2C) < 1$, and the third term $\frac{a_1^{pk_1+qk_2}\exp(-na_3^2/a_2)}{n} \precsim 1/n$ because $pk_2 + qk_2 \prec n$ and $na_3^2/a_2 \asymp n$, which implies $a_1^{pk_1+qk_2}\exp(-na_3^2/a_2) \rightarrow 0$ as $n \rightarrow +\infty$. Finally, we conclude
\bee\nonumber
\E\|\xi\big\{\tilde{\M \Sigma}_{0,\mathcal{B}}(k_1,k_2)-\M\Sigma^{*,\mathcal{B}}_2(k_2)\otimes \M\Sigma^{*,\mathcal{B}}_1(k_1)\big\}\big\|^2_2 \precsim \frac{pk_1 + qk_2}{n}
\ee
when $pk_1 + qk_2 \prec n$.
\par
\noindent\textbf{(c) When }$pk_1 + qk_2 \asymp n$\textbf{ :} under this scenario, we have $pk_1 + qk_2 \leq Cn$ with some $C >0$ as $n$ is sufficiently large. Taking $t_n = t_0 \equiv \max\{a_3^2/a_2^2 + 1,C^2\log^2(a_1)/a_3^2\}\geq a_3^2/a_2^2$, the right-hand side of \eqref{lm:srate:targetf} becomes
\bee\nonumber
&t_0 + \int_{t_0}^{+\infty}a_1^{pk_1 + qk_2}\exp\big[-na_3\sqrt{t}\big]dt 
\\
&= t_0 + \frac{2a_1^{pk_1+qk_2}\sqrt{t_0}\exp(-na_3\sqrt{t_0})}{na_3} + \frac{2a_1^{pk_1+qk_2}\exp(-na_3\sqrt{t_0})}{n^2a_3^2} 
\\
&\precsim t_0 + \frac{a_1^{pk_1 + q k_2}\exp(-na_3\sqrt{t_0})}{n}
\\
&\precsim t_0 + \frac{a_1^{C_{\text{up}}n}\exp\big\{-na_3\frac{C_{\text{up}}\log(a_1)}{a_3}\big\}}{n} \\
& = t_0 + 1/n \\
& \precsim t_0,
\ee 
where the first equality is derived in the same way as \eqref{lm:srate:con1target}. Thus $\E\|\xi\big\{\tilde{\M \Sigma}_{0,\mathcal{B}}(k_1,k_2)-\M\Sigma^{*,\mathcal{B}}_2(k_2)\otimes \M\Sigma^{*,\mathcal{B}}_1(k_1)\big\}\big\|^2_2  \precsim 1 \asymp \frac{pk_1 + qk_2}{n}$ when $pk_1 + qk_2 \asymp n$.
\par
Combining scenarios \textbf{(a)}--\textbf{(c)}, we finish our {Step 1} by {\color{black} showing}
\bee\label{lm:srate:ratestep1}
\E\|\xi\big\{\tilde{\M \Sigma}_{0,\mathcal{B}}(k_1,k_2)-\M\Sigma^{*,\mathcal{B}}_2(k_2)\otimes \M\Sigma^{*,\mathcal{B}}_1(k_1)\big\}\big\|^2_2 \precsim 
\begin{cases}
\frac{pk_1 + qk_2}{n} & pk_1 + qk_2 \precsim n
\\
\Big(\frac{pk_1 + qk_2}{n}\Big)^2 & pk_1 + qk_2 \succ n.
\end{cases}
\ee
\\
\par
\noindent{\textbf{Step 2: }}In this step, we bound the second term on the right-hand side of \eqref{lm:srate:bstart}: $\frac{1}{pq}\E\big\|\xi\big[\vecc(\bar{\M X})\vecc(\bar{\M X})^{\T}\HDBand \big]\big\|_{2}^2$. In particular, we will show that the convergence rate of the second term on right-hand side of \eqref{lm:srate:bstart} is at the same asymptotic order with the first term of it. 

For simplicity, denote $\mathcal{H}^{\mathcal{B}}_{n,k_1,k_2} \equiv \xi\big[\vecc(\bar{\M X})\vecc(\bar{\M X})^{\T}\HDBand \big]$. By triangle inequality,
\bee\label{lm:srate:final}
\E\big\|\xi\big[\vecc(\bar{\M X})\vecc(\bar{\M X})^{\T}\HDBand \big]\big\|_{2}^2 &= \E\|\mathcal{H}^{\mathcal{B}}_{n,k_1,k_2} \|_2^2
\\
&\leq 2\E \|\mathcal{H}^{\mathcal{B}}_{n,k_1,k_2}  - \E\mathcal{H}^{\mathcal{B}}_{n,k_1,k_2} \|_2^2 + 2\|\E\mathcal{H}^{\mathcal{B}}_{n,k_1,k_2} \|_2^2.
\ee 
Therefore, we bound the two terms on the right-hand side, respectively. 
\par
\noindent{\textbf{(a): }}The proof technique to bound $2\E \|\mathcal{H}^{\mathcal{B}}_{n,k_1,k_2}  - \E\mathcal{H}^{\mathcal{B}}_{n,k_1,k_2} \|_2^2$ is almost the same as the technique used in {Step 1}.  Similar to the derivation for $\xi\big\{\tilde{\M \Sigma}_{0,\mathcal{B}}(k_1,k_2)-\M\Sigma^{*,\mathcal{B}}_2(k_2)\otimes \M\Sigma^{*,\mathcal{B}}_1(k_1)\big\} = \xi\big\{\tilde{\M \Sigma}_{0,\mathcal{B}}(k_1,k_2)\big\} -\E[\xi\big\{\tilde{\M \Sigma}_{0,\mathcal{B}}(k_1,k_2)\big\}] $ in \eqref{l:srate:vdu:0}-\eqref{l:srate:vdu2}, for $\mathcal{H}^{\mathcal{B}}_{n,k_1,k_2}  - \E\mathcal{H}^{\mathcal{B}}_{n,k_1,k_2}$, we can also derive
\bee\nonumber
&\bds v^\T \cdot(\mathcal{H}^{\mathcal{B}}_{n,k_1,k_2}  - \E\mathcal{H}^{\mathcal{B}}_{n,k_1,k_2} )\cdot\bds u 
\\
&=\bds v^\T\cdot\Big[\xi\big[\vecc(\bar{\M X})\vecc(\bar{\M X})^{\T}\HDBand \big] - \E\xi\big[\vecc(\bar{\M X})\vecc(\bar{\M X})^{\T}\HDBand \big] \Big]\cdot \bds u 
\\
&= [\bds v\circ \vecc\{B_{k_2}(\bds 1_q)\}]^\T \times \Big[\xi\{\vecc(\bar{\M X})\vecc(\bar{\M X})^{\T}\} - \E\big[\xi\{\vecc(\bar{\M X})\vecc(\bar{\M X})^{\T}\}\big]\Big]\times[\bds u\circ \vecc\{B_{k_1}(\bds 1_p)\}],
\ee
where $\bds u\in\mathcal{U}_{q^2}, \bds v\in\mathcal{U}_{p^2}$. Similarly to \eqref{l:srate:vdu:3}-\eqref{l:srate:vdu:final}, we have
\bee\label{l:srate:sec1}
\big\|&\mathcal{H}^{\mathcal{B}}_{n,k_1,k_2}  - \E\mathcal{H}^{\mathcal{B}}_{n,k_1,k_2} \big\|_{2} 
\\
= &\sup_{\bds u\in\mathcal{U}_{q^2}, \bds v\in\mathcal{U}_{p^2}}\big|\bds v^\T \cdot(\mathcal{H}^{\mathcal{B}}_{n,k_1,k_2}  - \E\mathcal{H}^{\mathcal{B}}_{n,k_1,k_2} )\cdot\bds u\big|
\\
= &\sup_{\bds u\in\mathcal{U}_{q^2}, \bds v\in\mathcal{U}_{p^2}}\big|[\bds v\circ \vecc\{B_{k_2}(\bds 1_q)\}]^\T \times \Big[\xi\{\vecc(\bar{\M X})\vecc(\bar{\M X})^{\T}\} - \E\big[\xi\{\vecc(\bar{\M X})\vecc(\bar{\M X})^{\T}\}\big]\Big]\times[\bds u\circ \vecc\{B_{k_1}(\bds 1_p)\}]\big|
\\
= & \sup_{\bds u\in\mathcal{U}_{q^2}, \bds v\in\mathcal{U}_{p^2}}\underbrace{\|\bds v\circ \vecc\{B_{k_2}(\bds 1_q)\}\|\cdot\|\bds u\circ \vecc\{B_{k_1}(\bds 1_p)\}\|}_{\leq 1\text{ as }\bds u\in\mathcal{U}_{q^2}, \bds v\in\mathcal{U}_{p^2}}\times|{\underbrace{\tilde{\bds v}}_{\in\mathcal{U}_{q^2}^{\mathcal{B}}(k_2)}}^\T \times\Big[ \xi\{\vecc(\bar{\M X})\vecc(\bar{\M X})^{\T}\} - 
\\
&\E\big[\xi\{\vecc(\bar{\M X})\vecc(\bar{\M X})^{\T}\}\big]\Big]\times \underbrace{\tilde{\bds u}}_{\in\mathcal{U}_{p^2}^{\mathcal{B}}(k_1)}|
\\
\leq &\sup_{\bds u^* \in\mathcal{U}_{p^2}^{\mathcal{B}}(k_1),\bds v^*\in\mathcal{U}_{q^2}^{\mathcal{B}}(k_2)}\big|{\bds v^*}^\T\cdot\Big[\xi\{\vecc(\bar{\M X})\vecc(\bar{\M X})^{\T}\} - \E\big[\xi\{\vecc(\bar{\M X})\vecc(\bar{\M X})^{\T}\}\big]\Big]\cdot \bds u^*\big|,
\ee
where we let $\tilde{\bds v} = \frac{\bds v\circ \vecc\{B_{k_2}(\bds 1_q)\}}{\|\bds v\circ \vecc\{B_{k_2}(\bds 1_q)\}\|}, \tilde{\bds u} = \frac{\bds u\circ \vecc\{B_{k_1}(\bds 1_p)\}}{\|\bds u\circ \vecc\{B_{k_1}(\bds 1_p)\}\|}$. {\color{black} And further we have} 
\bee\nonumber
&\sup_{\bds u^* \in\mathcal{U}_{p^2}^{\mathcal{B}}(k_1),\bds v^*\in\mathcal{U}_{q^2}^{\mathcal{B}}(k_2)}\Big|{\bds v^*}^\T\cdot\Big[\xi\{\vecc(\bar{\M X})\vecc(\bar{\M X})^{\T}\} - \E\big[\xi\{\vecc(\bar{\M X})\vecc(\bar{\M X})^{\T}\}\big]\Big]\cdot \bds u^*\Big|
\\
&=\sup_{\bds u^* \in\mathcal{U}_{p^2}^{\mathcal{B}}(k_1),\bds v^*\in\mathcal{U}_{q^2}^{\mathcal{B}}(k_2)}\Big|[\bds v^*\circ \vecc\{B_{k_2}(\bds 1_q)\}]^\T \times \Big[\xi\{\vecc(\bar{\M X})\vecc(\bar{\M X})^{\T}\} - \E\big[\xi\{\vecc(\bar{\M X})\vecc(\bar{\M X})^{\T}\}\big]\Big]
\\
&\times[\bds u^*\circ \vecc\{B_{k_1}(\bds 1_p)\}]\Big|
\\
&\leq \small{\sup_{\bds u\in\mathcal{U}_{q^2}, \bds v\in\mathcal{U}_{p^2}}\Big|[\bds v\circ \vecc\{B_{k_2}(\bds 1_q)\}]^\T \times \Big[\xi\{\vecc(\bar{\M X})\vecc(\bar{\M X})^{\T}\} - \E\big[\xi\{\vecc(\bar{\M X})\vecc(\bar{\M X})^{\T}\}\big]\Big]}
\\
&\times [\bds u\circ \vecc\{B_{k_1}(\bds 1_p)\}]\Big| 
\\
&=\big\|\mathcal{H}^{\mathcal{B}}_{n,k_1,k_2}  - \E\mathcal{H}^{\mathcal{B}}_{n,k_1,k_2} \big\|_{2}, 
\ee
where the last equality {\color{black} can be derived similarly using the argument in the} first three lines of 
\eqref{l:srate:sec1}. This implies
\bee\label{lm:srate:bar1}
&\big\|\mathcal{H}^{\mathcal{B}}_{n,k_1,k_2}  - \E\mathcal{H}^{\mathcal{B}}_{n,k_1,k_2} \big\|_{2}  
\\
&= \sup _{\bds u^* \in\mathcal{U}_{p^2}^{\mathcal{B}}(k_1),\bds v^*\in\mathcal{U}_{q^2}^{\mathcal{B}}(k_2)}\Big|{\bds v^*}^\T\cdot\Big[\xi\{\vecc(\bar{\M X})\vecc(\bar{\M X})^{\T}\} - \E\big[\xi\{\vecc(\bar{\M X})\vecc(\bar{\M X})^{\T}\}\big]\Big]\cdot \bds u^*\Big|.
\ee
On the other hand we note that $\xi\{\hat{\M \Sigma}_0 - \E\hat{\M \Sigma}_0\} = \xi\{\hat{\M \Sigma}_0\} - \E\xi\{\hat{\M \Sigma}_0\}$ by Lemma \ref{lemma:xi}. Similar to \eqref{l:srate:vsu:1}-\eqref{l:srate:back}, we can also show for $\xi\{\vecc(\bar{\M X})\vecc(\bar{\M X})^{\T}\} - \E \xi\{\vecc(\bar{\M X})\vecc(\bar{\M X})^{\T}\}$,
\bee\label{lm:srate:bar2}
&{\bds v^*}^\T\Big[\xi\{\vecc(\bar{\M X})\vecc(\bar{\M X})^{\T}\} - \E\big[\xi\{\vecc(\bar{\M X})\vecc(\bar{\M X})^{\T}\}\big]\Big]\bds u^* 
\\
&= \underbrace{\vecc(\bar{\M X})^\T\times\Big[\mathscr{V}^*({\bds v}^*)\otimes\mathscr{U}({\bds u^*})\Big]\times\vecc(\bar{\M X}) - \E\Big[\vecc(\bar{\M X})^\T\times\big\{\mathscr{V}^*({\bds v}^*)\otimes\mathscr{U}({\bds u^*})\big\}\times\vecc(\bar{\M X})\Big]}_{\equiv \bar{I}(\bds u^*,\bds v^*)},
\ee
where $\mathscr{V}^*({\bds v}^*)\otimes\mathscr{U}({\bds u^*})$ is defined in the same way as \eqref{l:srate:vsu:2}. Combining \eqref{lm:srate:bar1} and \eqref{lm:srate:bar2}, we show
\bee\nonumber
\big\|\mathcal{H}^{\mathcal{B}}_{n,k_1,k_2}  - \E\mathcal{H}^{\mathcal{B}}_{n,k_1,k_2} \big\|_{2} = \sup _{\bds u^* \in\mathcal{U}_{p^2}^{\mathcal{B}}(k_1),\bds v^*\in\mathcal{U}_{q^2}^{\mathcal{B}}(k_2)}|\bar{I}(\bds u^*,\bds v^*)|.
\ee
Now let $\M M = (\M \Sigma^*)^{1/2}\mathscr{V}^*({\bds v}^*)\otimes\mathscr{U}({\bds u^*})(\M \Sigma^*)^{1/2}, \bar{\mathbf{v}}_s = (\M \Sigma^*)^{-1/2}\vecc(\bar{\M X}) = \frac{1}{n}\sum_{i = 1}^n(\M \Sigma^*)^{-1/2}\vecc(\M X_i)$. In Step 1.2, we have already shown $\|\M M\|_\F \leq 1/\varepsilon^2_0$ and $\|\mathbf{v}^\T(\M\Sigma^*)^{-1/2}\vecc({\M X}_i)\|_{\psi_2}\leq C'/\sqrt{\rho}\varepsilon_0$ for any $\|\mathbf{v}\| = 1$. It is easy to see $\E \bar{\mathbf{v}}_s = 0, \cov(\bar{\mathbf{v}}_s) = \M I_{pq}$. Moreover, it is easy to check that $\{\mathbf{v}^\T(\M \Sigma^*)^{-1/2}\vecc(\M X_i)\}_{i = 1}^n$ are i.i.d. mean-zero, sub-Gaussian random variables with $\|\mathbf{v}^\T(\M\Sigma^*)^{-1/2}\vecc({\M X}_i)\|_{\psi_2}\leq C'/\sqrt{\rho}\varepsilon_0$ for any $\|\mathbf{v}\| = 1$, thus by { Proposition 2.6.1} in \citet{vershynin2018high}, we have
\bee\nonumber
\|\mathbf{v}^\T\bar{\mathbf{v}}_s\|_{\psi_2} & = \frac{1}{n}\Big\|\sum_{i = 1}^n\mathbf{v}^\T(\M \Sigma^*)^{-1/2}\vecc(\M X_i)\Big\|_{\psi_2}\\
&\leq C_{\psi_2}\frac{1}{\sqrt{n}}\|\mathbf{v}^\T(\M \Sigma^*)^{-1/2}\vecc(\M X_i)\|_{\psi_2}
\\
&\leq \frac{C_{\psi_2}C'}{\sqrt{\rho}\varepsilon_0}\frac{1}{\sqrt{n}},
\ee
where $C_{\psi_2}$ is a fixed constant for any $\|\mathbf{v}\| = 1$. {\color{black} Therefore, }by \eqref{l:srate:hwi}
\bee\label{lm:srate:barcon}
&\p\Big(|\bar{I}(\bds u^*,\bds v^*)|\geq t\Big) 
\\
&= \p\Big(|\bar{\mathbf{v}}_s^\T \M M \bar{\mathbf{v}}_s - \E\bar{\mathbf{v}}_s^\T \bar{\M M \mathbf{v}}_s|\geq t\Big)
\\
&\leq2\exp\Bigg[-\min\Big\{\frac{n^2t^2}{C^2(C'C_{\psi_2}/\sqrt{\rho}\varepsilon_0)^4\varepsilon_0^2},\frac{nt}{C(C'C_{\psi_2}/\sqrt{\rho}\varepsilon_0)^2\varepsilon_0}\Big\}\Bigg] 
\\
&= 2\exp\Big[-n\min\Big\{9a_2'nt^2,3a_3't\Big\}\Big],
\ee
where $a_2' = \frac{1}{9C^2(C'C_{\psi_2}/\sqrt{\rho}\varepsilon_0)^4\varepsilon_0^2}, a_3' = \frac{1}{3C(C'C_{\psi_2}/\sqrt{\rho}\varepsilon_0)^2\varepsilon_0}$. Then an exact same argument as deriving $1/3$-net in {\color{black} Steps 1.3 and 1.4} implies
\bee\label{lm:srate:tailbar}
&\p\Big(\big\|\mathcal{H}^{\mathcal{B}}_{n,k_1,k_2}  - \E\mathcal{H}^{\mathcal{B}}_{n,k_1,k_2} \big\|^2_{2} \geq t\Big) 
\\
&= \p\Big(\sup _{\bds u^* \in\mathcal{U}_{p^2}^{\mathcal{B}}(k_1),\bds v^*\in\mathcal{U}_{q^2}^{\mathcal{B}}(k_2)}|\bar{I}(\bds u^*,\bds v^*)|^2 \geq t\Big)
\\
&\leq\p\Big(3^2\sup_{\bds u^{(1/3),*}\in \mathcal{N}_u\atop
\bds v^{(1/3),*}\in \mathcal{N}_v}|\bar{I}(\bds u^{(1/3),*},\bds v^{(1/3),*})|^2 \geq {t}\Big) 
\\
&\leq 7^{c_upk_1 + c_v qk_2}\times \p\Big(3^2|\bar{I}(\bds u^{(1/3),*},\bds v^{(1/3),*})|^2 \geq {t}\Big) 
\\
&= 7^{c_upk_1 + c_v qk_2}\times \p\Big(|\bar{I}(\bds u^{(1/3),*},\bds v^{(1/3),*})| \geq \sqrt{t}/3\Big)
\\
&\leq 2 a_1^{pk_1 + qk_2}\exp\Big[-n\min\Big\{a_2'nt,a_3'\sqrt{t}\Big\}\Big],
\ee
where the last inequality holds by \eqref{lm:srate:barcon} and $a_1 = 7^{\max(c_u,c_v)}>1$. %same as \textbf{Step 1.4}. 
\par
If we compare tail probability bound of $\big\|\mathcal{H}^{\mathcal{B}}_{n,k_1,k_2}  - \E\mathcal{H}^{\mathcal{B}}_{n,k_1,k_2} \big\|^2_{2}$ in \eqref{lm:srate:tailbar},
\bee\nonumber
2a_1^{pk_1 + qk_2}\exp\big[-n\min\{a_2t,a_3\sqrt{t}\}\big],
\ee with the tail probability bound of $\|\xi\big\{\tilde{\M \Sigma}_{0,\mathcal{B}}(k_1,k_2)-\M\Sigma^{*,\mathcal{B}}_2(k_2)\otimes \M\Sigma^{*,\mathcal{B}}_1(k_1)\big\}\big\|^2_2$ in \eqref{lm:srate:con}, 
\bee\nonumber
2 a_1^{pk_1 + qk_2}\exp\big[-n\min\big\{a_2'nt,a_3'\sqrt{t}\big\}\big],
\ee
one can see that when $n\rightarrow +\infty$, for each $t$, $\big\|\mathcal{H}^{\mathcal{B}}_{n,k_1,k_2}  - \E\mathcal{H}^{\mathcal{B}}_{n,k_1,k_2} \big\|^2_{2}$ has a sharper or equally sharp tail probability bound compared to $\|\xi\big\{\tilde{\M \Sigma}_{0,\mathcal{B}}(k_1,k_2)-\M\Sigma^{*,\mathcal{B}}_2(k_2)\otimes \M\Sigma^{*,\mathcal{B}}_1(k_1)\big\}\big\|^2_2$, because for each $t$, $ n\min\{a_2'nt,a_3'\sqrt{t}\}\succsim n\min\{a_2t,a_3\sqrt{t}\}$, and thus $\exp\big[-n\min\big\{a_2'nt,a_3'\sqrt{t}\big\}\big]\precsim \exp\big[-n\min\{a_2t,a_3\sqrt{t}\}\big]$. Since the bound of $\E\|\xi\big\{\tilde{\M \Sigma}_{0,\mathcal{B}}(k_1,k_2)-\M\Sigma^{*,\mathcal{B}}_2(k_2)\otimes \M\Sigma^{*,\mathcal{B}}_1(k_1)\big\}\big\|^2_2$ in {Step 1} is based on the tail probability bound \eqref{lm:srate:con} of $\big\|\xi\big\{\tilde{\M \Sigma}_{0,\mathcal{B}}(k_1,k_2)-\M\Sigma^{*,\mathcal{B}}_2(k_2)\otimes \M\Sigma^{*,\mathcal{B}}_1(k_1)\big\}\big\|^2_2$, if we use the same argument in {Step 1.4} to derive the bound of $\big\|\mathcal{H}^{\mathcal{B}}_{n,k_1,k_2}  - \E\mathcal{H}^{\mathcal{B}}_{n,k_1,k_2} \big\|^2_{2}$ based on \eqref{lm:srate:tailbar}, we can also show
\bee\label{lm:srate:HB}
\E \big\|\mathcal{H}^{\mathcal{B}}_{n,k_1,k_2}  - \E\mathcal{H}^{\mathcal{B}}_{n,k_1,k_2} \big\|^2_{2} \precsim 
\begin{cases}
\frac{pk_1 + qk_2}{n} & pk_1 + qk_2 \precsim n
\\
\Big(\frac{pk_1 + qk_2}{n}\Big)^2 & pk_1 + qk_2 \succ n,
\end{cases}
\ee
 which implies $\E \big\|\mathcal{H}^{\mathcal{B}}_{n,k_1,k_2}  - \E\mathcal{H}^{\mathcal{B}}_{n,k_1,k_2} \big\|^2_{2}$ is in the same asymptotic order of $\E\|\xi\big\{\tilde{\M \Sigma}_{0,\mathcal{B}}(k_1,k_2)-\M\Sigma^{*,\mathcal{B}}_2(k_2)\otimes \M\Sigma^{*,\mathcal{B}}_1(k_1)\big\}\big\|^2_2$ in \eqref{lm:srate:ratestep1}, when $n\rightarrow +\infty$.
\\
\par
\noindent\textbf{(b)}: Next we show $\|\E\mathcal{H}^{\mathcal{B}}_{n,k_1,k_2} \|_2^2$ is negligible compared to $\E\|\xi\big\{\tilde{\M \Sigma}_{0,\mathcal{B}}(k_1,k_2)-\M\Sigma^{*,\mathcal{B}}_2(k_2)\otimes \M\Sigma^{*,\mathcal{B}}_1(k_1)\big\}\big\|^2_2$. For $\E \mathcal{H}^{\mathcal{B}}_{n,k_1,k_2} = \E\xi\big[\vecc(\bar{\M X})\vecc(\bar{\M X})^{\T}\HDBand \big]$, the $(m_2 - 1)*q + l_2, (m_1 - 1)*p + l_1$th entry is 
\begin{align*}\nonumber
&\E \Big\{\vecc(\bar{\M X})_{l_1,l_2}\vecc(\bar{\M X})_{m_1,m_2}\times\M I\big(|l_1 - m_1|\leq k_1, |l_2 - m_2|\leq k_2\big)\Big\} 
\\
&= \frac{1}{n^2}\E\Big\{\sum_{1\leq i_1,i_2\leq n}\vecc(\M X_{i_1})_{l_1,l_2}\vecc(\M X_{i_2})_{m_1,m_2}\Big\}\times\M I\big(|l_1 - m_1|\leq k_1, |l_2 - m_2|\leq k_2\big) %\tag*{(By definition)}
\\
&=\frac{1}{n^2}\Big[\sum_{1\leq i\leq n}\E\{\vecc(\M X_{i})_{l_1,l_2}\vecc(\M X_{i})_{m_1,m_2}\}\Big]\times\M I\big(|l_1 - m_1|\leq k_1, |l_2 - m_2|\leq k_2\big) %\tag*{(By $\vecc(\M X_i)\}_{i = 1}^n$ \text{ i.i.d.})}
\\
&=\frac{1}{n}\sigma^{(1)}_{l_1,m_1}\sigma^{(2)}_{l_2,m_2}\times\M I\big(|l_1 - m_1|\leq k_1, |l_2 - m_2|\leq k_2\big).
\end{align*}
This implies $\E \mathcal{H}^{\mathcal{B}}_{n,k_1,k_2} = \frac{1}{n}\xi\big\{\M \Sigma_2^{*,\mathcal{B}}(k_2)\otimes \M \Sigma_1^{*,\mathcal{B}}(k_1)\big\} = \frac{1}{n}\vecc\big\{\M \Sigma_2^{*,\mathcal{B}}(k_2)\big\}\vecc\big\{\M \Sigma_2^{*,\mathcal{B}}(k_2)\big\}^\T$ and 
\bee\label{lm:srate:HB2}
\|\E \mathcal{H}^{\mathcal{B}}_{n,k_1,k_2}\|_2^2 &= \Big\|\frac{1}{n}\vecc\big\{\M \Sigma_2^{*,\mathcal{B}}(k_2)\big\}\vecc\big\{\M \Sigma_2^{*,\mathcal{B}}(k_2)\big\}^\T\Big\|_2^2 
\\
&=\frac{1}{n^2}\|\vecc\big\{\M \Sigma_2^{*,\mathcal{B}}(k_2)\big\}\vecc\big\{\M \Sigma_2^{*,\mathcal{B}}(k_2)\big\}^\T\|_{\F}^2 
\\
& = \frac{1}{n^2}\big\|\xi\big\{\M \Sigma_2^{*,\mathcal{B}}(k_2)\otimes \M \Sigma_1^{*,\mathcal{B}}(k_1)\big\}\big\|_\F^2
\\
&=\frac{1}{n^2}\|\M \Sigma_2^{*,\mathcal{B}}(k_2)\otimes\M \Sigma_1^{*,\mathcal{B}}(k_1)\|_\F^2 \\
&= \frac{1}{n^2}\|\M \Sigma_2^{*,\mathcal{B}}(k_2)\|_\F^2\|\M \Sigma_1^{*,\mathcal{B}}(k_1)\|_\F^2
\\
&\precsim\frac{pq}{n^2},
\ee
where the second equality holds because $\text{Rank}\big[\vecc\big\{\M \Sigma_2^{*,\mathcal{B}}(k_2)\big\}\vecc\big\{\M \Sigma_2^{*,\mathcal{B}}(k_2)\big\}^\T\big] = 1$, the last equality holds by property of Kronecker product (See \citet{lancaster1972norms}), and the last inequality holds by,
\bee
\|\M \Sigma_1^{*,\mathcal{B}}(k_1)\|^2_\F &=\sum_{|l_1 - m_1|\leq k_1} \{\sigma^{(1)}_{l_1,m_1}\}^2
\\
&=\sum_{l_1 = m_1} \{\sigma^{(1)}_{l_1,m_1}\}^2 + \sum_{|l_1 - m_1|\leq k_1\atop l_1\neq m_1} \{\sigma^{(1)}_{l_1,m_1}\}^2
\\
&\leq p\|\M\Sigma_1^*\|_{\max}^2 + pC_0^2
\\
&\precsim p.
\ee
Here we use the condition of $\M \Sigma^*_1\in\mathcal{F}(\varepsilon_0, \alpha)$ as an example. Same result also holds for $\M \Sigma^*_1\in\mathcal{M}(\varepsilon_0, \alpha)$.
\par
Now we compare the rate of $\E\|\xi\big\{\tilde{\M \Sigma}_{0,\mathcal{B}}(k_1,k_2)-\M\Sigma^{*,\mathcal{B}}_2(k_2)\otimes \M\Sigma^{*,\mathcal{B}}_1(k_1)\big\}\big\|^2_2$ in \eqref{lm:srate:ratestep1} {\color{black}with the order $pq/n^2$.}
\par\noindent \textbf{i. When $pk_1+qk_2 \precsim n$: }we have $\frac{pk_1 + qk_2}{n}\precsim 1$ and thus
\bee\label{lm:srate:compare}
\frac{pk_1 + qk_2}{n}&\succsim \Big(\frac{pk_1 + qk_2}{n}\Big)^2 
= \frac{2pqk_1k_2 + p^2k_1^2 + q^2k_2^2}{n^2} 
\succsim \frac{pqk_1k_2}{n^2}
\succsim \frac{pq}{n^2}.
\ee
\textbf{ii. When $pk_1 + qk_2\succ n$: } same as \eqref{lm:srate:compare}, we have  $\big(\frac{pk_1 + qk_2}{n}\big)^2\succsim \frac{pq}{n^2}$. 
\par
 Summarizing the results in {\color{black} these two scenarios}, we conclude $\E\|\mathcal{H}^{\mathcal{B}}_{n,k_1,k_2}\|_2^2$ is  negligible compared to the rate of $\E\|\xi\big\{\tilde{\M \Sigma}_{0,\mathcal{B}}(k_1,k_2)-\M\Sigma^{*,\mathcal{B}}_2(k_2)\otimes \M\Sigma^{*,\mathcal{B}}_1(k_1)\big\}\big\|^2_2$ in \eqref{lm:srate:ratestep1}.
\\
\par
Results in {\color{black} these two scenarios} show that  $\|\E\mathcal{H}^{\mathcal{B}}_{n,k_1,k_2} \|_2^2$ is always negligible, compared to the rate of $\E \|\mathcal{H}^{\mathcal{B}}_{n,k_1,k_2}  - \E\mathcal{H}^{\mathcal{B}}_{n,k_1,k_2}\|^2_2$ in \eqref{lm:srate:HB}. Then by \eqref{lm:srate:final}, we finally show $\E\|\mathcal{H}^{\mathcal{B}}_{n,k_1,k_2}\|_2^2$ is in the same asymptotic order of $\E\|\xi\big\{\tilde{\M \Sigma}_{0,\mathcal{B}}(k_1,k_2)-\M\Sigma^{*,\mathcal{B}}_2(k_2)\otimes \M\Sigma^{*,\mathcal{B}}_1(k_1)\big\}\big\|^2_2$ in \eqref{lm:srate:ratestep1}, when $n\rightarrow +\infty$, i.e.,
\bee\label{lm:srate:step2}
\E\|\mathcal{H}^{\mathcal{B}}_{n,k_1,k_2}\|_2^2 \precsim 
\begin{cases}
\frac{pk_1 + qk_2}{n} & pk_1 + qk_2 \precsim n
\\
\Big(\frac{pk_1 + qk_2}{n}\Big)^2 & pk_1 + qk_2 \succ n.
\end{cases}
\ee
\\
\par
\noindent\textbf{Step 3:} Finally, combining \eqref{lm:srate:bstart}, \eqref{lm:srate:ratestep1} and \eqref{lm:srate:step2}, we show
\bee\label{lm:srate:b}
\frac{1}{pq}\E{\|\xi\big\{\tilde{\M \Sigma}_{\mathcal{B}}(k_1,k_2)\big\} - \xi\{\M\Sigma^{*,\mathcal{B}}_2(k_2)\otimes \M\Sigma^{*,\mathcal{B}}_1(k_1)\}\|_{2}^2}  &\precsim \frac{1}{pq}
\begin{cases}
\frac{pk_1 + qk_2}{n} & pk_1 + qk_2 \precsim n
\\
\Big(\frac{pk_1 + qk_2}{n}\Big)^2 & pk_1 + qk_2 \succ n.
\end{cases}
\\
& \asymp \begin{cases}
\frac{k_1}{qn} + \frac{k_2}{pn} & pk_1 + qk_2 \precsim n
\\
\frac{pk^2_1}{qn^2} + \frac{qk^2_2}{pn^2} & pk_1 + qk_2 \succ n. 
\end{cases}
\ee
\\

Next we prove Lemma \ref{lemma:srate} when $\eta = \mathcal{T}$ for the proposed tapering estimator. {\color{black} We note the proof procedures for the proposed tapering estimators are analogous to the proof for the proposed banded estimator, and most of the proof techniques can be directly applied to the proposed tapering estimator case. Thus we omit those details that are similar and only focus on parts that are different. }
\par
Similar to \eqref{lm:srate:bstart}, we have
\bee\label{lm:srate:tstart}
&\frac{1}{pq}\E\|\xi\big\{\tilde{\M \Sigma}_{\mathcal{T}}(k_1,k_2)\big\} - \xi\{\M\Sigma^{*,\mathcal{T}}_2(k_2)\otimes \M\Sigma^{*,\mathcal{T}}_1(k_1)\}\|_{2}^2
\\
&\precsim \frac{1}{pq}\E\big\|\underbrace{\xi\big\{\tilde{\M \Sigma}_{0,\mathcal{T}}(k_1,k_2)-\M\Sigma^{*,\mathcal{T}}_2(k_2)\otimes \M\Sigma^{*,\mathcal{T}}_1(k_1)\big\}}_{\Delta_n^{\mathcal{T}}}\big\|_{2}^2 
\\
&+ \frac{1}{pq}\E\big\|\xi\big[\vecc(\bar{\M X})\vecc(\bar{\M X})^{\T}\HDTaper \big]\big\|_{2} ^ 2 .
\ee
Similar to $\eta = \mathcal{B}$, we use three steps to show the desired result. In {Step 1}, we bound the first term $\frac{1}{pq}\E\big\|{\xi\big\{\tilde{\M \Sigma}_{0,\mathcal{T}}(k_1,k_2)-\M\Sigma^{*,\mathcal{T}}_2(k_2)\otimes \M\Sigma^{*,\mathcal{T}}_1(k_1)\big\}}\big\|_{2}^2$ on the right-hand side of \eqref{lm:srate:tstart}. In {Step 2} we show the error rate of second term $ \frac{1}{pq}\E\big\|\xi\big[\vecc(\bar{\M X})\vecc(\bar{\M X})^{\T}\HDTaper \big]\big\|_{2} ^ 2$ is in the same asymptotic order of the first term. In {Step 3} we combine the bounds of two terms together and finally show \eqref{lemma:srate:res} for $\eta = \mathcal{T}$.
\\
\par 
\noindent \textbf{Step 1: }For the first term on the right-hand side of \eqref{lm:srate:tstart}, similar to \eqref{l:srate:vdu:3} we can show
\bee\nonumber
|\bds v^\T&\Delta_n^{\mathcal{T}}\bds u| = \Big|\big[\bds v\circ\vecc\big\{ T_{k_2}(\bds 1_q)\big\}\big]^\T\times \xi\big\{\hat{\M \Sigma}_0 - \E\hat{\M \Sigma}_0\big\}\times\big[\bds u\circ\vecc\big\{ T_{k_1}(\bds 1_p)\big\}\big]\Big|
\ee
for any $\bds v\in\mathcal{U}_{q^2}, \bds u \in \mathcal{U}_{p^2}$. By definition of $T_k(\cdot)$ in \eqref{def:Tk}, we know $\frac{\bds v\circ\vecc\{ T_{k_2}(\bds 1_q)\}}{\|\bds v\circ\vecc\{ T_{k_2}(\bds 1_q)\}\|} \in \mathcal{U}_{q^2}^{\mathcal{B}}(k_2)$ and $\frac{\bds u\circ\vecc\{ T_{k_1}(\bds 1_p)\}}{\|\bds u\circ\vecc\{ T_{k_1}(\bds 1_p)\}\|} \in \mathcal{U}_{p^2}^{\mathcal{B}}(k_1)$ since for the coordinates that are non-zero in $\vecc\{ T_{k_2}(\bds 1_q)\}$, the corresponding coordinates in $\vecc\{B_{k_1}(\bds 1_p)\}$  must also be non-zero. Therefore,
\bee\nonumber
|\bds v^\T \Delta_n^{\mathcal{T}}\bds u| &= \Big|\big[\bds v\circ\vecc\big\{ T_{k_2}(\bds 1_q)\big\}\big]^\T\times \xi\big\{\hat{\M \Sigma}_0 - \E\hat{\M \Sigma}_0\big\}\times\big[\bds u\circ\vecc\big\{ T_{k_1}(\bds 1_p)\big\}\big]\Big|
\\
&= \underbrace{\|\bds u\circ\vecc\big\{ T_{k_1}(\bds 1_p)\big\}\|\|\bds u\circ\vecc\big\{ T_{k_1}(\bds 1_p)\big\}\|}_{\leq 1}
\\
&\times \Big|{\underbrace{\Big[\frac{\bds v\circ\vecc\big\{ T_{k_2}(\bds 1_q)\big\}}{\|\bds v\circ\vecc\big\{ T_{k_2}(\bds 1_q)\big\}\|}\Big]}_{\in\mathcal{U}^{\mathcal{B}}_{q^2}(k_2)}}^\T\times \xi\big\{\hat{\M \Sigma}_0 - \E\hat{\M \Sigma}_0\big\}\times\underbrace{\Big[\frac{\bds u\circ\vecc\big\{ T_{k_1}(\bds 1_p)\big\}}{\|\bds u\circ\vecc\big\{ T_{k_1}(\bds 1_p)\big\}\|}\Big]}_{\in \mathcal{U}_{p^2}^{\mathcal{B}}(k_1)}\Big|
\\
&\leq  \Big|{\underbrace{\Big[\frac{\bds v\circ\vecc\big\{ T_{k_2}(\bds 1_q)\big\}}{\|\bds v\circ\vecc\big\{ T_{k_2}(\bds 1_q)\big\}\|}\Big]}_{\in\mathcal{U}^{\mathcal{B}}_{q^2}(k_2)}}^\T\times \xi\big\{\hat{\M \Sigma}_0 - \E\hat{\M \Sigma}_0\big\}\times\underbrace{\Big[\frac{\bds u\circ\vecc\big\{ T_{k_1}(\bds 1_p)\big\}}{\|\bds u\circ\vecc\big\{ T_{k_1}(\bds 1_p)\big\}\|}\Big]}_{\in \mathcal{U}_{p^2}^{\mathcal{B}}(k_1)}\Big|,
\ee
which implies
\bee\nonumber
\|\xi\big\{\tilde{\M \Sigma}_{0,\mathcal{T}}(k_1,k_2)-\M\Sigma^{*,\mathcal{T}}_2(k_2)\otimes \M\Sigma^{*,\mathcal{T}}_1(k_1)\big\}\|_2 &= \sup_{\bds u\in\mathcal{U}_{p^2}, \bds v\in\mathcal{U}_{q^2}}|\bds v^\T\Delta_n^{\mathcal{T}}\bds u| 
\\
&\leq \sup_{\bds u^*\in\mathcal{U}^{\mathcal{B}}_{p^2}(k_1), \bds v^*\in\mathcal{U}_{q^2}^{\mathcal{B}}(k_2)}|{\bds v^*}^\T\times\xi\big\{\hat{\M \Sigma}_0 - \E\hat{\M \Sigma}_0\big\}\times \bds u^*| 
\\
&=\|\xi\big\{\tilde{\M \Sigma}_{0,\mathcal{B}}(k_1,k_2)-\M\Sigma^{*,\mathcal{B}}_2(k_2)\otimes \M\Sigma^{*,\mathcal{B}}_1(k_1)\big\}\|_2, 
\ee 
where the last equality holds by \eqref{l:srate:vdu:final}. Then we have 
$$\E\|\xi\big\{\tilde{\M \Sigma}_{0,\mathcal{T}}(k_1,k_2)-\M\Sigma^{*,\mathcal{T}}_2(k_2)\otimes \M\Sigma^{*,\mathcal{T}}_1(k_1)\big\}\|^2_2\leq  \E \|\xi\big\{\tilde{\M \Sigma}_{0,\mathcal{B}}(k_1,k_2)-\M\Sigma^{*,\mathcal{B}}_2(k_2)\otimes \M\Sigma^{*,\mathcal{B}}_1(k_1)\big\}\|^2_2.$$
{\color{black} Thus,} $\E\|\xi\big\{\tilde{\M \Sigma}_{0,\mathcal{T}}(k_1,k_2)-\M\Sigma^{*,\mathcal{T}}_2(k_2)\otimes \M\Sigma^{*,\mathcal{T}}_1(k_1)\big\}\|^2_2$ has the same rate as the rate of $\E \|\xi\big\{\tilde{\M \Sigma}_{0,\mathcal{B}}(k_1,k_2)-\M\Sigma^{*,\mathcal{B}}_2(k_2)\otimes \M\Sigma^{*,\mathcal{B}}_1(k_1)\big\}\|^2_2$  given in \eqref{lm:srate:ratestep1}.
\\
\par
\noindent\textbf{Step 2: }For the second term on the right-hand side of \eqref{lm:srate:tstart}, we denote $\mathcal{H}^{\mathcal{T}}_{n,k_1,k_2} \equiv \xi\big[\vecc(\bar{\M X})\vecc(\bar{\M X})^{\T}\HDTaper \big]$ and bound it similarly to \eqref{lm:srate:final},
\bee\label{lm:srate:t2}
\E\big\|\xi\big[\vecc(\bar{\M X})\vecc(\bar{\M X})^{\T}\HDTaper \big]\big\|_{2} ^ 2  \precsim \E\|\mathcal{H}^{\mathcal{T}}_{n,k_1,k_2} - \E\mathcal{H}^{\mathcal{T}}_{n,k_1,k_2}\|_2^2 + \|\E\mathcal{H}^{\mathcal{T}}_{n,k_1,k_2}\|_2^2.
\ee
\textbf{(a): }Similar to \eqref{l:srate:sec1}, by definition of $T_{k}(\cdot)$, for any $\bds v\in\mathcal{U}_{p^2}, \bds u\in\mathcal{U}_{q^2}$, let $\tilde{\bds v}' = \frac{\bds v\circ \vecc\{T_{k_2}(\bds 1_q)\}}{\|\bds v\circ \vecc\{T_{k_2}(\bds 1_q)\}\|}, \tilde{\bds u}' = \frac{\bds u\circ \vecc\{T_{k_2}(\bds 1_q)\}}{\|\bds u\circ \vecc\{T_{k_2}(\bds 1_q)\}\|}$. We have already shown $\tilde{\bds v}'  \in \mathcal{U}_{q^2}^{\mathcal{B}}(k_2)$ and $\tilde{\bds u}' \in \mathcal{U}_{p^2}^{\mathcal{B}}(k_1)$. Then we have
\bee\nonumber
&\big\|\mathcal{H}^{\mathcal{T}}_{n,k_1,k_2}  - \E\mathcal{H}^{\mathcal{T}}_{n,k_1,k_2} \big\|_{2} 
\\
&= \sup_{\bds u\in\mathcal{U}_{q^2}, \bds v\in\mathcal{U}_{p^2}}\underbrace{\|\bds v\circ \vecc\{T_{k_2}(\bds 1_q)\}\|\cdot\|\bds u\circ \vecc\{T_{k_1}(\bds 1_p)\}\|}_{\leq 1\text{ as }\bds u\in\mathcal{U}_{q^2}, \bds v\in\mathcal{U}_{p^2}}\times|{\underbrace{\tilde{\bds v}'}_{\in\mathcal{U}_{q^2}^{\mathcal{B}}(k_2)}}^\T \times\Big[ \xi\{\vecc(\bar{\M X})\vecc(\bar{\M X})^{\T}\} - 
\\
&\E\big[\xi\{\vecc(\bar{\M X})\vecc(\bar{\M X})^{\T}\}\big]\Big]\times \underbrace{\tilde{\bds u}'}_{\in\mathcal{U}_{p^2}^{\mathcal{B}}(k_1)}|
\\
&\leq \sup_{\bds u^* \in\mathcal{U}_{p^2}^{\mathcal{B}}(k_1),\bds v^*\in\mathcal{U}_{q^2}^{\mathcal{B}}(k_2)}\big|{\bds v^*}^\T\cdot\Big[\xi\{\vecc(\bar{\M X})\vecc(\bar{\M X})^{\T}\} - \E\big[\xi\{\vecc(\bar{\M X})\vecc(\bar{\M X})^{\T}\}\big]\Big]\cdot \bds u^*\big|
\\
&=\big\|\mathcal{H}^{\mathcal{B}}_{n,k_1,k_2}  - \E\mathcal{H}^{\mathcal{B}}_{n,k_1,k_2} \big\|_{2}, \quad  \text{(By \eqref{lm:srate:bar1})}
\ee
where $\tilde{\bds v}' = \frac{\bds v\circ \vecc\{T_{k_2}(\bds 1_q)\}}{\|\bds v\circ \vecc\{T_{k_2}(\bds 1_q)\}\|}, \tilde{\bds u}' = \frac{\bds u\circ \vecc\{T_{k_2}(\bds 1_q)\}}{\|\bds u\circ \vecc\{T_{k_2}(\bds 1_q)\}\|}$. This implies $$\E\big\|\mathcal{H}^{\mathcal{T}}_{n,k_1,k_2}  - \E\mathcal{H}^{\mathcal{T}}_{n,k_1,k_2} \big\|_{2}^2\precsim \E\big\|\mathcal{H}^{\mathcal{B}}_{n,k_1,k_2}  - \E\mathcal{H}^{\mathcal{B}}_{n,k_1,k_2} \big\|_{2}^2,$$
and thus by \eqref{lm:srate:HB} the convergence rate of $\E\big\|\mathcal{H}^{\mathcal{T}}_{n,k_1,k_2}  - \E\mathcal{H}^{\mathcal{T}}_{n,k_1,k_2} \big\|_{2}^2$ is the same as the rate $\begin{cases}
\frac{pk_1 + qk_2}{n} & pk_1 + qk_2 \precsim n
\\
\Big(\frac{pk_1 + qk_2}{n}\Big)^2 & pk_1 + qk_2 \succ n.
\end{cases}$ in \eqref{lm:srate:ratestep1}.
\par
\noindent\textbf{(b): }Similar to \eqref{lm:srate:HB2}, we have
\bee\nonumber
\|\E \mathcal{H}^{\mathcal{T}}_{n,k_1,k_2}\|_2^2 &= \Big\|\frac{1}{n}\vecc\big\{\M \Sigma_2^{*,\mathcal{T}}(k_2)\big\}\vecc\big\{\M \Sigma_2^{*,\mathcal{T}}(k_2)\big\}^\T\Big\|_2^2 
\\
&=\frac{1}{n^2}\|\vecc\big\{\M \Sigma_2^{*,\mathcal{T}}(k_2)\big\}\vecc\big\{\M \Sigma_2^{*,\mathcal{T}}(k_2)\big\}^\T\|_{\F}^2 
\\
& = \frac{1}{n^2}\big\|\xi\big\{\M \Sigma_2^{*,\mathcal{T}}(k_2)\otimes \M \Sigma_1^{*,\mathcal{T}}(k_1)\big\}\big\|_\F^2
\\
&=\frac{1}{n^2}\|\M \Sigma_2^{*,\mathcal{T}}(k_2)\otimes\M \Sigma_1^{*,\mathcal{T}}(k_1)\|_\F^2 
\\
&= \frac{1}{n^2}\|\M \Sigma_2^{*,\mathcal{T}}(k_2)\|_\F^2\|\M \Sigma_1^{*,\mathcal{T}}(k_1)\|_\F^2
\\
&\leq \frac{1}{n^2}\|\M \Sigma_2^{*,\mathcal{B}}(k_2)\|_\F^2\|\M \Sigma_1^{*,\mathcal{B}}(k_1)\|_\F^2
\\
&\precsim\frac{pq}{n^2},
\ee
where $\|\M \Sigma_1^{*,\mathcal{T}}(k_1)\|_\F^2\leq \|\M \Sigma_1^{*,\mathcal{B}}(k_1)\|_\F^2, \|\M \Sigma_2^{*,\mathcal{T}}(k_2)\|_\F^2\leq \|\M \Sigma_2^{*,\mathcal{B}}(k_2)\|_\F^2$ holds by definition. Thus, similarly to {Step 2}, the $\eta = \mathcal{B}$ case, the convergence rate of $\|\E \mathcal{H}^{\mathcal{T}}_{n,k_1,k_2}\|_2^2$ is also negligible compared to the rate $\begin{cases}
\frac{pk_1 + qk_2}{n} & pk_1 + qk_2 \precsim n
\\
\Big(\frac{pk_1 + qk_2}{n}\Big)^2 & pk_1 + qk_2 \succ n.
\end{cases}$ in \eqref{lm:srate:ratestep1}. 
\\
\par
\noindent\textbf{Step 3:} With \eqref{lm:srate:tstart} and \eqref{lm:srate:t2}, we can summarize the results above and conclude $\E{\|\xi\big\{\tilde{\M \Sigma}_{\mathcal{T}}(k_1,k_2)\big\} - \xi\{\M\Sigma^{*,\mathcal{T}}_2(k_2)\otimes \M\Sigma^{*,\mathcal{T}}_1(k_1)\}\|_{2}^2}$ also has the convergence rate, $\begin{cases}
\frac{pk_1 + qk_2}{n} & pk_1 + qk_2 \precsim n
\\
\Big(\frac{pk_1 + qk_2}{n}\Big)^2 & pk_1 + qk_2 \succ n.
\end{cases}$ same with \eqref{lm:srate:ratestep1}. \\
This is equivalent to the statement that $\frac{1}{pq}\E{\|\xi\big\{\tilde{\M \Sigma}_{\mathcal{T}}(k_1,k_2)\big\} - \xi\{\M\Sigma^{*,\mathcal{T}}_2(k_2)\otimes \M\Sigma^{*,\mathcal{T}}_1(k_1)\}\|_{2}^2}$ has same convergence rate as $\frac{1}{pq}\E{\|\xi\big\{\tilde{\M \Sigma}_{\mathcal{B}}(k_1,k_2)\big\} - \xi\{\M\Sigma^{*,\mathcal{B}}_2(k_2)\otimes \M\Sigma^{*,\mathcal{B}}_1(k_1)\}\|_{2}^2}$ in \eqref{lm:srate:b}, which finishes our proof.
\end{proof}

%Recalling the notation in Section \ref{notation:robust}, we use the following Lemmas \ref{lemma:robust:1}--\ref{lemma:robust:3} to prove Theorem \ref{T4}. 
{\color{black} Here we present Lemmas \ref{lemma:robust:1}--\ref{lemma:robust:3} that are used to prove Theorem \ref{T4}, which shows the convergence rate of the proposed robust covariance estimate. }
\begin{lemma}\label{lemma:robust:1}
Let $\vecc(\M X_1),\vecc(\M  X_2),\cdots,\vecc(\M X_n)$ be i.i.d random vectors in $\RR^{pq}$ with true covariance $\M \Sigma^* = \M \Sigma^*_2 \otimes \M \Sigma^*_1$. Assume $\E(|x^{(i)}_{l_1,l_2}\cdot x^{(i)}_{m_1,m_2}|^2) \leq M < +\infty$ for $1\leq i \leq n, 1\leq l_1,m_1 \leq p$, $1\leq l_2,m_2 \leq q$, where $M$ is a constant that does not depend on $n, l_1,m_1, l_2,m_2$.
\par
For the proposed robust banded estimator, when $\M \Sigma_1^{*}\in \mathcal{F}(\varepsilon_0, \alpha_1),\M \Sigma_2^{*} \in \mathcal{F}(\varepsilon_0, \alpha_2)$ or $\M \Sigma_1^{*}\in \mathcal{M}(\varepsilon_0, \alpha_1),\M \Sigma_2^{*} \in \mathcal{M}(\varepsilon_0, \alpha_2)$, for $\eta\in\{\mathcal{B},\mathcal{T}\}$ we have
\begin{align}\label{l:rfrate:res1}
\frac{1}{pq}\E\big\|\widecheck{\M \Sigma}_{\mathcal{\eta}}(k_1,k_2) - \M \Sigma_{\mathcal{R}}^{*,\eta}(k_1,k_2)\big\|^2_\F \precsim \frac{k_1k_2}{n
}.
\end{align}
\end{lemma}
\begin{proof}[Proof of Lemma \ref{lemma:robust:1}]%In the following proof, we use the notation introduced in Section \ref{sec:notation}.
\par
The notation of this proof is mainly contained in Section \ref{notation:robust}. In Section \ref{notation:robust}, we have shown $\widecheck{\M\Sigma}$ as a sample covariance estimator of i.i.d. random vectors $\vecc({\widecheck{\M X}}^c_1),\dots,\vecc({\widecheck{\M X}}^c_n)\in \RR^{pq}$, and have defined $\widecheck{\M\Sigma}_{\eta}(k_1,k_2)$ with $\eta\in\{\mathcal{B},\mathcal{T}\}$ as the doubly banded/tapering matrix of $\widecheck{\M\Sigma}$. This is similar to conditions of Lemma \ref{l:frate} that $\hat{\M \Sigma}$ is a sample covariance estimator of i.i.d. random vectors $\vecc(\M X_1),\dots,\vecc(\M X_n)\in\RR^{pq}$ and $\hat{\M \Sigma}_\eta(k_1,k_2)$ is doubly banded/tapering matrix of $\hat{\M \Sigma}$.
\par
 Furthermore, since $\widecheck{x}^{(i)}_{l_1,l_2} = \text{sgn}(x^{(i)}_{l_1,l_2})(|x^{(i)}_{l_1,l_2}|\wedge \tau)$, we have $|\widecheck{x}^{(i)}_{l_1,l_2}| \leq |{x}^{(i)}_{l_1,l_2}|$ and 
\bee\label{lm:robust:1:1}
&\E\Bigg(\Big|\big\{\widecheck{x}^{(i)}_{l_1,l_2} - \E\widecheck{x}^{(i)}_{l_1,l_2}\big\}\cdot \big\{\widecheck{x}^{(i)}_{m_1,m_2} - \E\widecheck{x}^{(i)}_{m_1,m_2}\big\}\Big|^2\Bigg)
\\
&= \E\Bigg(\Big|\widecheck{x}^{(i)}_{l_1,l_2}\widecheck{x}^{(i)}_{m_1,m_2} -\widecheck{x}^{(i)}_{l_1,l_2}\E\widecheck{x}^{(i)}_{m_1,m_2} - \widecheck{x}^{(i)}_{m_1,m_2}\E\widecheck{x}^{(i)}_{l_1,l_2} +\E\widecheck{x}^{(i)}_{l_1,l_2} \E\widecheck{x}^{(i)}_{m_1,m_2}\big\}\Big|^2\Bigg)
\\
&\precsim \E\Big(\Big|\widecheck{x}^{(i)}_{l_1,l_2}\widecheck{x}^{(i)}_{m_1,m_2}\Big|^2\Big) + \E\Big(\Big|\widecheck{x}^{(i)}_{l_1,l_2}\E\widecheck{x}^{(i)}_{m_1,m_2}\Big|^2\Big) + \E\Big(\Big|\widecheck{x}^{(i)}_{m_1,m_2}\E\widecheck{x}^{(i)}_{l_1,l_2}\Big|^2\Big) + \E\Big(\Big|\E\widecheck{x}^{(i)}_{m_1,m_2}\E\widecheck{x}^{(i)}_{l_1,l_2}\Big|^2\Big)
\\
&\precsim \E\Big(\Big|{x}^{(i)}_{l_1,l_2}{x}^{(i)}_{m_1,m_2}\Big|^2\Big) + \E\Big(\Big|{x}^{(i)}_{l_1,l_2}\E\widecheck{x}^{(i)}_{m_1,m_2}\Big|^2\Big) + \E\Big(\Big|{x}^{(i)}_{m_1,m_2}\E\widecheck{x}^{(i)}_{l_1,l_2}\Big|^2\Big) + \Big|\E\widecheck{x}^{(i)}_{m_1,m_2}\Big|^2\times\Big|\E\widecheck{x}^{(i)}_{l_1,l_2}\Big|^2,
\ee
where the first inequality holds by triangle inequality and the second inequality holds by $|\widecheck{x}^{(i)}_{l_1,l_2}| \leq |{x}^{(i)}_{l_1,l_2}|$. Then by Cauchy-Schwarz inequality we have
\bee\label{lm:robust:1:2}
\big|\E\widecheck{x}^{(i)}_{l_1,l_2}\big|^2 &\leq \E\big(\big|\widecheck{x}^{(i)}_{l_1,l_2}\big|^2\big)\leq \E\big(\big|{x}^{(i)}_{l_1,l_2}\big|^2\big) \leq\Big\{\E\big(\big|{x}^{(i)}_{l_1,l_2}\big|^4\big)\Big\}^{1/2}
\leq \sqrt{M}
\ee
for any $1\leq l_1 \leq p, 1\leq  l_2 \leq q$. The finial inequality above holds because $\E(|x^{(i)}_{l_1,l_2}\cdot x^{(i)}_{m_1,m_2}|^2) \leq M $ for any $1\leq l_1,m_1\leq p$ and $1\leq l_2,m_2 \leq q$,  Finally, by combining \eqref{lm:robust:1:1} and \eqref{lm:robust:1:2}, we have the following \textit{finite} entrywise fourth order moment bound for $\vecc(\widecheck{\M X}^c_i)$,
\bee\nonumber
&\E\Bigg(\Big|\big\{\widecheck{x}^{(i)}_{l_1,l_2} - \E\widecheck{x}^{(i)}_{l_1,l_2}\big\}\cdot \big\{\widecheck{x}^{(i)}_{m_1,m_2} - \E\widecheck{x}^{(i)}_{m_1,m_2}\big\}\Big|^2\Bigg)
\\
&\precsim  \E\Big(\Big|{x}^{(i)}_{l_1,l_2}{x}^{(i)}_{m_1,m_2}\Big|^2\Big) + M\times\E\Big(\Big|{x}^{(i)}_{l_1,l_2}\Big|^2\Big) + M\times\E\Big(\Big|{x}^{(i)}_{m_1,m_2}\Big|^2\Big) + M
\\
&\leq \E\Big(\Big|{x}^{(i)}_{l_1,l_2}{x}^{(i)}_{m_1,m_2}\Big|^2\Big) + M\times\Big\{\E\Big(\Big|{x}^{(i)}_{l_1,l_2}\Big|^4\Big)\Big\}^{1/2} + M\times\Big\{\E\Big(\Big|{x}^{(i)}_{m_1,m_2}\Big|^4\Big)\Big\}^{1/2} + M
\\
&\leq M + M\times{M}^{1/2} + M \times M^{1/2} + M
\\
&\leq 2M + 2M^{3/2},
\ee
where the second inequality holds by Cauchy-Schwarz inequality and the third inequality holds by $\E(|x^{(i)}_{l_1,l_2}\cdot x^{(i)}_{m_1,m_2}|^2) \leq M $, for any $1\leq l_1,m_1\leq p$ and $1\leq l_2,m_2 \leq q$. This fourth moment bound for $\vecc(\widecheck{\M X}^c_i)$ is also similar to the fourth moment bound for $\vecc(\M X_i)$ in Lemma \ref{l:frate} that $\E(|x^{(i)}_{l_1,l_2}\cdot x^{(i)}_{m_1,m_2}|^2) \leq M$. Also, $\vecc(\widecheck{\M X}^c_i)$ is  mean zero. The $\M \Sigma_{\mathcal{R}}^{*,\mathcal{\eta}}$ in Lemma \ref{lemma:robust:1} is the doubly banded/tapering matrix of $\cov(\widecheck{\M X}_i)$ and it can be seen as an analogy of $\M \Sigma_2^{*,\eta}(k_2)\otimes \M \Sigma_1^{*,\eta}(k_1)$ which is the doubly banded/tapering matrix of $\cov(\M X_i)$, i.e., $\M\Sigma^*$.
\par
In sum, based on all the similarities of conditions in Lemma \ref{l:frate} and Lemma \ref{lemma:robust:1}, we can directly use the entrywise proof arguments we use in {Proof of Lemma \ref{l:frate}} to prove Lemma \ref{lemma:robust:1}.
\par
For the proposed robust banded estimator, by applying similar arguments in the derivation of \eqref{T_3_1}--\eqref{l:f:final:1}, we can show $\frac{1}{pq}\E\big\|\widecheck{\M \Sigma}_{\mathcal{\mathcal{B}}}(k_1,k_2) - \M \Sigma_{\mathcal{R}}^{*,\mathcal{B}}(k_1,k_2)\big\|^2_\F \precsim \frac{k_1k_2}{n
}$.
\par
For the proposed robust tapering estimator, by applying similar arguments in the derivation of \eqref{l:frate:t1}--\eqref{l:frate:t3}, we  have
$
\frac{1}{pq}\E\big\|\widecheck{\M \Sigma}_{\mathcal{\mathcal{T}}}(k_1,k_2) - \M \Sigma_{\mathcal{R}}^{*,\mathcal{T}}(k_1,k_2)\big\|^2_\F \precsim \frac{k_1k_2}{n} .
$
\end{proof}
To account for diverging $ \tau$, we have the following new results.
\begin{lemma}\label{lemma:robust:2}
Let $\vecc(\M X_1),\vecc(\M  X_2),\cdots,\vecc(\M X_n)$ be i.i.d random vectors in $\RR^{pq}$ with true covariance $\M \Sigma^* = \M \Sigma^*_2 \otimes \M \Sigma^*_1$ where $\M \Sigma_1^{*}\in \mathcal{F}(\varepsilon_0, \alpha_1),\M \Sigma_2^{*} \in \mathcal{F}(\varepsilon_0, \alpha_2)$ or $\M \Sigma_1^{*}\in \mathcal{M}(\varepsilon_0, \alpha_1),\M \Sigma_2^{*} \in \mathcal{M}(\varepsilon_0, \alpha_2)$. 
\par
Under Assumptions \ref{A:minev}--\ref{A:psi2}, for $\eta\in\{\mathcal{B},\mathcal{T}\}$ we have 
\bee\label{lemma:srate:res:rb2}
\frac{1}{pq}\E\big\|\xi\{\widecheck{\M \Sigma}_{\eta}(k_1,k_2)\} - \xi\{{\M \Sigma}^{*,\eta}_{\mathcal{R}}(k_1,k_2)\}\big\|_2^2\precsim (J_n\tau)^4\times\begin{cases}
\frac{k_1}{qn} + \frac{k_2}{pn} & {\rm ~if~} pk_1 + qk_2 \precsim n
\\
\frac{pk^2_1}{qn^2} + \frac{qk^2_2}{pn^2} & {\rm ~if~} pk_1 + qk_2 \succ n. 
\end{cases}
\ee 
\end{lemma}
\begin{proof}[Proof of Lemma \ref{lemma:robust:2}]
The notation of this proof is mainly contained in Section \ref{notation:robust}. The conditions of Lemma \ref{lemma:robust:2} for samples $\vecc(\M X_1),\dots,\vecc(\M X_n)$, are similar to the conditions of Lemma \ref{lemma:srate} for samples $\vecc(\widecheck{\M X}^c_1),\dots,\vecc(\widecheck{\M X}^c_n)$. We can directly use similar arguments in the proof of Lemma \ref{lemma:srate}, to finish our proof here. For simplicity, we only discuss the condition that $\eta = \mathcal{B}$. 
\par
The only difference between this proof and the proof of Lemma \ref{lemma:srate}  is that in Lemma \ref{lemma:srate}, $\vecc(\M X_1),\dots,\vecc(\M X_n)$ are sub-Gaussian random variables with fixed parameter $\rho$. By {Proposition 2.5.2} in \citet{vershynin2018high}, we know $\|\vecc(\M X_i)\|_{\psi_2} \leq C_\rho$ for some fixed constant $C_\rho$ only depending on $\rho$. But for $\widecheck{\M X}^c_i$, the $\psi_2$--norm is not bounded. Since $\widecheck{x}^{(i)}_{l_1,l_2} = \text{sgn}(x^{(i)}_{l_1,l_2})(|x^{(i)}_{l_1,l_2}|\wedge \tau)$ for all $1\leq l_1\leq p, 1\leq l_2\leq q$, we know that all coordinates of $\vecc(\widecheck{\M X}_i)$ are bounded in $[-\tau,\tau]$. So the absolute values of all coordinates of $\E\vecc(\widecheck{\M X}_i)$ are bounded by $\tau$. Combining with the fact that all coordinates of $\E\vecc(\widecheck{\M X}_i)$ are constants and thus independent, with some fixed constant $C_K>0$, we have 
\begin{align*}
\|\E\vecc(\widecheck{\M X}_i)\|_{\psi_2} &\leq C_{K}\max_{1\leq l_1\leq p, 1\leq l_2\leq q}\big\|\E[\widecheck{\M X}_i]_{l_1,l_2}\big\|_{\psi_2}
\\
&\leq C_K\tau/\sqrt{\log(2)},
\end{align*}
where the first inequality is by Lemma 3.4.2 in \citet{vershynin2018high} and the second inequality is by (2.17) of \citet{vershynin2018high}. Under Assumption \ref{A:psi2}, we have
\bee\label{lemma:robust:2:hatx:psi2}
\|\vecc(\widecheck{\M X}_i^c)\|_{\psi_2} &= \|\vecc(\widecheck{\M X}_i) - \E\{\vecc(\widecheck{\M X}_i)\}\|_{\psi_2}
\\
&\leq \|\vecc(\widecheck{\M X}_i)\|_{\psi_2} + \|\E\{\vecc(\widecheck{\M X}_i)\}\|_{\psi_2}
\\
&\leq \big\{J_n + C_K/\sqrt{\log(2)}\big\}\times \tau
\\
&\precsim J_n\tau,
\ee
by $J_n\succsim 1$ and  triangle inequality.
\par
Similar to \eqref{step1:final}, similar arguments to {Step 1.1} in the {Proof of Lemma \ref{lemma:srate}} can show 
\bee\nonumber
\big\|\xi\{\widecheck{\M \Sigma}_{\eta}(k_1,k_2)\} - \xi\{{\M \Sigma}^{*,\eta}_{\mathcal{R}}(k_1,k_2)\}\big\|_2 = \sup_{{\bds v}^*\in \mathcal{U}^{\mathcal{B}}_{q^2}(k_2), \bds u^*\in \mathcal{U}^{\mathcal{B}}_{p^2}(k_1)}\Big|\frac{1}{n}\sum _{i = 1}^nI^{\mathcal{R}}_i(\bds u^*,\bds v^*)\Big|,
\ee
where $I^{\mathcal{R}}_i(\bds u^*,\bds v^*)\equiv \vecc(\widecheck{\M X}^c_i)^\T\cdot\mathscr{V}^*({\bds v}^*)\otimes\mathscr{U}({\bds u^*})\cdot\vecc(\widecheck{\M X}^c_i) - \E\big\{\vecc(\widecheck{\M X}^c_i)^\T\cdot\mathscr{V}^*({\bds v}^*)\otimes\mathscr{U}({\bds u^*})\cdot\vecc(\widecheck{\M X}^c_i)\big\}$. Here $\bds{\mathscr{U}}^*(\bds u^*), \bds{\mathscr{V}}^*(\bds v^*)$ are previously defined as $\bds{\mathscr{U}}^*(\bds u^*) \in \mathbb{R}^{p\times p}, \bds{\mathscr{V}}^*(\bds v^*)\in\mathbb{R}^{q\times q}$ such that $\vecc\{\bds{\mathscr{U}}^*(\bds u^*)\} = \bds u^*, \vecc\{\bds{\mathscr{V}}^*(\bds v^*)\} = \bds v^*$. We now let $\M M = (\M \Sigma_{\mathcal{R}}^*)^{1/2} \mathscr{V}^*({\bds v}^*)\otimes\mathscr{U}({\bds u^*})(\M \Sigma_{\mathcal{R}}^*)^{1/2}$ and $\mathbf{v}_s = (\M \Sigma_{\mathcal{R}}^*)^{-1/2}\vecc(\widecheck{\M X}^c_i)$ in the Hanson-Wright  inequality \eqref{l:srate:hwi}.  To account for the effect of the divergence of $\tau$, we observe
\bee\nonumber
\|\mathbf{v}^\T\mathbf{v}_s\|_{\psi_2} &= \|\mathbf{v}^\T(\M \Sigma_{\mathcal{R}}^*)^{-1/2}\vecc(\widecheck{\M X}^c_i)\|_{\psi_2}
\\
&= \|(\M \Sigma_{\mathcal{R}}^*)^{-1/2}\mathbf{v}\|\times \Big\|
\Big\{\frac{(\M \Sigma_{\mathcal{R}}^*)^{-1/2}\mathbf{v}}{\|(\M \Sigma_{\mathcal{R}}^*)^{-1/2}\mathbf{v}\|}\Big\}^\T\vecc(\widecheck{\M X}^c_i)\Big\|_{\psi_2}
\\
&\leq\|(\M \Sigma_{\mathcal{R}}^*)^{-1/2}\|_2 \times\|\mathbf{v}\|\times \big\{J_n + C_K/\sqrt{\log(2)}\big\}\tau \quad (\text{By \eqref{lemma:robust:2:hatx:psi2}})
\\
&= \lambda_{\min}^{-1/2}(\M\Sigma^*_{\mathcal{R}})\big\{J_n + C_K/\sqrt{\log(2)}\big\}\tau 
\\
&\leq \big\{J_n + C_K/\sqrt{\log(2)}\big\}\tau/\sqrt{\varepsilon_0'}  \quad (\text{By Assumption \ref{A:minev}})
\\
&\precsim J_n \tau.
\ee
Similar to \eqref{l:srate:con1}, we can also show
\bee\nonumber
\|\M M\|_\F &\leq \|(\M \Sigma_{\mathcal{R}}^*)^{1/2}\|_2^2\times \|\mathscr{V}^*({\bds v}^*)\otimes\mathscr{U}^*(\bds u^*)\|_2\leq 1/\varepsilon_0'
\ee
by Assumption \ref{A:minev} and $\|\mathscr{V}^*({\bds v}^*)\otimes\mathscr{U}^*(\bds u^*)\|_2 = \|\mathscr{V}^*({\bds v}^*)\|_2\cdot\|\mathscr{U}^*(\bds u^*)\|_2\leq \|\mathscr{V}^*({\bds v}^*)\|_\F\cdot\|\mathscr{U}^*(\bds u^*)\|_\F = \|{\bds v}^*\|\cdot\|\bds u^*\| = 1$. Similar arguments to \eqref{l:srate:hwi}--\eqref{l:srate:subexp} can show 
\bee\label{l:rb2:subexp}
\p\Big(\Big|\frac{1}{n}\sum_{i = 1}^nI^{\mathcal{R}}_{i}(\bds u^*,\bds v^*)\Big|\geq t\Big) \leq2\exp\Big[-nC^{(5)}\min\{t^2/(J_n\tau)^4,t/(J_n\tau)^2\}\Big]
\ee
for some constant $C^{(5)}>0$. In comparison with the tail probability of $\Big|\frac{1}{n}\sum_{i = 1}^nI^{}_{i}(\bds u^*,\bds v^*)\Big|$ in \eqref{l:srate:subexp} for the non-robust estimators, the new tail probability for the robust estimators \eqref{l:rb2:subexp} includes $J_n\tau$. Then, by using the new tail bound \eqref{l:rb2:subexp}, a similar argument to the {Step 1.3}--{Step 3} in proof of Lemma \ref{lemma:srate} can show \eqref{lemma:srate:res:rb2} directly, for $\eta = \mathcal{B}$. Same as the proof of Lemma \ref{lemma:srate}, the results for proposed tapering estimator can be shown similarly.
\end{proof}

\begin{lemma}\label{lemma:rb}
Let $\vecc(\M X_1),\vecc(\M  X_2),\cdots,\vecc(\M X_n)$ be i.i.d random vectors in $\RR^{pq}$ with true covariance $\M \Sigma^* = \M \Sigma^*_2 \otimes \M \Sigma^*_1$. Assume $\E(|x^{(i)}_{l_1,l_2}\cdot x^{(i)}_{m_1,m_2}|^\zeta) \leq M < +\infty$ for any $1\leq l_1,m_1\leq p$ and $1\leq l_2,m_2 \leq q$, where $\zeta > 1$ is the order of heavy-tailedness, and $M$ is a constant that does not depend on $n, l_1,m_1, l_2,m_2$. Then we have 
\bee\nonumber
\|\M\Sigma^*_{\mathcal{R}} - \M\Sigma^*\|_{\max} \precsim \tau^{-2(\zeta - 1)}
\ee 
\end{lemma}
\begin{proof}[Proof of Lemma \ref{lemma:rb}] The notation of this proof is mainly contained in Section \ref{notation:robust}.{ Similar to} our previous proofs, without loss of generality, we assume $\vecc(\M X_i)$ is mean zero. Since $\cov\{\vecc(\widecheck{\M X}_i)\} = \M \Sigma^*_{\mathcal{R}}$ and $\cov\{\vecc(\M X_i)\} = \M \Sigma^*$, we know any entry of $\M\Sigma^*_{\mathcal{R}} - \M\Sigma^*$ can be written as
\bee\label{lemma:rb:1}
&\E\Big[(\widecheck{x}^{(i)}_{l_1,l_2} - \E \widecheck{x}^{(i)}_{l_1,l_2})\times (\widecheck{x}^{(i)}_{m_1,m_2} - \E \widecheck{x}^{(i)}_{m_1,m_2})\Big] - \E\Big[x^{(i)}_{l_1,l_2}\times x^{(i)}_{m_1,m_2}\Big]
\\
&=\E\Big[\widecheck{x}^{(i)}_{l_1,l_2}\widecheck{x}^{(i)}_{m_1,m_2} - \E\widecheck{x}^{(i)}_{l_1,l_2} \times \E\widecheck{x}^{(i)}_{m_1,m_2}\Big] - \E\Big[x^{(i)}_{l_1,l_2}\times x^{(i)}_{m_1,m_2}\Big]
\\
&=\E\Big[\widecheck{x}^{(i)}_{l_1,l_2}\widecheck{x}^{(i)}_{m_1,m_2} - x^{(i)}_{l_1,l_2} x^{(i)}_{m_1,m_2}\Big] -\E\widecheck{x}^{(i)}_{l_1,l_2} \times \E\widecheck{x}^{(i)}_{m_1,m_2}
\\
&=\E\Big[\widecheck{x}^{(i)}_{l_1,l_2}\widecheck{x}^{(i)}_{m_1,m_2} - x^{(i)}_{l_1,l_2} x^{(i)}_{m_1,m_2}\Big] -\Big[\E\widecheck{x}^{(i)}_{l_1,l_2} - \E x^{(i)}_{l_1,l_2} \Big]\times \Big[\E\widecheck{x}^{(i)}_{m_1,m_2} - \E x^{(i)}_{m_1,m_2} \Big]
\ee
for any $1\leq l_1,m_1 \leq p$ and $1\leq l_2,m_2 \leq q$. Here the last equality holds because $\vecc(\M X_i)$ is mean zero. 
\par
For the first term on the right-hand side of \eqref{lemma:rb:1}, since $\widecheck{x}^{(i)}_{l_1,l_2} = \text{sgn}(x^{(i)}_{l_1,l_2})(|x^{(i)}_{l_1,l_2}| \wedge \tau)$, we can show 
\bee\nonumber
&\big|x^{(i)}_{l_1,l_2} x^{(i)}_{m_1,m_2} - \widecheck{x}^{(i)}_{l_1,l_2}\widecheck{x}^{(i)}_{m_1,m_2}\big|  
\\
& = \Big|\text{sgn}\big(x^{(i)}_{l_1,l_2}\big)\times \text{sgn}\big(x^{(i)}_{m_1,m_2}\big)\times\Big[|x^{(i)}_{l_1,l_2}|\times|x^{(i)}_{m_1,m_2}| - \big\{|x^{(i)}_{l_1,l_2}| \wedge \tau\big\}\times \big\{|x^{(i)}_{m_1,m_2}| \wedge \tau\big\}\Big]\Big|
\\
&=\Big||x^{(i)}_{l_1,l_2}|\times|x^{(i)}_{m_1,m_2}| - \big\{|x^{(i)}_{l_1,l_2}| \wedge \tau\big\}\times \big\{|x^{(i)}_{m_1,m_2}| \wedge \tau\big\}\Big|
\\
&=
\begin{cases}
0 & |x^{(i)}_{l_1,l_2}|\leq \tau, |x^{(i)}_{m_1,m_2}|\leq \tau;
\\
\Big||x^{(i)}_{l_1,l_2}x^{(i)}_{m_1,m_2}| -  \tau \big|x^{(i)}_{m_1,m_2}\big| \Big| & |x^{(i)}_{l_1,l_2}|> \tau, |x^{(i)}_{m_1,m_2}|\leq \tau;
\\
\Big||x^{(i)}_{l_1,l_2}x^{(i)}_{m_1,m_2}| -  \tau \big|x^{(i)}_{l_1,l_2}\big| \Big| & |x^{(i)}_{m_1,m_2}|> \tau, |x^{(i)}_{l_1,l_2}|\leq \tau;
\\
\Big||x^{(i)}_{l_1,l_2}x^{(i)}_{m_1,m_2}| -  \tau^2 \Big| & |x^{(i)}_{m_1,m_2}|> \tau, |x^{(i)}_{l_1,l_2}|> \tau
\end{cases}
\\
&\leq |x^{(i)}_{l_1,l_2}x^{(i)}_{m_1,m_2}|\times \M I(|x^{(i)}_{l_1,l_2}|> \tau\text{ or } |x^{(i)}_{m_1,m_2}|> \tau).
\ee
For the last inequality above, we give it a case-by-case explanation. When $|x^{(i)}_{l_1,l_2}| \leq \tau$ and $|x^{(i)}_{m_1,m_2}|\leq \tau$, the value on the right-hand side of the third equality above equals $0$. When $|x^{(i)}_{l_1,l_2}| > \tau$ and $|x^{(i)}_{m_1,m_2}|\leq \tau$, we have $|x^{(i)}_{l_1,l_2}x^{(i)}_{m_1,m_2}| >  \tau \big|x^{(i)}_{m_1,m_2}\big|$ and therefore $\big||x^{(i)}_{l_1,l_2}x^{(i)}_{m_1,m_2}| -  \tau |x^{(i)}_{m_1,m_2}| \big| = |x^{(i)}_{l_1,l_2}x^{(i)}_{m_1,m_2}| -  \tau \big|x^{(i)}_{m_1,m_2}\big| < |x^{(i)}_{l_1,l_2}x^{(i)}_{m_1,m_2}|$. When $|x^{(i)}_{m_1,m_2}|> \tau, |x^{(i)}_{l_1,l_2}|\leq \tau$, a symmetric argument can show $\big||x^{(i)}_{l_1,l_2}x^{(i)}_{m_1,m_2}| -  \tau |x^{(i)}_{l_1,l_2}| \big| = |x^{(i)}_{l_1,l_2}x^{(i)}_{m_1,m_2}| -  \tau \big|x^{(i)}_{l_1,l_2}\big| < |x^{(i)}_{l_1,l_2}x^{(i)}_{m_1,m_2}|$. When $|x^{(i)}_{m_1,m_2}| >  \tau, |x^{(i)}_{l_1,l_2}|> \tau$, we have $|x^{(i)}_{l_1,l_2}x^{(i)}_{m_1,m_2}| >  \tau^2 $ and therefore $\big||x^{(i)}_{l_1,l_2}x^{(i)}_{m_1,m_2}| -  \tau^2 \big| = |x^{(i)}_{l_1,l_2}x^{(i)}_{m_1,m_2}| -  \tau^2 \leq |x^{(i)}_{l_1,l_2}x^{(i)}_{m_1,m_2}|$. Combining all the above arguments, we obtain the last inquality. 
\par
We also note that for any $1\leq l_1\leq p, 1\leq l_2 \leq q$, by high-order Markov inequality, we have
\bee\label{l:robust:max:1}
\p\big(|x^{(i)}_{l_1,l_2}|> \tau\big) \leq \frac{\E|x^{(i)}_{l_1,l_2}|^{2\zeta}}{\tau^{2\zeta}}\leq \frac{M}{\tau^{2\zeta}}.
\ee
With all results above, we bound the first term on the right-hand side of \eqref{lemma:rb:1} by 
\bee\label{l:robust:max:2}
&\Big|\E\big[\widecheck{x}^{(i)}_{l_1,l_2}\widecheck{x}^{(i)}_{m_1,m_2} - x^{(i)}_{l_1,l_2} x^{(i)}_{m_1,m_2}\big]\Big|
\\
&\leq \E\Big|\widecheck{x}^{(i)}_{l_1,l_2}\widecheck{x}^{(i)}_{m_1,m_2} - x^{(i)}_{l_1,l_2} x^{(i)}_{m_1,m_2}\Big|
\\
&\leq \E\Big\{|x^{(i)}_{l_1,l_2}x^{(i)}_{m_1,m_2}|\times \M I(|x^{(i)}_{l_1,l_2}|> \tau\text{ or } |x^{(i)}_{m_1,m_2}|> \tau)\Big\}
\\
&\leq \Big[\E\Big\{|x^{(i)}_{l_1,l_2}x^{(i)}_{m_1,m_2}|^{\zeta}\Big\}\Big]^{1/\zeta}\times \Big[\E\Big\{\M I^{}(|x^{(i)}_{l_1,l_2}|> \tau\text{ or } |x^{(i)}_{m_1,m_2}|> \tau)\Big\}\Big]^{(\zeta - 1)/\zeta}
\\
&= \Big[\E\Big\{|x^{(i)}_{l_1,l_2}x^{(i)}_{m_1,m_2}|^{\zeta}\Big\}\Big]^{1/\zeta}\times \Big\{\p\big(|x^{(i)}_{l_1,l_2}|> \tau\text{ or } |x^{(i)}_{m_1,m_2}|> \tau\big)\Big\}^{(\zeta - 1)/\zeta}
\\
&\leq M^{1/\zeta} \times \Big\{\p\big(|x^{(i)}_{l_1,l_2}|> \tau\big)+\p\big(|x^{(i)}_{m_1,m_2}|> \tau\big)\Big\}^{(\zeta - 1)/\zeta}
\\
&\leq 2^{(\zeta - 1)/\zeta}M\cdot\tau^{-2(\zeta - 1)},
\ee
where the third inequality holds by H\"{o}lder  inequality.
\par
For the second item on the right-hand side of \eqref{lemma:rb:1}, since $\widecheck{x}^{(i)}_{l_1,l_2} = \text{sgn}(x^{(i)}_{l_1,l_2})(|x^{(i)}_{l_1,l_2}| \wedge \tau)$, we have
\bee\nonumber
\big|\widecheck{x}^{(i)}_{l_1,l_2} -  x^{(i)}_{l_1,l_2}\big|  &= \begin{cases}
0 & |x^{(i)}_{l_1,l_2}|\leq \tau
\\
|x^{(i)}_{l_1,l_2} - \tau| & |x^{(i)}_{l_1,l_2}|> \tau
\end{cases}
\\
&\leq |x_{l_1,l_2}^i|\times\M I(|x^{(i)}_{l_1,l_2}|> \tau).
\ee
Similar to \eqref{l:robust:max:1} and \eqref{l:robust:max:2}, we can derive
\bee\nonumber
\big|\E(\widecheck{x}^{(i)}_{l_1,l_2} -  x^{(i)}_{l_1,l_2})\big|&\leq \E\big|\widecheck{x}^{(i)}_{l_1,l_2} -  x^{(i)}_{l_1,l_2}\big|
\\
&\leq \E\Big[|x_{l_1,l_2}^i|\times\M I(|x^{(i)}_{l_1,l_2}|> \tau)\Big]
\\
&\leq \Big\{\E|x_{l_1,l_2}^i|^{2\zeta}\Big\}^{1/2\zeta}\times \big\{\p(|x^{(i)}_{l_1,l_2}|> \tau)\big\}^{\frac{2\zeta - 1}{2\zeta}}
\\
&\leq  \{\E|x_{l_1,l_2}^i|^{2\zeta}\Big\}^{1/2\zeta}\times \Bigg\{\frac{M}{\tau^{2\zeta}}\Bigg\}^{\frac{2\zeta - 1}{2\zeta}}
\\
&\leq M\cdot\tau^{-(2\zeta - 1)}
\ee
for any $1\leq l_1 \leq p, 1\leq l_2 \leq q$. Then the second term of \eqref{lemma:rb:1} can be bounded by
\bee\nonumber
\Big|\Big[\E\widecheck{x}^{(i)}_{l_1,l_2} - \E x^{(i)}_{l_1,l_2} \Big]\times \Big[\E\widecheck{x}^{(i)}_{m_1,m_2} - \E x^{(i)}_{m_1,m_2} \Big]\Big|&=\big|\E(\widecheck{x}^{(i)}_{l_1,l_2} -  x^{(i)}_{l_1,l_2})\big|\times \big|\E(\widecheck{x}^{(i)}_{m_1,m_2} -  x^{(i)}_{m_1,m_2})\big|
\\
&\leq M^{2}\cdot \tau^{-2(2\zeta - 1)}.
\ee
\par
Summarizing the results above, we final show
\bee\nonumber
\|\M\Sigma^*_{\mathcal{R}} - \M \Sigma^*\|_{\max} &= \max_{1\leq l_1,m_1 \leq p\atop 1\leq l_2,m_2\leq q}\Bigg|\E\Big[(\widecheck{x}^{(i)}_{l_1,l_2} - \E \widecheck{x}^{(i)}_{l_1,l_2})\times (\widecheck{x}^{(i)}_{m_1,m_2} - \E \widecheck{x}^{(i)}_{m_1,m_2})\Big] - \E\Big[x^{(i)}_{l_1,l_2}\times x^{(i)}_{m_1,m_2}\Big]\Bigg|
\\
&=\max_{1\leq l_1,m_1 \leq p\atop 1\leq l_2,m_2\leq q}\Big|\E\Big[\widecheck{x}^{(i)}_{l_1,l_2}\widecheck{x}^{(i)}_{m_1,m_2} - x^{(i)}_{l_1,l_2} x^{(i)}_{m_1,m_2}\Big]\Big| 
\\
&+\max_{1\leq l_1,m_1 \leq p\atop 1\leq l_2,m_2\leq q}\Big\{\Big|\E\widecheck{x}^{(i)}_{l_1,l_2} - \E x^{(i)}_{l_1,l_2} \Big|\times \Big|\E\widecheck{x}^{(i)}_{m_1,m_2} - \E x^{(i)}_{m_1,m_2} \Big|\Big\}
\\
&\leq 2^{(\zeta - 1)/\zeta}M\cdot\tau^{-2(\zeta - 1)} + M^{2}\cdot \tau^{-2(2\zeta - 1)}
\\
&\precsim \tau^{-2(\zeta - 1)},
\ee
where the last inequality holds because $\tau \succsim 1$ and $2(2\zeta - 1)>2(\zeta - 1)$ when $\zeta > 1$.
\end{proof}

\begin{lemma}\label{lemma:robust:3}
Let $\vecc(\M X_1),\vecc(\M  X_2),\cdots,\vecc(\M X_n)$ be i.i.d random vectors in $\RR^{pq}$ with true covariance $\M \Sigma^* = \M \Sigma^*_2 \otimes \M \Sigma^*_1$. Assume $\E(|x^{(i)}_{l_1,l_2}\cdot x^{(i)}_{m_1,m_2}|^\zeta) \leq M < +\infty$ for any $1\leq l_1,m_1\leq p$ and $1\leq l_2,m_2 \leq q$, where $\zeta > 1$ and $M$ is a constant that does not depend on $n, l_1,m_1, l_2,m_2$. When $\M \Sigma_1^{*}\in \mathcal{F}(\varepsilon_0, \alpha_1),\M \Sigma_2^{*} \in \mathcal{F}(\varepsilon_0, \alpha_2)$ or $\M \Sigma_1^{*}\in \mathcal{M}(\varepsilon_0, \alpha_1),\M \Sigma_2^{*} \in \mathcal{M}(\varepsilon_0, \alpha_2)$, we have 
\bee\nonumber
\frac{1}{pq}\big\|{\M \Sigma}^{*,\eta}_{\mathcal{R}}(k_1,k_2)-\M\Sigma_2^{*,{\eta}}(k_2)\otimes \M\Sigma_1^{*,{\eta}}(k_1)\big\|_\F^2 \precsim {k_1k_2}\tau^{-4(\zeta - 1)}.
\ee
\end{lemma}
\begin{proof}[Proof of Theorem \ref{lemma:robust:3}] The notation of this proof is mainly contained in Section \ref{notation:robust}. {We only consider the general case that $\M \Sigma_1^{*}\in \mathcal{F}(\varepsilon_0, \alpha_1),\M \Sigma_2^{*} \in \mathcal{F}(\varepsilon_0, \alpha_2)$. The bound for $\M \Sigma_1^{*}\in \mathcal{M}(\varepsilon_0, \alpha_1),\M \Sigma_2^{*} \in \mathcal{M}(\varepsilon_0, \alpha_2)$ holds as a special case.  By definitions of ${\M \Sigma}^{*,\eta}_{\mathcal{R}}(k_1,k_2)$ and $\M\Sigma_2^{*,{\eta}}(k_2)\otimes \M\Sigma_1^{*,{\eta}}(k_1)$, we know
\bee\label{l:robust:3:1}
&{\M \Sigma}^{*,\mathcal{B}}_{\mathcal{R}}(k_1,k_2) - \M\Sigma_2^{*,{\mathcal{B}}}(k_2)\otimes \M\Sigma_1^{*,{\mathcal{B}}}(k_1) =\big\{{\M \Sigma}^{*}_{\mathcal{R}}-\M \Sigma^*\big\}\circ \big\{B_{k_2}(\M 1_q)\otimes B_{k_1}(\M 1_p)\big\};
\\
&{\M \Sigma}^{*,\mathcal{T}}_{\mathcal{R}}(k_1,k_2) - \M\Sigma_2^{*,{\mathcal{T}}}(k_2)\otimes \M\Sigma_1^{*,{\mathcal{T}}}(k_1) =\big\{{\M \Sigma}^{*}_{\mathcal{R}}-\M \Sigma^*\big\}\circ \big\{T_{k_2}(\M 1_q)\otimes T_{k_1}(\M 1_p)\big\}.
\ee
So there are at most $(2k_1 + 1)(2k_2 + 1)pq$ non-zero entries in ${\M \Sigma}^{*,\eta}_{\mathcal{R}}(k_1,k_2) - \M\Sigma_2^{*,{\eta}}(k_2)\otimes \M\Sigma_1^{*,{\eta}}(k_1)$ for both $\eta =\mathcal{B}$ or $\mathcal{T}$. Also by definition of $T_k(\cdot)$, $B_k(\cdot)$ and \eqref{l:robust:3:1}, we know that each non-zero entry in ${\M \Sigma}^{*,\eta}_{\mathcal{R}}(k_1,k_2) - \M\Sigma_2^{*,{\eta}}(k_2)\otimes \M\Sigma_1^{*,{\eta}}(k_1)$ is either equal to the corresponding entry in $\M\Sigma^*_{\mathcal{R}} - \M \Sigma^*$, or is a shrinkage of corresponding entry in $\M\Sigma^*_{\mathcal{R}} - \M \Sigma^*$. Thus we have
\bee\nonumber
\big\|{\M \Sigma}^{*,\eta}_{\mathcal{R}}(k_1,k_2) - \M\Sigma_2^{*,{\eta}}(k_2)\otimes \M\Sigma_1^{*,{\eta}}(k_1)\big\|^2_\F \leq (2k_1 + 1)(2k_2 + 1)pq\times\|\M\Sigma^*_{\mathcal{R}} - \M\Sigma^*\|_{\max}^2.
\ee
Combining with Lemma \ref{lemma:rb}, we finally show
\bee\nonumber
\frac{1}{pq}\big\|{\M \Sigma}^{*,\eta}_{\mathcal{R}}(k_1,k_2) - \M\Sigma_2^{*,{\eta}}(k_2)\otimes \M\Sigma_1^{*,{\eta}}(k_1)\big\|^2_\F &\precsim \frac{1}{pq}(2k_1 + 1)(2k_2 + 1)pq\times\Big(\tau^{-2(\zeta - 1)}\Big)^2
\\
&\asymp{k_1 k_2}\cdot\tau^{-4(\zeta - 1)}.
\ee}
\end{proof}
{\color{black} Here we present Lemmas \ref{l:dt}--\ref{l:sp:out} that are used to prove Theorem \ref{Ts}, which shows the spectral-norm convergence rate of the doubly tapering covariance estimate. 
\begin{lemma}\label{subG:dev}
Let $\M X_1,\M X_2,\dots,\M X_n\in \RR^d$ be i.i.d. sub-Gaussian random variables with parameter $\rho$. Suppose the covariance matrix of $\M X_1 $ is $\M \Sigma$. If $\lambda_{\max}(\M\Sigma)\leq \lambda_0$, there exist $\zeta>0$ only determined by $\rho, \lambda_0$ and some constants $C_1,C_2>0$ such that,
\bee\label{l4:res}
\Pr\Big[\Big|\frac{1}{n}\sum_{i  =1}^{n}\mathbf{ u}^{\T}\big\{(\M X_i - \E \M X_i)(\M X_i - \E \M X_i)^{\T} - \M\Sigma\big\}\bold v\Big|>t\Big]\leq  
\begin{cases}
2\exp(-C_1n\rho^2t^2) & t \in [0,\zeta]

\\
2\exp(-C_2n\rho t) & t \in(\zeta,+\infty)
\end{cases}
\ee
for all $\bold u,\bold v \in \RR^d$, where $\|\bold u\|_2 = \|\bold v\|_2 = 1$.
\end{lemma}
\begin{proof}[Proof of Lemma \ref{subG:dev}] 
We give a proof based on sub-Gaussian/sub-Exponential random variable property and non-asymptotic concentration inequality. We note that similar result could also be shown via large deviation theory (see e.g., \cite{saulis1991limit}, \cite{Bickel2008threshold}). For random variable $X$, recall the definitions of $\|\cdot\|_{\psi_1}, \|\cdot\|_{\psi_2}$ given in Section \ref{def:sesgrv} as
\bee\label{def:psi}
\|X\|_{\psi_1} &\equiv \inf\{t>0:\E \exp(|X|/t)\leq 2\},
\\
\|X\|_{\psi_2} &\equiv \inf\{t>0:\E \exp(X^2/t^2)\leq 2\}.
\ee 
And recall that $X$ is a sub-exponential random variable if
\bee\nonumber
\E\Big[\exp\big(\lambda(X - \mu)\big)\Big]\leq \exp\Big(\frac{v^2\lambda^2}{2}\Big)
\ee
for some non-negative parameters $(v,\alpha)$ for all $|\lambda|<1/\alpha$.
\par
By definition \eqref{subga}, it is easy to see both $\mathbf{u}^{\T}(\M X_i - \E \M X_i)$ and $ \mathbf{v}^{\T}(\M X_i - \E \M X_i)$ are sub-Gaussian random variables parametrized with $\rho$. By Proposition 2.5.2 in \cite{vershynin2018high}, we have both $\|\mathbf{u}^{\T}(\M X_i - \E \M X_i)\|_{\psi_2}$ and $\|\mathbf{v}^{\T}(\M X_i - \E \M X_i)\|_{\psi_2} \precsim 1/\sqrt{\rho}$. By Lemma 2.7.7 in \cite{vershynin2018high}, $\mathbf{u}^{\T}(\M X_i - \E \M X_i)(\M X_i - \E \M X_i)^\T \mathbf{v} = \mathbf{u}^{\T}(\M X_i - \E \M X_i)\times \mathbf{v}^{\T}(\M X_i - \E \M X_i)$ is a sub-exponential random variable and 
\bee\nonumber
\|\mathbf{u}^{\T}(\M X_i - \E \M X_i)(\M X_i - \E \M X_i)^\T \mathbf{v}\|_{\psi_1}\precsim 1/\rho.
\ee
By a careful comparison of Proposition 2.7.1 in \cite{vershynin2018high} and our definition of sub-Gaussian, we can see that $\mathbf{u}^{\T}(\M X_i - \E \M X_i)(\M X_i - \E \M X_i)^\T \mathbf{v}$ is sub-exponential, parametrized with $(C/\rho,C'/\rho)$ for some non-negative fixed constants $C,C'$. By the concentration results of sub-exponential random variables (see e.g. (2.18) in \cite{wainwright2019high}), we finally have
\bee\nonumber
\Pr\Big(\Big|\frac{1}{n}\sum_{i  =1}^{n}\mathbf{ u}^{\T}\big\{(\M X_i - \E \M X_i)(\M X_i - \E \M X_i)^{\T} - \M\Sigma\big\}\bold v\Big|>t\Big)\leq
\begin{cases}
2\exp(-C_1n\rho^2t^2) & t \in [0,\zeta]

\\
2\exp(-C_2n\rho t) & t \in(\zeta,+\infty)
\end{cases}
\ee
where $\zeta = C_3/\rho$ for some fixed $C_1,C_2,C_3 > 0$.
\end{proof}
\begin{lemma}\label{l:dt}
Let $\vecc(\M X_1), \vecc(\M X_2),\dots,\vecc(\M X_n)$ be i.i.d. sub-Gaussian random vectors in $\RR^{pq}$ with true covariance $\M \Sigma^* = \M \Sigma^*_2 \otimes \M \Sigma^*_1$, where $\M \Sigma_1^{*}\in \mathcal{F}(\varepsilon_0, \alpha_1),\M \Sigma_2^{*} \in \mathcal{F}(\varepsilon_0, \alpha_2)$. We have\bee\nonumber
\E\big(\|\DTCov - \M \Sigma_2^{*,\mathcal{T}}(k_2)\otimes \M \Sigma_1^{*,\mathcal{T}}(k_1)\|_2^2\big) \precsim \frac{k_1k_2 + \log\big(\max\{p,q\}\big)}{n},
\ee
when $p \succsim n^{b}$ or $q \succsim n^{b}$ for some $b > 0$.  %{\color{blue}our theorem is also good for exponential order. Even $p \asymp exp(n^2)$ the error bound still holds. But when $p \asymp exp(n^2)$, the results deduced via this lemma is useless for consistency. This condition just rule some slow divergent case, which the banding is not necessary, i.e. $\hat{\M \Sigma}$ is rate optimal.}
\end{lemma}
\begin{proof}[Proof of Lemma \ref{l:dt}]
This proof can be seen as a generalization of  proof of Theorem 2 in \citet{Cai2010}. In Section \ref{notation:main}, we have defined $\hat{\M\Sigma}_0 \equiv \frac{1}{n} \sum_{i=1}^{n}\left(\vecc(\mathbf{X}_{i})\right)\left(\vecc(\mathbf{X}_{i})\right)^{\T} = [\hat{\sigma}^
{(0)}_{((l_1,m_1),(l_2,m_2))}]$ and $\DTCovZ \equiv \hat{\M \Sigma}_0 \HDTaper$. By triangle inequality,
\bee\label{lm:dt:main}
\E\|\DTCov - \M \Sigma_2^{*,\mathcal{T}}(k_2)\otimes \M \Sigma_1^{*,\mathcal{T}}(k_1)\|_2^2 &\precsim \E \|\DTCovZ - \M \Sigma_2^{*,\mathcal{T}}(k_2)\otimes \M \Sigma_1^{*,\mathcal{T}}(k_1)\|_2^2 
\\
& +\E \|\DTCovZ - \DTCov\|_2^2.
\ee
In the following steps 1 and 2, we bound first and second terms on the right-hand side of \eqref{lm:dt:main}, respectively.
\\
\par
\noindent{\textbf{Step 1 (Bound of $\|\DTCovZ - \M \Sigma_2^{*,\mathcal{T}}(k_2)\otimes \M \Sigma_1^{*,\mathcal{T}}(k_1)\|_2$): }}
By definition, $\DTCovZ$ has the form
\bee\label{def:Mr2k1k}
\begin{pmatrix}
T_{k_2}(\bds 1_{q})_{1,1}\times \hat{\M \Sigma}_0^{1,1} \circ T_{k_1}\big(\M 1_p\big) & \cdots\cdots & T_{k_2}(\bds 1_{q})_{1,q}\times \hat{\M \Sigma}_0^{1,q} \circ T_{k_1}\big(\M 1_p\big) \\
\vdots  & \vdots  & \vdots
\\
T_{k_2}(\bds 1_{q})_{q,1}\times \hat{\M \Sigma}_0^{q,1} \circ T_{k_1}\big(\M 1_p\big)  &  \dots    & T_{k_2}(\bds 1_{q})_{q,q}\times \hat{\M \Sigma}_0^{q,q} \circ T_{k_1}\big(\M 1_p\big)
\end{pmatrix},
\ee
where $\hat{\M \Sigma}_0^{l_2,m_2}$ is the $l_2 m_2$th $p\times p$ sub-block matrix of $\hat{\M \Sigma}_0$. In the following,  We use shorthand notation $\hat{\M \Sigma}^{l_2,m_2}_{0}(k_1) \equiv \hat{\M \Sigma}_0^{l_2,m_2} \circ T_{k_1}\big(\M 1_p\big)$ for simplicity. Then we define 
\bee\nonumber
\M M_{r}^{(2),k_1,k} = \Big[\hat{\M \Sigma}^{l_2,m_2}_{0}(k_1) \cdot \M I(r\leq l_2 < r + k, r \leq m_2 < r + k)\Big]_{1\leq l_2,m_2 \leq q} \in \mathbb{R}^{pq \times pq}.
\ee
Without loss of generality, we assume $k_2$ is an even number. Then a similar argument as the Proof of Lemma 1 in \citet{Cai2010} leads to
\bee\label{l:dt:eq}
\DTCovZ = (k_2/2)^{-1}\big(\M S^{(k_2)} - \M S^{(k_2/2)}\big),
\ee
where $\M S^{(k)} = \sum_{r = 1 - k}^{q} \M M_r^{(2),k_1,k}$. Then similar to Lemma 2 in \citet{Cai2010}, for a given $r$, we can see that $\M M^{(2),k_1,k}_{k + r},\M M^{(2),k_1,k}_{2k + r},\dots,\M M^{(2),k_1,k}_{q + r}$ are disjoint diagonal blocks, and thus
\bee\label{Skdecom}
\|\M S^{(k)}  - \E \M S^{(k)}\|_2&\leq \sum_{r_1 = 1 }^k\Big\|\sum_{r_2 = -1}^{q/k}\M M_{r_2k + r_1}^{(2),k_1,k} - \E\big[\M M_{r_2k + r_1}^{(2),k_1,k}\big]\Big\|_2 
\\
&\leq k\max_{1 \leq r_1 \leq k}\Big\|\sum_{r_2 = -1}^{q/k}\M M_{r_2k + r_1}^{(2),k_1,k} - \E\big[\M M_{r_2k + r_1}^{(2),k_1,k}\big]\Big\|_2
\\
&\leq k\max_{1 - k\leq r\leq q} \Big\|\M M_{r}^{(2),k_1,k} - \E\big[\M M_{r}^{(2),k_1,k}\big]\Big\|_2.
\ee
Note here by definition \eqref{def:Mr2k1k}, $\M M_{r_2k + r_1}^{(2),k_1,k}$ can still be a non-zero matrix when $r_2 = -1$. Since $\M M_{r_2k + r_1}^{(2),k_1,k_2/2}$ is a diagonal sub-block of $\M M_{r_2k + r_1}^{(2),k_1,k_2}$, one has $\Big\|\M M_{r}^{(2),k_1,k_2/2} - \E\big[\M M_{r}^{(2),k_1,k_2/2}\big] \Big\|_2\leq \Big\|\M M_{r}^{(2),k_1,k_2} - \E\big[\M M_{r}^{(2),k_1,k_2}\big] \Big\|_2$. Combining it with \eqref{l:dt:eq}, one has
\bee\label{lm:dt:trans1}
\|\DTCovZ - \M \Sigma_2^{*,\mathcal{T}}(k_2)\otimes \M \Sigma_1^{*,\mathcal{T}}(k_1)\|_2 \leq 3\max_{1\leq r\leq q - k_1 + 1} \Big\|\M M_{r}^{(2),k_1,k_2} - \E\big[\M M_{r}^{(2),k_1,k_2}\big]\Big\|_2.
\ee
Note $\E \DTCovZ = \M \Sigma_2^{*,\mathcal{T}}(k_2)\otimes \M \Sigma_1^{*,\mathcal{T}}(k_1)$. For $\M M_{r}^{(2),k_1,k_2}$, by definition, we have
\bee\label{l:dt:largeM}
\M M_{r}^{(2),k_1,k_2} = 
\begin{pmatrix}
&\ddots &&&&&
\\
&& \hat{\M \Sigma}^{r,r}_{0}(k_1) & \cdots & \hat{\M \Sigma}^{r,r+k_2 - 1}_{0}(k_1) 
\\
&& \vdots & \cdots & \vdots
\\
&& \hat{\M \Sigma}^{r+k_2 - 1,r}_{0}(k_1)  & \cdots & \hat{\M \Sigma}^{r+k_2 - 1,r+k_2 - 1}_{0}(k_1) 
\\
&& & & & \ddots
\end{pmatrix}.
\ee
Now we will present it via some diagonal sub-block matrices, similar to (\ref{l:dt:eq}). For each $\hat{\M \Sigma}^{l_2,m_2}_{0}(k_1)$,  we further define 
\bee\nonumber
\mathbf{m}_{r'}^{(1)}(l_2,m_2,k) &\equiv \big[\hat{\sigma}^
{(0)}_{((l_1,m_1),(l_2,m_2))}\cdot \M I(r' \leq l_1 < r' + k, r' \leq m_1 < r' + k)\big]_{1 \leq l_1 \leq p, 1\leq m_1 \leq p} \in \mathbb{R}^{p \times p}
\\
\bds{\mathcal{m}}_{r'}^{(1)}(l_2,m_2,k) &\equiv \big[\hat{\sigma}^
{(0)}_{((l_1,m_1),(l_2,m_2))}\big]_{r' \leq l_1 < r' + k, r' \leq m_1 < r' + k} \in \mathbb{R}^{d \times d},
\ee
where $\bds{\mathcal{m}}_{r'}^{(1)}(l_2,m_2,k)$ can be seen as a ``compressed'' version of $\mathbf{m}_{r'}^{(1)}(l_2,m_2,k)$ that only preserves the non-zero $d\times d$ sub-block. We note that since we can choose $r' < 0$, $d$ will {\em not always} be equal to $k$. Correspondingly, we also define
\bee\nonumber
&\M M^{(1),k_2,k}_{r,r'} \equiv 
\begin{pmatrix}
\mathbf{m}_{r'}^{(1)}(r,r,k) & \dots & \mathbf{m}_{r'}^{(1)}(r,r + k_2 -1,k)
\\
\vdots & \cdots & \vdots
\\
\mathbf{m}_{r'}^{(1)}(r+k_2 - 1,r,k) & \cdots &  \mathbf{m}_{r'}^{(1)}(r+k_2-1,r+k_2-1,k)
\end{pmatrix}\in \mathbb{R}^{pk_2 \times pk_2}
\\
&\mathscr{M}^{(1),k_2,k}_{r,r'} \equiv 
\begin{pmatrix}
\bds{\mathcal{m}}_{r'}^{(1)}(r,r,k) & \dots & \bds{\mathcal{m}}_{r'}^{(1)}(r,r + k_2 -1,k)
\\
\vdots & \cdots & \vdots
\\
\bds{\mathcal{m}}_{r'}^{(1)}(r+k_2 - 1,r,k) & \cdots &  \bds{\mathcal{m}}_{r'}^{(1)}(r+k_2-1,r+k_2-1,k)
\end{pmatrix}\in \mathbb{R}^{dk_2 \times dk_2}.
\ee
\begin{figure}
\label{fig:vis_mm:1}
   \centering
\includegraphics[width=0.8\textwidth]{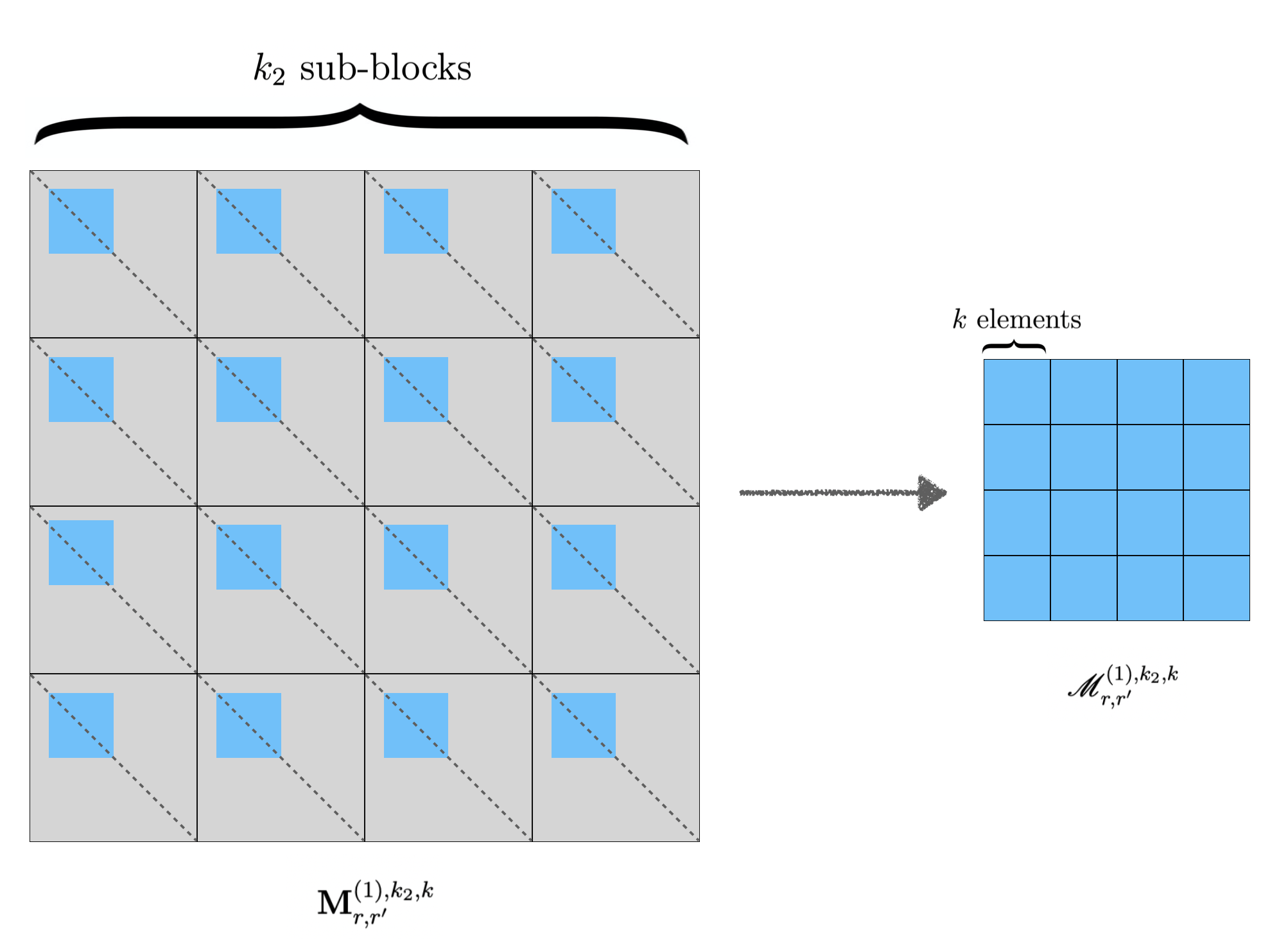}
\caption{Visualization of $\M M^{(1),k_2,k}_{r,r'}$ and $\mathscr{M}^{(1),k_2,k}_{r,r'}$ when $r' > 0$. $\mathscr{M}^{(1),k_2,k}_{r,r'}$ is assembled by the corresponding diagonal blocks in each sub-block of $\M M^{(1),k_2,k}_{r,r'}$}
\end{figure}
The relationship of $\M M^{(1),k_2,k}_{r,r'} $ and $\mathscr{M}^{(1),k_2,k}_{r,r'}$ can be visualized in Figure \ref{fig:vis_mm:1}. Since $\hat{\M \Sigma}^{l_2,m_2}_{0}(k_1)$ is {equivalent} to tapering with a bandwidth of $k_1$ on $\hat{\M \Sigma}_0^{l_2,m_2}$. The same argument as in the proof of Lemma 1 in \citet{Cai2010} also leads to
\bee\nonumber
\hat{\M \Sigma}^{l_2,m_2}_{0}(k_1) = (k_1/2)^{-1}\big({\mathcal{S}}_{l_2,m_2}^{(k_1)} - {\mathcal{S}}_{l_2,m_2}^{(k_1/2)}\big),
\ee
where $\mathcal{S}^{(k)}_{l_2,m_2} = \sum_{r' = 1 - k}^p\mathbf{m}_{r'}^{(1)}(l_2,m_2,k)$.
Thus, by \eqref{l:dt:largeM} we have 
\begin{align}\nonumber
&\M M_{r}^{(2),k_1,k_2} 
\\
& = (k_1/2)^{-1}\begin{pmatrix}
&\ddots &&&&&
\\
&& {\mathcal{S}}_{r,r}^{(k_1)} - {\mathcal{S}}_{r,r}^{(k_1/2)} & \cdots & {\mathcal{S}}_{r,r + k_2 - 1}^{(k_1)} - {\mathcal{S}}_{r,r + k_2 - 1}^{(k_1/2)}
\\
&& \vdots & \cdots & \vdots
\\
&& {\mathcal{S}}_{r+k_2 - 1,r}^{(k_1)} - {\mathcal{S}}_{r+k_2 - 1,r}^{(k_1/2)}  & \cdots &{\mathcal{S}}_{r+k_2 - 1,r+k_2 - 1}^{(k_1)} - {\mathcal{S}}_{r+k_2 - 1,r+k_2 - 1}^{(k_1/2)}
\\
&& & & & \ddots
\end{pmatrix}
\\
&=
(k_1/2)^{-1}\sum_{r' = 1 - k_1}^p
\begin{pmatrix}
\ddots 
\\\nonumber
& 
\begin{pmatrix}
\mathbf{m}_{r'}^{(1)}(r,r,k_1) & \cdots & \mathbf{m}_{r'}^{(1)}(r,r + k_2 - 1,k_1) 
\\
\vdots & \vdots & \vdots
\\
 \mathbf{m}_{r'}^{(1)}(r + k_2 - 1,r,k_1) & \cdots &  \mathbf{m}_{r'}^{(1)}(r + k_2 - 1,r + k_2 - 1,k_1)
\end{pmatrix}
\\
&& \ddots
\end{pmatrix}
\\
&-(k_1/2)^{-1}\sum_{r' = 1 - k_1/2}^p
\begin{pmatrix}
\ddots 
\\\nonumber
& 
\begin{pmatrix}
\mathbf{m}_{r'}^{(1)}(r,r,k_1/2) & \cdots & \mathbf{m}_{r'}^{(1)}(r,r + k_2 - 1,k_1/2) 
\\
\vdots & \vdots & \vdots
\\
 \mathbf{m}_{r'}^{(1)}(r + k_2 - 1,r,k_1/2) & \cdots &  \mathbf{m}_{r'}^{(1)}(r + k_2 - 1,r + k_2 - 1,k_1/2)
\end{pmatrix}
\\
&& \ddots
\end{pmatrix}
\\\nonumber
&= \begin{pmatrix}
\ddots 
\\
& (k_1/2)^{-1}\big({\M S'}^{(k_1)}_{r,k_2} - {\M S'}^{(k_1/2)}_{r,k_2}\big)
\\
&& \ddots
\end{pmatrix},
\end{align}
where ${\M S'}^{(k)}_{r,k_2} \equiv \sum_{r' = 1 - k}^p \M M^{(1),k_2,k}_{r,r'}$. By (\refeq{lm:dt:trans1}), we then have
\bee\label{lm:dt:target2}
&\|\DTCovZ - \M \Sigma_2^{*,\mathcal{T}}(k_2)\otimes \M \Sigma_1^{*,\mathcal{T}}(k_1)\|_2 
\\
&\leq 3\max_{r\leq q - k_2 + 1} \Big\|\M M_{r}^{(2),k_1,k_2} - \E\big[\M M_{r}^{(2),k_1,k_2}\big]\Big\|_2
\\
&= 3\max_{r\leq q - k_2 + 1}\Big\|(k_1/2)^{-1}\big({\M S'}^{(k_1)}_{r,k_2} - {\M S'}^{(k_1/2)}_{r,k_2}\big) - \E \big[(k_1/2)^{-1}\big({\M S'}^{(k_1)}_{r,k_2} - {\M S'}^{(k_1/2)}_{r,k_2}\big)\big]\Big\|_2
\\
&\leq 6\times \frac{1}{k_1}\max_{r\leq q - k_2 + 1}\Big\|{\M S'}^{(k_1)}_{r,k_2}  - \E \big[{\M S'}^{(k_1)}_{r,k_2} \big]\Big\|_2 + 3\times \frac{1}{k_1/2}\max_{r\leq q - k_2 + 1}\Big\|{\M S'}^{(k_1/2)}_{r,k_2}  - \E \big[{\M S'}^{(k_1/2)}_{r,k_2} \big]\Big\|_2,
\ee
by triangle inequality. Thus, to {reach the final conclusion}, it is left to study the upper bound of $\|{\M S'}^{(k)}_{r,k_2} - \E {\M S'}^{(k)}_{r,k_2}\|_2$. We first prove the following claim.
{\begin{claim}\label{claim:1}
For a $pk_2 \times pk_2$ symmetric matrix $\M M$ such that
\bee\nonumber
\M M = \begin{pmatrix}
\begin{pmatrix}
\bds{m}_{1}^{1,1}
\\
& \ddots 
\\
&& \bds{m}_{\ell}^{1,1}
\end{pmatrix}_{p\times p} & \dots & 
\begin{pmatrix}
\bds{m}_{1}^{1,k_2}
\\
& \ddots 
\\
&& \bds{m}_{\ell}^{1,k_2}
\end{pmatrix}_{p\times p}
\\
\vdots & \vdots & \vdots
\\
\begin{pmatrix}
\bds{m}_{1}^{k_2,1}
\\
& \ddots 
\\
&& \bds{m}_{\ell}^{k_2,1}
\end{pmatrix}_{p\times p} & \dots &
\begin{pmatrix}
\bds{m}_{1}^{k_2,k_2}
\\
& \ddots 
\\
&& \bds{m}_{\ell}^{k_2,k_2}
\end{pmatrix}_{p\times p}
\end{pmatrix}
\ee
where $\{\bds{m}_{\ell_0}^{l_2,m_2}\}_{1\leq l_2,m_2\leq k_2}$ is a set of square matrices of dimension $d_{\ell_0}\times d_{\ell_0}$ for any $1 \leq \ell_0 \leq \ell$. Elements in $\M M$ that are not entries in $\bds{m}^{l_2,m_2}_{\ell_0}$ are all $0$. Then we have
\bee\nonumber
\|\M M\|_2 = \max_{1\leq \ell_0 \leq \ell}\big\{\big\|\mathscr{M}_{\ell_0}\big\|_2\big\}, 
\ee
where $\mathscr{M}_{\ell_0} = \begin{pmatrix} 
\bds{m}_{\ell_0}^{1,1} & \dots & \bds{m}_{\ell_0}^{1,k_2}
\\
\vdots & \vdots & \vdots
\\
\bds{m}_{\ell_0}^{k_2,1} & \dots & \bds{m}_{\ell_0}^{k_2,k_2}
\end{pmatrix}_{k_2d_{\ell_0} \times k_2d_{\ell_0}}$.
\end{claim}}
\begin{proof}[Proof of Claim \ref{claim:1}]
Since $\M M$ is symmetric, for $\bds v \in \mathbb{R}^{pk_2}$ we have 
\bee\nonumber
\|\M M\|_2 = \max_{\|\bds v\|_2 = 1}\bds v^\T \M M\bds v.
\ee
If we write $\bds v^\T = (\underbrace{\bds v_{(1),1}^\T, \dots,\bds v_{(1), \ell}^\T}_{\in \mathbb{R}^p}, \underbrace{\bds v_{(2),1}^\T, \dots,\bds v_{(2), \ell}^\T}_{\in \mathbb{R}^p},\dots\dots,\underbrace{\bds v_{(k),1}^\T, \dots,\bds v_{(k), \ell}^\T}_{\in \mathbb{R}^p})$, where $\bds v_{(l_1), \ell_0} \in \mathbb{R}^{d_{\ell_0}}$ for any $1\leq l_1 \leq d_{\ell_0}$ and $1 \leq \ell_0 \leq \ell$. We also  define 
\bee\label{def:vl0}
\bds{v}_{\ell_0}^\T \equiv (\bds{v}^\T_{(1),\ell_0}, \bds{v}^\T_{(2),\ell_0}, \dots, \bds{v}^\T_{(k_2),\ell_0}).
\ee
With simple algebra, one can see
\bee\nonumber
\bds{v}^\T \M M \bds{v} = \sum_{1\leq \ell_0 \leq \ell}\bds{v}_{\ell_0}^\T\mathscr{M}_{\ell_0}\bds{v}_{\ell_0}.
\ee
Thus,
\bee\label{cliam:1}
\|\M M\|_2 &= \max_{\sum_{1\leq \ell_0 \leq \ell}\|\bds v_{\ell_0}\|^2_2 = 1}\sum_{1\leq \ell_0 \leq \ell}\bds{v}_{\ell_0}^\T\mathscr{M}_{\ell_0}\bds{v}_{\ell_0} 
\\
&=  \max_{\sum_{1\leq \ell_0 \leq \ell}\|\bds v_{\ell_0}\|^2_2 = 1}\sum_{1\leq \ell_0 \leq \ell}\|\bds{v}_{\ell_0}\|_2^2 \Big(\frac{\bds{v}_{\ell_0}}{\|\bds{v}_{\ell_0}\|_2}\Big)^\T\mathscr{M}_{\ell_0}\Big(\frac{\bds{v}_{\ell_0}}{\|\bds v_{\ell_0}\|}\Big)
\\
&\leq \max_{\sum_{1\leq \ell_0 \leq \ell}\|\bds v_{\ell_0}\|^2_2 = 1}\sum_{1\leq \ell_0\leq \ell}\|\bds{v}_{\ell_0}\|_2^2\cdot \|\mathscr{M}_{\ell_0}\|_2 
\\
&\leq \max_{1\leq \ell_0 \leq \ell}\big\{\big\|\mathscr{M}_{\ell_0}\big\|_2\big\},
\ee
where the last inequality holds because $\sum_{1\leq \ell_0 \leq \ell}\|\bds v_{\ell_0}\|^2_2 = 1$ and $\sum_{1\leq \ell_0\leq \ell}\|\bds{v}_{\ell_0}\|_2^2\cdot \|\mathscr{M}_{\ell_0}\|_2$ can be seen as a weighted average of $\{\|\mathscr{M}_{\ell_0}\|_2\}_{1\leq \ell_0 \leq \ell}$. 
\par
On the other hand, there exists an $\ell_0^*$ such that $\|\mathscr{M}_{\ell_0^*}\|_2 = \max_{1\leq \ell_0\leq\ell}\{\|\mathscr{M}_{\ell_0}\|_2\}$. By properties of spectral norm, we can find a ${\bds v}^*_{\ell_0^*}\in \RR^{k_2d_{\ell_0^*}}$ such that $\|\bds v^*_{\ell_0^*}\| = 1$ and ${\bds{v}^*_{\ell_0^*}}^\T\mathscr{M}_{\ell_0^*}\bds{v}^*_{\ell_0^*} = \|\mathscr{M}_{\ell_0^*}\|_2 = \max_{1\leq \ell_0 \leq \ell}\big\{\big\|\mathscr{M}_{\ell_0}\big\|_2\big\}$. Then we define $\bds v^*$ as
\bee
\bds v^* = ({{{\bds v}^*_{(1),1}}^\T,\dots,{\bds v^*_{(1),\ell_0^*}}^\T_{} \dots,{{\bds v}_{(1), \ell}^*}^\T},\dots\dots,{{{\bds v}^*_{(k),1}}^\T,\dots,{\bds v^*_{(k),\ell_0^*}}^\T_{} \dots,{{\bds v}_{(k), \ell}^*}^\T}),
\ee
where $\bds v^*_{(k'),\ell_0^*}$ is the $k'$th sub-vector of $\bds v^*_{\ell_0^*}$ as defined in \eqref{def:vl0} for all $1\leq k'\leq k$. All other $\bds v^*_{(k'),\ell_0}$ are  defined to be zero vectors when $\ell_0\neq \ell_0^*$ and $1\leq k'\leq k$. Then we show that
\bee\label{cliam:2}
\max_{1\leq \ell_0 \leq \ell}\big\{\big\|\mathscr{M}_{\ell_0}\big\|_2\big\} &= \big\|\mathscr{M}_{\ell^*_0}\big\|_2
\\
&= {\bds{v}^*_{\ell_0^*}}^\T\mathscr{M}_{\ell_0^*}\bds{v}^*_{\ell_0^*} + \sum_{\ell_0\neq \ell_0^*}\underbrace{{\bds{v}^*_{\ell_0}}^\T}_{=\bds 0}\mathscr{M}_{\ell_0}\bds{v}^*_{\ell_0}
\\
&=\sum_{1\leq \ell_0 \leq \ell}{\bds{v}^*_{\ell_0}}^\T\mathscr{M}_{\ell_0}\bds{v}^*_{\ell_0}
\\
&= {{\bds v^*}^\T} \M M\bds v^*
\\
&\leq \|\M M\|_2,
\ee
where the last inequality holds by directly checking that $\|\bds v^*\| = 1$. Combining \eqref{cliam:1} and \eqref{cliam:2}, we finally have
\bee\nonumber
\|\M M\|_2 = \max_{1\leq \ell_0 \leq \ell}\big\{\big\|\mathscr{M}_{\ell_0}\big\|_2\big\}.
\ee
\end{proof}
Then similar to \eqref{Skdecom}, by definition we have
\bee\nonumber
&\Big\|{\M S'}^{(k)}_{r,k_2} - \E {\M S'}^{(k)}_{r,k_2}\Big\|_2 
\\
&\leq \sum_{r_3 = 1 }^k\Big\|\sum_{r_4 = -1}^{p/k}\M M_{r, r_4k + r_3}^{(1),k_2,k} - \E\big[\M M_{r, r_4k + r_3}^{(1),k_2,k}\big]\Big\|_2 
\\
&\leq k\max_{1 \leq r_3 \leq k}\Big\|\sum_{r_4 = -1}^{p/k}\M M_{r, r_4k + r_3}^{(1),k_2,k} - \E\big[\M M_{r, r_4k + r_3}^{(1),k_2,k}\big]\Big\|_2.
\ee 
We now take $\sum_{r_4 = -1}^{p/k}\M M_{r, r_4k + r_3}^{(1),k_2,k} - \E\big[\M M_{r, r_4k + r_3}^{(1),k_2,k}\big]$ as the $\M M$ in Claim \ref{claim:1}, and we have
\bee\nonumber
\Big\|{\M S'}^{(k)}_{r,k_2} - \E {\M S'}^{(k)}_{r,k_2}\Big\|_2 \leq k\max_{1-k\leq r'\leq p}\big\|\mathscr{M}_{r,r'}^{(1),k_2,k} - \E\big[\mathscr{M}_{r,r'}^{(1),k_2,k}\big]\big\|_2.
\ee
Combining it with (\ref{lm:dt:target2}), we obtain
\bee\label{lm:dt:target3}
&\|\DTCovZ - \M \Sigma_2^{*,\mathcal{T}}(k_2)\otimes \M \Sigma_1^{*,\mathcal{T}}(k_1)\|^2_2 
\\
&\leq 81\max_{1 - k\leq r'\leq p,\atop {1\leq r \leq q - k_2 + 1,\atop k\in\{k_1,k_1/2\}}}\big\|\mathscr{M}_{r,r'}^{(1),k_2,k} - \E\big[\mathscr{M}_{r,r'}^{(1),k_2,k}\big]\big\|^2_2
\\
&\leq 81{\max_{1-k_1\leq r'\leq p \atop1\leq r \leq q - k_2 + 1}\big\|\mathscr{M}_{r,r'}^{(1),k_2,k_1} - \E\big[\mathscr{M}_{r,r'}^{(1),k_2,k_1}\big]\big\|^2_2}
\\
&\leq 81\underbrace{\max_{1\leq r'\leq p - k_1 + 1\atop1\leq r \leq q - k_2 + 1}\big\|\mathscr{M}_{r,r'}^{(1),k_2,k_1} - \E\big[\mathscr{M}_{r,r'}^{(1),k_2,k_1}\big]\big\|^2_2}_{N_{k_1,k_2}}.
\ee
To understand the second inequality in \eqref{lm:dt:target3}, we first note $\mathscr{M}_{r,r'}^{(1),k_2,k_1/2} - \E\big[\mathscr{M}_{r,r'}^{(1),k_2,k_1/2}\big]$ is blocked by submatices $\bds{\mathcal{m}}_{r'}^{(1)}(l_2,m_2,k_1/2) - \E \bds{\mathcal{m}}_{r'}^{(1)}(l_2,m_2,k_1/2)$, and $\mathscr{M}_{r,r'}^{(1),k_2,k_1} - \E\big[\mathscr{M}_{r,r'}^{(1),k_2,k_1}\big]$ is blocked by submatices $\bds{\mathcal{m}}_{r'}^{(1)}(l_2,m_2,k_1) - \E \bds{\mathcal{m}}_{r'}^{(1)}(l_2,m_2,k_1)$. Furthermore, by definition, each of the $\bds{\mathcal{m}}_{r'}^{(1)}(l_2,m_2,k_1/2) - \E \bds{\mathcal{m}}_{r'}^{(1)}(l_2,m_2,k_1/2)$ is a sub-matrix of the corresponding $\bds{\mathcal{m}}_{r'}^{(1)}(l_2,m_2,k_1) - \E \bds{\mathcal{m}}_{r'}^{(1)}(l_2,m_2,k_1)$. Here we note the fact that the spectral norm of a submatrix is always smaller or equal to the spectral norm of the original matrix; see e.g. (2.3.13) in \citet{golub2013matrix} for details. Since  $\mathscr{M}_{r,r'}^{(1),k_2,k_1/2} - \E\big[\mathscr{M}_{r,r'}^{(1),k_2,k_1/2}\big]$ is a submatrix of $\mathscr{M}_{r,r'}^{(1),k_2,k_1} - \E\big[\mathscr{M}_{r,r'}^{(1),k_2,k_1}\big]$ block-wisely, one can easily show $$\big\|\mathscr{M}_{r,r'}^{(1),k_2,k_1/2} - \E\big[\mathscr{M}_{r,r'}^{(1),k_2,k_1/2}\big]\big\|_2\leq \big\|\mathscr{M}_{r,r'}^{(1),k_2,k_1} - \E\big[\mathscr{M}_{r,r'}^{(1),k_2,k_1}\big]\big\|_2$$ for any $1-k_1/2\leq r'\leq p$ and $1\leq r \leq q-k_2 + 1$. Thus the second inequality of \eqref{lm:dt:target3} follows.
\par
The last inequality of \eqref{lm:dt:target3} holds for the same reason as above. In particular, because any sub-block $\mathbf{m}_{r'}^{(1)}(l_2,m_2,k_1/2) - \E \mathbf{m}_{r'}^{(1)}(l_2,m_2,k_1/2)$ in $\mathscr{M}_{r,r'}^{(1),k_2,k_1/2} - \E\big[\mathscr{M}_{r,r'}^{(1),k_2,k_1/2}\big]$ with $1-k_1\leq r'\leq 0$, is a submatrix of the corresponding $\mathbf{m}_{1}^{(1)}(l_2,m_2,k_1) - \E \mathbf{m}_{1}^{(1)}(l_2,m_2,k_1)$ in $\mathscr{M}_{r,1}^{(1),k_2,k_1} - \E\big[\mathscr{M}_{r,1}^{(1),k_2,k_1}\big]$, one has $\mathscr{M}_{r,r'}^{(1),k_2,k_1/2} - \E\big[\mathscr{M}_{r,r'}^{(1),k_2,k_1/2}\big]$ is generally a submatrix of $\mathscr{M}_{r,1}^{(1),k_2,k_1} - \E\big[\mathscr{M}_{r,1}^{(1),k_2,k_1}\big]$ block-wisely. We  then have $\big\|\mathscr{M}_{r,r'}^{(1),k_2,k_1} - \E\big[\mathscr{M}_{r,r'}^{(1),k_2,k_1}\big]\big\|_2\leq \big\|\mathscr{M}_{r,1}^{(1),k_2,k_1} - \E\big[\mathscr{M}_{r,1}^{(1),k_2,k_1}\big]\big\|_2$ when $1-k_1\leq r'\leq 0$ and $1\leq r \leq q-k_2 + 1$. Similarly, we also have $\big\|\mathscr{M}_{r,r'}^{(1),k_2,k_1} - \E\big[\mathscr{M}_{r,r'}^{(1),k_2,k_1}\big]\big\|_2\leq \big\|\mathscr{M}_{r,p-k_1 + 1}^{(1),k_2,k_1} - \E\big[\mathscr{M}_{r,p-k_1 + 1}^{(1),k_2,k_1}\big]\big\|_2$ when $p-k_1 + 2\leq r'\leq p$ and $1\leq r \leq q-k_2 + 1$.  Thus the last inequality of \eqref{lm:dt:target3} follows.
\par
With (\ref{lm:dt:target3}), we focus on the concentration of $N_{k_1,k_2}$. Adapted from proof of Lemma 3 in \cite{Cai2010}, we have
\bee\label{l:dt:key1}
\p(N_{k_1,k_2} > x)&=\p\Big\{\max_{1\leq r'\leq p - k_1 + 1\atop1\leq r \leq q - k_2 + 1}\big\|\mathscr{M}_{r,r'}^{(1),k_2,k_1} - \E\big[\mathscr{M}_{r,r'}^{(1),k_2,k_1}\big]\big\|^2_2 > x\Big\}
\\
&\leq pq \max_{1\leq r'\leq p - k_1 + 1\atop1\leq r \leq q - k_2 + 1}\p\Big\{\big\|\mathscr{M}_{r,r'}^{(1),k_2,k_1} - \E\big[\mathscr{M}_{r,r'}^{(1),k_2,k_1}\big]\big\|^2_2 > x\Big\}
\\
&\precsim pq \cdot 5^{k_1 k_2}\sup_{\mathbf{v}_i \in \mathcal{U}_{k_2 d}, r, r'} \p\Big\{\big| \mathbf{v}_i^\T\big(\mathscr{M}_{r,r'}^{(1),k_2,k_1} - \E\big[\mathscr{M}_{r,r'}^{(1),k_2,k_1}\big]\big)\mathbf{v}_i\big|^2 > x\Big\}.
\ee
Here $\mathcal{U}_{k_2 d}$ is the set of unit spheres in $\RR^{k_2d}$ where $\mathscr{M}^{(1),k_2,k_1}_{r,r'} \in \RR^{k_2d \times k_2d}$, and $d\leq k_1$ is determined by $k_1$ and $r'$. By the structure of $\mathscr{M}_{r,r'}^{(1),k_2,k_1}$, we can apply Lemma \ref{subG:dev} to show
\bee\label{l:dt:key2}
\p(N_{k_1,k_2} > x) &= \p\Big\{\max_{1\leq r'\leq p - k_1 + 1\atop1\leq r \leq q - k_2 + 1}\big\|\mathscr{M}_{r,r'}^{(1),k_2,k_1} - \E\big[\mathscr{M}_{r,r'}^{(1),k_2,k_1}\big]\big\|^2_2 > x\Big\}
\\
&\precsim pq 5^{k_1k_2}\exp(-nC'_2x)
\ee
for some $C'_2,\zeta' > 0$ when $|x| \leq \zeta'$. Then by Cauchy-Schwarz inequality,
\bee\label{l:dt:key3}
&\E\big\|\DTCovZ - \M \Sigma_2^{*,\mathcal{T}}(k_2)\otimes \M \Sigma_1^{*,\mathcal{T}}(k_1)\|^2_2 
\\
&\leq 81\E N_{k_1,k_2} 
\\
&\precsim x + \E \big(N_{k_1,k_2}\M I(N_{k_1,k_2} > x)\big)
\\
&\leq x+ \sqrt{\E(N_{k_1,k_2}^2)}\sqrt{\p(N_{k_1,k_2} > x)} 
\\
&\precsim x + \sqrt{\E(\|\DTCovZ\|_\F^4 + \|\M \Sigma_2^{*,\mathcal{T}}(k_2)\otimes \M \Sigma_1^{*,\mathcal{T}}(k_1)\|_\F^4)}\cdot\sqrt{\p(N_{k_1,k_2} > x)}
\\
&\precsim x + p^2q^2\sqrt{pq5^{k_1k_2}}\exp(-nxC'_2/2).
\ee
Take $x =C\frac{\log(pq) + k_1k_2}{nC_2'/2}$ with sufficient large $C>0$,
\bee\label{l:dt:key4}
\E\big\|\DTCovZ - \M \Sigma_2^{*,\mathcal{T}}(k_2)\otimes \M \Sigma_1^{*,\mathcal{T}}(k_1)\|^2_2 \precsim \frac{\log(\max\{p,q\}) + k_1k_2}{n},
\ee
when $p$ or $q$ is $\succsim n^{b}$ for some $b >0$.
\\
\par
\noindent{\textbf{Step 2 (Bound of $\E\big\|\DTCovZ - \DTCov\|^2_2$): }} It is easy to see that $\DTCov - \DTCovZ = (\hat{\M \Sigma}_0  - \hat{\M \Sigma}) \HDTaper = \vecc(\bar{\M X})\vecc(\bar{\M X})^\T\HDTaper.$ We define $\tilde{\mathscr{M}}^{(1),k_2,k}_{r,r'}$ based on $\vecc(\bar{\M X})\vecc(\bar{\M X})^\T\HDTaper$ in the same way as $\mathscr{M}^{(1),k_2,k}_{r,r'}$ based on $\DTCovZ$. Similar to \eqref{lm:dt:target3}, we can show
\bee\nonumber
\E\|\vecc(\bar{\M X})\vecc(\bar{\M X})^\T\HDTaper\|^2_2\leq 81\E \max_{r,r'}\|\tilde{\mathscr{M}}^{(1),k_2,k_1}_{r,r'}\|_2^2.
\ee 
Now, to bound $\E \big\|\DTCovZ - \DTCov\|^2_2 = \E \|\vecc(\bar{\M X})\vecc(\bar{\M X})^\T\HDTaper\|_2^2$, it is left to get the concentration bound of $\max_{r,r'}\|\tilde{\mathscr{M}}^{(1),k_2,k_1}_{r,r'}\|_2^2$ .
\par
By the property of sub-Gaussian random variables, there exists $\rho'>0$ such that 
\bee\nonumber
\p\Big\{\Big|\mathbf{v}^\T\Big[\vecc(\bar{\M X}) - \E\{\vecc(\bar{\M X})\}\Big]\Big|>x\Big\} = \p\Big\{\Big|\mathbf{v}^\T\cdot \vecc(\bar{\M X})\Big|>x\Big\}\leq e^{-\rho' nx^2}
\ee
{\color{black} for any} $\|\mathbf{v}\| = 1$. Here $\rho'> 0$ is a constant {\color{black} that depends on the constant $\rho$}. Treating $\vecc(\bar{\M X})$ as a sub-Gaussian random variable with parameter $n\rho'$, by sub-Gaussian property
\bee\nonumber
\p\Big(|\mathbf{v}^\T\cdot \vecc(\bar{\M X}) \vecc(\bar{\M X})^\T \cdot \mathbf{v}|>x\Big) &= \p[\{\mathbf{v}^\T\cdot \vecc(\bar{\M X})\}^2>x]
\\
&= \p\{|\mathbf{v}^\T\cdot \vecc(\bar{\M X})|>\sqrt{x}\}
\\
&\leq \exp(-\rho'nx)
\ee
%\bee
%\p\Big(|\mathbf{v}^\T\cdot \vecc(\bar{\M X}) \vecc(\bar{\M X})^\T \cdot \mathbf{v} - \mathbf{v}^\T\E[\vecc(\bar{\M X}) \vecc(\bar{\M X})^\T]\mathbf{v}|>x\Big)\leq \exp(-C_2''n^2x^2)
%\ee
for some $C_2''>0$. Thus, by the construction of $\tilde{\mathscr{M}}^{(1),k_2,k}_{r,r'}$, similar to \eqref{l:dt:key1}-\eqref{l:dt:key2}, we have
\bee\nonumber
&\p\Big\{\max_{1\leq r'\leq p - k_1 + 1\atop1\leq r \leq q - k_2 + 1}\big\|\tilde{\mathscr{M}}_{r,r'}^{(1),k_2,k_1}\big\|^2_2 > x\Big\} 
\\
&\precsim pq 5^{k_1k_2}\sup_{\|\bds{v}\| = 1,\bds{v}\in \RR^{d\times d},r,r'}\p\Big\{\Big| \bds{v}^\T\tilde{\mathscr{M}}_{r,r'}^{(1),k_2,k_1} \bds{v}\Big|^2 > x\Big\}
\\
&\precsim pq 5^{k_1k_2}\sup_{\|\mathbf{v}\| = 1,\mathbf{v}\in \RR^{pq\times pq}}\p\Big[\Big|\mathbf{v}^\T\cdot \Big\{\vecc(\bar{\M X}) \vecc(\bar{\M X})^\T \Big\} \cdot \mathbf{v}\Big|^2>x\Big]
\\
&\precsim pq5^{k_1k_2}\exp(-\rho'n\sqrt{x}).
\ee
Following similar arguments to \eqref{l:dt:key3}-\eqref{l:dt:key4}, we can  show
\bee\nonumber
\E\big\|\DTCovZ - \DTCov\|^2_2 \precsim \Big[\frac{\log(\max\{p,q\}) + k_1k_2}{n}\Big]^2,
\ee
which is negligible compared to $\E\big\|\DTCovZ - \M \Sigma_2^{*,\mathcal{T}}(k_2)\otimes \M \Sigma_1^{*,\mathcal{T}}(k_1)\|^2_2$ when $p$ or $q$ is $\succsim n^{b}$ for some $b >0$.
\end{proof}
\begin{lemma}\label{l:sp:out}
For $\M \Sigma_1^{*}\in \mathcal{F}(\varepsilon_0, \alpha_1),\M \Sigma_2^{*} \in \mathcal{F}(\varepsilon_0, \alpha_2)$ and $\M\Sigma^* = \M \Sigma_2^*\otimes \M \Sigma_1^*$, we have
\bee\nonumber
\label{l:sp:outer:res1}
\|\M \Sigma_2^{*,\mathcal{T}}(k_2)\otimes\M \Sigma_1^{*,\mathcal{T}}(k_1) - \M \Sigma^*\|_2^2 \precsim \M I(k_1< 2p-2)\cdot k_1^{-2\alpha_1} +  \M I(k_2< 2q - 2)\cdot k_2^{-2\alpha_2}.
\ee
\end{lemma}
\begin{proof}[Proof of Lemma \ref{l:sp:out}] The notations of this proof are mainly contained in Section \ref{notation:main}. We can directly decompose
\bee\nonumber
&\M \Sigma_2^{*,\mathcal{T}}(k_2)\otimes\M \Sigma_1^{*,\mathcal{T}}(k_1) - \M \Sigma^* 
\\
&= \M \Sigma_2^{*,\mathcal{T}}(k_2)\otimes\M \Sigma_1^{*,\mathcal{T}}(k_1) - \M \Sigma_2^*\otimes \M \Sigma_1^*
\\
&= \big\{\M \Sigma_2^{*,\mathcal{T}}(k_2) - \M \Sigma_2^*\big\}\otimes\M \Sigma_1^{*,\mathcal{T}}(k_1) + \M \Sigma_2^*\otimes\big\{\M \Sigma_1^{*,\mathcal{T}}(k_1)-\M \Sigma_1^*\big\}.
\ee
By triangle inequality, we have
\bee\label{main:lm:s:outer}
&\big\|\M \Sigma_2^{*,\mathcal{T}}(k_2)\otimes\M \Sigma_1^{*,\mathcal{T}}(k_1) - \M \Sigma^*\big\|^2_2 
\\
&\leq 2\big\|\big\{\M \Sigma_2^{*,\mathcal{T}}(k_2) - \M \Sigma_2^*\big\}\otimes\M \Sigma_1^{*,\mathcal{T}}(k_1)\big\|^2_2 + 2\big\|\M \Sigma_2^*\otimes\big\{\M \Sigma_1^{*,\mathcal{T}}(k_1)-\M \Sigma_1^*\big\}\big\|^2_2
\\
&=2\big\|\M \Sigma_2^{*,\mathcal{T}}(k_2) - \M \Sigma_2^*\big\|^2_2\times\big\|\M \Sigma_1^{*,\mathcal{T}}(k_1)\big\|^2_2 + 2\big\|\M \Sigma_2^*\big\|^2_2\times\big\|\M \Sigma_1^{*,\mathcal{T}}(k_1)-\M \Sigma_1^*\big\|^2_2,
\ee
where the last equality holds by \citet{lancaster1972norms}.
\par
We first bound $\big\|\M \Sigma_1^{*,\mathcal{T}}(k_1)-\M \Sigma_1^*\big\|_2$ and $\big\|\M \Sigma_2^{*,\mathcal{T}}(k_2)-\M \Sigma_2^*\big\|_2$. By definition of $T_k(\cdot)$, we know the absolute value of $l_1m_1$th entry of $\M \Sigma_1^{*,\mathcal{T}}(k_1)-\M \Sigma_1^*$ is less or equal to $|\sigma_{l_1m_1}^{(1)}|$ when $|l_1 - m_1| > \lfloor \frac{k_1}{2}\rfloor$; equal to $0$ when $|l_1 - m_1|\leq \lfloor \frac{k_1}{2}\rfloor$. Note that $1\leq l_1,m_1 \leq p$. Then we have
\bee\nonumber
& \big\|\M \Sigma_1^{*,\mathcal{T}}(k_1)-\M \Sigma_1^*\big\|_{\infty}  \\
&~= \max_{1\leq l_1\leq p}\sum_{m_1 = 1}^{p} \big|[\M \Sigma_1^{*,\mathcal{T}}(k_1)-\M \Sigma_1^*]_{l_1m_1}\big|
\\
&~= \max_{1\leq l_1 \leq p}\Big[\sum_{m_1: |m_1 - l_1|>\lfloor \frac{k_1}{2}\rfloor} \underbrace{\big|[\M \Sigma_1^{*,\mathcal{T}}(k_1)-\M \Sigma_1^*]_{l_1m_1}\big|}_{\leq |\sigma^{(1)}_{l_1m_1}|}
+ \sum_{m_1: |m_1 - l_1|\leq\lfloor \frac{k_1}{2}\rfloor} \underbrace{\big|[\M \Sigma_1^{*,\mathcal{T}}(k_1)-\M \Sigma_1^*]_{l_1m_1}\big|}_{=0}\Big]
\\
&~\leq \max_{1\leq l_1 \leq p}\sum_{m_1: |m_1 - l_1|>\lfloor \frac{k_1}{2}\rfloor} |\sigma_{l_1m_1}^{(1)}|
\\
&~\leq 
\begin{cases}
C_0\big(\big\lfloor k_1/2\big\rfloor\big)^{-\alpha_1} & k_1< 2p-2;
\\
0 & k_1\geq 2p-2,
\end{cases}
\\
&~=\M I(k_1<2p-2)\times C_0\big(\big\lfloor k_1/2\big\rfloor\big)^{-\alpha_1} 
\\
&~\asymp \M I(k_1<2p-2)k_1^{-\alpha_1}.
\ee
When $k_1 < 2p-2$, the second inequality can be directly derived from the fact that $\M\Sigma_1^* \in \mathcal{F}(\varepsilon_0,\alpha) $ and \eqref{A1}. When $k_1 \geq 2p-2$, the second inequality holds because there are no $l_1,m_1$ such that $1\leq l_1,m_1\leq p$ satisfying $|m_1 - l_1|>\lfloor k_1/2 \rfloor\geq (2p-2)/2>p-1$. So $\max_{1\leq l_1 \leq p}\sum_{m_1: |m_1 - l_1|>\lfloor \frac{k_1}{2}\rfloor} |\sigma_{l_1m_1}^{(1)}| = 0$ as there is no element in the sum. Since $\M \Sigma_1^{*,\mathcal{T}}(k_1)-\M \Sigma_1^*$ is symmetric, we also have $\|\M \Sigma_1^{*,\mathcal{T}}(k_1)-\M \Sigma_1^*\|_1 = \|\M \Sigma_1^{*,\mathcal{T}}(k_1)-\M \Sigma_1^*\|_{\infty}$. Then by Lemma \ref{lm:matrixholder}, we finally show
\bee\nonumber
\|\M \Sigma_1^{*,\mathcal{T}}(k_1)-\M \Sigma_1^*\|_2 &\leq \sqrt{\|\M \Sigma_1^{*,\mathcal{T}}(k_1)-\M \Sigma_1^*\|_1\times\|\M \Sigma_1^{*,\mathcal{T}}(k_1)-\M \Sigma_1^* \|_{\infty}} 
\\
&= \big\|\M \Sigma_1^{*,\mathcal{T}}(k_1)-\M \Sigma_1^*\big\|_{\infty}
\\
&\precsim \M I(k_1<2p-2)k_1^{-\alpha_1}.
\ee
A symmetric argument can also show $\|\M \Sigma_2^{*,\mathcal{T}}(k_2)-\M \Sigma_2^*\|_2 \precsim \M I(k_2<2q-2)k_2^{-\alpha_2}$.
\par
Next we bound $\big\|\M \Sigma_1^{*,\mathcal{T}}(k_1)\big\|_2$ and $\big\|\M \Sigma_2^*\big\|_2$. By definitions of $\M \Sigma_1^{*,\mathcal{T}}(k_1)$, we know entrywisely $\M \Sigma_1^{*,\mathcal{T}}(k_1)$ is a shrinkage of $\M\Sigma_1^*$. So we know that
\bee\nonumber
\|\M \Sigma_1^{*,\mathcal{T}}(k_1)\|_\infty &=\max_{1\leq l_1\leq p}\sum_{m_1 = 1}^p\big|\big[\M \Sigma_1^{*,\mathcal{T}}(k_1)\big]_{l_1m_1}\big|
\\
&\leq\max_{1\leq l_1\leq p}\sum_{m_1 = 1}^p\big|\big[\M \Sigma_1^{*}\big]_{l_1m_1}\big|
\\
&=\max_{1\leq l_1\leq p}\sum_{m_1 = 1}^p\big|\sigma^{(1)}_{l_1m_1}\big|
\\
&\leq C_0,
\ee
where the last inequality holds because $\M \Sigma_1^*\in\mathcal{F}(\varepsilon_0,\alpha_1)$. Then since $\M \Sigma_1^{*,\mathcal{T}}(k_1)$ is symmetric, we know $\|\M \Sigma_1^{*,\mathcal{T}}(k_1)\|_\infty = \|\M \Sigma_1^{*,\mathcal{T}}(k_1)\|_1$ and then
\bee\nonumber
\|\M \Sigma_1^{*,\mathcal{T}}(k_1)\|_2 &\leq \sqrt{\|\M \Sigma_1^{*,\mathcal{T}}(k_1)\|_\infty \times \|\M \Sigma_1^{*,\mathcal{T}}(k_1)\|_1}
\\
&=\|\M \Sigma_1^{*,\mathcal{T}}(k_1)\|_\infty
\\
&\leq C_0.
\ee
Similarly we can also show $\|\M \Sigma_2^{*}\|_2 \leq C_0$.
\par
By \eqref{main:lm:s:outer} and the bounds of $\big\|\M \Sigma_2^{*,\mathcal{T}}(k_2) - \M \Sigma_2^*\big\|_2, \big\|\M \Sigma_1^{*,\mathcal{T}}(k_1)\big\|_2 , \big\|\M \Sigma_2^*\big\|_2,\big\|\M \Sigma_1^{*,\mathcal{T}}(k_1)-\M \Sigma_1^*\big\|_2$, we finally show
\bee\nonumber
&\big\|\M \Sigma_2^{*,\mathcal{T}}(k_2)\otimes\M \Sigma_1^{*,\mathcal{T}}(k_1) - \M \Sigma^*\big\|^2_2 
\\
&\leq 2C_0^2\big\|\M \Sigma_2^{*,\mathcal{T}}(k_2) - \M \Sigma_2^*\big\|^2_2 + 2C_0^2\times\big\|\M \Sigma_1^{*,\mathcal{T}}(k_1)-\M \Sigma_1^*\big\|^2_2
\\
&\precsim  \M I(k_1< 2p-2)\cdot k_1^{-2\alpha_1} +  \M I(k_2< 2q - 2)\cdot k_2^{-2\alpha_2},
\ee
which completes the proof.
\end{proof}
}
\section{Proof of Theorems}\label{sec:thmpf}
\subsection{Proof of Theorem \ref{T2}}
As discussed in Section \ref{sec:discu:thm1}, our proof strategy is to upper bound the target error by two error terms, and then use Lemmas \ref{l:outer}--\ref{lemma:srate} to bound the two error terms respectively. The notation of this proof is mainly contained in Section \ref{notation:main}. By triangle inequality, we have
\bee\label{pf:t2:triangle}
&\frac{\E\|\hat{\M\Sigma}^\eta_2\kii\otimes\hat{\M\Sigma}^\eta_1\ki - \M\Sigma^*\|_{\F}^2}{pq} 
\\
&\leq 2\cdot \Bigg[ {{\frac{\E\|\M\Sigma^{*,\eta}_2(k_2)\otimes \M\Sigma^{*,\eta}_1(k_1) - \M\Sigma^*\|_{\F}^2}{pq}}} +{{\frac{\E\|\hat{\M\Sigma}^\eta_2\kii\otimes\hat{\M\Sigma}^\eta_1\ki - \M\Sigma^{*,\eta}_2(k_2)\otimes \M\Sigma^{*,\eta}_1(k_1)\|_{\F}^2}{pq}}}\Bigg].
\ee
\par
Recall that $\tilde{\M \Sigma}_{\eta}(k_1,k_2)$ is the doubly banded/tapering matrix of $\hat{\M \Sigma}$ with bandwidths $k_1$ and $k_2$. Since $\M\Sigma^{*}_1,\M\Sigma_2^*$ are in $\mathcal{F}(\varepsilon_0,\alpha)$ or $\mathcal{M}(\varepsilon_0,\alpha)$ matrix class, we know $\M\Sigma^{*}_1$ and $\M\Sigma_2^*$ are positive-definitive and thus have non-zero diagonal entires. Therefore $\M\Sigma^{*,\eta}_1(k_1)$ and $\M\Sigma^{*,\eta}_2(k_2)$ are not zero matrices. Then, taking $\tilde{\M \Sigma}^\diamond = \tilde{\M \Sigma}_{{\eta}}(k_1,k_2)$ and $\tilde{\M\Sigma}^* = \M\Sigma^{*,\eta}_2(k_2)\otimes \M\Sigma^{*,\eta}_1(k_1)$ in  Lemma \ref{T1}, we have
\bee\label{T2:Bbasic}
\frac{\|\hat{\M\Sigma}^\eta_2\kii\otimes\hat{\M\Sigma}^\eta_1\ki - \M\Sigma^{*,\eta}_2(k_2)\otimes \M\Sigma^{*,\eta}_1(k_1)\|^2_{\F}}{{pq}}\leq 8 \cdot\frac{\big\|\xi\big\{\tilde{\M \Sigma}_\eta(k_1,k_2)\big\}-\xi\big\{ \M\Sigma^{*,\eta}_2(k_2)\otimes \M\Sigma^{*,\eta}_1(k_1)\big\}\big\|^2_2}{{pq}}.
\ee
Taking expectation of \eqref{T2:Bbasic} and combining it with \eqref{pf:t2:triangle}, we have 
\bee\label{T2:main}
&\E\Bigg(\frac{\|\hat{\M\Sigma}^\eta_2\kii\otimes\hat{\M\Sigma}^\eta_1\ki - \M\Sigma^*\|^2_{\F}}{{pq}}\Bigg)
\\
&\leq 16\cdot\frac{\E\big\|\xi\big\{\tilde{\M \Sigma}_\eta(k_1,k_2)\big\}-\xi\big\{ \M\Sigma^{*,\eta}_2(k_2)\otimes \M\Sigma^{*,\eta}_1(k_1)\big\}\big\|^2_2}{{pq}} + 2\cdot{{\frac{\|\M\Sigma^{*,\eta}_2(k_2)\otimes \M\Sigma^{*,\eta}_1(k_1) - \M\Sigma^*\|_{\F}^2}{pq}}}
\\
&\asymp \frac{1}{pq}\E\big\|\xi\big\{\tilde{\M \Sigma}_{\eta}(k_1,k_2)\big\}-\xi\big\{\M\Sigma_2^{*,{\eta}}(k_2)\otimes \M\Sigma_1^{*,{\eta}}(k_1)\big\}\big\|_2^2
+\frac{1}{pq}\big\|\M\Sigma_2^{*,{\eta}}(k_2)\otimes \M\Sigma_1^{*,{\eta}}(k_1)-\M\Sigma^{*}\big\|^2_\F,
\ee
where the first inequality above holds because the term $\frac{1}{pq}\big\|\M\Sigma_2^{*,{\eta}}(k_2)\otimes \M\Sigma_1^{*,{\eta}}(k_1)-\M\Sigma^{*}\big\|^2_\F$ is nonrandom. \\

\par
We will now bound the two terms on the right-hand side of the above inequality via Lemmas \ref{l:outer}--\ref{lemma:srate}. For the first term in \eqref{T2:main}, since sub-Gaussian condition directly implies {\color{black} finite} fourth moment condition, we can bound it through Lemma \ref{l:frate} by
\bee\label{T2:B1:a}
&\frac{1}{pq}\E\big\|\xi\big\{\tilde{\M \Sigma}_{\eta}(k_1,k_2)\big\}-\xi\big\{\M\Sigma_2^{*,{\eta}}(k_2)\otimes \M\Sigma_1^{*,{\eta}}(k_1)\big\}\big\|_2^2 
\\
&\leq \frac{1}{pq}\E\big\|\xi\big\{\tilde{\M \Sigma}_{\eta}(k_1,k_2)\big\}-\xi\big\{\M\Sigma_2^{*,{\eta}}(k_2)\otimes \M\Sigma_1^{*,{\eta}}(k_1)\big\}\big\|_\F^2  \quad (\text{By }\|\cdot\|_2\leq \|\cdot\|_\F)
\\
&=\frac{1}{pq}\E\big\|\tilde{\M \Sigma}_{\eta}(k_1,k_2) - \M\Sigma_2^{*,{\eta}}(k_2)\otimes \M\Sigma_1^{*,{\eta}}(k_1)\big\|_\F^2 \quad \text{(By Lemma \ref{lemma:xi})} 
\\
&\precsim \frac{k_1k_2}{n}. \quad \text{(By Lemma \ref{l:frate})}
\ee
On the other hand, we can also bound it through Lemma \ref{lemma:srate}:
\bee\label{T2:B1:b}
\frac{1}{pq}\E{\|\xi\big\{\tilde{\M \Sigma}_{\mathcal{\eta}}(k_1,k_2)\big\} - \xi\{\M\Sigma^{*,\eta}_2(k_2)\otimes \M\Sigma^{*,\eta}_1(k_1)\}\|_{2}^2}{}\precsim \begin{cases}
\frac{k_1}{qn} + \frac{k_2}{pn} & pk_1 + qk_2 \precsim n
\\
\frac{pk^2_1}{qn^2} + \frac{qk^2_2}{pn^2} & pk_1 + qk_2 \succ n.
\end{cases}
\ee
Combining the two upper bounds together, we finally bound the first term in \eqref{T2:main}:
\bee\label{T2:B1:2}
\frac{1}{pq}\E\big\|\xi\big\{\tilde{\M \Sigma}_{\eta}(k_1,k_2)\big\}-\xi\big\{\M\Sigma_2^{*,{\eta}}(k_2)\otimes \M\Sigma_1^{*,{\eta}}(k_1)\big\}\big\|_2^2 \precsim
\begin{cases} \big(\frac{k_1k_2}{n}\big)\wedge\big(\frac{k_1}{qn} + \frac{k_2}{pn}\big)& pk_1 + qk_2 \precsim n
\\
\big(\frac{k_1k_2}{n}\big)\wedge\big(\frac{pk^2_1}{qn^2} + \frac{qk^2_2}{pn^2}  \big) & pk_1 + qk_2 \succ n.
\end{cases}
\ee
In addition, note that $\frac{k_1k_2}{n}\succsim \frac{\max(k_1,k_2)}{n} \asymp \frac{k_1 }{n} + \frac{k_2}{n}\succsim \frac{k_1}{qn} + \frac{k_2}{pn}$, we have
\bee\label{T2:add}
\big(\frac{k_1k_2}{n}\big)\wedge\big(\frac{k_1}{qn} + \frac{k_2}{pn}\big) \asymp \frac{k_1}{qn} + \frac{k_2}{pn}.
\ee
For the second term on the right-hand side of \eqref{T2:main}, by Lemma \ref{l:outer}, we directly have $\frac{1}{pq}\big\|\M\Sigma_2^{*,{\eta}}(k_2)\otimes \M\Sigma_1^{*,{\eta}}(k_1)-\M\Sigma^{*}\big\|^2_\F \precsim\M I_{\eta,p}(k_1)\cdot k_1^{-\tilde{\alpha}_1} +  \M I_{\eta,q}(k_2)\cdot k_2^{-\tilde{\alpha}_2}$, for either $\M \Sigma_1^{*}\in \mathcal{F}(\varepsilon_0, \alpha_1),\M \Sigma_2^{*} \in \mathcal{F}(\varepsilon_0, \alpha_2)$ or $\M \Sigma_1^{*}\in \mathcal{M}(\varepsilon_0, \alpha_1),\M \Sigma_2^{*} \in \mathcal{M}(\varepsilon_0, \alpha_2)$.  Combining \eqref{T2:main}, \eqref{T2:B1:2} and \eqref{T2:add}, we can show
\bee\label{T2:B1:out}
&\E\Big(\frac{\|\hat{\M\Sigma}^\eta_2\kii\otimes\hat{\M\Sigma}^\eta_1\ki - \M\Sigma^*\|_{\F}^2}{pq}\Big) 
\\
&\precsim
\begin{cases}\frac{k_1}{qn} + \frac{k_2}{pn}+ \M I_{\eta,p}(k_1)\cdot k_1^{-\tilde{\alpha}_1} +  \M I_{\eta,q}(k_2)\cdot k_2^{-\tilde{\alpha}_2}, & pk_1 + qk_2 \precsim n
\\
\big(\frac{k_1k_2}{n}\big)\wedge\big(\frac{pk^2_1}{qn^2} + \frac{qk^2_2}{pn^2}  \big)+ \M I_{\eta,p}(k_1)\cdot k_1^{-\tilde\alpha_1} +  \M I_{\eta,q}(k_2)\cdot k_2^{-\tilde\alpha_2}, & pk_1 + qk_2 \succ n.
\end{cases}
\ee
\qed
\subsection{Proof of Theorem \ref{T3}}
Same arguments with the proof of Theorem \ref{T2} can be applied directly to show Theorem \ref{T3}. The only difference is that the sub-Gaussian condition no longer holds for Theorem \ref{T3}. As a result, we can not bound $\frac{1}{pq}\E\|\xi\big\{\tilde{\M \Sigma}_{\mathcal{\eta}}(k_1,k_2)\big\} - \xi\{\M\Sigma^{*,\eta}_2(k_2)\otimes \M\Sigma^{*,\eta}_1(k_1)\}\|_{2}^2$ by both error rates \eqref{T2:B1:b} and \eqref{T2:B1:a}. Instead, we can only bound it by \eqref{T2:B1:a}. Keeping other arguments unchanged in the proof of Theorem \ref{T2}, the desired rate in Theorem \ref{T3} can be derived.\qed
\subsection{Proof of Theorem \ref{T:low}}
We first introduce some notation. We use parameter set $\Theta = \{0,1\}^k$ to identify the underlying distribution $\mathbb{P}(\bds\theta)$ of observations with $\bds\theta \in \Theta$. We denote the Hamming distance for $\bds\theta,\bds\theta' \in \Theta$ as $ H(\bds\theta,\bds\theta') = \sum_{l = 1}^k|\theta_l - \theta_l'|$, where $\theta_l$ is the $l$th coordinate of $\bds \theta$. The following proof generalizes the proof techniques in \citet{Cai2010} for the vector-valued data, where the key component is the following Assouad's Lemma. 
\par
\begin{lemma}[Assouad's Lemma]\label{lm:assou}
Let $\Theta = \{0,1\}^k$ and let $\bds T$ be any estimator based on observations from a distribution in $\{\mathbb{P}(\bds\theta),\bds\theta\in \Theta\}$. For any $s > 0$ and distance metric $d(\cdot,\cdot)$ of target parameters, 
\bee\nonumber
\max_{\bds\theta\in\Theta}2^{\color{black}s}\E_{\bds\theta}\big[d(\M T,\psi(\bds\theta))\big] \geq \min_{ H(\bds\theta,\bds\theta')\geq 1}\frac{d^{\color{black}s}\{\psi(\bds\theta),\psi(\bds\theta')\}}{ H(\bds\theta,\bds\theta')}\cdot \frac{k}{2}\cdot \min_{ H(\bds\theta,\bds\theta') = 1}\|\mathbb{P}(\bds\theta)\wedge\mathbb{P}(\bds\theta')\|,
\ee
where $\|{\mathbb{P}}(\bds\theta)\wedge{\mathbb{P}}(\bds\theta')\|$ is defined as 
\bee\label{def:||}
\|{\mathbb{P}}(\bds\theta)\wedge{\mathbb{P}}(\bds\theta')\|_{} \equiv 1 -\frac{1}{2} \|{\mathbb{P}}(\bds\theta)-{\mathbb{P}}(\bds\theta')\|_1.
\ee
\end{lemma}
\
\par
Our proof strategy is sketched as follows. 
\begin{itemize}
\item[(1)] We first propose a series of underlying distribution $\mathbb{P}(\bds \theta)$.
\item[(2)] We prove that all proposed underlying distributions $\big\{{\mathbb{P}}(\bds \theta)\big| \bds \theta \in{\Theta}\big\} \subseteq \mathcal{P}^n_{\varepsilon_0,\alpha_1,\alpha_2}$. Recall that in Theorem \ref{T:low}, $\mathcal{P}^n_{\varepsilon_0,\alpha_1,\alpha_2}$ is defined as the set of distributions of $\{\vecc(\M X_i)\}_{i = 1}^n$, such that $\vecc(\M X_1), \vecc(\M X_2),\dots,\vecc(\M X_n)$ are i.i.d. sub-Gaussian random vectors in $\mathbb{R}^{pq}$ with any true covariance $\M \Sigma^* = \M \Sigma^*_2 \otimes \M \Sigma^*_1$, where $\M \Sigma_1^{*}\in \mathcal{M}(\varepsilon_0, \alpha_1),\M \Sigma_2^{*} \in \mathcal{M}(\varepsilon_0, \alpha_2)$.  
\item [(3)] Since all proposed distributions are in $\mathcal{P}^n_{\varepsilon_0,\alpha_1,\alpha_2}$, we can apply the Assouad's Lemma on the proposed underlying distributions $\mathbb{P}(\bds \theta)$, and derive the minimax lower bound over the distribution class $\mathcal{P}^n_{\varepsilon_0,\alpha_1,\alpha_2}$. The lower bound in \eqref{lm:low:3:2} can be decomposed into three error terms: $B_1$, $B_2$ and $B_3$.
\item[(4)] We finally obtain the lower bounds for $B_1$, $B_2$ and $B_3$, respectively. Summarizing all the results together, we obtain the target lower bound.
\end{itemize}
We now complement all the details of the proof. 
\par
\
\par
\noindent{\textbf{(i). Construction of $\mathbb{P}(\bds \theta)$:}}
Let $k_p = \min\big\{(nq)^{\frac{1}{2\alpha_1 + 2}}, p/2\big\}, k_q = \min\big\{(np)^{\frac{1}{2\alpha_2 + 2}}, q/2\big\}$, and the dimension of parameter set $\Theta$ be $k_pp - k_p(k_p + 1)/2 + k_qq -k_q(k_q + 1)/2$, i.e., ${\Theta} = \{0,1\}^{k_pp - k_p(k_p + 1)/2 + k_qq -k_q(k_q + 1)/2}$. Then we define the corresponding distribution $\mathbb{P}(\bds \theta)$  as the joint distribution of $n$ i.i.d. samples $\vecc(\M X_1),\dots,\vecc(\M X_n)$, where each $\vecc(\M X_i)$ follows $pq$-dimensional Gaussian distribution $\mathbf N\big(\bds 0_{pq},\M \Sigma^*(\bds \theta)\big)$ with $\M \Sigma^*(\bds \theta)$ parameterized by $\bds \theta \in {\Theta}$. 
%and let the corresponding distribution of this $\{\vecc(\M X_i)\}_{i = 1}^n$ be 
%\bee\label{low:def:dist}\tilde{\mathbb{P}}(\bds \theta)\equiv \underbrace{\mathbf N\big(\bds 0_{pq},\M \Sigma^*(\bds \theta)\big)\times \dots\times \mathbf N\big(\bds 0_{pq},\M \Sigma^*(\bds \theta)\big)}_{n \text{ itmes}}.
%\ee 
We propose a specific one-to-one corresponding $\M\Sigma^*(\bds \theta)$, between $\bds \theta\in\Theta$ and the covariance of $\vecc(\M X_i)$ for each underlying distribution $\mathbb{P}(\bds \theta)$.
\par
Since ${\Theta} = \{0,1\}^{k_pp - k_p(k_p + 1)/2 + k_qq -k_q(k_q + 1)/2}$, we can write it as $\Theta = \Theta^{(p)}\times \Theta^{(q)}$ where $\Theta^{(p)} = \{0,1\}^{k_pp - k_p(k_p + 1)/2},\Theta^{(q)} = \{0,1\}^{k_qq - k_q(k_q + 1)/2}$. Therefore for any $\bds{\theta} \in \Theta$, we have $\bds \theta = (\bds \theta_p,\bds \theta_q)$ for some $\bds \theta_p \in \Theta^{(p)}$ and $\bds \theta_q\in\Theta^{(q)}$. 
\par
Define the following sets.
\bee\label{low:def:2}
&\tilde{\varTheta}^{(p)} \equiv \Big\{{\bds \vartheta_p = \{\theta_{l_1m_1}^{(p)}|\theta_{l_1m_1}^{(p)}  = 0 \text{ or }1, 1\leq|l_1 - m_1|\leq k_p,1\leq l_1\leq m_1\leq p\}}\Big| \text{ all possible }\bds \vartheta_p\Big\};
\\
&\tilde{\varTheta}^{(q)} \equiv \Big\{{\bds \vartheta_q = \{\theta_{l_2m_2}^{(q)}|\theta_{l_2m_2}^{(q)}  = 0 \text{ or }1, 1\leq|l_2 - m_2|\leq k_q,1\leq l_2\leq m_2\leq q\}}\Big| \text{ all possible }\bds \vartheta_q\Big\};
\ee 
For each $\bds\vartheta_p \in \tilde{\varTheta}^{(p)}$, there are totally $k_pp - k_p(k_p + 1)/2$ different pairs of $(l_1,m_1)$, such that each $\theta_{l_1m_1}^{(p)}\in \bds\vartheta_p$ can take either $0$ or $1$ without constraints. Thus $|\tilde\varTheta^{(p)}| = 2^{k_pp - k_p(k_p + 1)/2}$. Similarly,  $|\tilde\varTheta^{(q)}| = 2^{k_qq - k_q(k_q + 1)/2}$. Therefore, for each $\bds \theta_p\in\{0,1\}^{k_pp - k_p(k_p + 1)/2}$ and $\bds \theta_q\in\{0,1\}^{k_qq - k_q(k_q + 1)/2}$, {\color{black} there exists a one-to-one correspondence between $\bds \theta_p$ and $\bds\vartheta_p$, and a one-to-one correspondence between $\bds \theta_q$ and $\bds\vartheta_q$.} So for $\bds \theta = (\bds \theta_p,\bds \theta_q) \in\{0,1\}^{k_pp - k_p(k_p + 1)/2 + k_qq -k_q(k_q + 1)/2}$, there exists a one-to-one correspondence between $\bds \theta$ and $(\bds \vartheta_p,\bds \vartheta_q).$ Without loss of generality, through this proof, we fix these two one-to-one correspondences, i.e., we fix on specific functions $h_1(\cdot), h_2(\cdot)$ such that $\bds\vartheta_p = h_1(\bds \theta_p), \bds\vartheta_q = h_2(\bds \theta_q)$ and $h_1^{-1}(\bds\vartheta_p) = (\bds \theta_p), h_2^{-1}(\bds\vartheta_q) = \bds \theta_q$.
\par
Then for each $\bds \theta  = (\bds \theta_p, \bds \theta_q) \in \Theta$, we define  $\M \Sigma^*(\bds \theta) = \M \Sigma_2^*(\bds \theta_q)\otimes \M \Sigma_1^*(\bds \theta_p)$ where $\M \Sigma_2^*(\bds \theta_q), \M \Sigma_1^*(\bds \theta_p)$ are constructed as follows. Given a positive constant $\gamma > 0$, we define 
\bee\label{def:low:sigma12}
&\M \Sigma_1^*(\bds \theta_p) = \M 1_{p} +\gamma \M F_p(\bds \vartheta_p),
\\
&\M \Sigma_2^*(\bds \theta_q) = \M 1_{q} +\gamma \M F_q(\bds \vartheta_q),
\ee where $\bds 1_d$ is a $d\times d$ identity matrix, $\bds\vartheta_p = h_1(\bds \theta_p), \bds\vartheta_q = h_2(\bds \theta_q)$, and $\M F_p(\bds \vartheta_p), \M F_q(\bds \vartheta_q)$ are defined as follows:
\bee\label{def:Fpq}
\M F_p(\bds \vartheta_p) &\equiv [f_{l_1,m_1}^{\bds \vartheta_p}]_{p\times p} \in \RR^{p\times p} \text{ where }f^{\bds \vartheta_p}_{l_1,m_1} = \begin{cases}\theta_{l_1m_1}^{(p)}(nq)^{-1/2} & 1\leq |l_1 - m_1| \leq k_p\text{ and } l_1\leq m_1
\\
\theta_{m_1l_1}^{(p)}(nq)^{-1/2} & 1\leq |l_1 - m_1| \leq k_p\text{ and } l_1> m_1
\\
0 & \text{Otherwise}
\end{cases},
\\
\M F_q(\bds \vartheta_q) &\equiv [f_{l_2,m_2}^{\bds \vartheta_q}]_{q\times q} \in \RR^{q\times q} \text{ where }f^{\bds \vartheta_q}_{l_2,m_2} = \begin{cases}\theta_{l_2m_2}^{(q)}(np)^{-1/2} & 1\leq |l_2 - m_2| \leq k_q \text{ and } l_2\leq m_2
\\
\theta_{m_2l_2}^{(q)}(np)^{-1/2} & 1\leq |l_2 - m_2| \leq k_q \text{ and } l_2>m_2
\\
0 & \text{Otherwise}
\end{cases}.
\ee
\par
\
\par
\noindent{\textbf{(ii). Proving $\big\{{\mathbb{P}}(\bds \theta)\big| \bds \theta \in{\Theta}\big\} \subseteq \mathcal{P}^n_{\varepsilon_0,\alpha_1,\alpha_2}$:}} We first prove $\M \Sigma_1^*(\bds \theta_1) \in \mathcal{M}(\varepsilon_0,\alpha_1)$, when $\gamma\leq C_1$ with $C_1$ defined in \eqref{A2}, and $n$ is larger than $N_{\varepsilon_0}$ for some positive constant $N_{\varepsilon_0}$ only depending on $\varepsilon_0$. First, when $l_1 \neq m_1$ and $|l_1 - m_1|\leq k_p$, the $l_1m_1$th element of $\M \Sigma_1^*(\bds \theta_1) = \M 1_p + \gamma\M F_p(\bds\vartheta_p)$ is $0$ or $\gamma (nq)^{-1/2}$. Therefore, for $\gamma\leq C_1$, one has
\bee\nonumber
0 &\leq \gamma (nq)^{-1/2} \leq C_1 k_p^{-(\alpha_1 + 1)}\leq C_1|l_1 - m_1|^{-(\alpha_1 + 1)},
\ee
since $k_p = \min\big\{(nq)^{\frac{1}{2\alpha_1 + 2}}, p/2\big\} \leq (nq)^{\frac{1}{2\alpha_1 + 2}}$ and $|l_1 - m_1|\leq k_p$. When $|l_1 - m_1| > k_p$, the $l_1m_1$th element is $0 \leq C_1|l_1 - m_1|^{-(\alpha_1 +1)}$. So $\M \Sigma_1^*(\bds \theta_1)$ satisfies 
\bee\label{lower:ii:1}
\Big|\big[\M \Sigma_1^*(\bds \theta_1)\big]_{l_1,m_1}\Big|\leq C_1|l_1 - m_1|^{-(\alpha_1 + 1)},
\ee
for any $1\leq l_1\leq p$ and $1\leq m_1\leq p$ that $l_1\neq m_1$.
\par
On the other hand, one has
\begin{align}\nonumber
\|\M \Sigma_1^*(\bds \theta_1) - \M 1_p\|_2 
&= \|\gamma\M F_p(\bds \vartheta_p)\|_2 
\\\nonumber
&\leq \|\gamma\M F_p(\bds \vartheta_p)\|_{1} \quad\text{(By Lemma \ref{lm:matrixholder})}
\\\nonumber
&\leq \gamma(2k_p + 1)\cdot (nq)^{-1/2}
\\\nonumber
&\leq \gamma3k_p (nq)^{-1/2}
\\\label{tlow:con1}
&\leq \gamma3(nq)^{\frac{1}{2\alpha_1 + 2}}(nq)^{-1/2} 
\\\nonumber
&= \gamma 3(nq)^{\frac{1}{2\alpha_1 + 2} - \frac{1}{2}} 
\\\nonumber
&\rightarrow 0
\end{align}
when $n \rightarrow +\infty$, since $\alpha_1 > 0$ and $k_p = \min\big\{(nq)^{\frac{1}{2\alpha_1 + 2}},p/2\big\}$. By Lemma \ref{lm:weyl} (Weyl's Theorem), since both $\M \Sigma_1^*(\bds \theta_1)$ and $\M 1_p$ are symmetric matrices,
\bee\label{tlow:con2}
\max_{1\leq i \leq p}|\lambda_i\{\M \Sigma_1^*(\bds \theta_1)\} - 1| &= \max_{1\leq i \leq p}|\lambda_i\{\M \Sigma_1^*(\bds \theta_1)\} - \lambda_i\{\M 1_p\}|
\\\nonumber
&\leq \|\M \Sigma_1^*(\bds \theta_1) - \M 1_p\|_2 
\\\nonumber
&\rightarrow 0
\ee
as $n\rightarrow +\infty$. Here $\lambda_i(\M M)$ is the $i$th largest eigenvalue of matrix $\M M$. This implies all eigenvalues of $\M \Sigma_1^*(\bds \theta_1)$ will uniformly {converge to} $1$ when $n\rightarrow +\infty$. For any $\varepsilon_0 < 1$, by \eqref{tlow:con1} and \eqref{tlow:con2}, there exists $N_{\varepsilon_0}$ such that when $n>N_{\varepsilon_0}$ one has $\max_{1\leq i \leq p}|\lambda_i\{\M \Sigma_1^*(\bds \theta_1)\} - 1| \leq \|\M \Sigma_1^*(\bds \theta_1) - \M 1_p\|_2  < \min\{1/\varepsilon_0 - 1, 1 - \varepsilon_0\}$. Thus 
\bee\label{lower:ii:2}
\lambda_i\{\M \Sigma_1^*(\bds \theta_1)\} \in [\varepsilon_0,1/\varepsilon_0]
\ee for any $1\leq i\leq p$. 
\par
In the following proof, we take $\gamma \leq C_1$ and $n\geq N_{\varepsilon_0}'$. Then, combining \eqref{lower:ii:1} and \eqref{lower:ii:2}, we prove $\M \Sigma_1^*(\bds \theta_p)$ satisfies \eqref{A2} and therefore $ \M \Sigma_1^*(\bds \theta_p)\in \mathcal{M}(\varepsilon_0,\alpha_1)$. By similar argument, we can also prove $\M \Sigma_2^*(\bds \theta_q) \in \mathcal{M}(\varepsilon_0,\alpha_2)$. Then with $\gamma \leq C_1$ and $n\geq N_{\varepsilon_0}'$, we conclude that the corresponding $\vecc(\M X_i)$ is Gaussian for any $\mathbb{P}(\bds{\theta})$, and thus sub-Gaussian with separable covariance $\M\Sigma^*(\bds{\theta}) = \M\Sigma_2^*(\bds{\theta}_q)\otimes \M\Sigma^*_1(\bds{\theta}_p)$ such that $ \M \Sigma_1^*(\bds \theta_p)\in \mathcal{M}(\varepsilon_0,\alpha_1)$ and $\M \Sigma_2^*(\bds \theta_q) \in \mathcal{M}(\varepsilon_0,\alpha_2)$. In addition, all $\vecc(\M X_i)$ are i.i.d. for all $1\leq i\leq n$. To this end, we have proved $\big\{{\mathbb{P}}(\bds \theta)\big| \bds \theta \in{\Theta}\big\} \subseteq \mathcal{P}^n_{\varepsilon_0,\alpha_1,\alpha_2}$.
\
\par
\
\par
\noindent{\textbf{(iii). Application of Assouad's Lemma:}} Since $\big\{{\mathbb{P}}(\bds \theta)\big| \bds \theta \in{\Theta}\big\} \subseteq \mathcal{P}^n_{\varepsilon_0,\alpha_1,\alpha_2}$, by definition we have
\bee\label{lm:low:3}
&\inf_{\hat{\M \Sigma}_n}\sup_{\{\vecc(\M X_i)\}_{i = 1}^n\sim \mathbb{P};\atop\mathbb{P}\in \mathcal{P}^n_{\varepsilon_0,\alpha_1,\alpha_2}}\E\Bigg(\frac{\|\hat{\M \Sigma}_n - \M \Sigma_2^* \otimes \M \Sigma_1^*\|_\F^2}{pq}\Bigg) 
\\
&\succsim \inf_{\hat{\M \Sigma}_n}\sup_{\{\vecc(\M X_i)\}_{i = 1}^n\sim\mathbb{{P}}(\bds \theta)\atop \bds \theta \in{\Theta}}\E\Bigg(\frac{\|\hat{\M \Sigma}_n - \M \Sigma_2^* \otimes \M \Sigma_1^*\|_\F^2}{pq}\Bigg)
\\
&=\inf_{\hat{\M \Sigma}_n}\max_{\bds\theta\in{\Theta}}\E_{\{\vecc(\M X_i)\}_{i = 1}^n \sim {\mathbb{P}}(\bds \theta)}\Bigg(\frac{\|\hat{\M \Sigma}_n - \M \Sigma_2^* \otimes \M \Sigma_1^*\|_\F^2}{pq}\Bigg),
\ee
where $\hat{\M \Sigma}_n$ can be any possible covariance estimator based on $\{\vecc(\M X_i)\}_{i = 1}^n$. In Lemma \ref{lm:assou}, we can set  $\psi(\bds \theta) = \M \Sigma^*(\bds \theta), d\{\M \Sigma^*(\bds\theta),\M \Sigma^*(\bds \theta')\} = \|\M \Sigma^*(\bds \theta) - \M \Sigma^*(\bds \theta')\|_{\F}/\sqrt{pq}$, $s = 2$, and  $\M T = \hat{\M\Sigma}_n$. Then for any covariance estimator $\hat{\M \Sigma}_n$, we have
\bee\nonumber
&\max_{\bds\theta\in{\Theta}}\E_{\{\vecc(\M X_i)\}_{i = 1}^n \sim {\mathbb{P}}(\bds \theta)}\Bigg(\frac{\|\hat{\M \Sigma}_n - \M \Sigma_2^* \otimes \M \Sigma_1^*\|_\F^2}{pq}\Bigg)
\\
&\geq \frac{1}{4}\underbrace{\min_{ H(\bds\theta,\bds\theta')\geq 1\atop \bds \theta,\bds \theta'\in{ \Theta}}\frac{\|\M \Sigma^*(\bds \theta) - \M \Sigma^*(\bds \theta')\|_\F^2}{ pqH(\bds\theta,\bds\theta')}}_{B_1}
\\
&\cdot \underbrace{\frac{k_pp - k_p(k_p + 1)/2 + k_q q -k_q(k_q + 1)/2}{2}}_{B_2}\cdot \underbrace{\min_{ H(\bds\theta,\bds\theta') = 1 \atop{\bds\theta,\bds \theta'\in \M \Theta}}\|\tilde{\mathbb{P}}(\bds\theta)\wedge\tilde{\mathbb{P}}(\bds{\theta}')\|}_{B_3}.
\ee
Combining it with \eqref{lm:low:3} yields
\bee\label{lm:low:3:2}
\inf_{\hat{\M \Sigma}_n}\sup_{\{\vecc(\M X_i)\}_{i = 1}^n\sim \mathbb{P};\atop\mathbb{P}\in \mathcal{P}^n_{\varepsilon_0,\alpha_1,\alpha_2}}\E\Bigg(\frac{\|\hat{\M \Sigma}_n - \M \Sigma_2^* \otimes \M \Sigma_1^*\|_\F^2}{pq}\Bigg)&\succsim \inf_{\hat{\M \Sigma}_n}\frac{1}{4}B_1B_2B_3 = \frac{1}{4}B_1B_2B_3.
\ee
\
\par
\
\par
\noindent{\textbf{(iv). Bound of $B_1$:}} Recall that for any $\bds \theta, \bds \theta' \in {\M \Theta}$, they can be represented as $\bds \theta = (\bds \theta_p,\bds \theta_q), \bds \theta' = (\bds\theta'_p, \bds \theta_q')$ where $\bds \theta_p,\bds \theta_p' \in \{0,1\}^{k_pp - k_p(k_p + 1)/2}$ and $\bds \theta_q,\bds \theta_q' \in \{0,1\}^{k_qq - k_q(k_q + 1)/2}$. Then 
\bee\label{low:HH}
H(\bds \theta,\bds \theta') &=H\Big((\bds\theta_p,\bds \theta_q),(\bds\theta'_p,\bds \theta'_q)\Big) 
\\
&= \sum_{l = 1}^{k_pp - k_p(k_p + 1)/2}|\theta_{l,(p)} - \theta_{l,(p)}'| + \sum_{l = 1}^{k_qq - k_q(k_q + 1)/2}|\theta_{l,(q)} - \theta_{l,(q)}'|
\\
&=H(\bds \theta_p,\bds \theta_p') + H(\bds \theta_q,\bds \theta_q'),
\ee where $\theta_{l,(p)}, \theta_{l,(p)}'$ are the $l$th coordinates of $\bds \theta_p, \bds\theta_p'$ and $\theta_{l,(q)}, \theta_{l,(q)}'$ are the $l$th coordinates of $\bds \theta_q, \bds\theta_q'$.  Therefore, $B_1$ can be rewritten as 
\bee\label{low:b1:1}
B_1 &= \min_{ H(\bds\theta,\bds\theta')\geq 1\atop \bds \theta,\bds \theta'\in\tilde{\M \Theta}}\frac{\|\M \Sigma^*(\bds \theta) - \M \Sigma^*(\bds \theta')\|_\F^2}{ pqH(\bds\theta,\bds\theta')} 
\\
&= \min_{t\geq 1}\min_{ H(\bds\theta,\bds\theta') =t\atop \bds \theta,\bds \theta'\in\tilde{\M \Theta}}\frac{\|\M \Sigma^*(\bds \theta) - \M \Sigma^*(\bds \theta')\|_\F^2}{ pqH(\bds\theta,\bds\theta')}
\\
&=\min_{t\geq 1}\min_{{H(\bds\theta_p,\bds\theta_p') =t_1, H(\bds\theta_q,\bds\theta_q') =t_2
\atop t_1 + t_2 = t; \bds \theta,\bds \theta'\in\tilde{\M \Theta}}\atop t_1 \geq 0, t_2 \geq 0}\frac{\|\M \Sigma^*(\bds \theta) - \M \Sigma^*(\bds \theta')\|_\F^2}{ pqt}.
\ee
\par
Now for any given $t_1, t_2\geq 0$ such that $t = t_1 + t_2$, we derive the lower bound of $\frac{\|\M \Sigma^*(\bds \theta) - \M \Sigma^*(\bds \theta')\|_\F^2}{ pqt}$. Given ${\bds \theta}_q,{\bds \theta}'_q$,  we define the following three sets:
\bee\nonumber
\mathcal{D}_1 &=\big\{(l_2,m_2)\big|[\M \Sigma_2^*(\bds \theta_q)]_{l_2,m_2} = 0, [\M \Sigma_2^*(\bds \theta_q')]_{l_2,m_2} \neq 0\big\};
\\
\mathcal{D}_2 &= \big\{(l_2,m_2)\big|[\M \Sigma_2^*(\bds \theta_q)]_{l_2,m_2} \neq 0, [\M \Sigma_2^*(\bds \theta_q')]_{l_2,m_2} = 0\big\};
\\
\mathcal{D}_3 &= \big\{(l_2,m_2)\big|[\M \Sigma_2^*(\bds \theta_q)]_{l_2,m_2} \neq 0, [\M \Sigma_2^*(\bds \theta_q')]_{l_2,m_2}\neq 0\big\}.
\ee
We also define $D_j=\textbf{Card}(\mathcal{D}_j)$ for $ j=1, 2, 3$, where $\textbf{Card}$ is the cardinality of a set.
\par
Next, we derive the lower bounds of $\{D_j\}_{j = 1}^3$ in terms of $t_1,t_2$. If $(l_2,m_2) \in \mathcal{D}_1$, we have $l_2 \neq m_2$ since all diagonal elements of $\M \Sigma_2^*(\bds \theta_q)$ are non-zero. By the definitions in \eqref{low:def:2} and \eqref{def:Fpq}, we know $(l_2,m_2)\in\mathcal{D}_1$ if and only if the corresponding $[\mathbf{F}_{q}(\bds\vartheta_q)]_{l_2,m_2} = 0, [\mathbf{F}_{q}(\bds\vartheta'_q)]_{l_2,m_2} \neq 0$. Then $(l_2,m_2)\in\mathcal{D}_1$ is further equivalent to $\theta^{(q)}_{l_2,m_2} = 0, \theta'^{(q)}_{l_2,m_2} \neq 0$ when $l_2 < m_2$; and is equivalent to $\theta^{(q)}_{m_2,l_2} = 0, \theta'^{(q)}_{m_2,l_2} \neq 0$ when $l_2 > m_2$ by checking definitions in \eqref{def:Fpq}. We finally summarize the above results and show
\begin{align}\nonumber
D_1 &= \textbf{Card}(\mathcal{D}_1) 
\\
&= \textbf{Card}\big\{(l_2,m_2)\big|[\M \Sigma_2^*(\bds \theta_q)]_{l_2,m_2} = 0, [\M \Sigma_2^*(\bds \theta_q')]_{l_2,m_2} \neq 0\big\}
\\\nonumber
&=\textbf{Card}\big\{(l_2,m_2)\big|[\mathbf{F}_{q}(\bds\vartheta_q)]_{l_2,m_2} = 0, [\mathbf{F}_{q}(\bds\vartheta'_q)]_{l_2,m_2} \neq 0\big\}
\\\nonumber
&=\textbf{Card}\Big[\big\{(l_2,m_2)|\theta^{(q)}_{l_2,m_2} = 0, \theta'^{(q)}_{l_2,m_2} \neq 0,l_2 < m_2\big\}\bigcup\big\{(l_2,m_2)|\theta^{(q)}_{m_2,l_2} = 0, \theta'^{(q)}_{m_2,l_2} \neq 0,l_2 > m_2\big\}\Big]
\\\nonumber
&=2\textbf{Card}\Big[\big\{(l_2,m_2)|\theta^{(q)}_{l_2,m_2} = 0, \theta'^{(q)}_{l_2,m_2} \neq 0,l_2 < m_2\big\}\Big] \quad {\text{(By symmetry)}}
\\\nonumber
&=2\textbf{Card}\Big[\big\{(l_2,m_2)|\theta^{(q)}_{l_2,m_2} = 0, \theta'^{(q)}_{l_2,m_2} \neq 0,\theta^{(q)}_{l_2,m_2}\in\bds \vartheta_q, {\theta'}^{(q)}_{l_2,m_2}\in\bds \vartheta_q'\big\}\Big]\tag*{(In \eqref{low:def:2}, all $\theta^{(q)}_{l_2,m_2}\in{\bds\vartheta}_q$ and ${\theta'}_{l_2,m_2}^{(q)}\in{\bds\vartheta}'_q$ are defined only when $l_2\leq m_2$)}
\\\nonumber
&=2\textbf{Card}\{l|\theta_{l,(q)} = 0, \theta_{l,(q)}' =1, \theta_{l,(q)}\in \bds \theta_q, \theta_{l,(q)}' \in \bds \theta_q'\},
\end{align}
where the last equality holds because we have fixed one-to-one correspondences between $\bds \theta_q$ and $\bds \vartheta_q$, and between $\bds \theta'_q$ and $\bds \vartheta_q'$. By symmetry, we can also show $D_2 = 2\textbf{Card}\{l|\theta_{l,(p)} = 0, \theta_{l,(p)}' =1, \theta_{l,(p)}\in \bds \theta_p, \theta_{l,(p)}' \in \bds \theta_p'\}$. Then we have
\bee\label{lowerbound:con:1}
t_2 &= H(\bds \theta_q,\bds \theta_q') 
\\&= \sum_{l = 1}^{k_qq - k_q(k_q + 1)/2}|\theta_{l,(q)} - \theta_{l,(q)}'| 
\\
&= \sum_{l\in\{l|\theta_{l,(q)} = 0, \theta_{l,(q)}' =1\}}|\theta_{l,(q)} - \theta_{l,(q)}'| + \sum_{l\in\{l|\theta_{l,(q)} = 1, \theta_{l,(q)}' =0\}}|\theta_{l,(q)} - \theta_{l,(q)}'|
\\
&= \frac{1}{2}(D_1 + D_2).
\ee
Similarly, we have
\begin{align}\nonumber
t_1 &=\frac{1}{2}\Big[\textbf{Card}\big[\big\{(l_1,m_1)\big|[\M \Sigma_1^*(\bds \theta_p)]_{l_1,m_1} = 0, [\M \Sigma_1^*(\bds \theta_p')]_{l_1,m_1} \neq 0\big\}\big]
\\\label{low:t1}
&+\textbf{Card}\big[\big\{(l_1,m_1)\big|[\M \Sigma_1^*(\bds \theta_p)]_{l_1,m_1} \neq 0, [\M \Sigma_1^*(\bds \theta_p')]_{l_1,m_1} = 0\big\}\big]\Big].
\end{align}
On the other hand, since by definition the diagonal elements of both $\M \Sigma_2^*(\bds{\theta}_q)$ and $\M \Sigma_2^*(\bds{\theta}'_q)$ are non-zero, we have 
\bee\label{lowerbound:con:2}
D_3\geq q.
\ee
\par
With the lower bounds of $D_1,D_2,D_3$, we can finally bound $\|\M \Sigma^*(\bds \theta) - \M \Sigma^*(\bds \theta')\|_\F^2$ by
\begin{align}\nonumber
&\|\M\Sigma^*(\bds \theta) - \M \Sigma^*(\bds \theta')\|_\F^2 
\\\nonumber
&= \|\M \Sigma^*_2(\bds \theta_q)\otimes\M \Sigma_1^*(\bds \theta_p) - \M \Sigma^*_2(\bds \theta'_q)\otimes\M \Sigma_1^*(\bds \theta'_p)\|_\F^2
\\\nonumber
&=\underbrace{\textbf{Card}\big[\big\{(l_2,m_2)\big|[\M \Sigma_2^*(\bds \theta_q)]_{l_2,m_2} = 0, [\M \Sigma_2^*(\bds \theta_q')]_{l_2,m_2} \neq 0\big\}\big]}_{=D_1} \times[\gamma\times(np^{-1/2})]^2\|\M \Sigma_1^*(\bds\theta_p')\|_\F^2 
\\\nonumber
&+\underbrace{\textbf{Card}\big[\big\{(l_2,m_2)\big|[\M \Sigma_2^*(\bds \theta_q)]_{l_2,m_2} \neq 0, [\M \Sigma_2^*(\bds \theta_q')]_{l_2,m_2} = 0\big\}\big]}_{=D_2} \times[\gamma\times(np^{-1/2})]^2\|\M \Sigma_1^*(\bds\theta_p)\|_\F^2
\\\nonumber
&\small{+\underbrace{\textbf{Card}\big[\big\{(l_2,m_2)\big|[\M \Sigma_2^*(\bds \theta_q)]_{l_2,m_2} \neq 0, [\M \Sigma_2^*(\bds \theta_q')]_{l_2,m_2} \neq 0,l_2 \neq m_2\big\}\big]}_{=D_3 - q \text{ (all diagonal entries of $\M \Sigma_2^*(\bds \theta_q),\M \Sigma_2^*(\bds \theta_q')$ are non-zero)}}\times[\gamma\times(np^{-1/2})]^2\|\M \Sigma_1^*(\bds\theta_p) - \M \Sigma_1^*(\bds\theta_p')\|_\F^2}
\\\label{lowerbound:con:3:1}
&+q\|\M \Sigma_1^*(\bds\theta_p) - \M \Sigma_1^*(\bds\theta_p')\|_\F^2
\\\nonumber
&=D_1 \gamma^2(np)^{-1}\|\M \Sigma_1^*(\bds\theta_p')\|_\F^2 + D_2 \gamma^2(np)^{-1}\|\M \Sigma_1^*(\bds\theta_p)\|_\F^2 + \underbrace{(D_3-q)}_{\geq 0, \text{ by \eqref{lowerbound:con:2}}}\gamma^2(np)^{-1}\|\M \Sigma_1^*(\bds\theta_p) - \M \Sigma_1^*(\bds\theta_p')\|_\F^2
\\\nonumber
&+q\|\M \Sigma_1^*(\bds\theta_p) - \M \Sigma_1^*(\bds\theta_p')\|_\F^2 
\\\label{lowerbound:con:3}
&\geq D_1 \gamma^2(np)^{-1}\|\M \Sigma_1^*(\bds\theta_p')\|_\F^2 + D_2 \gamma^2(np)^{-1}\|\M \Sigma_1^*(\bds\theta_p)\|_\F^2+q\|\M \Sigma_1^*(\bds\theta_p) - \M \Sigma_1^*(\bds\theta_p')\|_\F^2. 
\end{align}
\begin{figure}
   \centering
\includegraphics[width=0.8\textwidth]{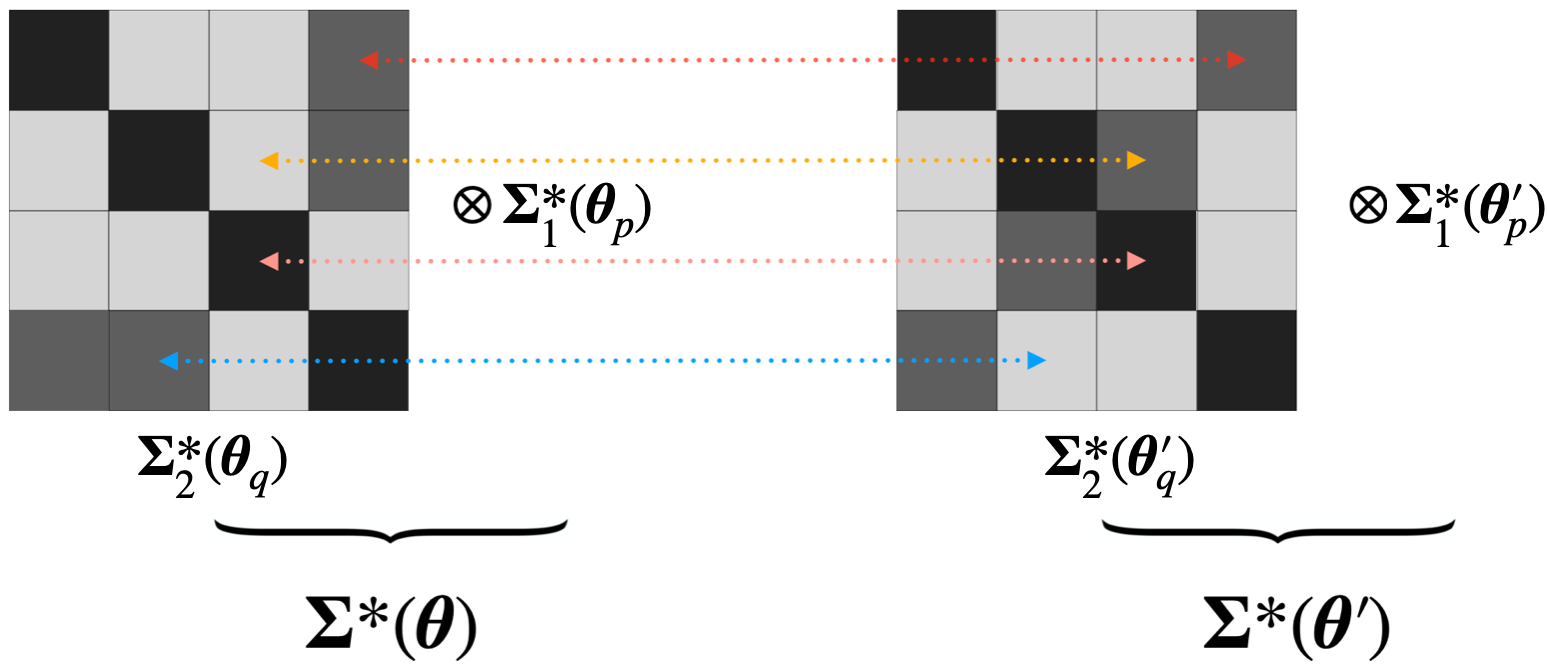}
\caption{Illustration of \eqref{lowerbound:con:3:1}'s derivation. The 
off-white block means the corresponding $[\M \Sigma_2^*(\bds \theta_q)]_{l_2,m_2}$ and $[\M \Sigma_2^*(\bds \theta_q')]_{l_2,m_2}$ both are $0$. The dark grey block means the corresponding $[\M \Sigma_2^*(\bds \theta_q)]_{l_2,m_2}$ and $[\M \Sigma_2^*(\bds \theta_q')]_{l_2,m_2}$ are off-diagonal entries and non-zero, which both are $\tau n^{-1/2}p^{-1/2}$. The black block means the corresponding $[\M \Sigma_2^*(\bds \theta_q)]_{l_2,m_2}$ and $[\M \Sigma_2^*(\bds \theta_q')]_{l_2,m_2}$ are diagonal entries and non-zero, which both are $1$. An entry in $\M \Sigma^*(\bds \theta)- \M \Sigma^*(\bds \theta')$ must be zero if its corresponding $[\M \Sigma_2^*(\bds \theta_q)]_{l_2,m_2}$ and $[\M \Sigma_2^*(\bds \theta_q')]_{l_2,m_2}$ are $0$ (i.e., $[\M \Sigma_2^*(\bds \theta_q)]_{l_2,m_2}$ and $[\M \Sigma_2^*(\bds \theta_q')]_{l_2,m_2}$ are colored with off-white). So an entry in $\M \Sigma^*(\bds \theta) - \M \Sigma^*(\bds \theta')$ is non-zero only if its corresponding $[\M \Sigma_2^*(\bds \theta_q)]_{l_2,m_2},[\M \Sigma_2^*(\bds \theta_q')]_{l_2,m_2}$ are entries of paired blocks that are connected with red, yellow, pink, or blue dashed lines.}\label{fig:vis_mm}
\end{figure}
The derivation of \eqref{lowerbound:con:3:1} can be illustrated by Figure \ref{fig:vis_mm}. In particular, any entry of $\M\Sigma^*(\bds \theta) - \M \Sigma^*(\bds \theta')$ can be represented as  $[\M \Sigma_2^*(\bds \theta_q)]_{l_2,m_2}\cdot [\M \Sigma_2^*(\bds \theta_p)]_{l_1,m_1} - [\M \Sigma_2^*(\bds \theta_q')]_{l_2,m_2}\cdot [\M \Sigma_2^*(\bds \theta_p')]_{l_1,m_1}$, for $1\leq l_1,m_1 \leq p$ and $1\leq l_2,m_2\leq q$. And it is non-zero only if the pair of $[\M \Sigma_2^*(\bds \theta_q)]_{l_2,m_2},[\M \Sigma_2^*(\bds \theta_q')]_{l_2,m_2}$ is one type of the paired blocks shown in Figure \ref{fig:vis_mm}. Note the paired blocks are connected with red/yellow/pink/blue dashed lines. These four paired blocks represent four different scenarios, and we will derive \eqref{lowerbound:con:3:1} by considering each of the four scenarios, respectively.
\begin{itemize}
\item[\textbf{(1)}] When $[\M \Sigma_2^*(\bds \theta_q')]_{l_2,m_2} \neq 0,[\M \Sigma_2^*(\bds \theta_q)]_{l_2,m_2} = 0$ and $1\leq l_2,m_2\leq q$, they are entries of paired blocks connected with the yellow dashed line. There are $D_1$ different pairs of $(l_2,m_2)$ satisfying this condition. One can easily see that $[\M \Sigma_2^*(\bds \theta_q')]_{l_2,m_2}$ must be off-diagonal entry. If we fix one pair of $(l_2,m_2)$, denoted by $(\tilde{l}_2,\tilde{m}_2)$, then for any  $(l_1,m_1)$, we have
$$\underbrace{[\M \Sigma_2^*(\bds \theta_q)]_{\tilde{l}_2,\tilde{m}_2}}_{=0}\cdot [\M \Sigma_2^*(\bds \theta_p)]_{l_1,m_1} - [\M \Sigma_2^*(\bds \theta_q')]_{\tilde{l}_2,\tilde{m}_2}\cdot [\M \Sigma_2^*(\bds \theta_p')]_{l_1,m_1} = -\gamma (np)^{-1/2}[\M \Sigma_2^*(\bds \theta_p')]_{l_1,m_1}
$$
since $[\M \Sigma_2^*(\bds \theta_q')]_{\tilde{l}_2,\tilde{m}_2} = -\gamma (np)^{-1/2}$ by definition \eqref{def:low:sigma12}. Then for fixed $(\tilde{l}_2,\tilde{m}_2)$, the sum of squares of all entries in $\M\Sigma^*(\bds \theta) - \M \Sigma^*(\bds \theta')$ over all different $(l_1,m_1)$ is $[\gamma n^{-1/2}p^{-1/2}]^2\times\|\M \Sigma_1^*(\bds\theta_p')\|_\F^2 $. Thus by counting all different pairs of $(l_2,m_2)$ satisfying the condition in this scenario, the sum of squares for all entries in $\M\Sigma^*(\bds \theta) - \M \Sigma^*(\bds \theta')$ satisfying $[\M \Sigma_2^*(\bds \theta_q')]_{l_2,m_2} \neq 0,[\M \Sigma_2^*(\bds \theta_q)]_{l_2,m_2} = 0$ is $D_1[\gamma n^{-1/2}p^{-1/2}]^2\times\|\M \Sigma_1^*(\bds\theta_p')\|_\F^2$, which is the first term. 
\item[\textbf{(2)}] When $[\M \Sigma_2^*(\bds \theta_q')]_{l_2,m_2} = 0,[\M \Sigma_2^*(\bds \theta_q)]_{l_2,m_2} \neq 0$ and $1\leq l_2,m_2\leq q$, they are entries of paired blocks connected with blue dashed line. By a similar argument used in scenario \textbf{(1)}, we can show that the sum of squares for all entries in $\M\Sigma^*(\bds \theta) - \M \Sigma^*(\bds \theta')$ satisfying $[\M \Sigma_2^*(\bds \theta_q')]_{l_2,m_2} = 0,[\M \Sigma_2^*(\bds \theta_q)]_{l_2,m_2} \neq 0$ is $D_2[\gamma n^{-1/2}p^{-1/2}]^2\times\|\M \Sigma_1^*(\bds\theta_p)\|_\F^2$, which is the second term in \eqref{lowerbound:con:3:1}
\item[\textbf{iii.}] When $[\M \Sigma_2^*(\bds \theta_q')]_{l_2,m_2} \neq 0,[\M \Sigma_2^*(\bds \theta_q)]_{l_2,m_2} \neq 0$ and $1\leq l_2 \neq m_2\leq q$ (off-diagonal), they are entries of paired blocks connected with red dashed line. There are totally $D_3 -q$ different  pairs of $(l_2,m_2)$ satisfying this condition. Since $(l_2,m_2)$ is off-diagonal, by definition we know $[\M \Sigma_2^*(\bds \theta_q')]_{l_2,m_2} = [\M \Sigma_2^*(\bds \theta_q)]_{l_2,m_2} = \gamma n^{-1/2}p^{-1/2}$. If we fix a pair of $(l_2,m_2)$, denoted by $(\tilde{l}_2,\tilde{m}_2)$, for any $(l_1,m_1)$, we have 
\bee\nonumber
&[\M \Sigma_2^*(\bds \theta_q)]_{\tilde{l}_2,\tilde{m}_2}\cdot [\M \Sigma_2^*(\bds \theta_p)]_{l_1,m_1} - [\M \Sigma_2^*(\bds \theta_q')]_{\tilde{l}_2,\tilde{m}_2}\cdot [\M \Sigma_2^*(\bds \theta_p')]_{l_1,m_1} 
\\
&=\gamma (np)^{-1/2}\times\Big\{[\M \Sigma_2^*(\bds \theta_p)]_{l_1,m_1}-[\M \Sigma_2^*(\bds \theta_p')]_{l_1,m_1}\Big\},
\ee
which, similar to the arguments in scenarios \textbf{(1)} and \textbf{(2)}, implies the sum of squares for all entries in $\M\Sigma^*(\bds \theta) - \M \Sigma^*(\bds \theta')$ that satisfies $[\M \Sigma_2^*(\bds \theta_q')]_{l_2,m_2} \neq 0,[\M \Sigma_2^*(\bds \theta_q)]_{l_2,m_2} \neq 0$ and $l_2 \neq m_2$, is $(D_3 - q)\times \{\gamma n^{-1/2}p^{-1/2}\}^2\times\|\M \Sigma_1^*(\bds\theta_p) - \M \Sigma_1^*(\bds\theta_p')\|_\F^2$, which is exactly the third term in \eqref{lowerbound:con:3:1}.
\item[\textbf{iv.}] When $[\M \Sigma_2^*(\bds \theta_q')]_{l_2,m_2} \neq 0,[\M \Sigma_2^*(\bds \theta_q)]_{l_2,m_2} \neq 0$,  where $1\leq l_2 = m_2\leq q$ (diagonal), they are entries of paired blocks connected with pink dashed line. By checking definition \eqref{def:low:sigma12}, we know all diagonal entries $[\M \Sigma_2^*(\bds \theta_q')]_{l_2,m_2}= [\M \Sigma_2^*(\bds \theta_q)]_{l_2,m_2} =1$ satisfying this condition and there are totally $q$ different diagonal pairs of $(l_2,m_2)$. For any $(l_1,m_1)$, if we fix on a $(\tilde{l}_2,\tilde{m}_2)$ such that $\tilde{l}_2 = \tilde{m}_2$,  we have
\bee\nonumber
&{[\M \Sigma_2^*(\bds \theta_q)]_{\tilde{l}_2,\tilde{m}_2}}\cdot [\M \Sigma_2^*(\bds \theta_p)]_{l_1,m_1} - [\M \Sigma_2^*(\bds \theta_q')]_{\tilde{l}_2,\tilde{m}_2}\cdot [\M \Sigma_2^*(\bds \theta_p')]_{l_1,m_1}  
\\
&= [\M \Sigma_2^*(\bds \theta_p)]_{l_1,m_1} - [\M \Sigma_2^*(\bds \theta_p')]_{l_1,m_1},
\ee 
which implies, with similar argument as previous conditions, the sum of squares of all entries in $\M\Sigma^*(\bds \theta) - \M \Sigma^*(\bds \theta')$ that satisfies $[\M \Sigma_2^*(\bds \theta_q')]_{l_2,m_2} \neq 0,[\M \Sigma_2^*(\bds \theta_q)]_{l_2,m_2} \neq 0$ and $l_2 = m_2$, is $q \times \|\M \Sigma_1^*(\bds\theta_p) - \M \Sigma_1^*(\bds\theta_p')\|_\F^2 $, which is the fourth term in \eqref{lowerbound:con:3:1}.
\end{itemize}
\par
Summarizing results in \textbf{i}--\textbf{iv}, we finally prove \eqref{lowerbound:con:3}. 

To further lower bound the right-hand side of \eqref{lowerbound:con:3}, we need to obtain lower bounds for $\|\M \Sigma_1^*(\bds\theta_p')\|_\F^2, \|\M \Sigma_1^*(\bds\theta_p)\|_\F^2$ and $\|\M \Sigma_1^*(\bds\theta_p) - \Sigma_1^*(\bds\theta_p')\|_\F^2$. We note by definition, all diagonal elements in $p \times p$ matrices $\M \Sigma_1^*(\bds\theta_p')$ and $\M \Sigma_1^*(\bds\theta_p)$ are $1$. Thus 
\bee\label{lowerbound:con:3:4}
\|\M \Sigma_1^*(\bds\theta_p')\|_\F^2\geq p \text{ and }\|\M \Sigma_1^*(\bds\theta_p)\|_\F^2\geq p.
\ee
Also by definition, all non-zero entries in $\M \Sigma_1^*(\bds\theta_p) - \M \Sigma_1^*(\bds\theta_p')$ are off-diagonal and equal to $\gamma n^{-1/2}q^{-1/2}$ or $-\gamma n^{-1/2}q^{-1/2}$. And an entry in $\M \Sigma_1^*(\bds\theta_p) - \M \Sigma_1^*(\bds\theta_p')$ is non-zero if and only if, either $[\M \Sigma_1^*(\bds \theta_p)]_{l_1,m_1} \neq 0, [\M \Sigma_1^*(\bds \theta_p')]_{l_1,m_1} = 0$ or $[\M \Sigma_1^*(\bds \theta_p)]_{l_1,m_1} = 0, [\M \Sigma_1^*(\bds \theta_p')]_{l_1,m_1} \neq 0$. Then we have
\bee\label{lowerbound:con:3:5}
\|\M \Sigma_1^*(\bds\theta_p) - \Sigma_1^*(\bds\theta_p')\|_\F^2 &= \{\gamma n^{-1/2}q^{-1/2}\}^2 \times \Bigg[\textbf{Card}\big[\big\{(l_1,m_1)\big|[\M \Sigma_1^*(\bds \theta_p)]_{l_1,m_1} = 0, [\M \Sigma_1^*(\bds \theta_p')]_{l_1,m_1} \neq 0\big\}\big]
\\
&+\textbf{Card}\big[\big\{(l_1,m_1)\big|[\M \Sigma_1^*(\bds \theta_p)]_{l_1,m_1} \neq 0, [\M \Sigma_1^*(\bds \theta_p')]_{l_1,m_1} = 0\big\}\big]\Bigg]
\\
&=\gamma^2n^{-1}q^{-1}\times 2t_1,
\ee
where the last equality holds by \eqref{low:t1}. 
Combining \eqref{lowerbound:con:1}, \eqref{lowerbound:con:2}, \eqref{lowerbound:con:3}, \eqref{lowerbound:con:3:4} and \eqref{lowerbound:con:3:5}, we finally have
\bee\nonumber
\|\M\Sigma^*(\bds \theta) - \M \Sigma^*(\bds \theta')\|_\F^2&\geq D_1 \gamma^2(np)^{-1}\|\M \Sigma_1^*(\bds\theta_p')\|_\F^2 + D_2 \gamma^2(np)^{-1}\|\M \Sigma_1^*(\bds\theta_p)\|_\F^2+q\|\M \Sigma_1^*(\bds\theta_p) - \M \Sigma_1^*(\bds\theta_p')\|_\F^2
\\
&\geq D_1 \gamma^2(np)^{-1}p + D_2 \gamma^2(np)^{-1}p+q2t_1\gamma^2(nq)^{-1}
\\
&=\gamma^2(n)^{-1}\underbrace{(D_1 + D_2)}_{=2t_2} +2t_1\gamma^2(n)^{-1} 
\\
&= 2(t_1 + t_2)\frac{\gamma^2}{n} 
\\
&= \frac{2t\gamma^2}{n}.
\ee
Combining with \eqref{low:b1:1}, one has
\bee\label{res:B1:final}
B_1 &= \min_{t\geq 1}\min_{ H(\bds\theta_p,\bds\theta_p') =t_1, H(\bds\theta_q,\bds\theta_q') =t_2
\atop t_1 + t_2 = t; \bds \theta,\bds \theta'\in\tilde{\M \Theta}}\frac{\|\M \Sigma^*(\bds \theta) - \M \Sigma^*(\bds \theta')\|_\F^2}{ pqt}
\\
&\geq \min_{t\geq 1}\min_{ H(\bds\theta_p,\bds\theta_p') =t_1, H(\bds\theta_q,\bds\theta_q') =t_2
\atop t_1 + t_2 = t; \bds \theta,\bds \theta'\in\tilde{\M \Theta}}\frac{1}{pqt} \times\frac{2t\gamma^2}{n}  
\\
&= \frac{2\gamma^2}{npq}.
\ee
\par
\
\par

\noindent{\textbf{(v). Bound of $B_2$:}} By the definitions of $k_p$ and $k_q$, $B_2$ can be directly lower bounded by
\bee\label{res:B2:final}
B_2 &= \frac{k_pp - k_p(k_p + 1)/2 + k_qq -k_q(k_q + 1)/2}{2} 
\\
&\geq \frac{k_pp - k_p(p/2 + 1)/2 + k_qq -k_q(q/2 + 1)/2}{2}
\\
&\succsim \frac{k_p(\frac{3}{4}p - \frac{1}{2}) +k_q(\frac{3}{4}q - \frac{1}{2})}{2}
\\
&\succsim k_p p+ k_qq\asymp\max\{k_pp,k_qq\}.
\ee
\par
\
\par

\noindent{\textbf{(vi). Bound of $B_3$:}} Recall $B_3 = \min_{ H(\bds\theta,\bds\theta') = 1\atop \bds \theta,\bds \theta'\in{\M \Theta}}\|{\mathbb{P}}(\bds\theta)\wedge{\mathbb{P}}(\bds\theta')\|$. We focus on the lower bound of $\|{\mathbb{P}}(\bds\theta)\wedge{\mathbb{P}}(\bds\theta')\|$ when $H(\bds\theta,\bds\theta') = 1$. In the following proof, we fix a pair of $\bds \theta,\bds \theta'$ with $H(\bds\theta,\bds\theta') = 1$. By \eqref{low:HH}, we know,
\bee\nonumber
H(\bds\theta_p,\bds\theta'_p) + H(\bds\theta_q,\bds\theta'_q)= 1.
\ee
Thus either $H(\bds\theta_p,\bds\theta'_p) = 1, H(\bds\theta_q,\bds\theta'_q) = 0$ or $H(\bds\theta_p,\bds\theta'_p) = 0, H(\bds\theta_q,\bds\theta'_q) = 1$. By symmetry, without loss of generality, we assume $H(\bds\theta_p,\bds\theta'_p) = 1, H(\bds\theta_q,\bds\theta'_q) = 0$, and so $\bds\theta_q = \bds\theta'_q$. Then $\M \Sigma^*_2(\bds \theta_q)$ and $\M \Sigma^*_2(\bds \theta_q')$ are equal, and $\M \Sigma^*_1(\bds \theta_p)$ and $\M \Sigma^*_1(\bds \theta_p')$ have only two different entries. The two entires that have different values in $\M \Sigma^*_1(\bds \theta_p)$ and $\M \Sigma^*_1(\bds \theta_p')$ are off-diagonal and are symmetric about the diagonal. By the definition of ${\mathbb{P}}(\bds \theta)$, we know ${\mathbb{P}}(\bds \theta)$ is the joint distribution of $n$ i.i.d. $\vecc(\M X_i) \sim \mathbf{N}\big\{\bds{0}_{pq},\M \Sigma^*_2(\bds \theta_q)\otimes \M \Sigma^*_1(\bds \theta_p)\big\}$.  By the definition in \eqref{def:||}, we have 
\bee\label{low:final:4}
\|{\mathbb{P}}(\bds\theta)\wedge{\mathbb{P}}(\bds\theta')\|_{} = 1 -\frac{1}{2} \|{\mathbb{P}}(\bds\theta)-{\mathbb{P}}(\bds\theta')\|_1,
\ee
where $\|{\mathbb{P}}(\bds\theta)-{\mathbb{P}}(\bds\theta')\|_1$ is the $L_1$ norm of ${\mathbb{P}}(\bds\theta)-{\mathbb{P}}(\bds\theta')$. To show the lower bound of $\|{\mathbb{P}}(\bds\theta)\wedge{\mathbb{P}}(\bds\theta')\|$ is a positive constant, we only need to show $\|{\mathbb{P}}(\bds\theta)-{\mathbb{P}}(\bds\theta')\|_1 $ can be bounded by a sufficient small constant as $n\rightarrow +\infty$. Define $K(\cdot\mid \cdot)$ as the Kullback-Leibler (KL) divergence. By Pinsker's inequality and the KL divergence of multivariate Gaussian distributions, we have
\bee\label{low:final:1}
\|{\mathbb{P}}(\bds\theta)-{\mathbb{P}}(\bds\theta')\|_1^2&\leq 2K\{{\mathbb{P}}(\bds\theta)\mid {\mathbb{P}}(\bds\theta')\} 
\\
&= 2n\Big[\underbrace{\frac{1}{2} \operatorname{tr}\left\{\M\Sigma^*(\bds\theta){\M\Sigma^*}(\bds\theta')^{-1}\right\}-\frac{pq}{2}}_{B_{31}} -\underbrace{\frac{1}{2} \log \operatorname{det}\left\{\M\Sigma^*(\bds\theta){\M\Sigma^*}(\bds\theta')^{-1}\right\}}_{B_{32}}\Big],
\ee
where $\M \Sigma^*(\bds \theta) = \M \Sigma^*_2(\bds \theta_q)\otimes \M \Sigma^*_1(\bds \theta_p), \M \Sigma^*(\bds \theta') = \M \Sigma^*_2(\bds \theta'_q)\otimes \M \Sigma^*_1(\bds \theta_p') = \M \Sigma^*_2(\bds \theta_q)\otimes \M \Sigma^*_1(\bds \theta_p')$ since $\bds \theta_q = \bds \theta_q'$. Here we note ${\M\Sigma^*}(\bds\theta')$ is invertible because we have shown in \eqref{lower:ii:2} that the smallest eigenvalue of ${\M\Sigma^*_1}(\bds\theta_p')$ is lower bounded by some positive constant, and same for ${\M\Sigma^*_2}(\bds\theta_q') = {\M\Sigma^*_2}(\bds\theta_q)$. Then by the property of Kronecker product matrix, we know ${\M\Sigma^*}(\bds\theta')$ is also invertible, and furthermore ${\M\Sigma^*}(\bds\theta')^{-1}= {\M\Sigma^*_2}(\bds\theta_q)^{-1}\otimes {\M\Sigma^*_1}(\bds\theta_p')^{-1}$. 
\par
We then define $\M D \equiv \M \Sigma^*_1(\bds \theta_p) - \M \Sigma^*_1(\bds \theta_p')$ and have 
\bee\label{low:diff:sigma}
\M \Sigma^*(\bds \theta) - \M \Sigma^*(\bds \theta')  &= \M \Sigma^*_2(\bds \theta_q)\otimes \M \Sigma^*_1(\bds \theta_p) - \M \Sigma^*_2(\bds \theta_q)\otimes \M \Sigma^*_1(\bds \theta_p')
\\
&=\M \Sigma_2^*(\bds\theta_q)\otimes\{\M \Sigma^*_1(\bds \theta_p) - \M \Sigma^*_1(\bds \theta_p')\}
\\
&=\M \Sigma_2^*(\bds\theta_q)\otimes \M D.
\ee
Next, we calculate $B_{31}$ and $B_{32}$ respectively.
\begin{itemize}
\item[\textbf{i.}] For $B_{31}$, by \eqref{low:diff:sigma}, we can directly see
\bee\label{low:b3:3}
\operatorname{tr}\left\{\M\Sigma^*(\bds\theta){\M\Sigma^*}(\bds\theta')^{-1}\right\} &= \tr\Big[\Big\{{\M\Sigma^*}(\bds\theta') +\M \Sigma_2^*(\bds\theta_q)\otimes \M D\Big\}\times{\M\Sigma^*}(\bds\theta')^{-1}\Big]
\\
&=\tr\Big[\M 1_{pq} +\{\M \Sigma_2^*(\bds\theta_q)\otimes \M D\}\times {\M\Sigma^*}(\bds\theta')^{-1}\Big]
\\
&= pq + \tr[\{\M \Sigma_2^*(\bds\theta_q)\otimes \M D\}\times{\M\Sigma^*}(\bds\theta')^{-1}].
\ee
Then we have
\bee\label{low:b3:1}
B_{31} &= \frac{1}{2}\operatorname{tr}\left\{\M\Sigma^*(\bds\theta){\M\Sigma^*}(\bds\theta')^{-1}\right\}-\frac{pq}{2}
\\
&=\frac{pq}{2} + \frac{1}{2}\tr[\{\M \Sigma_2^*(\bds\theta_q)\otimes \M D\}\times{\M\Sigma^*}(\bds\theta')^{-1}]-\frac{pq}{2} 
\\
&= \frac{1}{2}\tr[\{\M \Sigma_2^*(\bds\theta_q)\otimes \M D\}\times{\M\Sigma^*}(\bds\theta')^{-1}].
\ee
By the multiplication property of Kronecker product matrices ({Lemma 4.2.10} in \citet{horn1991topics}), we have
\bee\label{low:b3:2}
\{\M \Sigma_2^*(\bds\theta_q)\otimes \M D\}\times{\M\Sigma^*}(\bds\theta')^{-1} &= \big\{\M \Sigma_2^*(\bds\theta_q)\otimes \M D\big\}\times\big\{\M \Sigma_2^*(\bds\theta_q)^{-1}\otimes \M \Sigma_1^*(\bds\theta_p')^{-1}\big\}
\\
&= \big\{\M \Sigma_2^*(\bds\theta_q)\times\M \Sigma_1^*(\bds\theta_p')^{-1}\big\}\otimes\big\{\M D\times\M \Sigma_1^*(\bds\theta_p')^{-1}\big\} 
\\
&= \M 1_q\otimes\big\{\M D\times\M \Sigma_1^*(\bds\theta_p')^{-1}\big\}.
\ee
Combining \eqref{low:b3:1} and \eqref{low:b3:2}, by the trace property of Kronecker product matrix, we have
\bee\label{low:final:2}
B_{31} &= \frac{1}{2}\tr[\M 1_q\otimes\big\{\M D\times\M \Sigma_1^*(\bds\theta_p')^{-1}\big\}] 
\\
&= \frac{q}{2}\tr[\M D\times\M \Sigma_1^*(\bds\theta_p')^{-1}].
\ee
\par
\item[\textbf{ii.}] For $B_{32}$, by \eqref{low:b3:3} and \eqref{low:b3:2}, we have 
\bee\nonumber
B_{32} &= \frac{1}{2} \log \operatorname{det}\left\{\M\Sigma^*(\bds\theta){\M\Sigma^*}(\bds\theta')^{-1}\right\} 
\\
&= \frac{1}{2}\log\det\Big[\M 1_{pq} +\{\M \Sigma_2^*(\bds\theta_q)\otimes \M D\}\times {\M\Sigma^*}(\bds\theta')^{-1}\Big]
\\
&=\frac{1}{2}\log\det\Big[\M 1_{pq} +\M 1_q\otimes\big\{\M D\times\M \Sigma_1^*(\bds\theta_p')^{-1}\big\}\Big].
\ee
Denote $\Lambda_{pq} \equiv \{\tilde{\lambda}_i\}_{i = 1}^{pq}$, all eigenvalues of $\M 1_q\otimes\big\{\M D\times\M \Sigma_1^*(\bds\theta_p')^{-1}\big\}$, and $\Lambda_p \equiv \{\bar{\lambda}_i\}_{i = 1}^{p}$, all eigenvalues of $\M D\times\M \Sigma_1^*(\bds\theta_p')^{-1}$. By basic property of Kronecker product ({By Theorem 4.2.12} in \citet{horn1991topics}), we have $\Lambda_{pq} = \bigcup_{i = 1}^q\Lambda_p$. Next we focus on the property of eigenvalues in $\Lambda_p$, i.e., the eigenvalues of $\M D\times\M \Sigma_1^*(\bds\theta_p')^{-1}$. Since $\M D\times\M \Sigma_1^*(\bds\theta_p')^{-1}$ is similar to symmetric matrix $\M \Sigma_1^*(\bds\theta_p')^{-1/2}\M D\M \Sigma_1^*(\bds\theta_p')^{-1/2}$, $\Lambda_p$ is also the set of eigenvalues of $\M \Sigma_1^*(\bds\theta_p')^{-1/2}\M D\M \Sigma_1^*(\bds\theta_p')^{-1/2}$. We have
\bee\label{low:haha}
\|\M \Sigma_1^*(\bds\theta_p')^{-1/2}\M D\M \Sigma_1^*(\bds\theta_p')^{-1/2}\|_2&\leq\|\M \Sigma_1^*(\bds\theta_p')^{-1/2}\|_2^2\times\|\M D\|
\\
&\leq\|\M \Sigma_1^*(\bds\theta_p')^{-1/2}\|_2^2\times\|\M D\|_\F
\\
&\leq \frac{1}{\sqrt{\varepsilon_0}^2}\|\M \Sigma^*_1(\bds \theta_p) - \M \Sigma^*_1(\bds \theta_p')\|_\F 
\\
&=\frac{\sqrt{2(nq)^{-1}\gamma^2}}{\varepsilon_0},
\ee
where the third inequality holds because $\M \Sigma_1^*(\bds \theta_1') \in \mathcal{M}(\varepsilon_0,\alpha_1)$, and the last equality holds because $H(\bds \theta_p,\bds \theta_p') = 1$ and all off-diagonal entries of $\M \Sigma_1^*(\bds\theta_p)$ or $\M\Sigma_1^*(\bds\theta_p')$ are $\gamma n^{-1/2}q^{-1/2}$. Then all eigenvalues of $\Lambda_q$ are real since $\M \Sigma_1^*(\bds\theta_p')^{-1/2}\M D\M \Sigma_1^*(\bds\theta_p')^{-1/2}$ is symmetric. In addition all values in $\Lambda_q$ fall within the interval $\Big[-\frac{\sqrt{2(nq)^{-1}\tau^2}}{\varepsilon_0},\frac{\sqrt{2(nq)^{-1}\tau^2}}{\varepsilon_0}\Big]$, where $\frac{\sqrt{2(nq)^{-1}\tau^2}}{\varepsilon_0}\rightarrow 0$ as $n\rightarrow +\infty$. Since $\Lambda_{pq} = \bigcup_{i = 1}^q\Lambda_p$, the eigenvalues of $\M 1_q\otimes\big\{\M D\times\M \Sigma_1^*(\bds\theta_p')^{-1}\big\}$ have the same property. Then similar to proof of Lemma 5 in \citet{Cai2010}, Taylor expansion yields
\bee\label{low:final:3}
B_{32} &= \frac{1}{2}\log\det\Big[\M 1_{pq} +\M 1_q\otimes\big\{\M D\times\M \Sigma_1^*(\bds\theta_p')^{-1}\big\}\Big] 
\\
&={\frac{1}{2}\tr\Big[\M 1_q\otimes\big\{\M D\times\M \Sigma_1^*(\bds\theta_p')^{-1}\big\}\Big]}- R
\\
&=\frac{q}{2}\tr[\M D\times\M \Sigma_1^*(\bds\theta_p')^{-1}] - R
\\
&=B_{31} - R,
\ee
where the last equality holds by \eqref{low:final:2}, and $R \leq c_3\sum_{\tilde{\lambda}_i \in\Lambda_{pq}}\tilde{\lambda}_i^2$ for some fixed constant $c_3>0$.
\end{itemize}
\par
Combining \eqref{low:final:1}, \eqref{low:final:2} and \eqref{low:final:3}, we have
\begin{align}
\nonumber
\|{\mathbb{P}}(\bds\theta)-{\mathbb{P}}(\bds\theta')\|_1^2&\leq 2n[B_{31} - B_{32}]
\\\nonumber
& = 2nR 
\\\nonumber
&\leq 2nc_3\sum_{\tilde{\lambda}_i \in\Lambda_{pq}}\tilde{\lambda}_i^2
\\\nonumber
&=2nc_3q\sum_{\bar{\lambda}_i\in\Lambda_p}\bar{\lambda}_i^2 \quad  (\text{By } \Lambda_{pq} = \bigcup_{i = 1}^q\Lambda_p)
\\\nonumber
&=2nc_3q\|\M \Sigma_1^*(\bds\theta_p')^{-1/2}\M D\M \Sigma_1^*(\bds\theta_p')^{-1/2}\|_\F^2
\\\nonumber
&\leq 2nc_3q \|\M \Sigma_1^*(\bds\theta_p')^{-1/2}\|^4\|\M D\|^2_\F
\\\nonumber
&\leq  \frac{4c_3}{\varepsilon_0^2}\gamma^2,
\end{align}
where the third equality holds because $\Lambda_p$ are eigenvalues of  $\M \Sigma_1^*(\bds\theta_p')^{-1/2}\M D\M \Sigma_1^*(\bds\theta_p')^{-1/2}$, and the last inequality can be derived similarly to \eqref{low:haha}. With the inequality above, if $\gamma = \varepsilon_0/2\sqrt{c_3}$, we have $0\leq \|{\mathbb{P}}(\bds\theta)-{\mathbb{P}}(\bds\theta')\|_1\leq 1$. Then by \eqref{low:final:4}, one has
\bee\label{res:B3:final}
1/2\leq B_3 \leq 1.
\ee
Note that throughout the proof, we additionally need $\gamma \leq C_1$. Thus, we finally should pick $\gamma = \min\{\varepsilon_0/2\sqrt{c_3}, C_1\}$. 
\
\par
\
\par
\noindent{\textbf{Summary: }}We prove our target results by combining \eqref{lm:low:3:2}, \eqref{res:B1:final}, \eqref{res:B2:final}, \eqref{res:B3:final} and the definitions of $k_p,k_q$. In particular, one has
\bee\nonumber
&\inf_{\hat{\M \Sigma}_n}\sup_{\{\vecc(\M X_i)\}_{i = 1}^n\sim \mathbb{P};\atop\mathbb{P}\in \mathcal{P}^n_{\varepsilon_0,\alpha_1,\alpha_2}}\E\Bigg(\frac{\|\hat{\M \Sigma}_n - \M \Sigma_2^* \otimes \M \Sigma_1^*\|_\F^2}{pq}\Bigg)
\\
&\succsim \frac{1}{4}B_1B_2B_3
\\
&\succsim \frac{1}{4}\frac{1}{npq}\max\{k_pp,k_qq\}
\\
&\asymp\max\Big\{\frac{k_p}{qn},\frac{k_q}{pn}\Big\} 
\\
&\asymp\max\Big[\min\{(nq)^{\frac{1}{2\alpha_1 + 2}},p\}/qn,\min\{(np)^{\frac{1}{2\alpha_2 + 2}},q\}/pn\Big] 
\\
&=\max\Bigg[\min\Big\{\frac{p}{qn},(nq)^{\frac{1}{2\alpha_1 + 2} - 1}\Big\}, \min\Big\{\frac{q}{pn},(np)^{\frac{1}{2\alpha_2 + 2} - 1}\Big\}\Bigg].
\ee
\qed
\subsection{Proof of Theorem \ref{T:indi}}
\subsubsection{Preliminary}
We first present some preliminary to simplify our proof. Some relevant definitions in this proof can be found in Section \ref{notation:main}. 
\par
For simplicity, we only consider the convergence of $\big\|c_{1,\mathcal{B}}\hat{\M \Sigma}_1^{\mathcal{B}}(k_1) - \M\Sigma_1^*\big\|_\F^2$. When $\eta=\mathcal{T}$, the proof techniques are exactly the same. In addition, the convergence of $\big\|c_{2,\eta}\hat{\M \Sigma}_2^\eta(k_2) - \M\Sigma_2^*\big\|_\F^2$ can be shown directly by symmetry. 
\par 
Without loss of generality, we assume $p\geq q$. For particular $\hat{\M\Sigma}_1^{\mathcal{B}}(k_1)$ and $\hat{\M\Sigma}_2^{\mathcal{B}}(k_2)$, since  $\hat{\M\Sigma}_2^{\mathcal{B}}(k_2)\otimes \hat{\M\Sigma}_1^{\mathcal{B}}(k_1)$ is the optimal Frobenius-norm Kronecker product approximation of $\tilde{\M\Sigma}^{\MB}(k_1,k_2)$, \citet{Pitsianis1997} show that $\tilde{\M\Sigma}^{\MB}(k_1,k_2)$ can be presented by the following SVD form,
\bee\label{thm:ind:pre}
\xi\big\{\tilde{\M\Sigma}^{\MB}(k_1,k_2)\big\} = 
\begin{bmatrix}
\M U_1&\M U_1^{\bot}
\end{bmatrix}
\begin{bmatrix}
\hat{\M\Lambda}&\M 0_{q^2\times(p^2-q^2)}
\end{bmatrix}
\begin{bmatrix}
\M V_1^{\T}
\\
(\M V_1^{\bot})^{\T}
\end{bmatrix}, 
\ee
where $\hat{\M\Lambda} \equiv \diag\big\{
\hat{\varsigma}_1,\dots,\hat{\varsigma}_{\min{(p^2,q^2)=q^2 }}\big\} $, $\hat{\varsigma}_1\geq \hat{\varsigma}_2\geq \ldots \geq \hat{\varsigma}_{q^2}>0$, $\M U_1 \equiv \vecc\big\{\hat{\M \Sigma}_2^\MB(k_2)\big\}\big/\big\|\vecc\big\{\hat{\M \Sigma}_2^\MB(k_2)\big\}\big\|$, and $\M V_1 \equiv {\vecc\big\{\hat{\M \Sigma}_1^\MB(k_1)\big\}}\big/{\big\|\vecc\big\{\hat{\M \Sigma}_1^\MB(k_1)\big\}\big\|}$. On the other hand, by Lemma \ref{lemma:xi},
\bee\nonumber
\xi\big\{\M\Sigma_2^{*,\mathcal{B}}(k_2)\otimes\M\Sigma_1^{*,\mathcal{B}}(k_1)\big\} = \vecc\big\{\M\Sigma_2^{*,\mathcal{B}}(k_2)\big\}\cdot \vecc\big\{\M\Sigma_1^{*,\mathcal{B}}(k_1)\big\}^\T.
\ee
 Similar to \eqref{thm:ind:pre}, the SVD form of $\xi\big\{\M\Sigma_2^{*,\mathcal{B}}(k_2)\otimes\M\Sigma_1^{*,\mathcal{B}}(k_1)\big\} $ is 
\bee\nonumber
\xi\big\{\M\Sigma_2^{*,\mathcal{B}}(k_2)\otimes\M\Sigma_1^{*,\mathcal{B}}(k_1)\big\}  &= \begin{pmatrix}\tilde{\M U}_0 & \tilde{\M U}_{0}^{\perp}
\end{pmatrix}\cdot
\begin{pmatrix}\tilde{\varsigma} & \bds 0_{1\times (p^2 - 1)}
\\
\bds 0_{(q^2 - 1)\times1} & \bds 0_{(q^2 - 1)\times (p^2 - 1)}
\end{pmatrix}\cdot
\begin{pmatrix}
\tilde{\M V}_0^\T
\\
\big(\tilde{\M V}^\perp_0\big)^\T
\end{pmatrix}
\\
&=\tilde{\varsigma}\cdot\tilde{\M U}_0\tilde{\M V}_0^\T,
\ee
where $\tilde{\M U}_0 \equiv \vecc\big\{\M\Sigma_2^{*,\mathcal{B}}(k_2)\big\}\big/\big\|\vecc\big\{\M\Sigma_2^{*,\mathcal{B}}(k_2)\big\}\big\|$, $\tilde{\M V}_0 \equiv \vecc\big\{\M\Sigma_1^{*,\mathcal{B}}(k_1)\big\}\big/\big\|\vecc\big\{\M\Sigma_1^{*,\mathcal{B}}(k_1)\big\}\big\|$ and $\tilde{\varsigma} \equiv \|\vecc\big\{\M\Sigma_1^{*,\mathcal{B}}(k_1)\big\}\|\cdot \|\vecc\big\{\M\Sigma_2^{*,\mathcal{B}}(k_2)\big\}\|$.  Then by Lemma 1 in \citet{cai2018rate},
\bee\label{thm:pf:ind:pre:2}
\inf_{\delta\in\{0,1\}}\|\M V_1 - (-1)^{\delta}\cdot\tilde{\M V}_0\|^2\leq 2|\sin\Theta(\M V_1, \tilde{\M V}_0)|^2.
\ee
When \eqref{thm:pf:ind:pre:2} is satisfied with $\delta = 0$, i.e.,
\bee\label{thm:pf:ind:pre:3}
\|\M V_1 - \tilde{\M V}_0\|^2\leq 2|\sin\Theta(\M V_1, \tilde{\M V}_0)|^2,
\ee
we define $c_{a,\eta}$ for $a\in\{1,2\}$ and $\eta\in\{\mathcal{B},\mathcal{T}\}$ as
\bee\label{normalize:constant}
c_{a,\eta} = \begin{cases}
\|\M\Sigma_1^{*,\mathcal{B}}(k_1)\|_\F/\big\|\hat{\M\Sigma}_a^{\mathcal{B}}(k_a)\big\|_\F & \text{when } \eta = \mathcal{B}
\\
\|\M\Sigma_1^{*,\mathcal{T}}(k_1)\|_\F/\big\|\hat{\M\Sigma}_a^{\mathcal{T}}(k_a)\big\|_\F & \text{when } \eta = \mathcal{T}.
\end{cases}
\ee
When \eqref{thm:pf:ind:pre:2} is satisfied with $\delta = 1$, i.e., $\|\M V_1 + \tilde{\M V}_0\|^2\leq 2|\sin\Theta(\M V_1, \tilde{\M V}_0)|^2$, we define 
\bee\nonumber
c_{a,\eta} = \begin{cases}
-\|\M\Sigma_1^{*,\mathcal{B}}(k_1)\|_\F/\big\|\hat{\M\Sigma}_a^{\mathcal{B}}(k_a)\big\|_\F & \text{when } \eta = \mathcal{B}
\\
-\|\M\Sigma_1^{*,\mathcal{T}}(k_1)\|_\F/\big\|\hat{\M\Sigma}_a^{\mathcal{T}}(k_a)\big\|_\F & \text{when } \eta = \mathcal{T}.
\end{cases}
\ee
\par
In the main proof, we specifically consider the case that $\delta = 0$, and therefore $c_{a,\eta}$ is selected as \eqref{normalize:constant}. We also note the main proof can be directly adapted to the case of $\delta = 1$.
\subsubsection{Main Proof}
Now we bound $\big\|c_{1,\mathcal{B}}\hat{\M \Sigma}_1^{\mathcal{B}}(k_1) - \M\Sigma_1^*\big\|_\F^2/p$, where $c_{1,\mathcal{B}}$ is previously defined as $c_{1,\mathcal{B}}  = \|\M\Sigma_1^{*,\mathcal{B}}(k_1)\|_\F/\|\hat{\M\Sigma}_1^{\mathcal{B}}(k_1)\|_\F$. By triangle inequality, we have
\bee\label{thm:submatrix:final:1}
\frac{1}{p}\big\|c_{1,\mathcal{B}}\hat{\M \Sigma}_1^{\mathcal{B}}(k_1) - \M\Sigma_1^*\big\|_\F^2 &\leq \frac{2}{p}\big\|c_{1,\MB}\hat{\M \Sigma}_1^\MB(k_1) - \M\Sigma_1^{*,\mathcal{B}}(k_1)\big\|_\F^2 + \frac{2}{p}\big\| \M\Sigma_1^{*,\mathcal{B}}(k_1) - \M\Sigma_1^*\big\|_\F^2
\\
&\asymp\underbrace{\frac{1}{p}\Big\|\frac{\|\M\Sigma_1^{*,\mathcal{B}}(k_1)\|_\F}{\|\hat{\M \Sigma}_1^\MB(k_1) \|_\F}\hat{\M \Sigma}_1^{\mathcal{B}}(k_1) - \M\Sigma_1^{*,\mathcal{B}}(k_1)\Big\|_\F^2}_{I_1} + \underbrace{\frac{1}{p}\big\| \M\Sigma_1^{*,\mathcal{B}}(k_1) - \M\Sigma_1^*\big\|_\F^2}_{I_2}.
\ee
\par
We first give a brief summary of our proof framework. We will bound $I_2$ by Lemma \ref{l:outer}. Then  we will bound $I_1$ by combining the unilateral singular subspace perturbation bound \citep{cai2018rate} with the proving techniques we used in Lemma \ref{l:frate} and \ref{lemma:srate}.
\par
We first bound $I_2$. This can be regarded as a special degenerate case of Lemma \ref{l:outer} when $q=1,k_2 = 1$ and $\M\Sigma_2^* = 1$.
%Considering a degenerate scenario such that $q=1,k_2 = 1$ and $\M\Sigma_2^* = 1$ in Lemma \ref{l:outer}, 
The term $I_2$ can be directly bounded by 
\bee\label{thm:submatrix:final:2}
I_2 &= \frac{1}{p}\big\| \M\Sigma_1^{*,\mathcal{B}}(k_1) - \M\Sigma_1^*\big\|_\F^2
\\
&\leq \M I_{n,p}(k_1)\cdot k_1^{-\tilde\alpha_1},
\ee
where $\small\M I_{\eta,d}(k)$ and $\tilde{\alpha}_a$ are defined in Theorem \ref{T:indi}.
\par
Next, we bound $I_1$. By definitions of $\mathcal{F}(\varepsilon_0,\alpha),\mathcal{M}(\varepsilon_0,\alpha)$ and the matrix norm relationship $\|\M M\|_2\leq \|\M M\|_1$ for any symmetric matrix $\M M$, we have 
\bee\nonumber
\|\M\Sigma_1^{*,\mathcal{B}}(k_1) - \M\Sigma_1^{*}\|_2 &\leq \|\M\Sigma_1^{*,\mathcal{B}}(k_1) - \M\Sigma_1^{*}\|_1\leq C_0 k_1^{-\alpha_1}\precsim 1,
\ee
and $\|\M\Sigma_1^*\|_2 =\lambda_{\max}(\M\Sigma_1^*)\leq 1/\varepsilon_0$. Therefore by triangle inequality, $\|\M\Sigma_1^{*,\mathcal{B}}(k_1)\|_2\leq \|\M\Sigma_1^{*,\mathcal{B}}(k_1) - \M\Sigma_1^*\|_2 + \|\M\Sigma_1^*\|_2\precsim 1$. Since $\|\M M\|^2_\F\leq d\|\M M\|_2^2$ for any $d\times d$ matrix $\M M$, we have
\bee\nonumber
\|\M\Sigma_1^{*,\mathcal{B}}(k_1)\|_\F^2 &\leq p\|\M\Sigma_1^{*,\mathcal{B}}(k_1)\|_2^2 \precsim p .
\ee
Then $I_1$ can be bounded by
\bee\label{I1:form1}
I_1 &= \frac{1}{p}\Bigg\|\frac{\|\M\Sigma_1^{*,\mathcal{B}}(k_1)\|_\F}{\|\hat{\M \Sigma}_1^\MB(k_1) \|_\F}\hat{\M \Sigma}_1^\MB(k_1) - \M\Sigma_1^{*,\mathcal{B}}(k_1)\Bigg\|_\F^2
\\
&\leq \frac{\|\M\Sigma_1^{*,\mathcal{B}}(k_1)\|_\F^2}{p}\cdot\Big\|\frac{1}{\|\hat{\M \Sigma}_1^\MB(k_1) \|_\F}\hat{\M \Sigma}_1^\MB(k_1) - \frac{1}{\|\M\Sigma_1^{*,\mathcal{B}}(k_1)\|_\F}\M\Sigma_1^{*,\mathcal{B}}(k_1)\Big\|_\F^2
\\
&\precsim \frac{p}{p}\cdot \Bigg\|\frac{\hat{\M \Sigma}_1^\MB(k_1)}{\|\hat{\M \Sigma}_1^\MB(k_1) \|_\F} - \frac{\M\Sigma_1^{*,\mathcal{B}}(k_1)}{\|\M\Sigma_1^{*,\mathcal{B}}(k_1)\|_\F}\Bigg\|_\F^2
\\
&= \Bigg\|\frac{\vecc\big\{\hat{\M \Sigma}_1^\MB(k_1)\big\}}{\|\vecc\big\{\hat{\M \Sigma}_1^\MB(k_1)\big\}\|} - \frac{\vecc\big\{\M\Sigma_1^{*,\mathcal{B}}(k_1)\big\}}{\|\vecc\big\{\M\Sigma_1^{*,\mathcal{B}}(k_1)\big\}\|}\Bigg\|^2
\\
&=\|\M V_1 - \tilde{\M V}_0\|^2.
\ee
Combining \eqref{thm:pf:ind:pre:3} and \eqref{I1:form1}, we have
\bee\label{I1:form2}
I_1 &\precsim 2|\sin\Theta(\M V_1, \tilde{\M V}_0)|^2.
\ee
The $|\sin\Theta(\M V_1, \tilde{\M V}_0)|$ can be seen as the right singular subspace perturbation, of original matrix $\xi\big\{\M\Sigma_2^{*,\mathcal{B}}(k_2)\otimes\M\Sigma_1^{*,\mathcal{B}}(k_1)\big\}$ and its perturbation $\xi\big\{\tilde{\M\Sigma}^{\MB}(k_1,k_2)\big\}$. We use the rate-optimal subspace perturbation bound by \citet{cai2018rate}, to study the right singular subspace perturbation  $|\sin\Theta(\M V_1, \tilde{\M V}_0)|$.
\par
To make the proof comparable, we align most of our notation with \citet{cai2018rate}. In particular, we take 
\bee\nonumber
&\M X \equiv \xi\big\{\M\Sigma_2^{*,\mathcal{B}}(k_2)\otimes\M\Sigma_1^{*,\mathcal{B}}(k_1)\big\} ,
\\
&\hat{\M X} \equiv \xi\big\{\tilde{\M\Sigma}^{\MB}(k_1,k_2)\big\},
\\
&\M Z \equiv \xi\big\{\tilde{\M\Sigma}^{\MB}(k_1,k_2)\big\} - \xi\big\{\M\Sigma_2^{*,\mathcal{B}}(k_2)\otimes\M\Sigma_1^{*,\mathcal{B}}(k_1)\big\} ,
\ee
and thus $\M Z = \hat{\M X} - \M X$. Let $\M P_{\M M}$ be the projection operator of matrix $\M M$. When $\M M$ has orthonormal columns, In Section 2.1 of \citet{cai2018rate}, it has been shown that $\M P_{\M M} = \M M\M M^\T$. Then we further define $\M Z_{11}, \M Z_{12}, \M Z_{21}, \M Z_{22}$ as 
\bee\nonumber
\begin{pmatrix}
\M P_{\tilde{\M U}_0} 
\\
 \M P_{\tilde{\M U}_0^\perp}
\end{pmatrix}\cdot\M Z \cdot
\begin{pmatrix}
\M P_{\tilde{\M V}_0^{}} 
&
\M P_{{\tilde{\M V}_0}^\perp}
\end{pmatrix}
&=
\begin{pmatrix}
\tilde{\M U}_0\tilde{\M U}^\T_0\M Z \tilde{\M V}_0\tilde{\M V}_0^\T &\tilde{\M U}_0\tilde{\M U}^\T_0\M Z \tilde{\M V}_0^\perp\big(\tilde{\M V}_0^\perp\big)^\T
\\
\tilde{\M U}^\perp_0\big(\tilde{\M U}^\perp_0\big)^\T\M Z \tilde{\M V}_0\tilde{\M V}_0^\T & \tilde{\M U}^\perp_0\big(\tilde{\M U}^\perp_0\big)^\T\M Z\tilde{\M V}_0^\perp\big(\tilde{\M V}_0^\perp\big)^\T
\end{pmatrix}
\\
&\equiv
\begin{pmatrix}
\M Z_{11}&\M Z_{12}
\\
\M Z_{21} & \M Z_{22}
\end{pmatrix},
\ee
and $z_{b_1b_2} = \|\M Z_{b_1b_2}\|_2$ where $b_1,b_2\in\{1,2\}$.  We  note that since spectral norm is orthogonally invariant, we have
\bee\label{def:zab}
\begin{pmatrix}
z_{11} & z_{12}
\\
z_{21} & z_{22}
\end{pmatrix} &= \begin{pmatrix}
\|\tilde{\M U}_0\tilde{\M U}^\T_0\M Z \tilde{\M V}_0\tilde{\M V}_0^\T\|_2 & \|\tilde{\M U}_0\tilde{\M U}^\T_0\M Z \tilde{\M V}_0^\perp\big(\tilde{\M V}_0^\perp\big)^\T\|_2
\\
\|\tilde{\M U}^\perp_0\big(\tilde{\M U}^\perp_0\big)^\T\M Z \tilde{\M V}_0\tilde{\M V}_0^\T\|_2 & \|\tilde{\M U}^\perp_0\big(\tilde{\M U}^\perp_0\big)^\T\M Z\tilde{\M V}_0^\perp\big(\tilde{\M V}_0^\perp\big)^\T\|_2
\end{pmatrix}
\\
&=\begin{pmatrix}
\|\tilde{\M U}^\T_0\M Z \tilde{\M V}_0\|_2 & \|\tilde{\M U}^\T_0\M Z \tilde{\M V}_0^\perp\|_2
\\
\|\big(\tilde{\M U}^\perp_0\big)^\T\M Z \tilde{\M V}_0\|_2 & \|\big(\tilde{\M U}^\perp_0\big)^\T\M Z\tilde{\M V}_0^\perp\|_2
\end{pmatrix}.
\ee
In addition, since the spectral norm of any matrix $\M M$ is always larger or equal to the spectral norm of its submatrix, we have
\bee\nonumber
z_{ab} &\leq \Bigg\|\begin{pmatrix}
\tilde{\M U}^\T_0\M Z \tilde{\M V}_0 & \tilde{\M U}^\T_0\M Z \tilde{\M V}_0^\perp
\\
\big(\tilde{\M U}^\perp_0\big)^\T\M Z \tilde{\M V}_0 & \big(\tilde{\M U}^\perp_0\big)^\T\M Z\tilde{\M V}_0^\perp
\end{pmatrix}\Bigg\|_2
\\
&= \Bigg\|\begin{pmatrix}
\tilde{\M U}^\T_0
\\
\big(\tilde{\M U}^\perp_0\big)^\T
\end{pmatrix} \cdot\M Z\cdot \begin{pmatrix}
\tilde{\M V}_0  & \tilde{\M V}_0 ^\perp
\end{pmatrix} \Bigg\|_2
\\
&= \big\|\M Z\big\|_2 \quad (\|\cdot\|_2\text{ orthogonally invariant}),
\ee
for any $a,b\in\{1,2\}$. On the other hand, we have proved in \eqref{T2:B1:2} and \eqref{T2:add} that
\bee\label{zab:neg:2}
\frac{1}{pq}\E\big\|\M Z\big\|_2^2 &= \frac{1}{pq}\E\big\|\xi\big\{\tilde{\M\Sigma}^{\MB}(k_1,k_2)\big\} - \xi\big\{\M\Sigma_2^{*,\mathcal{B}}(k_2)\otimes\M\Sigma_1^{*,\mathcal{B}}(k_1)\big\}\big\|_2^2
\\
&\precsim
\begin{cases}\frac{k_1}{qn} + \frac{k_2}{pn}& pk_1 + qk_2 \precsim n
\\
\big(\frac{k_1k_2}{n}\big)\wedge\big(\frac{pk^2_1}{qn^2} + \frac{qk^2_2}{pn^2}  \big) & pk_1 + qk_2 \succ n
\end{cases}
\\
& = r_{\text{var}\mid \hat{\M \Sigma}_2\otimes \hat{\M\Sigma}_1}(k_1,k_2\mid p,q).
\ee
Therefore, for $k_1,k_2,p,q$ such that make $r_{\text{var}}(k_1,k_2\mid p,q) \rightarrow 0$, it directly implies $\E\|\M Z\|_2^2 \precsim pq\cdot r_{\text{var}\mid \hat{\M \Sigma}_2\otimes \hat{\M\Sigma}_1}(k_1,k_2\mid p,q)$. Since $\mathcal{L}^2$ convergence directly implies $\mathcal{O}_{\mathbb{P}}$ convergence (in particular, this can be shown by Markov inequality easily), we have
\bee\label{zab:neg}
z_{ab} &\leq \|\M Z\|_2
\\
&=\mathcal{O}_{\mathbb{P}}\big[\sqrt{pq}\big\{r_{\text{var}}(k_1,k_2\mid p,q)\big\}^{1/2}\big]
\\
&= \mathcal{o}_{\mathbb{P}}(\sqrt{pq}),
\ee
for $a,b\in\{0,1\}$.
\par
 We also let $\alpha\equiv \sigma_{\min}\big(\tilde{\M U}_0^\T\hat{\M X}\tilde{\M V}_0\big) = |\tilde{\M U}_0^\T\hat{\M X}\tilde{\M V}_0|$ and $\beta \equiv \big\|\big(\tilde{\M U}_0^{\perp}\big)^\T\hat{\M X}\big(\tilde{\M V}_{0}^\perp\big)\big\|_2$. By Theorem 1 in \citet{cai2018rate}, we have
\bee\label{thm:subm:target}
|\sin\Theta(\M V_1, \tilde{\M V}_0)|\leq \frac{\alpha z_{12} + \beta z_{21}}{\alpha^2 - \beta^2 - z_{12}^2\wedge z_{21}^2},
\ee
if $\alpha^2 >\beta^2 + z_{12}^2 \wedge z_{21}^2$. Next, we bound all quantities defined above. 
\par
We first lower bound $\alpha$. By definition, we can write
\bee\nonumber
\tilde{\M U}_0^\T\hat{\M X}\tilde{\M V}_0 &= \tilde{\M U}_0^\T\M Z\tilde{\M V}_0 + \tilde{\M U}_0^\T\M X\tilde{\M V}_0
\\
&= \tilde{\M U}_0^\T\M Z\tilde{\M V}_0 + \tilde{\varsigma}\cdot\tilde{\M U}_0^\T\tilde{\M U}_0\tilde{\M V}_0^\T\tilde{\M V}_0
\\
&=\tilde{\M U}_0^\T\M Z\tilde{\M V}_0 + \tilde{\varsigma}
\\
&= \tilde{\M U}_0^\T\M Z\tilde{\M V}_0+ \|\vecc\big\{\M\Sigma_1^{*,\mathcal{B}}(k_1)\big\}\|\cdot \|\vecc\big\{\M\Sigma_2^{*,\mathcal{B}}(k_2)\big\}\|
\\
&= \tilde{\M U}_0^\T\M Z\tilde{\M V}_0+ \|\M\Sigma_1^{*,\mathcal{B}}(k_1)\|_\F\cdot \|\M\Sigma_2^{*,\mathcal{B}}(k_2)\|_\F,
\ee
and thus $\alpha = |\tilde{\M U}_0^\T\hat{\M X}\tilde{\M V}_0| \geq \Big|\|\M\Sigma_1^{*,\mathcal{B}}(k_1)\|_\F\cdot \|\M\Sigma_2^{*,\mathcal{B}}(k_2)\|_\F - |\tilde{\M U}_0^\T\M Z\tilde{\M V}_0|\Big|$. We consider three terms 
$\|\M\Sigma_1^{*,\mathcal{B}}(k_1)\|_\F\cdot \|\M\Sigma_2^{*,\mathcal{B}}(k_2)\|_\F $ and $|\tilde{\M U}_0^\T\M Z\tilde{\M V}_0|$, respectively. 
 Note that $\M\Sigma_1^{*}\in\mathcal{M}(\varepsilon_0,\alpha_1)$ or $\mathcal{F}(\varepsilon_0,\alpha_1)$, and any matrix in  $\mathcal{M}(\varepsilon_0,\alpha_1)$ or $\mathcal{F}(\varepsilon_0,\alpha_1)$ has the smallest eigenvalue that is greater or equal to $\varepsilon_0/2$. Thus we have all diagonal entries of $\M\Sigma_1^{*}$ are greater or equal to $\varepsilon_0/2$, since $\M\Sigma_1^{*}$ is positive definite. As $k_1\geq 1$ and $\M\Sigma_1^{*,\mathcal{B}}(k_1)$ shares same diagonals with $\M\Sigma_1^{*}$, we have
\bee\nonumber
\|\M\Sigma_1^{*,\mathcal{B}}(k_1)\|_\F &\geq\bigg[\sum_{i = 1}^p\Big[\M\Sigma_1^{*,\mathcal{B}}(k_1)\Big]_{ii}^2\bigg]^{1/2}
\geq \sqrt{p}\cdot\frac{\varepsilon_0}{2}\succsim \sqrt{p}.
\ee
Similarly, we have $\|\M\Sigma_2^{*,\mathcal{B}}(k_2)\|_\F \succsim  \sqrt{q}$ and thus $\|\M\Sigma_1^{*,\mathcal{B}}(k_1)\|_\F\cdot \|\M\Sigma_2^{*,\mathcal{B}}(k_2)\|_\F  \succsim \sqrt{pq}$. On the other hand, as $|\tilde{\M U}_0^\T\M Z\tilde{\M V}_0| = z_{11}$, by \eqref{zab:neg}, we have $|\tilde{\M U}_0^\T\M Z\tilde{\M V}_0| \prec \sqrt{pq}$ with probability approaching $1$. Thus $|\tilde{\M U}_0^\T\M Z\tilde{\M V}_0|$ is negligible compared to $\|\M\Sigma_1^{*,\mathcal{B}}(k_1)\|_\F\cdot \|\M\Sigma_2^{*,\mathcal{B}}(k_2)\|_\F$ with probability approaching $1$. In sum, with probability approaching $1$,
\bee\label{thm:submatrix:target:dom:1}
\alpha \succsim \sqrt{pq}.
\ee 
\par
Next we show the lower bound of the denominator on right-hand side of \eqref{thm:subm:target}, i.e., $\alpha^2 - \beta^2 - z_{12}^2\wedge z_{21}^2$. For $\beta$, we decompose $\big(\tilde{\M U}_0^{\perp}\big)^\T\hat{\M X}\big(\tilde{\M V}_{0}^\perp\big)$ as 
\bee\nonumber
\big(\tilde{\M U}_0^{\perp}\big)^\T\hat{\M X}\big(\tilde{\M V}_{0}^\perp\big) &= \big(\tilde{\M U}_0^{\perp}\big)^\T(\hat{\M X} - \M X)\big(\tilde{\M V}_{0}^\perp\big) + \big(\tilde{\M U}_0^{\perp}\big)^\T \M X\big(\tilde{\M V}_{0}^\perp\big)
\\
&=\big(\tilde{\M U}_0^{\perp}\big)^\T\M Z\big(\tilde{\M V}_{0}^\perp\big) + \tilde{\varsigma}\cdot\underbrace{\big(\tilde{\M U}_0^{\perp}\big)^\T\tilde{\M U}_0\tilde{\M V}_0^\T\big(\tilde{\M V}_{0}^\perp\big)}_{ = 0}
\\
&= \M Z_{22}.
\ee
So $z_{22} = \|\M Z_{22}\|_2 = \|\big(\tilde{\M U}_0^{\perp}\big)^\T\hat{\M X}\big(\tilde{\M V}_{0}^\perp\big)\|_2 = \beta$ by definition of $\beta$. Then by \eqref{zab:neg}, we have
\bee\label{thm:submatrix:target:dom:2}
\beta^2 &= \mathcal{O}_{\mathbb{P}}\big[{pq}\cdot r_{\text{var}\mid \hat{\M \Sigma}_2\otimes \hat{\M\Sigma}_1}(k_1,k_2\mid p,q)\big] 
\\
&=\mathcal{o}_{\mathbb{P}}(pq).
\ee
Also by \eqref{zab:neg} we have 
\bee\label{thm:submatrix:target:dom:3}
z_{12}^2 \wedge z_{21}^2 = \mathcal{o}_{\mathbb{P}}(pq).
\ee Combining \eqref{thm:submatrix:target:dom:1}, \eqref{thm:submatrix:target:dom:2} and \eqref{thm:submatrix:target:dom:3}, with probability approaching $1$,
\bee\label{thm:submatrix:final1}
\alpha^2 - \beta^2 - z_{12}^2\wedge z_{21}^2 &\asymp \alpha^2\succsim pq,
\ee
and 
\bee\label{thm:submatrix:final2}
&\alpha^2 \succ\beta^2 + z_{12}^2 \wedge z_{21}^2.
\ee
\par
Now we give a sharper upper bound for the numerator of \eqref{thm:subm:target}, i.e., $\alpha z_{12} + \beta z_{21}$. Thus, we need to bound $z_{12},z_{21}$ carefully. We focus on $z_{12}$ as the bound for $z_{21}$ can be derived similarly by symmetry.

First, by \eqref{zab:neg:2} and \eqref{zab:neg}, we have a trivial bound 
\bee\label{submatrix:z12:1}
z_{12} &= \mathcal{O}_{\mathbb{P}}\Big[\sqrt{pq}\cdot \big\{r_{\text{var}\mid \hat{\M \Sigma}_2\otimes \hat{\M\Sigma}_1}(k_1,k_2\mid p,q)\big\}^{1/2}\Big]
\\
&= \begin{cases}\mathcal{O}_{\mathbb{P}}\Big[\sqrt{\frac{pk_1}{n}} + \sqrt{\frac{qk_2}{n}}\Big]& pk_1 + qk_2 \precsim n
\\
\mathcal{O}_{\mathbb{P}}\Big[\big(\sqrt{\frac{k_1k_2pq}{n}}\big)\wedge\big(\frac{pk_1}{n} + \frac{qk_2}{n}  \big)\Big] & pk_1 + qk_2 \succ n
\end{cases},
\ee
where $r_{\text{var}\mid \hat{\M \Sigma}_2\otimes \hat{\M\Sigma}_1}(k_1,k_2\mid p,q)$ is previously defined in \eqref{def:rwhole}. 
\par
In addition, by \eqref{def:zab}, we know $z_{12} = \|\tilde{\M U}^\T_0\M Z \tilde{\M V}_0^\perp\|_2$. By property of spectral norm, we have
\bee\nonumber
z_{12} &= \|\tilde{\M U}^\T_0\M Z \tilde{\M V}_0^\perp\|_2= \sup_{\bds u\in \mathcal{U}_{p^2}}\big|\tilde{\M U}^\T_0\M Z \tilde{\M V}_0^\perp\bds u\big|.
\ee
Since $\tilde{\M V}_0^\perp$ has orthonormal columns, we have $\tilde{\M V}_0^\perp\bds u\in \mathcal{U}_{p^2}$ if $\bds u\in\mathcal{U}_{p^2}$. Let $\bds u' = \tilde{\M V}_0^\perp\bds u$, one has
\bee\label{thm:submatrix:z12}
z_{12} &= \sup_{\bds u'  = \tilde{\M V}_0^\perp\bds u,
\atop \bds u\in\mathcal{U}_{p^2}}\big|\tilde{\M U}^\T_0\M Z \bds u'\big|\leq\sup_{\bds u''\in\mathcal{U}_{p^2}}\big|\tilde{\M U}^\T_0\M Z \bds u''\big|.
\ee  
As $\tilde{\M U}_0 \in \mathcal{U}_{q^2}$ and $\bds u''\in \mathcal{U}_{p^2}$, we employ same techniques, i.e. Hanson-Wright type inequality and $\epsilon$-nets, introduced in the proof of Lemma \ref{lemma:srate}, to derive the upper bounds of $\sup_{\bds u''\in\mathcal{U}_{p^2}}\big|\tilde{\M U}^\T_0\M Z \bds u''\big|$. In the following proof, we will use some notation defined in the proof of Lemma \ref{lemma:srate}. We first use triangle inequality to decompose
\bee\label{thm:submatrix:uzu}
\sup_{\bds u''\in\mathcal{U}_{p^2}}\big|\tilde{\M U}^\T_0\M Z \bds u''\big| &= \sup_{\bds u''\in\mathcal{U}_{p^2}}\Big|\tilde{\M U}^\T_0\Big[ \xi\big\{\tilde{\M\Sigma}^{\MB}(k_1,k_2)\big\} - \xi\big\{\M\Sigma_2^{*,\mathcal{B}}(k_2)\otimes\M\Sigma_1^{*,\mathcal{B}}(k_1)\big\}\Big] \bds u''\Big|
\\
&\leq
\sup_{\bds u''\in\mathcal{U}_{p^2}}\Big|\tilde{\M U}^\T_0\Big[ \xi\big\{\tilde{\M\Sigma}_0^{\MB}(k_1,k_2)\big\} - \xi\big\{\M\Sigma_2^{*,\mathcal{B}}(k_2)\otimes\M\Sigma_1^{*,\mathcal{B}}(k_1)\big\}\Big] \bds u''\Big|
\\
&+ \sup_{\bds u''\in\mathcal{U}_{p^2}}\Big|\tilde{\M U}^\T_0\Big[ \xi\big\{\tilde{\M\Sigma}_0^{\MB}(k_1,k_2)\big\} - \xi\big\{\tilde{\M\Sigma}^{\MB}(k_1,k_2)\big\}\Big] \bds u''\Big|
\\
&=\sup_{\bds u''\in\mathcal{U}_{p^2}}\Big|\tilde{\M U}^\T_0\Delta^{\mathcal{B}}_n \bds u''\Big| + \sup_{\bds u''\in\mathcal{U}_{p^2}}\Big|\tilde{\M U}^\T_0\mathcal{H}_{n,k_1,k_2}^{\mathcal{B}} \bds u''\Big|.
\ee
Similar to \eqref{l:srate:vdu}--\eqref{step1:final},  we can show
\bee\nonumber
\sup_{\bds u''\in\mathcal{U}_{p^2}}\Big|\tilde{\M U}^\T_0\Delta^{\mathcal{B}}_n \bds u''\Big| = \sup_{\bds u^*\in \mathcal{U}^{\mathcal{B}}_{p^2}(k_1)}\Big|\frac{1}{n}\sum _{i = 1}^nI_i\big(\bds u^*,\tilde{\M U}^\T_0\big)\Big|,
\ee
recalling that $I_i(\cdot,\cdot)$ is defined in \eqref{l:srate:back}. In \eqref{lm:srate:coreprob}, we show the tail probability bound of $\Big|\frac{1}{n}\sum _{i = 1}^nI_i\big(\bds u^*,\bds v^* = \tilde{\M U}^\T_0\big)\Big|$ for any unit vector $\bds u^*$. Then by a similar $\varepsilon-$net argument used in step 1.3 and \eqref{lm:srate:con} in proof of Lemma \ref{lemma:srate}, we have
\bee\label{tm:submatrix:main1}
\p\Bigg\{\Big(\sup_{\bds u''\in\mathcal{U}_{p^2}}\Big|\tilde{\M U}^\T_0\Delta^{\mathcal{B}}_n \bds u''\Big|\Big)^2 \geq t\Bigg\} &\leq 2a_1^{pk_1}\exp\Big[-n\min\{a''_2t,a''_3\sqrt{t}\}\Big]
\\
&= 2\exp\Big[pk_1/\log(a_1)-n\min\{a''_2t,a''_3\sqrt{t}\}\Big],
\ee
where $a_1 = 7^{\max(c_u,c_v)}>1, a''_2 =4 C''/(9K_3^2(\rho,\varepsilon_0)), a_3 = 2C''/3K_3(\rho,\epsilon)$ and $C'', K_3(\rho,\epsilon)$ are defined in the proof of lemma \ref{lemma:srate}. Note here we only account for the complexity of $\mathcal{U}^{\mathcal{B}}_{p^2}(k_1)$, instead of accounting for the complexities of both $\mathcal{U}^{\mathcal{B}}_{p^2}(k_1)$ and $\mathcal{U}^{\mathcal{B}}_{q^2}(k_2)$ in the original arguments in the proof of Lemma \ref{lemma:srate}.
\par
Now we choose $t = t_n = C_2\cdot \begin{cases}
\frac{pk_1}{n} & pk_1 \precsim n
\\
\frac{p^2k_1^2}{n^2} & pk_1 \succ n
\end{cases}$ with any positive $C_2$ that is larger than $1/\log(a_1)$. If $pk_1\precsim n$, since $p\geq1, k_1\geq 1$, we have 
\bee\label{tm:submatrix:main2}
\p\Bigg\{\Big(\sup_{\bds u''\in\mathcal{U}_{p^2}}\Big|\tilde{\M U}^\T_0\Delta^{\mathcal{B}}_n \bds u''\Big|\Big)^2 \geq t_n\Bigg\}  &\leq 2\exp\Big[\{1/\log (a_1)-C_2\}pk_2\Big]
\\
&\leq 2\exp\Big[1/\log (a_1)-C_2\Big] ,
\ee
where the right-hand side of the above inequality can be arbitrarily small when $C_2$ is sufficiently large. If $pk_1\succ n$, we have $p^2k^2_1/n^2 \succ pk_1/n$ and therefore when $n\rightarrow +\infty$,
\bee\nonumber
pk_1/\log(a_1)-nC_2\min\{a''_2t_n,a''_3\sqrt{t_n}\} &= pk_1/\log(a_1) - nC_2\min\{a''_2p^2k^2_1/n^2,a''_3pk_1/n\}
\\
&= pk_1/\log(a_1) - nC_2a''_3pk_1/n
\\
&=\{1/\log(a_1) - C_2a_3''\}\cdot pk_1
\\
&\rightarrow -\infty,
\ee
when $C_2$ is sufficiently large. This implies when $C_2$ is sufficiently large, 
\bee\label{tm:submatrix:main:alter}
\p\Bigg\{\Big(\sup_{\bds u''\in\mathcal{U}_{p^2}}\Big|\tilde{\M U}^\T_0\Delta^{\mathcal{B}}_n \bds u''\Big|\Big)^2 \geq t_n\Bigg\} \rightarrow 0,
\ee
if $pk_1 \succ n$ and $n\rightarrow +\infty$.
\par
 Combining \eqref{tm:submatrix:main1}, \eqref{tm:submatrix:main2} and \eqref{tm:submatrix:main:alter}, we conclude
\bee\nonumber
\Big(\sup_{\bds u''\in\mathcal{U}_{p^2}}\Big|\tilde{\M U}^\T_0\Delta^{\mathcal{B}}_n \bds u''\Big|\Big)^2 = 
\begin{cases}
\mathcal{O}_{\mathbb{P}}\big(\frac{pk_1}{n}\big) & pk_1 \precsim n
\\
\mathcal{O}_{\mathbb{P}}\big(\frac{p^2k^2_1}{n^2}\big) & pk_1 \succ n
\end{cases}.
\ee
Also similar to the proof of Lemma \ref{lemma:srate}, $\sup_{\bds u''\in\mathcal{U}_{p^2}}\Big|\tilde{\M U}^\T_0\mathcal{H}_{n,k_1,k_2}^{\mathcal{B}} \bds u''\Big|$ is negligible compared to $\sup_{\bds u''\in\mathcal{U}_{p^2}}\Big|\tilde{\M U}^\T_0\Delta^{\mathcal{B}}_n \bds u''\Big|$. By \eqref{thm:submatrix:z12}, \eqref{thm:submatrix:uzu} and the above result, we finally have another upper bound for $z_{12}$,
\bee\label{submatrix:z12:2}
z_{12} = \begin{cases}
\mathcal{O}_{\mathbb{P}}\big(\sqrt{\frac{pk_1}{n}}\big) & pk_1 \precsim n
\\
\mathcal{O}_{\mathbb{P}}\big({\frac{pk_1}{n}}\big) & pk_1 \succ n
\end{cases}.
\ee
Now we combine the rates in \eqref{submatrix:z12:1} and \eqref{submatrix:z12:2} for $z_{12}$.
\begin{itemize} 
\item[(i)] When $pk_1 \precsim n$, \eqref{submatrix:z12:2} yields $z_{12} = \mathcal{O}_{\mathbb{P}}\Big(\sqrt{\frac{pk_1}{n}}\Big)$. Since 
$$
\sqrt{\frac{pk_1}{n}}\precsim \sqrt{\frac{pk_1}{n}}+\sqrt{\frac{qk_2}{n}},
$$ 
we know the rate $\mathcal{O}_{\mathbb{P}}\Big(\sqrt{\frac{pk_1}{n}}\Big)$ is faster than the rate of $z_{12}$  implied in \eqref{submatrix:z12:1}. Therefore $\mathcal{O}_{\mathbb{P}}\Big(\sqrt{\frac{pk_1}{n}}\Big)$ is the optimal rate of $z_{12}$ when $pk_1\precsim n$.
\item[(ii)] When $pk_1 \succ n$, \eqref{submatrix:z12:2} yields the rate $z_{12} = \mathcal{O}_{\mathbb{P}}({\frac{pk_1}{n}})$   and \eqref{submatrix:z12:1} yields the rate $z_{12} = \mathcal{O}_{\mathbb{P}}\Big\{\big(\sqrt{\frac{k_1k_2pq}{n}}\big)\wedge\big(\frac{pk_1}{n} + \frac{qk_2}{n}  \big)\Big\}$. Since ${\frac{pk_1}{n}} \precsim \frac{pk_1}{n} + \frac{qk_2}{n} $, and therefore
$$
\frac{pk_1}{n} \wedge \big(\sqrt{\frac{k_1k_2pq}{n}}\big)\wedge\big(\frac{pk_1}{n} + \frac{qk_2}{n}  \big) \asymp \big(\sqrt{\frac{k_1k_2pq}{n}}\big) \wedge \frac{pk_1}{n},
$$
we know the rate of $z_{12}$ when $pk_1 \succ n$, is  $\mathcal{O}_{\mathbb{P}}\Big(\sqrt{\frac{k_1k_2pq}{n}}\wedge\frac{pk_1}{n}\Big)$ after combining \eqref{submatrix:z12:1} and \eqref{submatrix:z12:2}.
\end{itemize}
In summary, we have
\bee\label{submatrix:z12:final1}
z_{12} =  \begin{cases}\mathcal{O}_{\mathbb{P}}\Big(\sqrt{\frac{pk_1}{n}} \Big)& pk_1 \precsim n
\\
\mathcal{O}_{\mathbb{P}}\Big(\sqrt{\frac{k_1k_2pq}{n}}\wedge\frac{pk_1}{n}\Big) & pk_1 \succ n
\end{cases}.
\ee
By symmetry, we also have
\bee\label{submatrix:z12:final2}
z_{21} =  \begin{cases}\mathcal{O}_{\mathbb{P}}\Big(\sqrt{\frac{qk_2}{n}} \Big)& qk_2 \precsim n
\\
\mathcal{O}_{\mathbb{P}}\Big(\sqrt{\frac{k_1k_2pq}{n}}\wedge\frac{qk_2}{n}\Big) & qk_2 \succ n
\end{cases}.
\ee
\par
By \eqref{thm:submatrix:final2}, the condition \eqref{thm:subm:target} holds with  probability approaching $1$. Then, by \eqref{I1:form2}, \eqref{thm:subm:target}, \eqref{thm:submatrix:final1}, \eqref{submatrix:z12:final1} and \eqref{submatrix:z12:final2},  we have with probability approaching $1$,
\bee\label{thm:submatrix:final:3}
I_1 &\precsim |\sin\Theta(\M V_1,\tilde{\M V}_0)|^2
\\
&\precsim \frac{\alpha^2 z^2_{12} + \beta^2 z^2_{21}}{(\alpha^2 - \beta^2 - z_{12}^2\wedge z_{21}^2)^2}
\\
&\asymp\frac{\alpha^2 z^2_{12} + \beta^2 z^2_{21}}{\alpha^4}
\\
&\precsim \frac{z^2_{12} }{\alpha^2} + \frac{\beta^2 z^2_{21}}{\alpha^4},
\ee
which implies
\bee\nonumber
\\
I_1&=\mathcal{O}_{\mathbb{P}}\Bigg\{ \begin{cases}
\frac{k_1}{qn} & pk_1 \precsim n
\\
\frac{k_1k_2}{n}\wedge\frac{pk_1^2}{qn^2} & pk_1 \succ n
\end{cases}+ r_{\text{var}\mid \hat{\M \Sigma}_2\otimes \hat{\M\Sigma}_1}(k_1,k_2\mid p,q) \cdot \begin{cases}
\frac{k_2}{pn} & qk_2 \precsim n
\\
\frac{k_1k_2}{n}\wedge\frac{qk_2^2}{pn^2} & qk_2 \succ n
\end{cases} \Bigg\} 
\\
&=\mathcal{O}_{\mathbb{P}}\big\{ r_{\text{var}\mid \hat{\M \Sigma}_1}(k_1,k_2\mid p,q)\big\} + \mathcal{O}_{\mathbb{P}}\big\{r_{\text{var}\mid \hat{\M \Sigma}_2\otimes \hat{\M\Sigma}_1}(k_1,k_2\mid p,q) \cdot  r_{\text{var}\mid \hat{\M \Sigma}_2}(k_1,k_2\mid p,q)\big\},
\ee
where $$r_{\text{var}\mid \hat{\M \Sigma}_1}(k_1,k_2\mid p,q)= \begin{cases}
\frac{k_1}{qn} & pk_1 \precsim n
\\
\frac{k_1k_2}{n}\wedge\frac{pk_1^2}{qn^2} & pk_1 \succ n
\end{cases},$$ and $$r_{\text{var}\mid \hat{\M \Sigma}_2}(k_1,k_2\mid p,q)= \begin{cases}
\frac{k_2}{pn} & qk_2 \precsim n
\\
\frac{k_1k_2}{n}\wedge\frac{qk_2^2}{pn^2} & qk_2 \succ n
\end{cases}.$$ Combining \eqref{thm:submatrix:final:1}, \eqref{thm:submatrix:final:2} and \eqref{thm:submatrix:final:3}, we finally have
\bee\nonumber
&\frac{1}{p}\big\|c_{1,\mathcal{B}}\hat{\M \Sigma}_1^{\mathcal{B}}(k_1) - \M\Sigma_1^*\big\|_\F^2 
\\
&= \mathcal{O}_{\mathbb{P}}\Big\{ r_{\text{var}\mid \hat{\M \Sigma}_1}(k_1,k_2\mid p,q)+r_{\text{var}\mid \hat{\M \Sigma}_2\otimes \hat{\M\Sigma}_1}(k_1,k_2\mid p,q) \cdot  r_{\text{var}\mid \hat{\M \Sigma}_2}(k_1,k_2\mid p,q)+ \M I_{n,p}(k_1)\cdot k_1^{-\tilde\alpha_1}\Big\}.
\ee
\subsection{Proof of Theorem \ref{T4}}
The notation of this proof is mainly contained in Section \ref{notation:robust}. By setting $\widetilde{\boldsymbol\Sigma}^\diamond = \widecheck{\M \Sigma}_{\eta}(k_1,k_2)$ and $\tilde{\M\Sigma}^* = \M\Sigma^*_2\otimes \M\Sigma_1^*$ in Lemma \ref{T1} and taking expectation, we have
\bee\nonumber
\E\Big(\frac{\|\hat{\M \Sigma}_2^{\mathcal{R},\eta}\kii\otimes\hat{\M \Sigma}_1^{\mathcal{R},\eta}\ki - \M\Sigma^*\|_{\F}^2}{pq}\Big) &\leq 8\cdot \frac{1}{pq}\E\big\| \xi\{\widecheck{\M \Sigma}_{\eta}(k_1,k_2)\} -\xi\{\M\Sigma^*\}\big\|^2_2.
\ee
Note that for the right-hand side, we can decompose
\bee\nonumber
\xi\{\widecheck{\M \Sigma}_{\eta}(k_1,k_2)\} - \xi\{\M\Sigma^*\}   &= \Big[\xi\{\widecheck{\M \Sigma}_{\eta}(k_1,k_2)\} - \xi\{{\M \Sigma}^{*,\eta}_{\mathcal{R}}(k_1,k_2)\}\Big] 
\\
&+ \Big[\xi\big\{{\M \Sigma}^{*,\eta}_{\mathcal{R}}(k_1,k_2)\big\} - \xi\big\{\M\Sigma_2^{*,{\eta}}(k_2)\otimes \M\Sigma_1^{*,{\eta}}(k_1)\big\}\Big]
\\
&+\Big[\xi\big\{\M\Sigma_2^{*,{\eta}}(k_2)\otimes \M\Sigma_1^{*,{\eta}}(k_1)\big\} - \xi\{\M\Sigma^*\} \Big].
\ee
By triangle inequality, we then have
\begin{align}\nonumber
&\E\Big(\frac{\|\hat{\M \Sigma}_2^{\mathcal{R},\eta}\kii\otimes\hat{\M \Sigma}_1^{\mathcal{R},\eta}\ki - \M\Sigma^*\|_{\F}^2}{pq}\Big)
\\\nonumber
&\leq 8\cdot 3 \cdot\Bigg[\frac{1}{pq}\E \big\|\xi\{\widecheck{\M \Sigma}_{\eta}(k_1,k_2)\} - \xi\{{\M \Sigma}^{*,\eta}_{\mathcal{R}}(k_1,k_2)\}\big\|_2^2 
\\\nonumber
&+ \frac{1}{pq}\E\big\|\xi\big\{{\M \Sigma}^{*,\eta}_{\mathcal{R}}(k_1,k_2)\big\} - \xi\big\{\M\Sigma_2^{*,{\eta}}(k_2)\otimes \M\Sigma_1^{*,{\eta}}(k_1)\big\}\big\|_2^2 +\frac{1}{pq}\E\big\|\xi\big\{\M\Sigma_2^{*,{\eta}}(k_2)\otimes \M\Sigma_1^{*,{\eta}}(k_1)\big\} - \xi\{\M\Sigma^*\}\big\|_2^2\Bigg]
\\
&\precsim\frac{1}{pq}\E\big\|\xi\{\widecheck{\M \Sigma}_{\eta}(k_1,k_2)\} - \xi\{{\M \Sigma}^{*,\eta}_{\mathcal{R}}(k_1,k_2)\}\big\|_2^2 
+ \frac{1}{pq}\big\|{\M \Sigma}^{*,\eta}_{\mathcal{R}}(k_1,k_2)-\M\Sigma_2^{*,{\eta}}(k_2)\otimes \M\Sigma_1^{*,{\eta}}(k_1)\big\|_\F^2 
\\\label{pf:rbthem:final:0}
&+\frac{1}{pq}\big\|\M\Sigma_2^{*,{\eta}}(k_2)\otimes \M\Sigma_1^{*,{\eta}}(k_1) - \M\Sigma^*\big\|_\F^2,
\end{align}
where the last inequality holds by Lemma \ref{lemma:xi} and the fact that the last two terms are nonrandom. 
\par
Our proof strategy is to bound the three terms on the right-hand side of the above inequality, respectively. Under the condition of Theorem \ref{T4}, the Lemmas \ref{lemma:robust:2}--\ref{lemma:robust:3} always hold. In addition, when $\zeta \geq 2$, $\E(|x^{(i)}_{l_1,l_2}\cdot x^{(i)}_{m_1,m_2}|^\zeta) \leq M$ implies 
\bee\nonumber
\E(|x^{(i)}_{l_1,l_2}\cdot x^{(i)}_{m_1,m_2}|^2)&\leq \E(|x^{(i)}_{l_1,l_2}\cdot x^{(i)}_{m_1,m_2}|^{2\cdot \frac{\zeta}{2}})^{2/\zeta}\leq M^{2/\zeta}\leq M
\ee
by Cauchy-Schwarz inequality, and thus Lemma \ref{lemma:robust:1} holds. Therefore, when $\zeta \geq 2$, we use Lemmas \ref{lemma:robust:1}--\ref{lemma:robust:3} to bound the three error terms. When $1<\zeta<2$, we use Lemmas \ref{lemma:robust:2}--\ref{lemma:robust:3} to bound the three error terms.
\par
\
\par
\noindent
 We first consider when $\zeta \geq 2$. For the first term, combining Lemmas \ref{lemma:robust:1} and \ref{lemma:robust:3}, we have
\bee\label{pf:rbthem:final:1}
\frac{1}{pq}\E\big\|\xi\{\widecheck{\M \Sigma}_{\eta}(k_1,k_2)\} - \xi\{{\M \Sigma}^{*,\eta}_{\mathcal{R}}(k_1,k_2)\}\big\|_2^2 \precsim \frac{k_1k_2}{n}.
\ee
Also for the first term, by Lemma \ref{lemma:robust:2} we have
\bee\label{pf:rbthem:final:2}
\frac{1}{pq}\E\big\|\xi\{\widecheck{\M \Sigma}_{\eta}(k_1,k_2)\} - \xi\{{\M \Sigma}^{*,\eta}_{\mathcal{R}}(k_1,k_2)\}\big\|_2^2\precsim (J_n\tau)^4\times\begin{cases}
\frac{k_1}{qn} + \frac{k_2}{pn} & {\rm ~if~} pk_1 + qk_2 \precsim n
\\
\frac{pk^2_1}{qn^2} + \frac{qk^2_2}{pn^2} & {\rm ~if~} pk_1 + qk_2 \succ n.
\end{cases}
\ee
For the second term, by Lemma \ref{lemma:robust:3}, we have
\bee\label{pf:rbthem:final:3}
\frac{1}{pq}\big\|{\M \Sigma}^{*,\eta}_{\mathcal{R}}(k_1,k_2)-\M\Sigma_2^{*,{\eta}}(k_2)\otimes \M\Sigma_1^{*,{\eta}}(k_1)\big\|_\F^2 \precsim{k_1k_2}\cdot\tau^{-4(\zeta - 1)}.
\ee
For the third term, by Lemma \ref{l:outer}, we have
\bee\label{pf:rbthem:final:4}
\frac{1}{pq}\big\|\M\Sigma_2^{*,{\eta}}(k_2)\otimes \M\Sigma_1^{*,{\eta}}(k_1) - \M\Sigma^*\big\|_\F^2 \precsim \M I_{\eta,p}(k_1)\cdot k_1^{-\tilde{\alpha}_1} +  \M I_{\eta,q}(k_2)\cdot k_2^{-\tilde{\alpha}_2},
\ee
where $$\small\M I_{\eta,d}(k) = \begin{cases}
\M I(k< d-1), & \eta = \mathcal{B};
\\
\M I (k< 2d-2), & \eta = \mathcal{T},
\end{cases}$$
and  $$\small\tilde{\alpha}_{a} = \begin{cases}
2\alpha_a & \text{when }\M \Sigma_1^{*}\in \mathcal{F}(\varepsilon_0, \alpha_1),\M \Sigma_2^{*} \in \mathcal{F}(\varepsilon_0, \alpha_2);
\\
2\alpha_a + 1 & \text{when } \M \Sigma_1^{*}\in \mathcal{M}(\varepsilon_0, \alpha_1),\M \Sigma_2^{*} \in \mathcal{M}(\varepsilon_0, \alpha_2),
\end{cases}$$ for $a\in\{1,2\}$ and $\eta\in\{\mathcal{T},\mathcal{B}\}$.
Combining \eqref{pf:rbthem:final:0}, \eqref{pf:rbthem:final:1}, \eqref{pf:rbthem:final:2}, \eqref{pf:rbthem:final:3} and \eqref{pf:rbthem:final:4}, we show
\bee\nonumber
&\E\Big[\frac{\|\hat{\M \Sigma}_2^{\mathcal{R},\eta}\kii\otimes\hat{\M \Sigma}_1^{\mathcal{R},\eta}\ki - \M\Sigma^*\|_{\F}^2}{pq}\Big] 
\\
&= \Bigg[\frac{k_1k_2}{n} \Bigg]\wedge \Bigg[(J_n\tau)^4\Big\{\frac{k_1}{qn} + \frac{k_2}{pn}\Big\} \Bigg]  + k_1k_2\cdot\tau^{-4(\zeta - 1)} + \M I_{\eta,p}(k_1)\cdot k_1^{-\tilde{\alpha}_1} +  \M I_{\eta,q}(k_2)\cdot k_2^{-\tilde{\alpha}_2}
\\
&\precsim \Bigg[\frac{k_1k_2}{n} + k_1k_2\cdot\tau^{-4(\zeta - 1)}\Bigg]\wedge \Bigg[(J_n\tau)^4\Big\{\frac{k_1}{qn} + \frac{k_2}{pn}\Big\} + k_1k_2\cdot\tau^{-4(\zeta - 1)}\Bigg] + \M I_{\eta,p}(k_1)\cdot k_1^{-\tilde{\alpha}_1} +  \M I_{\eta,q}(k_2)\cdot k_2^{-\tilde{\alpha}_2},
\ee
when $pk_1 + qk_2\precsim n$ and 
\bee\nonumber
&\E\Big[\frac{\|\hat{\M \Sigma}_2^{\mathcal{R},\eta}\kii\otimes\hat{\M \Sigma}_1^{\mathcal{R},\eta}\ki - \M\Sigma^*\|_{\F}^2}{pq}\Big] 
\\
&\precsim \Bigg[\frac{k_1k_2}{n} + k_1k_2\cdot\tau^{-4(\zeta - 1)}\Bigg]\wedge \Bigg[(J_n\tau)^4\Big\{\frac{pk^2_1}{qn^2} + \frac{qk^2_2}{pn^2}\Big\} + k_1k_2\cdot\tau^{-4(\zeta - 1)}\Bigg] 
\\
&+ \M I_{\eta,p}(k_1)\cdot k_1^{-\tilde{\alpha}_1} +  \M I_{\eta,q}(k_2)\cdot k_2^{-\tilde{\alpha}_2},
\ee
when $pk_1 + qk_2\succsim n$. We discuss these two scenarios, respectively.
\begin{itemize}
\item When $pk_1 + qk_2 \precsim n$, by choosing $\tau \asymp \{pqk_1k_2n/(pk_1J_n^4 + qk_2J_n^4) \}^{1/(4\zeta)}$, the {\color{black} optimal} rate of $(J_n\tau)^4\Big\{\frac{k_1}{qn} + \frac{k_2}{pn}\Big\} + k_1k_2\cdot\tau^{-4(\zeta - 1)}$ is attained by 
\bee\label{rb:eq1}
(J_n\tau)^4\Big\{\frac{k_1}{qn} + \frac{k_2}{pn}\Big\} + k_1k_2\cdot\tau^{-4(\zeta - 1)}\asymp (k_1k_2)^{1/\zeta} \Big\{\frac{k_1J_n^4}{qn} + \frac{k_2J_n^4}{pn}\Big\}^{1-1/\zeta}.
\ee
If $(k_1k_2)^{1/\zeta} \Big\{\frac{k_1J_n^4}{qn} + \frac{k_2J_n^4}{pn}\Big\}^{1-1/\zeta}\precsim \frac{k_1k_2}{n}$, by setting $\tau \asymp \{pqk_1k_2n/(pk_1J_n^4 + qk_2J_n^4) \}^{1/(4\zeta)}$, we have
$$
\E\Big[\frac{\|\hat{\M \Sigma}_2^{\mathcal{R},\eta}\kii\otimes\hat{\M \Sigma}_1^{\mathcal{R},\eta}\ki - \M\Sigma^*\|_{\F}^2}{pq}\Big] \precsim (k_1k_2)^{1/\zeta} \Big\{\frac{k_1J_n^4}{qn} + \frac{k_2J_n^4}{pn}\Big\}^{1-1/\zeta} +  \M I_{\eta,p}(k_1)\cdot k_1^{-\tilde{\alpha}_1} +  \M I_{\eta,q}(k_2)\cdot k_2^{-\tilde{\alpha}_2}.
$$
If $(k_1k_2)^{1/\zeta} \Big\{\frac{k_1J_n^4}{qn} + \frac{k_2J_n^4}{pn}\Big\}^{1-1/\zeta}\succsim \frac{k_1k_2}{n}$, by setting $\tau=+\infty$, we have
$$
\E\Big[\frac{\|\hat{\M \Sigma}_2^{\mathcal{R},\eta}\kii\otimes\hat{\M \Sigma}_1^{\mathcal{R},\eta}\ki - \M\Sigma^*\|_{\F}^2}{pq}\Big] \precsim \frac{k_1k_2}{n}+  \M I_{\eta,p}(k_1)\cdot k_1^{-\tilde{\alpha}_1} +  \M I_{\eta,q}(k_2)\cdot k_2^{-\tilde{\alpha}_2}.
$$
\item
When $pk_1 + qk_2 \succsim n$, by choosing $\tau \asymp \{pqk_1k_2n^2/(pk_1J_n^2 + qk_2J_n^2)^2 \}^{1/(4\zeta)}$, the {\color{black} optimal} rate of $(J_n\tau)^4\Big\{\frac{pk^2_1}{qn^2} + \frac{qk^2_2}{pn^2}\Big\} +k_1k_2\cdot\tau^{-4(\zeta - 1)}$ is attained by 
\bee\label{rb:eq2}
(J_n\tau)^4\Big\{\frac{pk^2_1}{qn^2} + \frac{qk^2_2}{pn^2}\Big\} + k_1k_2\cdot\tau^{-4(\zeta - 1)} \asymp (k_1k_2)^{1/\zeta} \Big\{\frac{pk^2_1J_n^4}{qn^2} + \frac{qk^2_2J_n^4}{pn^2}\Big\}^{1 - 1/\zeta}.
\ee
If $(k_1k_2)^{1/\zeta} \Big\{\frac{pk^2_1J_n^4}{qn^2} + \frac{qk^2_2J_n^4}{pn^2}\Big\}^{1 - 1/\zeta}\precsim \frac{k_1k_2}{n}$, by setting $\tau \asymp \{pqk_1k_2n^2/(pk_1J_n^2 + qk_2J_n^2)^2 \}^{1/(4\zeta)}$, we have
\bee\nonumber
\E\Big[\frac{\|\hat{\M \Sigma}_2^{\mathcal{R},\eta}\kii\otimes\hat{\M \Sigma}_1^{\mathcal{R},\eta}\ki - \M\Sigma^*\|_{\F}^2}{pq}\Big] &\precsim (k_1k_2)^{1/\zeta} \Big\{\frac{pk^2_1J_n^4}{qn^2} + \frac{qk^2_2J_n^4}{pn^2}\Big\}^{1 - 1/\zeta}
\\
&+  \M I_{\eta,p}(k_1)\cdot k_1^{-\tilde{\alpha}_1} +  \M I_{\eta,q}(k_2)\cdot k_2^{-\tilde{\alpha}_2}.
\ee
If $(k_1k_2)^{1/\zeta} \Big\{\frac{pk^2_1J_n^4}{qn^2} + \frac{qk^2_2J_n^4}{pn^2}\Big\}^{1 - 1/\zeta}\succsim \frac{k_1k_2}{n}$, by setting $\tau=+\infty$, we have
$$
\E\Big[\frac{\|\hat{\M \Sigma}_2^{\mathcal{R},\eta}\kii\otimes\hat{\M \Sigma}_1^{\mathcal{R},\eta}\ki - \M\Sigma^*\|_{\F}^2}{pq}\Big] \precsim \frac{k_1k_2}{n}+  \M I_{\eta,p}(k_1)\cdot k_1^{-\tilde{\alpha}_1} +  \M I_{\eta,q}(k_2)\cdot k_2^{-\tilde{\alpha}_2}.
$$
\end{itemize}
\par
\
\par
\noindent
Now we consider when $1<\zeta<2$. The only difference  is Lemma \ref{lemma:robust:1} no longer holds when $1<\zeta<2$. Thus, we can not apply Lemma \ref{lemma:robust:1} to bound the first term on the right-hand side of \eqref{pf:rbthem:final:0}. By applying other arguments similar to the case that $\zeta \geq 2$, we can show 
\bee\nonumber
&\E\Big[\frac{\|\hat{\M \Sigma}_2^{\mathcal{R},\eta}\kii\otimes\hat{\M \Sigma}_1^{\mathcal{R},\eta}\ki - \M\Sigma^*\|_{\F}^2}{pq}\Big] 
\\
&\precsim (J_n\tau)^4\Big\{\frac{k_1}{qn} + \frac{k_2}{pn}\Big\} + k_1k_2\cdot\tau^{-4(\zeta - 1)} + \M I_{\eta,p}(k_1)\cdot k_1^{-\tilde{\alpha}_1} +  \M I_{\eta,q}(k_2)\cdot k_2^{-\tilde{\alpha}_2},
\ee
when $pk_1 + qk_2\precsim n$; and 
\bee\nonumber
&\E\Big[\frac{\|\hat{\M \Sigma}_2^{\mathcal{R},\eta}\kii\otimes\hat{\M \Sigma}_1^{\mathcal{R},\eta}\ki - \M\Sigma^*\|_{\F}^2}{pq}\Big] 
\\
&\precsim (J_n\tau)^4\Big\{\frac{pk^2_1}{qn^2} + \frac{qk^2_2}{pn^2}\Big\} + k_1k_2\cdot\tau^{-4(\zeta - 1)} + \M I_{\eta,p}(k_1)\cdot k_1^{-\tilde{\alpha}_1} +  \M I_{\eta,q}(k_2)\cdot k_2^{-\tilde{\alpha}_2},
\ee
when $pk_1 + qk_2\succsim n$. Then similar to \eqref{rb:eq1} and \eqref{rb:eq2}, the optimal rate is attained by
\bee\nonumber
\E&\Big[\frac{\|\hat{\M \Sigma}_2^{\mathcal{R},\eta}\kii\otimes\hat{\M \Sigma}_1^{\mathcal{R},\eta}\ki - \M\Sigma^*\|_{\F}^2}{pq}\Big]
\\ 
&\small{\precsim \begin{cases}
(k_1k_2)^{1/\zeta}\cdot\big(\frac{k_1J_n^4}{qn} + \frac{k_2J_n^4}{pn}\big)^{1 - 1/\zeta}+ \M I_{\eta,p}(k_1)\cdot k_1^{-\tilde{\alpha}_1} +  \M I_{\eta,q}(k_2)\cdot k_2^{-\tilde{\alpha}_2}, & pk_1 + qk_2 \precsim n,
\\
(k_1k_2)^{1/\zeta}\cdot\big(\frac{pk^2_1J_n^4}{qn^2} + \frac{qk^2_2J_n^4}{pn^2}  \big)^{1 - 1/\zeta}+ \M I_{\eta,p}(k_1)\cdot k_1^{-\tilde\alpha_1} +  \M I_{\eta,q}(k_2)\cdot k_2^{-\tilde\alpha_2}, & pk_1 + qk_2 \succ n,
\end{cases}}
\ee
when setting $\tau \asymp \begin{cases}
\big\{pqk_1k_2n/(pk_1J_n^4 + qk_2J_n^4) \big\}^{1/(4\zeta)}
& \text{ when } pk_1 + qk_2 \precsim n,
\\
\big\{pqk_1k_2n^2/(pk_1J_n^2 + qk_2J_n^2)^2 \big\}^{1/(4\zeta)} &\text{ when } pk_1 + qk_2 \succ n
.\end{cases}$
\qed
{\color{black}
\subsection{Proof of Theorem \ref{Ts}}
The notation of this proof is mainly contained in Section \ref{notation:main}. Note $\mathcal{F}(\varepsilon_0,\alpha)$ is a more general matrix class than $\mathcal{M}(\varepsilon_0,\alpha)$. Therefore, we only need to focus on the case that $\M \Sigma_1^{*}\in \mathcal{F}(\varepsilon_0, \alpha_1),\M \Sigma_2^{*} \in \mathcal{F}(\varepsilon_0, \alpha_2)$.
\par
We first prove a general result that $\small\E\big\|\M\Sigma^*-\hat{\M\Sigma}_2^{\eta,\mathcal{S}}(k_2)\otimes \hat{\M\Sigma}_1^{\eta,\mathcal{S}}(k_1)\big\|_2^2$ can be upper bounded by $4\E\|  \M\Sigma^* - \tilde{\M\Sigma}_{\eta}(k_1,k_2)\|_{2}^2$ for $\eta\in\{\mathcal{T},\mathcal{B}\}$. %{\color{red} Any reference about the result, in previous theorem etc?}{\color{orange}We shall be the first to derive such type of results. I think more former statement is note accurate. I have changed from ``we first note'' to ``first we prove''. In particular, it has been proved in \eqref{s257}}
In particular, observing that by \eqref{sec:dis:s},
\bee\nonumber
\big\| \widetilde{\M\Sigma}_{\eta}(k_1,k_2)-\hat{\M\Sigma}_2^{\mathcal{\eta,\mathcal{S}}}(k_2) \otimes \hat{\M\Sigma}_1^{\mathcal{\eta,\mathcal{S}}}(k_1)\big\|_2^2 &=  \min_{\bSigma_1, \bSigma_2}\big\| \widetilde{\M\Sigma}_{\eta}(k_1,k_2)-\bSigma_2 \otimes \bSigma_1\big\|_2^2
\\
&\leq\big\| \widetilde{\M\Sigma}_{\eta}(k_1,k_2)-\bSigma^*_2 \otimes \bSigma^*_1\big\|_2^2,
\ee
we have,
\begin{align}\nonumber
\small\E\big\|\M\Sigma^*-\hat{\M\Sigma}_2^{\eta,\mathcal{S}}(k_2)\otimes \hat{\M\Sigma}_1^{\eta,\mathcal{S}}(k_1)\big\|_2^2&\leq 2\E\|\M\Sigma^* - \tilde{\M\Sigma}_{\eta}(k_1,k_2)\|_{2}^2 + 2 \E\big\|\widetilde{\M\Sigma}_{\eta}(k_1,k_2)-\hat{\M\Sigma}_2^{\eta,\mathcal{S}}(k_2)\otimes \hat{\M\Sigma}_1^{\eta,\mathcal{S}}(k_1)\big\|_2^2
\tag*{\text{(Triangle inequality)}}
\\\nonumber
&\leq 2\E\|  \M\Sigma^* - \tilde{\M\Sigma}_{\eta}(k_1,k_2)\|_{2}^2 + 2 \E\big\| \widetilde{\M\Sigma}_{\eta}(k_1,k_2)-\bSigma^*_2 \otimes \bSigma^*_1\big\|_2^2
\\\label{s257}
&= 4\E\|  \M\Sigma^* - \tilde{\M\Sigma}_{\eta}(k_1,k_2)\|_{2}^2.
\end{align}
\par
 We now decompose $\tilde{\M \Sigma}_{\mathcal{T}}(k_1,k_2) - \M\Sigma^*$ as 
\bee\nonumber
&\tilde{\M \Sigma}_{\mathcal{T}}(k_1,k_2) - \M\Sigma^*
\\
&=\big[\tilde{\M \Sigma}_{\mathcal{T}}(k_1,k_2)-\M \Sigma_2^{*,\mathcal{T}}(k_2)\otimes\M \Sigma_1^{*,\mathcal{T}}(k_1)\big] + \big[\M \Sigma_2^{*,\mathcal{T}}(k_2)\otimes\M \Sigma_1^{*,\mathcal{T}}(k_1) - \M\Sigma^*\big]
\\
&=\big[\tilde{\M \Sigma}_{\mathcal{T}}(k_1,k_2)-\M \Sigma_2^{*,\mathcal{T}}(k_2)\otimes\M \Sigma_1^{*,\mathcal{T}}(k_1)\big] + \big[\M \Sigma_2^{*,\mathcal{T}}(k_2)\otimes\M \Sigma_1^{*,\mathcal{T}}(k_1) - \M\Sigma^*_2\otimes \M\Sigma_1^*\big],
\ee
as $\M\Sigma^* = \M\Sigma_2^*\otimes \M\Sigma_1^*$. Then triangle inequality implies
\bee\nonumber
\|\tilde{\M \Sigma}_{\mathcal{T}}(k_1,k_2) - \M\Sigma^*\|_2^2 \leq 2\big\|\tilde{\M \Sigma}_{\mathcal{T}}(k_1,k_2)-\M \Sigma_2^{*,\mathcal{T}}(k_2)\otimes\M \Sigma_1^{*,\mathcal{T}}(k_1)\big\|_2^2 + 2\big\|\M \Sigma_2^{*,\mathcal{T}}(k_2)\otimes\M \Sigma_1^{*,\mathcal{T}}(k_1) - \M\Sigma^*\big\|_2^2.
\ee
For the expectations of two terms on the right-hand side, we can bound the first term by Lemma \ref{l:dt} and the second term by Lemma \ref{l:sp:out}. Then we finally show
\bee\nonumber
&\E\|\tilde{\M \Sigma}_{\mathcal{T}}(k_1,k_2) - \M\Sigma^*\|_2^2 
\\
&\precsim 2\times\frac{k_1k_2+\log(\max\{p,q\})}{n} + 2\times\Big[\M I(k_1< 2p-2)\cdot k_1^{-2\alpha_1} +  \M I(k_2< 2q - 2)\cdot k_2^{-2\alpha_2}.\Big]
\\
&\asymp \frac{k_1k_2 + \log(\max\{p,q\})}{n} + \M I(k_1< 2p-2)\cdot k_1^{-2\alpha_1} +  \M I(k_2< 2q-2)\cdot k_2^{-2\alpha_2},
\ee
which combining with \eqref{s257} proves \eqref{Ts:res1}. \par
On the other hand, when neither of $p,q$ is diverging polynomially, our $\tilde{\M \Sigma}_{\mathcal{T}}(k_1,k_2)$ is actually a sample covariance estimator as we have chosen $k_1 = 2p - 2$, $k_2 = 2q -2$, and there is no tapering. Then we have 
\bee\label{lasteq}
\E\Big({\|\tilde{\M\Sigma}_{\mathcal{T}}(k_1,k_2) - \M\Sigma^*\|_{2}^2}\Big)\precsim \frac{pq}{n}
\ee
which is a well-known result for sample covariance estimator, under such low-dimensional scenario; see e.g. \citet{Cai2010}. Combining \eqref{lasteq} with \eqref{s257} immediately yields \eqref{Ts:res2}.
\qed
}
}

\end{document}